\documentclass[10pt, a4paper, UKenglish, twoside,notcite,final,pdftex]{memo-l}
\usepackage{amsfonts,amssymb,amsmath,amsthm,comment,mathtools}
\usepackage[UKenglish]{babel}
\usepackage{color}
\usepackage{caption}
\usepackage{subcaption}
\usepackage{pgf,tikz}
\usetikzlibrary{arrows}
\usepackage{mathtools}
 \usepackage{longtable}
\usepackage{bbm}
\usepackage[shortlabels]{enumitem} 
 
\usepackage[overload]{textcase}

\usepackage[utf8]{inputenc}
\usepackage[T1]{fontenc}

\usepackage{hyperref}
\usepackage{amsrefs}
\usepackage{color}

\numberwithin{equation}{section}
\usepackage{cancel}
\usepackage{graphicx}
\theoremstyle{plain}
\usepackage[mathscr]{eucal}
\usepackage{bbold}
\usepackage{upgreek}
\usepackage[]{listofsymbols}

\newtheorem{theorem}{Theorem}[chapter]
\newtheorem{proposition}[theorem]{Proposition}
\newtheorem{lemma}[theorem]{Lemma}
\newtheorem{corollary}[theorem]{Corollary}

\theoremstyle{definition}
\newtheorem{definition}[theorem]{Definition}
\newtheorem{example}[theorem]{Example}

\newtheorem{remark}[theorem]{Remark}

\numberwithin{section}{chapter}
\theoremstyle{plain}
\newtheorem{step}{Step}

\DeclarePairedDelimiter\abs{\lvert}{\rvert}%
\DeclarePairedDelimiter\norm{\lVert}{\rVert}%

\makeatletter
\newcommand{\opnorm}{\@ifstar\@opnorms\@opnorm}
\newcommand{\@opnorms}[1]{%
  \left|\mkern-3mu\left|\mkern-3mu\left|
   #1
  \right|\mkern-3mu\right|\mkern-3mu\right|
}
\newcommand{\@opnorm}[2][]{%
  \mathopen{#1|\mkern-3mu#1|\mkern-3mu#1|}
  #2
  \mathclose{#1|\mkern-3mu#1|\mkern-3mu#1|}
}
\makeatother

\DeclareMathOperator{\tr}{Tr}
\newcommand{\dist}{\mathrm{dist}}
\newcommand{\diam}{\mathrm{diam}}
\renewcommand{\d}{\mathrm{d}}

\newcommand{\Z}{\mathbb{Z}}
\newcommand{\R}{\mathbb{R}}

\newcommand{\p}{\varphi}

\newcommand{\rhoMT}{{\varrho}}
\newcommand{\TN}{T_N}
\newcommand{\HN}{\mathrm{H}_N}
\newcommand{\HL}{\mathrm{H}_\Lambda}
\newcommand{\wpzc}{\mathpzc{w}}
\newcommand{\Bcal}{\mathcal{B}}
\newcommand{\Ccal}{\mathcal{C}}
\newcommand{\Gcal}{\mathcal{G}}
\newcommand{\Hcal}{\mathcal{H}}
\newcommand{\Pcal}{\mathcal{P}}
\newcommand{\Ical}{\mathcal{I}}

\newcommand{\Kcal}{\mathcal{K}}
\newcommand{\Mcal}{\mathcal{M}}
\newcommand{\Ocal}{\mathcal{O}}
\newcommand{\Rcal}{\mathcal{R}}
\newcommand{\Qcal}{\mathcal{Q}}
\newcommand{\Scal}{\mathcal{S}}
\newcommand{\Ucal}{\mathcal{U}}
\newcommand{\Vcal}{\mathcal{V}}
\newcommand{\Wcal}{\mathcal{W}}
\newcommand{\Xcal}{\mathcal{X}}
\newcommand{\Zcal}{\mathcal{Z}}

\newcommand{\Ascr}{\mathscr{A}}
\newcommand{\Cscr}{\mathscr{C}}
\newcommand{\Escr}{\mathscr{E}}

\newcommand{\Lscr}{\mathscr{L}}
\newcommand{\Mscr}{\mathscr{M}}
\newcommand{\Nscr}{\mathscr{N}}
\newcommand{\Qscr}{\mathscr{Q}}
\newcommand{\Uscr}{\mathscr{U}}
\newcommand{\Vscr}{\mathscr{V}}
\newcommand{\Wscr}{\mathscr{W}}
\newcommand{\Zscr}{\mathscr{Z}}

\newcommand{\QU}{\mathscr{Q}_\mathscr{U}}

\newcommand{\Pclk}{{\Pcal^{\rm cl}_k}}
\newcommand{\Pck}{{\Pcal^{\rm c}_k}}
\newcommand{\Pc}{{\Pcal^{\rm c}}}
\newcommand{\Pckp}{{\Pcal^{\rm c}_{k+1}}}
\newcommand{\Ccl}{{\Ccal^{\rm cl}}}

\newcommand{\weightzeta}{{\overline{\zeta}}}

\newcommand{\myW}{\mathcal{W}}
\newcommand{\myZ}{\mathcal{Z}}
\newcommand{\myT}{{\boldsymbol{T}}}
\newcommand{\Malt}{{\boldsymbol{M}}}
\newcommand{\myS}{{\boldsymbol{S}}}
\newcommand{\myR}{\boldsymbol{R}}

\newcommand{\AP}{A_{\mathcal{P}}}
\newcommand{\AB}{A_{\mathcal{B}}}
\newcommand{\vertiii}[1]{{\left\vert\kern-0.25ex\left\vert\kern-0.25ex\left\vert#1 \right\vert\kern-0.25ex\right\vert\kern-0.25ex\right\vert}}
\newcommand{\vertiiiReg}[1]{{\vert\kern-0.25ex\vert\kern-0.25ex\vert #1   \vert\kern-0.25ex\vert\kern-0.25ex\vert}}

 \usepackage{dsfont}
\newcommand{\1}{{\mathds1}}
\newcommand{\nullL}{{\mathsf N}}
\newcommand{\Ran}{{R_0}}

\newcommand{\mf}{\mathfrak}
\DeclareMathOperator{\loc}{loc}
\DeclareMathOperator{\sgn}{sgn}
\DeclareMathOperator{\Span}{span}
\DeclareMathOperator{\sym}{sym}
\DeclareMathOperator{\tay}{Tay}

\DeclareMathAlphabet{\mathpzc}{OT1}{pzc}{m}{it}
\newcommand{\mpzc}{\mathpzc{m}}
\newcommand{\hpzc}{\mathpzc{h}}

\newcommand{\C}{\mathbb{C}}
\newcommand{\NN}{\mathbb{N}}
\newcommand{\myB}{\mathcal{B}}

\newcommand{\pn}{{r_0}}  
\newcommand{\pphi}{{p_\Phi}} 
\newcommand{\BX}{\mathcal X}   
\newcommand{\BY}{\mathcal Y}   
\newcommand{\BXp}{  {\mathcal X}'}   
\newcommand{\BXt}{\tilde {\mathcal X}} 
\newcommand{\headingh}{h}
\newcommand{\headingAkkp}{\boldsymbol{A}_{k:k+1}^X}
\newcommand{\headingAk}{\boldsymbol{A}_{k}^X}
\newcommand{\headingRk}{\boldsymbol{R}_{k+1}^{(\boldsymbol{q})}}
\newcommand{\headingC}{\boldsymbol{C}^{(\boldsymbol{q})}}
\newcommand{\headingB}{\boldsymbol{B}^{(\boldsymbol{q})}}
\newcommand{\headingA}{\boldsymbol{A}^{(\boldsymbol{q})}}
\newcommand{\headingd}{d}
\newcommand{\headingr}{r}

\DeclareMathOperator{\Id}{Id}

\definecolor{forestgreen}{cmyk}{0.91,0,0.88,0.12}

\newcommand{\gray}[1]{{\color{gray}#1}}

\usepackage{overpic}

\begin{document}

\frontmatter

\title{Cauchy-Born Rule  from Microscopic Models \\
with Non-convex Potentials}

\author{Stefan Adams \dag}
\address{Mathematics Institute, University of Warwick, Coventry CV4 7AL, United Kingdom}
\email{S.Adams@warwick.ac.uk}

\author{Simon Buchholz}
\address{Institute for Applied Mathematics \& Hausdorff Center for Mathematics, Universit\"at Bonn, Endenicher Allee 60, D-53115 Bonn, Germany}
\email{buchholz@iam.uni-bonn.de}

\author{Roman Koteck\'y}
\address{Center for Theoretical Study, Charles University, Jilsk\'a 1, Prague, Czech Republic}
\email{kotecky@cts.cuni.cz}
\address{Mathematics Institute, University of Warwick, Coventry CV4 7AL, United Kingdom}

\author{Stefan M\"uller}
\address{Institute for Applied Mathematics \& Hausdorff Center for Mathematics, Universit\"at Bonn, Endenicher Allee 60, D-53115 Bonn, Germany}
\email{stefan.mueller@hcm.uni-bonn.de}

\date{May 30, 2024}

\subjclass[2010]{Primary 82B28, 60G60; Secondary 82B41, 60K35, 74B20}

\keywords{Renormalisation group; random field of gradients; nonlinear elasticity; surface tension; multi-scale analysis}  

 
 \begin{abstract}
 
We study gradient field models on an integer lattice with non-convex interactions. 
These models emerge in distinct branches of physics and mathematics under various names. 
In particular, as zero-mass lattice (Euclidean) quantum field theory, models of random interfaces, 
and as mass-spring models of nonlinear elasticity.
Our attention is mostly  devoted to the latter with random vector valued fields as displacements for atoms of crystal structures,
where our aim is to prove the strict convexity of the free energy as a function of affine deformations 
for low enough temperatures and small enough deformations. 
This claim can be interpreted as a form of verification of the Cauchy-Born rule at small non-vanishing temperatures 
for a class of these models.   
We also show that  the scaling limit of the Laplace transform of the corresponding Gibbs measure (under a proper rescaling)
corresponds to the Gaussian gradient field with a particular covariance. 

The  proofs are based on a multi-scale (renormalisation group analysis) technique
needed in view of strong correlations of the studied gradient fields. 
To cover a sufficiently wide class of models, we extend these techniques from the standard case 
with rotationally symmetric nearest neighbour interaction to a more general situation with finite range interactions without any symmetry.
Our presentation is entirely self-contained covering the details of the needed renormalisation group methods.
\end{abstract} 

\maketitle
 
\tableofcontents  

\mainmatter

\chapter[Introduction] {Introduction} 
 \label{se:introduction}

This paper has two related goals.

Firstly, we seek to establish  uniform convexity properties for a class of lattice gradient models with non-convex microscopic interactions.

Secondly, we extend the rigorous renormalisation group techniques developed by Brydges and collaborators 
to vector models with finite range interactions without discrete rotational symmetry. 
The lack of  such symmetry,  leads to a need of a significant enlargement of  the set of relevant terms in the renormalisation treatment.

Regarding the first goal, we consider  \emph{gradient random fields} $\{\varphi(x)\}_{x\in\Lscr}$
indexed by a lattice $\Lscr$ with values in $\R^m$, $\varphi(x)\in\R^m$. 
The term \emph{gradient} is referring to the assumption that the distribution depends only on  gradients 
$\nabla_{e}\varphi(x)=\varphi(x+e)-\varphi(x)$.
These type of fields are used as effective models of crystal deformation or phase separation.
In the former case, where $m=3$ and $\Lscr\subset\R^3$, the value $\varphi(x)$ plays the role of a displacement
of an atom labelled by a site $x$ of a crystal under deformation. 
In the latter case, with $m=1$ and $\Lscr=\Z^2$,
the model is a discretization of a phase separation in $ \mathbb{R}^{3} $ with $\varphi(x)\in\R$ 
corresponding to the position of the (microscopic) phase separation surface. 
The model is a reasonably effective approximate description in spite of the fact that it ignores overhangs 
of the separation surface as well as any correlations inside and between the coexisting phases.

The distribution of random fields $\varphi(x)$ is given in terms of a Gibbs measure with Hamiltonian 
defined by means of  a finite range gradient potential.
To make this precise, let $A \subset \Z^d$ be a finite set and let  $U : (\R^m)^{A} \to \R$ be an interaction potential defined
on maps $\eta: A \to \R^m$. The fact that the interaction $U$ leads to a gradient random field amounts to an assumption
 that $U$ is invariant with respect to translations in $\R^m$, i.e. 
$U(t_a \eta) = U(\eta)$ where $t_a \eta(x) = \eta(x) + a$ for $x \in A$ and $a \in \R^m$.  As a result, for connected sets
$A$, $U(\psi)$ depends only on the discrete gradient of $\eta$.  For $x \in \Z^d$, we use $\tau_x(A)$ to denote the translation of
$A$ by $x$, $\tau_x(A) = x + A = \{ x+y : y \in A\}$ and for a map $\p: \Z^d \to \R^m$ we use $\tau_x \p(y)$ 
to denote the translated map $\tau_x \p(y) :=\p(y-x)$ and $\p_A$ to denote the restriction of $\p$ to the set $A$.
Then, for a finite set $\Lambda$ we introduce the Hamiltonian
$$ \HL(\p) = \sum_{x \in \Lambda:  \tau_x (A) \cap \Lambda \ne \emptyset}  U( (\tau_{-x}\p)_A).$$
Note that $(\tau_{-x}\p)_A$ depends only on the values of $\p$ on $\tau_x(A)$. 

The finite volume Gibbs distribution with a boundary condition $\psi_{\Z^d\setminus \Lambda}$ is given by the measure
$$
\gamma_{\beta,\Lambda}^{\psi}({\rm d}\varphi)=\frac{1}{Z_{\beta,\Lambda}(\psi)}\exp\big(-\beta \HL(\varphi)\big)\prod_{x\in\Lambda}{\rm d}\varphi(x)\prod_{x\in\Z^d\setminus \Lambda}\delta_{\psi(x)}({\rm d}\varphi(x)),
$$
where the partition function---the normalisation constant $ Z_{\beta,\Lambda}(\psi)$---is the integral of the density.
One is particularly interested in the boundary conditions 
$$
\psi_{F}(x)=Fx
$$
corresponding to a linear deformation $F:\R^d\to\R^m$. 

An object of basic relevance in this context is the  \textit{free energy\/} defined by the limit
$$
W_{\!\beta}(F)=-\lim_{\Lambda\uparrow\mathbb{Z}^d}\frac{1}{\beta|\Lambda|}\log Z_{\beta,\Lambda}(\psi_{F})
$$

In the scalar case, $m=1$,  the  map $F$ is actually a linear functional determined by a vector---a macroscopic tilt  $u\in\R^d$
defining the boundary condition  $\psi_{u}(x)=(u,x)$. The  free energy $W_{\!\beta}(u)$  then   corresponds 
to the interface free energy/surface tension $\sigma_\beta(u)$ with a given tilt---the price to pay for tilting a macroscopically flat interface.

Our main results relate to the strict convexity of the interface free energy $\sigma_\beta(u)$ 
as a function of the tilt $u$ in the scalar case 
and of the free energy $W_{\!\beta}(F)$ as a function of affine deformation $F$ in the vector case, respectively.
Actually, instead of random fields with affine Dirichlet boundary conditions $F$ we use the Funaki-Spohn trick 
that amounts to considering the fixed affine field $F$ with added periodic random field with vanishing mean,
see Section 2 for a detailed description.

Existing results concerning strict convexity  pertain to  the scalar case.
For a uniformly strictly convex symmetric nearest neighbour potential $U\in C^2(\R)$, 
Funaki and Spohn  \cite{FS97} established the \emph{convexity} of $ \sigma_{\beta} $ 
and use it in the derivation of the hydrodynamical limit of the Landau-Ginsburg model. This result was strengthened
to the claim of the \emph{strict convexity} of the surface tension  in \cite{DGI00} and \cite{GOS01}. 
 The strongest version proving that  $ \sigma_{\beta} \in C^{2,\gamma}$ under a mild additional regularity assumption 
 on the interaction $U$ was obtained very recently by Armstrong and Wu in \cite{AW19}.
 All those results use the random walk representation of Helffer and Sj\"ostrand
in a crucial way. This representation  requires the condition that the potential $U$ is uniformly strictly convex.

The case of a non-convex potential was treated by Cotar, Deuschel, and M\"uller \cite{CDM09}, 
who showed the strict convexity of the surface tension in the small $ \beta $ (high temperature) regime for potentials  of the form
$$
U(t)=U_0(t)+g(t), 
$$ 
where $U_0 $ is uniformly strictly convex and $ g \in C^2(\mathbb{R})$ has a negative bounded second derivative 
such that $ \sqrt{\beta}\norm{g^{\prime\prime}}_{L^1(\mathbb{R})} $ is sufficiently small. 
The unpublished precursor \cite{AKM16} of the present paper was still dealing with a nearest neighbour scalar case
obtaining the strict convexity of the surface tension for a class of  non-convex---and not necessarily symmetric---potentials with
$\beta$ large enough  (low temperatures) and sufficiently small tilts.

In the present paper we include  $\R^m$-valued random fields and consider microscopic  multibody interactions 
 which can have a general finite range
and go beyond nearest-neighbour or next-nearest neighbour internactions (both generalisations are 
crucial for applications to nonlinear elasticity).
All our main general results are collected in Chapter~\ref{sec:setting}.
An additional difficulty in the case of vector random field with application to the mass-spring models of 
discrete nonlinear elasticity stems from the fact that frame invariance implies that the interaction is necessarily
invariant under rotations.
This leads to a degeneracy of the quadratic form given by the second derivative of
the interaction $U$ at its minimum
that prevents the convexity of the free energy.
As explained in Chapter~\ref{sec:elasticity} devoted to a detailed discussion of this case,
 we will overcome this difficulty by adding a suitable discrete null Lagrangian to the Hamiltonian.
As a result, we arrive at proving  strict convexity of $W_{\!\beta}(F)$ when restricted to symmetric matrices $F$.

Such convexity results for the free energy are closely related to the so-called Cauchy-Born rule
in crystal elasticity.  This rule originates in work of Cauchy and was extended by 
Born and Huang \cite{BH54}, see Ericksen's paper  \cite{Eri08} for a  detailed discussion 
and a review or the literature. The Cauchy-Born rule for atomistic models at zero temperature states that the minimisers of the discrete
energy $\HL$ on finite domains with affine boundary conditions are given by the corresponding affine lattice deformations.
It  thus provides a crucial bridge between atomistic and continuum theories, since one can 
define an energy density function in the continuum theory by computing the energy per unit volume
for an affine transformation of the lattice. In the language of the calculus of variations, 
the fact that affine boundary conditions for the discrete energy lead to affine minimisers, corresponds to discrete quasiconvexity of the local interaction energy $U$. Strict quasiconvexity requires that the affine deformation is
the only minimiser for affine boundary condition and represents a stability property of the affine minimiser.
Often, the Cauchy-Born rule is only required for affine maps which are sufficiently close to a rigid motion. 

There is no obvious counterpart for the Cauchy-Born rule at finite temperature,  since the basic object
is not  a single deformation anymore, but a random field with the  probability measure  $\gamma_{\beta,\Lambda}^{\psi_F}$. 
Here it is useful to employ the Funaki-Spohn trick introduced in \cite{FS97}---replacing affine boundary condition by a  
periodic random field with the additional  fixed affine deformation 
$F$ incorporated directly in the Hamiltonian, cf.  the formulas \eqref{eq:defHF}, \eqref{eq:defWN}, and \eqref{eq:finite_volume_de} defining the corresponding Hamiltonian $\HN^F(\p) $, free energy
$W_{\!\beta,N}(F)$, and Gibbs measure $ \gamma_{\beta,N}^F(\d \p)$.

A weak version of the  Cauchy-Born at finite temperature is   the statement
that the free energy  density function $ F \mapsto W_{\beta, \Lambda}(F) $ is quasiconvex,
i.e., minimisers of the continuum  energy functional  $ \varphi \mapsto \int_\Omega W_{\beta, \Lambda}(\nabla \varphi(x)) \, dx$ 
subject to affine boundary conditions are affine. In Theorem~\ref{T:deW}   
we show that this is the case for sufficiently large $\beta$ (small temperature),
  discrete torus $\Z^d/ L^N \Z^d$ large enough, affine deformation $F$  sufficiently close to a rigid motion, 
  and the microscopic interaction $U$ satisfying certain regularity conditions.

In fact, we show that $W_{\beta, \Lambda}$ is strictly quasiconvex in this regime. Strict quasiconvexity  is
a stability condition which  penalizes fluctations at a mesoscopic scale. 
We note in passing that, under very general conditions on $U$,  Koteck\'y and Luckhaus \cite{KL14} have
shown  (for slightly different boundary conditions) that the thermodynamic limit 
$W_\beta:= \lim_{\Lambda \uparrow \Z^d} W_{\beta, \Lambda}$ exists
and is quasiconvex for all affine maps, not just those close to a rigid motion. 
 Their result, however, does not give strict quasiconvexity and hence does not penalize fluctuations for approximate minimisers. 
 
 A stronger version of the Cauchy-Born rule at finite temperature consists in requiring that the 
 probability measure $\gamma_{\beta,\Lambda}^{F}$ is concentrated near $0$ and approximately Gaussian.  
 Proving such a result is currently beyond the scope of our approach. We can show, however, that the measures 
 $\gamma_{\beta,\Lambda}^{F}$ converge to a Gaussian measure centered at zero  if we perform a suitable  rescaling of space,   
 see Theorem~\ref{T:descaling}.
 
 There is a sizable literature on the Cauchy Born rule at finite temperature, see e.g.,
 \cite{SSR99,DTMP05,XY06,XY08,YMLL15,LXY21,DWYY22}. In these works the point
 of view is slightly different from ours. The Cauchy-Born rule is seen as a method to
 approximate the discrete problem at finite temperature by a continuous problem. To  perform
 explicit calculations,  the harmonic approximation is used, i.e., instead of a measure like
 $\gamma^F_{\beta,N}$ one looks at a Gaussian measure which involves the second derivatives of the
 internal energy at the state considered.  
 The purpose of our work is exactly a rigorous analysis of the non-Gaussian case.

Our  proofs are based on multi-scale (renormalisation group analysis) techniques. 
However, to cover a sufficiently wide class of models, we need to extend these techniques from the standard case 
with rotationally symmetric nearest neighbour interaction to a more general situation with finite range interactions without any symmetry.

The second goal of the present paper thus is to show in detail how  the rigorous renormalisation approach of 
Brydges and collaborators (see \cite{BY90} for early work, \cite{Bry09} for a survey and 
\cite{BS15I, BS15II, BS15V, BBS19} for recent developments) 
can be extended to accommodate our class of models without a discrete rotational symmetry of the interaction.

In accordance with the general renormalisation group strategy,
the resulting partition function $Z_{\beta,\Lambda}(\psi_F)$ is obtained by a sequence of ``partial integrations'' (labelled by an index $k$).
The ``partial integrations'' are based on a finite range decomposition of the covariance of a linearization of the Hamiltonian.
The result of each of them is expressed in terms of two functions: the \emph{``irrelevant'' polymers} $K_k$ 
that are decreasing with each subsequent integration, and the \emph{``relevant'' ideal Hamiltonians} $H_k$---homogeneous quadratic functions of
gradients $\nabla\varphi$ parametrized by a fixed finite number of parameters. 
To fine-tune the procedure so that the final integration  yields a result with a straightforward bound 
we need to assure the smoothness of the procedure with respect  to the parameters of a suitably chosen ``seed Hamiltonian''. 
It turns out, however,  that the derivatives with respect to those parameters  lead to a loss of regularity of functions $K_k$ and $H_k$ 
considered as elements in a scale of Banach spaces if one uses standard extension of the finite range decomposition 
by Brydges and collaborators to the case without rotational symmetry.
We show that this problem can be overcome by using a refined finite range decomposition\cite{Buc18}, 
see Chapter~\ref{sec:FRD} for further discussion.

A brief overview of the general strategy of proof and the main new ideas compared to earlier work is given
in Section~\ref{se:main_new_ideas}. 
A summary of the concrete technical implementation of our approach is given in Chapter~\ref{sec:explanation}.

\chapter[Setting and Main Results] {Setting and Main Results}\label{sec:setting}

In the present chapter we formulate the main results of this work. 
They come in two forms.
First, we describe 'concrete' results with a general  class of Hamiltonians meeting particular assumptions allowing for certain 
non-convex interactions.  
To include finite range potentials (as opposed to just nearest neighbour interactions of standard gradient models), we introduce the equivalent notion
of generalized gradient models.
In Section~\ref{sec:ggm}  we describe the results for this general class of models---Theorem~\ref{thm:strictconvexity}  (local strict convexity of the free energy as a function of the imposed affine boundary conditions)
 and  Theorem~\ref{th:scalinglimit_concrete} (Gaussian scaling limit).
 Section~\ref{sec:discrete_elasticity_main} is then devoted to the corresponding results  for discrete elasticity, 
see Theorems~\ref{T:deW} and~\ref{T:descaling}.

Secondly, in Section~\ref{sec:main_abstract} we discuss the corresponding  abstract results, Theorems~\ref{th:pertcomp_E}  and~\ref{th:scalinglimit}, 
 formulated in terms of a rather general perturbation of Gaussian  measures. 
By splitting the concrete Hamiltonian arising for generalized gradient models into a quadratic part and a perturbation, one  easily sees that
the abstract results presented in Section~\ref{sec:main_abstract} imply  the concrete results for generalized gradient models, see Sections~\ref{sec:abstract_to_ggm} and~\ref{sec:emb}.

The proof of our basic results, Theorems~\ref{th:pertcomp_E}  and~\ref{th:scalinglimit}, will cover the main part of this work. 
The proof is based on a rigorous renormalisation approach  which has been systematically developed by Brydges, Slade, and Bauerschmidt 
over the last decades, see   \cite{Bry09, BS10,BBS19} for surveys and additional references to earlier and related work. 
 We review this approach,
and its modifications and extensions needed for our setting, in detail in Chapter~\ref{sec:explanation} below.  In Section~\ref{se:main_new_ideas} of the current chapter we give a very quick preview of
our adaptation of  the renormalisation group approach.  
We believe that it is of independent interest and might be useful
for the use of renormalisation group methods in other contexts, too.

Returning to a more precise description of the 'concrete'  results for generalized gradient models, the main assertion of 
Theorem~\ref{thm:strictconvexity} can be 
expressed informally as follows. 
If the (microscopic) local interaction energy
is sufficiently regular  (including subexponential growth bounds for finitely many  derivatives) and 
has a 'non degenerate'  strict minimum at zero, then the finite volume free energies (as functions of the imposed affine boundary condition $F$)
 are uniformly convex on a fixed small ball $B_{\delta}(0)$, provided
the inverse temperature $\beta$ satisfies $\beta \ge \beta_0$ for $\beta_0$ sufficiently large.
 Here the precise definition of 'non-degenerate' is given by the inequalities     \eqref{eq:V1bis_new} and  \eqref{eq:V3_new}.    
They involve a condition at $0$ and a global condition with uniformity at  $\infty$. Of course  $0$ may be replaced by any linear map and the conditions and the conclusion
are invariant under subtracting an affine function from the local interaction energy. 

Actually, we show more. The finite volume free energies admit a low temperature
 expansion of the form $W_{\!\beta,N}(F) = \Uscr(\bar F) + c_{\beta,N} + \frac1\beta \Wcal_{\beta,N}$,
 where $\Uscr(\bar F)$ is the interaction energy of the generalized gradient field $\bar F$ corresponding to the linear deformation $F$ (expressed in terms of  
 generalized interaction potential $\Uscr$ corresponding to the finite range potential $U$ of the original model), $c_{\beta,N}$ are constants, and 
 $F\mapsto \Wcal_{\beta, N}(F)$ is a function whose  the $C^r$ norm  is bounded on $B_{\delta}(0)$, uniformly in $N$ and in $\beta \ge \beta_0$.
 
Theorem~\ref{th:scalinglimit_concrete}
 asserts that under the same assumptions  on the microscopic energy, the  finite volume Gibbs measures approach, after suitable
rescaling,  a Gaussian measure in the infinite volume limit. More precisely, we show that, at least for a subsequence, the Laplace transforms
of the Gibbs measures converge. The relevant scaling corresponds to a central limit theorem scaling of the discrete gradient fields, see
Remark~\ref{re:correlation_structure_concrete}.

The setting of discrete elasticity poses an additional difficulty. As dictated by the condition of \emph{frame invariance}---an indispensable feature of microscopic models of nonlinear elasticity---the local interaction energy
is invariant under the action of the group of  $\mathrm{SO}(d) = \{ A \in \R^{d \times d} : A^T A = \1, \, \det A=1\}$ of proper rotations. Thus the minimum of the local interaction energy,
which is typically realized at the identity deformation, can never be non-degenerate. Nonetheless we identify  natural conditions on the local
interaction energy which allow us to conclude that the finite volume free energies is as convex as possible. More precisely, the free energy
is a function on $d \times d$  matrices and  we show that (for $\beta \ge \beta_0$)
the restriction 
to a small ball  around
the identity  in the space of symmetric matrices is uniformly strictly convex. It is easy to see that the $\mathrm{SO}(d)$ invariance of the 
local interaction energy implies that the finite volume free energies $W_{\!\beta, N}$ satisfy $W_{\!\beta, N}(QF) = W_{\!\beta,N}(F)$ for every
$Q \in \mathrm{SO}(d)$ and every $F \in \R^{d \times d}$.  Since every $d \times d$ matrix $F$ can be written as $F = Q S$ with
$Q \in SO(d)$ and $S$ symmetric, the free energy is determined by its values on symmetric matrices and uniform convexity on symmetric
matrices is the natural convexity condition. We also show that the free energy is uniformly
quasiconvex at matrices in $B_{\delta/2}(\1)$, see  \eqref{eq:uniformly_quasiconvex}.

The key idea to overcome the degeneracy induced by $\mathrm{SO}(d)$ invariance is described in detail in the next chapter, 
Chapter~\ref{sec:elasticity}.
The point is that we manage to employ a standard technique from nonlinear elasticity:
find a function---a  so-called discrete null Lagragian---such that  
i)  adding it to the local interaction energy, the Gibbs measure for fixed affine boundary conditions is unchanged, and 
ii)  the sum of the two energies has a nondegenerate minimum (in the sense of 
inequalities     \eqref{eq:V1bis_new} and  \eqref{eq:V3_new}) at the identity map.

\section{General setup}\label{sec:setup}
Fix an odd integer $L\geq 3$ and a dimension $d\ge 2$. Let $T_N=(\mathbb{Z}/(L^N\mathbb{Z}))^d$ be 
the $d$-dimensional \emph{discrete torus} of side length $L^N$ where $N$ is a positive integer.
We equip $T_N$ with the quotient distances  $\abs{\cdot}$ and  $\abs{\cdot}_\infty$ induced by the Euclidean 
and   maximum norm respectively.
Define the space of $m$-component fields on $T_N$ as
\begin{align}
\label{eq:defofVN}
\Vcal_N=\{\varphi: T_N\rightarrow \R^m\}=(\R^m)^{T_N}. 
\end{align}
Since the energies we consider are shift invariant we are only interested in gradient fields. 
However,  the condition of being a gradient is not entirely straightforward in dimension $d\geq 2$; 
thus  we rather work with  usual  fields modulo a  constant or, equivalently, with fields with the vanishing average  
\begin{align}
\label{eqdefofadmissible}
\varphi\in\Xcal_N=\Bigl\{\varphi\in \Vcal_N: \sum_{x\in T_N}\varphi(x)=0\Bigr\}
\end{align}
that are in one-to-one correspondence with \emph{gradient fields}.
Let the dot denote the standard scalar product on $\R^m$ which is later extended to $\mathbb{C}^m$. 
For $\varphi,\psi\in \Xcal_N$ the expression
\begin{align}
\label{E:scprodT}
(\varphi, \psi) =\sum_{x\in T_N} \varphi(x)\cdot \psi(x) 
\end{align} 
defines a scalar product on $\Xcal_N$ and this turns $\Xcal_N$ into a Hilbert space.
We use  $\lambda_N $ for  the $m(L^{Nd}-1)$-dimensional Hausdorff measure on $ \Xcal_N $, 
equip the space $ \Xcal_N $ with the $\sigma$-algebra 
$ \boldsymbol{\Bcal}_{\Xcal_N}$ 
induced by the Borel $\sigma$-algebra with respect to the product topology, 
and use  $ \Mcal_1(\Xcal_N)=\Mcal_1(\Xcal_N,  \boldsymbol{\Bcal}_{\Xcal_N}) $ 
to denote the set of probability measures on $ \Xcal_N $, 
referring to  elements in $ \Mcal_1(\Xcal_N) $ as  \emph{random gradient fields}.

The \emph{discrete forward} and \emph{backward derivatives} are defined by
\begin{align}
\begin{split}
(\nabla_i\varphi)_s(x)   & =\varphi_s(x+e_i)-\varphi_s(x)\qquad s\in\{1,\ldots, m\}, \quad i\in\{1,\ldots, d\},\\
(\nabla_i^*\varphi)_s(x) & =\varphi_s(x-e_i)-\varphi_s(x) \qquad s\in\{1,\ldots, m\}, \quad i\in\{1,\ldots, d\}.
\end{split}
\end{align}
Here $e_i$ are the standard coordinate unit vectors in $\mathbb{Z}^d$.
Forward and backward derivatives are adjoints of each other. 
We use $\nabla \varphi (x)$ and $\nabla^* \varphi (x)$ for the corresponding $m\times d$ matrices.

In this article we study a class of random gradient fields defined in terms of Gibbs measures
introduced by means of Hamiltonians $\HN^F: \Xcal_N\to \R   $ that are in  turn determined by a finite range potential 
$U:( \R^m)^{A}\to\R$ and a linear map $F : \R^d \to \R^m$ which serves as a boundary condition.

Here,  $A\subset \Z^d$  is a finite set and  we
use $\Ran$ to denote  the range of the potential $U$,  $\Ran=\diam_\infty A$. 
We assume  that $U$ is invariant with respect to translations in $\R^m$, i.e., 
$U(\psi)=U(t_a \psi)$
for any  $\psi\in  ( \R^m )^{A}$ with $(t_a \psi)_s(x)= \psi_s(x) +a_s$, $a\in \R^m$.
For connected sets $A$ this is equivalent to saying that $U(\psi)$ depends on $\psi$ only through 
$\nabla \psi$. 

Given $U$, we want to define a Hamiltonian which describes the action of $U$ on fields $\eta: \Z^d \to \R^m$ by
adding the action of $U$ on translations of $\eta$, restricted to $A$. Actually, we are interested in fields
which satisfy an affine boundary condition at $\infty$. To make this precise, we use the Funaki-Spohn trick \cite{FS97}
to consider fields of the form $F + \p: \Z^d \to \R^m$ where $F: \R^d \to \R^m$ is a linear map and $\p$ with values in $ \Xcal_N$ is a periodic random field.

We will frequently identify a map $\p: T_N \to \R^m$ with an $L^N \Z^d$-periodic map from $Z^N$ 
to $\R^m$, denoted by the same symbol.
Also, we often identify $F$ with a $m \times d$ matrix and we write $F$ also for the restriction of the map to $\Z^d$, $A$
or the set $x + A = \{ x + y : y \in A\}$. For any map $\eta: \Z^d \to \R^m$ we write $\eta_B$ for the restriction 
of $\eta$ to the set $B \subset \Z^d$. For $x \in \Z^d$ we define the translated map  $\tau_x \eta$ 
by  $\tau_x \eta(y) = \eta(y-x)$.

Note that for a linear map $F$ the translated maps $\tau_x F$ are affine maps which differ from $F$ only a by
constant map. Thus for an $L^N \Z^d$-periodic map $\p$,  the map
$$ x \mapsto U( (\tau_{-x} (F + \p))_A) $$
is also $L^N \Z^d$-periodic, since $U$ is invariant under shifts by a constant map.
Hence we can view $x \mapsto U( (\tau_{-x} (F + \p))_A)$ as a map from $T_N$ to $\R$ and we can define
the Hamiltonian $\HN^F: \Xcal_N \to \R$ by
\begin{align*}
\HN^F(\p)=\sum_{x\in T_N}  U( (\tau_{-x} (F + \p))_A).
\end{align*}
Using again that $\tau_{-x}F$ and $F$ differ only by a constant map,  we can rewrite the expression of $\HN^F$ as
\begin{align}
\label{eq:defHF}
\HN^F(\p)=\sum_{x\in T_N}  U( F +(\tau_{-x} \p)_A).
\end{align}
Note that $(\tau_{-x} \p)_A$ depends only on the values of $\p$ on the translated set $x + A$.

For a concrete example one might consider the following simple model of two-dimensional discrete
elasticity with nearest-neighbour and next-nearest-neighbour interactions: take $m=d=2$, $A = \{0,1\}^2$
and 
$$ U(\eta) = \tfrac12 \sum_{x,y \in \{0,1\}^2 : |x-y| = 1} U_1( |\eta(x) - \eta(y)|) +  \sum_{x,y \in \{0,1\}^2 : |x-y| = \sqrt 2} 
U_2( |\eta(x) - \eta(y)|)$$
where $U_1$ and  $U_2$ are functions defined on   $[0, \infty)$.
In this case one gets
\begin{align*}
 H^F_N(\p) = & \,  \sum_{x,y \in T_N : |x-y| = 1} U_1(| F(x-y) + \p(x) - \p(y))|  \\
  + & \,  \sum_{x,y \in T_N : |x-y| = \sqrt 2} 
U_2(| F(x-y) + \p(x) - \p(y))|.
\end{align*}

For the choice $U_i(t)  =  \frac12 K_i (t- a_i)^2$ this energy is considered in in \cite[eqn.\  (3.3)]{FT02}.

One can also consider additional multibody contributions to the energy,
for example a term $U_3(v)$ where $v$ is oriented area of the quadrilateral with corners 
$$(g_1, g_2, g_3, g_4): =     \left( \eta((0,0), \eta((0,1)),   \eta((1,1)) \eta((0,1))\right)$$ given by 
$$ v = \tfrac12 \det(g_2-g_1, g_4-g_1) + \tfrac12 \det(g_4-g_3, g_2 - g_3).$$
see \cite[p.\ 458]{FT02}.  The zero-temperature problem of minimizing the Hamiltonian $\HN$ subject to affine
Dirichlet boundary conditions and subject the constraint  that the map $\p$ is discretely orientation preserving (i.e., the oriented
area of every quadrilateral defined by $\p( x + \{0,1\}^2)$ is nonnegative) is discussed in \cite{FT02}.
For a discussion which microscopic interactions $U$ are allowed in our approach in the context of
discrete elasticity, see Sections~\ref{sec:discrete_elasticity_main} and~\ref{se:elasticity_examples}.

The \emph{finite volume gradient Gibbs measure $\gamma_{\beta,N}^F$  under a deformation $F$} is then defined as
\begin{align}  
\label{eq:finite_volume_de}
\gamma_{\beta,N}^F(\d \p)=\frac{1}{Z_{\beta,N}(F,0)}\exp\left(-\beta \HN^F(\p)\right)\lambda_N(\d\p),
\end{align}
where  $Z_{\beta,N}(F,0)$ is the normalizing \emph{partition function}.
Here the normalization $Z_{\beta,N}(F,0)$ is a special case, with $f=0$, of a useful 
\emph{generalized  partition function $Z_{\beta,N}(F,f)$ with a source term $f\in\Xcal_N$} defined by 
\begin{align}\label{eq:defofZ}
	Z_{\beta,N}(F,f)=\int_{\Xcal_N}\exp\bigl(-\beta \HN^F(\p)+(f,\p)\bigr)\lambda_N(\d\p).
\end{align}
In particular,  the generalized partition function characterizes  the Gibbs measure $\gamma_{\beta,N}^F$ and will be used to analyse its scaling limit. 

While  a natural  long-term objective is the specification of the gradient Gibbs measures 
with a given deformation as it was done in \cite{FS97} for the scalar case with convex interactions, 
in the present paper we will restrict our attention to the analysis of the partition function $Z_{\beta,N}(F,0)$ 
and the scaling limit of the partition function $Z_{\beta,N}(F,f)$. 
In particular, we investigate local convexity properties of the functions
\begin{align}  
\label{eq:defWN}
W_{\!\beta,N}(F)=-  \frac{\ln\, Z_{\beta,N}(F,0)}{\beta L^{Nd}}
\end{align}
and of the free energy 
\begin{align}
\label{eq:defW}
W_{\!\beta}(F)=\lim_{N \to \infty} W_{\!\beta, N}(F) = -\lim_{N\rightarrow  \infty}\frac{\ln\, Z_{\beta,N}(F,0)}{\beta L^{Nd}}.
\end{align}
For the scaling limit of the gradient Gibbs measure  we will analyse the  Laplace transform
\begin{align}
\label{eq:defLaplace}
\lim_{N\rightarrow \infty}\mathbb{E}_{\gamma_{\beta,N}^F}e^{(f_N,\p)}= 
\lim_{N\rightarrow \infty}\frac{Z_{\beta,N}(F,f_N)}{Z_{\beta,N}(F,0)}, 
\end{align}
where $f_N\in\Xcal_N$ is the rescaled discretization $f_N(x)=L^{-N(\frac{d+2}{2})}f(L^{-N}x)$  
of a smooth function $f:\mathbb{T}^d \rightarrow \R^m$ with average zero. Here $\mathbb{T}^d$ is a shorthand for the torus $(\R/\mathbb{Z})^d$.
The function $f_N$ is a slowly varying test function that allows us to examine the long distance behaviour
of the random field $\p$.

Let us remark that when  $m=d$, this is the setting for microscopic model of nonlinear elasticity  
with $F$ representing an affine deformation applied to a solid 
as will be discussed in detail in Chapter~\ref{sec:elasticity}.
In the scalar case, $m=1$, the model describes the behaviour of a random microscopic interface 
and  the map $F$ actually determines a vector---a macroscopic tilt  $u\in\R^d$
applied to the discrete interface and the free energy $W_{\!\beta}(F)$ 
then  corresponds to the interface free energy/surface tension $\sigma_\beta(u)$ with a given tilt.

\section{Main results for generalized gradient models}\label{sec:ggm}

In this section we first recall the natural equivalence of shift invariant interaction energies $U$ on a field $\varphi$
with energies $\Uscr$ that  act on the generalized  (discrete) gradients $D\varphi$ of $\varphi$, see Lemma~\ref{L:U-Uscr}.
Then we introduce an orthogonal decomposition of the Euclidean space of generalized   gradients $\Gcal$ into the subspace $\Gcal^\nabla$ corresponding to the first order
gradients (or, equivalently, linear maps) and a space $\Gcal^\perp$ corresponding to the higher gradients. In particular,  for every symmetric operator $\boldsymbol Q: 
\Gcal \to \Gcal$ we can define the operator $\boldsymbol Q^\nabla$ by restricting the corresponding quadratic form to $\Gcal^\nabla$, or, equivalently, to the corresponding linear maps,
see  \eqref{eq:defQscr_nabla} and  \eqref{eq:define_bsQ_nabla}.

With this notation in place, we state the main results for generalized gradient models: Theorem~\ref{thm:strictconvexity} 
covers locally uniform convexity of the free energy, while
Theorem~\ref{th:scalinglimit_concrete}  shows that  the limit of  an appropriate scaling of the Gibbs state is Gaussian.
\smallskip

We begin with the equivalence of shift invariant interaction energies and energies on generalized gradients.
In the previous section we considered finite range interactions with finite support $A \subset \Z^d$.  
We may assume, without loss of generality,  that $A$ is a  discrete cube of side $\Ran$. After a possible shift we may thus assume that
 $A = Q_{\Ran} :=\{0,\ldots, {\Ran}\}^d$.
We introduce the $m$ dimensional space of shifts $\Vcal_{Q_{\Ran}}=\{(a,\ldots,a)\in (\R^m)^{Q_{\Ran}}: a\in \R^m\}$ and its orthogonal complement $\Vcal_{Q_{\Ran}}^\perp$ in $(\R^m)^{Q_{\Ran}}$.
General interactions of range ${\Ran}$ are thus functions on the $m(({\Ran}+1)^{d}-1)$-dimensional space $\Vcal_{Q_{\Ran}}^\perp\simeq (\R^m)^{Q_{\Ran}}/\sim$ of local field configurations, where 
 the equivalence relation $\sim$ identifies configurations that  differ only by a constant field.
However, for our analysis it is more convenient to use an equivalent formulation with a space of local deformations introduced in terms
of  higher order derivatives of the fields.

We consider sets of multiindices $\Ical$ satisfying
\begin{align}\label{eq:def_Ical}
\{e_i\in \R^d:\,1\leq i\leq d\}\subset \Ical\subset 
\{\alpha\in \mathbb{N}_0^d\setminus \{(0,\ldots,0)\}:\;  \abs{\alpha}_\infty\leq {\Ran}\}.
\end{align}
Moreover we define the specific set $\Ical_{\Ran} :=\{\alpha\in \mathbb{N}_0^d\setminus \{(0,\ldots,0)\}:\;  \abs{\alpha}_\infty\leq {\Ran}\}$.
Note that the case $\Ical=\{e_1,\ldots ,e_d\}$ corresponds to nearest neighbour interactions.

Further, we consider the  vector space 
\begin{align}\label{eq:def_Gcal}
\Gcal=(\R^m)^\Ical
\end{align}
of extended gradients equipped with the standard scalar product 
\begin{equation}
\label{E:scprodG}
(z,z)_{\Gcal}=\sum_{\alpha\in\Ical} z_\alpha\cdot z_\alpha,  \  z=(z_\alpha)_{\alpha\in\Ical}\in \Gcal
\end{equation}	 
and the corresponding norm $|\cdot|_{\Gcal}$. When the set $\Gcal$  is clear from context,  we omit the subscript $\Gcal$.
We write $\Gcal_{\Ran} $ if $\Ical=\Ical_{\Ran}$.	
 For any $\p\in\Xcal_N$ and any $x\in T_N$, we then use 
$D\p(x)$ to denote \emph{the extended gradient}---the vector $(\nabla^\alpha \p(x))_{\alpha\in\Ical}\in\Gcal$
with $\nabla^\alpha \p(x)= \prod_{j=1}^d  \nabla_j^{\alpha(j)} \p(x)$.

Assuming that $L>{\Ran}+1$, so that the definition of  $D\p$ does not wrap around the torus, 
we have the following equivalence.
\begin{lemma}
\label{L:U-Uscr}
There exists  a linear isomorphism  $\Pi:\Gcal_{\Ran}\to \Vcal_{Q_{\Ran}}^\perp$ 
inducing  a one-to-one correspondence between functions on $\Vcal_{Q_{\Ran}}^\perp$ 
and those on $\Gcal_{\Ran}$. 
Namely, for any $U:\Vcal_{Q_{\Ran}}^\perp\to {\R}$, 
there is $\Uscr:\Gcal_{\Ran}\to \R$ such that $\Uscr(D\psi(0))=U(\psi)$ for any $\psi\in \Vcal_{Q_{\Ran}}^\perp$.
\end{lemma}
\begin{proof}
Both spaces $\Gcal_{\Ran}$ and $\Vcal_{Q_{\Ran}}^\perp$  have the same dimension $m(({\Ran}+1)^d-1)$.
The isomorphism between them can be explicitly given by the map 
$\Vcal_{Q_{\Ran}}^\perp\ni \psi\mapsto  D\psi(0)\in \Gcal_{\Ran}$.
This map is linear and injective ($D\psi_1(0)=D\psi_2(0)$ implies $\psi_1-\psi_2\in \Vcal_{Q_{\Ran}}$).
We define  $\Pi$ to be its inverse.

For any $U: \Vcal_{Q_{\Ran}}^\perp \to \R$, we define $\Uscr:\Gcal_{\Ran}\to \R$ by $\Uscr(z)=U(\Pi(z))$. 
Given that $\Pi$ is an isomorphism, we have $\Uscr(D\psi(0))=U(\psi)$ for any $\psi\in \Vcal_{Q_{\Ran}}^\perp$.
\end{proof}

There are obvious generalisations of the previous lemma to index sets $\Ical$ 
with the property that if $\alpha\in \Ical$ and $\beta\leq \alpha$ then $\beta\in \Ical$. 
In particular a similar statement holds for $A=\{0,e_1,\ldots,e_d\}$ and $\Ical=\{e_1,\ldots,e_d\}$.
\medskip

We next consider an orthogonal  decomposition of the space $\Gcal$ into a space isomorphic to linear maps
and the orthogonal space that corresponds to higher gradients. Define
\begin{align} \label{eq:G_nabla}
\Gcal^\nabla= & \, \{z\in \Gcal: z_\alpha=0 \text{ for }  \abs{\alpha}_1\neq 1\},  \\
 \Gcal^\perp= & \, \{z\in \Gcal: z_\alpha=0 \text{ for } \abs{\alpha}_1 = 1\},
 \end{align}
 where, $\abs{\alpha}_1=\sum_{i=1}^d \abs{\alpha_i}$.
 Note that the  dimension of $\Gcal^\nabla$ is  $md$ while the dimension of $\Gcal^\perp$
is $m(|\Ical|-d)$).
For any   $z\in\Gcal$ we denote by  $z^\nabla$ and $z^\perp$ the orthogonal projections of $z$
to  $\Gcal^\nabla$ and $\Gcal^\perp$, respectively.
We refer to $z^\nabla$ as to the \emph{gradient components} of $z$. 
In the following, we introduce  the convention of---unless  explicitly stated otherwise---skipping the index $1$  in the norm for 
multiindices (i.e., $\abs{\alpha}=\abs{\alpha}_1$) and $2$ in the Euclidean norm for vectors (i.e., $\abs{z}=\abs{z}_2$).

The vector space of linear maps $F:\R^d\to\R^m$ can be identified with the $md$-dimensional space $\Gcal^\nabla$ 
employing the isomorphism  $F\mapsto \overline{F}=DF(x)$ (for any $x\in \R^d$).
On $\mathop{\mathrm{Lin}}(\R^d; \R^m) \simeq \R^{m \times d}$  we define the usual Hilbert-Schmidt  scalar product by
\begin{align}
(F, G) = \sum_{i=1}^d  F e_i \cdot G e_i = \sum_{i=1}^d \sum_{s=1}^m F_{i,s} G_{i,s}.
\end{align}
With this  scalar product the isomorphism $F \mapsto \overline{F}$ becomes an isometry;
often we will  not   distinguish between $\abs{F}$ and $\abs{\overline F}$.

Let $\boldsymbol{Q}: \mathcal G \to \mathcal G$ be a symmetric  linear operator and consider, for $\alpha, \beta \in \Ical$, the
canonical projection $\pi_\alpha: \Gcal \to \R^m$ and injection  $\imath_\alpha: \R^m \to \mathcal G$  given by
$\pi_\alpha(z) = z_\alpha$ and $(\imath_\beta(y))_\gamma = \delta_{\beta \gamma} y$ where $ \delta_{\beta \gamma}$
is the Kronecker delta symbol. Define the operators $\boldsymbol Q_{\alpha \beta}: \R^m \to \R^m$ by 
$\boldsymbol Q_{\alpha \beta} = \pi_\alpha \circ  \boldsymbol Q \circ\imath_\beta$. 
If $\Qscr(z) =(z, \boldsymbol Q z)_\Gcal$ is the quadratic form induced by $\boldsymbol Q$ then 
$\Qscr(z) = \sum_{\alpha,\beta \in \Ical} (z_\alpha, \boldsymbol Q_{\alpha \beta} z_\beta)_{\R^m}$.
In particular
\begin{align}
\label{eq:defQ}
\Qscr(D\p(x))=\sum_{\alpha,\beta\in \Ical}  \nabla^\alpha\p(x)\cdot \boldsymbol{Q}_{\alpha\beta}\nabla^\beta\p(x). 
\end{align}
We define $\Qscr^\nabla$ as the restriction of $\Qscr$ to $\Gcal^\nabla$. Thus, 
\begin{equation}  \label{eq:defQscr_nabla}
\Qscr^\nabla(z) =  \sum_{|\alpha| = |\beta|=1} (z_\alpha, \boldsymbol Q_{\alpha \beta} z_\beta)_{\R^m} = \sum_{i,j=1}^d (z_i, Q_{ij} z_j)_{\R^m} \quad 
\text{for all $z \in \Gcal^\nabla$.}
\end{equation}
Here for $\alpha = e_i$ and $\beta = e_j$ we use the abbreviations $\boldsymbol Q_{ij} = \boldsymbol Q_{e_i e_j}$ and $z_i = z_{e_i}$.
By $\boldsymbol{Q}^\nabla$ we denote the  symmetric linear operator from $\Gcal^\nabla$ itself which corresponds to $\Qscr^\nabla$. 
By a slight abuse of notation,  we also denote by $\boldsymbol{Q}^\nabla$ the operator
from $\mathrm{Lin}(\R^d;\R^m)$ (or $\R^{m \times d}$) to itself
which is induced by the isometry $F \mapsto \bar F = DF$ from $\mathrm{Lin}(\R^d;\R^m)$ to $\Gcal^\nabla$. 
Thus
\begin{equation}  \label{eq:define_bsQ_nabla}
(F, \boldsymbol Q^{\nabla} G) := (\overline F, \boldsymbol Q^\nabla \overline G) = (\overline F, \boldsymbol Q \overline G)
\end{equation}
for all $F, G \in \R^{m \times d}$. 
For $i,j =1, \ldots, d$ and $s, t =1, \ldots, m$, we define the components of $\boldsymbol Q^{\nabla}$ by
\begin{align}  \label{eq:nabla_QU_components}
(\boldsymbol{Q}^\nabla)_{i,j;s,t} = (e_s \otimes e_i,  \boldsymbol Q^\nabla e_t \otimes e_j) =( e_s,   \boldsymbol Q_{ij} e_t)
\end{align}
\smallskip

We next formulate conditions on $\Uscr$ which are sufficient for our main results.
First note that if $\Uscr: \Gcal_{\Ran} \to \R$ is the energy that corresponds to a shift invariant interaction $U: \Vcal_{Q_{\Ran}}$ via
Lemma~\ref{L:U-Uscr}, then we get $U(\psi+F)=\Uscr(D\psi(0)+\overline{F})$ 
for any $\psi\in\Vcal_{Q_{\Ran}}^\perp$  where $\overline F = DF$ as above.
Thus the Hamilitonian $\HN^F(\p)$ can be expressed as
\begin{equation}
\label{eq:U-Uscr}
\HN^F(\p)=\sum_{x\in T_N} U( F+ (\tau_{-x}\p)_A) = \sum_{x\in T_N} \Uscr( \overline{F}+ D\p(x)).
\end{equation}

Substituting this equality into \eqref{eq:defofZ} and \eqref{eq:finite_volume_de}, we get for any function  $\Uscr$ on $\Gcal$
directly the corresponding free energy $W_{\!\beta,N}^\Uscr(F)$ and Gibbs state  $\gamma_{\beta,N}^{\Uscr,F}(\d \p)$,
\begin{equation}Z^{\Uscr}_{\beta,N}(F,f) = 
\int_{\Xcal_N}\exp\bigl(-\beta \sum_{x \in T_N} \Uscr(F+D\p(x)) +(f,\p)\bigr)\lambda_N(\d\p),
\end{equation}
\begin{equation} 
W_{\!\beta,N}^\Uscr(F)=-  \frac{\ln\, Z_{\beta,N}^\Uscr(F,0)}{\beta L^{Nd}}
\end{equation}
\begin{equation}    \label{eq:Gibbs_Uscr} 
 \gamma_{\beta,N}^{\Uscr,F}(\d \p)=\frac{1}{Z_{\beta,N}(F,0)}\exp\left(-\beta \sum_{x \in T_N} \Uscr(F+D\p(x))
\right)\lambda_N(\d\p).
\end{equation}
We will usually drop the index $\Uscr$ when the function $\Uscr$ is clear from the context.

For any twice differentiable function  $\Uscr$ on $\Gcal$ we define the symmetric quadratic form $\QU$ by
\begin{equation}  
\label{eq:define_QU_V}
\QU(z) := D^2\Uscr(0)(z,z).
\end{equation}
For integers  $r_0\geq 3$ and $r_1 \ge 0$ we now consider the following conditions
\begin{equation} \label{eq:V1_new}   
 \Uscr\in C^{r_0+r_1} (\Gcal);
\end{equation}
\begin{equation}  \label{eq:V1bis_new}
\omega_0 \abs{z}^2 \le \Qscr_\Uscr(z) \le \omega_0^{-1} \abs{z}^2  \quad \text{for some $\omega_0\in(0,1)$;}
\end{equation} 
\begin{equation}
\label{eq:V3_new}     \Uscr(z) - D\Uscr(0) z - \Uscr(0)  \geq  \omega\abs{z}^2   \text{ for all } z \in \mathcal G,   \text{and some
$\omega \in (0,  \frac{\omega_0}{8})$;} 
\end{equation}                 
\begin{equation}
\label{eq:V4_new}   \lim_{t \to \infty} t^{-2} \ln \Psi(t) = 0 
\text{ where }  
\Psi(t) := \sup_{\abs{z} \le t}   \sum_{3 \le \abs{\alpha} \le r_0+r_1}   \frac{1}{\alpha!} \abs{\partial^\alpha \Uscr(z)}.
\end{equation}

\begin{theorem}
\label{thm:strictconvexity} Suppose that a function  $\Uscr$ on $\Gcal$ satisfies  the conditions  \eqref{eq:V1_new}--\eqref{eq:V4_new}  for some 
$r_0\ge 3$ and $r_1 \geq  2$.
Then there exist positive constants $\beta_0 $ and $ \delta$
such that the corresponding free energies  $W_{\!\beta,N}^{\Uscr}\bigr|_{B_{\delta}(0)}$ are  $C^{r_1}$ and uniformly  convex for $\beta\geq \beta_0$.
In particular,  $D^2 W_{\!\beta, N}^{\Uscr}(F)(\dot{F},\dot{F}) \ge \frac{\omega_0}{4}  \abs{\dot{F}}^2$. 
Also, every limit $W_{\!\beta}^{\Uscr} = \lim_{\ell \to \infty} W_{\!\beta,N_\ell}^{\Uscr}$ is uniformly convex. 
\end{theorem}

The proof of Theorem~\ref{thm:strictconvexity} (in Section~\ref{sec:abstract_to_ggm})  
actually gives the following stronger result. For $\beta \geq \beta_0$ there exist functions 
$\Wcal_{\!\beta,N} \in C^{r_1}(B_\delta(0))$ and constants $c_{\beta,N}$ such that 
\begin{align}  \label{eq:beta_expansion_ggm}
W_{\!\beta, N}^{\Uscr}(F)= & \, \Uscr(\overline{F})+ c_{\beta,N} + \frac{\Wcal_{\beta,N}^{\Uscr}(F)}{\beta} 
\end{align}
and the $C^{r_1}(B_\delta(0))$ norm of the functions $\Wcal_{\beta,N}^{\Uscr}$ is bounded, 
uniformly for $N \ge 1$ and  $\beta \ge \beta_0$. See \eqref{eq:sigmaexprN} and Theorem~\ref{th:bound_Wcal}.

Let us remark that for a long time even for gradient interface models with uniformly convex potentials
it was only known that the free energy is in $C^{1,1}$ and the question
whether the free energy is in $C^2$ was an open problem (see, e.g., \cite{FS97}).
Only very recently it was shown in \cite{AW19} that the free energy
is in $C^{2,\alpha}_{\mathrm{loc}}$ for some $\alpha>0$ provided  that the 
local interaction is described by an isotropic  nearest neighbour interaction energy 
$U(\varphi_A) = \sum_{i=1}^d V(\varphi(e_i) - \varphi(0))$, with $A = \{0, e_1, \ldots, e_d\}$,  and that  the interaction potential
$V  : \R \to \R$   is symmetric, uniformly convex, and in $C^{2,\gamma}$ for some $\gamma>0$.
The result in \cite{AW19} is global and requires less regularity on the local interaction then  Theorem~\ref{thm:strictconvexity},
 but it crucially relies on global uniform convexity, because it uses, among various other tools,  the Helffer-Sj\"ostrand operator.
 We are in particular interested in application to nonlinear elasticity where global convexity is not possible in view of the fundamental
 $\mathrm{SO}(d)$ symmetry (see Section~\ref{sec:continuous_elasticity}) and where one needs to go beyond nearest neighbour interactions.
  In that setting, local uniform convexity of the  free energies  at sufficiently low temperatures as expressed in 
  Theorem~\ref{thm:strictconvexity} appears to be the best possible result. The loss of $r_0 =3$ derivatives in passing from 
  $\Uscr$ to the perturbative contribution $\Wcal_{\beta,N}$ to the free energy in  \eqref{eq:beta_expansion_ggm} 
  is due to our method and we do not know whether this can be improved.
 
We also note  that by Remark \ref{rem:coercive} in Section \ref{sec:reformulation_de}, the conditions  \eqref{eq:V1bis_new}
 and \eqref{eq:V3_new} can be replaced by the following weaker conditions. There exists $\omega'_0$ and $\omega' \in (0, \frac{\omega'_0}{8})$ such that 
\begin{equation}  \label{eq:V1bis_new_weakened}
\omega'_0 \abs{z^\nabla}^2 \le \Qscr_\Uscr(z) \le {\omega'_0}^{-1} \abs{z}^2  \quad \text{for all $z \in \Gcal$,}
\end{equation} 
\begin{equation}
\label{eq:V3_new_weakened}    
 \Uscr(z) - D\Uscr(0) z - \Uscr(0)  \geq  \omega'\abs{z^\nabla}^2   \text{ for all } z \in \mathcal G.
 \end{equation}
 Here $z^\nabla$ denotes the gradient component of $z$, i.e. the orthogonal projection 
of $z$ onto the space $\Gcal^\nabla$, defined in  \eqref{eq:G_nabla}.
Thus it suffices to have a lower bound in terms of the gradient part $z^\nabla$ rather than the full vector $z$. 
\medskip

We now turn to the  scaling limit of the Gibbs measures. The result below  is a generalisation of a result first shown in \cite{Hil15}.
 Recall that the gradient Gibbs 
measure $\gamma_{\beta,N}^{\Uscr,F}$ with deformation $F$   is given by 
  \eqref{eq:Gibbs_Uscr}. 
  
By $\boldsymbol{Q}_\Uscr$ we denote the symmetric operator on $\mathcal G$ which corresponds to
the symmetric bilinear form $D^2\Uscr(0)$. For the definition of the $\boldsymbol{Q}_\Uscr^\nabla$, the restriction $\boldsymbol{Q}_\Uscr$
to linear maps, and its components, see 
 \eqref{eq:define_bsQ_nabla} and \eqref{eq:nabla_QU_components}.

\begin{theorem} \label{th:scalinglimit_concrete}
Suppose that the assumptions  \eqref{eq:V1_new}--\eqref{eq:V4_new}  hold with $r_0 =3$ and $r_1 = 0$.
There exists an integer $L_0$ such that for every odd integer $L \ge L_0$ there is a positive  constant $\delta$ with the following property.
For every $\beta \ge 1$ and every $F \in B_\delta(0)$ there exists a subsequence  $(N_\ell)$ and 
a matrix $\boldsymbol{q}\in \R^{(m\times d)\times (m\times  d)}_{\mathrm{sym}}$
such that for $f\in C^\infty(\mathbb{T}^d, \R^m)$ with $\int_{\mathbb{T}^d} f=0$ 
and  $f_N(x)=L^{-N\frac{d+2}{2}}f(L^{-N}x)$, we have
\begin{equation}  \label{eq:convergence_laplace_transform_ggm}
\lim_{\ell\rightarrow \infty}   \mathbb{E}_{  \gamma_{  \beta , N_{\ell} }^{\Uscr,F} }e^{(f_{N_\ell},\varphi)}= 
e^{\frac1{2\beta}(f,\mathscr{C}^{\Uscr}_{\mathbb{T}^d} f)}.
\end{equation}
Here $ \gamma_{  \beta , N_{\ell} }^{\Uscr,F}$ is the Gibbs state corresponding to the potential $\Uscr$ 
and  $\mathscr{C}^{\Uscr}_{\mathbb{T}^d}$ is the inverse of the operator $\mathscr{A}^{\Uscr}_{\mathbb{T}^d}$ 
acting on $u\in H^1(\mathbb{T}^d,\R^d)$ with vanishing mean $\int u=0$  by
\begin{equation}  
\label{eq:limiting_operator2}
(\mathscr{A}^{\Uscr}_{\mathbb{T}^d}u)_s=-\sum_{t=1}^m \sum_{i,j=1}^d (\boldsymbol{Q}_\Uscr^{\nabla}-\boldsymbol{q})_{i,j;s,t} \partial_i\partial_j u_t. 
\end{equation}
\end{theorem}

\begin{remark}\hfill 
\label{re:correlation_structure_concrete}
\begin{itemize}[leftmargin=0.6cm]
\item[(1)]
Note that the rescaling $L^{-\frac{Nd}{2}}$ would correspond to a central limit law behaviour 
of the  random field. 
Due to the strong correlations we need to use the stronger rescaling with $L^{-N(\frac{d+2}{2})}$.
One easily sees that the scaling limit of the gradient field $\nabla \p$ involves the central limit scaling, 
cf. e.g.\ \cite{BS11} and \cite{NS97}.
\item[(2)] 
The matrix $\boldsymbol{q} $ and thus also  the operators $\mathscr{A}^{\Uscr}_{\mathbb{T}^d}$  and $ \mathscr{C}^{\Uscr}_{\mathbb{T}^d}$
depend on $\beta$ and $F$.  In order to avoid  overloaded   notation we omit an explicit reference to this dependence.

Note, however, that the limiting covariance is dominated by the gradient-gradient contribution
of the interaction while the higher order terms are not directly present, see also \cite{NS97}. 
In other words, the limiting covariance $\mathscr{C}$ depends only on the action of $\boldsymbol Q$ on the subspace $\mathcal G^\nabla$, 
defined in \eqref{eq:G_nabla}
and hence only on $\boldsymbol Q^\nabla$. Recall  that $\mathcal G^\nabla$ is naturally identified with the space of linear maps from $\R^d$ to $\R^m$ or,
equivalenty, with the space of matrices $\R^{m \times d}$.
There is, in general, an implicit dependence  of the limiting covariance   on the higher gradient terms 
$\boldsymbol{Q}_{\alpha \beta}$  of the quadratic interaction,  as well as  on the nonlinear terms in the interaction, through the matrix $\boldsymbol{q}$.
In fact such a dependence on  higher gradient terms of the quadratic interaction arises already in the Gaussian setting.

The higher  gradient  terms can change the local correlation structure. 
They have, however, little influence on the  long distance correlations because, roughly speaking, their long wave Fourier modes 
are very small and decay like  $\abs{p}^{\abs{\alpha} + \abs{\beta} }$ with  $\abs{\alpha} + \abs{\beta} \geq 3$, compared to
a  $\abs{p}^2$  decay for the gradient-gradient interaction. 
In term of the proof, the main point is that in equation   \eqref{eq:covariance_scaling_higher}  the terms with $|\alpha| \ge 2$ or $|\beta| \ge 2$ converge to zero.
\end{itemize}
\end{remark}

\begin{remark}[Passage to subsequences]  \label{re:subsequences}
Hilger  \cite{Hil18} has shown that for a scalar model with nearest neighbour
interactions one has convergence for the full sequence, both for Laplace transform of the Gibbs measures,
see  \eqref{eq:convergence_laplace_transform_ggm},  and for the free energy. 
Her approach follows closely the idea of Brydges and Slade  \cite[Section1.8.3]{BS15V} who compare the 
finite volume renormalisation flow to an infinite volume flow. We believe that the reasoning in \cite{Hil18} can
be adapted to our situation, but since the discussion of the finite volume renormalisation flow is already
technically quite demanding, we do not discuss such an extension in this work. 
The existence of the thermodynamic limit $ \lim_{N \to \infty} W_{\!\beta, N}(F)$ can probably be obtained
from  less sophisticated arguments, see, for example,  \cite{KL14} for results for slightly
different boundary conditions which require only  very weak conditions on the local  interaction $U$.
\end{remark}

\section{Main results for discrete elasticity}   \label{sec:discrete_elasticity_main}

 In this section we consider models of discrete elasticity and analyse local convexity properties of the free energy
 and the scaling limit of Gibbs measures. Indeed, the study of such models is a key  motivation for the present work
 and it is the reason why we considered vector-valued fields and interactions beyond nearest neighbour interactions
 in the previous  section. An additional difficulty in discrete nonlinear elasticity is that the invariance
of the local interaction energy under rotations leads to a degeneracy of the quadratic form $\Qscr_\Uscr$ 
which for elasticity corresponds to the second derivative of $\Uscr$ at the identity map.
In particular,  the  lower bound in  \eqref{eq:V1bis_new}  cannot hold. 
Hence  the results in the previous  section  cannot be applied directly. 
We will overcome this difficulty by adding a suitable discrete null Lagrangian,  see Definition~\ref{de:null_lagrangian},    
equation  \eqref{eq:hqsecway} and  Lemmas~\ref{le:equivalence_de_ggm} and \ref{le:embedding_de}  in the next  chapter.  
In this section we focus on the statement of the main results for discrete elasticity.

 We consider the general setting of  Section~\ref{sec:setup}  with $m=d$. 
For simplicity (and without loss of generality), we suppose that
the support set $A$ of the potential $U$ contains the unit cell of $\Z^d$, i.e.,  $\{0,1\}^d\subset A$.
As before we use the splitting $(\R^d)^A= \Vcal_A \oplus \Vcal_A^\perp$,
where  $\Vcal_A\sim \R^d$ is the $d$-dimensional subspace of shifts 
\begin{align}
\Vcal_A=\{(a,\dots,a)\in (\R^d)^A: a\in \R^d\},
\end{align}
 and $\Vcal_A^\perp$ is the  $d(\abs{A}-1)$-dimensional orthogonal complement of $\Vcal_A$.

For a linear map $F:\R^{d}\to \R^{d}$ we consider the extension to $(\R^{d})^A$   given by  $(F\psi)(x)=F(\psi(x))$. 
For ease of notation we will use the same symbol $F$ for the original map and the extension to $(\R^{d})^A$
and similarly for the extension to $(\R^d)^{\Z^d}$.

We assume that the potential $U:\bigl(\R^{d}\bigr)^A\rightarrow \R$  satisfies the following conditions: 

\begin{enumerate}[(H1),leftmargin=0.9cm]
\item[(H1)] \textit{Invariance under rotations and shifts:}
We have
\begin{align}
U(\psi)=U(\mathbf{R}(t_a \psi)) 
\end{align}
for any   $\psi\in  ( \R^d )^{A}$ and any $\mathbf{R}\in SO(d)$, $a\in \R^d$, with
	   
 \noindent
 $\mathbf{R}(t_a \psi)(x)= \mathbf{R}(\psi(x)+a).$
\item[(H2)] 
\textit{Ground state:} 
$U(\psi)\geq 0$ and $U(\psi)=0$ if and only if $\psi$ is a rigid body rotation, i.e., 
there exists  $\mathbf{R}\in \mathrm{SO}(d)$ and $a\in\R^d$ such that $\psi(x)=\mathbf{R}x+a$ for any $x\in A$.
\item[(H3)] \textit{Smoothness and convexity:}
Let $\1\in  ( \R^d )^{A}$ denote the identity configuration $\1(x)=x$. Assume that  $U$ is a $C^2$ function and 
$D^2U (\1)$ is positive definite on the subspace orthogonal to shifts and infinitesimal rotations given by skew-symmetric linear maps.
\item[(H4)] \textit{Growth at infinity:}
\begin{align}  \label{eq:H4_growth_at_infinity}
\liminf_{\psi\in \Vcal_A^\perp,\, |\psi|\rightarrow \infty}\frac{U(\psi)}{|\psi|^d}>0. 
\end{align}
\item[(H5)] \textit{Additional smoothness and subgaussian bound:} 
\begin{align}  \label{eq:subgaussian_bound}
 \lim_{|\psi| \to \infty}  |\psi|^{-2}  \ln \Big(  \sum_{2 \le |\alpha|_{1} \le  r_0+r_1}  \frac{1}{\alpha!}  \abs{\partial_\psi^\alpha U(\psi)} \Big) = 0,
\end{align}
where  we use the notation $\partial_\psi^\alpha U(\psi)= \prod_{x\in A} \prod_{s=1}^d \frac{\partial^{|\alpha|}}{\partial_{\psi_s(x)}^{\alpha(x,s)}}U(\psi)$
for any multiindex $\alpha: A \times \{1,\dots, d\} \to \mathbb N$.
\end{enumerate}
The first four conditions are the same as in \cite{CDKM06}. The last condition is a minor
additional  regularity assumption for the potential. It was stated as a separate item to make clear
that it is only required in the renormalisation group analysis but not in 
the convexification argument in Section~\ref{sec:discrete_null_lagrangians} below.

In \cite{CDKM06} the assumptions (H1) to (H4) are used to prove that the Cauchy-Born rule holds at zero temperature,
in the sense that the energy minimiser subject to affine boundary conditions is affine.
Here we use this result as a starting point for a study of the Gibbs distribution for the Hamiltonian $\HN$ at low temperatures.
The ground state in the setting of discrete elasticity corresponds to the linear  deformation given by the identity map.
Therefore we consider deformations $F\in \R^{d\times d}$   for which  $F-\1$  is small.  

We recall from  \eqref{eq:defHF} that the Hamiltonian $\HN^F$ is defined by
\begin{align}  \label{eq:defHF_elastic}
\HN^F(\p)=\sum_{x\in T_N} U(( \tau_{-x}\p)_A +F).
\end{align}
and we recall the definition of  the corresponding partition function $Z_{\beta,N}(F,0)$
in \eqref{eq:defofZ} and the finite volume free energy
\begin{align}  \label{eq:W_N_beta_elastic}
W_{\!\beta, N}^{U}(F) = - \frac{\ln Z_{\beta,N}(F, 0)}{\beta L^{Nd}}
\end{align} 
in \eqref{eq:defWN}.
 
Note that $W_{\!\beta, N}^{U}$ inherits the rotational invariance of $U$, i.e.
 \begin{align}  \label{eq:WNbeta_frame_indifferent}
 W_{\!\beta, N}^{U}(RF) = W_{\!\beta, N}^{U}(F) \quad \hbox{for all $R \in SO(d)$.}
 \end{align}
This follows immediately from the fact that the Hausdorff measure $\lambda_N$ on the space $ \Xcal_N$  of 
$L^N$ periodic fields with average zero is invariant under the map $\p \mapsto R \p$.

\begin{theorem}\label{T:deW}
Suppose that the potential $U$ satisfies the assumptions $(H1)$ to $(H5)$ with $r_0 =3$ and $r_1 \ge 2$.

Then for all sufficiently large odd $L$ there exist  positive constants $\beta_0 $, $c_0$, and $ \delta$   such that, 
for any $\beta\geq \beta_0$  and any $N \ge 1$ the functions $W_{\!\beta,N}^{U}: B_\delta(\1)\to \R$ are in  $C^{r_1}$	
and  the restrictions $W_{\!\beta,N}^{U}:B_\delta(\1) \cap \R_{\mathrm{sym}}^{d\times d}\to \R$ are uniformly convex.
Moreover, there exists a subsequence $(N_\ell)$ such that $W_{\!\beta,N_\ell}^{U}$ converges in $C^{r_1-1}$ to  
a  free energy $W_{\!\beta}^{U}$.
For  $r_1 \ge 3$  the restrictions of $W_{\!\beta}^{U}$ to  $B_\delta(\1)\cap \R_{\mathrm{sym}}^{d\times d}$ is uniformly convex.

In addition, there exists a null Lagrangian $\nullL$ (actually a multiple of the determinant function) such that
 the functions $W_{\!\beta,N}^{U} + \nullL$ are uniformly convex in $B_\delta(\1)$. 
 Finally, the functions $W_{\!\beta,N}^{U}$ are uniformly quasiconvex at matrices $G \in B_{\delta/2}(\1)$,
i.e., for all open and bounded sets $\Omega \subset \R^d$ and for all $C^1$ functions $\phi: \Omega \to \R^d$ with compact support,
\begin{equation}   \label{eq:uniformly_quasiconvex}
\int_\Omega  \bigl(W_{\!\beta, N}^{U}(G + \nabla \phi)   -W_{\!\beta, N}^{U}(G)\bigr) \, dx  \ge c_0 \int_\Omega |\nabla \phi|^2 \, dx.
\end{equation}
\end{theorem}
\medskip

In view of the rotational invariance of $W_{\!\beta, N}$, see   \eqref{eq:WNbeta_frame_indifferent}, we cannot expect convexity on $B_\delta(\1)$. 
Convexity on the symmetric matrices is a natural substitute since $W_{\!\beta, N}$ is determined by its values on symmetric matrices. 
Indeed $W_{\!\beta, N}(F) = W_{\!\beta, N}( \sqrt{F^T F})$ for all $F \in B_\delta(\1)$.

We also get the following result for the scaling limit of the Gibbs measures,  see 
 \eqref{eq:define_bsQ_nabla} and  \eqref{eq:nabla_QU_components}, for the definitions of $\boldsymbol Q_U^\nabla$
 and  $(\boldsymbol{Q}^\nabla_U)_{i,j;s,t}$, respectively. Here $Q_U$ is the symmetric linear operator on $\Gcal$ 
 induced by the symmetric bilinear form $D^2 \Uscr(\1)$.

\begin{theorem}\label{T:descaling}  
Suppose the potential $U$ satisfies the assumptions $(H1)$ to $(H5)$ with $r_0 =3$ and $r_1\ge 2$.
Then there exists an integer $L_0$ such that for every odd integer $L \ge L_0$ there is a positive  constant $\delta $ with the following property.
For every $\beta \ge 1$ and every $F \in B_\delta(0)$ there exists a subsequence  $(N_\ell)$ and 
a matrix $\boldsymbol{q}\in \R^{(m\times d)\times (m\times  d)}_{\mathrm{sym}}$
such that for $f\in C^\infty(\mathbb{T}^d,\R^d)$ with $\int_{\mathbb{T}^d} f=0$ and  $f_N(x)=L^{-N\frac{d+2}{2}}f(L^{-N}x)$,
\begin{align}
\lim_{\ell\rightarrow \infty}\mathbb{E}_{\gamma_{\beta,N_\ell}^{U,F}}e^{(f_{N_\ell},\p)}= e^{\frac1{2 \beta}(f,\mathscr{C}_{\mathbb{T}^d} f)}  .                                                      
\end{align}
Here, $\mathscr{C}_{\mathbb{T}^d} :f\to \mathscr{A}_{\mathbb{T}^d}^{-1}f$ is the inverse of the operator $\mathscr{A}_{\mathbb{T}^d}$ 
acting on functions  $u\in H^1(\mathbb{T}^d,\R^d)$ with  $\int u=0$   by
\begin{align}
(\mathscr{A}_{\mathbb{T}^d} u)_s =-\sum_{t=1}^d \sum_{i,j=1}^d (\boldsymbol{Q}_U^\nabla - \boldsymbol{q})_{i,j;s,t} \partial_i \partial_j u_t. 
\end{align}
\end{theorem}

For a discussion why only the restriction $\Qscr_U^\nabla$ and not the full quadratic form $\Qscr_U$ appears in the limiting covariance,  
see Remark~\ref{re:correlation_structure_concrete}.
The operators  $\boldsymbol Q^\nabla$ and $\boldsymbol Q^\nabla - \boldsymbol q$  are not positively definite 
on the set of all matrices because skew-symmetric matrices are in their null space. 
These operators  are, however, positive definite on symmetric matrices. 
By Korn's inequality this implies that $\mathscr{A}$ is an elliptic operator and that its inverse $\mathscr C$ is well-behaved. 
Actually we will see in the proof of Theorem~\ref{T:descaling} that the operator $\mathscr A$ can be also written in terms
of $\boldsymbol Q^\nabla_{U + \nullL}$ such that 
$\boldsymbol Q^\nabla_{U +\nullL}$ and  $\boldsymbol Q^\nabla_{U +\nullL} - \boldsymbol q$ are positive definite. 

 Regarding the choice of a subsequence, see Remark~\ref{re:subsequences}.

\section[Main results for abstract perturbations]{Main results for abstract perturbations}  \label{sec:main_abstract}

To set the stage, we  derive an expansion for the free energy and the Gibbs measure, by splitting the
interaction energy $\Uscr$ into a quadratic term and a remainder.

As before, we assume that the quadratic form $\QU = D^2 \Uscr(0)$  satisfies  the bounds
\begin{equation}
\label{eq:Qlowerbound}
\omega_0 \abs{z}^2 \leq \QU(z)\leq \omega_0^{-1}\abs{z}^2\  \text{ for all }\   z\in\Gcal
\end{equation}
for some $\omega_0\in (0,1)$. 

Similarly to \cite{AKM16}, we introduce  the  function $\overline{\Uscr}:\Gcal\times \R^{m\times d}\rightarrow \R$  by
\begin{equation}
\label{eq:defofUbar}
\overline{\Uscr}(z,F)=\Uscr(z+\overline{F})- \Uscr(\overline{F})- D\Uscr(\overline{F})(z)-\frac{\QU(z)}{2}. 
\end{equation}
It describes the remainder of the Taylor expansion of $\Uscr(z+\overline{F})$ around $\overline{F}$
collecting all third order terms plus the difference $D^2\Uscr(\overline{F})(z,z)-D^2\Uscr(0)(z,z)$ 
since we want to keep only the quadratic term  that does not depend on $\overline{F}$.
Notice that the function  $\Vscr(z)  = \overline{\Uscr}(z,0)=\Uscr(z)- \Uscr(0) - D\Uscr(0)z - \frac{\QU(z)}{2}$
is the third order remainder of the Taylor expansion of $\Uscr$ yielding  $\Vscr(0) = D\Vscr(0) = D^2 \Vscr(0) =0$.

The Hamiltonian $\HN^F$ can be expressed in terms of the function $\overline{\Uscr}$  as
\begin{align}
\label{eq:finalEnergyManip}
\HN^F(\p) = & \sum_{x\in T_N} \Uscr(D\p(x)+\overline{F}) \\
= &\  L^{Nd}\Uscr(\overline{F})+                                                                             
\sum_{x\in T_N} \Bigl(\overline{\Uscr}(D\p(x),F)+\frac{\Qscr_\Uscr(D\p(x))}{2}\Bigr)\notag,
\end{align}
where we used that the terms linear in $D\p(x)$ cancel because $\sum_{x\in T_N}D\p(x)=0$ in the periodic setting.
Using equation \eqref{eq:finalEnergyManip} we can rewrite the partition function \eqref{eq:defofZ}  as
\begin{equation}
Z_{\beta,N}(F,f)= e^{-\beta L^{Nd}\Uscr(\overline{F})}
\int_{\Xcal_N} e^{(f,\p)} e^{-\beta \sum_{x\in T_N} (\overline{\Uscr}(D\p(x),F) 
+\frac{\QU(D\p(x))}{2})}\lambda_N(\d\p).                                                            
\end{equation}
The positive quadratic form $\beta\QU$ defines the Gaussian probability measure
\begin{equation}  
\label{eq:gaussian_QU}
\mu_\beta(\d \p)=\frac1{Z_{\beta,N}^{\QU}}  \exp\bigl(-\tfrac{\beta}{2}\sum_{x\in T_N}
\QU(D\p(x))\bigr) \lambda_N(\d\p)                                       
\end{equation}
with an appropriate normalization  factor $Z_{\beta,N}^{\QU}$.

Thus
\begin{align}
Z_{\beta,N}(F,f)=e^{-\beta L^{Nd}\Uscr(\overline{F})} Z_{\beta,N}^{\QU} \int_{\Xcal_N} e^{(f,\p)}e^{-\beta\sum_{x\in T_N} 
\overline{\Uscr}(D\p(x),F)}\,\mu_\beta(\d\p).
\end{align}
Finally, rescaling the field by $\sqrt{\beta}$,  introducing the Mayer function corresponding to the remainder $\overline{\Uscr}$,
\begin{align}
\label{eq:defofKuVbeta}
\Kcal_{F,\beta,\Uscr}(z)=\exp\bigl(-\beta \overline{\Uscr}(\tfrac{z}{\sqrt{\beta}},F)\bigr)-1 ,
\end{align}
and using the shorthand $\mu=\mu_1$, we express  the partition function $Z_{\beta,N}(F,f) $ in terms of the  polymer expansion  
\begin{align}  
\begin{split}
\label{eq:initialfinal}
Z_{\beta,N}(F,f) 
& =e^{-\beta L^{Nd}\Uscr(\overline{F})} Z_{\beta,N}^{\QU} \int_{\Xcal_N} 
e^{(f,\frac{\p}{\sqrt \beta})}e^{-\beta\sum_{x\in T_N} \overline{\Uscr}(\frac{D\p(x)}{\sqrt{\beta}},F)}\,\mu(\d\p)\\
& =e^{-\beta L^{Nd}\Uscr(\overline{F})} Z_{\beta,N}^{\QU}\int_{\Xcal_N} 
e^{(\frac{f}{\sqrt \beta},\p)}\prod_{x\in T_N}(1+\Kcal_{F,\beta,\Uscr}(D\p(x)))\,\mu(\d\p) \\
& =e^{-\beta L^{Nd}\Uscr(\overline{F})} Z_{\beta,N}^{\QU}\int_{\Xcal_N}
e^{(\frac{f}{\sqrt \beta},\p)}\sum_{X\subset T_N}\prod_{x\in  X}\Kcal_{F,\beta,\Uscr}(D\p(x))\,\mu(\d\p).
\end{split}
\end{align}
Here we use the convention $\prod_{x\in\emptyset}\Kcal(D\p(x))=1$ for the empty product.
The integral in the last expression gives the perturbative contribution 
\begin{align}
\label{eq:Zpertcomp}
\Zcal_{\beta,N}\Bigl(F,\frac{f}{\sqrt \beta}\Bigr)=\int_{\Xcal_N}e^{(\frac{f}{\sqrt \beta},\p)}
\sum_{X\subset T_N}\prod_{x\in X}\Kcal_{F,\beta,\Uscr}(D\p(x))\,\mu(\d\p) .                                                                
\end{align}
Introducing  the perturbative components of the free energy by
\begin{align}\label{eq:defvarsigmaN}
\Wcal_{\beta,N}(F)=- \frac{\ln\Zcal_{\beta,N}(F,0)}{L^{Nd}}\text{  and  }
\Wcal_\beta(F)=-\lim_{N\rightarrow \infty}\frac{\ln\Zcal_{\beta,N}(F,0)}{L^{Nd}},
\end{align}
we can rewrite the $W_{\!\beta, N}$ and the  free energy  $W_{\!\beta}$ defined in   \eqref{eq:defWN} and \eqref{eq:defW} as
\begin{align}
\label{eq:sigmaexprN}
W_{\!\beta, N}(F)= & \, \Uscr(\overline{F})+\frac{\Wcal_{\beta,N}(F)}{\beta}-     \frac{1}{\beta L^{Nd}}\ln  Z_{\beta,N}^{\QU}, \\
\label{eq:sigmaexpr}
W_{\!\beta}(F)=  & \, \Uscr(\overline{F})+\frac{\Wcal_\beta(F)}{\beta}-\lim_{N\rightarrow \infty}\frac{1}{\beta L^{Nd}}\ln  Z_{\beta,N}^{\QU}.
\end{align}
Here, in both expressions, the last term is a constant independent of $F$.
Equations \eqref{eq:initialfinal}--\eqref{eq:Zpertcomp}  and  \eqref{eq:defvarsigmaN}--\eqref{eq:sigmaexprN} 
give the desired splitting of the partition function and the free energy into  a leading order low temperature contribution and a remainder. 
 The main task is to control the remainder.

\medskip

To do so, we generalize  the formulation slightly. Instead of a particular $\Kcal_{F,\beta,\Uscr}$ as  above,
we consider a general function $\Kcal: \Gcal\to\R$  and instead of the quadratic form $\QU$ depending on $\Uscr$,  
we consider a general positive definite quadratic form $\Qscr$. We then  define  the  partition function
\begin{align}
\label{eq:Zpertcomp_general}
\Zscr_{N}(\Kcal,\Qscr,f)=\int_{\Xcal_N}e^{(f,\p)}
\sum_{X\subset T_N}\prod_{x\in X}\Kcal(D\p(x))\,\mu_{\Qscr}(\d\p)                                                             
\end{align}
with  the Gaussian  probability  measure 
\begin{align}  
\label{eq:gaussian_Q}
\mu_{\Qscr}(\d \p)=\frac1{Z_{\beta,N}^{\Qscr}}  \exp\bigl(-\tfrac{1}{2}
\sum_{x\in T_N}\Qscr(D\p(x))\bigr) \lambda_N(\d\p).                                       
\end{align}
Introducing the free energies
\begin{equation}
\label{eq:defvarsigmaN_general}
\Wscr_{\!N}(\Kcal,\Qscr)=-\frac{\ln\Zscr_{N}(\Kcal,\Qscr,0)}{L^{Nd}}
\end{equation}
and 
\begin{equation}\label{eq:defvarsigma_general}
\Wscr(\Kcal,\Qscr)=-\lim_{N\rightarrow \infty}\frac{\ln\Zscr_{N}(\Kcal,\Qscr,0)}{L^{Nd}},
\end{equation}
we readily get
\begin{equation}\label{eq:Wcal=Wscr}
\Wcal_{\beta,N}^{\Uscr}(F)=\Wscr_{\!N}(\Kcal_{F,\beta,\Uscr},\Qscr_{\Uscr}) \ \text{ and }\
\Wcal_\beta^{\Uscr}(F)=\Wscr(\Kcal_{F,\beta,\Uscr},\Qscr_{\Uscr}).
\end{equation}

The key result of this paper consists in a good control  of the behaviour of the partition function $\Zscr_{N}(\Kcal,\Qscr,f)$ 
and thus also  the corresponding $\Wscr_{\!N}(\Kcal,\Qscr)$ and $\Wscr(\Kcal,\Qscr)$ for a  class  of admissible perturbations $\Kcal$. 
We first introduce an appropriate  function space which encodes  natural conditions on the perturbations $\Kcal$. 
We will then formulate conditions on the concrete  interaction
energies $\Uscr$ that guarantee that $K_{F,\beta,\Uscr}$ (accompanied by $\Qscr=\Qscr_{\Uscr}$) belongs to that  space.

 Let $\Qscr: \Gcal \rightarrow \R$ be a positive definite quadratic form, $\zeta \in (0,1)$,  and $r_0\geq 3$ an integer. 
We define the Banach space $\boldsymbol{E}_{\zeta, \Qscr}$ consisting of functions
$\Kcal:\Gcal=\left(\R^m\right)^{\Ical}\rightarrow \R$ such that that the following norm is finite
\begin{equation}
\label{eq:normE}
\norm{ \Kcal }_{\zeta, \Qscr} = 
\sup_{z\in \Gcal} \sum_{\abs{\alpha}_1\leq r_0} \frac1{\alpha !}\abs{\partial^\alpha \Kcal(z)}e^{-\frac{1}{2}(1-\zeta)\Qscr(z)} .                                           
\end{equation}
We will usually use the abbreviations
\begin{equation}  \label{eq:normE_abbr_tracking}
\boldsymbol E = \boldsymbol E_{\zeta, \Qscr},\quad \norm{ \cdot }_{\zeta} = \norm{ \cdot }_{\zeta, \Qscr}.
\end{equation}

The following theorem then provides bounds for the perturbative free energy. 

\begin{theorem}   \label{th:pertcomp_E}
 Fix the spatial dimension $d$, the number of components $m$, the range of interaction $\Ran$,
 the set of multiindices 
 $\{e_1,\ldots,e_d\}\subset\Ical \subset\{\alpha\in \mathbb{N}_0^d\setminus \{(0,\ldots,0)\}:\;  \abs{\alpha}_\infty\leq \Ran\}$,
 real constants  $\omega_0 > 0$, $\zeta \in (0,1)$ and  an integer $r_0 \ge 3$.
 For $\Kcal \in \boldsymbol E$,  let $\Wscr_{\!N}(\Kcal,\Qscr)$ be defined by 
 \eqref{eq:Zpertcomp_general} and \eqref{eq:defvarsigmaN_general}.
  
Then there exist  $L_0\in \mathbb{N}$ such that for every odd integer $L \ge L_0$ 
there exists a positive constant $\rhoMT$   with the following properties. 
For any integer $N \ge 1$ and any quadratic form   $\Qscr$ on $\Gcal= (\R^m)^\Ical$  that  satisfies the bounds 
\begin{equation}
\label{eq:Qlowerbound_again}
\omega_0 \abs{z}^2 \leq \Qscr(z)\leq \omega_0^{-1}\abs{z}^2\  \text{ for all }\   z\in\Gcal,
\end{equation} 
the map $ \Kcal \mapsto \overline\Wscr_{\!N}(\Kcal)$ defined as $\overline{\Wscr}_N(\Kcal)=\Wscr_{\!N}(\Kcal, \Qscr)$   
is $C^\infty$ in $B_\rhoMT(0) \subset \boldsymbol E_{\zeta, \Qscr}$ and its derivatives are bounded independently of $N$, i.e., 
\begin{equation}  
\label{eq:bound_D_W_K}
\frac1{\ell!} \norm{ D^\ell \overline{\Wscr}_N(\Kcal)(\dot \Kcal, \ldots, \dot \Kcal) } \le 
C_\ell\, \norm{ \dot \Kcal}^{\ell}_{\zeta, \Qscr} \text{ for all } \Kcal \in B_\rhoMT(0) \text{ and } \ell \in \mathbb{N}.
\end{equation}
In particular there exist $\overline\Wscr\in C^r(B_\rhoMT(0))$ and a subsequence $N_n\to \infty$ such that $\overline\Wscr_{N_n}$ 
converges to $\overline\Wscr$ for all $r \in \mathbb{N}$    and the derivatives  of $\overline\Wscr$ are bounded as  in \eqref{eq:bound_D_W_K}.
\end{theorem}
This is the main technical theorem of the paper. 
The main steps of the proof will be summarised in Chapter~\ref{sec:explanation}  and the theorem
 will be eventually proven in Chapter~\ref{sec:proofs}. 

An immediate consequence of   Theorem~\ref{th:pertcomp_E} is that the function 
 $F\mapsto \Wscr_{\!N}(\Kcal_F,\Qscr)$ has the desired smoothness properties. 
 Here $\R^{m \times d}\ni F\mapsto \Kcal_F\in \boldsymbol{E}$   is a function that satisfies  suitable conditions and   $\Qscr$ is a fixed quadratic form.
 Specifically, we have the following result.

\begin{theorem}   
\label{th:pertcomp}    
Let $d$, $m$, $\Ran$, $\Ical$, $\omega_0$, $\zeta$, $r_0$, $L\geq L_0$, $\rhoMT$ be 
 as in Theorem~\ref{th:pertcomp_E}. Fix a quadratic from $\Qscr$ which satisfies \eqref{eq:Qlowerbound_again}.
Let $r_1 \ge 2$ be an integer.
Then  for each integer $N \ge 1$, each open set $\Ocal \subset \R^{m \times d}$
and any map $\Ocal \ni F\rightarrow \Kcal_F\in \boldsymbol{E}_{\zeta, \Qscr}$ of class $C^{r_1}$ that  satisfies the bounds	
\begin{align}
\sup_{F \in \Ocal} \norm{\Kcal_F}_{ \zeta, \Qscr}  & < \rhoMT, \\
\Theta := \sup_{F \in \Ocal} \sum_{\abs{\gamma} \le r_1}  \frac{1}{\gamma!}\norm{  \partial^\gamma_F \Kcal_F}_{ \zeta, \Qscr} &  < \infty,
\end{align}
the function $F \mapsto \Wcal_N(F):=\Wscr_{\!N}(\Kcal_F,\Qscr)$ is in $C^{r_1}(\Ocal)$, and the derivatives 

\noindent
 $\abs{\partial_F^\alpha \Wcal_N(F)}$, $\abs{\alpha} \le r_1$, can be bounded in terms of  
 $\Theta$.
 
In particular, there exists $\Wcal\in C^{r_1-1,1}(\Ocal)$ and a subsequence $N_n \to \infty$ such that $\Wcal_{N_n} \to \Wcal$ 
in  $C^{r_1-1}$, in $C^{r_1-1}(\Ocal)$ or in $C^{r_1-1}_{\rm loc}(\Ocal)$  and the derivatives of $\Wcal$ up to order $r_1-1$ 
as well as the Lipschitz constant of the $(r_1-1)$-st derivative are bounded in terms of $L$ and $\Theta$.
\end{theorem}

\begin{proof}
The claim follows from   Theorem~\ref{th:pertcomp_E} and the chain rule.
\end{proof}
\bigskip

Finally we address the scaling limit of the model.
This is a statement about  the Laplace transform of the measure with density
$$
\sum_{X\subset T_N}\prod_{x\in X} \Kcal(D\p(x))\,\mu_{\Qscr}(\d \p)/\Zscr(\Kcal,\Qscr,0).
$$
Recall that  $\boldsymbol{Q}$ is the operator associated to the quadratic form $\Qscr$ via \eqref{eq:defQ}.  
Moreover, the  operator $\boldsymbol{Q}^\nabla$ is obtained by restriction the quadratic form $\Qscr$ to
the subspace $\Gcal^\nabla$ which corresponds to linear maps or, equivalently to matrices in $\R^{m \times d}$,
see  \eqref{eq:define_bsQ_nabla} and  \eqref{eq:nabla_QU_components} 
for the definition of $Q^\nabla$  and its components $(\boldsymbol{Q}^\nabla_U)_{i,j;s,t}$.

\begin{theorem} \label{th:scalinglimit} 
Fix the spatial dimension $d$, the number of components $m$, the range of interaction $\Ran$,
the set of multiindices 
$$
\{e_1,\ldots,e_d\}\subset\Ical \subset\{\alpha\in \mathbb{N}_0^d\setminus \{(0,\ldots,0)\}:\;  \abs{\alpha}_\infty\leq \Ran\},
$$
real constants  $\omega_0 > 0$, $\zeta \in (0,1)$ and an integer $r_0 \ge 3$. 
Let $\mathbb{T}^d = \R^d/\Z^d$. 
For $f\in C^\infty(\mathbb{T}^d, \R^m)$ with $\int_{\mathbb{T}^d}f=0$ 
define $f_N\in \mathscr{V}_N$ by $f_N(x)=L^{-N\frac{d+2}{2}}f(L^{-N}x)$.
 
Then there exists  $L_0\in \mathbb{N}$ such that for every odd integer $L \ge L_0$ 
there exists a constant $\rhoMT  > 0$ with the following properties. 
For any quadratic form   $\Qscr$ on $\Gcal= (\R^m)^\Ical$  that  satisfies the bounds 
\begin{equation}
\label{eq:Qlowerbound_again2}
\omega_0 \abs{z}^2 \leq \Qscr(z)\leq \omega_0^{-1}\abs{z}^2\  \text{ for all }\   z\in\Gcal
\end{equation}
and any $\Kcal\in B_{\rhoMT}(0)\subset \boldsymbol{E}_{\zeta,\Qcal}$ there is a subsequence $N_\ell\to \infty$   
and a matrix $\boldsymbol{q}\in \R^{(m\times d)\times (m\times  d)}_{\mathrm{sym}}$ 
such that for all $f\in C^\infty(\mathbb{T}^d, \R^m)$, the partition function
 $\Zscr_N(\Kcal,\Qscr,f_N)$ defined by \eqref{eq:Zpertcomp_general} satisfies
\begin{equation}
\label{eq:scaling_limit_abstract}
\lim_{\ell\rightarrow \infty}  \frac{\Zscr_{N_\ell}(\Kcal,\Qscr,f_{N_\ell})}{\Zscr_{N_\ell}(\Kcal,\Qscr,0)}= e^{\frac1{2}(f,\mathscr{C}f)},                                                        
\end{equation}
where $\mathscr{C}: f\to \mathscr{A}^{-1}f$ is the inverse of the operator 
$\mathscr{A}$ acting on $u\in H^1(\mathbb{T}^d,\R^d)$ with $\int u=0$  by
\begin{equation}  
\label{eq:limiting_operator}
(\mathscr{A}u)_s=-\sum_{t=1}^m \sum_{i,j=1}^d (\boldsymbol{Q}^{\nabla}-\boldsymbol{q})_{i,j;s,t} \partial_i\partial_j u_t. 
\end{equation}
\end{theorem}

For the choice of the scaling and the fact that $\mathscr A$ depends only on the restriction  $\boldsymbol Q^\nabla$ of $\boldsymbol Q$  to $\Gcal^\nabla$, 
see  Remark~\ref{re:correlation_structure_concrete}.
Regarding the choice of subsequence, see Remark~\ref{re:subsequences}.
\bigskip

\section[New ideas in the proof of the abstract perturbation result]{The proof of the abstract perturbation result: \\
informal preview  and indication of new ideas}
\label{se:main_new_ideas}

In this section we give a brief and  informal  preview of the general strategy of 
the proofs, as well as  a short  description of the novelties of our approach.
A  summary of the concrete technical implementation of our  approach  is given  in Chapter~\ref{sec:explanation}.

Recall that our goal is to analyse 
 the integral in \eqref{eq:Zpertcomp_general}, namely
\begin{align}
\Zscr_{N}(\Kcal,\Qscr,f)=\int_{\Xcal_N}e^{(f,\p)}
\sum_{X\subset T_N}\prod_{x\in X}\Kcal(D\p(x))\,\mu_{\Qscr}(\d\p)                                                             
\end{align}
and, in particular, to  study the perturbative contribution to the free energy $\Wscr_{\!N}(\Kcal,\Qscr)=-L^{-Nd}{\ln\Zscr_{N}(\Kcal,\Qscr,0)}$ 
defined in \eqref{eq:defvarsigmaN_general}. 
Since we are interested in convexity properties, we need to control the second derivatives of $\Wscr_{\!N}(\Kcal,\Qscr)$ with respect to $\Kcal$. 

We first recall  a heuristic argument why a brute force analysis
 leads to a logarithmic divergence for the second derivatives of  $\Wscr$ and thus a multiscale analysis is required.
 
Then we briefly describe a particular multiscale strategy which has been developed and refined  over the last 30 years, 
in particular by Brydges and collaborators, and is the basis of our approach. 
This strategy involves three key ingredients:
\begin{itemize}
\item Finite-range decomposition of the Gaussian measure $\mu_\Qscr$.
\item Reblocking.
\item Extraction of relevant or marginal terms and fine-tuning of the initial Gaussian measure.
\end{itemize}

We close the Section with a description of  our key  modifications and refinements of this strategy: 

\begin{itemize}
\item A new finite range decomposition,
\item A new large field regulator, and
\item A modified reblocking scheme.
\end{itemize}
Readers familiar with the general multiscale analysis developed by Brydges and collaborators 
may directly move to the subsection on the  new finite range decomposition.

\subsection*{The need for a multiscale analysis}
The main interest of this paper lies in the exploration of convexity properties of the free energy as a function of the deformation $F$, 
i.e., bounds on its second derivative. By the chain rule applied to $F\to \Kcal_F$
 (combining \eqref{eq:defofKuVbeta} with \eqref{eq:defofUbar} as summarized at the beginning of  
 Section~\ref{sec:abstract_to_ggm} below) this boils down to  bounds on the expression
\begin{align}
D^2_{\Kcal} \Wcal(\Kcal,\Qscr)(\dot{\Kcal},\dot{\Kcal}).
\end{align}
 As an example,  let us consider the case $\Kcal=0$. The evaluation of the second derivative yields
\begin{align}
\begin{split}
D^2_{\Kcal} \Wcal(0,\Qscr)(\dot{\Kcal},\dot{\Kcal})
&=L^{-Nd}\biggl(\sum_{x,y\in T_N} \int_{\Xcal_N} \dot{\Kcal}(D\p(x))
\dot{\Kcal}(D\p(y))\, \mu_\Qscr(\d \p)
\\
&\quad - \sum_{x,y\in T_N}\int_{\Xcal_N} 
\dot{\Kcal}(D\p(x))\, \mu_{\Qscr}(\d \p)
\int_{\Xcal_N} 
\dot{\Kcal}(D\p(y))\, \mu_\Qscr(\d \p)\biggr)
\end{split}
\end{align}
 Assuming further that $\dot{\Kcal}$ is centred with respect to 
 $\mu_{\Qscr}$  we find the simpler expression
 \begin{align}
\begin{split}
D^2_{\Kcal} \Wcal(0,\Qscr)(\dot{\Kcal},\dot{\Kcal})
&=L^{-Nd}\sum_{x,y\in T_N} \int_{\Xcal_N} \dot{\Kcal}(D\p(x))
\dot{\Kcal}(D\p(y))\, \mu_\Qscr(\d \p)
\\
&=\sum_{x\in T_N} \int_{\Xcal_N} \dot{\Kcal}(D\p(x))
\dot{\Kcal}(D\p(0))\, \mu_\Qscr(\d \p)
\end{split}
\end{align}
where we used translation invariance in the second step.
While our setting is very general, the key difficulty is already present when 
$\mu_\Qscr(\d \p)$ is the scalar ($m=1$) discrete massless Gaussian free field, i.e., 
the operator $\Qscr=\Delta$ is the discrete lattice Laplacian $\Delta$.
 It is well known that the gradient of the discrete Gaussian free field decays critically, i.e.,
\begin{align}
\mathbb{E}_{\mu_{\Delta}} (\nabla_i\p(x)\nabla_j \p (y))=
\int_{\Xcal_N} \nabla_i\p(x)
\nabla_j\p(y)\, \mu_{\Delta}(\d \p)
\sim |x-y|^{-d}
\end{align}
and a similar statement holds for general $\Qscr$ satisfying \eqref{eq:Qlowerbound_again}.
Therefore the best possible general bound for the expression
$\mathbb{E}_{\mu_\Qscr} (\dot{\Kcal}(D\p(x))\dot{\Kcal}(D \p (y)))$ is also of order $|x-y|^{-d}$. We conclude that then for a 'typical' $\dot\Kcal$,
\begin{align}
\sum_{x\in T_N} \left|\int_{\Xcal_N} \dot{\Kcal}(D\p(x))
\dot{\Kcal}(D\p(0))\, \mu_\Qscr(\d \p)\right|
\gtrsim \sum_{x\in T_N\setminus \{0\}} |x|^{-d}\to \infty
\end{align}
as $N\to \infty$.  
This shows that the long range correlations of the Gaussian fields prevent a brute force analysis and,  
instead, we need to carefully exploit cancellations of the terms.

\subsection*{Finite range decomposition of the Gaussian measure}
To avoid the logarithmic divergence we use a multiscale analysis and begin by 
decomposing  the Gaussian field as a sum of Gaussian fields that capture the correlations on specific scales. 
The key ingredient  is the convolution  property of Gaussian random variables, namely
if $\p_1\sim N(0, \Cscr_1)$ and $\p_2\sim N(0, \Cscr_2)$ are independent Gaussian fields with covariance operators 
$\Cscr_1$ and $\Cscr_2$ then $\p_1+\p_2\sim N(0, \Cscr_1+\Cscr_2)$ is again a Gaussian field
and its covariance is the sum of covariances of $\p_1$ and $\p_2$.
In terms of integrals this identity can be expressed using $\Ascr_1=\Cscr_1^{-1}$, $\Ascr_2=\Cscr_2^{-1}$, and
$\Ascr_3=(\Cscr_1+\Cscr_2)^{-1}$ as follows
\begin{align}
\int f(\p)\, \mu_{\Ascr_3}(\d \p) =
 \int f(\p_1+\p_2)
\mu_{\Ascr_1}(\d \p_1)\, \mu_{\Ascr_2}(\d \p_2)
=
 \int f(\p)
\,(\mu_{\Ascr_1}\ast \mu_{\Ascr_2})(\d \p).
\end{align}
Here we used that $\p_i$ and $-\p_i$ follow the same distribution.
We consider a decomposition $\mu_\Qscr=\mu_1\ast \ldots\ast\mu_{N+1}$ of the Gaussian measure $\mu_\Qscr$ with $\mu_k$  capturing
 the correlation structure of $\mu_\Qscr$ on the length scale $\sim L^k$. 
We can then rewrite the integral in  \eqref{eq:Zpertcomp_general} as 
\begin{align}
\Zscr_{N}(\Kcal,\Qscr,0)=\int_{\Xcal_N}
\sum_{X\subset T_N}\prod_{x\in X}\Kcal\biggl(\sum_{k=1}^{N+1} D\p_k(x)\biggr)\,\mu_1(\d\p_1)\ast\ldots \ast\mu_{N+1}(\d\p_{N+1})                                                             
\end{align}
and integrate out the fluctuations scale by scale.
One possibility to obtain such a decomposition is to use a dyadic decomposition of the covariances in Fourier space,
i.e., the (discrete) Fourier transform of the kernel of the covariance operator $\Cscr_k$ of $\mu_k$ 
satisfies $\hat{\Ccal}_k(p)=\boldsymbol{1}_{p\in A_k} \hat{\Ccal}(p)$ where $\Ccal $ is the kernel 
generating the covariance operator $\Cscr$ of $\mu_\Qscr$ and $A_k=\{L^{-k-1}\leq |p|\leq L^{-k}\}$ is 
an annulus in the Fourier space. 
For a definition  and relevant properties of the discrete Fourier transform we refer to Section~\ref{sec:FRD}. 
The problem when working with such a decomposition is that, while $\mu_1$ essentially captures the correlation structure 
of $\mu_{\Qscr}$ on length scale $L$, it also weakly correlates points with larger distance. 
This causes substantial problems for the bookkeeping of the exponential number of terms. 
Note that the initial integrand given by 
\begin{align}\label{eq:product_structure}
\sum_{X\subset T_N} \prod_{x\in X} \Kcal(\nabla\p(x))
= \prod_{x\in T_N}(1+\Kcal(\nabla\p(x)))
\end{align}
contains $2^{|T_N|}$ terms but those can be expressed  as a product of $|T_N|$ terms. 
However, after integrating out the measure $\mu_1$ for a dyadic decomposition in Fourier space this structure is completely lost 
and the resulting combinatorial explosion complicates the control of the integral. 

Therefore, we  use  a more involved \emph{finite range decomposition} of the measure $\mu_\Qscr$. 
Finite range decompositions are again  decompositions  of the Gaussian $\p$ into a sum of Gaussian fields $\p_k$
with distributions $\mu_k$ that have the property that $\mu_k$ captures the correlation structure on length scale $L^k$. 
In addition, however, they satisfy $\mathbb{E}_{\mu_k}(\nabla_i\p(x)\nabla_j\p(x))=0$ for $|x-y|\gtrsim L^k$, 
i.e., the fields are independent for well separated points.
This implies that for $x,y$ such that $|x-y|\gtrsim L^k$
\begin{align}
\begin{split}
\int_{\Xcal_N} 
\Kcal(\nabla &(\p_k+\psi)(x))\Kcal(\nabla(\p_k+\psi)(y))\, \mu_k(\d\p_k)
\\&=
\int_{\Xcal_N} 
\Kcal(\nabla(\p_k+\psi)(x))\, \mu_k(\d\p_k)
\int_{\Xcal_N} 
\Kcal(\nabla(\p_k+\psi)(y))\, \mu_k(\d\p_k).
\end{split}
\end{align}

\subsection*{Reblocking}
It turns out that after integrating out the measures $\mu_k$ of a finite range decomposition we can retain a similarly simple structure 
as in \eqref{eq:product_structure} by collecting terms on length scale $L^k$.
We refer to  this process as \emph{reblocking}   since it bears some  resemblance to the block spin approach.
 The combination of integrating out one fluctuation field followed by a reblocking operation will be called \emph{renormalisation map} or renormalisation transform. 
The reblocking operation allows us to find simple expression for the functionals
\begin{align}
F_k(\p) =\int_{\Xcal_N}
\sum_{X\subset T_N}\prod_{x\in X}\Kcal\left(D\p(x) + \sum_{i=1}^{k} D\p_i(x)\right)\,\mu_1(\d\p_1)\ast\ldots \ast\mu_{k}(\d\p_{k})                                                              
\end{align}
that we obtain after integrating out the  first $k$ fluctuation fields.
 The actual implementation of the reblocking procedure and thus an explicit definition of the renormalisation transform can be found
in Section~\ref{sec:defRenormTransform}.

\subsection*{Extraction of relevant terms and fine-tuning}
The proper definition of the renormalisation map is slightly more  involved because  for naive definitions or the renormalisation map
 the natural  norm of the functionals  $F_k$
 will not be bounded uniformly in $k$. 
Instead, in every step we need to extract a finite number of low degree polynomials in the gradient field and introduce a separate bookkeeping for those terms. 
For gradient fields those terms are 
 \begin{align}\label{eq:relevant_terms}
 1, \,  \nabla^\alpha \p_i(x),\,
   \nabla^\beta \p_i(x) \nabla^\gamma \p_j(x)
 \end{align}
 where $1\leq |\alpha|\leq 1 + \lfloor d/2\rfloor$, $|\beta|=|\gamma|=1$ and $1\leq i, j\leq m$.
  In the renormalisation nomenclature they  correspond to the relevant (and marginal) operators. 
  The proper treatment of those terms is implicitly contained in the definition 
of the renormalisation map in Section~\ref{sec:defRenormTransform}.
However, the complete definition of the extraction operator
appears later in Section~\ref{se:projection_Pi2}. 

Finally we note that the Gaussian measure $\mu_\Qscr$ does not capture the long range behaviour of the model (e.g., the covariance structure)
because  the perturbation $\Kcal$ affects this in a non-trivial way (see also the statement of Theorem~\ref{th:scalinglimit}). 
In fact, the bookkeeping of the relevant operators  captures the effect of $\Kcal$ on the long range behaviour. 
Mathematically it is difficult to  track those terms explicitly. 
Instead, we employ a careful fine-tuning procedure replacing the  quadratic form $\Qscr$ in the  initial  integral \eqref{eq:Zpertcomp_general} 
with respect to the measure $\mu_{\Qscr}$ by a  replacement $\mu_{\Qscr'}$ judicially chosen so that 
$\mu_{\Qscr'}$ captures the model behaviour on large scales (i.e., the scaling limit). 
This is accomplished by making sure that the total contribution of the
relevant operators in \eqref{eq:relevant_terms} that we collect during the renormalisation steps
 explicitly vanishes as we integrate out all the fluctuations  from $\mu_1$  up to $\mu_N$.
 Eventually, the convenient operator $\Qscr'$   is found by a certain  fixed point argument (see Section~\ref{sec:finetuning}).

\subsection*{A new finite range decomposition}
We now highlight the two main differences of our approach  to previous works, namely 
 the use of improved finite range decompositions  and the construction of new large field regulators.
 We begin with the finite range decomposition.
 
One key difference with respect to earlier works is that we consider an anisotropic setting. 
The fine-tuning argument then requires us to carry out the whole analysis
not only for a single Gaussian measure $\mu_\Qscr$, but for a family of Gaussian measures from which we 
find the optimal choice $\mu_{\Qscr'}$ by a fixed point argument.
Therefore we need an entire family of finite range decompositions parametrized by the quadratic form $\Qscr$. 
In contrast, in  isotropic settings, e.g., for the $\p^4$-model, it is sufficient to consider a fixed finite range decomposition
 and simply rescale the fields. Suitable decompositions for the anisotropic setting  were constructed in \cites{AKM13,Bau13}. 
  One additional difficulty, which arises in the application of the fixed point argument,
  is that we need to  consider derivatives of the finite range decomposition with respect to the implicated quadratic forms $\Qscr$. 
  Those  derivatives, however, are unbounded for the decompositions from \cite{AKM13} and  \cite{Bau13}. 
In the precursor work  \cite{AKM16} this problem was solved by a rather involved analysis of fixed point arguments for maps with loss of regularity.
 Instead,  we rely here on the finite range decomposition constructed in \cite{Buc18} which is a small modification of the earlier constructions.  
  For this decomposition the derivatives with respect to $\Qscr$ exist and all our maps are differentiable without loss of regularity. 
  This allows us to avoid substantial technical difficulties mentioned above. 
  We review the properties of the finite range decomposition in Section~\ref{sec:FRD}. 
  This is then used as a crucial ingredient in the proof of the smoothness of the renormalisation map which can be found in Section~\ref{sec:smoothness}.
 
\subsection*{A new large field regulator}
An important technical innovation of this work is the construction of a new large field regulator. 
While providing a more detailed explanation and motivation of our construction in Section~\ref{sec:weights}, here we only give a brief overview.
 Generally the  \emph{large field problem}  refers to the problem of bounding the contribution of large fields $|\p|$ 
 to a functional integral $\int F(\p)\,\mu(\d \p)$. 
 Typically, the functionals $F$ are unbounded for large fields $|\p|$ but this should cause no problems because the measure $\mu$ is rapidly decaying.
 However, a formal treatment of this issue can be challenging. 
 Similar to previous works we consider norms that essentially boil down to 
 \begin{align}
 |F_k|_k = \sup_{\p} \frac{|F_k(\p)|}{w_k(\p)},
 \end{align}
 i.e., a weighted $L^\infty$ norm where the large field regulator $w_k$ determines the allowed growth behaviour as $|\p|\to \infty$ 
 (the actual norm on $F_k$ is more complicated but this does not matter here).
 The initial weight $w_0(\p)$ determines the class of perturbations $\Kcal$ that can be handled ($\Kcal$ can grow at most as fast as $w_0$)
 and therefore the results are more general the larger $w_0$ is, while the renormalisation analysis becomes more involved.
 Previously (see \cite{AKM16}) the large field regulators $w_k$ were carefully designed
 and complicated expressions which involved a large constant $h$, e.g.,
 \begin{align}
 w_k(\p)\propto \exp\left({h^{-2} \left(W_1(\nabla\p)+W_2(\nabla^2\p)+W_3(\nabla^3 \p)\right)}\right)
 \end{align}
 for some polynomials $W_1$, $W_2$, and $W_3$.
The key observation here is that we can instead choose
\begin{align}
w_0(\p) \propto \exp\left({\frac{(1-\varepsilon) (\p, \Qscr\p)}{2}}\right)
\end{align}
for any positive constant $\varepsilon>0$ and then essentially define $w_{k+1}\propto w_k\ast \mu_{k+1}$.
It is clear that this choice is optimal (up to $\varepsilon$) as the functionals $F$ must be integrable with respect to $\mu$. 
We emphasize that this has the consequence that we just need to assume any quadratic lower bound 
$\omega|z|^2$ for the potential $\Uscr$ in \eqref{eq:V3_new} while earlier approaches would require the lower bound 
$(1-\varepsilon)\Qscr_\Uscr(z)/2$ for some small $\varepsilon>0$, i.e., 
the potential would have to grow essentially as fast at $\infty$ as the quadratic Taylor expansion at the origin.
To clarify the difference we consider the example of gradient interface models where the Hamiltonian is of the form 
\begin{align}
\HN(\p) = \sum_{x\in T_N} \sum_{i=1}^d V(\nabla_i \p(x)).
\end{align}
The  interaction potential $V$ is often decomposed as follows $V(s)=s^2/2 + W(s)$ with $W(0)=W'(0)=W''(0)=0$.
The results in \cite{AKM16} then only cover $W(s)\geq -\varepsilon s^2/2$, i.e.,
$V(s)\geq (1-\varepsilon) s^2/2$ for some small
$\varepsilon >0$ while here we can handle all $W$ such that
$W(s)\geq -(1-\varepsilon) s^2$, i.e., $V(s)\geq \varepsilon s^2$ for arbitrarily small $\varepsilon >0$.
This improvement is of particular  importance for discrete elasticity
where we can handle a very general class of potentials in \gray{the} Theorems~\ref{T:deW} and \ref{T:descaling}. 
On the other hand it is unclear whether there are any potentials of interest  that satisfy the stronger growth conditions that were required in earlier works.

\subsection*{Further modifications}
We use a slightly different reblocking scheme  where we assign the contribution of small polymers (see Section~\ref{se:polymers} 
for definitions regarding the reblocking) to a cube on the next scale while previous works distributed their contribution to several  polymers. 
This simplifies  working out an explicit definition of the renormalisation map in Section~\ref{sec:defRenormTransform}.

Finally, we  treat \emph{vector valued fields} and \emph{general finite range interaction potentials} 
instead of just nearest neighbour interactions, but this is mostly a matter of additional notation.

\section[From abstract perturbations to generalized gradient models]
{The abstract perturbation theorem implies the results \\ for the generalized gradient models}  \label{sec:abstract_to_ggm}

The connection between the concrete conditions on 	$\Uscr$ used in Theorems~\ref{thm:strictconvexity} 
and~\ref{th:scalinglimit_concrete} for generalized gradient models and the abstract results in Theorems~\ref{th:pertcomp_E}, \ref{th:pertcomp},    
and~\ref{th:scalinglimit} is given by the following result.
Recall from  \eqref{eq:defofKuVbeta}  that  $\Kcal_F=\Kcal_{F,\beta, \Uscr}$ is defined as 
\begin{align}
\Kcal_F(z) = \Kcal_{F,\beta,\Uscr}(z)=
\exp\bigl(-\beta\overline{\Uscr}(\frac{z}{\sqrt{\beta}},F)\bigr)-1, 
\end{align}
 with 
\begin{equation}
\label{eq:defofUbar_repeat}
\overline{\Uscr}(z,F)=\Uscr(z+\overline{F})- \Uscr(\overline{F})- D\Uscr(\overline{F})(z)-\frac{\QU(z)}{2},
\end{equation}
where $\Qscr_\Uscr(z) = D^2 \Uscr(0)(z,z)$,
recall  \eqref{eq:defofUbar}.

\begin{proposition}  
\label{prop:embedding}                      
Let $r_0\geq 3$ and $r_1 \ge 0$ be integers and assume that 
\begin{equation} \label{eq:V1}   
 \Uscr\in C^{r_0+r_1} (\Gcal).   
\end{equation}
and
\begin{equation}  \label{eq:V1bis}
\omega_0 \abs{z}^2 \le \Qscr_\Uscr(z) \le \omega_0^{-1} \abs{z}^2
\end{equation} 
for some $\omega_0\in(0,1)$.
Let  $0<\omega \le \frac{\omega_0}{8}$ and suppose  that $\Uscr : \mathcal G \to \R$ satisfies the additional conditions
\begin{equation}
\label{eq:V3}     \Uscr(z) - D\Uscr(0) z - \Uscr(0)  \geq  \omega\abs{z}^2   \text{ for all } z \in \mathcal G,   \text{and} 
\end{equation}                 
\begin{equation}
\label{eq:V4}   \lim_{t \to \infty} t^{-2} \ln \Psi(t) = 0 
\text{ where }  
\Psi(t) := \sup_{\abs{z} \le t}   \sum_{3 \le \abs{\alpha} \le r_0+r_1}   \frac{1}{\alpha!} \abs{\partial^\alpha \Uscr(z)}.
\end{equation}

Then there exist  $\tilde \zeta$ (depending on $\omega$ and $\omega_0$), 
$\delta_0 > 0$ (depending on $\omega$, $\omega_0$ and  $\Psi(1)$), 
 $C_1$ (depending on $\omega$, $r_0$ and the function $\Psi$)
and $\Theta$ (depending on  $\omega$, $r_0$, $r_1$ and the function  $\Psi$)   
such  for all $\delta \in (0,  \delta_0]$  and all $\beta \ge 1$
the map $B_\delta(0)  \ni F \mapsto \Kcal_F = \Kcal_{F,\beta, \Uscr} \in \boldsymbol E = \boldsymbol E_{\tilde \zeta}$ is $C^{r_1}$ and satisfies
\begin{equation}  \label{eq;Kscr_smallness}
 \norm{\Kcal_F}_{\tilde \zeta, \Qscr_\Uscr} \le C_1 (\delta + \beta^{-1/2})
\end{equation}
as well as
\begin{equation} \label{eq;Kscr_bound_derivatives} 
\sum_{\abs{\gamma} \le r_1} \frac{1}{\gamma!} \,   \norm{ {\partial}^\gamma_F   \Kcal_F}_{\tilde \zeta,\Qscr_\Uscr}  \le \Theta.
 \end{equation}
In particular, given $\rhoMT >0$, there exists $\delta > 0$ and $\beta_0 \ge 1$ (both depending on  $ \omega$, $ \omega_0$,  and the function $\Psi$)
 such that, for all $\beta \ge \beta_0$ and all $F \in B_\delta$, we have \eqref{eq;Kscr_bound_derivatives} and 
\begin{equation}  
\label{eq;Kscr_bound_derivatives_rho}
\norm{\Kcal_{F}}_{\tilde \zeta,\Qscr_\Uscr} \le \rhoMT.
\end{equation}
\end{proposition}

The proof is  more or less straightforward, but a bit lengthy and we postpone it to  the end of this chapter,
Section~\ref{sec:emb}.
The proof shows
that we may take $\tilde \zeta = \frac{\omega \omega_0}{2}$.
Explicit expressions for $\delta_0$, $C_1$, and $\Theta$ are given in   \eqref{eq:define_delta_0},  \eqref{eq:define_C1_Kscr},  
and  \eqref{eq:define_Theta_Kscr},  respectively. 
The proof also shows that the dependence of $\delta_0$ and $C_1$ on the number of derivatives of $\Uscr$ can be slightly improved. 
If we set $\Psi_r(t) := \sup_{\abs{z} \le t}  \sum_{3\le \abs{\alpha} \le r}   \frac{1}{\alpha!} \abs{\partial^\alpha \Uscr(z)}$,
then $\delta_0$ depends on $\omega$, $\omega_0$, and $\Psi_3(1)$, 
while  $C_1$ depends on $\omega$, $r_0$, and the function $\Psi_{r_0}$.
\footnote{Recall our convention that for multiindices $\abs{\alpha}=\abs{\alpha}_1$.}

\begin{remark}    \label{re:singular_perturbation} \hfill
\begin{enumerate}[1.,leftmargin=0cm]
\item
 We may assume without loss of generality that $\Uscr(0) = D\Uscr(0) = 0$ since both, the Mayer function $\overline \Uscr$, 
and assumptions of the proposition, are invariant under adding an affine function to $\Uscr$. 
The lower  growth assumption \eqref{eq:V3} is then much weaker than the corresponding condition in \cite{AKM16}. 
Indeed,  assumption   \eqref{eq:V3} only requires any quadratic bound from below while in \cite{AKM16} the condition 
$ \Uscr(z)\geq \frac{\Qscr(z)}{2}-\varepsilon \abs{z}^2$ for some small $\varepsilon> 0$ was imposed,  
i.e., the potential was assumed to grow almost as fast as the quadratic approximation at $0$. 
\item
Let us emphasize that we do not require that $0$ is a minimum of the potential.
In fact, the theorem remains valid by replacing $0$ by  $z_0$ in all four assumptions, once $z_0$ is such that
\begin{align}
\Uscr(z)-D \Uscr(z_0)(z-z_0)-\Uscr(z_0)\geq \omega \abs{z-z_0}^2
\end{align}
for some $\omega\geq 0$.     
\item
The proposition can be generalized to some singular potentials,
e.g., it is possible to consider potentials $\Uscr+\Vscr$ where $\Uscr$ is as before and $\Vscr:\mathcal{G}\to \R\cup \{\infty\}$ satisfies $\Vscr(z)\geq 0$, 
$0\notin \mathrm{supp}\, \Vscr$, and $e^{-\Vscr}\in C^{r_0+r_1}$. 
The one dimensional potential $\Vscr(x)=\eta(x) \abs{x-2}^{-1}$ where $\eta\in C^\infty_c((1,\infty))$ satisfies $e^{-\Vscr(x)}\in C^\infty(\R)$, 
hence non-trivial examples for such potentials with singularities exist.

Let us briefly indicate the necessary extensions to prove this result.
Suppose that $\varepsilon >0$ is chosen such that
 $\dist (0,\mathrm{supp}\,\Vscr)\geq \varepsilon$.
We choose $\delta_0\leq \varepsilon/2$.
On the complement of the support of $\overline\Vscr$ we can argue as in the proof of Proposition~\ref{prop:embedding} below. 
If $(z/\sqrt{\beta},F)$ is in the support of  $\overline{\Vscr}$ and $\abs{F}\leq \delta_0$ we conclude that $\abs{z}\geq \varepsilon \sqrt{\beta}/2$. 
In this regime we use the estimate
\begin{align}
\abs*{ e^{-\beta \overline \Vscr (\frac{z}{\sqrt{\beta}},F)-\beta \overline \Uscr (\frac{z}{\sqrt{\beta}},F)}-1}_{T_{z,F}}
 \leq
  \abs*{ e^{-\beta \overline \Vscr (\frac{z}{\sqrt{\beta}},F)}}_{T_{z,F}} \abs*{ e^{-\beta \overline \Uscr (\frac{z}{\sqrt{\beta}},F)}}_{T_{z,F}} +\abs{1}_{T_{z,F}}
\end{align}
where $\abs{\cdot}_{T_{z,F}}$ is defined in \eqref{eq:defT_zFnorm} below.
Then the first term can be controlled by the assumption on $\Vscr$ and the second term is bounded in \eqref{eq:finalDerivativeEstimate}.
The condition $\abs{z} \geq \varepsilon \sqrt{\beta}/2$ implies that 
when multiplied with the weight of the $\lVert\cdot\rVert_{\boldsymbol{E}}$ norm both summands are exponentially small in $\beta$.
\end{enumerate}
\end{remark}

Combining Proposition \ref{prop:embedding} and Theorem \ref{th:pertcomp}  we get the following result.
\begin{theorem}   \label{th:bound_Wcal}
 Assume that $\Uscr$ satisfies  \eqref{eq:V1}-- \eqref{eq:V4}. Then there exists and $L_0$ such that for every odd integer $L \ge L_0$
there exist positive $\delta_1$ and $\beta_1 $  such that for $\beta \geq \beta_1$ the functions $\Wcal_{\beta,N}^\Uscr$ 
defined in \eqref{eq:defvarsigmaN} are  $C^{r_1}(B_{\delta_1}(0))$ and their $C^{r_1}$ norm is bounded uniformly for all  $N \ge 1$ and $\beta \ge \beta_1$.
\end{theorem}

\begin{proof}[Proof of Theorem~\ref{thm:strictconvexity}]
Recall from \eqref{eq:sigmaexprN} that
\begin{align}
\label{eq:sigmaexprN_repeat}
W_{\!\beta, N}(F)= & \, \Uscr(\overline{F})+\frac{\Wcal_{\beta,N}(F)}{\beta}-     \frac{1}{\beta L^{Nd}}\ln  Z_{\beta,N}^{\QU}
\end{align}
and note that the last term is independent of $F$. 
It follows from Theorem~\ref{th:bound_Wcal} that there is a constant $\Xi>0$ independent of $\beta$ and $\delta$ 
such that  $|D^2\Wcal_{\beta,N}(\dot{F},\dot{F})|\leq \Xi |\dot{F}|^2$ in $B_\delta(0)$ for $\beta\geq \beta_1$ and $\delta\leq \delta_1$.
The bound \eqref{eq:V4} on the third derivative of $\Uscr$ implies that there is a $\delta_2>0$ such that, for $\delta\leq \delta_2$ and $F \in B_\delta(0)$,
$$
\abs{D^2 \Uscr(\overline F)(z,z) -\Qscr_\Uscr(z)} 
=  \abs{D^2 \Uscr(\overline F)(z,z) - D^2 \Uscr(0)(z,z)} \le \frac{\omega_0}{2} \abs{z}^2
$$
and thus
\begin{equation}
D^2 \Uscr(\overline{F}) \geq \frac{\omega_0}{2} \abs{z}^2.
\end{equation}
Let $\beta_2=4 \Xi/ \omega_0$. 
Then for $\beta\geq \max(\beta_1,\beta_2)$, $\delta\leq \min(\delta_1,\delta_2)$, and $F\in B_\delta(0)$,
\begin{align}
D^2W_{\!\beta,N} (F)(\dot{F},\dot{F}) &=
D^2 \Uscr(\overline{F})(\overline{\dot{F}},\overline{\dot{F}})
+\frac{D^2\Wcal_{\beta,N}( F)}{\beta}(\dot{F},\dot{F})\\
& \geq \frac{\omega_0}{2} \abs{\overline{\dot{F}}}^2 -\frac{\Xi}{ 4\Xi/  \omega_0} 
\abs{\dot{F}}^2 \geq \frac{\omega_0}{4}  \abs{\dot{F}}^2.\notag
\end{align}
The assertion for the limit  $W_{\!\beta},$ follows from the fact  that the pointwise limit of uniformly convex functions is uniformly convex. 
\end{proof}

\begin{proof}[Proof of Theorem~\ref{th:scalinglimit_concrete}]
Combining \eqref{eq:initialfinal}, \eqref{eq:Zpertcomp}, and \eqref{eq:Zpertcomp_general}, we get
\begin{equation}
\mathbb{E}_{ \gamma_{  \beta,N  }^{F,\Uscr} }e^{(f_{N},\p)}
=\frac{Z_{\beta,N}(F,f_N)}{Z_{\beta,N}(F,0)}\Zcal_{\beta,N}
=\frac{\Zcal_{\beta,N}(F,\frac{f_N}{\sqrt{\beta}})}{\Zcal_{\beta,N}(F,0)}
=\frac{\Zscr_N(\Kcal_{F,\beta,\Uscr},\Qscr_{\Uscr},\frac{f_N}{\sqrt{\beta}})} {\Zscr_N(\Kcal_{F,\beta,\Uscr},\Qscr_{\Uscr},0)}.
\end{equation}
The assumptions ensure that Theorem~\ref{th:scalinglimit} can be applied which implies the claim.
\end{proof}

\section{Embedding of the initial perturbation} \label{sec:emb}

In this section we prove Proposition~\ref{prop:embedding} which shows that the
conditions imposed in the setting of generalized gradient models are sufficient to apply the abstract
perturbation theory in Section~\ref{sec:main_abstract}.
The proof is  more or less straightforward, but a bit lengthy, and may be skipped on first reading.

\begin{proof}[Proof of Proposition~\ref{prop:embedding}]
The main point is to obtain the additional factor $\beta^{-1/2} + \delta$ in \eqref{eq;Kscr_smallness}
which can be made as small as desired by taking $\delta$ small and $\beta$ large. 
This factor essentially comes from the third order Taylor expansion. 
We may assume that $\Uscr(0) = D\Uscr(0) = 0$ since the second and higher order derivatives of $\Uscr$
(and thus also the function $\overline \Uscr$) and the assumptions in 
Proposition~\ref{prop:embedding} are invariant under addition of an affine function to $\Uscr$.
The rest of the argument is then essentially an exercise in estimating polynomials and their exponentials.
Observe that, for functions $f \in C^{r_0}(\mathcal G)$, the norms $\abs{ f}_{T_{z}}$ introduced in Appendix A amount to
$$
\abs{ f}_{T_{z}} = \sum_{\abs{\alpha} \le r_0} \frac{1}{\alpha!} \abs*{ \partial^\alpha_z f(z)}
$$
(see Example~\ref{Ex:T0norm} and equation~\eqref{eq:T0_norm_for_Rp_new}).

The proof of Proposition~\ref{prop:embedding} can be split into the following steps:
\medskip

\begin{step} 
For any $f \in C^{r_0}(\mathcal G)$  we have
\begin{equation} 
\label{eq:exp_bound_Kscr_1}
\abs{e^f }_{T_z} \le  e^{f(z)}   (1 + \abs{f}_{T_z})^{r_0}
\end{equation}
and
\begin{equation} 
\label{eq:exp_bound_Kscr_2}
\abs{e^f -1 }_{T_z} \le  \max(e^{f(z)},1)    (1 + \abs{f}_{T_z})^{r_0}  \, \abs{f}_{T_z}.
\end{equation}
\end{step} 
\medskip

We first note that for $f_1, f_2 \in C^{r_0}(\mathcal G)$ we have  $\abs{ f_1 f_2}_{T_z} \le \abs{ f_1}_{T_z} \, \abs{f_2}_{T_z}$. 
This follows abstractly from  Proposition~\ref{pr:product_estimate_taylor} and Example~\ref{Ex:T0norm} in the appendix. 
Alternatively one can easily verify this by a direct calculation using that the (truncated) product of Taylor polynomials 
is the Taylor polynomial of the product. 
To prove \eqref{eq:exp_bound_Kscr_1} we  set $\tilde f(y) = f(y)- f(z)$. 
Then  $ e^{f(y)} = e^{f(z)} e^{\tilde f(y)}$. 
Since $\tilde f(z) = 0$ the $r_0$-th order Taylor polynomial of $e^{\tilde f}$ at $z$ agrees with the Taylor polynomial 
of $\sum_{m=0}^{r_0}  \frac{1}{m!}  (\tilde f)^m$. 
By the triangle inequality and the  product property we get
$$
 \abs{ e^{\tilde f}}_{T_{z}} \le \sum_{r=0}^{r_0}  \frac{1}{r!}  \abs{\tilde f}_{T_{z}}^r\le (1 + \abs{ \tilde f}_{T_{z}})^{r_0} 
\le  (1 + \abs{f}_{T_{z}})^{r_0}. 
$$
This finishes the proof of  \eqref{eq:exp_bound_Kscr_1}.
Now \eqref{eq:exp_bound_Kscr_2} follows from the identity
$$ 
e^f -1 = \int_0^1   e^{\tau f}   f \, \d \tau,
$$
Jensen's inequality, and the product property.
 
 We will now use the claims of Step 1 for 
 \begin{align}
 \label{eq:Ubeta}
 f(z) = \Uscr_\beta(z, F) = \beta \overline \Uscr( \frac{z}{\sqrt \beta}, F).
 \end{align}
\medskip

\begin{step} 
For $\beta \ge 1$ and $\abs{F} \le \delta \le 1$, we have
\begin{equation}  \label{eq:pointwise_bound_U_beta}
\abs{\Uscr_\beta( \cdot, F)}_{T_z} \le  (\beta^{-1/2} +  \delta)  \widetilde  \Psi(\abs{z}) 
\end{equation}
where 
\begin{equation}
\widetilde \Psi(t)  :=  3(1 + t)^3  \,   \Psi(t+1).    
\end{equation}
\end{step} 
\medskip

Actually, we show a slightly  stronger bound,
\begin{equation} \label{eq:pointwise_bound_U_beta_optimal}
\abs{\Uscr_\beta( \cdot, F)}_{T_z}  \le \bigl[ 3 (1+\abs{z}_\infty^2) \abs{\overline F} + (1+\abs{z}_\infty)^3 \beta^{-1/2}  \bigr]  
\, \Psi\bigl( \tfrac{|z|}{\sqrt \beta} + \delta\bigr).
\end{equation}

Let us remark that in this proof $D$ refers as usual to  total derivatives and $\partial$ to partial derivatives.
Without reference to $z$ or $F$, the derivatives $\partial \Uscr$ (or $\partial_i \Uscr$) 
and $D \Uscr$ refer to the derivatives of the function $\Uscr$ evaluated at the corresponding argument, 
while $\partial_{z_i} \Uscr(\frac{z}{\sqrt{\beta}}+\overline F)$ and $D_z \Uscr(\frac{z}{\sqrt{\beta}}+\overline F)$ 
refer to the derivatives of the map $z\to  \Uscr(\frac{z}{\sqrt{\beta}}+\overline F)$.
Clearly $\partial_{z_i} \Uscr(\frac{z}{\sqrt{\beta}}+\overline F)=\frac{1}{\sqrt{\beta}}\partial_i \Uscr(\frac{z}{\sqrt{\beta}}+\overline F)$ and
$ \partial_{F_i} \Uscr(\frac{z}{\sqrt{\beta}}+\overline F)=\partial_i \Uscr(\frac{z}{\sqrt{\beta}}+\overline F)$.

For derivatives of the  third or higher  order we use that
$$
\partial^\alpha_z \Uscr_\beta(z, F)  = \beta^{1 - \frac{|\alpha|}{2} } \, \partial^\alpha \Uscr( \frac{z}{\sqrt \beta} + \overline F),  
$$
which yields
$$   
\sum_{3 \le |\alpha| \le r_0}  \frac{1}{\alpha!} |\partial^\alpha_z \Uscr_\beta(z, F)|      \le \beta^{-1/2} \, \Psi(\frac{\abs{z}}{\sqrt \beta} + \delta).
$$
To estimate the lower order terms, we use the third order Taylor expansion in $z$. This yields
\begin{align}
\Uscr_\beta(z, F) = \,  \frac12 D^2 \Uscr(\overline F)&(z,z) -\frac12 D^2 \Uscr(0)(z,z)\\ +& \beta^{-1/2} \int_0^1 \frac{(1-\tau)^2}{2}   
D^3 \Uscr(\overline F+\frac{\tau z}{\sqrt \beta} )(z,z,z) \, \d \tau ,  \notag
\end{align}
\begin{align}
D_z \Uscr_\beta(z, F)(\dot z) =  \, D^2 \Uscr(\overline F)&(z,\dot z) -D^2 \Uscr(0)(z,\dot z) \\+& \beta^{-1/2} \int_0^1  (1-\tau )  
D^3 \Uscr(\overline F+\frac{\tau z}{\sqrt \beta} )(z,z,\dot z) \, \d \tau , \notag
\end{align}
\begin{align}
D^2_z \Uscr_\beta(z, F)(\dot z_1, \dot z_2) =  \, D^2\Uscr(\overline F)&(\dot z_1,\dot z_2) - D^2 \Uscr(0)(\dot z_1,\dot z_2)\\+ &\beta^{-1/2} \int_0^1  
D^3 \Uscr(\overline F+\frac{\tau z}{\sqrt \beta})(z,\dot z_1,\dot z_2) \, \d \tau .\notag
\end{align}
Using further the bound  
\begin{equation}
\label{eq:D2F-D20}
\abs{D^2 \Uscr(\overline F)(\dot z_1, \dot z_2) - D^2 \Uscr(0)(\dot z_1, \dot z_2)} \le  \int_0^1 D^3 \Uscr(\tau \overline F)(\overline F, \dot z_1, \dot z_2) \, \d \tau 
\end{equation}
combined with
$$  
\frac{1}{3!} |  D^3 \Uscr(\tau \overline F)(\overline F, z, z) |\le
   \sum_{|\alpha|=3} \frac{1}{\alpha!}  \Big|\partial^\alpha \Uscr (\tau \overline F)\Big| \,  |z|_\infty^2 \, 
|\overline F|_\infty,
$$
as well as
\begin{align}
\frac{1}{3!}  \Big|D^3 \Uscr(\overline F+\frac{\tau z}{\sqrt \beta} )(z,z,z)\Big| \le&
\frac{1}{3!}  \sum_{i_1, i_2, i_3=1}^{\mathrm{dim} \mathcal G} |\partial_{i_1} \partial_{i_2} \partial_{i_3} \Uscr(\overline F+ \frac{\tau z}{\sqrt \beta} ) |\,  |z|_\infty^3  \notag\\
=&  
\sum_{|\alpha|=3}  \frac{1}{\alpha!} \Big|\partial^\alpha   \Uscr (\overline F+\frac{\tau z}{\sqrt \beta} )\Big| \,  |z|_\infty^3\notag
\end{align}
with   $\int_0^1 \frac{(1-\tau)^2}{2}d\tau=\frac{1}{3!}$, we deduce that
\begin{equation}  \label{eq:bound_sup_Ubeta}
|\Uscr_\beta(z, F)|  \le(3 |z|_\infty^2 \, |\overline F|_\infty + |z|_\infty^3   \beta^{-1/2} )\,  \Psi(\frac{|z|}{\sqrt \beta} + \delta).
\end{equation}
Reasoning similarly for the first and second derivatives of $\Uscr_\beta$, we obtain \eqref{eq:pointwise_bound_U_beta_optimal}. 
Since $|\overline F|_\infty \le |\overline F| = |F|$ we deduce  \eqref{eq:pointwise_bound_U_beta}.
  \medskip

\begin{step}  There exist $\delta_0 >0$ such that 
 \begin{equation}  \label{eq:U_beta_quadratic_below}
 -\Uscr_\beta(z, F) \le \frac12 \QU(z) - \frac\omega2 \, |z|^2   \text{ for all }  \beta \ge 1,  \, \, z \in \mathcal G, \, \, F \in B_{\delta_0}(0).
 \end{equation}
\end{step} 
 \medskip
 
 Using  the definitions \eqref{eq:defofUbar} and   \eqref{eq:Ubeta}, we need to show that 
$$
\beta\bigl( \Uscr(\overline{F}+\frac{z}{\sqrt\beta})-\Uscr(\overline{F})-D \Uscr(\overline{F})(\frac{z}{\sqrt\beta})  \bigr)
\ge  \frac\omega2 \, \abs{z}^2 .
$$ 
 For $F=0$ this follows directly from the assumption \eqref{eq:V3},
$$
 \beta\bigl( \Uscr(\frac{z}{\sqrt\beta})-\Uscr(0)-D \Uscr(0)(\frac{z}{\sqrt\beta})  \bigr)
\ge  \beta \omega \, \abs*{{\frac{z}{\sqrt\beta}}}^2 = \omega \, \abs{z}^2 \ge \frac\omega2 \, \abs{z}^2.
$$
 
 This can be extended to the case when $F$ is small if compared with $z/\sqrt{\beta}$.
On the other hand, if  $F$ is comparable or bigger than  $z/\sqrt{\beta}$, we can rely on the third order  Taylor expansion around 0.

Indeed,  consider first the case when $\frac{z}{\sqrt \beta}$ is large. 
Let  $\kappa := \frac{9}{\omega \omega_0} \ge 9$ and assume that $\frac{|z|}{\sqrt{\beta}}\geq \kappa \delta$ and $|\overline{F}| = |F| \leq \delta$.
The estimate \eqref{eq:V3} with the assumption $\Uscr(0) = D\Uscr(0) = 0$  implies
\begin{equation}
\label{eq:betaV}
\beta \Uscr(\overline{F}+\frac{z}{\sqrt{\beta}}) \geq                                                                                        
\omega \beta\bigl| \overline{F}+\frac{z}{\sqrt{\beta}} \bigr|^2 \ge  \omega \beta
\bigl( \frac{\abs{z}}{\sqrt{\beta}}-\abs{\overline{F}}\bigr)^2  \ge 
\omega\bigl(1-\frac1\kappa\bigr)^2 \abs{z}^2 .
\end{equation}
 For $z$ and $F$ as before and  using $D\Uscr(0)=0$, $D^2\Uscr(0)=\QU$, and the third order Taylor expansion, we bound
\begin{equation}
\beta\abs*{D\Uscr(\overline{F})(\frac{z}{\sqrt{\beta}})} \leq  
 \beta \abs*{D^2\Uscr(0)(\overline{F},\frac{z}{\sqrt{\beta}})}  + 
 \sup_{|\xi| \le |\overline{F}|}\frac{\beta}2  \abs*{  D^3\Uscr(\xi)(\overline{F},\overline{F},\frac{z}{\sqrt{\beta}}) }                                                                     
\end{equation}
Evaluating the first term as
\begin{align}
\beta\abs*{ D^2\Uscr(0)(\overline{F},\frac{z}{\sqrt{\beta}})}  \le& 
\beta \abs*{ D^2\Uscr(0)(\overline{F},\overline{F})}^{1/2}\abs*{ D^2\Uscr(0)(\frac{z}{\sqrt{\beta}},\frac{z}{\sqrt{\beta}})}^{1/2}\notag\\
\le&
\frac{\beta}{\omega_0} |F| \frac{|z|}{\sqrt{\beta}}  \le 
\frac{\beta}{\kappa \omega_0}  \bigl(\frac{|z|}{\sqrt{\beta}}\bigr)^2= \frac{|z|^2}{\kappa \omega_0}\notag
\end{align}
and the second term, assuming that $3 \Psi(1) \delta\le 1$,  as
$$
\sup_{|\xi| \le |\overline{F}|}\frac{\beta}2  \abs*{  D^3\Uscr(\xi)(\overline{F},\overline{F},\frac{z}{\sqrt{\beta}}) }\leq 
3 \beta \Psi(1) \delta \frac{1}{\kappa}  \abs*{\frac{z}{\sqrt \beta}}^2  \le \,  
 \frac{|z|^2}{\kappa} ,     
$$
we get the bound
\begin{equation}
\label{eq:betaV1}
\beta\abs*{D\Uscr(\overline{F})(\frac{z}{\sqrt{\beta}})} \leq   \bigl( 1+\frac1{\omega_0}  \bigr) \frac{\abs{z}^2}{\kappa}.
\end{equation}

Similarly,  assuming again that $\delta \le  \frac{1}{ 3 \Psi(1)}$,  we get
\begin{align}
\label{eq:betaV0}
\beta|\Uscr(\overline{F})|\leq &   
\beta \abs*{D^2\Uscr(0)(\overline{F},\overline{F})}  +  
\sup_{|\xi| \le |\overline{F}|}\frac{\beta}2  \abs*{  D^3\Uscr(\xi)(\overline{F},\overline{F},\overline{F}) }  \\  \le&
\beta\Bigl(\frac{\delta^2}{\omega_0} + 3\Psi(1)\delta^3\Bigr)\le (1+\frac1{\omega_0})\frac{|z|^2}{\kappa^2}.    \notag                                                             
\end{align}

Combining  the bounds \eqref{eq:betaV}, \eqref{eq:betaV1} and \eqref{eq:betaV0} imply \eqref{eq:U_beta_quadratic_below} once
$$
\bigl( 1+\frac1{\omega_0}  \bigr) \frac1{\kappa}(1+\frac1{\kappa})\le \omega\bigl(\bigl(1-\frac1\kappa\bigr)^2-\frac12\bigr).
$$
For this to be true, it suffices when	
$$
2(1+\frac1{\kappa})\le \kappa\omega_0\omega\bigl(\bigl(1-\frac1\kappa\bigr)^2-\frac12\bigr).
$$	
Indeed, with the choice $\kappa=\frac9{\omega\omega_0}\ge 9$, the left hand side is bounded from above by $2(1+1/9)=20/9$
while the right hand side from below by $9 ((8/9)^2-1/2)=  47/18>20/9$.

It remains to consider the case $|z|/ \sqrt \beta < \kappa \delta$.  
We choose 
 \begin{equation} \label{eq:define_delta_0}
  \delta_0 := 
\min\Big(\frac{1}{1+\kappa}, \frac{3\omega_0}{16 \kappa \Psi(1)}\Big) .
 \end{equation}
 With $|z|/ \sqrt \beta < \kappa \delta$ and $ \delta \le \delta_0$, we get  $\frac{|z|}{\sqrt \beta} + \delta \le (\kappa + 1)  \delta \le 1$ . 
 Hence,  from   \eqref{eq:bound_sup_Ubeta} with $|z|_\infty \le |z|$, $\kappa\ge 9$,   and assuming $\omega\le \frac{\omega_0}2$, we get
\begin{align}
 |\Uscr_\beta(z, F)| \le (3 + \kappa) \delta  \Psi(1)  |z|^2   \le&  \frac43 \kappa   \delta  \Psi(1)  |z|^2  \le   \frac14 \omega_0 |z|^2
\notag\\ \le & \frac12 (\omega_0 -   \omega) |z|^2 \le \frac12 \QU(z) - \frac12 \omega |z|^2.\notag
 \end{align} 
 Thus \eqref{eq:U_beta_quadratic_below} holds for this choice of $\delta_0$ and $|z|/ \sqrt \beta \le \kappa \delta$. 
 Finally for   $\delta \le \delta_0$ also the condition $\delta \le  \frac{1}{ 3 \Psi(1)}$ is satisfied and thus
 \eqref{eq:U_beta_quadratic_below} holds for all $z$ and all $F \in B_{\delta_0}(0)$. 
 \medskip

\begin{step}  Let  $0<\delta<\delta_0$ with $\delta_0 \le 1$  given by \eqref{eq:define_delta_0}. 
Then, with $\tilde\zeta=\frac{\omega \omega_0}{2}$, we have
\begin{equation} \label{eq:global_bound_U_beta}
\|  e^{- \Uscr_\beta(\cdot, F)}-1 \|_{\tilde\zeta, \Qscr_{\Uscr}}  \le C_1( \delta+\beta^{-1/2} ) \text{ for all } \beta \ge 1, \, \, F \in B_{\delta}(0).
\end{equation}
\end{step} 
 \medskip
 
Combining \eqref{eq:exp_bound_Kscr_2},  \eqref{eq:pointwise_bound_U_beta} and \eqref{eq:U_beta_quadratic_below}
and using that $\beta^{-1/2} + \delta \le 2$ we get

\begin{equation}   \label{eq:pointwise_exp_Ubeta_minus_1}
| e^{- \Uscr_\beta(\cdot, F)} -1 |_{T_z} 
 \le  e^{  \frac12 \QU(z) - \frac\omega2  |z|^2 }        (\beta_0^{-1/2} + \delta_0)\,   \widetilde  \Psi(|z|) \,  (1 +  2 \widetilde  \Psi(|z|))^{r_0}.
 \end{equation}
Given that $\frac12  \frac{\omega \omega_0}{2} \QU(z) \le \frac14 \omega |z|^2$ we have
\begin{equation} \label{eq:estimates_weights_norm_zeta}
e^{-\frac12 (1- \frac{\omega \omega_0}{2}) \QU(z)} \le e^{-\frac12 \QU(z)} \, e^{\frac14 \omega |z|^2}.
\end{equation}
 Thus  multiplying   \eqref {eq:pointwise_exp_Ubeta_minus_1} by the weight  $e^{-\frac12 (1- \frac{\omega \omega_0}{2}) \QU(z)}$  
 and setting  
\begin{equation} \label{eq:define_C1_Kscr}
C_1 = \sup_{t \ge 0}  e^{-   \frac{\omega}{4}  t^2 } \,  \widetilde \Psi(t) ( 1 +2  \widetilde \Psi(t))^{r_0}  < \infty  \quad
\hbox{with} \quad \widetilde \Psi(t)  =  3 (1 + t)^3  \,   \Psi(t+1), 
\end{equation}
we get  \eqref{eq:global_bound_U_beta}, thus completing Step 4.
\medskip

The estimates \eqref{eq:pointwise_exp_Ubeta_minus_1} and \eqref{eq:estimates_weights_norm_zeta} 
imply that the assumptions of Lemma \ref{le:criterion_diff_zeta} below hold. 
This shows that $F\to \Kcal_F$ is continuous. 
Together with the result of the previous step this ends the proof for $r_1=0$.

It remains to show the bound \eqref{eq;Kscr_bound_derivatives} for the derivatives with respect to $F$.
Considering  first  the case $|\gamma|=1$, we need to estimate
\begin{align}
 \frac{\partial}{\partial F_i} e^{- \Uscr_\beta(z, F)} = - e^{- \Uscr_\beta(z, F)} \, \frac{\partial}{\partial F_i} \Uscr_\beta(z, F).
 \end{align}
By the chain and product rules, the derivatives  $\partial^\alpha_z$ of this expression exist for $|\alpha| \le r_0$. 
Moreover by \eqref{eq:exp_bound_Kscr_1}
and the product property of the $| \cdot |_{T_z}$ norm,
\begin{equation} \label{eq:pointwise_bound_first_exp_Ubeta}
 \Big| \frac{\partial}{\partial F_i} e^{- \Uscr_\beta(\cdot , F)}  \Big|_{T_z} 
 \le e^{- \Uscr_\beta(z, F)} ( 1+ | \Uscr_\beta(\cdot, F)|_{T_z})^{r_0}
\, \Big|   \frac{\partial}{\partial F_i} \Uscr_\beta(z, F) \Big|_{T_z}.
\end{equation}
 Then it remains to bound 
$\Big|   \frac{\partial}{\partial F_i} \Uscr_\beta(z, F) \Big|_{T_z}$.

For the higher derivatives with respect to $F$ the combinatorics becomes more complicated. Therefore,
it is actually useful to introduce the  norm $\abs{\cdot}_{T_{z,F}}$ for Taylor polynomials in two variables 
(see Appendix \ref{se:polynomials_several_variables}),
\begin{align}
\label{eq:defT_zFnorm}
| f|_{T_{z,F}} := \sum_{|\alpha| \le r_0} \sum_{|\gamma| \le r_1}
\frac{1}{\alpha!} \frac{1}{\gamma!}  \big| \partial^\alpha_z \partial^\gamma_F  f(z,F) \big|.
\end{align}
Note that, in particular, the expression $\Big|   \frac{\partial}{\partial F_i} \Uscr_\beta(z, F) \Big|_{T_z}$ is controlled by this norm.
As a preparation we show an estimate similar to the result of Step 2 for the $|\cdot |_{T_{z,F}}$ norm of $\Uscr_\beta(z, F)$.
\medskip

\begin{step} For $\beta\geq 1$ and $|F|\leq 1$ we have
\begin{align}\label{eq:prop2.6Step5}
\begin{split}
 |\Uscr_{\beta}(z,F)|_{T_{z,F}} 
  &\leq 2^{r_0+r_1+1}(1+|z|)^3 \Psi(|z|+1).
 \end{split}
\end{align} 
\end{step} 
\medskip

To estimate the terms in the definition of $|\cdot|_{T_{z,F}}$ norm, we distinguish three cases depending on the order of derivatives.

For $|\gamma|=0$ we have shown in Step 2 that for $\beta\geq 1$ and $|F|\leq 1$,
\begin{align}\label{eq:T_zF1}
 \sum_{|\alpha| \le r_0} \frac{1}{\alpha!}   \big| \partial^\alpha_z   \Uscr_\beta(z,F) \big|
 =|\Uscr_{\beta}(\cdot,F)|_{T_z}\leq 6(1+|z|)^3 \Psi(1+|z|). 
\end{align}

For $|\gamma|\geq 1$ and $|\alpha| \geq 2$ we use that
$ | \partial^\alpha_z \partial^\gamma_{F} \Uscr_\beta(z, F) | = \beta^{1-|\alpha|/2}\big| \partial^{\alpha + \gamma} \Uscr(\frac{z}{\sqrt \beta} + \overline F)\big|$.
The  combinatorial identity 
\begin{align}
 \sum_{\alpha +\gamma = \kappa} \frac{1}{\alpha!} \frac{1}{\gamma!} = \frac{1}{\kappa!}  \sum_{\alpha +\gamma = \kappa} \frac{\kappa!}{\alpha! \gamma!} = \frac{1}{\kappa!} 2^{|\kappa|}
 \end{align}
then implies
\begin{align}\label{eq:T_zF2}
\begin{split}
 \sum_{\substack{2 \le |\alpha| \le r_0, \\
  1\leq  |\gamma| \le r_1}}  \frac{1}{\alpha!} \frac{1}{\gamma!} |  \partial^\alpha_z \partial^\gamma_F \Uscr_\beta(z, F) |
& \le 2^{r_0+r_1}  \sum_{3\leq |\kappa|\le r_0 + r_1} \frac{1}{\kappa!}  \Big| \partial^{\kappa} \Uscr(\frac{z}{\sqrt \beta} + \overline F)\Big|
\\ &
\le 2^{r_0+r_1} \Psi(|z| + 1).
\end{split}
\end{align}

For the terms with $\alpha = 0$ 
and $|\gamma|\geq 1$, one differentiates  with respect to $F$ the second order Taylor expansion of $\Uscr_\beta$ in the variable $z$,
\begin{align}\label{eq:secondOrderTaylorUbeta}
 \Uscr_\beta(z, F) = \int_0^1 (1-\tau)  D^2 \Uscr(\tau \frac{z}{\sqrt \beta}+\overline F)(z,z) \, d\tau - \frac{ D^2 \Uscr(0)(z,z)}{2}.
\end{align} 
 Using the identity
\begin{align}
 \sum_{|\gamma| = k} \frac{1}{\gamma!} |\partial^\gamma f(F)| = \frac{1}{k!}   \sum_{i_1, \ldots i_k=1}^{\mathrm{dim} \mathcal G} |  \partial^{i_1}  \ldots \partial^{i_k} f(F)|
\end{align}
valid for any $f\in C^k(\mathcal{G})$, we get
\begin{align}
\sum_{j_1,\dots, j_\ell}\sum_{\alpha:\abs{\alpha}=k}\frac{1}{\alpha!} |\partial_{j_1}\dots \partial_{j_\ell} \partial^\alpha f(z)|=
\frac{(k+\ell)!}{k!}\sum_{\overline\alpha:\abs{\overline\alpha}=k+\ell}\frac{1}{\overline\alpha!} |\partial^{\overline\alpha} f(z)|.
\end{align}
Hence \eqref{eq:secondOrderTaylorUbeta} implies
\begin{align}\label{eq:T_zF3}
 \sum_{1\leq |\gamma| \le r_1} \frac{1}{\gamma!} | \partial^\gamma_F  \Uscr_\beta(z, F) | \le  
 \frac{(r_1+2)!}{2 r_1!} 
 |z|_\infty^2  \Psi(|z|+1)  \le \frac{(r_1+2)^2}{2} |z|^2 \Psi(|z|+1)
\end{align}
Similarly, the Taylor expansion for the derivative, 
\begin{align}
D_z\Uscr_\beta(z,F)(\dot z) = \int_0^1  D^2 \Uscr(\tau \frac{z}{\sqrt \beta}+\overline F)(z, \dot z) \, \d \tau -   D^2 \Uscr(0)(z, \dot z),
\end{align}
implies that 
\begin{align}\label{eq:T_zF4}
 \sum_{|\alpha|=1}  \sum_{1 \leq |\gamma| \le r_1} \frac{1}{\gamma!} | \partial^\alpha_z \partial^\gamma_F  \Uscr_\beta(z, F) | \le 
 \frac{(r_1+2)!}{ r_1!} 
  |z|_\infty  \Psi(|z|+1) \le (r_1+2)^2   |z| \Psi(|z|+1).
\end{align}
Thus, combining \eqref{eq:T_zF1}, \eqref{eq:T_zF2}, \eqref{eq:T_zF3}, and \eqref{eq:T_zF4}
we obtain \eqref{eq:prop2.6Step5} since $(r_1+2)^2\leq 8\cdot 2^{r_1} \leq 2^{r_0+r_1}$.
\medskip

\begin{step} Derivatives with respect to $F$.
\smallskip

\noindent
Let $\delta_0$ and $\tilde\zeta$ be like in Step 4.
The map $B_{\delta_0}(0)\ni F \mapsto e^{-\Uscr_\beta(\cdot, F)}\in \boldsymbol E$ is $r_1$ times continuously differentiable
and
\begin{equation}\label{eq;Kscr_bound_derivatives_repeated}
\sum_{|\gamma| \le r_1} \frac{1}{\gamma!} \,   \norm{ {\partial}^\gamma_F   \Kcal_F}_{\tilde \zeta,\Qscr_\Uscr}  \le \Theta.
 \end{equation}
with $\Theta$ depending on  $\Psi$, $\omega$, $r_0$,  $r_1$,  and $\Ran$.
\end{step} 
\medskip

By the chain and product rule it follows that the the derivatives $\partial^\alpha_z \partial^\gamma_F e^{\Uscr_\beta(z, F)}$ exists 
for all $|\alpha| \le r_0$ and $|\gamma| \le r_1$ and are continuous in $(z, F)$. 
To get a bound for $| \partial^\gamma_F e^{\Uscr_\beta(z, F)}|_{T_z}$ and to prove higher differentiability
of $F \mapsto e^{\Uscr_\beta(\cdot, F)}$, we proceed similarly to Step 4.
As shown in  Appendix \ref{se:polynomials_several_variables} the product property extends to the norm $|\cdot|_{T_{z,F}}$.

From the product property one deduces  as in Step 1 that
\begin{align}
 \sum_{|\gamma| \le r_1}  \frac{1}{\gamma!}  \big|  \partial_F^{\gamma}  e^{f(\cdot, F)} |_{T_z}  = 
|e^f|_{T_{z,F}} \le e^{f(z,F)} \, (1 + |f|_{T_{z,F}})^{r_0+r_1}.
\end{align}
For $f =- \Uscr_\beta$  we find with the results of Step 3 and Step 5 that
\begin{align}
\begin{split}
\label{eq:finalDerivativeEstimate}
\sum_{|\gamma| \le r_1}  \frac{1}{\gamma!}  \big|  \partial_F^{\gamma}  &e^{-\Uscr_\beta(\cdot, F)} |_{T_z}  
\leq e^{\frac12 \QU(z)-\frac{\omega}{2} |z|^2}\left(1+2^{r_0+r_1+1}(1+|z|)^3\Psi(|z|+1)\right)^{r_0+r_1} 
\\
&\leq
e^{\frac12 (1-\tilde \zeta)\QU(z)} e^{-\frac14 \omega|z|^2}\left(1+2^{r_0+r_1+1}(1+|z|)^3\Psi(|z|+1)\right)^{r_0+r_1},
\end{split}
\end{align}
where we used \eqref{eq:estimates_weights_norm_zeta} and the definition of $\tilde \zeta$ in the second step.
Invoking also  Lemma~\ref{le:criterion_diff_zeta} below,  it follows by induction in $|\gamma|$ that the map
$ F \mapsto e^{-\Uscr_\beta(\cdot, F)}$ is $r_1$ times continuously differentiable as a map from $B_{\delta_0}(0)$ to $\boldsymbol E$. 
Moreover \eqref{eq:finalDerivativeEstimate} implies the estimate 
\eqref{eq;Kscr_bound_derivatives_repeated}
for the higher derivatives with 
\begin{equation}  \label{eq:define_Theta_Kscr}
\Theta =  (|\mathcal G |+1)^{r_1}  \sup_z       e^{- \frac14 \omega  |z|^2} \left(1+2^{r_0+r_1+1}(1+|z|)^3\Psi(|z|+1)\right)^{r_0+r_1},
\end{equation}
where $(|\mathcal G |+1)^{r_1} \ge \abs{\{\gamma: |\gamma|\leq r_1\} }$ counts the number
of terms in the sum $\sum_{|\gamma|\leq r_1}$ which arises because we interchange 
the sum with the supremum in the definition of the $\lVert\cdot\rVert_{\tilde{\zeta},\QU}$ norm.
\end{proof}

\begin{lemma}   \label{le:criterion_diff_zeta}
Let $\Ocal$ be an open set in a finite dimensional space  and $h:\Ocal \times \Gcal\to \R$ a map satisfying  two conditions:
\begin{enumerate}[label=(\roman*), leftmargin=0.9cm]
\item
\label{it:criterion_diff_zeta1} 
For each $(F,z) \in \Ocal \times \mathcal G$ and each $\alpha$ with $|\alpha| \le r_0$ the partial derivatives
$ \partial_z^\alpha  h(F,z)$
exist and are continuous in $\Ocal \times \mathcal G$,
\item
\label{it:criterion_diff_zeta2} 
$\lim_{|z| \to \infty} e^{-\frac12 (1-\zeta) \QU(z)} \sup_{F \in \Ocal} |  h(F,\cdot)|_{T_z} = 0$.
\end{enumerate}
Define the function $g : \Ocal \to \boldsymbol E_\zeta$ by taking $(g(F))(z)=h(F,z)$. 
Then $g \in C^0(\Ocal,\boldsymbol E_\zeta)$. Moreover, if the conditions \ref{it:criterion_diff_zeta1} 
and \ref{it:criterion_diff_zeta2} hold for all partial derivatives
$h_i(F,z)=\frac{\partial }{\partial F_i}h(f,z)$ then $g\in C^1(\Ocal,\boldsymbol E_\zeta)$.
\end{lemma}

\begin{proof}
To prove that $F\to h(F,\cdot)$ is a continous map from $\Ocal$ to $\boldsymbol E_\zeta$ note that $h$ is uniformly continuous on compact subsets of $\Ocal \times \mathcal G$. Let $\delta > 0$. 
By assumption there exists an $R$ such that $\sup_{F \in \Ocal}  e^{-\frac12 (1-\zeta) \Qscr(z)} |h(F,\cdot)|_{T_z} \le \delta$ if $|z| > R$. 
Let $F_k \to F$. 
Then 
\begin{align*}
 \limsup_{k \to \infty} \| h(F_k, \cdot) -& h(F, \cdot) \|_\zeta
 = \, 
\limsup_{k \to \infty} \,  \sup_{z \in \mathcal G}  e^{\frac12(1-\zeta) \QU(z)}  | h(F_k, \cdot) - h(F, \cdot) |_{T_z}  \\
 \le  & \, 2 \delta + \limsup_{k \to \infty}  \sup_{|z| \le R}  e^{-\frac12 (1-\zeta) \QU(z)}  | h(F_k, \cdot) - h_i(F, \cdot) |_{T_z}
 = 2 \delta
\end{align*}
by uniform continuity on compact sets. Since $\delta > 0$ was arbitrary, this  establishes that $g\in C^0(\Ocal,\boldsymbol E_\zeta)$.

Assume now that all partial derivatives $h_i=\frac{\partial}{\partial F_i}h$ satisfy \ref{it:criterion_diff_zeta1} and \ref{it:criterion_diff_zeta2}.
The same reasoning as before implies that $F \mapsto h_i(F,\cdot)$ is a continuous map from $\Ocal$ to $\boldsymbol E_\zeta$.
Then we use that
\begin{align}
 h( F + \eta e_i, z) - h(F, z) - h_i(F,z)\eta = \int_0^1 [h_i( F + t\eta,z) - h_i(F,z)]  \, \eta \, \d t, 
\end{align}
divide by $\eta$, use Jensen's inequality for $| \cdot |_{T_z}$ and take the limit $\eta \to 0$ to  show that the map 
$g:\Ocal \to \boldsymbol E_\zeta$  has partial derivatives  given by $h_i( F, \cdot)$. 
Moreover these partial derivatives  are continuous. 
Since $\Ocal$ is in a finite dimensional space this implies the assertion that $g\in C^1(\Ocal, \boldsymbol E_\zeta)$. 
\end{proof}

\chapter{Discrete Nonlinear Elasticity} \label{sec:elasticity}

In this chapter we show how the results for discrete elasticity can be deduced from the results for non-degenerate gradient models. 
The main difficulty is that in discrete elasticity  the local interaction energy  $U$ is invariant under the action of
the group $\mathrm{SO}(d)$ of proper rotations. 
Thus the second derivative of $U$ at the identity (which is the natural minimum energy state) is degenerate. 
We will overcome this degeneracy by adding a so-called discrete null Lagrangian $\nullL$ to $U$.

This will allow us to break the degeneracy of $U$, see Lemma~\ref{le:embedding_de}, in particular inequalities \eqref{eq:V1bis2} and \eqref{eq:V32}.
 At the same time,  for fixed linear boundary conditions $F$,  the addition of $\nullL$ will not change the Gibbs measure 
 and will  effect the free energy only in a trivial way, see  Lemma~\ref{le:null_lagrangian_periodic} as well as  equations   
  \eqref{eq:equivalence_myW} and \eqref{eq:equivalence_laplace}.
Thus the desired assertions for discrete elasticity follow easily from  Theorems~\ref{th:bound_Wcal} and~\ref{th:scalinglimit_concrete}, 
see the end  of Section~\ref{sec:reformulation_de} for the detailed argument. 

To set the stage and to motivate the concept of  discrete null Lagrangians,  
we begin this chapter by a very brief discussion of continuous nonlinear elasticity and continuous null Lagrangians. 
Readers familiar with these topics may directly advance  to Section~\ref{sec:discrete_null_lagrangians}.

\section{Continuous nonlinear elasticity and continuous null Lagrangians}  \label{sec:continuous_elasticity}

In this section we discuss three points:  

\noindent
(i) the general setting of continuous nonlinear elasticity and the fact that the
stored energy density cannot be convex due to the $\mathrm{SO}(d)$ symmetry; 

\noindent
(ii) the definition and characterisation of null Lagrangians in this setting;

\noindent
(iii) a simple  argument that the stored energy density can be made uniformly convex in the neighbourhood of the identity
by adding a suitable null Lagrangian. 

These results are not really needed for the rest of this work, 
but they provide background and motivation for our constructions in discrete elasticity.
Regarding nonlinear elasticity, we just give a bare bones description, see, e.g.,  \cite{Bal77, Cia88, Bal02}
for further information. 

Let $d \ge 2$. A $d$-dimensional elastic body is described by an open and bounded set $\Omega \subset \R^d$, called the \emph{reference configuration} of
the body. A map $\varphi: \Omega \to \R^d$ is called a deformation of the body and the image $\varphi(\Omega)$ is called the \emph{deformed configuration}
or current configuration. To such a map $\varphi$ one assigns the \emph{elastic energy}
\begin{equation} \label{eq:elastic_energy_1}
I(\varphi) = \int_\Omega W(x, \nabla \varphi(x)) \, dx. 
\end{equation}
One usually chooses the reference configuration $\Omega$ so that the identity map is a global minimiser of $I$. If, in addition, the body is homogeneous, then
the expression for the elastic energy reduces to 
\begin{equation} \label{eq:elastic_energy_2}
I(\varphi) = \int_\Omega W( \nabla \varphi(x)) \, dx. 
\end{equation}
In what follows, we will work in this setting.
In particular,  the function $W: \R^{d \times d} \to \R$ has a global minimum at the identity matrix $\1$. 
This function is referred to as the \emph{stored energy density}
\footnote{One can also consider functions $W: \R^{d \times d} \to \R \cup \{\infty\}$. This allows one to incorporate
certain constraints, e.g., the constraint $\det \nabla \varphi > 0$ which states that deformations  should be infinitesimally orientation preserving.}
We may assume without loss of generality that $W(\1) = 0$ and hence $W(F) \ge 0$ for all $F \in \R^{d \times d}$.

The equilibrium states of an elastic body correspond to minimisers or, more generally, stationary points of the elastic energy $I$,
subject to suitable boundary conditions. Here we focus on Dirichlet boundary conditions
\begin{equation} \label{eq:dirichlet_bc}
u = u_0 \quad \text{on $\partial \Omega$,}
\end{equation}
where $u_0 : \partial \Omega \to \R^d$ is a given map.

Many of the analytical subtleties of nonlinear elasticity arise from the fact that the elastic energy is invariant under rigid motions of the deformed configuration. 
This is equivalent to the invariance property
\begin{equation} \label{eq:SOd_invariance_continuous}
W(QF) = W(F) \quad \text{for all $Q \in \mathrm{SO}(d)$ and all $F \in \R^{d \times d}$.}
\end{equation}
It follows that  $W= 0$ on $\mathrm{SO}(d)$. For a realistic elastic body one also has $W(F) > 0$ if $F \notin \mathrm{SO}(d)$
since every deformation which is not a rigid motion costs energy. 
In particular,  $W$ \emph{cannot be convex} since then $W$ would vanish on the convex hull of $\mathrm{SO}(d)$
 which is strictly larger than $\mathrm{SO}(d)$. The failure of convexity rules out the use of soft functional analytic methods to study
 question such as the existence of minimisers of $I$ subject to the Dirichlet boundary condition  \eqref{eq:dirichlet_bc}.
 Indeed the existence of minimisers was only established in the pioneering work of Ball \cite{Bal77}.
 
 The failure of uniform convexity can already been seen on the infinitesimal level. If $W$ is a $C^2$ function then the invariance property
 \eqref{eq:SOd_invariance_continuous} implies that 
 \begin{equation} \label{eq:degeneracy_D2W_continuous}
 D^2 W(\1)(G, G) = 0 \quad \text{for all skew-symmetric $d \times d$ matrices $G$}
 \end{equation}
since the tangent space of $\mathrm{SO}(d)$ at the identity consists of the  skew-symmetric $d \times d$ matrices.
\medskip

This is the same difficulty as that one we face in the discrete setting. We will  show  below that we can modify $W$ in such a way that it becomes
uniformly convex near $\1$ without changing the minimisers of the Dirichlet problem. 
The key concept is that of null Lagrangians. Informally, a function $f: \R^{d \times d}  \to \R$ is a null Lagrangian
if $\int_\Omega f(\nabla \varphi) \, dx$ depends only on the boundary values of $\varphi$. 
A precise definition, for a general case of $\R^m$-valued vector fields, is as follows.

\begin{definition} \label{de:null_lagrangian_continuous} 
Let $\Omega \subset \R^d$ be bounded and open (and non-empty). A continuous map $f: \R^{m \times d} \to \R$ is a \emph{null Lagrangian} if and only if
\begin{equation}  \label{eq:continuous_NL1}
\int_\Omega f(\nabla \varphi) \, dx = \int_\Omega f(\nabla \psi) \, dx 
\end{equation}
whenever
\begin{equation} \label{eq:continuous_NL2}
\varphi, \psi \in C^2(\overline \Omega, \R^m) \quad \text{and $\psi-\varphi$ has compact support in $\Omega$.}
\end{equation}
\end{definition}
Here $C^2(\overline \Omega, \R^m)$ denotes the space of $C^2$  functions such that  their derivatives up to order $2$ 
have continuous extensions to the closure $\overline \Omega$. 

\begin{remark} \label{re:continous_NL}\hfill
\begin{enumerate}[1.,leftmargin=0.3cm]
\item From the definition, the notion of a null Lagrangian may depend on the choice of the open, non-empty set $\Omega$. 
Proposition~\ref{pr:null_lagrangian_continuous} below shows that it does not.
\item Also, by Proposition~\ref{pr:null_lagrangian_continuous} \gray{below}, every null Lagrangian is a polynomial of degree $\le \min(m,d)$. 
Thus  the property \eqref{eq:continuous_NL1} extends to functions $\varphi, \psi$ in the Sobolev space $W^{1,p}(\Omega; \R^m)$ 
as long as $\varphi-\psi \in W^{1,p}_0(\Omega; \R^m)$ and $p \ge \min(m,d)$. 
In particular, we see that for Dirichlet boundary conditions, the minimisers of the functional $I$ do not change if we replace $W$ by $W+f$. 
\end{enumerate}
\end{remark}

The following result shows that null Lagrangians exist  and gives  an explicit characterisation.
 It is a special case of results for more general integrands, see \cite{BCO81} and the references therein.

\begin{proposition}  \label{pr:null_lagrangian_continuous} 
Let $f: \R^{m \times d} \to \R$ be continuous and let $\Omega \subset \R^d$ be bounded, open and non-empty.
Then the  following properties are equivalent.
\begin{enumerate}[label=(\roman*),leftmargin=0.9cm]
\item  \label{it:NL1}  The function $f$ is a null Lagrangian;
\item   \label{it:NL2} For every $F \in \R^{m \times d}$ and   $\varphi \in C_c^\infty(\Omega; \R^m)$ one has
\begin{equation}   \label{eq:null_lagrangian_affine}
\int_\Omega \bigl(f(F +\nabla  \varphi)  - f(F)\bigr) \, dx = 0;
\end{equation}
\item    \label{it:NL3} The function $f$ is rank-one affine, i.e.,  for every $F \in \R^{m \times d}$,  $a \in \R^m$ and $\xi \in \R^d$ the map 
$t \mapsto f(F + t a \otimes \xi)$ is an affine function on $\R$;
\item     \label{it:NL4}    The expression $f(F) $ is the sum of a constant term and  a linear combination of minors (subdeterminants) of $F$;
\item  \label{it:NL5} The function $f$ is $C^2$ and for every $\varphi \in C^2(\Omega)$ we have
\begin{equation}  \label{eq:EL1}
 \int_\Omega  Df(\nabla \varphi) \cdot D \eta \, dx = 0 \quad \text{for all $\eta \in C_c^2(\Omega; \R^m)$}
\end{equation}
or, equivalently, 
\begin{equation}    \label{eq:EL2} -\mathrm{div} Df(D\varphi) = 0 \quad \text{for all $\varphi \in C^2(\Omega;\R^m)$.}
\end{equation}
\end{enumerate}
\end{proposition}

In  \eqref{eq:EL1} the dot $\cdot$ denotes the usual scalar product on the space of  $m \times d$ matrices and $C_c^2(\Omega; \R^m)$
denotes the space of  those $C^2$ maps from $\Omega$ to $\R^m$ that  have compact support in $\Omega$. 
The property   \eqref{eq:EL2} motivates the name 'null Lagrangian': for \emph{every} $C^2$ map $\varphi$ the 
Euler-Lagrange equations are satisfied identically. 

\begin{proof} This is well known. We include a proof  for the convenience of the reader. 
We show that  
\eqref{it:NL1} $\Longrightarrow$ \eqref{it:NL2} $\Longrightarrow$ \eqref{it:NL3} $\Longrightarrow$ \eqref{it:NL4} $\Longrightarrow$ \eqref{it:NL1} and
\eqref{it:NL1} $\Longleftrightarrow$ \eqref{it:NL5}. 

\eqref{it:NL1} $\Longrightarrow$ \eqref{it:NL2}: \quad Obvious.

\eqref{it:NL2} $\Longrightarrow$ \eqref{it:NL3}:  \quad We have to show that, for every $\lambda \in (0,1)$. 
\begin{equation} \label{eq:rank_one_convexity}
f(F)  = \lambda f(F - (1-\lambda) a \otimes \xi) + (1-\lambda) f(F + \lambda a \otimes \xi).
\end{equation}
First note that   by approximation  \eqref{eq:null_lagrangian_affine} also holds for every Lipschitz map $\varphi$ which has compact support in $\Omega$. 
Let $h: \R \to \R$ be a one-periodic Lipschitz  function such that $h' = - (1-\lambda)$ on $(0, \lambda)$ and $h' = \lambda$ on $(\lambda, 1)$,
let $\eta \in C_c^\infty(\Omega)$ and let $k$ be an integer.
Apply  \eqref{eq:null_lagrangian_affine} with the maps $\varphi = \eta  a \frac1k h(k x \cdot \xi)$. 
Letting $k \to \infty$ we see that 
\begin{equation}
\int_\Omega   \lambda f(F -  \eta (1-\lambda) a \otimes \xi) + (1-\lambda) f(F +  \eta \lambda a \otimes \xi) - f(F) = 0.
\end{equation}
Choosing a sequence of functions $\eta_\ell \in C_c^\infty(\Omega)$ that converges to a characteristic function we deduce
\eqref{eq:rank_one_convexity} 

\eqref{it:NL3} $\Longrightarrow$ \eqref{it:NL4}: \quad 
This is a nice exercise in multilinear algebra, see, e.g.,   \cite[Theorem 6.1]{CDKM06}.

\eqref{it:NL4} $\Longrightarrow$ \eqref{it:NL1}: \quad It suffices to show the result for a single $r \times r$ subdeterminant. 
We give a proof based on differential forms which nicely encodes the multilinear properties of subdeterminants.
Alternatively one can prove the result by directly using the definition of determinants. This is particularly easy for dimensions $d=2$ 
and $d=3$ which are the most interesting cases, see the end of the proof. 

After a possible permutation of the coordinates in $\R^m$ and $\R^d$ we may assume that 
$f(\nabla \varphi) = \det \frac{\partial (\varphi_1, \ldots, \varphi_r)}{\partial(x_1, \ldots, x_r)}$. 
Let $\Omega' \subset \Omega$ be an open set with smooth boundary which includes the support of $\varphi - \psi$. 
Writing $dx$ for the differential form $dx^1 \wedge \ldots \wedge dx^d$ and using the Leibniz rule
for differential forms and Stokes theorem we get
\begin{align}  \label{eq:NL_differential_forms}
 \, \int_{\Omega'} f(\nabla \varphi) dx 
 =& \,  \int_{\Omega'} d\varphi_1 \wedge d\varphi_2  \ldots \wedge d\varphi_r  \wedge dx^{r+1} \ldots \wedge dx^d    \nonumber \\
=& \, \int_{\Omega'} d (\varphi_1 d\varphi_2   \ldots \wedge d\varphi_r \wedge dx^{r+1} \ldots\wedge dx^d)     \nonumber \\
= & \, \int_{\partial \Omega'} \varphi_1 d\varphi_2   \ldots \wedge d\varphi_r \wedge dx^{r+1} \ldots\wedge dx^d.    \nonumber
\end{align}
Since $\Omega'$ includes the support of $\varphi-\psi$  we can replace $\varphi$ by $\psi$ in the last integral.
Using   \eqref{eq:NL_differential_forms}  for $\psi$ instead of $\varphi$ we conclude that 
$ \int_{\Omega'} f(\nabla \varphi)  dx =  \int_{\Omega'} f(\nabla \psi)  dx$.
Moreover, $\nabla \psi= \nabla \varphi$ on $\Omega \setminus \Omega'$ and thus 
$ \int_{\Omega} f(\nabla \varphi)  dx =  \int_{\Omega} f(\nabla \psi)  dx$.

For $r=2$ the calculation above simply boils down to the identity
\begin{align}
\det  \frac{\partial (\varphi_1, \varphi_2)}{\partial(x_1,  x_2) } = \frac{\partial}{\partial x_1}  
\Bigl( \varphi_1  \frac{\partial \varphi_2}{\partial x_2} \Bigr)
-   \frac{\partial}{\partial x_2}  \left( \varphi_1  \frac{\partial \varphi_2}{\partial x_1} \right),
\end{align}
and  for $r=3$ a similar explicit formula which expresses the determinant as a divergence 
using $\varphi_1$ and $2 \times 2$ minors of $\varphi_2$, $\varphi_3$ is easily derived. 

\eqref{it:NL1} $\Longrightarrow$ \eqref{it:NL5}: \quad By \eqref{it:NL4} the function $f$ is a polynomial and hence $C^2$. 
Apply the condition in the definition of the null Lagrangian with $\psi = \varphi+t \eta$, differentiate with respect to $t$ and
evaluate at $t=0$.

\eqref{it:NL5} $\Longrightarrow$ \eqref{it:NL1}: \quad Set $\varphi_t = \varphi + t (\psi-\varphi)$
and  $h(t) = \int_\Omega f(\nabla \varphi_t)$. Note that  $\eta:= \psi-\varphi \in C_c^2(\Omega; \R^m)$. 
Thus  \eqref{eq:EL1} implies that $h'(t) = 0$ for all $t \in [0,1]$. Hence $h(0) = h(1)$ which is the desired conclusion. 
\end{proof}

\medskip

Finally we show that by adding a null Lagrangian we can locally achieve uniform convexity provided that $D^2W(\1$) positive
definite in the subspace perpendicular to the tangent space of $\mathrm{SO}(d)$ at $\1$.

\begin{proposition}  \label{pr:uniform_convexity continuous} Suppose that $W: \R^{d \times d} \to \R$
is $C^2$ in a neighbourhood of $\1$ and
satisfies
\begin{equation}
W(QF) = W(F) \quad \text{for all $Q \in \mathrm{SO(d)}$ and $F \in \R^{d \times d}$}
\end{equation}
and
\begin{equation}
DW(\1) = 0.
\end{equation}
Suppose further  that there exists a $c > 0$ such that  
\begin{equation} \label{eq:D2_sym}
D^2 W(\1)(G, G) \ge c |G|^2 \quad \text{for all symmetric $d \times d$ matrices $G$.}
\end{equation}
Then for all sufficiently small $\alpha >0$ there exists a  $c' >0$ and a $\delta > 0$ such that  
\begin{equation}
D^2(W + \alpha \det)(F)(G,G) \ge c' |G|^2 \quad \text{for all $G \in \R^{d \times d}$ and all $F \in B_\delta(\1)$.}
\end{equation}
\end{proposition}

\begin{proof} It suffices to show that for all sufficiently small $\alpha > 0$ there exists a $c' >0$ such that
\begin{equation}  \label{eq:D2Walpha_1}
D^2 (W + \alpha \det)(\1)(F,F) \ge 2 c' |F|^2.
\end{equation}
Then the assertion follows by the continuity of the second derivative.

The main point is to show that $D^2 \det(\1)$ is positive definite on skew-symmetric matrices. 
To do so,  note that for a skew-symmetric $G$ the matrix $e^G$ is in $\mathrm{SO}(d)$. 
Differentiating the identity $1 = \det e^{tG}$ twice and evaluating at $t=0$, we 
get
\begin{align}  \label{eq:D2det_a}
0 = \frac{d}{dt}|_{t=0} D\det(e^{t G})(e^{tG} G) = D^2 \det(\1)(G,G) + D\det(\1)(G^2).
\end{align}
Thus
\begin{align}  \label{eq:D2det_b}
D^2 \det(\1)(G,G) = -\tr G^2 = \tr G^T G = |G|^2 \quad \text{for all skew symmetric $G$.}
\end{align}

Now, let $G$ be skew-symmetric and let $H$ be symmetric. Computing the second derivative
of $t \mapsto W(e^{tG}(\1 + t H)) = W(\1 + t H)$ at $t=0$ and using that $DW(\1) = 0$, we get
\begin{align}
D^2 W(\1)(G+H, G+H) = D^2 W(\1)(H,H) \ge c |H|^2.
\end{align}

Combining this with the estimates
$D^2 \det(\1)(H,H) \le C |H|^2$  and
\begin{equation}
2 D^2 \det(\1)(G,H) \le 2 C |G| |H| \le \frac12  |G|^2 + 2   C^2  |H|^2, 
\end{equation}
we get
\begin{equation}
D^2(W + \alpha \det)(\1)(G+H, G+H) \ge  (c -  \alpha(C + 2 C^2) |H|^2 + \frac{\alpha}{2}  |G|^2.
\end{equation}
Since every matrix $F$ can be written as $F = G+ H$ with $G$ skew-symmetric and $H$ symmetric and since $|G+H|^2 = |G|^2 + |H|^2$
we see that  \eqref{eq:D2Walpha_1} holds provided that $0 < \alpha < \frac12 \frac{c}{C+2 C^2}$ and 
$c' = \frac12 \min(c- \alpha(C + 2 C^2), \alpha/2)$.
\end{proof}

\section{Discrete null Lagrangians}  \label{sec:discrete_null_lagrangians}

We saw in \eqref{eq:U-Uscr} that the Hamiltonian $\HN^F$ can be formulated in terms of a potential $U$ 
with finite range support $A$ as well as in terms of the generalized gradient potential $\Uscr$.
In discrete elasticity, however,
the potential  $U$ and thus also $\Uscr$ has a degenerate minimum due to the invariance under rotiatons
and we cannot directly apply the results stated in the in Section~\ref{sec:ggm}
on generalized gradient models.
Instead we first need to gain local coercivity. This can be done with the help of an  addition of a discrete null  Lagrangian.

Let us first  introduce the concept of discrete null Lagrangians.
\begin{definition}  \label{de:null_lagrangian}
A function $\nullL: \left(\R^d\right)^A\rightarrow \R$ is called a \emph{discrete null Lagrangian} 
 if for any finite set $\Lambda\subset\mathbb{Z}^d$ and  any  $\p, \tilde{\p}\in (\R^d)^{\Z^d}$ 
 such that $\p(x)=\tilde{\p}(x)$ for all $x\notin \Lambda$ we have the following identity
\begin{align}\label{eq:defofNullLagrangian}
\sum_{x \in \Lambda_A} N( (\tau_{-x} \p)_A)=  
\sum_{x \in \Lambda_A} N( (\tau_{-x} \tilde\p)_A) 
\quad \hbox{where $\Lambda_A := \{ x \in \Z^d : (x+A) \cap \Lambda \neq \emptyset\}$.}
\end{align}
\end{definition}

If $\nullL$ is a discrete null Lagrangian and  $\p(x)=Fx$ for $x\notin \Lambda$ then, in particular,
\begin{align}
\sum_{x \in \Lambda_A} N( (\tau_{-x} \p)_A)=  
\sum_{x \in \Lambda_A} N( \tau_{-x} F) 
	\end{align}
It is useful to note that \eqref{eq:defofNullLagrangian} holds if and only if
\begin{align}  \label{eq:criterion_null_Lagrangian}
\sum_{x \in \Lambda'} \nullL( (\tau_{-x} \p)_A)=         
		\sum_{x \in \Lambda'}    \nullL( (\tau_{-x} \p)_A)
\quad \hbox{for some  finite $\Lambda'$ with $\Lambda_A \subset \Lambda' \subset \Z^d$.}
\end{align}
This follows immediately from the observation that $x \in \Lambda' \setminus \Lambda_A$ implies that
$x+A \subset \Z^d \setminus \Lambda$ and hence $(\tau_{-x} \p)_A  = (\tau_{-x} \tilde \p)_A$.

\begin{example} \label{ex:linear_NL} 
Let $A=\{0, y\}$ where $y\in \Z^d$ and  $\nullL(\p)  =\p(y) - \p(0). $ 
Then $\nullL$ is a discrete null Lagrangian. To see this we use the criterion \eqref{eq:criterion_null_Lagrangian}.
Consider a cube  $\Lambda'$ which is so large that  
\begin{align}  \label{eq:linear_Lambda_prime_large}
\Lambda_A \subset \Lambda', \quad \Lambda  \cap \big( (y + \Lambda') \setminus \Lambda'\big) = \emptyset \quad \hbox{and} \quad 
\Lambda  \cap \big( \Lambda' \setminus (y + \Lambda')\big) = \emptyset.
\end{align}
Now
$$ 
\sum_{x \in \Lambda'} \nullL( (\tau_{-x}\p)_A) = \sum_{ z \in (y + \Lambda' )\setminus \Lambda'} \p(z)  - \sum_{z \in \Lambda' \setminus (y +\Lambda')} \p(z).
$$ 
Thus the assertion follows from \eqref{eq:linear_Lambda_prime_large} since $\p = \tilde \p$ in $\Z^d  \setminus \Lambda$. \\
It follows that all the maps $\p \mapsto \nabla^\alpha \p(0)$ with $\alpha \ne 0$ are discrete null Lagrangians.
\end{example}

\begin{example}  \label{ex:det_NL} An important example of a nonlinear discrete null Lagrangian is given by the discrete determinant. 
For $d=2$ and a map $\psi:\{0,1\}^2 \to \R^2$ one defines the discrete determinant as the oriented area of the
oriented quadrilateral generated by the points $\psi(0)$, $\psi(e_1)$, $\psi(e_1+e_2)$, $\psi(e_2)$. 
Since this quadrilateral can be seen as the sum of the two oriented triangles 
$\psi(0), \psi(e_1), \psi(e_2)$ and $ \psi(e_1+e_2), \psi(e_2), \psi(e_1)$ this oriented area is given by 
\begin{align*} 
\nullL_{\rm det}(\psi)  =  & \,  \frac12  \det( \psi(e_1) - \psi(0) \, | \, \psi(e_2) - \psi(0)) \\
+  & \, \frac12 \det( \psi(e_2) - \psi(e_1+e_2)   \, | \,  \psi(e_1) - \psi(e_1+e_2)      )
\end{align*}
Then for a square 
$Q_l = \{0, 1, \dots, l\}^2$  with the oriented boundary 
$$
\vec P_l= \bigl((0,0), (1,0),\!..,(l,0),(l,1),\!..,(l,l), (l-1,l),\!.., (0,l), (0,l-1),\!..,(0,0)\bigr),
$$ 
the sum $\sum_{x \in Q_l} \nullL((\tau_{-x} \p)_{\{0,1\}^2})$ is the sum of the oriented areas of the image of the subsquares  
of $(0,l)^2$ of size $1$ and hence the oriented area of the oriented polygon $\p(\vec P_l)$.
Thus it follows from the criterion \eqref{eq:criterion_null_Lagrangian} that 
 $\nullL$ is a discrete null Lagrangian (given $\Lambda$, take $\Lambda' = Q_l-\lfloor \frac{l}2\rfloor (1,1,\dots, 1)$ with sufficiently large $l$).

To generalize the discrete determinant to higher dimensions, it is useful to reformulate first the case $d=2$ and to express 
$\nullL_{{\rm det}}$  with the help of the continuous null Lagrangian $\det \nabla \psi$
for $\psi: \Omega \subset \R^2 \to \R^2$. For  $\p: \Z^2 \to \R^2$ define $I(\p): \R^2 \to \R^2$ as the multilinear interpolation, i.e., 
for $x \in \Z^2$ the map $I(\p_{|x + [0,1]^2})(y)$ is the unique map which is affine in each coordinate direction $y_i$ 
and agrees with $\p$ on $x + \{0,1\}^2$. 
Note that $I(\p)$ is defined consistently along the  lines $x_i \in \Z$ and is continuous on $\R^2$.
Also, note that $I(\p)(\vec P_1)$ is the boundary of the polygon generated by the points $\p(0)$, $\p(e_1)$, $\p(e_1+e_2)$, $\p(e_2)$.
Thus
\begin{equation} 
\label{E:NQ}
\nullL_{\rm{det}}(\p) = \int_{(0,1)^2} \det \nabla I(\p) \, \d x
\end{equation} 
and 
\begin{equation} 
\label{E:Nb}
\sum_{x \in Q_l} \nullL((\tau_{-x} \p)_{\{0,1\}^2}) = \sum_{x \in Q_l} 
 \int_{x+(0,1)^2} \det \nabla I(\p) \, \d x = \int_{(0,l)^2} \det \nabla I(\p) \, \d x.
\end{equation} 
By Proposition~\ref{pr:null_lagrangian_continuous} and Definition~\ref{de:null_lagrangian_continuous}, the
integral on the right hand side  of \eqref{E:Nb} depends only on $I\p_{| \partial (0,l)^2}$ and thus only 
on $\p_{| \vec P_l}$. This gives another proof that $\nullL_{ \rm{det}}$ is a discrete null Lagrangian.

For $d \ge 3$ and $\p : \Z^d \to \R^d$ we define the multilinear interpolation in the same way. We then define 
the  discrete determinant $\nullL_{\rm{det}}: (\R^d)^{\{0,1\}^d} \to \R$ by
\begin{align}  \label{eq:def_discrete_det}
 \nullL_{\rm{det}} (\p) = \int_{(0,1)^d} \det \nabla I(\p) \, \d x.
 \end{align}
It follows that, in analogy with \eqref{E:Nb}, we have
\begin{equation} 
\label{E:Nb_higher_d}
\sum_{x \in Q_l}  \nullL_{\rm{det}} ((\tau_{-x} \p)_{\{0,1\}^d}) = \int_{(0,l)^d} \det \nabla I(\p) \, \d x.
\end{equation} 
Thus the left hand side of    \eqref{E:Nb_higher_d}  depends only on $I \varphi|_{\partial(0,l)^d}$ and hence
only on $\varphi_{ \mathbb Z^d \cap \partial(0,l)^d}$. This shows that $\nullL_{\rm{det}}$ is a discrete null Lagrangian.

Note that for each $y \in (0,1)^d$ the expression $\nabla I(\p)(y)$ is a linear combination of the values
$\p(x)$ for $x \in \{0,1\}^d$. Thus $\nullL_{\rm{det}}$ is a homogeneous polynomial of degree $d$ on $ (\R^d)^{\{0,1\}^d}$. 
If $F: \R^d \to \R^d$ is linear and  $\p_F(y) = Fy$ for all $y \in \{0,1\}^d$ then $I\p(y) = Fy$ and thus
\begin{align}  \label{eq:consistencey_discrete_det}
 \nullL_{\rm{\det}}(F) = \det F.
\end{align}
\end{example}

Discrete null Lagrangians are defined using Dirichlet boundary conditions on $\Z^d$. 
One can extend the null Lagrangian property to periodic perturbations. 
We will only need the following result.

\begin{lemma}  \label{le:null_lagrangian_periodic}
Assume that $\nullL : (\R^d)^A \to \R$ is a discrete shift-invariant  null Lagrangian
and assume that $\nullL$ is bounded on bounded sets.  Let $F: \R^d \to \R^d$ be a linear map.
Assume that $|A|_\infty := \sup \{ |y|_\infty : y  \in A \} \le \frac18 L^N$. Then for all periodic functions $\p: T_N = \Z^d/ L^N \Z^d \to \R^d$
\begin{align}  \label{eq:null_lagrangian_periodic}
\sum_{x \in T_N} \nullL( (F + (\tau_{-x} \p)_A) = \sum_{x \in T_N} \nullL(F) = L^{Nd} \nullL(F).
\end{align}
where $F$ also denotes the restriction of $F$ to $A$.
\end{lemma}

\begin{proof}   
The proof is standard, but we include it for the convenience of the reader. Fix $F$ and $\p$. 
We use the set $[\frac{-L^N-1}{2}, \frac{L^N-1}{2}]^d$ as the fundamental domain of $T_N$. 
Extend $\p$ to an $L^N$-periodic function on $\Z^d$. Let $M$ be a large  odd integer and consider a cut-off function $\eta: \Z^d \to [0,1]$ such that
\begin{align}
\eta(x) = 
\begin{cases}
1 \quad &\hbox{if $|x|_\infty \le (M-2) \frac{L^N-1}{2}  + |A|_\infty$}, 
\\
 0 \quad &\hbox{if $|x|_\infty \ge  M \frac{L^N-1}{2}  -  2 |A|_\infty$},  
 \end{cases} 
\end{align}
Apply the criterion  \eqref{eq:criterion_null_Lagrangian} with the set $\Lambda' =  \Lambda_M := [- \frac{M L^N-1}{2} ,  \frac{M L^N- 1}{2}]^d$,   
$\tilde \p = 0$ and $\eta \p$ in place of $\p$ and use that $\tau_{-x} F$ and $F$ differ only by a constant map and that $\nullL$ is shift invariant.
This gives
\begin{align} \label{eq:nl_periodic_cut}
 \sum_{x \in \Lambda_M} \nullL( F + (\tau_{-x} (\eta \p))_A) = \sum_{x \in \Lambda_M} \nullL(F) = M^d L^{Nd} \nullL(F).
\end{align}
Now
$$  \sum_{x \in \Lambda_{M-2}} \nullL( F + (\tau_{-x} (\eta \p))_A)  = (M-2)^d \sum_{x \in T_N} \nullL(  F + 
(\tau_{-x}  \p)_A) $$
Since $\p$ is bounded on $\Z^d$,  we get  $|F y  + (\tau_{-x}(\eta \p))(y) | \le C$ for all $y \in A$ and $x \in \Z^d$.
Using the assumption that $\nullL$ is bounded on bounded sets we get 
$| \nullL (F + (\tau_{-x}(\eta \p))_A)| \le C'$
and 
$$ \sum_{x \in \Lambda_M \setminus \Lambda_{M-2} }  | \nullL (F + (\tau_{-x}(\eta \p))_A)| \le C' (M^d - (M-2)^d) L^{Nd}.$$
Dividing  \eqref{eq:nl_periodic_cut} by $M^d$ and passing to the limit $M \to \infty$,  we get 
\eqref{eq:null_lagrangian_periodic}.
\end{proof}

We now make use of discrete null Lagrangians for improving the convexity properties of elastic interaction energies.
Recall assumptions (H1) to  (H5) introduced in Section~\ref{sec:discrete_elasticity_main}.

\begin{theorem}\label{th:convexification}[Theorem 5.1 in \cite{CDKM06}]
	Under the assumptions (H1)-(H4) there is a shift invariant discrete null Lagrangian $\nullL\in C^\infty(\left(\R^d\right)^A,\R)$ 
	and a shift invariant function $E\in C^2(\left(\R^d\right)^A,\R)$ such that: 
\begin{enumerate}[label=(\roman*),leftmargin=0.9cm]
\item $E$ is uniformly convex on the subspace $\Vcal_A^\perp$ orthogonal to the shifts;
\item For  all $\psi\in \bigl(\R^d\bigr)^A$
\begin{align}\label{eq:propofHineq}
		U(\psi)+ \nullL(\psi)\geq E(\psi);
\end{align}
\item For $\psi\in \bigl(\R^d\bigr)^A$ that are close to rotations $\psi_\mathbf{R}(x)=\mathbf{R}x$ with $\mathbf{R}\in \mathrm{SO}(d)$,
\begin{align}\label{eq:propofHeq}
	U(\psi)+ \nullL(\psi) =E(\psi) .                    
\end{align}
\end{enumerate}

In fact one can take  $\nullL = \alpha \nullL_{\rm{det}}$ for some $\alpha \in \R$ where $\nullL_{\rm{det}}$ is  the discrete determinant
defined in  \eqref{eq:def_discrete_det}. 
Hence $\nullL$ is  a polynomial of degree $d$ and,  in particular, it is smooth. 
Moreover, $\nullL$ depends only  on the values of the deformation in one unit cell
whose corners are contained in $A$ and for affine maps $F:\mathbb{Z}^d\to \R^d$\gray{,} restricted to  $A$ 
one has  
\begin{align}\label{eq:NeqDet}
\nullL(F) = \alpha \det F.
\end{align}
\end{theorem}

\begin{remark}  \label{re:explain_null_lagrangian} 
The heart of the matter is to show that for small  $\alpha>0$ the quadratic form $D^2(U + \alpha  N_{\rm{det}})(z)$ is positive definite 
on $\mathcal V_A^\perp$ for  $z=\1$ (and hence for $z$ in a small neighbourhood of $\1$). This is easy. 
Indeed, $D^2 U(\1)$ is positive semidefinite on $\Vcal_A^\perp$ since 
$\1$ is a minimum of $U$ and, by assumption,
positive definite on the complement of the space $\Scal \subset \Vcal_A^\perp$ of skew symmetric linear maps. It thus suffices 
to show that $D^2 N_{\rm{det}}$ is positive definite on $\mathcal S$. For $F \in \Scal$ we have $N_{\rm{\det}}(F) = \det F$. 
Now by  \eqref{eq:D2det_b} we have $D^2\det(\1)(F,F) \ge |F|^2$ for skew symmetric $F$.
\end{remark}
 \medskip

\section{Reformulation of discrete elasticity as a  non-degenerate generalized gradient model} 

\sectionmark{Discrete elasticity as generalized gradient models}\label{sec:reformulation_de}

In the following we will  rephrase the model
given by the Hamiltonian \eqref{eq:U-Uscr} in the setting introduced in Section \ref{sec:setup}.
We first give an overview of the general strategy.
The key idea is to consider the energy given by $U+\nullL$ instead of 
$U$, where $\nullL$ is the discrete null Lagrangian in Theorem~\ref{th:convexification}.
 The function $U+\nullL$ is bigger than a strictly convex function and agrees with it in a neighbourhood of the identity. 
In particular the second derivative at the identity  is strictly positive (modulo shift invariance)
so it almost falls in the class of energies satisfying the assumptions of Proposition \ref{prop:embedding}
(up to a trivial shift from $0$ to  the generalized gradient $\1_{Q_\Ran}$  of the identity map on $Q_{\Ran}$).
One minor issue is that  relying on Lemma~\ref{L:U-Uscr}, we restricted the passage from finite range interaction $U$
on $Q_{\Ran}$  to generalized gradient interactions $\Uscr$ on $\Gcal_{\Ran}$. 
Given that  the interaction term $U$  depends only on the field in $A$,  its second derivative  will not, in general, 
satisfy  the coercivity condition \eqref{eq:Qlowerbound} when $A \subsetneq Q_{\Ran}$.
The addition of another null Lagrangian, however, gives us an energy that has a strictly positive  Hessian at the identity.

To implement the strategy described above,  we first note that the assumption $\{0,1\}^d\subset A$ implies for any   $\psi \in (\R^d)^A$ that
\begin{align*}
|\nabla \psi(0)|^2\leq 2\sum_{i=1}^d (|\psi(0)|^2+|\psi(e_i)|^2)\leq 2d |\psi|^2,
\end{align*}
with $|\cdot|$  the  Euclidean norm on $(\R^d)^A$.
Let $E: (\R^d)^A \to \R$ be the function in Theorem~\ref{th:convexification}.
Uniform convexity of $E$ orthogonal to shifts and shift invariance of $E$ imply that
 there is a constant $\mu>0$ such that, for $\psi\in \Vcal_A^\perp$  and all $G \in \R^{d \times d}$
\begin{align}\label{eq:Hcoerc1}
\begin{split}
E(G+\psi)&\geq E(G)+DE(G)(\psi) +{\mu}|\psi|^2\\ &
\geq E(G)+DE(G)(\psi) +\frac{\mu}{2d}|\nabla\psi(0)|^2.
\end{split}
\end{align}
Since the first and the last expression are shift invariant we conclude that,  for all
$\psi \in (\R^d)^A$ and all $G \in \R^{d \times d}$,
\begin{align}\label{eq:Hcoerc2}
E(G+\psi)\geq  E(G)+DE(G)(\psi) +\frac{\mu}{2d}|\nabla\psi(0)|^2.
\end{align} 

Taking $G$ as the identity, we see that the growth of $E$ is controlled from below by the gradient in one point. 
We now introduce a null Lagrangian that allows 
us to redistribute the gradient lower bound to gain coercivity on $(\R^d)^{Q_{\Ran}}$.
In view of future applications, we state the following lemma for general $m$-dimensional vector fields. 
For our applications in elasticity we only need the case $m=d$.
\begin{lemma}\label{le:N_0} Define
$\nullL_0: (\R^{m})^{Q_{\Ran}}\rightarrow \R$ by
\begin{align}
\nullL_0(\psi)=-  \sum_{i=1}^d \abs{\nabla_i \psi(0)}^2 + \frac{1}{\Ran(\Ran+1)^{d-1}}\sum_{i=1}^d\sum_{y,y+e_i\in Q_{\Ran}} 
\abs{\nabla_{i}\psi(y)}^2.                                                                                           
\end{align}
Then the function $\nullL_0$ is a null Lagrangian and ${\nullL_0}(\psi)=0$ if $\psi$ is the restriction of an affine map.
\end{lemma}
\begin{proof} This is similar to Example~\ref{ex:linear_NL}. Note that
$$
\#\{ y \in Q_{\Ran} : y + e_i \in Q_{\Ran}\} = \Ran(\Ran+1)^{d-1}.
$$ 
Thus
\begin{align} & \nullL_0(\psi) = \frac{1}{\Ran(\Ran+1)^{d-1}}\ \sum_{i=1}^d\sum_{y,y+e_i\in Q_{\Ran}}  \nullL_{y,i}(\psi)  \\
\hbox{where} \quad & \nullL_{y,i}(\psi) =| \nabla_i \psi(y)|^2 - |\nabla_i \psi(0) |^2.
\end{align}
Thus it suffices to show that  for all $y \in Q_{\Ran}$ with $y + e_i \in Q_{\Ran}$ the map  $\nullL_{y,i}: (\R^d)^{Q_{\Ran} }\to \R$ is a null Lagrangian. 
Fix such $y$ and $i$. We use the criterion  \eqref{eq:criterion_null_Lagrangian}. 
 Assume that $\tilde \p = \p$ in $\Z^d \setminus \Lambda$. Take $\Lambda'$ so large that $\Lambda_{Q_{\Ran}} \subset \Lambda'$ and 
$\big( (y + \Lambda') \Delta \Lambda' \big)  \cap \Lambda_{Q_{\Ran}} = \emptyset$. 
Here $\Delta$ denotes the symmetric set difference. 
We have
\begin{align*}
& \, \sum_{x \in \Lambda'}  \nullL_{y,i}(  (\tau_{-x}\p)_{Q_{\Ran}}) =  \sum_{x \in \Lambda'}  \bigl(|\nabla_i \p(y+x)|^2 - |\nabla_i \p(x)|^2 \bigr)\\
= & \,  \sum_{z \in (y + \Lambda') \setminus \Lambda'}  |\nabla_i \p(z)|^2
- \sum_{z \in \Lambda' \setminus (y + \Lambda')} |\nabla_i \p(z)|^2.
\end{align*}
Since $\big( (y + \Lambda') \Delta \Lambda' \big)  \cap \Lambda_{Q_{\Ran}} = \emptyset$,
the sums on the right hand side only involve points $z$ with $z \notin \Lambda_{Q_{\Ran}}$.
If $z \notin \Lambda_{Q_{\Ran}}$ then
$(z + Q_{\Ran}) \cap \Lambda = \emptyset$. In particular $z \notin \Lambda$ and
$z + e_i \notin \Lambda$ and hence $\nabla_i \tilde \p(z) = \nabla_i \p(z)$.
It follows that $\sum_{x \in \Lambda'}  \nullL_{y,i}(  (\tau_{-x}\p)_{Q_{\Ran}}) = 
\sum_{x \in \Lambda'}  \nullL_{y,i}(  (\tau_{-x}\tilde \p)_{Q_{\Ran}})$ and hence
$\nullL_{y,i}$ is a null Lagrangian.

Finally,  if  $\psi$ is the restriction of an affine map, then $\nabla_i \psi(y) = \nabla_i \psi(0)$ and therefore $\nullL_0(\psi) = 0$. 
\end{proof}
\smallskip

We use $\1_{Q_\Ran}$ to denote the identity map on $Q_\Ran$ and
define the  energies $\tilde{U},\tilde{\nullL}, \tilde{E}:(\R^d)^{Q_{\Ran}}\rightarrow \R$  by 
\begin{align}
\tilde{U}(\psi)&=U(\psi_A),    \label{eq:tildeU_de}  \\
\tilde{\nullL}(\psi)&=\nullL(\psi_A)+\frac{\mu}{2d}\nullL_0(\psi  - \1_{Q_0}),   \label{eq:tildeN_de}\\
\tilde{E}(\psi)&=E(\psi_A)+\frac{\mu}{2d}\nullL_0(\psi  - \1_{Q_0}).  \label{eq:tildeE_de}
\end{align}
Those functionals inherit the properties $\tilde{U}+\tilde{\nullL}\geq \tilde{E}$ with equality in a neighbourhood
of rotations (restrictions of rotations are still rotations). Moreover,  from \eqref{eq:Hcoerc2} 
with $G = \1$ we infer that
for any $\psi\in (\R^d)^{Q_{\Ran}}$,
\begin{align}\label{eq:hqfirstway}
\begin{split}
&\, \tilde{E}(\1_{Q_{\Ran}}+\psi)\\
=&\, E(\1_{A}+\psi_A)+\frac{\mu}{2d}\frac{1}{\Ran(\Ran+1)^{d-1}}\sum_{i=1}^d\sum_{x,x+e_i\in Q_{\Ran}}                           
|\nabla_{i}\psi(x)|^2  -\frac{\mu}{2d} |\nabla\psi(0)|^2	\\	
\geq&\, \tilde{E}(\1_{Q_{\Ran}})+D \tilde{E}(\1_{Q_{\Ran}})(\psi) 
+\frac{\mu}{2d}\frac{1}{\Ran(\Ran+1)^{d-1}}\sum_{i=1}^d\sum_{x,x+e_i\in Q_{\Ran}} |\nabla_{i}\psi(x)|^2                                                                              
\end{split}
\end{align}
where we used that $D\nullL_0(0)=0$ and $\nullL_0(0)=0$.

Recalling Lemma~\ref{L:U-Uscr},
we use the isomorphism $\Pi:\Gcal_{\Ran}\to \Vcal_{Q_{\Ran}}^\perp$ to define the functions $\Uscr,\Escr,\Nscr:\Gcal_{\Ran}\to \R$
by
\begin{align}  \label{eq:Uscr_de}	
\Uscr(z)=\tilde{U}(\Pi (z)+\1_{Q_{\Ran}}), \quad 
\Escr(z)=\tilde{E}(\Pi (z)+\1_{Q_{\Ran}}), \quad \hbox{and} \quad 
 \Nscr(z)=\tilde \nullL(\Pi (z)+\1_{Q_{\Ran}}).
\end{align}
In view  of \eqref{eq:propofHineq} and \eqref{eq:propofHeq}  the functions $\Uscr$, $\Escr$, and $\Nscr$ satisfy
\begin{align}\label{eq:propofHscrineq}
\Uscr(z)+\Nscr(z)&\geq \Escr(z),
\\ \label{eq:propofHscreq}
\Uscr(z)+\Nscr(z) &= \Escr(z) \quad \text{for $z$ close to $0$}.
\end{align}
Moreover, the definition of $\Uscr$ implies that 
\begin{align}  
U( (\1 + F+ (\tau_{-x} \p)_A)=\Uscr(\overline{F}+D\p(x)). 
\end{align}
Hence the Hamiltonian for the discrete elasticity model defined in
\eqref{eq:defHF_elastic} can be written as
\begin{align}  
\label{E:HN^1}
\HN^{\1 +F}(\p)=\sum_{x\in T_N} \Uscr(\overline{F}+D\p(x)). 
\end{align}
The functionals $\Uscr$, $\Nscr$, and $\Escr$ are differentiable
since they are a composition of a differentiable map  and a linear map.
Moreover,  \eqref{eq:propofHscrineq}, \eqref{eq:propofHscreq}, and the bound \eqref{eq:hqfirstway} 
imply that there is $\omega_1>0$ such that for all $z\in \Gcal_{\Ran}$
\begin{align}
\begin{split}
\Uscr(z)+\Nscr(z)&\geq \Escr(z)\geq  \Escr(0)+D\Escr(0)(z) + \omega_1 |z|^2 \\ 
& =  \, (\Uscr + \Nscr)(0) +  D(\Uscr + \Nscr)(0)z + \omega_1 |z|^2.   
\label{eq:hqsecway}  
\end{split}                                
\end{align}
where we used that $\sum_{i=1}^d\sum_{y,y+e_i\in Q_{\Ran}} \abs{\nabla_{i}\psi(y)}^2$ 
defines a norm on $\Vcal^\perp_{Q_{\Ran}} \simeq \Gcal_{\Ran}$ and  all norms on  a finite dimensional space are equivalent.

\smallskip

We now proceed as follows. 
We first show that the free energy and the Gibbs measure for the discrete elasticity
model with interaction potential $U$ can be expressed in terms of a generalized gradient model with potential $\Uscr+\Nscr$,
see  Lemma~\ref{le:equivalence_de_ggm} below. 
Next show that under the assumptions (H1)--(H5) the potential $\Uscr + \Nscr$ satisfies the conditions 
in Proposition~\ref{prop:embedding}, see Lemma~\ref{le:embedding_de} below. 
Once this is done the main results for discrete elasticity,  Theorem~\ref{T:deW}  and Theorem~\ref{T:descaling}, 
will follow easily from the abstract perturbation results,
Theorem~\ref{th:pertcomp} and Theorem~\ref{th:scalinglimit}.

As in \eqref{eq:define_QU_V}  and \eqref{eq:defofUbar}, we define the quadratic part
\begin{align}
\Qscr_{\Uscr + \Nscr}(z) := D^2(\Uscr + \Nscr)(0)(z,z)
\end{align}
and the function 
\begin{align}
\label{eq:defofU+Nbar}
\overline{(\Uscr+\Nscr)}(z,F)=(\Uscr+\Nscr)(z+\overline{F})- (\Uscr+\Nscr)(\overline{F})- 
D(\Uscr+\Nscr)(\overline{F})(z)-\frac{\Qscr_{\Uscr + \Nscr}(z)}{2}.
\end{align}
Note that \eqref{eq:hqsecway} implies 
\begin{align} \label{eq:lower_bound_Q_U+N}
 \Qscr_{\Uscr + \Nscr}(z) \geq 2\omega_1 |z|^2
\end{align}
Since $U$ and $N$  and  hence $\Uscr$ and $\Nscr$ are $C^2$ we also have
\begin{align} \label{eq:upper_bound_Q_U+N}
 \Qscr_{\Uscr + \Nscr}(z) \leq \frac{1}{\omega_2} |z|^2
\end{align}
for some $\omega_2 > 0$. 

Lemma~\ref{le:null_lagrangian_periodic} implies that
\begin{align}\label{eq:Nprop}
	L^{Nd}\Nscr(\overline{F})=\sum_{x\in T_N} \Nscr(D\p(x)+\overline{F}). 
\end{align}
Writing $\Uscr = - \Nscr + (\Uscr + \Nscr)$ and using   \eqref{E:HN^1}, \eqref{eq:Nprop}, and \eqref{eq:defofU+Nbar}, we find
\begin{align}
\begin{split}\label{eq:EandV}
\HN^{\1 +F}(\p) = -L^{Nd}\Nscr(\overline{F}) & + \sum_{x\in T_N}(\Uscr+\Nscr)(D\p(x)+\overline{F})
\\= L^{Nd}\Uscr(\overline F) + &\sum_{x\in T_N}\overline{(\Uscr+\Nscr)}(D\p(x),F)
\\&+ \sum_{x\in T_N}\Bigl(D(\Uscr+\Nscr)(\overline{F})(D\p(x))+\frac{\Qscr_{\Uscr + \Nscr}(D\p(x))}{2}\Bigr)
\\ = L^{Nd}\Uscr(\overline{F}) + &\sum_{x\in T_N} 
\Bigl(\overline{(\Uscr + \Nscr)}(D\p(x),F)+\frac{\Qscr_{\Uscr + \Nscr}(D\p(x))}{2} \Bigr). 
\end{split}
\end{align}
In the last equality we used  the identity $\sum_{x\in T_N}D\p(x)=0$. 
As a result,  the partition function for the discrete elasticity model defined in \eqref{eq:defofZ} can be expressed  as
\begin{align}  \label{eq:de_partition_rewritten}
Z_{\beta,N}^U(\1 + F,  f)) = e^{-\beta L^{Nd}\Uscr(\overline{F})}
Z_{\beta,N}^{\Qscr_{\Uscr + \Nscr}} 
\Zcal_{\beta,N}^{\Uscr + \Nscr}\Bigl(F, \frac{f}{\sqrt \beta}\Bigr) 
\end{align}
where 

\begin{align}   \label{eq:Gaussian_part_Qscr_Uscr_plus_Nscr}
Z_{\beta,N}^{\Qscr_{\Uscr + \Nscr}} = 
\int_{\Xcal_N}  e^{-\frac12 \beta \sum_{x \in T_N} \Qscr_{\Uscr+\Nscr}(D\p(x))} \, \mu(\d \p)   
\end{align}
and
\begin{align}\label{eq:Zpertcomp_rewritten}
\Zcal_{\beta,N}^{\Uscr +\Nscr}(F,f):=
\int_{\Xcal_N}e^{(f,\p)}\sum_{X\subset T_N}\prod_{x\in X}K_{F,\beta,\Uscr+\Nscr}(D\p(x))\, \mu(\d \p),                                                             
\end{align}
with $K_{F, \beta, \Uscr + \Nscr}$  defined by replacing $\Uscr$ by $\Uscr + \Nscr$ and 
$\overline{\Uscr }$ by $\overline{\Uscr + \Nscr}$ in   \eqref{eq:defofKuVbeta}, 
\eqref{eq:defofUbar} and \eqref{eq:define_QU_V}.
The calculations so far can be summarised as follows.

\begin{lemma} \label{le:equivalence_de_ggm}
Let $Z_{\beta,N}^{U}(F,0)$ denote the partition function of the discrete elasticity model
with interaction $U$ and deformation $F$, let $\gamma_{\beta,N}^{F, U}$ denote the corresponding
finite volume Gibbs measure, let 
\begin{align*}  
W_{\!\beta, N}^U(F) = - \frac{\ln Z_{\beta,N}(F, 0)}{\beta L^{Nd}},
\end{align*}
and let 
$$ 
\Wcal^{U}_{\beta,N}(F) = \beta(W_{\!\beta, N}^U(F) - U(F)) + \frac{\ln Z^{Q_U}_{\beta,N}}{L^{Nd}}
$$
be the perturbative component of the free energy introduced in \eqref{eq:defvarsigmaN}. 
 Let $Z_{\beta,N}^{\Uscr + \Nscr}(F,0)$ denote the partition function of the generalized
gradient model with interaction $\Uscr + \Nscr$ and deformation $F$,  let $\gamma_{\beta,N}^{F, \Uscr + \Nscr}$
be the corresponding Gibbs measure, and let 
$$ 
\Wcal^{\Uscr + \Nscr}_{\beta,N}(F)  = - \frac{  \ln \Zcal_{\beta,N}^{\Uscr + \Nscr}(F, 0)    }{   L^{Nd}    }
$$
be  the quantity in \eqref{eq:defvarsigmaN}.
Then 
\begin{align}  \label{eq:equivalence_Z}
Z_{\beta,N}^{U}(\1+ F,0) = & \,  e^{\beta L^{Nd} \Nscr( \overline F)}   Z_{\beta,N}^{\Uscr + \Nscr}(F,0),\\
  \label{eq:equivalence_myW}
 \Wcal^{U}_{\beta,N}(\1 + F) = &\,  \Wcal^{\Uscr + \Nscr}_{\beta,N}(F),\\
\label{eq:equivalence_laplace} 
\mathbb{E}_{   \gamma_{\beta,N}^{\1 + F, U}} e^{(f, \p)} = & \, \mathbb{E}_{ \gamma_{\beta,N}^{F, \Uscr + \Nscr} } e^{(f,\p)}.
\end{align}
\end{lemma}

\begin{proof} 
Equation \eqref{eq:initialfinal}, applied to $\Uscr + \Nscr$ instead of $\Uscr$ gives
\begin{align}  
\begin{split}
 Z_{\beta,N}^{\Uscr + \Nscr}(F,f) =\hspace{5pt} &   e^{-\beta L^{Nd}(\Uscr(\overline{F}) + \Nscr(\overline F))}   Z_{\beta,N}^{\Qscr_{\Uscr + \Nscr}} 
\Zcal_{\beta,N}^{\Uscr + \Nscr}\Bigl(F,\frac{ f}{\sqrt \beta}\Bigr)  \\
\underset{ \eqref{eq:de_partition_rewritten}}{=}  &
e^{-\beta L^{Nd}\Nscr(\overline{F})}     Z_{\beta,N}^{U}(\1+ F,f)    \label{eq:consequence_Zpertcomp}
\end{split}
\end{align}
Taking $f=0$ we get \eqref{eq:equivalence_Z}.
Dividing both sides by the corresponding expression for $f=0$ we get  \eqref{eq:equivalence_laplace}.
It follows from  
\eqref{eq:de_partition_rewritten} that 
\begin{align}
\begin{split}
 \, L^{Nd} \Wcal^{\Uscr + \Nscr}_{\beta,N}(F)& = - \ln Z_{\beta,N}^{U}(\1+ F,0) - \beta L^{Nd} \Uscr(\overline F) + \ln Z_{\beta,N}^{\Qscr_{\Uscr + \Nscr}}\\
 &=  \,  L^{Nd} \Wcal^{U}_{\beta,N}(\1 + F) +  \ln Z_{\beta,N}^{\Qscr_{\Uscr + \Nscr}} -  \ln Z_{\beta,N}^{\Qscr_{\Uscr}}
\end{split} 
 \end{align}
Thus \eqref{eq:equivalence_myW} follows if we can show that
\begin{align}  \label{eq:independence_Nscr_1}
Z_{\beta,N}^{\Qscr_{\Uscr + \Nscr}} = Z_{\beta,N}^{\Qscr_{\Uscr}}.
\end{align}
To prove  \eqref{eq:independence_Nscr_1},  note that Lemma~\ref{le:null_lagrangian_periodic}      implies that
$$ \sum_{x \in T_N} \Nscr(s D\p(x) + t D\p(x)) = \sum_{x \in T_N} \Nscr(0) = 0.$$
Taking the derivative with respect to $s$ and $t$ at $s=t=0$ we obtain that
$ \sum_{x \in T_N} \Qscr_\Nscr(D\p(x)) = 0$. This yields
\eqref{eq:independence_Nscr_1}.
\end{proof}

We now prove that the potential $\Uscr + \Nscr$ satisfies the conditions in Proposition \ref{prop:embedding}
so that we can apply the results from the previous chapter.

\begin{lemma}  \label{le:embedding_de}
Under the hypotheses (H1), (H2), (H3), (H4), and (H5) 
in Section~\ref{sec:discrete_elasticity_main} the function $\Uscr + \Nscr$ 
satisfies the assumptions of Proposition \ref{prop:embedding}, i.e.,
\begin{align}
\label{eq:V12}    \Uscr+\Nscr\in C^{r_0+r_1} (\Gcal_{\Ran}),                                                                      \\ 
\label{eq:V1bis2}  \omega_0 |z|^2 \le \Qscr_{\Uscr + \Nscr}(z) \le \omega_0^{-1}|z|^2, \\
\label{eq:V32}     \Uscr(z)+\Nscr(z)  -  \bigl(\Uscr(0)-\Nscr(0)\bigr) - D\bigl(\Uscr(0)+\Nscr(0)\bigr) (z) \geq \omega_1 |z|^2 ,  
\text{ and}                  \\ 
\label{eq:V42}   \lim_{t \to \infty} t^{-2} \ln \Psi(t) = 0 \quad \hbox{where} 
\quad \Psi(t) := \sup_{|z| \le t}  \sum_{3 \le |\alpha| \le r_0+r_1}   \frac{1}{\alpha!} |\partial^\alpha \bigl(\Uscr(z)+\Nscr(z)\bigr)|.
\end{align}	
with $\omega_0 = \min(2 \omega_1, \omega_2)$, where $\omega_1$ and $\omega_2$ are the constants in
\eqref{eq:hqsecway} and   \eqref{eq:upper_bound_Q_U+N}, respectively.
\end{lemma}

\begin{proof}
The first condition is a consequence of the smoothness of $\Qscr$, $\Nscr$, and $\Uscr$
which follows by the chain rule from the smoothness of $U$ postulated in (H5) and the smoothness of the polynomial $\nullL$.  
The second condition follows from  \eqref{eq:lower_bound_Q_U+N} and   \eqref{eq:upper_bound_Q_U+N}.
The third condition follows from  \eqref{eq:hqsecway}.
The last condition follows from the fact that the 
$\Uscr$-term is controlled by (H5) and the chain rule and  that $\Nscr$ is a polynomial.
\end{proof}

Finally we  show how to  deduce the results for the discrete elasticity model  from
the results for abstract perturbations of quadratic Hamiltonians.

\begin{proof}[Proof of Theorem~\ref{T:deW}]
\emph{Uniform smoothness of $W^{U}_{\!\beta,N}$ in $B_\delta(\1)$ and convergence}.\\
By \eqref{eq:equivalence_myW} and the definition of $\Wcal^{U}_{\!\beta, N}$  (see
 Lemma~\ref{le:equivalence_de_ggm}) we have
 \begin{align}   \label{eq:expansion_W_beta_elasticity}
 W^{U}_{\!\beta,N}(\1 + F)  = & \, U(\1+ F) + \frac1\beta \Wcal^{U}_{\!\beta,N}(\1+ F) - c_{\beta,N} \\
 = & \, \, U(\1+ F) + \frac1\beta \Wcal^{\Uscr + \Nscr}_{\!\beta,N}(F) - c_{\beta,N}  \nonumber 
 \end{align}
 where 
 \begin{align}  \label{eq:proof_2dot6_c_beta_N}  c_{\beta,N}  = 
 \frac{1}{\beta L^{Nd} } \ln Z^{Q_\Uscr}_{\beta,N} =  \frac{1}{\beta L^{Nd} } 
 \int_{\Xcal_N} e^{- \beta \sum_{x \in T_N} Q_{\Uscr}(D\p(x))}  \, \lambda_N(d\p)
 \end{align}
Recall that, for $G \in \R^{d \times d}$ we use the shorthand notation  $U(G)$ for the expression $U(\psi_G|_{A})$ where $\psi_G(x) = G x$.

 By Lemma~\ref{le:embedding_de} the potential $\Uscr + \Nscr$ satisfies the assumption of Proposition \ref{prop:embedding}
(with $\omega := \frac{\omega_0}{8}$).
Thus Lemma~\ref{le:embedding_de} and  Theorem~\ref{th:pertcomp} imply that 
there exists a $\delta >0$ such that the derivatives of $\Wcal^{\Uscr + \Nscr}_{\!\beta,N}$
up to order $r_1$ are uniformly bounded. The assertion about convergence of subsequences 
follow from the Arzela-Ascoli theorem.

\medskip
\emph{Convexity properties of $W^{U}_{\!\beta,N}$ and $W^{U}_{\!\beta,N} + \nullL$.}

 By (H3)  the  quadratic form $D^2 U(\1)$ is positive definite on symmetric matrices.
 Thus  
restriction of $U$ to $B_\delta(\1) \cap \R^{d \times d}_{\mathrm{sym}}$ is uniformly convex if $\delta > 0$ is chosen sufficiently small. Hence   \eqref{eq:expansion_W_beta_elasticity} and the uniform bounds
 on $D^2 \Wcal^{\Uscr + \Nscr}_{\!\beta,N}$  in $B_\delta(\1)$ imply that 
$W^{U}_{\!\beta,N}$ is uniformly convex on $B_\delta(\1) \cap \R^{d \times d}_{\mathrm{sym}}$
if  $\beta \ge \beta_0$ with $\beta_0$ sufficiently large.

By  Theorem~\ref{th:convexification}, the function $U+ \nullL$ agrees with a uniform convex function $E$ in $B_\delta(\1)$ if $\delta > 0$
is chosen sufficiently small, Thus the uniform  bounds
 on $D^2 \Wcal^{\Uscr + \Nscr}_{\!\beta,N}$ in $B_\delta(\1)$ imply that
 $W_{\!\beta,N}^{U} +\nullL$ is uniformly convex in $B_\delta(\1)$ if $\beta \ge \beta_0$, for sufficiently large $\beta_0$.
 
 \medskip
 \emph{Uniform quasiconvexity at points in $B_{\delta/2}(\1)$}.\\
We need to show  that there
exists a $c_0 > 0$ such that  for all open bounded set $\Omega$ and
all $C^1$ function $\phi: \Omega \to \R^d$ with compact support on $\Omega$ we have
\begin{equation} \label{eq:strict_qc_proof}
\int_\Omega W^{U}_{\!\beta, N}(G + \nabla \phi) -  W^{U}_{\!\beta, N}(G) \, dx \ge c_0 \int_\Omega
|\nabla \phi|^2 \, dx.
\end{equation}
whenever $G \in B_{\delta/2}(\1)$

We first note that that it suffices to show that
\begin{equation} \label{eq:global_convexity_W_plus_N}
(W_{\!\beta,N}^{U} + \nullL)(G+G') - (W_{\!\beta,N}^{U} + \nullL)(G) -   D(W_{\!\beta,N}^{U} + \nullL)(G) G'
\ge c_0 |G'|^2,
\end{equation}
for all $G \in B_{\delta/2}(\1)$ and all $G' \in \R^{d \times d}$, whenever  $\beta \ge \beta_0$ with $\beta_0$ sufficiently
large. 
Indeed, it suffices to apply \eqref{eq:global_convexity_W_plus_N} with $G' = \nabla \phi(x)$, integrate over
$x$ and use that
$$ \int_\Omega  \nabla \phi(x) \, dx = 0,$$
by the divergence theorem,  and
$$ \int_\Omega \nullL(G + \nabla \phi) - \nullL(G) \, dx  = 0$$
by Proposition~\ref{pr:null_lagrangian_continuous},   since $\nullL$ is a multiple of the determinant.

For $|G'| < \delta/2$ the inequality  \eqref{eq:global_convexity_W_plus_N}
follows from that fact that $W_{\!\beta,N}^{U} + \nullL$ is uniformly convex in $B_\delta(\1)$. 

For $|G'| \ge \delta/2$ we argue as follows.
By  Theorem~\ref{th:convexification},  $U + N \ge E$ for a $C^2$ map $E$ which is uniformly convex on $\Vcal_A^\perp$. Since $\nullL$ is a discrete null Lagrangian, it follows directly from the definition 
of $W_{\!\beta,N}^{U}$   (see \eqref{eq:defWN} and  \eqref{eq:defofZ}) that
\begin{align*}
W_{\!\beta,N}^{U} + \nullL = W_{\!\beta,N}^{U + \nullL} \ge W_{\!\beta,N}^{E} 
\end{align*}
We will show below that the uniform convexity of $E$ on $\Vcal_A^\perp$ implies that
\begin{equation}   \label{eq:lower_bound_free_E}
 W_{\!\beta,N}^{E} \ge E - c'_{\beta,N} 
\end{equation}
with 
\begin{equation}
c'_{\beta,N} := \frac{1}{\beta L^{Nd}}   \ln  \int_{\Xcal_N} e^{-\beta c_2 \sum_{x\in T_N}  |\nabla \p|^2(x)} d\lambda_N
\end{equation}
for some $c_2 > 0$.

On the other hand it follows from   \eqref{eq:expansion_W_beta_elasticity}, the uniform bounds on
$ \Wcal^{\Uscr + \Nscr}_{\!\beta,N}$ and on $D \Wcal^{\Uscr + \Nscr}_{\!\beta,N}$ in $B_\delta(\1)$, 
and the identity $U+N = E$ in $B_\delta(\1)$ that
\begin{eqnarray}   \label{eq:comparison_W_U_plusN_E}
(W_{\!\beta,N}^{U} + \nullL)(G) - E(G) &\le& c\beta^{-1} - c_{N,\beta}, \\
|D(W_{\!\beta,N}^{U} + \nullL)(G) - DE(G)| &\le& c \beta^{-1}.
  \label{eq:comparison_DW_U_plusN_E}
\end{eqnarray}
The uniform convexity of $E$ on  $\Vcal_A^\perp$ implies
that, for all $G' \in \R^{d\times d}$, 
\begin{equation}  \label{eq:convexity_E}
E(G+G')  - E(G) - DE(G) G' \ge c_1  |G'|^2.
\end{equation}
Thus 
\begin{align} \label{eq:global_convexity_W_plus_N_bis}
& \, (W_{\!\beta,N}^{U} + \nullL)(G+G') - (W_{\!\beta,N}^{U} + \nullL)(G) -   D(W_{\!\beta,N}^{U} + \nullL)(G) G'  \\
\ge & \, 
c_1 |G'|^2 - c'_{N,\beta} + c_{N, \beta} -  \frac{c}{\beta}( 1+ |G'|)
\end{align}
Thus,  for $c_0 \in (0, c_1)$,  we obtain the desired estimate \eqref{eq:global_convexity_W_plus_N}
for $|G'| \ge \delta_0/2$ and $\beta \ge \beta_0$ with $\beta_0$ sufficiently large, 
provided that we can show that
\begin{equation}   \label{eq:bound_cN_beta}
c_{\beta,N} - c'_{\beta,N}   \ge - c \beta^{-1}  
\end{equation}
for $\beta \ge \beta_0$.

Thus it only remains to show   \eqref{eq:lower_bound_free_E} and  \eqref{eq:bound_cN_beta}.
We first show   \eqref{eq:lower_bound_free_E}.
By \eqref{eq:Hcoerc2} we have
\begin{equation} \label{eq:pointwise_lower_bound_D2E}
E(F + \psi_A) \ge E(F) + DE(F)   \psi_A \ge c_2 |\nabla \psi(0)|^2  \quad \text{for all $\psi \in \Xcal_N$,}  
\end{equation}
with $c_2 = \mu/(2d)$.

By definition, every $\p \in \Xcal_N$ satisfies $\sum_{x \in T_N} \p(x) = 0$. Thus
$\sum_{x \in T_N} (\tau_{-x} \p)_A = 0$. 
Thus the definition    \eqref{eq:defHF_elastic}  of the Hamiltonian  $\HN^{E,F}$ and  \eqref{eq:pointwise_lower_bound_D2E}
imply that 
\begin{align*}
\HN^{E,F}(\p) =  & \,  \sum_{x\in T_N} E((\tau_{-x} \p)_A +F)  \ge L^{Nd} E(F) + c_2 \sum_{x \in T_N} |\nabla (\tau_{-x} \varphi)(0)|^2  \\ 
=  & \, 
L^{Nd} E(F) + c_2 \sum_{x \in T_N} |\nabla \varphi(x)|^2 .
\end{align*}
It follows that
\begin{align}   \label{eq:estimate_W_E_free}
W^{E}_{\!\beta,N}(F) = & \,  -\frac{1}{\beta L^{Nd}}  \ln \int_{\Xcal_N} e^{-\beta \HN^{E,F}(\p)} \lambda_N(d\p) 
 \\
\ge & \, E(F)   -\frac{1}{\beta L^{Nd}}   \ln  \int_{\Xcal_N} e^{-\beta c_2 \sum_{x\in T_N}  |\nabla \p|^2(x)} \lambda_N(d\p)
\nonumber
\end{align}
and this proves  \eqref{eq:lower_bound_free_E}.

To show  \eqref{eq:bound_cN_beta}, we  first use the
 the change  of variables $\p \mapsto \sqrt{\beta c_2} \p$ in the definition of $c'_{\beta,N}$.
Since $\lambda_N$ is the $L^{Nd}-1$ dimensional Hausdorff measure on $\Xcal_N$, we get 
 \begin{align}   \label{eq:bound_cN_prime_beta}  
 c'_{\beta,N} =  & \, \frac{1}{\beta L^{Nd}}  
 \ln  \int_{\Xcal_N} e^{-\beta c_2 \sum_{x\in T_N}  |\nabla \p|^2(x)}  \lambda_N(d\p)
\\
=  & \,  \frac1\beta \frac{L^{Nd}-1}{L^{Nd}} 
  \ln \sqrt{\beta c_2} + 
  \frac{1}{\beta L^{Nd}}   \ln  \int_{\Xcal_N} e^{- \sum_{x\in T_N}  |\nabla \p|^2(x)} \lambda_N(d\p).
   \nonumber 
  \end{align}
Next,  we derive a lower bound for the quantity $c_{\beta,N}$ given by \eqref{eq:proof_2dot6_c_beta_N}.
To get an upper bound for $Q_\Uscr(D\p(x))$,  we note that by assumption the elastic interaction energy
$U: (\R^d)^A \to \R$ is $C^2$. Hence the energy $\Uscr: \Gcal \to \R$ defined by 
  \eqref{eq:tildeU_de}  
   and \eqref{eq:Uscr_de} is also $C^2$. It follows that $Q_\Uscr(z) = D^2 \Uscr(0)(z,z)$ satisfies the estimate
$ Q_\Uscr(z) \le c |z|^2$.
 Now recall that
the generalized gradient $D\p(x)$ can be expressed as a linear combinations the values  $\nabla \p(y)$  
for $y \in x + Q_{R_0}$. Thus we have the  estimate
$$ \sum_{x \in T_N} |D\p(x)|^2 \le C \sum_{x \in T_N} |\nabla \p(x)|^2.$$
Hence 
$$ \sum_{x \in T_N} Q_\Uscr(D\p(x)) \le c_3 \sum_{x \in T_N} |\nabla \p(x)|^2,$$
and using the  change of variables $\p \mapsto \sqrt{ \beta c_3} \p$ we get,  as in   \eqref{eq:bound_cN_prime_beta}, 
\begin{align*}
c_{\beta,N} = & \,  \frac{1}{\beta L^{Nd}} \int_{\Xcal_N}  e^{-\beta \sum_{x \in T_N} 
Q_{\Uscr}(D\p(x))} \, \lambda_N(d\p)  \\
\ge & \,   
\frac1\beta \frac{L^{Nd}-1}{L^{Nd}} 
  \ln \sqrt{\beta c_3} + 
  \frac{1}{\beta L^{Nd}}   \ln  \int_{\Xcal_N} e^{- \sum_{x\in T_N}  |\nabla \p|^2(x)} \lambda_N(d\p)
\end{align*}
Subtracting the identity \eqref{eq:bound_cN_prime_beta}
 for $c'_{\beta,N}$,  we get  \eqref{eq:bound_cN_beta}.
\end{proof}

\begin{proof}[Proof of Theorem~\ref{T:descaling}]
By Lemma~\ref{le:embedding_de} the potential $\Uscr + \Nscr$ satisfies the assumption of Proposition \ref{prop:embedding}
(with $\omega := \frac{\omega_0}{8}$) and thus  Theorem~\ref{th:scalinglimit_concrete} can be applied to the
generalized gradient model. Together with \eqref{eq:equivalence_laplace} 
this gives
\begin{align}
\lim_{\ell\rightarrow \infty}    \mathbb{E}_{ \gamma_{\beta,N_\ell}^{\1 + F} }  e^{(f_{N_{\ell}},\p)}= 
e^{\frac1{2\beta}(f,\mathscr{C}_{\mathbb{T}_d}f)}                                                        
\end{align}
where $\mathscr{C}_{\mathbb{T}_d}$ is the inverse of the operator $\mathscr{A}_{\mathbb{T}_d}$ given by
\begin{align} \label{eq:limiting_operator3}
(\mathscr{A}_{\mathbb{T}_d}u)_s=-\sum_{t=1}^d \sum_{i,j=1}^d (\boldsymbol{Q}_{\Uscr+\Nscr} -\boldsymbol{q})_{i,j;s,t} \partial_i\partial_ju_t. 
\end{align}

It only remains to show that in the definition of $\mathscr A$ we may replace $\boldsymbol{Q}_{\Uscr+\Nscr}$
by $\boldsymbol{Q}^\nabla_U$.

First note that  the operator $\mathscr A_{\mathbb{T}_d}$ depends only on the action of $\boldsymbol Q_{\Uscr + \Nscr}$ 
on the subspace $\mathcal G^\nabla_{\Ran}$.
Now each $z \in \mathcal G^\nabla_{\Ran}$ is of the form $z = \overline{F} = DF$ where $F: \R^d \to \R^d$ is linear. By the definition 
of $\Uscr$ and $\Nscr$, see  \eqref{eq:tildeU_de}--\eqref{eq:tildeE_de} and \eqref{eq:Uscr_de}, we have
\begin{align}
(\Uscr + \Nscr)(\overline F) = (U + \nullL)(\1_A + F_A) + \frac{\mu}{2d} \nullL_0(\1_{Q_{\Ran}} +
 F_{Q_{\Ran}}) = (U + \nullL)(\1_A + F_A)
\end{align}
since $\nullL_0$ vanishes on linear maps. 
Thus 
\begin{align}
\boldsymbol{Q}_{\Uscr + \Nscr}(\overline F) = D^2(U + \nullL)(\1)(F_A, F_A) = \boldsymbol{Q}^\nabla_{U+\nullL}(F)
\end{align}
where we used the definition \eqref{eq:define_bsQ_nabla}  of $\boldsymbol{Q}^\nabla_{U+\nullL}$.
It follows that  the operator $\mathscr A_{\mathbb{T}_d}$ can be written as
\begin{align}  \label{eq:limiting_operator4}
(\mathscr{A}_{\mathbb{T}_d}u)_s=-\sum_{t=1}^d \sum_{i,j=1}^d (\boldsymbol{Q}_{U+\nullL}^\nabla-\boldsymbol{q})_{i,j;s,t} \partial_i\partial_j u_t. 
\end{align}
Now $\boldsymbol Q_{U +\nullL}^\nabla = \boldsymbol Q_U^\nabla + \boldsymbol Q_\nullL^\nabla$ and it only remains to show that $\boldsymbol Q_\nullL^\nabla$ generates the zero operator. Multiplying by a test function $g \in C^\infty(\mathbb{T}^d, \R^d)$, 
denoting the scalar product on $\R^d$ by $\cdot$, recalling that $\nullL(F_A) = \alpha \det F$ and
using  that  $\det$ is a null Lagrangian  on maps defined on 
$\mathbb{T}^d$, i.e. $\int_{\mathbb{T}^d} (\det(\1 + \nabla h) - \det \1 ) \, \d x = 0$  for all $h \in C^\infty(\mathbb{T}^d, \R^d)$, 
we get
\begin{align}
&- \int_{\mathbb{T}^d}  g \cdot \sum_{i,j=1}^d (\boldsymbol{Q}_{\nullL}^\nabla)_{i,j} \partial_i\partial_j f \, \d x 
=    \int_\mathbb{T^d}\sum_{i,j=1}^d \partial_i g  \cdot (\boldsymbol{Q}_\nullL^\nabla)_{ij}\partial_jf \\
= & \, \int_\mathbb{T^d}   \alpha D^2\det(\1)(\nabla f,\nabla g)                               
=\frac{\d}{\d s}\frac{\d}{\d t}_{| s=t=0} \int_\mathbb{T^d} \alpha \det(\1+s\nabla f+t\nabla g)=0.                                
\end{align}
Thus in \eqref{eq:limiting_operator4} we may replace $\boldsymbol Q_{U + \nullL}^\nabla$ by $\boldsymbol Q_U^\nabla$ and this finishes the proof of 
Theorem~\ref{T:descaling}.
\end{proof}
	
\begin{remark}\label{rem:coercive}
Completely independent from the analysis of discrete elasticity,  the null Lagrangian $\nullL_0$ introduced in 
Lemma~\ref{le:N_0}
can be used to gain coercivity in generalized gradient models with
$\R^m$ valued fields.
Recall that for $\Uscr: \Gcal \to \R$ we defined $\Qscr_{\Uscr}(z) = D^2 \Uscr(z,z)$. 
Recall also that for  $z \in \Gcal$ the expression  $z^\nabla$ denotes  the orthogonal projection of $z$ to the space 
$\Gcal^\nabla$ of gradients introduced in \eqref{eq:G_nabla}.
Assume that $\Uscr$ satisfies the conditions \eqref{eq:V1_new}  and \eqref{eq:V4_new},
but instead of  \eqref{eq:V1bis_new} and \eqref{eq:V3_new}  the function $\Uscr$ only satisfies the following weaker
condition.
There exists
$\omega'_0$ and $\omega' \in (0, \frac{\omega'_0}{8})$ such that 
\begin{equation}  \label{eq:V1bis_new_weakened}
\omega'_0 \abs{z^\nabla}^2 \le \Qscr_\Uscr(z) \le {\omega'_0}^{-1} \abs{z}^2  \quad \text{for all $z \in \Gcal$,}
\end{equation} 
\begin{equation}
\label{eq:V3_new_weakened}    
 \Uscr(z) - D\Uscr(0) z - \Uscr(0)  \geq  \omega'\abs{z^\nabla}^2   \text{ for all } z \in \mathcal G.
 \end{equation}
 Thus,  in comparison with  \eqref{eq:V1bis_new} and \eqref{eq:V3_new}, in the lower bounds $z$ has been replaced by $z^\nabla$.

 Let $\nullL_0$ denote the null Lagrangian in Lemma~\ref{le:N_0} and define $\Nscr_0$ by
  $\Nscr_0(z) =\nullL(\Pi(z))$, where $\Pi: \Gcal \to \Vcal^\perp$ is the linear  isomorphism introduced in 
 Lemma~\ref{L:U-Uscr}. Then $N_0(\psi) = \Nscr_0(D\psi(0))$ where $D\psi$ denotes the extended gradient
 introduced after \eqref{E:scprodG}. 
 Note that  $\psi \mapsto D\psi(0)$ is an isomorphism from $\Vcal^\perp$ to $\Gcal$ and that 
 for $z = D\psi(0)$ we have $z^\nabla = (\nabla_1 \psi(0), \ldots \nabla_d \psi(0))$. 
 Since $\psi \mapsto \sum_{i=1}^d\sum_{y,y+e_i\in Q_{\Ran}} 
\abs{\nabla_{i}\psi(y)}^2$ is a positive definite quadratic form on $\Vcal^\perp$ and since $\Pi$ and $\Pi^{-1}$ 
 are bounded as linear maps between finite dimensional vector spaces, it follows from  Lemma~\ref{le:N_0}  that
 $$ \mathcal \Nscr_0(z) \ge c_1 |z|^2 -  c_2 \abs{z^\nabla}^2$$ 
 for some $c_1, c_2 > 0$.      
 
Set 
$$ 
\Uscr' = \Uscr + \frac{\omega'}{c_2} \Nscr_0.
$$
Since $\Nscr(0) = 0$ and $D\Nscr(0) = 0$, the assumption \eqref{eq:V3_new_weakened}   implies that
\begin{equation}
 \Uscr'(z) - D\Uscr'(0) z - \Uscr'(0)  \geq \frac{ \omega' c_1}{c_2}\abs{z}^2   \text{ for all } z \in \mathcal G.
 \end{equation}
 Similarly we get
 \begin{equation} Q_{\Uscr'}(z) \ge \frac{ \omega' c_1}{c_2}\abs{z}^2   \text{ for all } z \in \mathcal G.
 \end{equation}
 Thus $\Uscr'$ satisfies  \eqref{eq:V1bis_new} and \eqref{eq:V3_new} with $\omega_0 = \frac{ \omega' c_1}{c_2}$
 and $\omega = \frac{ \omega' c_1}{16 c_2}$. Since $\Nscr_0$ is quadratic, the derivaties of order three and higher agree for $\Uscr$ and $\Uscr'$. 
 Hence $\Uscr'$ also satisfies  \eqref{eq:V1_new} and \eqref{eq:V4_new}.
 It follows that $\Uscr'$ satisfies the assumptions 
 of Theorems~\ref{thm:strictconvexity} and~\ref{th:scalinglimit_concrete}.
 
 Since $\nullL_0$ is a null Lagrangian and $\nullL_0$ vanishes on affine maps (i.e. $\Nscr_0(\overline F) = 0$ for a linear map $F$) 
 we get, as in  \eqref{eq:consequence_Zpertcomp},
 \begin{equation}  \label{eq:nullL_0_generalized_Z}
  Z_{\beta,N}^{\Uscr'}(F,f) =  Z_{\beta,N}^{\Uscr}(F,f) 
 \end{equation}
 Taking $f= 0$ we see that, in particular, 
 \begin{equation}
 W^{\Uscr'}_{\beta, N}(F) = W^{\Uscr}_{\beta, N}(F) 
 \end{equation}
Dividing  \eqref{eq:nullL_0_generalized_Z} by the corresponding identity for $f= 0$ we see
that
\begin{equation}
\mathbb{E}_{   \gamma_{\beta,N}^{F, \Uscr'}} e^{(f, \p)} =  \, \mathbb{E}_{ \gamma_{\beta,N}^{F, \Uscr} } e^{(f,\p)}.
\end{equation}
Hence the conclusions of Theorems~\ref{thm:strictconvexity} and~\ref{th:scalinglimit_concrete} hold for $\Uscr$
and not just for $\Uscr'$. 
\end{remark}

\section{Examples}  \label{se:elasticity_examples}

Our theory applies to a wide range of interactions, including  two-body interactions (as in \cite{FT02})
as well as multibody interactions which appear for example 
 in the Stillinger-Weber potential \cite{SW85} or the Tersoff potential \cite{Ter88}.
 Our definition of the Hamiltonian in  \eqref{eq:defHF} and our  hypotheses (H1) to (H5)  in Section~\ref{sec:discrete_elasticity_main} do, however, 
 impose some structural restrictions on the interaction potentials.
 
First, as in \cite{EM07} or \cite{LXY21}, we need to 
assume that the potentials have finite range, i.e., atoms which have a sufficiently large distance in the 
\emph{reference} lattice $\Z^d$ do not interact. 
As pointed out  in \cite[p.\ 1714]{LXY21},  the finite range assumption is common in atomistic simulations.

Secondly, in contrast to \cite{EM07, LXY21},  we
need to assume in addition a polynomial  lower growth bound of order $d$ for large discrete gradients, see \eqref{eq:H4_growth_at_infinity}.
The lower growth bound enters in two ways. We need a quadratic lower bound to control the contribution of
large fields, see the discussion of large field regulators in Section~\ref{se:main_new_ideas}.
Actually,  a certain quadratic lower bound for the total energy is also implicit in \cite{LXY21} through the use of the harmonic approximation. 

The stronger condition  of a lower growth bound of the order of the dimension $d$ comes from the fact that we want to express  
the free energy as convex function plus a null Lagrangian, namely the discrete determinant, see \cite{CDKM06}. 
Such lower growth conditions are common in continuum elasticity to ensure the existence of global minimisers.

 A third, more technical condition included in our hypotheses is that the potential should be
sufficiently regular, namely of class $C^{r_0 + r_1}$ with  $r_0 \ge 0$ and $r_1 \ge 3$.
In particular the potential should remain finite when the distance of the images of  two different lattice points goes
to zero,  while for most potentials used in practise the potential goes to $+\infty$ in this situation. 
This is mostly a technical restriction and we can, in fact, also allow certain singular potentials, see  
 the third point in  Remark~\ref{re:singular_perturbation}.

Regarding the first two restrictions, the finite range assumption and the lower growth bound, 
are both proxies for disallowing 
large scale rearrangement of the lattice. Indeed, if one allows arbitrary rearrangements of the lattice
then, in the thermodynamic limit, the free energy $W_\beta$ is invariant 
under the group $\mathrm{GL}(\Z,d)$ of linear lattice-preserving maps. From this one can deduce that  $W_\beta(F)$ depends only on 
the volume change $\det F$, i.e., in the thermodynamic limit  we obtain a fluid rather than solid
(for the rigorous argument at zero temperature, see  \cite{FT89}). The fluid-like behaviour 
in the thermodynamic limit is easy to understand conceptually: on an extremely long time scale
the body can not sustain any shear, due to the creation and motion of dislocations.

For the zero temperature case, E and Ming \cite{EM07} do not impose a lower growth bound. 
They do not look, however,  for global minimisers of the discrete energy, but only a for local  minimisers
 with respect to variations which are small in the $W^{1,\infty}$ norm. They 
 show that for sufficiently weak and regular applied forces 
 and affine boundary conditions sufficiently close to the identity
 there exist such local minimisers which are close to the identity in the $W^{1,\infty}$ norm. 
 Thus the local minimiser and the allowed competitors do not explore the range of large discrete gradients.
 For the case of positive temperature,  \cite[Thm. 1 and Thm. 2]{LXY21}  Luo, Xiang, and Yang
 prove a similar results for local minimisers of the free energy, using the harmonic approximation. 
 Thus again deformations with large discrete gradients do not appear explicitly, but only
 implicitly through the formula for the entropy of a Gaussian measure. 
 
 We remark in passing that  we use the lattice $\Z^d$ only to label the  interaction points. 
 In particular $\Z^d$,  can be replaced by a Bravais lattice. Using the flexibility of the microscopic energy $U$, 
 one can also extend the setting to complex lattices.

\chapter{Explanation of the Method}\label{sec:explanation}

In this chapter we outline our general approach. 
It follows closely the programme for the rigorous renormalisation group analysis  of  functional integrals 
which has been systematically developed by Brydges, Slade and collaborators over the last decades, see   \cite{Bry09, BS10,BBS19} 
for surveys and additional references to earlier and related work. 
Additional features in our setting are the fact that we need a whole family of finite range decompositions in order to have enough parameters 
for the fine-tuning process and that we work with an almost optimal family of weights or large-field regulators. 

Let us remark that the analysis of Brydges, Slade,   and collaborators for the 4-dimensional weakly self avoiding random walk 
and $|\p|^4$ theory  \cite{BBS14, BBS15, BS15V} also requires a family of finite range decompositions for the operators 
$(-\Delta + m^2)^{-1}$ where $m^2 \in [0, \delta]$ is a parameter. 
However, they only need continuity in $m^2$ of the renormalisation map while we need to show smoothness.

\section{Set-up}
We focus on an outline of the strategy to prove  Theorem~\ref{th:pertcomp_E}, the proof of Theorem~\ref{th:scalinglimit}  is very similar.
We want to study the integral
\begin{align}
\Zscr\equiv \Zscr_{N}(\Kcal,\Qscr,0) =   \int_{\Xcal_N}  \sum_{X \subset \TN}  K(X, \p) \, \mu^{(0)} (\d\p) 
\end{align}
where as before $\TN=(\mathbb{Z}/L^N\mathbb{Z})^d$,
\begin{align}\label{eq:explanation_K}
K(X, \p) = \prod_{x \in X}  \mathcal K(D\p(x)),
\end{align} 
and 
$\mu^{(0)}(\d \p)=\mu_{\Qscr}(\d \p)$ is the Gaussian measure given by 
$$ \mu^{(0)}(d \p)  =  \frac{1}{Z^{(0)}} e^{ - \frac12 \sum_{x \in \TN}  \Qscr(D\p(x))}   \lambda_N(\d\p).$$
Let us recall that $\nabla$ denotes the discrete forward difference operator here.
It turns out that it is convenient to embed this problem into a  more general family of problems of the form
\begin{equation}  \label{eq:family_for_Z}
{\Zscr}(H_0, K_0, \boldsymbol q): =  \int_{\Xcal_N} (e^{-H_0} \circ K_0)( \TN, \p) \mu^{(\boldsymbol q)}(\d\p).
\end{equation}
Here $\boldsymbol q$ is a small symmetric $md \times md$ matrix and $\mu^{(\boldsymbol q)}$ is the Gaussian measure given by 
$$ \mu^{(\boldsymbol q)}(d \p)  =  \frac{1}{Z^{(\boldsymbol q)}}    e^{ - \frac12 \sum_{x \in \TN} \Qscr(D\p(x))   - (\boldsymbol q \nabla \p(x), \nabla \p(x))}   \lambda_N(\d\p).$$
The circle product $\circ$ of maps  $F, G$ defined on the subsets of $\TN$ is given  by
\begin{align}\label{eq:def_circ_explanation}
F\circ G(X)=\sum_{Y\subset X} F(Y)G(X\setminus Y)
\end{align}
for $X\subset \TN$.
For the sake of the present definition, we just temporarily assume that for every  $X\subset \TN$,  functions
$F(X), G(X): \Xcal_N\to \R$ are given---some restrictions will be introduced later.
The sum includes the empty set for which we set $F(\emptyset) = G(\emptyset) = 1$. 
The definition of the circle product is motivated by the following property. If $F$ and $G$ factor, i.e. if $F(X) = \prod_{x \in X} F(\{x\})$ and 
$G(X) = \prod_{x \in X} G(\{x\})$ then 
\begin{align}\label{eq:fac_circ_explanation}
F \circ G(X) = \prod_{x \in X}(F + G)(\{x\}).
\end{align}
The term $H_0$ plays a special role which will be further discussed below. It only contains so called relevant terms,
namely constants and certain linear and quadratic expressions in $\p$. More specifically we assume that 
\begin{align}\label{eq:explain_additive_structure_ham}
H_0(X, \p) =& \sum_{x \in X}H_0(\{x\}, \p) \ \text{ with }\\ 
H_0(\{x \}, \p) = & a_\emptyset + \sum_{1\le |\alpha| \le \lfloor d/2 \rfloor + 1}    \sum_{i=1}^m  a_{i,\alpha}    \nabla^\alpha \p_i(x)
  + \frac12 ({\boldsymbol a} \nabla \p(x), \nabla \p(x))   \label{eq:explain_relevant_ham}
\end{align}
where $\boldsymbol a$ is a symmetric $md \times md$ matrix and $a_{i,\alpha}$ an $a_\emptyset$
are scalar coefficients. 

The original problem corresponds to the choices $\boldsymbol q = 0$, $H_0(X, \p) = 0$ and 
$K_0(X, \p) = \prod_{x \in X}  \mathcal K(D\p(x))$. 

\section{Finite range decomposition}\label{sec:FRD_overview}
The first idea is to replace the integration {with respect to } the Gaussian measure $\mu^{(\boldsymbol q)}$,
by a sequence of integrations {with respect to } Gaussian measures $\mu_{k}^{(\boldsymbol q)}$, $k=1, \ldots N+1$,
such that the measure $\mu_{k}^{(\boldsymbol q)}$ essentially detects the behaviour of the fields $\p$ on the spatial
scales between $L^{k-1}$ and $L^k$. 

More precisely,  we express  the  translation-invariant covariance operator $\mathscr{C}^{(\boldsymbol q)}$ 
of the Gaussian measure $\mu^{(\boldsymbol q)}$ as a sum of translation-invariant covariance operators with finite range, i.e., 
\begin{align}   \label{eq:explain_finite_range}
 \mathscr{C}^{(\boldsymbol{q})}=\sum_{k=1}^{N+1} \mathscr{C}^{(\boldsymbol{q})}_k,
 \text{ and the corresponding kernels satisfy   }         
\mathcal{C}^{(\boldsymbol{q})}_k=-C_k\, \, \text{for  }\, |x|_\infty\geq \frac{L^k}{2}.
\end{align}
Moreover,
the kernel
 $\mathcal{C}_{k}^{(\boldsymbol q)}$ behaves like the Green's function of the discrete Laplace operator
on scale $L^{k-1}$, i.e., 
\begin{align} 
\bigl|\nabla^\alpha\mathcal{C}_{k}^{(\boldsymbol q)}(x)  \bigr|\leq
\begin{cases}
C_{\alpha} L^{-(k-1)(d-2+|\alpha|)}\;\text{for}\;d+|\alpha|>2,\\
C_{\alpha}\ln(L) L^{-(k-1)(d-2+|\alpha|)}\;\text{for}\;d+|\alpha|=2.
\end{cases}
\end{align}
Then $\mu^{(\boldsymbol q)} = \mu^{(\boldsymbol q)}_{N+1}  \ast  \ldots \ast  \mu^{(\boldsymbol q)}_1$
and thus the quantity $\Zscr(H_0, K_0, \boldsymbol q)$ can be expressed as an $N+1$ fold integral. 
Alternatively we can
 define the convolution operator $\boldsymbol R^{(\boldsymbol q)}_k$ by 
$$  ( \boldsymbol R^{(\boldsymbol q)}_k F)(\psi)  = \int_{\Xcal_N}   F(\psi + \p)  \mu_{k}^{(\boldsymbol q)}(\d\p). $$
Then the integral we are interested in can be written as
\begin{align}  \label{eq:explain_R_representation}
 {\Zscr}(H_0, K_0, \boldsymbol q) = \big( \boldsymbol R^{(  \boldsymbol q)}_{N+1}   \boldsymbol R^{(  \boldsymbol q)}_{N} \ldots \boldsymbol  R^{(  \boldsymbol q)}_{1}  (e^{-H_0} \circ K_0)      \big)( \TN,0).
\end{align}

\section{The renormalisation map}
 In view of  \eqref{eq:explain_R_representation} the key idea is to define a map $ \boldsymbol{T}_k^{(\boldsymbol{q})}: (H_k, K_k) \mapsto (H_{k+1}, K_{k+1})$
 such that 
 \begin{align}  \label{eq:explain_prop_Tk}
 e^{-H_{k+1}} \circ K_{k+1} (\TN) = \boldsymbol{R}^{(\boldsymbol q)}_{k+1} (e^{-H_k} \circ K_{k}(\TN)).  
 \end{align}
  For ease of notation for the rest of the chapter we usually do not denote the dependence of the renormalisation map  on $\boldsymbol{q}$ explicitly, i.e. we write
$\boldsymbol{T}_k$ instead of $\boldsymbol{T}_k^{(\boldsymbol{q})}$.
Then 
\begin{equation}  \label{eq:explain_rep_by_N}
 {\mathcal Z}(H_0, K_0, \boldsymbol q) = \big( \boldsymbol R^{(  \boldsymbol q)}_{N+1} (e^{-H_N} \circ K_N)  \big)( \TN,0)
= \int_{\Xcal_N}  (e^{-H_N} \circ K_N)(   \TN,\p) \, \mu^{(\boldsymbol q)}_{ N+1}(\d\p).
\end{equation}

Of course, the property  \eqref{eq:explain_prop_Tk} does not determine $\boldsymbol{T}_k$ uniquely. 
Indeed, for any $\tilde{H}_k$ of the form  \eqref{eq:explain_additive_structure_ham} and \eqref{eq:explain_relevant_ham},
we can use \eqref{eq:fac_circ_explanation} to write
\begin{align}
\label{eq:tildeHk}
e^{-H_k}\circ K_k(X)=&(e^{-\tilde{H}_k}+e^{-H_k}-e^{-\tilde{H}_k})\circ K_k(X)\\
=&\bigl( e^{-\tilde{H}_k}\circ(e^{-H_k}-e^{-\tilde{H}_k})\circ K_k\bigr)(X)=e^{-\tilde{H}_k}\circ \tilde{K}_k(X),\notag
\end{align}
where $\tilde{K}_k=(e^{-H_k}-e^{-\tilde{H}_k})\circ K_k$.
The guiding principle for the definition of $\boldsymbol{T}_k$ is that we want $\boldsymbol{T}_k(0,0) = (0,0)$ 
and that the derivative of $\boldsymbol{T}_k$ at the origin is contracting in $K_k$ and expanding in $H_k$. 
This will allow us to apply the stable manifold theorem to show that the 
term on the right hand side of  \eqref{eq:explain_rep_by_N} is $1$ up to an exponentially small correction provided that
we chose $H_0$ suitably in dependence of $K_0$, see the next section. 
Indeed, the special form of relevant Hamiltonians given in   \eqref{eq:explain_relevant_ham}
stems from the fact that exactly monomials of this form 
do not lead to a contraction under application of ${\boldsymbol R}_k^{(\boldsymbol q)}$ 
if we equip the space of functionals with natural scale dependent norms.
See the last two paragraphs of Section~\ref{se:polymers} for further discussion on relevant vs. irrelevant monomials. 
The definition of the map $\boldsymbol{T}_k$ thus involves three key steps:
\begin{itemize}[leftmargin=0.4cm]
\item Integration against $\mu_{k+1}^{(\boldsymbol q)}$, i.e., application of $\boldsymbol{R}^{(\boldsymbol q)}_{k+1}$.
\item Extraction of the relevant terms,  see \eqref{eq:defoftildeHk}.
\item Coarse-graining to maps defined on disjoint blocks of size $L^{k+1}$ ($(k+1)$-blocks)
and their union ($(k+1)$-polymers) rather than single points and subsets of $\TN$, 
see \eqref{eq:defofHk+1} and \eqref{eq:defofKk+1}.
\end{itemize}
The motivation for the coarse graining is that  a   field $\p$ which is typical
under the next-scale measure $\mu_{k+2}^{(\boldsymbol q)}$ varies only slowly on scale $L^{k+1}$. 
The circle product  is adjusted to the coarse graining: 
for two maps $F, G$ on $k$-polymers the circle product is
defined as
$ F \circ G(X)   = \sum_{\text{$k$-polymer } Y, Y \subset X} F(Y) G(X \setminus Y)$. In particular for $k=N$ there are only two polymers,
the whole torus $\TN$ and the empty set. Thus the right hand side of  \eqref{eq:explain_rep_by_N}
simplifies further since $e^{-\HN} \circ K_N (\TN) = e^{-\HN}(\TN) + K_N(\TN)$.

The key results about the maps $\boldsymbol{T}_k$ are contained in Theorems~\ref{prop:smoothnessofS} and~\ref{prop:contractivity} below: 
they are smooth in a small neighbourhood  (uniformly in $k$ and $N$) and the derivatives at the origin are given by
$$ 
D \boldsymbol{T}_k(0) \binom{\dot H}{\dot K} = \begin{pmatrix} \boldsymbol{A}_k & \boldsymbol{B}_k \\
 0 & \boldsymbol{C}_k \end{pmatrix}  \binom{\dot H}{\dot K} 
$$
where 
\begin{equation}   \label{eq:explain_contraction}   
\norm{ \boldsymbol{A}_k^{-1} } \le c < 1, \quad \norm{ \boldsymbol{C}_k } \le c < 1.
\end{equation}
These estimates give a precise formulation of the idea that the flow is contracting in the $K$ variable and expanding in the $H$ variable. 

\section{Application of the stable manifold theorem and fine-tuning}
The uniform smoothness of the maps $\boldsymbol{T}_k$ and the contraction estimates 
 \eqref{eq:explain_contraction} allow us to apply a discrete version of the stable manifold theorem. 
 This guarantees  that  there exists a smooth function $\hat H_0$ such that for each sufficiently small $K_0$
  the flow starting with $(\hat H_0(K_0, \boldsymbol q), K_0)$ satisfies $\HN = 0$ and $\| K_N \| \le C \eta^N$  for a suitable $\eta < 1$.
 This is described in full  detail in Chapter~\ref{sec:finetuning}  below  (for a slightly modified situation). 
 
 The basic idea is very simple. 
 One considers the vector containing as its coordinates  $H$'s and $K$'s on all scales,  $ Z= (H_0, \dots, H_{N-1}, K_1, \dots, K_{N})$ 
 and a weighted norm $\| Z \| = \max (\max_{0 \le k \le N-1} \eta^{-k}  \| H_k\|,   \max_{ 1\le k \le N}  \eta^{-k} \|K_k\| )$. 
 The space of vectors with finite norm is denoted by $\mathcal Z$. 
 Then one reformulates the conditions that $(H_{k+1}, K_{k+1}) = \boldsymbol{T}_k(H_k, K_k)$ and $\HN = 0$ as a fixed point condition. 
  More precisely one defines  a map $\widetilde{\mathcal T}$ on $\mathcal Z$ which has $\boldsymbol q$ and $K_0$ as  additional parameters
 such that every $Z$ which satisfies $\widetilde{\mathcal T}(\boldsymbol{q},K_0, Z) = Z$ 
 also satisfies $(H_{k+1}, K_{k+1}) = \boldsymbol{T}_k(H_k, K_k)$  for $k \le N-2$ and $\boldsymbol{T}_{N-1}(H_{k-1}, K_{k-1}) = (0, K_N)$. 

 The contraction estimates  \eqref{eq:explain_contraction} will imply  that  that map 
 $\widetilde{\mathcal T}(\boldsymbol q, K_0, \cdot)$ does indeed have a fixed point $Z = \hat  Z(\boldsymbol q, K_0)$ for every small $K_0$. 
 Then the  map $\hat H_0$ is obtained by taking the $H_0$ component of $\hat Z$. 
 As a result, we get for each small $K_0$
 $$ 
 {\mathcal Z}(\hat H_0(K_0, \boldsymbol q), K_0, \boldsymbol q) = \int_{\mathcal{X}_N}  (1 + K_N(\TN, \p))  \,   \mu^{(\boldsymbol q)}_{N+1}(\d\p)
 $$
 where $K_N$ is exponentially small (and depends smoothly on $K_0$ and $\boldsymbol q$). 
 If, by chance,  $\hat H_0(K, 0) = 0$ then we have solved our original problem. 
 In general, there is, however, no reason why this should be true. 
 
 In the final step we will thus use the freedom to tune the  free parameter $\boldsymbol q$ 
 so that the effects of $\boldsymbol q$ in the Gaussian measure $\mu^{(\boldsymbol q)}$ 
 and the effect of $\hat H_0(\boldsymbol q, K_0)$ cancel exactly up to a constant term  which can be pulled out of the integral.
 Thus the final dependence of our original partition function $Z(\mathcal K)$ on $\mathcal K$ 
 is encoded in this constant term, up to an exponentially small term which comes from $K_N$. 
 This allows to reach easily the final conclusion.
 
 The details of this fine-tuning procedure are explained in Chapter~\ref{sec:finetuning}. 
 It is actually convenient to write an enlarged family of problems in a slightly different way. 
 Instead of working only with $\boldsymbol q$ as the main free parameter
 we use a full relevant Hamiltonian (see  \eqref{eq:explain_relevant_ham}) 
 as the free parameter and identify $\boldsymbol q$ with the quadratic part $\boldsymbol a$ of the relevant Hamiltonian. 
 Denoting the relevant Hamiltonian by $\mathcal H$ and its quadratic part by $\boldsymbol q(\mathcal H)$ 
 we are led to study the family of problems
 $$
  \int_{X_N}   \big(  e^{-H_0} \circ \hat K_0(\mathcal H, \mathcal K)   \big)  (\TN)  \, \mu^{(\boldsymbol q(\mathcal H))}(\d\p)
   \quad \hbox{with} \quad  \hat K_0(\mathcal H, \mathcal K) = e^{-\mathcal H} \mathcal K,
   $$
 see Chapter~\ref{sec:finetuning}. 
 We then show as above that there exists a function $\hat H_0$ such that the choice $H_0 = \hat H_0(\mathcal H, \mathcal K)$ leads to $\HN = 0$ 
 and an exponentially small $K_N$. 
 In fact we can use exactly the argument given above in connection with the observation that the map
 $(\mathcal H, \mathcal K) \to e^{-\mathcal H} \mathcal K$ is smooth. 
  It is then easy to see  that there exists a map $\hat{\mathcal H}$  such that 
  $\hat H_0( \hat{\mathcal H}(\mathcal K), \mathcal K) =\hat{ \mathcal H}(\mathcal K)$ and that the integral for $\mathcal H = \hat{\mathcal H}(\mathcal K)$ 
  and $H_0 = \hat H_0(\mathcal H, \mathcal K)$ agrees with our original integral up to a scalar factor.

 \section{A glimpse at the implementation of the strategy}
 Our main objects are relevant Hamiltonians $H_k$ and perturbations $K_k$. 
 The relevant Hamiltonians are described by a finite number of parameters:  
 $a_\emptyset$ for the constant part, $a_{\alpha, i}$ for the linear part and ${\boldsymbol a}$ for the quadratic part. 
 The $K_k$ are functions depending on  a $k$-polymers $X$ and the field $\p$. 
 One key ingredient is to design the renormalisation group (RG) maps $\boldsymbol{T}_k$ so that at each step the relevant terms are correctly extracted. 
 This can already be guessed at the level of the linearised problem. 
 Another key ingredient is to design norms for $H_k$ and $K_k$ which allow us to prove uniform smoothness and contraction estimates. 
 The construction of such norms will be described in detail in  Chapter~\ref{sec:description}. 
 Here we just mention three guiding principles
\begin{itemize}[leftmargin=0.4cm]
 \item The norms at the scale $k$ for the fields $\p$  should be such that a field which is 'typical' 
 under the measure $\mu^{(\boldsymbol q)}_k$ has norm approximately of the order $1$;
 \item For a fixed $k$ polymer, the norm on the functional $\p \mapsto K(X, \p)$ should be dual to the field norm. 
 For linear functionals it is clear what duality means. 
 Homogeneous polynomials of degree $r$ can be viewed as linear functionals on the $r$-fold tensor product of the space of fields 
 and there is a natural way to design norms which behave well under tensorisation (see Appendix~\ref{se:norms_polynomials});
 \item Our starting perturbation factors, i.e. $K_0(X) = \prod_{x \in X} K_0(\{x\})$. 
 This suggests that for small $K$ the size of $K(X)$ should decrease exponentially in the number of  blocks in $X$.  
 The property to factor is lost in the iteration. 
 To keep the idea that the contribution from large polymers is exponentially small, 
 we introduce in the definition of the norm of $K_k$ a weight $A^{|X|_k}$, where $|X|_k$ is the number of $k$-blocks in the polymer $X$.
 \end{itemize}
 
 Two further points turn out to be important. 
 First, while the factorisation property is in general lost, the finite range condition  \eqref{eq:explain_finite_range} on the covariance 
  in the finite range decomposition ensures that  the \emph{factorisation between polymers that are separated by one block} still holds. 
 Here we use the fact that we work on fields with zero average. 
 Thus the  action of the kernel $\mathcal{C}^{(\boldsymbol{q})}_k$ on fields by discrete convolution does not change  if we add a constant to the kernel. 
 Hence the    condition $\mathcal{C}^{(\boldsymbol{q})}_k= -C_k$ for  $|x|_\infty\geq \frac{L^k}{2}$ is equivalent 
 to assuming that $\mathcal{C}^{(\boldsymbol{q})}_k$ is supported in  $\{ x :  |x|_\infty < \frac{L^k}{2}  \}$.

 This factorisation property for polymers that are separated by one block allows us to  track  $K_k(X, \cdot)$ only for \emph{connected} polymers $X$. 
 The functional for general polymers is then obtained by multiplying over the connected components.
 
 The second point is the so called \emph{large field problem}. 
 With exponentially small probability, very large values of the field $\nabla \p(x)$ may arise. 
 Since a  typical perturbation $\mathcal K(D\p(x)) = e^{-\overline{\mathcal U}(D\p(x))} - 1$ contains also exponential terms, 
 care has to be taken that the integrals in each step are well-defined. 
 This problem is well known in rigorous renormalisation theory and handled by the introduction of carefully chosen weights, 
 or large-field regulators, in the norms of $K_k$. 
 In Chapter~\ref{sec:weights} we present a new construction of weights which leads to almost optimal weights. 

  \chapter{Choice of Parameters}\label{sec:tracking}
  The precise implementation of the RG construction involves a number of parameters which help to fine-tune the properties
  of the RG map and to ensure the key smoothness and contraction estimates in Theorems~\ref{prop:smoothnessofS} and \ref{prop:contractivity}
  from which the main results  Theorem \ref{th:pertcomp_E} and \ref{th:scalinglimit} can be deduced. 
  
  The purpose of this { chapter} is to give an overview over these parameters and to explain how they are chosen. 
  Detailed descriptions are given in the following { chapters}. Here we focus on a bird's eye view to emphasise the
  idea that the parameters can be sequentially chosen in such a way that all the restriction that arise in the following chapters
  can be eventually satisfied simultaneously. 
  
Readers who  want to plunge immediately into the details of the argument may skip this chapter upon first reading 
and refer back to   it for a quick overview why all the restrictions which appear in various parts of the proof are consistent.
  
Actually, a majority of the parameters can be chosen once and for all (depending possibly  only on some basic parameters
that can be fixed at the outset, like the dimension $d$ of the model and the maximal order $R_0$ of discrete derivatives in a coordinate direction). 
  We will refer to all these as \emph{fixed parameters} and we will not track how the various constants depend on them. 
  A list of  fixed parameters is given in Section~\ref{se:fixed_parameters} below. 
  We first discuss the free parameters that we will adjust to obtain the desired smoothness and contraction estimates.

  \section{The free parameters \texorpdfstring{$L$}{L}, \texorpdfstring{$\protect\headingh$}{h}, and \texorpdfstring{$A$}{A}}  \label{se:free_parameters}
  There are three free parameters, namely
 \begin{longtable}[t]{cp{11cm}}
 $L \in \NN$ & The \emph{size of the basic block}.\\
 $h \gg 1$    & A \emph{scaling factor in the norm for the fields};  
 the field norm on level $k$ involves a term $h_k^{-1}$ with $h_k = 2^k h$, see  \eqref{eq:primal_norm} and \eqref{eq:definition_weights}. 
 A field which is typical   on scale $k$ (i.e., under the measure $\mu_{k+1}$) has norm  of order $h_k^{-1}$. 
  Since the norms on functionals are defined by duality, the standard Hamiltonian
  $ H(\p) = \sum_{x \in B} |\nabla \p|^2$ for a block $B$ on scale $k$ has norm $h_k^2$. 
  In our earlier work \cite{AKM16} we used a scaling factor $h$ which was independent of $k$. 
  The reason we now need scaling factors  $h_k$ which grow sufficiently rapidly in $k$ is related to the new choice 
  of  nearly optimal weights (see Chapter~\ref{sec:weights}). 
  Among others, we want to bound the field norm by the increase in the logarithm of the 
  weights as we go from scale $k$ to $k+1$ (see \eqref{eq:w9}). 
  This essentially requires that $\sum_{k < N} h_k^{-2}$ can be bounded independently of $N$. 
  A similar issue arises for the estimates  \eqref{eq:w5} and  \eqref{eq:w6}.
  The choice of exponentially growing scaling factors is mostly for convenience. 
  We cannot  allow for faster than exponential growth
  because factors of $h_{k+1}/ h_k$ appear in the proof of the change of scale estimates in 
  Lemma~\ref{le:norms_pointwise} and Lemma~\ref{le_contraction_I}.\\
   $A \gg 1$ & A parameter which \emph{penalises the contributions of functionals
  defined on long polymers}. The norm on functionals involves a supremum over all $k$-polymers $X$
  of $A^{|X|} \| K(X)\|_{k}$ where $|X|$ denotes the number of $k$-blocks in $X$. 
 \end{longtable} 
  
  Our goal is to  show that there exists a number $L_0$ and functions $L \mapsto h_0(L)$ and
 $L \mapsto A_0(L)$ such that the renormalisation maps $\boldsymbol{T}_k = \boldsymbol{T}_k^{(\boldsymbol q)}$ have good properties 
 (in a suitable  small neighbourhood of $0$ and for sufficiently small $\boldsymbol q$) if
 $$ 
 L\ge L_0, \quad h \ge h_0(L), \quad \hbox{and} \quad A \ge A_0(L).
 $$ 
  
  In the following we first review the choice for the fixed parameters. Then we describe the key steps in the proof
  and discuss which restrictions on the free parameters $L$, $h$, and $A$ arise in each step.

  \section{Fixed parameters}    \label{se:fixed_parameters} 
  The following parameters are fixed once and for all and dependence on them is usually 
  not indicated in the following:
  \begin{longtable}[t]{cp{10.6cm}}
    $d$ & Spatial dimension.
    \\ \vspace{0.2cm} %
  $m$ & The number of components of the field $\p$.
  \\ \vspace{0.2cm} %
     $R_0$ & A nonzero integer which determines the maximal number of discrete (forward) derivatives through the set
     
   $\{e_1,\ldots,e_d\}\subset \Ical \subset \{\alpha\in \mathbb{N}_0^d\setminus \{(0,\ldots,0)\}:\;  \abs{\alpha}_\infty\leq R_0\} $.
   \\ \vspace{0.2cm} %
  $r_0 \ge 3$ & An integer that measures smoothness of the functionals in the field.
Loosely speaking, the restriction $r_0 \ge 3$ arises from the fact that the third order terms are always irrelevant,  but quadratic terms are not. 
More precisely, the condition $r_0 \ge 3$ is crucial for  the two-norm estimate   \eqref{eq:two_norm_concrete}. 
In particular,  this estimate  allows us to deduce the crucial contraction estimate for  $\boldsymbol{C}^{(\boldsymbol{q})}$ 
from a contraction estimate for the action of  the extraction operator $1 - \Pi_2$ on  Taylor polynomials at zero.  
See Lemma~\ref{le_contraction_I} and Lemma~\ref{le:contraction_single_block_new} in connection with   \eqref{eq:decomp_C} for further details. 
Our standard choice is
\begin{equation} \label{eq:def_r_0}
r_0 = 3.
\end{equation}
\\ \vspace{0.2cm}%
$r_1 \ge 2$ & An integer that measures smoothness with respect to external parameters (e.g., the deformation  $F$).
\\ \vspace{0.2cm}   $\pphi$ & Number of discrete derivatives in the definition of the field norm $| \phi|_{j, X}$. 
We need  $\pphi \ge \lfloor d/2 \rfloor + 2$ to get the right decay in $L$ in the Poincar\'e type estimate 
in Lemma~\ref{le:discrete_taylor_approximation}  which is the main ingredient in the proof of  the contraction estimate for $1 - \Pi_2$
  (see Lemma~\ref{le_contraction_I}). 
  We will take
  \begin{equation}  \label{eq:def_p_Phi}
  \pphi = \lfloor d/2 \rfloor +2.
  \end{equation}
  \\ \vspace{0.2cm} %
  $M$ & Number of discrete derivatives in the definition of the quadratic form $\boldsymbol{M}^X_k$
  in \eqref{eq:defofMk}.
  We need $M \ge \pphi + \lfloor d/2 \rfloor + 1$ to be able to apply the discrete Sobolev
  embedding and to get control of $p_\Phi$ discrete derivatives in the supremum norm.
   We will take 
   \begin{equation} \label{eq:def_M}
    M = p_\Phi + \lfloor d/2 \rfloor + 1 = 2 \lfloor d/2 \rfloor + 3.
    \end{equation}
    \\ \vspace{0.2cm} %
    $R$ & A geometric parameter that is used to define a neighbourhood around blocks (see \eqref{eq:nghbhdscompact}). 
 It determines the allowed range of dependence of the functionals on the first scale, e.g.,
$\Kcal(\{x\},\p)$, $\Hcal_0(\{x\},\p)$, and $\boldsymbol{M}_0^{\{x\}}$ 
(see \eqref{eq:explanation_K}, \eqref{eq:explain_additive_structure_ham}, and \eqref{eq:defofMk}) may only depend on $\p{\restriction_{x+[-R,R]^d}}$.
   This implies that we need that $R\ge \max(R_0,M,\pphi)=\max(R_0,M)$. We will take
       \begin{equation}  \label{eq:def_R}
    R = \max(R_0, M) = \max(R_0, 2 \lfloor d/2 \rfloor + 3).
    \end{equation}
  \\ \vspace{0.2cm} %
   $n$ & The number of discrete derivatives controlled in the finite range decomposition (see Theorem~\ref{thm:frd}). We need
  $n \ge 2M$ to control the integral of the weights against the Gaussian measures obtained by 
  the finite range decomposition (see  Theorem~\ref{th:weights_final}\ref{w:w8} and its proof in Lemma~\ref{prop:W3}) and we will take 
  \begin{equation}  \label{eq:def_n}
  n = 2M = 4 \lfloor d/2 \rfloor + 6.
  \end{equation}
  \\ \vspace{0.2cm} %
 $\tilde n$ & A secondary parameter in the finite range decomposition (see Theorem \ref{thm:frd}) which relates to the decay 
of  the derivative of the Fourier symbols with respect to the quadratic form we decompose.
 We need   $\tilde n \ge n +    \lfloor d/2 \rfloor + 1$ to bound the derivative of
the maps $\boldsymbol{R}_k^{(\boldsymbol{q})}$ with respect to $\boldsymbol{q}$ (see  Theorem \ref{prop:finalsmoothness}) and we will take
\begin{equation}  \label{eq:def_tilde_n}
 \tilde  n =    n +    \lfloor d/2 \rfloor + 1 = 2M +  \lfloor d/2 \rfloor + 1 = 5 \lfloor d/2 \rfloor + 7.
  \end{equation}
  \\ \vspace{0.2cm} %
   $\omega_0 >0$ & A parameter that controls the coerciveness and boundedness of the quadratic
  form $\Qscr$. We require (see  \eqref{eq:Qlowerbound_again})
  \begin{align}
	\omega_0 |z|^2 \leq \Qscr(z) \leq \omega_0^{-1}|z|^2\  \text{ for all }\   z\in\Gcal = (\mathbb{R}^m)^{\Ical}.
\end{align}
 \\\vspace{0.2cm} %
  $\zeta \in (0,1)$ & This parameter controls the exponential weight in the norm $\| \cdot \|_{ \zeta}$
which is defined in \eqref{eq:normE}  and  measures the allowed  growth  of the perturbation $\mathcal K(z)$  as $z\to \infty$. 
\\ \vspace{0.2cm} %
$\weightzeta \in (0,\frac14)$ & This parameter analogously controls the growth of the weights, see \eqref{eq:defAk}. 
To make the norms of the perturbation $\mathcal K$ 
and the corresponding functional $K$ consistent we choose $\weightzeta = \tfrac14 \zeta$, see  \eqref{eq:definition_weightzeta}
 as well as  \eqref{eq:consistency_weightzeta}, \eqref{eq:w-1} and Lemma~\ref{le:I}. \\
 $\eta \in (0, \frac23]$ & This parameter controls the rate of convergence of $\|H_k\|$ and $\|K_k\|$.
More precisely it appears in the definition of the norm of the vector $(H_0, \ldots, H_{N-1}, K_1, \ldots K_N)$.
Vectors with norm $\le 1$ satisfy $\| H_k \| \le \eta^k$ and $\|K_k \| \le \eta^k$, see 
   \eqref{eq:norm_mcZ}. For the purpose of the current paper we could take $\eta = \tfrac23$, but other applications
   require smaller values of $\eta$.  
  \end{longtable}

\section{Choice of the free parameters in the key steps  of the proof}

The key technical results are the uniform smoothness and contraction   estimates
for the renormalisation maps $\boldsymbol{T}_k$, see Theorem~\ref{prop:smoothnessofS} and Theorem~\ref{prop:contractivity}.
From these estimates our main results in Chapter~\ref{sec:setting} can be derived  by general abstract arguments.
This deduction is  described in detail in Chapters~\ref{sec:proofs} and~\ref{sec:finetuning}. 
We first  show discrete stable manifold theorem (Theorem~\ref{propfixedpoint}). Together with a second
 fixed point result  (Lemma~\ref{lemmafixedpoint}) this allows us to prove  a representation formula for the 
partition function (Theorem~\ref{maintheorem}). From this representation formula the desired results follow easily, 
see Sections~\ref{sec:proof_GGM} and~\ref{sec:proof_scaling_limit}.

In the following we thus focus on  key steps in the proof of the smoothness and contraction estimates, namely:
\begin{itemize}[leftmargin=0.4cm]
\item Set-up:
\begin{list}{$-$}{\setlength{\leftmargin}{0.4cm}}
\item The construction of a family of finite range decompositions.
\item The definition of the RG map and factorisation properties.
\item The construction of weights.
\end{list}
\item Estimates for the basic operations:
\begin{list}{$-$}{\setlength{\leftmargin}{0.4cm}}
\item  The product estimates and submultiplicativity of the norms.
\item The uniform boundedness and  smoothness of the integration map ${\boldsymbol R}^{(\boldsymbol q)}_{k+1}$. 
\item The estimates for the extraction map  $\Pi_2$ and for $(1 - \Pi_2)$ (with a change of scales).
\end{list}
\item 
The uniform estimates for  the derivatives of the RG map (Theorem~\ref{prop:smoothnessofS}).
\item The contraction estimates for the linearised RG map (Theorem~\ref{prop:contractivity}).
\end{itemize}

The most delicate   steps are the construction of the weights and the linear estimates. 
The point of the linear estimates is to obtain  a sufficiently small bound on the contraction condition,
while for the smoothness estimates we do not need so precise control of the constants.
Bad constants in the smoothness estimates only lead to small neighbourhood $B_\varrho$ in which Theorem 2.4. applies,
but that is not a problem. 

We now review the role of the free  parameters in the key steps.

\subsection*{Family of finite range decompositions} 
In Theorem~\ref{thm:frd} we obtain a finite range decomposition for 
all quadratic forms with $\omega_0/2  \le \Qscr \le 2 \omega_0^{-1}$.
Dependence  of the estimates on $L$ is expressed
explicitly.  The parameters $h$ and $A$ do not appear.
A key property is that the convolution operators $\boldsymbol R^{(\boldsymbol q)}_k$, which correspond to the
finite range composition for the quadratic forms $\Qscr^{(\boldsymbol q)}(z) = \Qscr(z) - ( \boldsymbol q z^\nabla, z^\nabla)$,
depend smoothly on $\boldsymbol q$, with bounds independent of $N$, see Theorem~\ref{prop:finalsmoothness}.

\subsection*{Definition of the RG map: locality, factorisation, geometric properties.}
To make the combinatorics of the coarse-graining   and the properties of the finite range decomposition
interact nicely, we define various neighbourhoods of a polymer  and locality conditions
on the functionals $\p \mapsto K(X, \p)$,
see Section~\ref{se:polymers}.

  Consistency of these definitions requires $L \ge 2^d + R$. 
  The construction involves a map $\pi$ which assigns to a  polymer $X$  at scale $k$ (a union of blocks of size $L^k$)
  a polymer $\pi(X)$ at scale $k+1$. In general $X$ is not contained in $\pi(X)$, but the condition $L \ge 2^d + R$ guarantees
  that the corresponding small scale neighbourhoods, defined in  \eqref{eq:nghbhdscompact}, satisfy $X^\ast \subset \pi(X)^\ast$. 
  To ensure that the RG map preserves the 
  factorization property we need the stronger relation 
  \begin{equation}  \label{eq:restriction_L_geometric}
  L \ge 2^{d+2} + 4 R,
  \end{equation}
  see Lemma~\ref{le:K_kplus_factors} and
Proposition~\ref{pr:properties_RG_map}.

\subsection*{Weights}
To deal with the large field problem we introduce families of weak  weights $w_k^X(\p)$ and $w^X_{k:k+1}(\p)$
as well as strong weights $W^X_k(\p)$ which depend on the field $\p$ and a polymer $X$. 
These weights need to satisfy certain natural supermultiplicativity properties and need  to be consistent 
with application of the integration map $\boldsymbol R^{(\boldsymbol q)}_{k+1}$. These properties are summarised in 
Theorem~\ref{th:weights_final}. 
They hold provided that the following constraints are satisfied
\begin{eqnarray} \label{eq:constraint_L_weights}
L &\ge& 2^{d+3} + 16 R,    \\
 \label{eq:constraint_h_weights}
h &\ge& C \delta^{-1/2}(L),
\end{eqnarray}
Here $\delta(L)$ is a parameter that appears in the construction of the weights, 
see Theorem~\ref{th:weights_final} and  \eqref{eq:defofdelta}.

For the contraction estimate in Theorem~\ref{prop:contractivity}, it is crucial that the constant $\AB$
in the integration estimate \eqref{eq:w8} for a single block  does not depend on $L$.
Here we use that the smoothness estimates in the finite range decomposition have the optimal 
dependence on $L$, see \eqref{eq:discretebounds}, \eqref{eq:discreteboundsfinal} and
\eqref{eq:trace_bound_single_block}. The optimal $L$-dependence is an important improvement of the finite range decomposition 
in \cite{Buc18}    over the one in \cite{AKM13}. This improvement is related to the fact the the decomposition in  \cite{Buc18}  is based on  Bauerschmidt's decomposition \cite{Bau13},   rather than on \cite{BT06}.
The free parameter $A$ does not appear in the construction of the weights. 

\subsection*{Submultiplicativity of the norms and product estimates    }
The RG  map  $\boldsymbol{T}_k$ can be written as  composition of  linear maps, the harmless map $H \mapsto e^{-H}$ 
for relevant Hamiltonians,  and a number of polynomial maps which arise from the combinatorics
of the circle product and the coarse-graining procedure. The key difficulty is that the degree of the polynomials
is not bounded independent of $N$. Hence a natural  idea is to work with norms which are submultiplicative
so that products and  polynomials (as well as  their derivatives) can be easily estimated.
Submultiplicativity of the relevant norms, defined in \eqref{strongnorm}-\eqref{middlenorm}, essentially follows from 
general facts about tensor product norms on (Taylor) polynomials (see
Proposition~\ref{pr:product_estimate_taylor} in  Appendix~\ref{se:norms_polynomials}) and the
supermultiplicativity of the weights. The details are described in
   Sections~\ref{se:pointwise_submult} and~\ref{se:norms_submult}.
The submultiplicativity estimates require only the conditions \eqref{eq:constraint_L_weights} and
\eqref{eq:constraint_h_weights}  already discussed above. 
The estimates are stated for norms involving  fixed polymers and thus  the parameter $A$ does not appear in these estimates. 
The parameter $A$ will enter when we consider  the final norms such as 
$\norm{F}_{k}^{(A)}=\sup_{X }\norm{F(X)}_{k,X} A^{|X|_k}$ where the supremum is taken over connected $k$-polymers.

The  submultiplicativity  estimates usually hold for products of functionals defined for disjoint polymers.  
One key feature of the circle product $\circ$  is that it contains  only such products. 
This is actually the reason for expressing the RG maps  in terms of circle products.
The importance of the circle product was already realised in the early works on rigorous renormalisation,
see \cite{BY90},   pp. 354--355 as well as Section 7 in \cite{BF92}.

\subsection*{Uniform boundedness and smoothness of the integration map ${\boldsymbol R}^{(\boldsymbol q)}_{k+1}$ }
Boundedness of map  ${\boldsymbol R}^{(\boldsymbol q)}_{k+1}$, which acts by convolution, follows from the integration estimates
\eqref{eq:w7} and \eqref{eq:w8} for the weights, the fact that convolution commutes with Taylor expansion and basic properties of the field 
norm, see Lemma~\ref{le:keyboundRk} and its proof. 
Here we only need the condition   \eqref{eq:constraint_L_weights}
 which implies the integration estimates for the weights. 

In \cite{AKM16} the bound for the derivatives $D_{\boldsymbol q}^\ell  {\boldsymbol R}^{(\boldsymbol q)}_{k+1} K(X, \cdot) $ was rather delicate. 
Indeed, $D_{\boldsymbol q}^\ell  {\boldsymbol R}^{(\boldsymbol q)}_{k+1} K(X, \cdot)$ could only be bounded by a stronger norm of $K$, 
more precisely a norm which involves more derivatives of $K$ with respect to the field. 
Therefore the proof of the abstract stable manifold theorem had to be based on an implicit function theorem with loss of regularity. 

In the present paper this difficulty is overcome by using the new finite range decomposition in \cite{Buc18} for which 
one can easily obtain bounds for $D_{\boldsymbol q}^\ell  {\boldsymbol R}^{(\boldsymbol q)}_{k+1} F(X,  \cdot) $ for local functionals $F(X, \cdot)$
in terms of $L^p$ norms of $F(X, \cdot)$, see Theorem~\ref{prop:finalsmoothness}. 
Thus the bound for $D_{\boldsymbol q}^\ell  {\boldsymbol R}^{(\boldsymbol q)}_{k+1} K(X, \cdot) $ follows again 
from the integration of the estimate for the weights. The only new ingredient is that we need a bound for the $L^p$ norm 
of the weight, for some $p > 1$, see \eqref{eq:w7} and \eqref{eq:w8}.

\subsection*{Estimates for \texorpdfstring{$\Pi_2$}{Pi2} and \texorpdfstring{$(1 - \Pi_2)$}{1-Pi2}}
A key step in the definition of the map $\boldsymbol{T}_k$ is the extraction of relevant terms. We need  that this leads
to a bounded map from $K_k$ to $H_{k+1}$ and, more importantly,  that due to the extraction of the relevant
terms the linearisation  $\boldsymbol{C}^{(\boldsymbol{q})}_k$ of the map $K_{k} \mapsto K_{k+1}$ is a strong contraction. 

As we will discuss below,   the main step is to analyse the extraction at the level of Taylor polynomials at 
zero. This leads to the definition of the projection $\Pi_2$ and the remainder map $1- \Pi_2$, see  Section~\ref{se:projection_Pi2}. 
The key properties of these maps are stated in Lemma~\ref{le:Pi2_bounded} and Lemma~\ref{le_contraction_I}. 
 The corresponding estimates only rely on the definition of the field norms in \eqref{eq:primal_norm}. 
The free parameter $A$ does not appear. The dependence on $h$ (or $h_k = 2^k h$) cancels exactly.
We only need the mild geometric condition 
\begin{equation} 
\label{eq:bound_L_Pi2} 
L \ge 2^d + R
\end{equation}
which ensures that for any block $B \in \mathcal{B}_k$  we have $B^\ast \subset B^+$ 
and that  the functionals $H(\p, B)$ and $K(\p, B)$ depend only on the field $\p$ restricted to $B^+$. 
Under the additional harmless condition $L \ge 7$ one gets that, for $k \le N-1$, the set $B^{++}$ does not wrap around the torus
and hence one can work with fields on $\Z^d$ rather than on the torus $T_N = \Z^d / L^N \Z^d$, see the text before 
 \eqref{eq:fields_on_Zd} for further discussion. 
 Note that the earlier condition  \eqref{eq:constraint_L_weights}  implies both    \eqref{eq:bound_L_Pi2} and $L \ge 7$. 
 In addition this ensures that $X^\ast\subset X^{+}$, hence $X^\ast$ does not wrap around the torus on scale $N-1$ either.

\subsection*{Uniform smoothness estimates for the RG map, 
see Chapter~\ref{sec:smoothness}}
The renormalisation map  $\boldsymbol{T}_k^{(\boldsymbol{q})}$ can be written as
\begin{align}  \label{eq:renorme_with_Sk}
(H_{k+1}, K_{k+1}) = 	\boldsymbol{T}_k^{(\boldsymbol{q})}(H_k,K_k)=(\boldsymbol{A}_k^{(\boldsymbol{q})}H_k 
	+\boldsymbol{B}_k^{(\boldsymbol{q})}K_k, \myS_k(H_k,K_k,\boldsymbol{q})) 
\end{align}
where $\boldsymbol{A}_k^{(\boldsymbol{q})}$ and $\boldsymbol{B}_k^{(\boldsymbol{q})}$ are linear operators,
see \eqref{eq:define_Sk}.
We first focus on the smoothness and uniform estimates for the nonlinear map $\boldsymbol{S_k}$.
This map can be written as a composition of 
three polynomial maps $P_1$, $P_2$, and $P_3$, the integration map ${\boldsymbol R}^{(\boldsymbol q)}_{k+1}$, 
the projection $\Pi_2$,  and the exponential map $E$ on  ideal Hamiltonians given by
$E(H) = e^H$, see \eqref{eq:decomposition_S}.  
The bounds for the exponential map $E$ are contained in Lemma~\ref{le:smoothness_exp}  and require no restrictions on $L$,  $h$ or $A$.  
The bounds for the  other maps follow from   the bounds for the product map, the integration map and the map $\Pi_2$. 
The required restrictions on $L$ and $h$ take the form 
\begin{eqnarray} \label{eq:constraint_L_smoothness}
L &\ge& \max(2^{d+3} + 16 R, 4d (2^d + R)),     \\
 \label{eq:constraint_h_smoothness}
h &\ge& C \delta^{-1/2}(L),
\end{eqnarray}
see  \eqref{eq:hL_for_products}--\eqref{eq:L_for_R1R2}  and Lemmas~\ref{le:P2}--\ref{lemmaS2}.
Here $\delta(L)$ is as above the  parameter that appears in the construction of the weights, 
see Theorem~\ref{th:weights_final} and  \eqref{eq:defofdelta}.

The additional lower bound on $L$ as compared to 
\eqref{eq:constraint_L_weights}  comes from the geometric condition
\eqref{eq:estimate_Ustar_blocks} which appears in the analysis of the map $P_1$ and  reads
$$ |U^\ast|_k\leq 2|U|_k  \text{ if } L \ge 4d (2^d + R).$$

The restrictions on $A$ come from the need to compensate certain combinatorial terms in the estimates
of the polynomials maps 
and these restrictions  take the form 
$$ A \ge A_0(L)$$
An explicit choice of $A_0(L)$ is given in   \eqref{eq:bound_A0_smooth}

Otherwise the dependence of constants $L,h$ and $A$  is tracked explicitly
 in Chapter~\ref{sec:smoothness} and we get  an
explicit bound for the final neighbourhood $U_{\rho, \kappa}$ on which 
$S$ is smooth. For $\kappa$ we can take the value in Proposition~\ref{prop:W3}
and $\rho$ can be taken of the form $\rho =c  A^{-2}$ where $c$ is given explicitly in terms
of a constant in the finite range decomposition and the bound for the map $\Pi_2$, see Section~\ref{se:proof_smoothness_S}.

\subsection*{Contraction estimates, see Chapter~\ref{sec:contraction} }

We finally discuss the estimates for the linearisation of $\boldsymbol{T}_k = \boldsymbol{T}_k^{(\boldsymbol{q})}$ at $H=K=0$. 
For ease of notation we drop the superscript $\boldsymbol{q}$. 
The linear maps $\boldsymbol{A}_k$ and $\boldsymbol{B}_k$ which appear in  \eqref{eq:renorme_with_Sk} are easily estimated, 
see Sections~\ref{se:bound_Aq} and~\ref{se:bound_Bq}.
Moreover we have $\partial_{H_k} \boldsymbol{S}_k(0,0) = 0$. 
Thus the heart of the matter is the estimate for $\boldsymbol{C}_k := \partial_{K_k} \boldsymbol{S}_k(0,0)$, i.e., 
the linearisation of the map $K_k \mapsto K_{k+1}$. 

This is on the one hand easier than the smoothness estimate discussed above  since $\boldsymbol{C}_k$ 
has a very explicit form (see below) and since we do not have to deal with 
nonlinear terms and the combinatorics of polynomial expressions in $H_k$ and $K_k$. 
On the other hand the estimate of $\boldsymbol{C}_k$ is actually the most delicate estimate since we need 
that the norm of $\boldsymbol{C}_k$ is small (uniformly in $N$), 
while  for the smoothness estimate we only needed uniform boundedness. 
Indeed, smallness of $\boldsymbol{C}_k$ leads to the strongest restrictions on $L$, see below.
                                   
The linear map $\boldsymbol{C}_k$ is given by
\begin{align}
& (\boldsymbol{C}_k \dot{K})(U,\p) \\ 
\nonumber 
= & \sum_{B:\, \overline{B}=U}(1-\Pi_2)\int_{\Xcal_N}\dot{K}(B,\p+\xi)\,\mu_{k+1}^{(\boldsymbol{q})}(\d\xi)+               
\sum_{\substack{X\in \Pck\setminus \mathcal{B}_k\\  \pi(X)=U}}\int_{\mathcal{X}_N}\dot{K}(X,\p+\xi)\,\mu_{k+1}^{(\boldsymbol{q})}(\d\xi),
\end{align}  
see \eqref{eq:defofCk}.                                 
 The second term only involves connected polymers $X$ which are not single blocks. 
 For those the number of blocks decreases under coarse  graining. 
 More precisely if we denote by $|X|_k$ the number of $k$-blocks in $X$ and by $|U|_{k+1}$ the number of $(k+1)$-blocks in $U$ 
 then  $|U|_{k+1} = |\pi(X)|_{k+1} \le  \frac{1}{1+ 2\upalpha} |X|_k$ for some explicit $\upalpha > 0$ given  in  Lemma \ref{le:app1}.

 Thus the second term has small operator norm if we choose $A$ sufficiently large. 
 Indeed we will show in Lemma~\ref{le:contrlarge}   that it suffices to take
\begin{align}
A \ge   \max \Bigl( \frac{8}{   {\eta}   } {\AP}^2L^d(2^{d+1}+1)^{d2^d}, 
\bigl(\frac{8  \AP}{  {\eta}    \updelta}\bigr)^{\frac{1+2\upalpha}{2\upalpha}} \Bigr)
\end{align}
where $\AP = \AP(L)$ is the constant from     Theorem~\ref{th:weights_final}~\ref{w:w7} 
and $\updelta$ is the constant from Lemma \ref{le:app2}, respectively. 	
		
For the smallness of the first term the crucial ingredient is the smallness of map $1- \Pi_2$, see  Lemma~\ref{le_contraction_I}. 
Actually that lemma gives an estimate for the Taylor polynomial (with respect to $\p$)  at zero.
 To get the necessary estimate for the Taylor polynomial at an arbitrary point $\p$ 
 we use that the natural norms of the  derivatives of  $K$ of  order three or higher 
exhibit a decay of order $L^{-3d/2}$ when we pass from scale $k$ to scale $k+1$, see  \eqref{eq:two_norm_concrete} in Lemma~\ref{le:norms_pointwise}.

The detailed estimates for the first term can be found in Lemmas~\ref{le:contraction_single_block_new}
  and~\ref{le:contraction_single_block_prime}.
  They  hold under the assumption
$$ 
L \ge \max\bigl(  (    4 \eta^{-1}   C' \AB C_1)^{\frac{1}{d'-d}},  (    32 \eta^{-1}   C' \AB (C_2 +1))^{\frac{2}{d}},  { 2^{d+3}+16R}  \bigr),
$$
where $\AB$ is the constant in the single-block integration estimate for the weights, see  \eqref{eq:L_0_single_block_contraction}.

\chapter{Description of the Multiscale Analysis}  \label{sec:description}

In this chapter we introduce the key elements of the multiscale analysis.
We state the existence of  a suitable  finite range decomposition, give a precise definition of the RG map, and define function spaces and norms. 
We continue to work on the discrete torus $T_N=(\mathbb{Z}/(L^N\mathbb{Z}))^d$.

\section{Finite range decompositions}\label{sec:FRD}

In this section we recall the properties of the finite range decomposition which we will use in the following. 
Finite range decompositions of the initial Gaussian measure were introduced by  Brydges, Guadagni, and Mitter \cite{BGM04} 
to better exploit  locality and independence properties and to avoid the use of cluster expansion, see \cite{BT06}  and \cite{AKM13} for further developments.
We will use the finite range decomposition as proven in \cite{Buc18} which is based on the approach by Bauerschmidt \cite{Bau13}. 
We now describe the set-up and the results in detail. 

Recall that $\mathcal{G}=(\mathbb{R}^m)^{\Ical}$ where 
$\{e_1,\ldots,e_d\}\subset\Ical\subset\{\alpha\in \mathbb{N}_0^d\setminus \{0,\ldots,0\}: |\alpha|_\infty\leq R_0\}$ 
and the extended gradient is the vector $D\p(x)=(\nabla^\alpha\p(x))_{\alpha\in \Ical}\in \mathcal{G}$ (see Section~\ref{sec:ggm}). 
For a positive definite quadratic form $\Qscr$ on $\mathcal{G}$ the expression
\begin{align}
\frac{e^{-\frac12 \sum_{x\in T_N} \Qscr(D\p(x))}}{Z}\lambda_N(\d \p)
\end{align}
 defines a Gaussian measure on $\mathcal{X}_N$, the space of vector valued fields with average zero (see \eqref{eqdefofadmissible}).
Denoting the generator of $\Qscr$ by $\boldsymbol{Q}:\Gcal\rightarrow\Gcal$, we get a corresponding elliptic finite difference
operator $\mathscr{A}$ on $\Xcal_N$
\begin{align}
\label{E:A_Q}
\mathscr{A}_{\boldsymbol{Q}}\p=\sum_{\alpha,\beta\in \Ical} (\nabla^\alpha)^\ast 
\boldsymbol{Q}_{\alpha\beta}\nabla^\beta\p.                         
\end{align}

We use  $\mathscr{C}_{\boldsymbol{Q}}=\mathscr{A}_{\boldsymbol{Q}}^{-1}:\Xcal_N\to\Xcal_N$
 to denote the covariance of the Gaussian measure generated by  $\mathscr{A}_{\boldsymbol{Q}}$.

The following theorem states the existence of a finite range decomposition for the family $\mathscr{C}_{\boldsymbol{Q}}$. 
To state the necessary bounds, we first briefly introduce some notation and facts concerning the discrete Fourier transform.
For a more detailed discussion see \cite{AKM13}.
The operator $\mathscr{A}_{\boldsymbol{Q}}:\mathcal{X}_N\to \mathcal{X}_N$ commutes with translations,
hence its inverse $\mathscr{C}_{\boldsymbol{Q}}$ also commutes with translations.
Thus there exists a unique kernel $\mathcal{C}_{\boldsymbol{Q}}:T_N\to \R^{m\times m}$ with 
$\sum_{x\in T_N} \mathcal{C}_{\boldsymbol{Q}} (x)=0$  such that 
\begin{align}
(\mathscr{C}_{\boldsymbol{Q}}\p)(x)=\sum_{y\in T_N}\mathcal{C}_{\boldsymbol{Q}}(x-y)\p(y). 
\end{align}

Recall that $L\geq 3$ is odd. 
We introduce the dual torus 
\begin{align}\label{eq:definition_dual_torus}
\widehat{T}_N=\left\{-\frac{(L^N-1)\pi}{L^N},-\frac{(L^N-3)\pi}{L^N},\ldots ,\frac{(L^N-1)\pi}{L^N} \right\}^d.
\end{align}

For $p\in \widehat{T}_N$, we define the functions $f_p:T_N\rightarrow \mathbb{C}$  by $f_p(x)=e^{i\langle p,x\rangle}$.
Then the Fourier transform $\widehat{\psi}:\widehat{T}_N\rightarrow \mathbb{C}$ of a function $\psi:T_N\rightarrow \mathbb{C}$ is defined by
\begin{align}
	\widehat{\psi}(p)=\sum_{x\in T_N} f_p(-x) \psi(x). 
\end{align}

For vector and matrix valued functions the Fourier transform is defined component-wise.
In particular, the Fourier transform diagonalises translation invariant operators
\begin{align}
	\widehat{\mathscr{C}_{\boldsymbol{Q}}\p}(p)=\widehat{\mathcal{C}}_{\boldsymbol{Q}}(p)\widehat{\p}(p). 
\end{align}
We will also use the Plancherel identity
\begin{align}\label{eq:Plancherel}
	(\p,\psi)_{T_N}=\frac{1}{L^{Nd}}\sum_{p\in \widehat{T}_N} \widehat{\p}(p)\widehat{\psi}(p). 
\end{align}

The discrete derivatives satisfy
\begin{align}\label{eq:def_qp}
	\widehat{\nabla \p}(p)=q(p)\widehat{\p}(p) 
\end{align}
with $q_j(p)=e^{ip_j}-1$ for $1\leq j\leq d$.
For $p\in \widehat{T}_N$ we have $\frac{|p|}{2}\leq |q(p)|\leq |p|$.
We use the shorthand $q(p)^\alpha=\prod_{i=1}^d q_i(p)^{\alpha_i}$ for a multiindex $\alpha\in \mathbb{N}_0^d$.
The Fourier transform of the kernel $\mathcal{A}_{\boldsymbol{Q}}$ of the operator $\mathscr{A}_{\boldsymbol{Q}}$ is
therefore given by
\begin{align}\label{eq:fourierA}
	\widehat{\mathcal{A}}_{\boldsymbol{Q}}(p)=\sum_{\alpha,\beta\in \Ical} \overline{q}(p)^\alpha \boldsymbol{Q}_{\alpha\beta}q(p)^\beta.                                                          
\end{align}
and $\widehat{\mathcal{C}}_{\boldsymbol{Q}}(p)=\left(\widehat{\mathcal{A}}_{\boldsymbol{Q}}(p)\right)^{-1}$

We consider the set of all  quadratic forms $\Qscr$ that satisfy,
\begin{align}\label{eq:conditionForQ}
\omega_0|z^\nabla|^2\leq \Qscr(z)&\leq \frac{1}{\omega_0} |z|^2
\end{align} 
for some constant $\omega_0\in (0,1)$ which is a slightly weaker condition than \eqref{eq:Qlowerbound_again}. 
Note that \eqref{eq:fourierA} then implies that there exists a constant $\omega$ such that
\begin{align}\label{Ahatestimate}
\begin{split}    
\omega|p|^2\leq \widehat{\mathcal{A}}_{\boldsymbol{Q}}(p)\leq \frac1\omega  |p|^2,                     \\
\frac{\omega}{|p|^2}\leq \widehat{\mathcal{C}}_{\boldsymbol{Q}}(p)\leq \frac{1}{\omega |p|^2},                                     
\end{split}                                                                     
\end{align}
where $\omega$  only depends on $\omega_0$, $R_0$, and $d$. 
\begin{theorem}[Theorem~2.5 in \cite{Buc18}]\label{thm:frd}
Fix  $\overline\omega_0>0$.
Consider the family of symmetric, positive operators $\boldsymbol{Q}:\Gcal\rightarrow \Gcal$ corresponding to quadratic forms $\Qscr$ 
that satisfy strict inequality \eqref{eq:conditionForQ} with $\overline\omega_0$. Let $L> 3$ be  odd, $N\geq 1$ as before and let $\tilde{n}>n$ be two integers. 
Then there exists a family of finite range decomposition $\mathscr{C}_{\boldsymbol{Q},k}$, $k=1,2,\dots,N+1$, 
of the operator $\mathscr{C}_{\boldsymbol{Q}}$ such that 
\begin{align}\begin{split}
 & \mathscr{C}_{\boldsymbol{Q}}=\sum_{k=1}^{N+1} \mathscr{C}_{\boldsymbol{Q},k},     \text{ with }                \\
 & \mathcal{C}_{\boldsymbol{Q}, k}(x)=-C_k\, \, \text{for  }\, |x|_\infty\geq \frac{L^k}{2}, 
 \end{split}\end{align}
where $C_k\geq 0$ is a constant, positive semi-definite matrix that is  independent of $\boldsymbol{Q}$.
The family $\mathscr{C}_{\boldsymbol{Q}, k}$ satisfies the following bounds where all constants may depend on $R$, $d$, $m$, 
$\overline\omega_0$, $n$, and $\tilde{n}$. 
The $\alpha$-th discrete derivative for all $\alpha$  with $|\alpha|\leq n$ is bounded by
\begin{align}\label{eq:discretebounds}
\sup_{\left|\dot{\boldsymbol{Q}}\right|\leq 1}\left|\nabla^\alpha
D_{\boldsymbol{Q}}^\ell\mathcal{C}_{\boldsymbol{Q},k}(x)(\dot{\boldsymbol{Q}},\ldots,\dot{\boldsymbol{Q}})\right|\leq
\begin{cases}
C_{\alpha,\ell} L^{-(k-1)(d-2+|\alpha|)}\;\text{for}\;d+|\alpha|>2\\
C_{\alpha,\ell}\ln(L) L^{-(k-1)(d-2+|\alpha|)}\;\text{for}\;d+|\alpha|=2.
\end{cases}
\end{align}
Further, for kernels in Fourier space we have the following lower bounds  with a constant $c>0$,  
\begin{align}\begin{split}\label{finalfrdlower}
\widehat{\mathcal{C}}_{\boldsymbol{Q},k}(p)         & \geq                           
\begin{cases}	
cL^{-2(d+\tilde{n})-1}L^{2j}L^{(k-j)(-d+1-n)} & \text{ for }  L^{-j-1}<|p|\leq L^{-j}\text{ and } j< k\\
cL^{-2(d+\tilde{n})-1}L^{2k}                           & \text{ for } |p|\leq L^{-k-1},
\end{cases}
\end{split}\end{align}
and similar upper bounds  with a constant $C$,
\begin{align}
\begin{split} \label{finalfrdupper}
\left|\widehat{\mathcal{C}}_{\boldsymbol{Q},k}(p)\right|               & \leq                           
\begin{cases}
CL^{2(d+\tilde{n})+1}L^{2j}L^{(k-j)(-d+1-n)}                  & \text{ for }                   L^{-j-1}<|p|\leq L^{-j}\text{ and } j< k\\
CL^{2k} & \text{ for } |p|\leq L^{-k-1}. 
\end{cases}
\end{split}
\end{align}
For the derivatives of the kernels with $\abs{\dot{\boldsymbol{Q}}}\leq 1$ and $\ell\geq 1$ we finally have the following stronger bounds
\begin{align}
\left|\frac{\d^\ell}{\d s^\ell}\widehat{\mathcal{C}}_{\boldsymbol{Q}+s\dot{\boldsymbol{Q}},k}(p)\right|\leq
\begin{cases}
C_{\ell}L^{(d+\tilde{n})+1}L^{2j}L^{(k-j)(-d+1-\tilde{n})}\hspace{-0.15cm} &\text{for }  L^{-j-1}<|p|\leq L^{-j},  j< k,\\
C_{\ell}L^{2k}  & \text{for } |p|\leq L^{-k-1}.
\end{cases}
\end{align}
The lower and upper bound can be combined to give, for $\ell\geq 1$
and $\boldsymbol{Q}$, $\widetilde{\boldsymbol{Q}}$ satisfying \ref{eq:Qlowerbound} 
\begin{align}\label{keyquotientbound}
\begin{split}
\left| \frac{\d^\ell}{\d s^\ell}
\widehat{\mathcal{C}}_{\boldsymbol{Q}+s\dot{\boldsymbol{Q}},k}(p)\right|& \cdot 
\left| \widehat{\mathcal{C}}_{\widetilde{\boldsymbol{Q}},k}(p)^{-1}\right| 
\\
&\leq \begin{cases}
K_\ell L^{4(d+\tilde{n})+2}L^{(k-j)(n-\tilde{n})}  \; & \text{ for }   L^{-j-1}<|p|\leq L^{-j}\text{ and } j< k,\\
K_\ell L^{2(d+\tilde{n})+1}\;          & \text{ for } |p|\leq L^{-k-1}, 
\end{cases}
\end{split}
\end{align}
where the constants $K_\ell$ do not depend on $N$ or $k$.
\end{theorem}
Let us recall one further theorem from \cite{Buc18}
that states that expectations with respect to $\mu_{\mathscr{C}_{k+1}^{\boldsymbol{Q}}}$ are differentiable in $\boldsymbol{Q}$. 
This will be a key ingredient in the proof of the smoothness of our renormalisation map.

\begin{theorem}[Theorem~4.5 in \cite{Buc18}]\label{prop:finalsmoothness}
  Let  $\mathscr{C}_{\boldsymbol{Q},k+1}$  a finite range decomposition as in Theorem
\ref{thm:frd}
with $\tilde{n}-n>d/2$ and $X\subset T_N$ be a subset with diameter
$D=\mathrm{diam}_\infty(X)\geq
L^k$. Let $F:\mathcal{V}_N\rightarrow \mathbb{R}$ be a
functional that is measurable
with respect to the $\sigma$-algebra generated by $\{\p(x)|\, x\in X\}$, i.e.,
$F$ depends only on the values of the field $\p$ in $X$. Then for $\ell\geq 1$ and $p>1$ the following bound holds
\begin{align}
 \left|\frac{\d^\ell}{\d t^\ell}\int_{\mathcal{X}_N}\left. F(\p)\, \mu_{\boldsymbol{Q}+t\boldsymbol{Q}_1,k+
 1}(\d\p)\right|_{t=0}\right|\leq C_{\ell,p}(L) (DL^{-k})^{\frac{d\ell}{2}} \left| \boldsymbol{Q}_1\right|^\ell
\lVert
F\rVert_{L^p(\mathcal{X}_N,\mu_{{\boldsymbol{Q},k+1}})}.
\end{align}
The constant depends in addition on $K_\ell$ from \eqref{keyquotientbound}
and therefore on $\omega_0$, $d$, $m$, $n$, $\tilde{n}$, and $R_0$.
\end{theorem}

We already explained  in Chapter \ref{sec:explanation} that in order to prove Theorem~\ref{th:pertcomp} 
it is not sufficient to decompose the Gaussian measure generated by $\Qscr$ but we have to consider small perturbations of this quadratic form. 
However, it is sufficient to consider only perturbations of the gradient-gradient term of the quadratic form. 
They are parametrized by symmetric maps $\boldsymbol{q}: \mathbb{R}^{d\times m}\rightarrow \mathbb{R}^{d\times m}$ 
 and we use $|\boldsymbol{q}| $ to denote   its operator norm with respect to the standard scalar product on $\mathbb{R}^{d\times m}$.
We consider the family of quadratic forms $\Qscr^{(\boldsymbol{q})}$ given by
\begin{align}
\Qscr^{(\boldsymbol{q})}(z)=\Qscr(z){-}z^\nabla\cdot \boldsymbol{q}z^\nabla
\end{align}
 and the corresponding family of operators
\begin{align}\label{eq:definition_Aq}
	\mathscr{A}^{(\boldsymbol{q})}=\sum_{\alpha,\beta\in \Ical}\sum_{i,j=1}^m (\nabla^{(\alpha,i)})^\ast 
	\boldsymbol{Q}_{(\alpha,i),(\beta,j)}\nabla^{(\beta, j)}  {-}\sum_{|\alpha|=|\beta|=1}\sum_{i,j=1}^m \boldsymbol{q}_{(\alpha,i),(\beta,j)}(\nabla^{(\alpha,i)})^\ast \nabla^{(\beta,j)} 
\end{align}
where $\nabla^{(\alpha,i)}\p(x)=\nabla^\alpha\p_i(x)$ and $\boldsymbol{Q}$ denotes the generator of $\Qscr$.
 The partition function of the Gaussian measure generated by $\mathscr{A}^{(\boldsymbol{q})}$ will be denoted by 
\begin{align}\label{eq:measurenormalisation}
	Z^{(\boldsymbol{q})}=\int_{\Xcal_N} e^{-\frac{1}{2}(\p,\mathscr{A}^{(\boldsymbol{q})}\p)}\,\d\p. 
\end{align}

In the following we will always assume that 
\begin{align}
\boldsymbol{q}\in B_{\kappa}=B_\kappa(0)\coloneqq \{\boldsymbol{q}\in \R^{(d\times m)\times (d\times m)}_{\mathrm{sym}}:\, |\boldsymbol{q}|\leq {\kappa} \}
\end{align}
 for some ${\kappa}$ with
$\kappa< \frac{\omega_0}{2}$.
Later we will impose additional conditions on $\kappa$.

Note that the family $\Qscr^{(\boldsymbol{q})}$ satisfies the condition \eqref{eq:conditionForQ} 
with $\overline\omega_0=\omega_0/2$ for $\boldsymbol{q}\in B_{\kappa}$.
To obtain our main results, we fix a finite range decomposition  
as in Theorem~\ref{thm:frd} with parameters $\overline\omega_0=\omega_0/2$ and
\begin{align}\label{eq:defn}
	n=2M=4\left\lfloor \frac{d}{2}\right\rfloor+6 ,\quad 
	\tilde{n}=n+\left\lfloor \frac{d}{2}\right\rfloor+1        
	=5\left\lfloor \frac{d}{2}\right\rfloor+7.         
\end{align}
The choice is related to the choice of the norms and a Sobolev embedding as we will see later. 
In particular we obtain a finite range decomposition $\mathscr{C}_k^{(\boldsymbol{q})}$
with kernels $\mathcal{C}_k^{(\boldsymbol{q})}$ with $1\leq k\leq N+1$
for $\boldsymbol{q}\in B_\kappa$ of the covariances $\mathscr{C}^{(\boldsymbol{q})}=\left(\mathscr{A}^{(\boldsymbol{q})}\right)^{-1}$.
To state the result in Theorem~\ref{th:weights_final}
in slightly bigger generality we will consider more general choices of parameters there.

The key property of these decompositions is their finite range   which implies
for a random Gaussian field $\p$ with covariance $\mathcal{C}^{(\boldsymbol{q})}_k$ 
that
$\mathbb{E}(\nabla_i\p(x)\nabla_j\p(y))=\nabla_j^\ast\nabla_i\mathcal{C}^{(\boldsymbol{q})}_k(x-y)=0$ if $|x-y|\geq L^k/2$, 
rendering the gradient variables $\nabla_i\p(x)$ and $ \nabla_j\p(y)$  to be independent.
In particular, this implies 
\begin{align}\label{eq:finiteRangeProp} 
\mathbb{E}(F_1(\nabla\p{\restriction_X})F_2(\nabla\p{\restriction_Y}))              
=\mathbb{E}(F_1(\nabla\p{\restriction_X}))\mathbb{E}(F_2(\nabla\p{\restriction_Y})) 
\end{align}
for sets $X$ and $Y$ such that $\mathrm{dist}(X,Y)\geq L^k/2$.
In analytic terms this means
\begin{align}
\int_{\Xcal_N} F_1(\nabla\p{\restriction_X})F_2(\nabla\p{\restriction_Y}) \,\mu_{\mathscr{C}^{(\boldsymbol{q})}_k}                                                
= \int_{\Xcal_N} F_1(\nabla\p{\restriction_X}) \,\mu_{\mathscr{C}^{(\boldsymbol{q})}_k}  \int_{\Xcal_N} F_2(\nabla\p{\restriction_Y})                             
\,\mu_{\mathscr{C}^{(\boldsymbol{q})}_k}                                                
\end{align}
We will use this \emph{factorization property} frequently in the following.
Also, we will often use the shorthand  $\mu_k^{(\boldsymbol{q})}=\mu_{\mathscr{C}^{(\boldsymbol{q})}_k}$,
dropping occasionally $\boldsymbol{q}$ from  the notation.

If $\p$ is distributed according to $\mu$ and  the fields $\p_k$ are independent and  distributed according to $\mu_k$,
the finite range decomposition amounts, in probabilistic language, to the claim that
\begin{align}
\label{p_k}
	\p\overset{\mathcal{D}}{=}\sum_{k=1}^{N+1} \p_k
\end{align}
in distribution. Or, from the analytic viewpoint, this can be expressed in terms of the convolution of measures, i.e.,
\begin{align}
	\mu=\mu_1\ast\ldots \ast\, \mu_{N+1}. 
\end{align}

The renormalisation maps  are then defined by succesive integrations,
\begin{align}\label{eq:defRenormMap}
	(\boldsymbol{R}_k^{(\boldsymbol{q})}F)(\p)=\int_{\Xcal_N} F(\p+\xi)\, \mu_k^{(\boldsymbol{q})}(\d \xi) 
	=F\ast \mu_k^{(\boldsymbol{q})}(\p)                                                  
\end{align}
for $1\leq k\leq N+1$.
Later we will define Banach spaces of functionals that will guarantee that the map
$\boldsymbol{R}_k^{(\boldsymbol{q})}$ is well defined and continuous.
If $F$ is integrable with respect to $\mu^{(\boldsymbol{q})}$ this definition implies
\begin{align}
\label{E:R...R}
	\begin{split}
	\int_{\Xcal_N} F(\p)\,\mu^{(\boldsymbol{q})}(\d \p)
	  & = \int_{\Xcal_N\times\ldots\times \Xcal_N} F\Bigl(\sum_{i=1}^{N+1}\p_i\Bigr)\,\mu_1^{(\boldsymbol{q})}(\d\p_1) 
	\ldots\mu_{N+1}^{(\boldsymbol{q})}(\d \p_{N+1})\\
	  & =(\boldsymbol{R}_{N+1}^{(\boldsymbol{q})}\ldots \boldsymbol{R}_{1}^{(\boldsymbol{q})})(F)(0)                                                 
	\end{split}
\end{align}

\section{Polymers and relevant Hamiltonians}  \label{se:polymers}

In this section 
 we define certain subsets of the torus that will be used to organize the multiscale analysis. 
We also introduce  local and shift-invariant functionals and the concept of relevant Hamiltonians.
To keep these definitions simple we introduce the constant
\begin{align}\label{eq:definition_R}
R=\max(R_0,2\lfloor d/2\rfloor+3)=\max(R_0,M)
\end{align}
depending only on the range of the interaction $R_0$ and the dimension $d$.

When working with subsets of the torus it is convenient to identify subsets of $\Z^d$ 
with their image under the projection $\Z^d\to T_N=(\mathbb{Z}/(L^N\mathbb{Z}))^d$ without reflecting this in the notation.

There is a natural hierarchical paving corresponding to the correlation range 
of random fields governed by Gaussian fields $\p_k $ introduced in \eqref{p_k}. 
Namely, for $ k=0,1,2,\ldots,N $, we pave the torus $T_N $ by $ L^{(N-k)d} $ disjoint cubes of side length $ L^k $. 
These cubes are all translates by vectors in $ L^k\Z^d $ ($L$ is odd) of the cube  $B_0\subset T_N$ which is the image
of the set $\{z\in \mathbb{Z}^d:\; |z_i|\leq \frac{L^k-1}{2}\}$ under the projection $\Z^d\to T_N$.
We call such cubes  \emph{$k$-blocks} or blocks of $ k$-th generation, and use 
$\Bcal_k=\Bcal_k(T_N)$  to denote the set of all $ k$-blocks in $T_N$,
$$
\Bcal_k=\Bcal_k(T_N)=\{B\colon B \mbox{ is a } k\mbox{-block}\},\quad k=0,1,\ldots,N.
$$
Single vertices of the torus are $0$-blocks, the starting generation for the renormalisation group transforms.. The only $N$-block is the torus $T_N$ itself, ${\Bcal}_N=\{T_N\}$.

\begin{figure}
\begin{overpic}[abs,unit=1mm,scale=.75,
]
{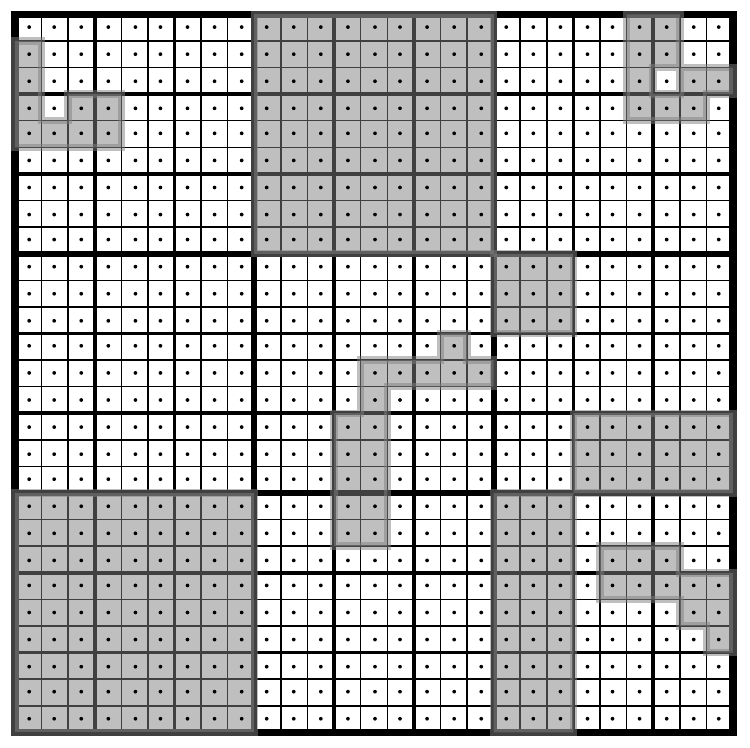}
\put(-4,80) {$X_1$}
\put(45,95) {$X_2$}
\put(82,95) {$X_1$}
\put(95,55) {$X_3$}
\put(95,35) {$X_4$}
\put(95,16) {$X_5$}
\put(65,-3) {$X_4$}
\put(45,-3) {$X_6$}
\put(16,-3) {$X_2$}
\end{overpic}
\caption{Torus $(\Z/(L\Z))^d$ with $d = 2$, $L = 3$, and $N = 3$ identified with  the square $\{z\in \Z^2: \abs{z}_\infty\le 13\}$. 
The sets $\Bcal_0$, $\Bcal_1$, and $\Bcal_2$ contain $27^2$, $9^2$, and $3^2$ blocks, respectively.
There are six connected polymers in the figure: 0-polymers $X_1$ (the two pieces touching the boundary are connected on the torus), 
$X_5$,  and $X_6$ (in the center),  1-polymers $X_3$ and $X_4$, and 2-polymer $X_2$ (again, the two pieces touching on the torus).
Only the 1-polymers $X_1$ and  $X_6$ are strictly disjoint from the other polymers---taking into account the connectedness on the torus,  
the  polymer $X_2$ is connected with each of the polymers $X_3$, $X_4$, and $X_5$.}
\label{fig:polymers}
\end{figure}

Next, we summarise a notation for particular unions and neighbourhoods of blocks:

\begin{itemize}[leftmargin=0.4cm]	
\item A union of $k$-blocks is called a $k$-\emph{polymer} and $\mathcal{P}_k$ will be the set of all $k$-polymers.
	Note that this definition of polymers  differs from the definitions inspired directly by  physics, in particular polymers need not be connected.
\item A nonempty set $X\subset T_N$ is \emph{connected} if for all $x,y\in X$ there is a sequence 
$x=x_0, x_1,\ldots , x_m=$
	$y$ with $x_i\in X$ for $0\leq i\leq m$ such that $|x_i-x_{i+1}|_\infty=1$ for $0\leq i\leq m-1$. 
	This notion corresponds to graph connectedness in the graph with vertices $\TN$ and edges between $x,y\in \TN$ if $|x-y|_\infty=1$.	 
	We illustrate connected polymers in Figure~\ref{fig:polymers}.
	  We say that connected $X, Y \subset T_N$ are \emph{touching} if $X\cup Y$ is connected.
\item Sets $A$ and $B$ are \emph{strictly disjoint} if $|a-b|_\infty>1$ for all $a\in A$ and $b\in B$. 
	 An important property is that for $X, Y\in \mathcal{P}_k$ such
	      that $X$ and $Y$ are strictly disjoint we have $\text{dist}(X,Y) > L^k$. 
	      If $\xi_k$ is distributed according to $\mu_k$ this implies that the gradient fields $\nabla \xi_k$ restricted to $X$ and $Y$ are independent by
	      the finite range property. 
\item We use $\Pck$  to denote the set of \emph{connected} $k$-\emph{polymers} (hence nonempty) and 
	 $\mathcal{C}(X)$  to denote the set of \emph{connected components of a polymer} $X$.
\item $\mathcal{B}_k(X)$ is the \emph{ set of  all } $k$-\emph{blocks} contained in a polymer $X$ 	with  $|X|_k$ denoting their \emph{number}.
\item The  \emph{closure} $\overline{X}\in\mathcal{P}_{k+1}$  of a $k$-polymer  $X\in \mathcal{P}_k$ 
	is the smallest $(k+1)$-polymer containing $X$. 
\item We say that a connected $k$-polymer $X\in \Pck$ is \emph{small}  if 
	$|X|_k\leq 2^d$. We use  $\mathcal{S}_k$ to denote the set of all small $k$-polymers. All other polymers
	      in $\Pck\setminus \mathcal{S}_k$ will be called \emph{large}. 
	      Small polymers are introduced because they need a special treatment in the renormalisation  procedure. 
	      The reason boils down to the fact  that for large polymers $X\in\Pck\setminus\mathcal{S}_k$ the
	      closure satisfies $|\overline{X}|_{k+1}\leq \alpha(d)|X|_k$ for some $\alpha(d)<1$. For $X\in\mathcal{S}_k$, 
	      however,  it is possible that  $|\overline{X}|_{k+1}=|X|_k$.
\item For any block $B\in \mathcal{B}_k$ and $k\geq 1$ let $\widehat{B}\in \mathcal{P}_k$ be the cube of side length $(2^{d+1}+1)L^k$ centred at $B$ 
              (i.e., with the same center as $B$). 
              Note that this is similar to the definition of the small set neighbourhood in \cite{AKM16} but the side length is slightly bigger.
	      For  $B\in  \mathcal{B}_0$ let $\widehat{B}\in  \mathcal{P}_0$ denote the cube centred at $B$ of side length $(2^{d+1}+2R+1)$ 
	      where $R$ denotes the range of the interaction as defined in \eqref{eq:definition_R}.
\item For any polymer $X\in\mathcal{P}_k$ and $k\geq 1$ we define the \textit{small  neighbourhood} $X^\ast\in \mathcal{P}_{k-1}$ of $X$ by
	      \begin{align}\label{eq:propOfSmallNb}
	      	X^\ast=\bigcup_{B\in \mathcal{B}_{k-1}(X)}\widehat{B}. 
	      \end{align}
	      For $k=0$ we define $X^\ast=X+[-R,R]^d\in \mathcal{P}_0$. 
	      Note that we view $\ast$ as a map from $\mathcal{P}_k$ to $\mathcal{P}_{k-1}$ for $k\geq 1$. 
	      In particular, $X^{\ast\ast}=(X^\ast)^\ast\in\mathcal{P}_{k-2}$ for $X\in\mathcal{P}_k$ and $k\geq 2$. 
	      If the scale of the considered polymer is not clear from the context it will be indicated explicitly.
	      The definition of $\widehat{B}$ implies that for  $X\in \mathcal{P}_k$, $k\geq 1$, and $x\in X^\ast$,
	      \begin{align}\label{eq:distXast}
	       \mathrm{dist}_\infty(x,X)\leq (2^d+R)L^{k-1}.
	      \end{align}
\item Finally, for any $X\in\mathcal{P}_k$ we define the \textit{large neighbourhood} 
	\begin{equation}
	\label{E:X+}
	X^+=\bigcup_{B\in \mathcal{B}_k  \text{ is   touching }X} B\ \text{ for } k\ge 1 \text{ and } X^+=X^\ast  \text{ for } k= 0.
	\end{equation}
	\end{itemize}
	Recall our convention to identify subsets of $\Z^d$ with subsets of $T_N$. 
For future reference we recapitulate the definitions of neighbourhoods:
\begin{align}
\begin{split}\label{eq:nghbhdscompact}
\widehat{B}&=	
\begin{cases}
B+[-2^d-R,2^d+R]^d\quad &\text{for $B\in \mathcal{B}_0$}\\
B+[-2^dL^k, 2^dL^k]^d\quad &\text{for $B\in \mathcal{B}_k$, $1\leq k\leq N-1$}
\end{cases}
\\
X^\ast&=
\begin{cases}
X+[-R,R]^d\quad &\text{for $X\in \mathcal{P}_0$}\\
X+[-2^d-R, 2^d+R]^d\quad &\text{for $X\in \mathcal{P}_1$}\\
X+[-2^dL^{k-1}, 2^dL^{k-1}]^d\quad &\text{for $X\in \mathcal{P}_k$, $2\leq k\leq N-1$}
\end{cases}
\\
X^+&=
\begin{cases}
X+[-R,R]^d\quad &\text{for $X\in \mathcal{P}_0$}\\
X+[-L^{k}, L^{k}]^d\quad &\text{for $X\in \mathcal{P}_k$, $1\leq k\leq N-1$}.
\end{cases}
\end{split}
\end{align}
Let us also collect several obvious geometric consequences of the definitions.

For strictly disjoint $U_1,U_2\in \Pcal_{k+1}$ and $L\geq 2^{d+2}+4R$ we have
\begin{align}\label{eq:distU1U2}
 \mathrm{dist}(U_1^\ast,U_2^\ast)\geq L^{k+1}-2(2^dL^k+R)\geq \frac{L^{k+1}}{2}.
\end{align}  
For $L\geq 2^d+R$ and $X\in \mathcal{P}_k$ we have
\begin{align}\label{XastsubsetXplus}
  X^\ast \subset X^+.
\end{align}
Indeed, for $k=0$ it holds by definition and
for $k\geq 1$ the inclusion follows from \eqref{eq:distXast}.
Moreover, for $L\geq 2^{d}+R$ and $k\geq 0$
\begin{align}\label{eq:polymernghd}
	X^\ast\subset X^+\subset Y^\ast \;\text{ for $X\in\mathcal{S}_k$ and $Y\in \mathcal{P}_{k+1}$ such that $X\cap Y\neq \emptyset$}. 
\end{align}
To verify the second inclusion,
let $B\in \mathcal{B}_{k}(X\cap Y)$.  We will show that then $X^+\subset \widehat{B}$ and thus $X^+\subset Y^\ast$.
Indeed, given that $X$ is small, it
 is contained in a cube of side length $(2^{d+1}-1)L^k$ centred at $B$.
For $k\ge 1$ this implies that
$X^+$ is contained in a cube of side length $(2^{d+1}+1)L^k$ centred at $B$, while for $k=0$ in a cube of
side length $2^{d+1}+2R+1$ centred at $B$. In both cases it implies that $X^+\subset \widehat{B}$.

Now we introduce the class of functionals we are going to work with.
We set
\begin{align}\label{eqstructoffunctionals}
\begin{split}
	M(\mathcal{P}_k,\Vcal_N)	
	=\{ F:\mathcal{P}_k\times \Vcal_N\rightarrow \mathbb{R}:\, F(X,\cdot)\in M(\Vcal_N),\; 
	F\; \text{local, transl.\ \& shift inv.}\}.
\end{split}	
	 \end{align}
 Here, $M(\Vcal_N)$ is the set of measurable real functions on $\Vcal_N$ 
with respect to the Borel $\sigma$-algebra.
Locality of $F$ is defined by assuming that $F(X,\p)$  depends only on the value of the field $\p$ on $X^\ast$, that is,
 assuming the equality
$F(X,\p)=F(X,\psi)$ to be valid whenever $\p{\restriction_{X^\ast}}=\psi{\restriction_{X^\ast}}$.  
The translation invariance of $F$ means that for any $a\in (L^k\mathbb{Z})^d$ we have
$F(\tau_a(X),\tau_a(\p))=F(X,\p)$, where $\tau_a(B)=B+a$ and $\tau_a\p(x)=\p(x-a)$.
Finally, for a local functional $F$ and a connected polymer $X$, the shift invariance means that $F(X,\p+\psi)=F(X,\p)$, 
where  $\psi$ is a constant function, $\psi(x)=c$ for $x\in X^\ast$. 
For general polymers $X$ we define the shift invariance by assuming that $F(X,\p+\psi)=F(X,\p)$ whenever
$\psi$ is a step function---a constant on each nearest neighbour graph-connected component of $X^\ast$.
Here nearest neighbour graph-connectedness refers to the usual nearest neighbour graph structure on $\TN$
(defining the set $E(\TN)$ of edges in $\TN$ as   $E(\TN)=\{\{x,y\}; x,y\in\TN$ such that $\abs{x-y}_2=1\}$ 
in contrast to the relation  $\abs{x-y}_\infty=1$ used when defining connectedness of polymers). 
Note that for $k\geq 1$ and $X\in \mathcal{P}_k$ the graph-connected components of $X^\ast$ agree with 
the connected components we defined before. 

It is convenient to define the  functionals on $\Vcal_N$ instead of $\Xcal_N$,
the space of fields with average zero which are in one-to-one correspondence with gradient fields.
 Nevertheless, all the measures $\mu_k^{(\boldsymbol{q})}$ appearing in the following are supported on $\Xcal_N$ which implies that
the functionals are only evaluated for $\p\in \Xcal_N$. 
Moreover the measures $\mu_k^{(\boldsymbol{q})}$ are absolutely continuous with respect to the Hausdorff measure on $\Xcal_N$.
Note that for $F\in M(\Vcal_N)$  such that $F(\p+c)=F(\p)$ for  any $\p\in \Vcal_N$ and any constant field $c\in \mathcal{V}_N$,  
the restriction $F{\restriction_{\mathcal{X}_N}}$ is measurable with respect to the Borel $\sigma$-algebra on $\mathcal{X}_N$. 
Indeed, the condition $F(\p+c)=F(\p)$ implies that for any Borel $O\subset \R$
with $A=(F{\restriction_{\mathcal{X}_N}})^{-1}(O)$ , we have
$F^{-1}(O)=A\times \Xcal_N^\perp \subset \Xcal_N\oplus \Xcal_N^\perp=\Vcal_N$ .

Let us  formulate an equivalent characterisation of shift invariance.
For any subset $X\subset \TN$ we introduce the set of edges $E(X)= \{\{x,y\}\in E(\Z^d); x,y\in X, \abs{x-y}_2=1\}$ 
and  the set of directed edges $\vec E(X)=\{(x,y),(y,x), \{x,y\}\in E\}$.
For $\p\in \Vcal_N$ we can view $\nabla \p$  as a
function  from $\vec{E}(\TN)$  to  $ \R^m$  by taking $\nabla \p((x,x+e_i))=\nabla_i\p(x)=\p(x+e_i)-\p(x)$ 
and $\nabla \p((x+e_i,x))=\nabla_i^*\p(x+e_i)=\p(x)-\p(x+e_i)$.

\begin{lemma}\label{le:shiftinvvsgradients}
A functional $F:\mathcal{P}_k\times\Vcal_N\rightarrow \mathbb{R}$
 is local and shift invariant  iff
 for each $X\in \mathcal{P}_k$ there is a functional $\widetilde{F}_X:\vec{E}(X^\ast)\to \mathbb{R}$ 
 such that $F(X,\p)=\widetilde{F}_X(\nabla\p{\restriction_{\vec{E}(X^\ast)}})$ for any $\p\in\Vcal_N$
 , i.\ e., $F(X,\cdot)$ is measurable with respect to the $\sigma$-algebra generated by
 $\nabla\p{\restriction_{\vec{E}(X^\ast)}}$.
\end{lemma}

\begin{proof}
We first observe that for a graph connected set $Y$, a fixed  $y\in Y$, and $\eta: \vec{E}(Y)\to \R^m$,
 there is at most one function $\widetilde{\p}:Y\to \R^m$ such that $\widetilde{\p}(y)=0$ and $\nabla\widetilde{\p}=\eta$.
 Note that a necessary condition is that  $\eta((x,y))=-\eta((y,x))$ for any $(x,y)\in \vec{E}(Y)$.)
Indeed,  if there were two such functions  $\widetilde{\p}_1$ and $\widetilde{\p}_2$ 
and a point $z\in Y$ such that $\widetilde{\p}_1(z)\neq \widetilde{\p}_2(z)$,  we would get a contradiction since 
$\widetilde{\p}_1(z)-\widetilde{\p}_2(z)= \sum_{i=0}^{n} \eta_{(x_i,x_{i+1})}- \sum_{i=0}^{n} \eta_{(x_i,x_{i+1})}=0$ for any path $x_0=y,x_1,\ldots , x_n=z$. 
Such a path exists since the graph $(Y,E(Y))$ is connected. 

Now, let $F$ be shift invariant and local, $Y_1,Y_2, \dots, Y_n$ be the graph connected components of $X^\ast$,
and let $y_i\in Y_i, i=1,2,\dots,n$.  
Note that the argument above implies that for $\eta: \vec{E}(X^\ast)\to \R^m$ there is at most one $\p\in \mathcal{V}_N$ such
that $\p(y_i)=0$ and $\nabla\p{\restriction_{\vec{E}(X^\ast)}}=\eta$.
Then we define $\widetilde{F}_X(\eta)=F(\p)$ if such a $\p$ exists 
and $\widetilde{F}_X(\eta)=0$ otherwise.
For $\p\in \mathcal{V}_N$ define $\widetilde{\p}(x)=\p(x)-\sum_i \p(y_i)\mathbf{1}_{Y_i}(x)$.
Then by shift invariance we have
\begin{align}
F(X,\p)=F(X,\widetilde{\p})=\widetilde{F}_X(\nabla\widetilde{\p}{\restriction_{\vec{E}(X^\ast)}})=\widetilde{F}_X(\nabla\p{\restriction_{\vec{E}(X^\ast)}}).
\end{align}
 The  opposite implication is obvious.
\end{proof}

In addition to the set $M(\mathcal{P}_k,\Vcal_N)$ of functionals  we consider its obvious generalizations
$M(\Pck,\Vcal_N)$, $M(\mathcal{S}_k,\Vcal_N)$ and $M(\mathcal{B}_k, \Vcal_N)$.
We often shorten the notation to  $M(\mathcal{P}_k), M(\Pck), M(\mathcal{S}_k)$, and $M(\mathcal{B}_k)$.
Note that there are two
canonical inclusions $\iota_1: M(\mathcal{B}_k)\rightarrow M(\Pck)$ and $\iota_2:
M(\Pck)\rightarrow M(\mathcal{P}_k)$
given by $(\iota_1F)(X,\p)=\prod_{B\in \mathcal{B}_k(X)} F(B,\p)$ and  $(\iota_2F)(X,\p)=\prod_{Y\in \mathcal{C}(X)} F(Y,\p)$, respectively.
In the following we will usually drop $\iota$ from the notation and  write $F(X,\p)=F^X(\p)$ for 
$F\in M(\mathcal{B}_k)$ and $F(X,\p)=\prod_{Y\in \mathcal{C}(X)} F(Y,\p)$  for $
F\in M(\Pck)$.
 Vice versa, for functionals  $F\in M(\Pcal_k)$ that factor, i.e., that satisfy $F(X\cup Y)=F(X)F(Y)$ 
 for strictly disjoint polymers $X,Y\in \Pcal_k$ we can consider the restriction $F{\restriction_{\Pck}}$ that then satisfies $\iota_2(F{\restriction_{\Pck}})=F$. 
We will sometimes suppress the restriction in the notation. 

This leads  to a useful convention:
\begin{align}\label{convention}
\begin{split}    
&\text{\emph{Functionals in $M(\Pcal_{k})$ that factor over strictly disjoint polymers}}\\
&\text{\emph{are identified with functionals in $M(\Pck)$. }}
\end{split}  
\end{align}

We endow  $ M(\mathcal{P}_k)$ with an associative and commutative product $\circ_k$ (circle product),
\begin{align}\label{eq:def_circle_product}
	(F_1\circ_k F_2)(X,\p)=\sum_{Y\in\mathcal{P}_k(X)}F_1(Y,\p)F_2(X\setminus Y,\p),\ \ F_1,F_2\in M(\mathcal{P}_k)
\end{align}
that is useful to streamline the notation.
Recall that we introduced this operation already in \eqref{eq:def_circ_explanation} in Chapter~\ref{sec:explanation}.
Note that the circle product and its domain of definition depends on the scale $k$.
We will mostly simply write $\circ$ as the scale $k$ will usually be clear from context.
 An important property of the circle product is that it serves as a shorthand for the expansion of the product
\begin{align}
\label{eq:Bcirc}
	(F_1+F_2)^X(\p)=(F_1\circ F_2)(X,\p)
\end{align}
with $F_1,F_2\in M(\mathcal{B}_k)$. In particular this corresponds to \eqref{eq:fac_circ_explanation}
for $k=0$.

Finally, we
introduce the space of relevant Hamiltonians $M_0(\mathcal{B}_k)\subset M(\mathcal{B}_k)$ given by all functionals of the
form 
\begin{align}   \label{eq:description_relevant_hamiltonian}
	H(B,\p)=\sum_{x\in B}\mathcal{H}(\{x\},\p) 
\end{align}
where $\Hcal(\{x \},\p)$ is a linear combination of the following  \textit{relevant monomials}\,:
\begin{itemize}[leftmargin=0.4cm]
\item The constant monomial ${ \Mscr_\emptyset}(\{x\})(\p) \equiv 1$;
\item the linear monomials ${ \Mscr_{i,\alpha}}(\{x\})(\p) := \nabla^{i, \alpha} \p(x) := \nabla^\alpha \p_i(x)$ with $1\le |\alpha| \le \lfloor d/2 \rfloor + 1$;
\item the quadratic monomials ${ \Mscr_{(i, \alpha), (j, \beta)}}(\{x\})(\p) = \nabla^{\alpha}\p_i(x) \, \nabla^\beta \p_j(x)$ with $|\alpha| = |\beta|=1$.
\end{itemize}
The rationale for declaring exactly these monomials as relevant is based on the following heuristic argument concerning
the decay of their expectations under the measures $\mu_{k}$:
Let us assign  the \textit{scaling dimension} $[\p] = \frac{d-2}{2}$ to the field $\p$, and the scaling dimension
$[{ \Mscr_\mpzc}] = r [\p] + \sum_{i=1}^r |\alpha_i|$
to a general monomial ${ \Mscr_\mpzc}(\{x\})(\p) = \nabla^{\alpha_1}   \p_{i_1}(x)  \cdots \nabla^{\alpha_r} \p_{i_r}(x)$ (with $\alpha_i \ne 0$). 
The relevance of the scaling dimension follows from the asymptotics $\mathbb{E}_{\mu_{k}} |{ \Mscr_\mpzc}(\{ x\})|^2 \sim L^{-2k[{ \Mscr_\mpzc}]}$ 
and the fact that, by the smoothness properties of correlations of $\mu_{k}$, we expect that the fields 
$\p(x)$ and $\p(y)$ are correlated only if $|x-y| \leq c L^d$. 
 As a result,  for a $k$-block $B$ we get  $\mathbb{E}_{\mu_{k}}\bigl( \sum_{x \in B} { \Mscr_\mpzc}(\{ x\})^2\bigr)  \sim L^{-2k[{ \Mscr_\mpzc}] + 2kd}$.
  Hence the relevant monomials are exactly those for which the expectation of  $ |\sum_{x \in B} { \Mscr_\mpzc}(\{ x\})|$ under $\mu_{k}$  
  is not expected to decay for large $k$. 
One often  calls the monomials for which this quantity grows with $k$ \textit{relevant}, 
those for which it remains of order $1$ \textit{marginal} and those for which it decays as \textit{irrelevant}. 
To avoid clumsy notation such as 'not irrelevant'  or 'relevant or marginal' we include marginal monomials into our list of relevant polynomials. 

 Any $H\in M_0(\mathcal{B}_k)$  is clearly shift invariant and local (the fact that $B+[-R,R]\cap T_N\subset B^+$ once $R\geq \lfloor d/2\rfloor +1$
 implies that  $H\in M_0(\mathcal{B}_k)$  and thus $M_0(\mathcal{B}_k)\subset M(\mathcal{B}_k)$). 

\section{Definition of the renormalisation map}\label{sec:defRenormTransform}

In this section    we define the flow of the functionals under the renormalisation maps \eqref{eq:defRenormMap}. 
Our definition essentially agrees with the definition of the renormalisation map in \cite{Bry09}
except for one important difference in the regrouping of the terms. 
Instead of distributing the contribution of a small polymer $X\in \Scal_k$ to all blocks $B'\in \mathcal{B}_{k+1}$ that intersect $X$, 
we add the contribution to a single block $B'$ such that $B'\cap X\neq \emptyset$. 
To implement this we consider a suitable  map  $\pi:\Scal_k\to \Bcal_{k+1}$ 
that is invariant under translation $\tau_a$ with $a \in (L^{k+1}\Z)^d$ and then add the contribution of $X$ to $B'=\pi(X)$ 
(see \eqref{eq:defofpi} and \eqref{eq:pifactor} below).
This change has no influence on our results but it simplifies some arguments, e.g., 
the projection in \eqref{eq:def_projection} below needs to be defined only for single blocks. 
Adding the contribution of $X$ to a block $B'$ such that $X\not\subset B'$
requires us to add contributions of the form $(e^{-H})^{-1}$
to $K$ (see \eqref{eq:renormmap_reblocked} below). In our case this poses no problem
because the domain of definition and the norm of $H$ is invariant under $H\to -H$.
This is not true in general, in particular for  $\p^4$-theory the coefficient of the $\p^4$ term must be positive 
to obtain functionals that are integrable with respect to the Gaussian measures $\mu_k$.

The flow will be described by two sequences of functionals $H_k$ and $K_k$.
The coordinate $H_k\in M_0(\mathcal{B}_k)$ stems from the finite dimensional space of relevant Hamiltonians 
and collects the relevant and marginal directions 
whereas the perturbation $K_k\in M(\Pck)$ is an element of an infinite dimensional space 
that collects all remaining irrelevant directions of the model.
We introduce the map $\myT_k$ that maps the operators $H_k$ and $K_k$ to the next scale operators $H_{k+1}$ and $K_{k+1}$. 
Formally it is given by a map
\begin{align}  
\label{eq:def_RG_map1}
\myT_k:M_0(\mathcal{B}_k)\times  M(\Pck)\times \mathbb{R}^{(d\times m)\times (d\times m)}\rightarrow M_0(\mathcal{B}_{k+1})\times M(\Pckp),
\end{align}
where we reflected the fact that it also depends on the \emph{a priori} tuning matrix $\boldsymbol{q}\in \mathbb{R}^{(d\times m)\times (d\times m)}$ 
which is mostly suppressed in the notation in this section.
Using $\boldsymbol{R}_{k}^{(\boldsymbol{q})}$ or a shorthand $\boldsymbol{R}_{k}$ for $\myT_k$ with a fixed $\boldsymbol{q}$,
the key requirement for the renormalisation transformation is the identity
\begin{align}\label{eq:mainproprenorm}
	\boldsymbol{R}_{k+1}^{(\boldsymbol{q})}(e^{-H_k}\circ K_k)(\TN,\p)=(e^{-H_{k+1}}\circ K_{k+1})(\TN,\p). 
\end{align}
Moreover the renormalisation transformation must be chosen in such a way that the  map $K_k\mapsto K_{k+1}$  is contracting.
For most polymers this will follow from the definition of the norms and the fact that typically the number of  blocks
decreases when the scale is changed, i.e., $|\overline{X}|_{k+1}< |X|_k$.
 However, for $k$-blocks $X\in \mathcal{B}_k$, and in general also for $X\in \mathcal{S}_k$,  this is not true, $|\overline{X}|_{k+1}= |X|_k$.  
 As a result, we have to subtract the dominant part of their contribution and include it in the Hamiltonian $H_{k+1}$.
The process of selection of the relevant part that is to be included to the space of relevant Hamiltonians determines a projection 
\begin{align}\label{eq:def_projection}
\Pi_2 : M(\Bcal_k) \to M_0(\Bcal_k).
\end{align}
Existence, boundedness and further properties of this projection are  discussed in 
Section~\ref{se:projection_Pi2} below. 
Slightly informally, $\Pi_2 F$ is defined as a  ``homogenization''  of  the second order Taylor expansion $T_2$ around zero.
Namely, considering the second order
Taylor expansion of $F(B)$ given by $\dot \p \mapsto  F(B)(0) + DF(B)(0)(\dot \p) + \frac12 D^2 F(0)(\dot \p, \dot \p)$, 
we define $\Pi_2 F$ as the ideal Hamiltonian $F(B)(0) + \ell(\dot{\p})+Q(\dot{\p},\dot{\p})$ 
where $\ell$ is the unique linear relevant Hamiltonian  
that satisfies the condition $\ell(\dot{\p})=DF(B,0)(\dot{\p})$ for all $\dot{\p}$ whose restriction to $B^{+}$ is a polynomial of degree $\lfloor d/2 \rfloor + 1$  
and similarly $Q$ is the unique quadratic relevant Hamiltonian that agrees with $D^2F(B,0)$ on all functions whose restriction to   $B^{+}$ is affine.
Note that $B^{+}$ does not wrap around the torus for $k\leq N-1$ and $L\geq 5$ and, 
as a consequence, the condition that $\p$ restricted to $B^{+}$ is a polynomial is well defined.

We defer the definition of $\boldsymbol{T}_k$ and first motivate its definition with a sequence of manipulations starting
with the left hand side of \eqref{eq:mainproprenorm}.
We define the relevant Hamiltonian on the next scale by
\begin{align}
	H_{k+1}(B',\p)=\sum_{B\in \mathcal{B}_k( B')} \widetilde{H}_{k}(B,\p) 
\end{align}
where $\widetilde{H}_k(B,\p)$ is defined by 
\begin{align}
\label{eq:tildeH}
	\widetilde{H}_k(B,\p)=\Pi_2\boldsymbol{R}_{k+1}H_k(B,\p) -\Pi_2\boldsymbol{R}_{k+1}K_k(B,\p). 
\end{align}
Note that we  need to subtract only the contributions that stem from a single block.
In the following we write 
\begin{align}\label{eq:defIItildeJ}
I(B,\p+\xi)=\exp(-H_k(B,\p+\xi)),\quad \tilde{I}(B,\p)=\exp(-\widetilde{H}_k(B,\p)),\quad
\text{and}\quad \tilde{J}=1-\tilde{I}.
\end{align}
Using repeatedly the identities  \eqref{eq:Bcirc}, we rewrite the initial integral in \eqref{eq:mainproprenorm} in terms of the next scale Hamiltonian,
\begin{align}
\begin{split}\label{eq:factorization_circle_product}
	\int_{\mathcal{X}_N}
	I(\p+\xi)&\circ K_k(\p+\xi)\,\mu_{k+1}(\d\xi)
	\\
	&=\int_{\mathcal{X}_N}\tilde{I}(\p)\circ (I(\p + \xi)-\tilde{I}(\p))\circ K_k(\p+\xi)\, \mu_{k+1}(\d\xi)
	\\
	&=
	\int_{\mathcal{X}_N}\tilde{I}(\p)\circ (1-\tilde{I})(\p)\circ (I-1)(\p+\xi)\circ K_k(\p+\xi)
	\, \mu_{k+1}(\d\xi)
	\\
	&=\tilde{I}(\p)\circ \left(\int_{\mathcal{X}_N} \tilde{J}(\p)\circ(I-1)(\p+\xi)\circ K_k(\p+\xi)     \, \mu_{k+1}(\d\xi)\right).
	\end{split}
\end{align}
This allows  us to introduce an intermediate perturbation functional 
$\widetilde{K}_k:\mathcal{P}_k\times \Vcal_N\times \Vcal_N\rightarrow \mathbb{R}$  by
\begin{align}
\label{tildeK}
	\widetilde{K}_k(X,\p,\xi)=(\tilde{J}(\p)\circ (I(\p+\xi)-1)\circ K_k(\p+\xi))(X). 
\end{align}
The initial integral \eqref{eq:defRenormMap} then becomes
\begin{align}\label{eq:RIKeqKtilde}
	\boldsymbol{R}_{k+1}(e^{-H_k}\circ K_k)(\TN,\p)=
	\sum_{X\in\mathcal{P}_k(\TN)}\tilde{I}^{\TN\setminus X}(\p)\int_{\Xcal_N}\widetilde{K}_k(X,\p,\xi)\, 
	\mu_{k+1}(\d \xi).                                                                                              
\end{align}
In the next step we regroup the terms in a such a way that
we obtain an expression in the form $e^{-H_{k+1}}\circ K_{k+1}$ with 
$H_{k+1}\in M_0(\mathcal{B}_{k+1})$ and $K_{k+1}\in M(\Pckp)$. 
For  $X\in \mathcal{P}_k\setminus\mathcal{S}_k$ we just include the contribution of the integral of $\widetilde{K}_k(X,\p,\xi)$ to 
the terms labelled by $U=\overline{X}$ in $K_{k+1}$. 
Introducing,   on the spaces $M(\mathcal{P}_{k})$, the norms with the weight $A^{\abs{X}_k}$ 
 we will prove the contractivity of the linearization of the map $\boldsymbol{T}_k$. 
 For     $A>1$ and  $X\in \mathcal{P}_k\setminus \mathcal{S}_k$ for which we can show that $\abs{X}_{k+1} < \abs{X}_k$, 
 this is based on the suppression factor  $A^{\abs{X}_{k+1}-\abs{X}_k}$.
However, for $X\in  \mathcal{S}_k$ this strategy does not work since we might have $\abs{X}_{k+1} = \abs{X}_k$. 
In this case, as explained above, we have to include the dominant part of their contribution into the Hamiltonian $H_{k+1}$ as anticipated in \eqref{eq:tildeH}.

In addition, for $X\in  \mathcal{S}_k$, we have to determine to which of the blocks $B'\in \mathcal{B}_{k+1}$, 
among those that intersect $X$,  we attribute the  corresponding contribution. 
This is achieved in the following claim.
There exists a map $\widetilde\pi:\Pck\rightarrow {\Pckp}$ that is translation invariant, i.e.,
$\widetilde\pi(\tau_a X)=\tau_a\widetilde\pi(X)$ for $a\in (L^{k+1}\mathbb{Z})^d$ and satisfies 
\begin{align}\label{eq:defofpi}
	\widetilde\pi(X)=
	\begin{cases}
	  & \overline{X}\quad\text{if }\; X\in \Pck\setminus \mathcal{S}_k,                  \\
	  & B' \quad \text{where }\; B'\in \mathcal{B}_{k+1}\; \text{with } B'\cap X\neq \emptyset\quad \text{for }\, X\in \mathcal{S}_k.
	\end{cases} 
\end{align}
We then extend $\widetilde\pi$ to a map $\pi:\mathcal{P}_k\rightarrow {\mathcal{P}_{k+1}}$ defined on all polymers by
\begin{align}\label{eq:pifactor}
	\pi(X)=\bigcup_{Y\in \mathcal{C}(X)} \widetilde\pi(Y) \quad \forall\; X\in \mathcal{P}_k.
\end{align}
To show the existence of  a map $\widetilde\pi$, it suffices to define the image $\widetilde\pi(X)$ for any $X\in \mathcal{S}_k$ and show
that the resulting $\widetilde\pi $ is translation invariant. 
``Unwrapping'' the torus $T_N$, it can be viewed as a projection $P:\Z^d\to T_N$  with the preimage of any point $ x\in T_N$ being the set 
$P^{-1}(\{x\})=\{\tau_a x, a\in (L^{N}\mathbb{Z})^d\}$.
The preimage of any $X\in \mathcal{S}_k$ is a collection of sets $\{\tau_a \widehat X, a\in (L^{N}\mathbb{Z})^d\}$, 
where $\widehat X\subset \Z^d$ can be chosen as  a connected set $\widehat X\subset \Z^d$ (recall that any $X\in \mathcal{S}_k$ is connected, 
see Section~\ref{se:polymers}) for which $X=P(\widehat X)$. 
For any $X\in \mathcal{S}_k$ consider the  $k$-block  $B(X)\in \mathcal{B}_k(X)$ such that the preimage of its centre in $\widehat X$  
is the first one in the lexicographic order in $\mathbb{Z}^d$  among the preimages in $\widehat X$ of centres of  $k$-blocks in $\mathcal{B}_k(X)$.
We  determine  the image $\widetilde\pi(X)$ as the $(k+1)$-block $B'=\overline{B(X)}$.
Translation invariance of the map $\widetilde\pi$  follows immediately from the fact that  $B(\tau_a X)=\tau_a B(X)$ for any $a\in (L^{k}\mathbb{Z})^d$.

We claim that for $X\in \mathcal{P}_{k}$ and $L\geq 2^d+ R$,
\begin{align}\label{eq:pisetincl}
	\mathcal{P}_{k-1}\ni X^\ast\subset \pi(X)^\ast\in \mathcal{P}_k. 
\end{align}
By \eqref{eq:pifactor}, it is sufficient to show this for $X$ connected.
If $X\in \Pck\setminus \mathcal{S}_k$ the claim follows  by \eqref{eq:defofpi}. 
For $X\in \mathcal{S}_k$ this is a consequence of \eqref{eq:polymernghd} applied with $Y=\pi(X)$.

We define the function $\chi:\mathcal{P}_k\times \mathcal{P}_{k+1}\rightarrow \mathbb{R}$ by
\begin{align}
	\chi(X,U)=\mathbb{1}_{\pi(X)=U}. 
\end{align}
This definition ensures that $\sum_{U\in\mathcal{P}_{k+1}}\chi(X,U)=1$.
Using the relation \\
$(\TN\setminus X)\cup(X\setminus U)=
T_N\setminus(X\cap U)=(\TN\setminus U)\cup (U\setminus X)$ we rearrange the right hand side of \eqref{eq:RIKeqKtilde},
\begin{align}\label{eq:renormmap_reblocked}
\begin{split}
	\boldsymbol{R}_{k+1}&(e^{-H_k}\circ K_k)(\TN,\p)
	\\
	&=\sum_{U\in \mathcal{P}_{k+1}}{\tilde I}^{\TN\setminus U}(\p)\left[\sum_{X\in \mathcal{P}_k}\chi(X,U) \tilde{I}^{U\setminus 
	X}(\p)\tilde{I}^{-(X\setminus U)}(\p)\int_{\mathcal{X}_N}\widetilde{K}_k(X,\p,\xi)\,\mu_{k+1}(\d\xi) \right]                                                                       
\end{split}\end{align}
where the shorthand expression $I^{-X}=(I^X)^{-1}$ was used. Therefore we define
\begin{align}\label{eq:defofKprime}
	K_{k+1}(U,\p)=\sum_{X\in \mathcal{P}_k}\chi(X,U) \tilde{I}^{U\setminus 
	X}(\p)\tilde{I}^{-(X\setminus U)}(\p)\int_{\mathcal{X}_N}\widetilde{K}_k(X,\p,\xi) \,\mu(\d\xi)    
\end{align}
for any $U\in \mathcal{P}_{k+1}$. 

\begin{lemma}   \label{le:K_kplus_factors}
	For $H_k, \widetilde{H}_k\in M_0(\mathcal{B}_k)$, $I$, $\tilde{I}$, and $\tilde{J}$ as in
	\eqref{eq:defIItildeJ}, and $L\geq 2^{d+2}+4R$ the functional $K_{k+1}$ defined in \eqref{eq:defofKprime} has the following properties
	\begin{enumerate}[label=\roman*)]
		\item If $K_k$ is  translation invariant on scale $k$,
		i.e., $K_k(X,\p)=K_k(\tau_aX,\tau_a\p)$ for $a\in (L^k\mathbb{Z})^d$ then  $K_{k+1}$ is
		translation invariant on scale $k+1$;
		\item If $K_k$ is local, i.e., $K_k(X,\p)$ only depends on the values of $\p$ in $X^\ast$ then $K_{k+1}$ is local;
		\item If $K_k$ is invariant under shifts
		then $K_{k+1}$ is also shift invariant;
		\item \label{it:4KinM} If $K_k\in M(\mathcal{P}_k)$ then $K_{k+1}\in M(\mathcal{P}_{k+1})$;
		\item \label{it:5K_factors} If $K_k {\in M(\mathcal{P}_k)}$ factors on the scale $k$, i.e.,
		      \begin{align}
		      	K_k(X_1\cup X_2,\p)=K_k(X_1, \p)K_k(X_2,\p) \quad \text{for strictly disjoint }\; X_1,X_2\in \mathcal{P}_k, 
		      \end{align}
		      then $K_{k+1}$ factors on scale $k+1$, i.e.,
		      \begin{align}
		      	K_{k+1}(U_1\cup U_2,\p)=K_{k+1}(U_1, \p)K_{k+1}(U_2,\p) 
			\quad \text{for strictly disjoint }\; U_1,U_2\in \mathcal{P}_{k+1}. 
		      \end{align}
	\end{enumerate}
The last property shows that for $K_k\in M(\Pck)$ we have
	$K_{k+1}\in M(\Pckp)$ using our convention \eqref{convention} to identify
	functionals in $M(\Pcal_{k+1})$ that factor with functionals in $M(\Pckp)$.
\end{lemma}

\begin{proof}
The first claim is a consequence of the translation invariance of $K_k, H_k$, and $\pi$.
	
For the second claim we observe that $I(X,\p)$, $\tilde{J}(X,\p)$, $\tilde{I}(X,\p)$, and $K_k(X,\p)$ only depend
on the values of 	$\p{\restriction{X^\ast}}$.
Moreover $\chi(X,U)=1$ implies by \eqref{eq:pisetincl} for $L\geq 2^d +R$ that $X^\ast\subset U^\ast$.
Since the $\ast$ operation is monotone and the renormalisation map $\boldsymbol{R}_{k+1}$ preserves locality, the functionals  $K_{k+1}(U,\p)$ only depend on $ \p{\restriction_{U^\ast}}$. 
	
To prove the shift invariance of $K_{k+1}$, we first notice that  $I$, $\tilde{I}$, and $\tilde{J}$ are shift invariant because they are compositions of
 $H_k, \widetilde{H}_k\in M_0(\mathcal{B}_k)$ with the exponential.
Thus, by Lemma \ref{le:shiftinvvsgradients}, the functionals  $K_k(X,\p)$, $I(X,\p)$, $\tilde{I}(X,\p)$, and $\tilde{J}(X,\p)$ can be
rewritten as functionals of  $\nabla \p{\restriction_{\vec{E}(X^\ast)}}$.
As in the locality argument this implies that $K_{k+1}$ can be written as 
$K_{k+1}(U,\p)=\widetilde{K}_k(U,\nabla\p{\restriction_{\vec{E}(U^\ast)},0})$ and is thus shift invariant.

The claim \ref{it:4KinM} is a consequence of the first three claims.
	
To prove the last claim, we observe that functionals $F, G\in M(\mathcal{P}_k)$ that factor on the scale $k$, satisfy   the equality 
$(F\circ G)(X\cup Y)=(F\circ G)(X) (F\circ G)(Y)$ whenever  the  polymers $X$ and $Y$ are strictly disjoint. 
Indeed,
\begin{align}\label{eq:circfactor}
\begin{split}		
		(F\circ G)(X\cup Y)
		&=\sum_{Z\subset X\cup Y} F(Z)G((X\cup Y)\setminus Z)
	\\	&=\sum_{\substack{Z_1\subset X \\ Z_2\subset Y}} F(Z_1\cup Z_2)G(X\cup Y\setminus (Z_1\cup Z_2))
\\		&=
		\sum_{Z_1\subset X }\sum_{Z_2\subset Y} F(Z_1) F(Z_2)G(X\setminus Z_1)G( Y\setminus  Z_2)
		\\
		&=(F\circ G)(X)(F\circ G)(Y).
		\end{split}
\end{align}
	Given that the circle product is associative, we can extend this to three functionals:
	the product $F\circ G\circ H$ factors  if $F$, $G$, and $H$ factor. In particular, the functional $\widetilde{K}_k$ 
	defined in \eqref{tildeK} factors on the scale $k$.

	 Let $U_1, U_2\in \mathcal{P}_{k+1}$ be strictly disjoint polymers and
	let $X\in \mathcal{P}_k$ be a polymer such that we have $\bigcup_{Y\in\mathcal{C}(X)}\pi(Y)=\pi(X)=U=U_1\cup U_2$. 
	We claim that there is a unique decomposition
	$X=X_1\cup X_2$ such that $X_1$ and $X_2$ are strictly disjoint and satisfy $\pi(X_i)=U_i$.

For the existence, consider $X_1=U_1^\ast\cap X$, $X_2=U_2^\ast\cap X$. 
Clearly, $X_1$ and $X_2$ are strictly disjoint and 	$X_1\cup X_2=X$
since by \eqref{eq:pisetincl} we know that $X\subset (U_1\cup U_2)^\ast=U_1^\ast \cup U_2^\ast$. 
The inclusions $\pi(X_i)\subset U_i^+$ together with $U_1^+\cap U_2=U_1\cap U_2^+=\emptyset$ 
and $U=\pi(X)=\pi(X_1)\cup \pi(X_2)$ imply that $\pi(X_i)=U_i$.

Uniqueness follows from the observation that $\pi(\widetilde{X}_i)=U_i$ implies by \eqref{eq:pisetincl} that $\widetilde{X}_i\subset U_i^\ast$, 
and thus $\widetilde{X}_i\subset X_i$.

	 Assuming $L\geq 2^{d+2}+4R$ and using \eqref{eq:distXast} and \eqref{eq:distU1U2}, we conclude that the distance between 
	 $X_1^\ast $ and $X_2^\ast$ is bigger than the range of $\mu_{k+1}$, 
	\begin{align}\label{eq:distX1X2}
		\mathrm{dist}(X_1^\ast,X_2^\ast)\geq \mathrm{dist}(U_1^{\ast},U_2^{\ast})
		\geq \frac{L^{k+1}}{2} .
	\end{align}
	Thus, using that $\widetilde{K}_k$ factors on scale $k$, we  get
	\begin{align}
	\begin{split}
		\int_{\Xcal_N}\widetilde{K}_k(X_1\cup                                                                                                          
		X_2,\p,\xi)\, \mu_{k+1}(\d \xi)
		&=\int_{\Xcal_N}\widetilde{K}_k(X_1,\p,\xi)\widetilde{K}_k(X_2,\p,\xi)\, \mu_{k+1}(\d\xi)
	\\		
		&=\int_{\Xcal_N}\widetilde{K}_k(X_1,\p,\xi)\,\mu_{k+1}(\d\xi)\int_{\Xcal_N}\widetilde{K}_k(X_2,\p,\xi)\, \mu_{k+1}(\d\xi).
		\end{split}
	\end{align}
	Finally, we observe that 
	\begin{align}
		\begin{split}
		(X_1\cup X_2)\setminus(U_1\cup U_2) & =(X_1\setminus U_1)\cup (X_2\setminus U_2)  \\
		(U_1\cup U_2)\setminus(X_1\cup X_2) & =(U_1\setminus X_1)\cup (U_2\setminus X_2). 
		\end{split}
	\end{align}
	The inclusion '$\subset$' holds in general, the other inclusion follows from $X_1\cap U_2=X_2\cap U_1=\emptyset$. 
	
	As a result, using manipulations similar to \eqref{eq:circfactor} for  strictly disjoint $U_1, U_2\in\mathcal{P}_{k+1}$, these facts imply 
	\begin{align}
	\begin{split}
		&K_{k+1}(U_1  \cup U_2,\p)                                                                                                
		    \\
		    &=\sum_{X\in \mathcal{P}_k}\chi(X,U_1\cup U_2) \tilde{I}^{(U_1\cup U_2)\setminus                                       
		X}(\p)\tilde{I}^{-(X\setminus (U_1\cup U_2))}(\p)\int_{\mathcal{X}_N}\widetilde{K}_k(X,\p,\xi)
		\,\mu_{k+1}(\d \xi)\\
		       & =\sum_{X\in \mathcal{P}_k}\mathbb{1}_{\pi(X)=U_1\cup U_2} \tilde{I}^{(U_1\cup U_2)\setminus                           
		X}(\p)\tilde{I}^{-(X\setminus (U_1\cup U_2))}(\p)\int_{\mathcal{X}_N}\widetilde{K}_k(X,\p,\xi)
		\,\mu_{k+1}(\d \xi)\\
		      &  =\hspace{-0.15cm}\sum_{X_1, X_2\in \mathcal{P}_k}\mathbb{1}_{\pi(X_1)=U_1}\mathbb{1}_{\pi(X_2)=U_2} 
		       \frac{\tilde{I}^{(U_1\cup U_2)\setminus 
		(X_1\cup X_2)}(\p)}{\tilde{I}^{(X_1\cup X_2)\setminus (U_1\cup U_2)}(\p)}
		\int_{\mathcal{X}_N}\widetilde{K}_k(X_1\cup X_2,\p,\xi)
		\,\mu_{k+1}(\d \xi) \\
		      &  =\hspace{-0.1cm}\sum_{X_1, X_2\in \mathcal{P}_k}\mathbb{1}_{\pi(X_1)=U_1}\mathbb{1}_{\pi(X_2)=U_2} 
		       \frac{\tilde{I}^{U_1\setminus  X_1}(\p)}{\tilde{I}^{X_2\setminus U_2}(\p)}
		       \frac{\tilde{I}^{U_2\setminus X_2}(\p)}{\tilde{I}^{X_1\setminus U_1}(\p)}\times\\
		       &\qquad\qquad\hspace{3cm}\times
		\int_{\mathcal{X}_N}\!\!\!\widetilde{K}_k(X_1,\p,\xi)\,\mu_{k+1}(\d \xi)\int_{\mathcal{X}_N}\!\!\!\widetilde{K}_k(X_2,\p,\xi) \,\mu_{k+1}(\d \xi)\\
		    &    =K_{k+1}(U_1,\p)K_{k+1}(U_2,\p)                                                                                          
	\end{split}\end{align}	
\end{proof}

For future reference we summarize a concise definition of $\boldsymbol{T}_k$.
Recall that we defined for $0\leq k\leq N-1$ 
and $H_k\in M_0(\Bcal_k)$ the next scale Hamiltonian $H_{k+1}\in M_0(\Bcal_{k+1})$ by
\begin{align}\label{eq:defofHk+1}
		H_{k+1}(B',\p)=\sum_{B\in \mathcal{B}_k( B')} \widetilde{H}_k(B,\p) 
	\end{align}
	where
	\begin{align}\label{eq:defoftildeHk}
		\widetilde{H}_k(B,\p)=\Pi_2\boldsymbol{R}_{k+1}^{(\boldsymbol{q})}H_k(B,\p)-\Pi_2\boldsymbol{R}_{k+1}^{(\boldsymbol{q})}K_k(B,\p). 
	\end{align}
	For 
	$K_k\in M(\Pck)$ we denote
	$\widetilde{K}_k(X,\p,\xi)=\left(1-e^{-\widetilde{H}_k(\p)}\right)\circ\left(e^{-H_k(\p+\xi)}-1\right)\circ K_k(X,\p+\xi)$ and we define $K_{k+1}\in
	M(\Pckp)$ for $U\in \Pckp$ by 
	\begin{align}\label{eq:defofKk+1}
	\begin{split}
		K_{k+1}(U,\p)&=\sum_{X\in \mathcal{P}_k}\chi(X,U) \exp\Big(-\hspace{-0.cm}\sum_{B\in \mathcal{B}_k(U\setminus   
		X)}\widetilde{H}_k(B,\p)\,+\sum_{B\in \mathcal{B}_k(X\setminus                   
		U)}\widetilde{H}_k(B,\p)\Big)\times
		\\
		&\hspace{3cm}\times\int_{\Xcal_N}\widetilde{K}_k(X,\p,\xi)\;\mu_{k+1}^{(\boldsymbol{q})}(\d\xi)
	\end{split}\end{align}
	where $\chi(X,U)=\mathbb{1}_{\pi(X)=U}$ and $\pi:\mathcal{P}_k\rightarrow \mathcal{P}_{k+1}$ was defined in \eqref{eq:pifactor}.
	Recalling our convention \eqref{convention} to identify functionals in $M(\Pcal_{k+1})$ that factor over strictly disjoint polymers 
	with functionals in $M(\Pckp)$, this definition agrees with the definition of $K_{k+1}$ in \eqref{eq:defofKprime}  because we showed in Lemma~\ref{le:K_kplus_factors}(v) that
	$K_{k+1}$ factors over connected components. 
	Moreover the equation \eqref{eq:defofKk+1} holds for all $U\in \Pcal_{k+1}$ for the extension of $K_{k+1}$ to $M(\Pcal_k)$.
	
\begin{definition}\label{def:Tk}
Let $0\leq k\leq N-1$. The renormalisation transformation 
\begin{align}  
		\boldsymbol{T}_k:M_0(\mathcal{B}_k)\times M(\Pck)\times \mathbb{R}^{(d\times m)\times (d\times m)}_{\mathrm{sym}}\rightarrow 
		M_0(\mathcal{B}_{k+1})\times M(\Pckp)                                                                             
	\end{align}
	is defined by
	\begin{align}   \label{eq:description_rg_trafo}
		\boldsymbol{T}_k(H_k,K_k,\boldsymbol{q})=(H_{k+1},K_{k+1})
	\end{align}
	where $H_{k+1}$ and $K_{k+1}$ are given by 
	\eqref{eq:defofHk+1} and \eqref{eq:defofKk+1} respectively.
\end{definition}
We have already shown the following result for $\boldsymbol{T}_k$.
\begin{proposition}  \label{pr:properties_RG_map}
For $L\geq 2^{d+2}+4R$ and $0\leq k\leq N-1$ the renormalisation transformation $\boldsymbol{T}_k$ is well defined
 and satisfies for $H_k\in M_0(\mathcal{B}_k)$,
$K_k\in M(\Pck)$, $H_{k+1}\in M_0(\mathcal{B}_{k+1})$, and
$K_{k+1}\in M(\Pckp)$ with $\boldsymbol{T}_k(H_k,K_k,\boldsymbol{q})=(H_{k+1},K_{k+1})$
the identity
\begin{align}\label{mainrenormprop_again}
\boldsymbol{R}_{k+1}^{(\boldsymbol{q})}(e^{-H_k}\circ K_k)(\TN,\p)=(e^{-H_{k+1}}\circ K_{k+1})(\TN,\p). 
\end{align}
\end{proposition}

\begin{proof}
Lemma~\ref{le:K_kplus_factors} shows that $K_{k+1}\in M(\Pckp)$. 
Therefore the map $\boldsymbol{T}_k$ is well defined. 
Equation \eqref{mainrenormprop_again}  follows from 
\eqref{eq:renormmap_reblocked}.
\end{proof}
Of course the condition  \eqref{mainrenormprop_again} is not sufficient for our analysis. 
In addition we need smoothness and boundedness results for the map $\boldsymbol{T}_k$. 
This requires that the spaces $M(\Pck)$ are equipped with a norm.
 In the next section we  will define the relevant norms which will allow us to establish the smoothness result
and to prove contraction properties of $\boldsymbol{T}_k$.

\section{Norms}\label{se:norms_new}
 
Next we introduce suitable norms on the space $M(\mathcal{P}_k,\Vcal_N)$ of local functionals (see \eqref{eqstructoffunctionals}). 
For any $F\in M(\Pck,\Vcal_N)$ and any $X\in \Pck$  we define $F(X)\in M(\Vcal_N)$ by $F(X)(\p)=F(X,\p)$. 
Fixing  now $r_0\in \mathbb{N}$ with $r_0\geq 3$, we introduce a norm  $\norm{F(X)}_{k,T_\p}$ 
based on a norm of the $\pn$-th order Taylor polynomial of the functional $F(X)$  at $\varphi$
as well as the norm  $\norm{F(X)}_{k,X}=\sup_\varphi  w^{-X}(\varphi) \norm{F(X)}_{k,T_\p}$, where
$w^{-X}(\varphi) = \frac{1}{w^X(\varphi)}$ and $w^X$ is an appropriately chosen weight function. 
The main difference  in comparison with  \cite{AKM16} (which was based on earlier work of Brydges et al.,  
cf.\ e.g.\  \cite{BS15I} and \cite{Bry09}) is in the choice of these weights. 
The current choice allows  us to relax substantially the growth condition for the potential.
An additional difference with respect to \cite{AKM16} is that we use a different norm on polynomials 
(essentially the projective instead of the injective norm on the dual tensor product, see Section  \ref{se:appendix_sym_norm} in the appendix).  
This is not crucial but it puts our approach in line with the much more general framework developed in \cite{BS15I,BS15II}.

The main observation for the definition of the norms on Taylor polynomials
is that the action of polynomials can be linearised by looking at their action  on (direct sums) of tensor products. 
More precisely a homogeneous polynomial $P^{(r)}$ of degree $r$ on the space  of fields $\BX$ can be uniquely identified with a symmetric
$r$-linear form and hence with an element  $\overline{P^{(r)}}$ in the dual of $\BX^{\otimes r}$ (see Lemma~\ref{le:extension_polynomial}).

To define a linear action of a general polynomial $P$ we recall that  $\oplus_{r=0}^\infty \BX^{\otimes  r}$
is the space of  sequences $g = (g^{(0)}, g^{(1)}, \ldots )$ 
 with $g^{(r)} \in \BX^{\otimes r}$ and with only finitely many non-vanishing terms.
Then we define the dual pairing
\begin{equation}  \label{eq:pairing_P_g}
\langle P, g \rangle = \sum_{r= 0}^\infty \langle \overline{P^{(r)}}, g^{(r)} \rangle
\end{equation}
with  the space of test functions
\begin{equation}
\Phi := \Phi_\pn := \{ g \in  \oplus_{r=0}^\infty \BX^{\otimes  r}  :   g^{(r)} = 0 \ \text{ for all }\ r > \pn\}.
\end{equation}
The restriction to the space $\Phi_\pn$ means that  the linear maps $P$ correspond to  polynomials of order at most $\pn$.

In the following we take $\BX=\Vcal_N$ as the space of fields with norms defined on $\Phi$ as follows. 
On $\Vcal_N^{\otimes 0} = \R$ we take the usual absolute value on $\R$. 
Let 
$$X \in \mathcal{P}_k \ \text{ and }\  j \in \{k, k+1\}.$$
For $\varphi \in \Vcal_N$ and $x \in \Lambda$  we define $\nabla^{i, \alpha}_x \varphi = (\nabla^\alpha \varphi_i)(x)$
and  consider the norms
\begin{align}  \label{eq:primal_norm}
 | \varphi |_{j,X} =& \,  \sup_{x \in X^\ast} \sup_{1 \le i \le m} \,  \sup_{1 \le |\alpha| \le \pphi}  
\wpzc_j(\alpha)^{-1}  \left|\nabla_x^{i, \alpha} \p \right|\\
 = & \,  \sup_{x \in X^\ast} \sup_{1 \le i \le m}   \, \sup_{1 \le |\alpha| \le \pphi}   
  \wpzc_j(\alpha)^{-1}  \left|\nabla^\alpha \p_i(x)\right|.
 \notag
\end{align}
where 
\begin{equation}  \label{eq;definition_pphi}
\pphi = \lfloor d/2 \rfloor + 2
\end{equation}
 and  the weights $ \wpzc_j(\alpha) $ are given by
\begin{equation}  \label{eq:definition_weights}
\wpzc_j(\alpha) =  h_j \,  L^{-j |\alpha|}  \,  L^{-j\frac{d-2}{2}}    \quad  \hbox{with} \quad h_j = 2^j h. 
\end{equation}
The $|\cdot|_{j,X}$-norm for the fields depends on a $k$-polymer $X$ and a scale $j \in \{k, k+1\}$ 
and it measures the size of the field in a weighted maximum-norm in a neighbourhood of this polymer. 
The weights  are chosen so that a typical value of the field $\xi$ distributed according
to $\mu_{j+1}^{(\boldsymbol{q})}$ has norm of order $h_j^{-1}$ (cf. \eqref{eq:discretebounds}).
The parameters $h_j$ allow to control the scaling of the field norms $|\cdot|_{j,X}$
and since norms are defined by duality the  parameter $h_j$ also appears in the norm for Hamiltonians $H\in M_0(\mathcal{B})$.
See Section~\ref{se:free_parameters} for further discussion why we choose scaling factor $h_j$ which grow with $j$.

Viewing homogeneous terms  $g^{(r)} \in \Vcal_N^{\otimes r}$ 
as maps (or more precisely  equivalence classes of maps modulo tensor products involving 
constant fields, see  Section \ref{se:standard_example} in the appendix) from $\Lambda^r$ to  $(\R^p)^{\otimes r}$ with
$\nabla^{\alpha_j}$ acting on the $j$-th argument of $g_{i_1  \ldots i_r}^{(r)}$,
we  introduce the norm 
\begin{align}   \label{eq:tensor_norm}
| g^{(r)} |_{j, X} = & \,  \sup_{x_1,  \ldots, x_r \in { X^\ast}}  \, \sup_{\mpzc \in \mf m_{\pphi,r}} \,  \wpzc_j(\mpzc)^{-1} \,  
\nabla^{\mpzc_1} \otimes \ldots \otimes \nabla^{\mpzc_r} g^{(r)}(x_1, \ldots, x_r)   \\
=   & \, \sup_{x_1,  \ldots, x_r \in { X^\ast}}  \, \sup_{\mpzc \in \mf m_{\pphi,r}} \,  \wpzc_j(\mpzc)^{-1} \,  
\nabla^{\alpha_1} \otimes \ldots \otimes \nabla^{\alpha_r} g_{i_1 \ldots i_r}^{(r)}(x_1, \ldots, x_r) 
\notag
\end{align}
where  $\mf m_{\pphi, r}$ is the set of $r$-tuplets $\mpzc =(\mpzc_1, \ldots, \mpzc_r)$ with
$\mpzc_\ell = (i_\ell,  \alpha_\ell)$ and $1 \le |\alpha_\ell| \le \pphi$ and 
\begin{equation}   \label{eq:definition_weights_r}
\wpzc_j(\mpzc) = \prod_{\ell=1}^r \wpzc_j(\alpha_\ell). 
\end{equation}

The  norm defined above is actually the injective  tensor norm on $(\Vcal_N, | \cdot |_{j,X})^{\otimes r}$, 
see  \eqref{eq:tensor_norm_appendix},
implying, in particular, that
$$ 
|  \varphi^{(1)} \otimes \ldots \otimes \varphi^{(r)} |_{j,X} = | \varphi^{(1)} |_{j,X} \ldots  | \varphi^{(r)} |_{j,X} 
\ \text{ for any }\  \varphi^{(1)},\dots, \varphi^{(r)} \in \BX.
$$
We now define a norm on the space $\Phi$ of test functions by
\begin{equation}   \label{eq: norm_on_Phi}
| g |_{j, X} =  \sup_{r \in \NN_0}  | g^{(r)}|_{j,X}=\sup_{r\leq r_0} |g^{(r)}|_{j,X}.
\end{equation}
and a dual norm on polynomials by
\begin{equation}
| P|_{j,X} :=  \sup \{ \langle P, g \rangle  : g \in \Phi, \, | g |_{j,X} \le 1 \}.
\end{equation}
Assume that $F  \in C^{\pn}(\Vcal_N)$  satisfies the locality condition  with respect to a polymer $X\in\Pc$,
$$
\label{eq:locality_F_norms} 
F(\p+\psi )= F(\p) \ \text{ if }\ \psi|_{X^*}=0.
$$
We define
the pairing
\begin{equation}  \label{eq:pairing_F_g}
\langle F, g \rangle_\p := \langle \textstyle{\tay_\p F}, g \rangle.
\end{equation}
and the norm 
\begin{equation}  \label{eq:def_k_X_Tp_norm}
\abs{F}_{j, X,T_\p} = | \hbox{$\tay_\varphi$}F |_{j,X}= \sup \{ \langle  F, g \rangle_\p : g \in \Phi, \, | g|_{j,X} \le 1\}.
\end{equation}
Here  $\tay_\varphi F$ denotes the Taylor polynomial of order $\pn$ of $F$ at $\varphi$. 

We remark in passing that the right hand side of  \eqref{eq:def_k_X_Tp_norm} 
may be infinite since $|\cdot |_{j, X}$ is only a seminorm, but this will not occur
in the cases we are interested in, namely when $F$ is local and shift invariant 
in the sense described in the paragraph following \eqref{eqstructoffunctionals}. 
More precisely the right hand side of \eqref{eq:def_k_X_Tp_norm} 
is finite if and only if 
$\tay_\p F(\dot \p +\dot  \psi) =\tay_\p F(\dot \p)$ for all $\dot \p \in \Vcal_N$ and all  $\dot \psi \in \Vcal_N$ with $|\dot \psi|_{j,X} = 0$ 
(to see this one uses the fact $\Vcal_N$ is finite dimensional and  the zero norm elements of $\Vcal_N^{\otimes r}$ 
are linear combinations of tensor products $\xi_1 \otimes \ldots \otimes \xi_r$ where at least one of the $\xi_i$ has zero norm).
Note that $|\dot \psi|_{k,X} = 0$ implies that $\dot \psi$ is constant on 
each graph-connected component 	of $X^\ast$ and therefore by the definition
of shift invariance $F(\p+\dot\psi)=F(\p)$ for all $\p\in \Vcal_N$.

The final  norms for the functional $F$  are weighted  sup-norms  over $\p$ of the norm $\abs{F}_{k, X,T_\p}$. 
Dividing the norm $\abs{F}_{k,X,\p}$ by a regulator $w_k(\p)$, we allow the functional to grow for large fields.
A way to think about these regulators is that $\abs{F(\p)}\leq \norm{F(X)}  \,  w_k(\p)$.  
This bound must behave well with respect  to integration against $\mu_{k+1}$ and satisfy certain submultiplicativity properties.
The exact definition of the regulator is slightly involved and will be given in the next section. 

Now, we define  a norm on the  class of functionals  $M(\Pck) = M(\Pck,\Vcal_N)$ defined in \eqref{eqstructoffunctionals}. 
Writing $F(X)(\varphi) = F(X, \varphi)$ for any  $F \in M(\Pck,\Vcal_N)$, we  sometimes  use the abbreviation
\begin{equation} \label{eq:abbreviate_FX_Tphi}
\abs{F(X)}_{k, T_\p} := \abs{F(X)}_{k,X, T_\p}.
\end{equation}
Let $W_k^X, w_k^X, w_{k:k+1}^X\in M(\mathcal{P}_k)$ be weight functions that will be defined in the next section.
Let us denote $W^{-X}_k=(W^X_k)^{-1}$ and similarly for $w$. 
The strong and weak norms, respectively, by
\begin{align}\label{strongnorm}
	\vertiii{F(X)}_{k,X}        & =\sup_{\p} \abs{F(X)}_{k, T_\varphi}     \,       W^{-X}_k(\p),        \\
	\label{weaknorm}
	\norm{F(X)}_{k,X}     & =\sup_{\p} \abs{F(X)}_{k, T_\varphi}    \,    w_k^{-X}(\p),       \\
	\label{middlenorm}
	\norm{F(X)}_{k:k+1,X} & =\sup_{\p}\abs{F(X)}_{k, T_\varphi}    \,      w_{k:k+1}^{-X}(\p). 
\end{align}
The last norm is a version of the weak norm which lies between the weak norms of scales $k$ and $k+1$. 
In fact we will use the strong norm only for functionals in $M(\mathcal{B})$ which already factor over single blocks. 
We write $\vertiii{F}_k=\vertiii{F(B)}_{k,B}$ where the right hand side is independent of $B$ by translation invariance.

Finally, for any $A\ge 1$ we define the global weak norm for $F\in M(\Pck)$  given 
by a weighted maximum of the weak norms over the connected polymers
\begin{align}\label{globalweaknorm}
	\norm{F}_{k}^{(A)}=\sup_{X\in \Pck}\norm{F(X)}_{k,X} A^{|X|_k}
\end{align}
and similarly we define the norms $\norm{\cdot}_{k:k+1}^{(A)}$.
For polymers $X$  that are not connected we will usually  estimate the norm of $F(X, \cdot)$ 
by  the product of the norms of $F(Y_i, \cdot)$ where $Y_1, Y_2, \ldots$ are the connected components of $X$. 
 In Lemma~\ref{le:submult} we will state submultiplicativity properties of the norms needed for these estimates.
With the norm  \eqref{globalweaknorm} we also consider the version where we replace the weak $k$ norm by the in-between $k:k+1$ norm. 
 
We finally introduce another norm on the space of relevant Hamiltonians (at scale $k$).
 Recall that we defined these  to be  functionals of the form
\begin{align} \label{eq:relevant_hamiltonian_exp} 
\begin{split} H_k(B, \p) = L^{dk} a_\emptyset &+ \sum_{x \in B} \sum_{(i, \alpha) \in \mf v_1}    a_{i, \alpha} \nabla^\alpha \p_i(x) \\
&+ \sum_{x \in B} \, \, \sum_{(i, \alpha), (j, \beta) \in \mf v_2} a_{(i, \alpha), (j, \beta)} \nabla^\alpha \p_i(x) \nabla^\beta \p_j(x).
\end{split}
\end{align}
Here $B$ is a $k$-block and the index sets $\mf v_1$ and $\mf v_2$ are given
by
\begin{align}
\mf v_1 := \{ (i, \alpha) : 1 \le i \le m, \,  \alpha \in \NN_0^{\mathcal U}, 1 \le |\alpha| \le \lfloor d/2\rfloor +1 \},
\end{align}
\begin{align}
\mf v_2 := \{ (i, \alpha), (j, \beta)  : 1 \le i,j  \le m, \,  \alpha, \beta \in \NN_0^{\mathcal U}, \,  |\alpha|= |\beta| = 1,  \, (i, \alpha) \le (j, \beta)  \},
\end{align}
where $\mathcal{U}=\{e_1,\ldots, e_d\}$.
The expression  $(i, \alpha) \le (j, \beta)$ refers to any ordering on $\{1, \ldots, m\} \times \{e_1, \ldots, e_d\}$, e.g. lexicographic ordering. 
We use ordered indices to avoid double  counting since $ \nabla^\alpha \p_i(x) \nabla^\beta \p_j(x) =  \nabla^\beta \p_j(x) \nabla^\alpha \p_i(x)$.
We now introduce a norm for relevant Hamiltonians which is expressed directly in terms of the coefficients $a_{\mpzc}$ and given by 
 \begin{align}  \label{hamiltoniannorm}
 \| H_k \|_{k,0} = L^{kd} \, |a_\emptyset| + \sum_{(i, \alpha) \in \mf v_1}
h_k L^{kd} L^{- k \frac{d-2}{2}} L^{-k |\alpha|} \, |a_{i, \alpha}| + \sum_{\mpzc \in \mf v_2}  
h_k^2\,  |a_\mpzc|.
\end{align}
The weights in front of the coefficients are chosen in such a way that the  norm $ \| \cdot \|_{k,0}$ is equivalent
(uniformly in $k$ and $N$) to the strong norm $\vertiii{\cdot}$ (see Lemma~\ref{le:Pi2_bounded} and Lemma~\ref{le:Htp_vs_Hk0} below). 
Intuitively the weight in $L$ can also be understood by recalling that the typical value of $|\nabla^\alpha \p_i(x)|$
under $\mu_{k+1}$ is of order $L^{-k |\alpha|} L^{-k \frac{d-2}{2}}$. 

Note that the norms  depend on the constants $h_k$, $A$ and also on $L$ that will be chosen later.
We will need one additional norm because the renormalisation map $\boldsymbol{R}_{k+1}$ does not preserve factorisation on scale $k$ 
so that we cannot rely on submultiplicativity. 
This norm will only be required in the smoothness result in Chapter~\ref{sec:smoothness} and we postpone the definition of the last norm to that chapter.

\section{Properties of the renormalisation map}

Our definition of the renormalisation transformation $\boldsymbol{T}_k$  in Definition~\ref{def:Tk} satisfies the condition \eqref{eq:mainproprenorm}. 
A second requirement for the map $\boldsymbol{T}_k$ is that it separates relevant and irrelevant contributions properly. 
Observe that the origin $(0,0)$ is a fixed point of the transformation for every $\boldsymbol{q}$. 
The separation of relevant and irrelevant contributions can be made precise by showing that 
the linearisation of $\boldsymbol{T}_k$ at the origin defines a hyperbolic dynamical system.
A close look at the definition of $\boldsymbol{T}_k$ reveals that $H_{k+1}$ is in fact a linear function of $K_{k}$ and $H_{k}$, i.e., we can write
\begin{align}   \label{eq:define_Sk}
	\boldsymbol{T}_k(H_{k},K_{k},\boldsymbol{q})=(\boldsymbol{A}_k^{(\boldsymbol{q})}H_{k} 
	+\boldsymbol{B}_k^{(\boldsymbol{q})}K_{k}, \myS_k(H_{k},K_{k},\boldsymbol{q})) 
\end{align}
where $\boldsymbol{A}_k^{(\boldsymbol{q})}$ and $\boldsymbol{B}_k^{(\boldsymbol{q})}$ are linear operators.
We need two theorems concerning the renormalisation transformation $\boldsymbol{T}_k$. The
first theorem states local smoothness of the map $\myS_k$ which is required to apply an implicit function theorem. Let us use 
$\mathcal{U}_{\rho,\kappa}\subset M_0(\mathcal{B}_k)\times M(\Pck)\times \mathbb{R}^{(d\times m)\times (d\times m)}_{\mathrm{sym}}$
to denote the subset
\begin{align}
\begin{split}
\mathcal{U}_{\rho,\kappa}=\{(H_{k};K_{k},\boldsymbol{q})\in  M_0(\mathcal{B}_k)& \times M(\Pck)\times 
\mathbb{R}^{(d m)\times (d m)}_{\mathrm{sym}}:\, \\
&   \lVert  H_{k}\Vert_{k,0}<\rho,\,  \lVert K_{k}\rVert_{k}^{(A)}<\rho,\, | \boldsymbol{q}|<\kappa\}.      
\end{split}          
\end{align}

\begin{theorem}\label{prop:smoothnessofS}\label{PROP:SMOOTHNESSOFS}
Let $L_0=\max(2^{d+3}+16R,4d(2^d+R))$. 
For every $L\geq L_0$ there are $h_0(L)$, $A_0(L)$, and $\kappa(L)$
such that for  $h\geq h_0(L)$ and $A\geq A_0(L)$ there exists $\rho=\rho(A)$ 
such that the map $\myS_k$ satisfies
\begin{align}
\myS_k\in C^\infty\left(\mathcal{U}_{\rho,\kappa},\, \left( M(\Pckp),\lVert\cdot\rVert_{k+1}^{(A)}\right)\right). 
\end{align}
Moreover there are constants $C=C_{j_1,j_2,j_3}(A,L)$ such that 
\begin{align}\label{eq:propS:Claim}
\lVert D_1^{j_1}D_2^{j_2}D_3^{j_3}\myS_k(H_{k},K_{k},\boldsymbol{q})(\dot{H}^{j_1},\dot{K}^{j_2},\dot{\boldsymbol{q}}^{j_3})\rVert_{k+1}^{(A)} \leq 
C\lVert\dot{H}\rVert_0^{j_1}\left(\lVert\dot{K}\rVert_{k}^{(A)}\right)^{j_2}\lVert\dot{\boldsymbol{q}}\rVert^{j_3}             
\end{align}
for any $(H_{k},K_{k},\boldsymbol{q})\in \mathcal{U}_{\rho,\kappa}$ and any $j_1,j_2,j_3\geq 0$.
\end{theorem}

Note that to avoid overloading by indices that are  evident from the context, 
we skip indicating the scale for variables $\dot{H}$ and $\dot{K}$.

The proof of this theorem can be found in Chapter~\ref{sec:smoothness}.
The second theorem concerns the hyperbolicity of the linearisation of the renormalisation transformation.
Recall that $\eta\in \left(0,\tfrac23\right]$ is a fixed parameter that controls the contraction rate of the renormalisation flow.
\begin{theorem}\label{prop:contractivity}
The first derivative of $\boldsymbol{T}_k$ at $H_k=0$ and $K_k=0$ has the triangular form
\begin{align}\label{eq:triangular_derivative}
D\boldsymbol{T}_k(0,0,\boldsymbol{q})\begin{pmatrix}
\dot{H} \\ \dot{K}
\end{pmatrix}=\begin{pmatrix}
\boldsymbol{A}_k^{(\boldsymbol{q})} & \boldsymbol{B}_k^{(\boldsymbol{q})} \\
0                 & \boldsymbol{C}_k^{(\boldsymbol{q})} 
\end{pmatrix}
\begin{pmatrix}
\dot{H} \\ \dot{K}
\end{pmatrix}
\end{align}
where
\begin{align}\label{eq:defofAk}
(\boldsymbol{A}_k^{(\boldsymbol{q})}\dot{H})(B',\p) & =\hspace{-0.15cm}\sum_{B\in\mathcal{B}_k(B')}\dot{H}(B,\p)+L^{(k+1)d}
\hspace{-0.15cm}\sum_{(i, \alpha), (j, \beta) \in \mf v_2}   \hspace{-0.15cm}
a_{(i, \alpha), (j, \beta)}    \,    (\nabla^\beta)^\ast \nabla^\alpha \mathcal{C}^{(\boldsymbol{q})}_{k+1, ij}(0) \\        
\label{eq:defofBk}
(\boldsymbol{B}_k^{(\boldsymbol{q})}\dot{K})(B',\p) & = -
\sum_{B\in\mathcal{B}_k(B')}\Pi_2\left(\int_{\Xcal_N}\dot{K}(B,\p+\xi)\, \mu_{k+1}^{(\boldsymbol{q})}(\d\xi)\right)          \\
\label{eq:defofCk}
(\boldsymbol{C}_k^{(\boldsymbol{q})}\dot{K})(U,\p)  & =\sum_{B\in \mathcal{B}_k:\, \overline{B}=U}(1-\Pi_2)\int_{\Xcal_N}\dot{K}(B,\p+\xi)\,\mu_{k+1}^{(\boldsymbol{q})}(\d\xi)+  \\
\notag & \hspace{4cm}           
+\sum_{\substack{X\in \Pck\setminus \mathcal{B}(X)\\  \pi(X)=U}}\int_{\mathcal{X}_N}\dot{K}(X,\p+\xi)\,\mu_{k+1}^{(\boldsymbol{q})}(\d\xi).
\end{align}
There exists a constant $L_0$  such that there are constants
$h_0=h_0(L)$, $A_0=A_0(L)$, and $\kappa(L) > 0$ such that for any
$L\geq L_0$, $A\geq A_0(L)$, $h\geq h_0(L)$ and for $| \boldsymbol{q} |<\kappa(L)$ the following bounds hold independent of $k$ and $N$
\begin{align}\label{eq:contractionABC}
\lVert \boldsymbol{C}_k^{(\boldsymbol{q})}\rVert\leq \frac34 \eta,\quad \lVert (\boldsymbol{A}_k^{(\boldsymbol{q})})^{-1}\rVert\leq \frac34,\; 
\text{and}\;\lVert\boldsymbol{B}_k^{(\boldsymbol{q})}\rVert\leq \frac{1}{3}.                                             
\end{align}
Here the norms denote the operator norms  $(M(\Pck),\lVert\cdot\rVert_{k}^{(A)})\to (M(\Pckp),\lVert\cdot\rVert_{k+1}^{(A)})$, \newline
$(M_0(\mathcal{B}_{k+1}),\lVert\cdot\rVert_{k+1,0})\to (M_0(\mathcal{B}_{k}),\lVert\cdot\rVert_{k,0}),  (M(\Pck),\lVert\cdot\rVert_{k}^{(A)})\to
(M_0(\mathcal{B}_{k+1}),,\lVert\cdot\rVert_{k+1,0})$.
Moreover the derivatives of the operators with respect to $\boldsymbol{q}$  are bounded:
\begin{align}\label{eq:qderivABC}
\lVert \partial^\ell_{\boldsymbol{q}} \boldsymbol{A}^{(\boldsymbol{q})}_k\dot{H}\rVert_0\leq C\lVert \dot{H}\rVert_0,\; 
\lVert \partial^\ell_{\boldsymbol{q}} \boldsymbol{B}_k^{(\boldsymbol{q})}\dot{K}\rVert_0\leq C\lVert \dot{K}\rVert,\;    
\lVert \partial^\ell_{\boldsymbol{q}} \boldsymbol{C}_k^{(\boldsymbol{q})}\dot{K}\rVert\leq C\lVert \dot{K}\rVert\;      
\end{align}
for some constant $C=C_\ell(A,L)$. The proof shows that $L_0$ only depends on $d$, $m$, $R_0$, 
and on $\zeta$ and $\omega_0$ through $A_{\Bcal}$ where $A_{\Bcal}$ comes from Theorem~\ref{th:weights_final}.
\end{theorem}
\begin{proof}
Here we only show the validity of  the expressions for the operators $\boldsymbol{A}_k^{(\bf{q})}$, $\boldsymbol{B}_k^{(\boldsymbol{q})}$, 
and $\boldsymbol{C}_k^{(\boldsymbol{q})}$ and the bound \eqref{eq:qderivABC}.
The bounds for the operator norms  will be shown in Chapter~\ref{sec:contraction} in Lemma~\ref{le:Aq}, 
Lemma~\ref{le:bound_BQ} and Lemma~\ref{le:contr}.
The proof of the bounds \eqref{eq:qderivABC} can be found in Corollary \ref{cor:bdAB} 
for the operators $\boldsymbol{A}^{(\boldsymbol{q})}_k$ and
$\boldsymbol{B}_k^{(\boldsymbol{q})}$. For $\boldsymbol{C}^{(\boldsymbol{q})}_k$ it follows from 
Theorem~\ref{prop:smoothnessofS} and the identity
\begin{align}
\partial^\ell_{\boldsymbol{q}} \boldsymbol{C}_k^{(\boldsymbol{q})}= \partial^\ell_{\boldsymbol{q}} \partial_{K_{k}} \boldsymbol{S}_k(0,0,\boldsymbol{q}). 
\end{align}
	
To obtain the formula for $\boldsymbol{A}_k^{(\boldsymbol{q})}$ we recall that by \eqref{eq:defofHk+1} and \eqref{eq:defoftildeHk} 
\begin{align}
(\boldsymbol{A}^{(\boldsymbol{q})}_k \dot{H})(B',\p)=\sum_{B\in \mathcal{B}_k(B')} \Pi_2\myR_{k+1}^{(\boldsymbol{q})}\dot{H}(B,\p). 
\end{align}
We write the Hamiltonian $\dot{H}$ as a sum of constant, linear and quadratic terms,
$\dot H(\p) =L^{dk} a_\emptyset + \ell(\p) + Q(\xi, \xi)$. Then
\begin{align}
\dot{H}(B,\p+\xi)=\dot{H}(B,\p)+Q(\xi, \xi)+\text{terms linear in $\xi$}
\end{align}
where in view of \eqref{eq:relevant_hamiltonian_exp}
\begin{align*}
Q(\xi, \xi) = \sum_{x \in B} \, \, \sum_{(i, \alpha), (j, \beta) \in \mf v_2} a_{(i, \alpha), (j, \beta)} \nabla^\alpha \xi_i(x) \nabla^\beta \xi_j(x).
\end{align*}

Linear terms vanish when integrated against $\mu_{k+1}(\d\xi)$.
Observe that the projection $\Pi_2$  preserves relevant Hamiltonians, i.e., $\Pi_2H_k=H_k$ for $H_{k}\in M_0(\mathcal{B}_k)$.
It remains to evaluate the integral of the quadratic form ${Q}(\xi,\xi)$. 
Since the covariance of $\mu_{k+1}$ is translation invariant we have for $\mathbb{E} = \mathbb{E}_{\mu_{k+1}}$
\begin{align}
\mathbb{E}(\nabla^\alpha \xi^i(x)\nabla^\beta \xi^j(y))=     \mathbb{E}((\nabla^\beta)^\ast\nabla^\alpha \xi^i(x)\xi^j(y)) = 
\big((\nabla^\beta)^\ast\nabla^\alpha \mathcal{C}_{ij}\big) (x-y).            
\end{align}
This implies that
\begin{align}
\int_{\Xcal_N} {Q}(\xi, \xi)\,\mu_{k+1}^{(\boldsymbol{q})}(\d\xi)=\sum_{x\in B}\sum_{(i, \alpha), (j, \beta) \in \mf v_2}   
a_{(i, \alpha), (j, \beta)}    \,    (\nabla^\beta)^\ast \nabla^\alpha \mathcal{C}^{(\boldsymbol{q})}_{k+1, ij}(0).                              
\end{align}
Summing over $B\in \mathcal{B}_k(B')$ we get the formula 
\eqref{eq:defofAk} for $\boldsymbol{A}^{(\boldsymbol{q})}$ using that $|B|=L^{dk}$ 
 and $|\mathcal{B}_k(B')|=L^d$ to obtain the prefactor $L^{(k+1)d}$ of the constant term. 
	
The formula for $\boldsymbol{B}_k^{(\boldsymbol{q})}$ is a direct consequence of the definitions \eqref{eq:defofHk+1} and \eqref{eq:defoftildeHk}.
	
We now derive the formula for $\boldsymbol{C}_k^{(\boldsymbol{q})}$.
Recall that we defined $\widetilde{K}_k(K_{k},H_{k})(\p,\xi)=(1 -e^{-\widetilde{H}_k(\p)}))\circ(e^{-H_k(\p+\xi)}-1)\circ K_k(\p+\xi)$. 
We calculate the derivative at 0 in direction $\dot{K}$, hence we set $H_k=0$ and 
\begin{align}
\widetilde{H}_k(B,\p)=-\Pi_2\boldsymbol{R}_{k+1}^{(\boldsymbol{q})} K_k(B,\p). 
\end{align}
This implies for the derivative of $\widetilde{K}_k$ at zero
\begin{align}
D_K\widetilde{K}_k(0)(\dot{K})(X,\p,\xi)=
\begin{cases}& \dot{K}(X,\p+\xi)-\Pi_2\boldsymbol{R}_{k+1}^{(\boldsymbol{q})}\dot{K}(X,\p)\text{ if $X\in\mathcal{B}_k$,} \\
& \dot{K}(X,\p+\xi)\text{ if $X\in \Pck\setminus \mathcal{B}_k$,}          \\
& 0\text{ if $X\in \mathcal{P}_k\setminus \Pck$}.                           
\end{cases}
\end{align}
The derivative vanishes for non-connected polymers because $K_{k}$ factors on scale $k$.
Now the definition \eqref{eq:defofKk+1} implies \eqref{eq:defofCk}.
	
Finally we show that the  derivative of $K_{k+1}$ with respect to $H_k$ vanishes. To this end we notice that
\begin{align}
D_H\widetilde{K}_k(0)(\dot{H})(X,\p,\xi)=                                             
\begin{cases}                                                                   
\dot{H}(X,\p+\xi)- \Pi_2\boldsymbol{R}_{k+1}^{(\boldsymbol{q})}\dot{H}(X,\p)\; \text{ for $X\in \mathcal{B}_k$}, \\
0 \;\text{ otherwise.}                                                          
\end{cases}                                                                     
\end{align}
Thus \eqref{eq:defofKk+1} implies that the derivative vanishes for $U\notin\mathcal{B}_{k+1}$ 
and we infer that for $B'\in \mathcal{B}_{k+1}$ 
\begin{align}
\begin{split}
D_HK_{k+1}(\dot{H})(B',\p)&=
\sum_{B\in \mathcal{B}_{k}(B')} \int_{\mathcal{X}_N} \dot{H}(B,\p+\xi)-
(\Pi_2\boldsymbol{R}_{k+1}^{(\boldsymbol{q})}\dot{H})(B,\p)\,\mu_{k+1}^{(\boldsymbol{q})} (\d \xi)\\
&=\sum_{B\in \mathcal{B}_{k}(B')} (\boldsymbol{R}_{k+1}^{(\boldsymbol{q})}\dot{H})(B,\p)-(\Pi_2\boldsymbol{R}_{k+1}^{(\boldsymbol{q})}\dot{H})(B,\p)=0
\end{split}	
\end{align}
where we used that $\boldsymbol{R}_{k+1}^{(\boldsymbol{q})}$ maps relevant Hamiltonians to relevant Hamiltonians as shown above and $\Pi_2$ is the identity on relevant Hamiltonians.
\end{proof}

\chapter{A New Large Field Regulator}\label{sec:weights}

In this chapter we construct a new large field regulator. It allows for
substantially rougher initial perturbations than the previous regulator in \cite{AKM16} or \cite{Bry09}.
Previously explicit estimates for carefully chosen Gaussian integrals were used to
construct the regulators. In the new approach we define the weights implicitly based
on the abstract formula for Gaussian integrals.
The next section contains the definition of the new large field regulator. Then, in Section~\ref{sec:LFR_comments}, we motivate the construction and relate it to earlier
work. Finally, in the remaining sections of this chapter, we prove the properties of the weights
stated in Theorem~\ref{th:weights_final} below.

\section{Introducing the weights}
\label{sec:NLRdef}

Recall that we defined the constant
\begin{align}\label{eq:defofMConstant}
M=M(d)=\pphi+\lfloor d/2\rfloor+1=2 \lfloor d/2\rfloor +3
\end{align}
 that is related to the discrete Sobolev embedding (note that compared to
\cite{AKM16} we changed $M$).
For $0\leq k \leq N$ and any $k$-polymer $X$  we define the linear operator $\boldsymbol{M}_k^X:\mathcal{X}_N\to \mathcal{X}_N$ by
\begin{align}\label{eq:defofMk}
	\boldsymbol M_k^X=\sum_{1\leq |\alpha|\leq M} L^{2k(|\alpha|-1)} (\nabla^\ast)^\alpha \chi_X \nabla^\alpha 
\end{align}
where $\chi_X:T_N\rightarrow \mathbb{R}$ is defined by 
\begin{align}
\chi_X(x)=\sum_{B\in \mathcal{B}_{k}(X)}
\boldsymbol{1}_{B^+}(x)
= \abs*{\{B\in \Bcal_k(X):  x\in B^+\}}.
\end{align}
Here $\boldsymbol{1}$ denotes the indicator function.
Recall (see Section~\ref{se:polymers}) that $B^+=(B+[-L^k,L^k]^d)\cap T_N$ for $k\geq 1$ and $B^+=(B+[-R,R]^d)\cap T_N$ for $k=0$.
Note that here and in the following we sometimes use the 
natural inclusion $\mathbb{R}\hookrightarrow \mathbb{R}^{m\times m}$ given by
$\lambda\to \lambda\, \mathrm{Id}$ 
without reflecting this in the notation.
Let us also introduce the operator
\begin{align}
\Malt_k=\sum_{1\leq |\alpha|\leq M} L^{2k(|\alpha|-1)} (\nabla^\ast)^\alpha\nabla^\alpha. 
\end{align}
The operators $\boldsymbol{M}_k^{\TN}$ and $\boldsymbol{M}_k$ are related by
\begin{align}\label{eq:MLambdaN}
\boldsymbol{M}_k^{\TN}=\Xi_k \sum_{1\leq |\alpha|\leq M} L^{2k(|\alpha|-1)} (\nabla^\ast)^\alpha\nabla^\alpha =\Xi_k \Malt_k
\end{align}
where $\Xi_k=|B^+|_k$, $B\in \mathcal{B}_k$ accounts for the sum over $\boldsymbol{1}_{B^+}$.
From the definition of $B^+$ we find
 $\Xi_0=(2R+1)^d$, $\Xi_N=1$, and $\Xi_k=3^d$ for $1\leq k<N$
 and therefore in particular
 \begin{align}\label{eq:thetamax}
 \Xi_k\leq \Xi_{\mathrm{max}}=(2R+1)^d. 
 \end{align}
 Note that
$\Malt_k$ is translation invariant and therefore diagonal in Fourier space.
 
 Recall that we consider the space $\Gcal=(\R^m)^\Ical$ where $\Ical$ satisfies  $\{e_1,\ldots,e_n\}\subset \Ical\subset \subset\{\alpha\in \mathbb{N}_0^d\setminus \{0,\ldots,0\}: |\alpha|_\infty\leq R_0\}$.
We assume that $\Qscr$ is a quadratic form on $\mathcal{G}$ 
that satisfies \eqref{eq:conditionForQ}.
From now on we use the shorthand notation 
 $\mathscr{A}= \mathscr{A}_{\boldsymbol{Q}}= \mathscr{A}^{(0)}$ for the operator generated by $\Qscr $ on $\mathcal{X}_N$ 
 (cf. \eqref{E:A_Q} and \eqref{eq:definition_Aq}),
\begin{align}\label{eq:A_again}
(\p,\mathscr{A}\p)=\sum_{x\in T_N} \Qscr(D\p(x)).
\end{align}

 Let $\weightzeta\in (0,\tfrac14)$ be a parameter. We will later set
\begin{align}\label{eq:definition_weightzeta}
\weightzeta=\zeta/4
\end{align}
 where
 $\zeta\in (0,1/2)$ is the parameter in the norm on $\boldsymbol{E}$
 that appears in Theorem \ref{th:pertcomp}.
Let $\delta_j=4^{-j}\delta>0$ be a sequence of real numbers
with $\delta$ to be specified later.
We define large field regulators $w_k^X,w_{k:k+1}^X$ for the weak norm for $0\leq k \leq N$ by
\begin{align}\label{eq:defofwk}
w_{k}^X(\p)=e^{\frac{1}{2}(\boldsymbol A_k^X\p,\p)},\quad     
w_{k:k+1}^X(\p)=e^{\frac{1}{2}(\boldsymbol A_{k:k+1}^X\p,\p)} 
\end{align}
where $\boldsymbol A_k^X$ and $\boldsymbol A_{k:k+1}^X$ are
linear symmetric operators on $\mathcal{X}_N$ that are defined iteratively by
\begin{equation}
\begin{aligned}\label{eq:defAk}
(\p,\boldsymbol A_0^X \p) &=(1-4\weightzeta)\sum_{x\in X}\Qscr(D\p(x))+\delta_0(\p,\boldsymbol M_0^X\p) \quad \text{for $X\in \mathcal{P}_0$}, \\
\boldsymbol A_{k:k+1}^X &=\left((\boldsymbol A_k^X)^{-1}-(1+\weightzeta)\mathscr C_{k+1}\right)^{-1} \quad \text{for $X\in \mathcal{P}_k$ and $0\leq k\leq N$},\\
\boldsymbol A_{k+1}^X &=\boldsymbol A_{k:k+1}^{X^{\ast}}+\delta_{k+1}\boldsymbol M_{k+1}^X \quad \text{for $X\in \mathcal{P}_{k+1}$ and $0\leq k\leq N-1$}.                                                                                           
\end{aligned}
\end{equation}
Here $\mathscr C_{k+1}$ is a finite range decomposition for the operator $\mathscr{A} = \mathscr{A}^{(0)}$ 
as in Theorem \ref{thm:frd}.
The definition of $\boldsymbol{A}_{k:k+1}^X$ is a bit sloppy because $\boldsymbol{A}_k^X$ is in general
not invertible, however the definition makes sense on the space $\ker (\boldsymbol{A}_k^X)^\perp$ and then $\boldsymbol{A}_{k:k+1}^X$ is
the extension by zero of this operator; see the beginning of Section~\ref{sec:proweights} and Lemma~\ref{prop:W1}~\ref{it:welldefined} below.
We use the neighbourhood $X^{\ast}$ in the definition of $\boldsymbol{A}_{k+1}^X$ to
account for the fact that in the reblocking step we also add contributions to $X$ that come from polymers that are not contained in $X$ but only in $X^\ast$.

We define the strong norm weight functions almost as in \cite{AKM16} by
\begin{align}  \label{eq:strong_weight}
	W_k^X(\p)=e^{\frac12(\boldsymbol z\p,\p)}\text{  with  } (\p,\boldsymbol{G}_k^X\p)=\frac{1}{h_k^2} \sum_{1\leq |\alpha|\leq \lfloor \frac{d}{2}\rfloor +1}L^{2k(|\alpha|-1)}(\nabla^\alpha 
	\p,\boldsymbol{1}_X                                                                                                                       
	\nabla^\alpha\p)                                                                                                                      
\end{align}
where as before $h_k=2^kh$   with $h=h(L)$ to be chosen later.

\section{Motivation and relation to earlier definitions}\label{sec:LFR_comments}
To motivate the definition of the weight functions, we add several observations.
In the evaluation of functional integrals $\int F(\p)\,\mu(\d\p)$ where $\mu$ is a Gaussian measure 
it is a well known problem that the functional $F$ is in general unbounded for large fields $\p$. 
This is the \textit{large field problem} that makes the construction of good norms for $F$ difficult. 
A more detailed discussion can be found in \cite{Bry09}.
In our approach  we defined the norms for $F$ in \eqref{weaknorm}  by  
$\norm*{F}_{k,X}=\sup_\p \abs{F(\p)}_{k,X,T_\p} (w_k^{X}(\p))^{-1}$ where $w_k^X$ are the weight functions. 
They regulate the allowed growth at infinity. The larger the weight function the weaker the norm. 
So the results get stronger, i.e., the class of admissible potentials is bigger, if we can choose $w_k$ bigger.
The growth assumptions for the potential $V$  in our theorems are weaker than those in \cite{AKM16} due to 
the larger weights that we construct in this chapter.

The weights need to satisfy two main requirements that ensure that the renormalisation map is bounded and that the norms are submultiplicative.
In the following we give a rough overview of those two requirements and why our new construction satisfies them. 
We first discuss the boundedness of the 
renormalisation map that constrains the growth of the weights.

The first key requirement for the norm is that the renormalisation map, i.e., convolution with the partial
measures $\mu_{k+1}$ is bounded. This yields the condition
\begin{align}
\begin{split}
\norm*{\int_{\Xcal_N} F(\p+\cdot)\, \mu_{k+1}(\d\p)}& =    
\sup_{\psi} w_{k+1}^{-X}(\psi) \abs*{\int_{\Xcal_N} F(\p+\psi)\, \mu_{k+1}(\d\p)}\\
& \leq \norm*{F} \sup_{\psi} w_{k+1}^{-X}(\psi) \int_{\Xcal_N} w_k^X(\p+\psi)\, \mu_{k+1}(\d\p).
\end{split}
\end{align}
In other words the renormalisation map is bounded if and only if
$w_k^X\ast \mu_{k+1}\lesssim  w_{k+1}^X$.
Therefore the optimal choice is $w_{k+1}^X\propto w_k^X\ast \mu_{k+1}$. In general this is a very 
implicit
definition that is not very useful. If, however, $w_k(\p)=e^{\frac12 (\p,\boldsymbol{A}_k^X\p)}$ is an  exponential of a quadratic
form, the convolution can be carried out explicitly and then the next weight has the same structure, i.e., it is again the exponential of 
a quadratic form.  Indeed, by general Gaussian calculus
the following identity holds for a given linear symmetric positive operator
$A$ on a finite dimensional vector space $V$ and a covariance operator $C$ 
\begin{align}\begin{split}\label{eq:GaussianCalculus}
	\int_{V} e^{\frac12 \left(A(\p+\psi),\p+\psi\right)}\,\mu_C(\d\psi)
	  & =\left(\frac{\det(C^{-1}-A)}{\det C^{-1}}\right)^{-\frac12}e^{\frac12( 
	(A^{-1}-C)^{-1}\p,\p)}
	\\
	  & =\det\left(\1-C^{\frac12}AC^{\frac12}\right)^{-\frac{1}2}e^{\frac12(    
	(A^{-1}-C)^{-1}\p,\p)}
	\end{split}
\end{align}
under the assumption that $A<C^{-1}$.
This implies that the next scale quadratic form is essentially given by the expression 
for $\boldsymbol{A}_{k:k+1}^X$ in \eqref{eq:defAk}.

We now briefly discuss the second key requirement for the norms of the functionals which is sub-multiplicativity for 
distant polymers, i.e., $\lVert F^XF^Y\rVert\leq \lVert F^X\rVert \cdot\lVert F^Y\rVert$ if $X$ and $Y$ are strictly
disjoint polymers. This condition is necessary to regroup the terms and estimate products.
Since the maximum norm is sub-multiplicative we find the condition 
$w_k^Xw_k^Y\geq w_k^{X\cup Y}$ for the weights. 
At first sight this might appear problematic because we have no explicit
expression for $w_k^X$. But it turns out that the finite range property of $\mu_{k+1}$ ensures that the
weight functions factor for strictly disjoint polymers if we choose
$w_{k+1}^X\propto w_k^X\ast \mu_{k+1}$.
To show this we note that $w_0^X(\p)$ only depends  on the values of $\p$ in a neighbourhood of $X$. 
The same is true for $w_k^X(\p)$ because it is a convolution of $w_0$ with some measure.
Then the factorisation follows by induction from the finite range property
\begin{align}\label{eq:factheu}
	w_{k+1}^{X\cup Y}=w_k^{X\cup Y}\hspace{-0.1cm}\ast \mu_{k+1}=                              
	(w_k^X\cdot w_k^Y)\ast \mu_{k+1}=(w_k^X\hspace{-0.1cm}\ast \mu_{k+1})( w_k^Y\hspace{-0.1cm}\ast\mu_{k+1}) 
	=w_{k+1}^Xw_{k+1}^Y.                                                        
\end{align}
Finally, let us  briefly mention why we need the second set of weights $w_{k:k+1}$ that includes the operator $\boldsymbol{M}_k^X$. 
The reason is twofold. On the one hand, in every step we also need to control contribution from 
the Hamiltonian terms on blocks
that are bounded in the strong norm but the
blocks are not separated from the considered polymer. Therefore sub-multiplicativity does not hold in this
case. Instead we add the operator $\boldsymbol{M}_k^X$ that allows us to bound the terms from the Hamiltonian. 
Secondly,
the field norm $|\p|_{k,X}$ must be controlled by the weight function $w_k^X$.
This is also guaranteed by
the addition of the term $\boldsymbol{M}_k^X$. It turns out, however, that this changes the weight functions only slightly for sufficiently
small prefactor $\delta$ (see Lemma~ \ref{le:opineq} below).

\section{Properties of the weight functions}
\label{sec:proweights}

Here and in the following we consider the extensions of the quadratic 
forms $\boldsymbol{G}_k^X$, $\boldsymbol{M}_k^X$, $\boldsymbol{A}_{k}^X$, and $\boldsymbol{A}_{k:k+1}^X$ 
to $\Vcal_N$ by
$\boldsymbol{G}_k^X\p=0, $ for $\p\in \Xcal_N^\perp=\{\text{constant fields}\}$ and similarly for the other forms.
Then we can also extend the weight functions $w_k^X$, $w_{k:k+1}^X$, and $W_k^X$ to $\Vcal_n$ using their definition \eqref{eq:defofwk} and \eqref{eq:strong_weight}. This extension has the property that $w_k^X(\p+\psi)=w_k(\p)$ if $\psi$ is a constant field. 

In the following theorem we collect the properties of the weight functions $w_k^X$, $w_{k:k+1}^X$, and $W_k^X$.
The claims of the theorem will be  reformulated and proven
directly in terms of the operators $\boldsymbol{A}_k^X$, $\boldsymbol{A}_{k:k+1}^X$, and $\boldsymbol{G}_k^X$ in the following sections.
We state our results for general values of $\pphi$, $M$, $n$,  and $\tilde n$
but we will later only use the weights for the parameters chosen as indicated before.
 Recall our convention that we do not indicate dependence on the fixed parameters $\omega_0$, 
 $\weightzeta$, $d$, $m$, $R_0$, $M$, $n$, and $\tilde{n}$.
\begin{theorem}\label{th:weights_final}
Consider $\Gcal$ as defined in \eqref{eq:def_Gcal} and let $\Qscr$ be a quadratic form
on $\Gcal$  satisfying 
\begin{align}
\omega_0 |z^\nabla|^2 \le \Qscr(z) \le \omega_0^{-1} |z|^2
\end{align}
 with a  constant $\omega_0\in (0,1)$ and let $\weightzeta\in (0,\frac14)$.
Let $M\geq\pphi+\lfloor \frac{d}{2}\rfloor +1$ and let  $\mathscr{C}_k^{(\boldsymbol{q})}$ be
a family of finite range decompositions for the quadratic forms  
$z \mapsto \Qscr(z) - (\boldsymbol q z^\nabla, z^\nabla)$,  with $n\geq 2M$ and $\tilde{n} > n$.
Then, for every 
\begin{align}\label{eq:condition_L_weights}
L\geq 2^{d+3}+16R,
\end{align} 
there are constants $\lambda >0$, $\delta(L)>0$,  $\kappa(L)$  (specified in \eqref{eq:defoflambda}, \eqref{eq:defofdelta},  
and \eqref{eq:defofkappa}) and $h_0(L)$ given by 
\begin{align}\label{eq:definition_h0}
h_0(L)=\delta(L)^{-\frac12} \max(8^{\frac12},c_d)
\end{align}
such that	the weight functions defined in \eqref{eq:defofwk}
and \eqref{eq:strong_weight}
are well-defined and satisfy:

\begin{enumerate}[label=\roman*)]
\item  
\label{w:w1}
For any $Y\subset X\in \mathcal{P}_k$, $0\leq k\leq N$, and $\p\in \mathcal{V}_N$, 
\begin{align}\label{eq:w1}
w_k^Y(\p)\leq w_k^X(\p)\quad \text{and} \quad w_{k:k+1}^Y(\p)\leq w_{k:k+1}^X(\p);
\end{align}
\item
\label{w:w2}
 The estimate 
\begin{align}\label{eq:w2}
w_k^X(\p)\leq \exp\left(\frac{(\p,\Malt_k\p)}{2\lambda}\right)
\quad \text{and} \quad w_{k:k+1}^X(\p)\leq \exp\left(\frac{(\p,\Malt_k\p)}{2\lambda}\right)
\end{align}
holds for $0\leq k\leq N$, $X\in \Pcal_k$, and $\p\in \Vcal_N$;
\item
\label{w:w3}
 For any strictly disjoint polymers $X,Y\in \mathcal{P}_k$, $0\leq k\leq N$,  and $\p\in \mathcal{V}_N$,
\begin{align}\label{eq:w3}
w_k^{X\cup Y}(\p)=w_k^X(\p)w_k^Y(\p);
\end{align}
\item
\label{w:w4}
 For any polymers $X,Y\in \mathcal{P}_k$  such that $\mathrm{dist}(X,Y)\geq \frac34 L^{k+1}$, $0\leq k\leq N$, and $\p\in \mathcal{V}_N$,
\begin{align}\label{eq:w4}
w_{k:k+1}^{X\cup Y}(\p)=w_{k:k+1}^X(\p)w_{k:k+1}^Y(\p);
\end{align}
\item 
\label{w:w4b}
For  any disjoint polymers $X,Y\in \mathcal{P}_k$, $0\leq k\leq N$, and 
$\p\in \mathcal{V}_N$,
\begin{align}\label{eq:w4b}
W_{k}^{X\cup Y}(\p)=W_{k}^X(\p)W_{k}^Y(\p);
\end{align} 
\item 
\label{w:w5}
For   $h \ge h_0(L)$,  disjoint polymers $X, Y\in \mathcal{P}_k$, $0\leq k\leq N$, and $\p\in \mathcal{V}_N$,
\begin{align}\label{eq:w5}
w_k^{X\cup Y}(\p)\geq w_k^X(\p)W_k^Y(\p);
\end{align}
\item 
\label{w:w6}
For  $h \ge h_0(L)$, $X\in \mathcal{P}_k$ and $U=\pi(X)\in \mathcal{P}_{k+1}$, $0\leq k\leq N-1$, and $\p\in \mathcal{V}_N$,
\begin{align}\label{eq:w6}
w_{k+1}^U(\p)\geq w_{k:k+1}^X(\p)\left(W_k^{U^+}(\p)\right)^2;
\end{align}
\item 
\label{w:w9}
For any $h \ge h_0(L)$ and  all $0\leq k\leq N-1$, $X\in \mathcal{P}_{k+1}$ and $\p\in \mathcal{V}_N$,
\begin{align}\label{eq:w9}
e^{\frac{|\p|_{k+1,X}^2}{2}}w_{k:k+1}^X(\p)\leq w_{k+1}^X(\p);
\end{align}
\item 
\label{w:w7}
Let $\rho = (1+\weightzeta)^{1/3}-1$.
There is a constant $A_\mathcal{P}=A_\mathcal{P}( L)$
such that for $\boldsymbol{q}\in B_\kappa$, $\overline{\rho}\in [0, \rho]$, $X\in \mathcal{P}_k$, $0\leq k\leq N$, and $\p\in \mathcal{V}_N$,
\begin{align}\label{eq:w7}
\left(\int_{\mathcal{X}_N} \left(w_k^X(\p+\xi)\right)^{1+\overline{\rho}}\,\mu_{k+1}^{(\boldsymbol{q})}(\d \xi)\right)^{\frac{1}{1+\overline{\rho}}}\leq \left(\frac{A_{\mathcal{P}}}{2}\right)^{|X|_k} w_{k:k+1}^X(\p);
\end{align}
\item 
\label{w:w8}
There is a constant $A_\mathcal{B}$ independent of $L$
such that for $\boldsymbol{q}\in B_\kappa$, $\overline{\rho}= [0, \rho]$ (with $\rho=(1+\weightzeta)^{1/3}-1$ as before), 
$X\in \mathcal{P}_k$, $0\leq k\leq N$, and $\p\in \mathcal{V}_N$,
\begin{align}\label{eq:w8}
\left(\int_{\mathcal{X}_N} \left(w_k^B(\p+\xi)\right)^{1+\overline{\rho}}\,\mu_{k+1}^{(\boldsymbol{q})}(\d \xi)\right)^{\frac{1}{1+\overline{\rho}}}\leq 
\frac{A_{\mathcal{B}}}{2} w_{k:k+1}^B(\p).
\end{align}
\end{enumerate}	
\end{theorem}

For the contraction estimate in Theorem~\ref{prop:contractivity} it is crucial that the constant $\AB$
in the integration estimate \eqref{eq:w8} for a single block  does not depend on $L$.
To show this we  use that the smoothness estimates in the finite range decomposition have the optimal 
dependence on $L$, see \eqref{eq:discreteboundsfinal} and
\eqref{eq:trace_bound_single_block} below. The optimal $L$ dependence is an important improvement of the finite range decomposition 
in \cite{Buc18}    over the one in \cite{AKM13}. This improvement is related to the fact the the decomposition in  \cite{Buc18}  is based on 
Bauerschmidt's decomposition \cite{Bau13},   rather than on \cite{BT06}.

\begin{proof}
The theorem  follows from a sequence of lemmas in the following sections.
Lemma~\ref{prop:W1} establishes basic properties of 
the operators $\boldsymbol{A}_k^X$ and $\boldsymbol{A}_{k:k+1}^X$ that imply
\ref{w:w1} and \ref{w:w2}.
Lemma~\ref{prop:W2} concerns factorisation properties of the 
operators $\boldsymbol{A}_k^X$ and $\boldsymbol{A}_{k:k+1}^X$ that allow us to conclude
\ref{w:w3}-\ref{w:w6}. 
Lemma~\ref{prop:W3} gives a bound on a particular  determinant that implies \ref{w:w7} and \ref{w:w8}.
Finally,  Lemma~ \ref{prop:W4} bounds the field norm $|\cdot|_{k,X}$ in terms of the weights. This easily yields property \ref{w:w9}.
\end{proof}

\section{The main technical matrix estimate}

In this section we prove a crucial technical estimate which shows that the 
iterative procedure   \eqref{eq:defAk} introducing the operators  $A_{k} \rightarrow A_{k:k+1} \rightarrow A_{k+1}$
is well-defined.

We first recall some standard facts about monotone matrix  functions. 
We say that two Hermitian matrices $A$ and $B$ satisfy $A \le B$ if $(Ax,x) \le (Bx, x)$ for all $x$. 
We say that a map $f$ from a subset  $U$ of  the Hermitian  matrices to the Hermitian  matrices is 
matrix monotone
if $A \le B$ implies $f(A) \le f(B)$ for all $A,B \in U$. 
\begin{lemma} \label{le:matrix_monotone}\hfill

\begin{enumerate}[label=(\roman*),leftmargin=0.7cm]
\item The map $A \mapsto -A^{-1}$ is matrix-monotone on the set of positive definite Hermitian matrices.
\item Let $C$ be Hermitian and positive definite. For positive definite Hermitian  matrices $A$ with $A < C^{-1}$
define
\begin{equation} \label{eq:algebraic_integration}
f(A) := (A^{-1} - C)^{-1}.
\end{equation}
Then $f$ is matrix monotone.
\item\label{it:matrix_monotone3} If we extend $f$ to Hermitian  matrices $A$ with $0 \le A < C^{-1}$  by
\begin{equation}  \label{eq:algebraic_integration2}
f(A) =  \begin{cases} ((A_{\ker A^\perp})^{-1} - (P_{\ker A^\perp} C P_{\ker A^\perp}))^{-1}  & \hbox{on $\ker A^\perp$,}\\
0 & \hbox{on $\ker A$.}
\end{cases}
\end{equation}
then the extended function is still matrix monotone. 
\item\label{it:matrix_monotone4} If $0 \le A < C^{-1}$ then $A^{1/2} C A^{1/2} <  \1$ and the extended function $f$ satisfies
\begin{equation}  \label{eq:algebraic_integration3}
f(A) = A^{1/2} (\1 - A^{1/2} C A^{1/2})^{-1} A^{1/2}.
\end{equation}
 {
There is the following absolutely convergent series representation for $f$ and $0\le A<C^{-1}$
\begin{align}\label{eq:algebraic_integration4}
f(A)=\sum_{i=0}^\infty A(CA)^i.
\end{align}
} 
\end{enumerate}
\end{lemma}

\begin{proof} The assertions are classical. We include a proof for the convenience of the reader.

The first assertion follows from L\"owner's theorem (\cite{Low34}, see also \cite{Han13}) since
the imaginary part of 
the map $\C \setminus \{0\} \ni z \mapsto - z^{-1} = \frac{-\bar z}{|z|^2}$ 
is non-negative in the upper half-plane.
Alternatively,  it can be proved elementary as follows. 
First, monotonicity is clear for $B = \1$ since for a positive definite symmetric matrix $A$ 
the condition $A \le \1$ is equivalent to $\mathrm{spec}(A) \subset (0,1]$ while the condition 
$A \ge 1$ is equivalent to $\mathrm{spec}(A) \subset [1, \infty)$. 
To prove the result for a general $B$ assume $A \le B$ and note that this implies
 $\overline{F}^T A F \le \overline{F}^T B F$ for all matrices $F$. 
 Taking $F = B^{-1/2}$ we get $B^{-1/2} A B^{-1/2} \le \1$ and thus
 $B^{1/2} A^{-1} B^{1/2} \ge \1$ which implies that $A^{-1} \ge B^{-1/2} \1 B^{-1/2} = B^{-1}$. 
 
 The second assertion follows by applying the monotonicity of the inversion map twice.
 
 The third assertion follows since the right hand side is the limit $\lim_{\varepsilon \downarrow 0} f(A + \varepsilon \1)$.
 
 The fourth assertion is clear for $0 < A < C^{-1}$.  Fix $A$ with $0 \le A < C^{-1}$. 
 Then there exist  $\delta > 0$ and $\varepsilon_0 > 0$
 such that for all $\varepsilon \in (0, \varepsilon_0)$ we have $A_\varepsilon := A + \varepsilon \1 \le (1- \delta) C^{-1}$. Thus
 $ C \le (1-\delta) A_\varepsilon^{-1}$ and hence $A_\varepsilon^{1/2} C A_\varepsilon^{1/2} \le (1-\delta) \1$.  
 Passing to the limit $\varepsilon \downarrow 0$ we get $A^{1/2} C A^{1/2} \le (1 - \delta) \1$ and the validity
 of   \eqref{eq:algebraic_integration3}.
  {
  Equation \eqref{eq:algebraic_integration4} follows
  from \eqref{eq:algebraic_integration3} by expanding the the Neumann series. 
} 
\end{proof}

We now  show the crucial technical lemma that allows us to find suitable bounds for the
operators $\boldsymbol{A}_k$ . Basically this lemma shows that for sufficiently small $\delta$ the $\boldsymbol{M}_k^X$ terms
are just a small perturbation of the operators.
\begin{lemma}\label{le:opineq}
Suppose that $\Qscr$, $M$, $n$, $\tilde{n}$, and $\mathscr{C}_k = \mathscr{C}_k^{(0)}$ satisfy the assumptions of Theorem \ref{th:weights_final}.
Then the following holds.	
For all  $\lambda\in (0,1/4)$ and $L\geq 3$ odd, there is a constant $\mu(\lambda, L) \geq 1$ such that 
for any $\varepsilon\in (0,1)$, $0\leq \delta<\frac{1+\varepsilon}{\mu}$ and for all $0\leq k \leq N-1$, the bound
\begin{align}\label{eq:basicweightineq}
\left(\lambda \Malt_{k}^{-1}+(1+\varepsilon)\sum_{j=k+2}^{N+1}            
\mathscr C_j\right)^{-1}+\delta \Malt_{k+1}\leq                                     
\left(\lambda \Malt_{k+1}^{-1}+(1+\varepsilon-\mu\delta)\sum_{j=k+2}^{N+1}\mathscr C_j\right)^{-1}
\end{align}
 holds in the sense of Hermitian  operators on $\mathcal{X}_N$.
\end{lemma}
\begin{remark}
The proof is quite technical and not very insightful. 
Therefore we first give a brief heuristic argument.  
All operators are diagonal in the Fourier space. 
Thus it is sufficient to show the bound for all Fourier modes 
$p\in \widehat{T}_N\setminus \{0\}$ of the kernels of the operators that actually are $m\times m$ matrices. 
Note that $\Malt_k$ acts diagonally with respect to the $m$ components and thus its Fourier modes are multiples of the identity.
We use  $\widehat{\Malt_k}(p)\in \R$ to  denote the coefficient of the Fourier mode and use
the embedding into $\R^{m\times m}$ when necessary.
Let $q(p)_j=e^{ip_j}-1$ and note that the Fourier multiplier of $\nabla$ is the vector $q(p)$ whose norm is of the order $ |p|$:
$|p|/2\leq|q(p)|\leq |p|$ for $p\in \widehat{T}_N$ (cf. \eqref{eq:def_qp}). 
Therefore we  can write 
\begin{align}\label{eq:Fourier_of_M_k}
\widehat{\Malt_k}(p)=\sum_{1\leq |\alpha|\leq M} L^{2k(|\alpha|-1)}|q(p)^{2\alpha}|
\end{align}
To shorten the notation we introduce the notation 
\begin{align}
\mathscr{C}_{k+1}^{N+1}=\sum_{j=k+1}^{N+1} \mathscr C_j.
\end{align}
	
There are two regimes $|p|\leq L^{-k}$ and $|p|\geq L^{-k}$ requiring different treatment.
	
Using \eqref{eq:Fourier_of_M_k},  for $|p|<L^{-k}$ we find that 
$\widehat{\Malt}_k(p)\approx |p|^2$.
Indeed, since, roughly speaking, $\widehat{\mathcal{C}}_j(p)\approx |p|^{-2}$ for $|p|\approx L^{-j}$,
we observe that  $|p|^{-2}\approx \widehat{{\mathcal C}}_{k+1}^{N+1}(p)$ for $|p|\leq L^{-k}$.
Hence, after factoring out the term $|p|^2$, we are left to show an inequality
of the type $\alpha^{-1}+\delta\leq (\alpha-\mu\delta)^{-1}$ for given real numbers $\alpha$ and $\delta$. 
This is true  for some large $\mu$ if $\alpha$ is uniformly bounded above and below and $\delta>0$ is bounded above.
	
For $|p|\geq L^{-k}$ the asymptotic behaviour is ${\widehat{\Malt}}_k(p)\approx |p|^{2M}L^{(2M-2)k}$
and $\widehat{\mathcal{C}}_k^{N+1}(p)\approx 0$.Then the bound is implied by $\widehat{\Malt}_k(p)\ll\widehat{\Malt}_{k+1}(p)$.
\end{remark}
\begin{proof}
Here, we implement rigorously the heuristics described above.
The first step is to compare the operators
${\Malt}_k$ and $\Malt_{k+1}$.
For $|p|\geq L^{-k}$ and $L\geq 8$ we observe using $|p|/2\leq |q(p)|\leq |p|$ that
\begin{align}
\begin{split}
4|q(p)|^2\leq 16|q(p)|^4L^{2k}&\leq \frac{32}{L^2}\sum_{|\alpha|=2} L^{2(k+1)(|\alpha|-1)}|q(p)^{2\alpha}|\\
&\leq \frac12\sum_{|\alpha|=2} L^{2(k+1)(|\alpha|-1)}|q(p)^{2\alpha}|\leq \frac12\widehat{\Malt}_{k+1}(p).
\end{split}
\end{align}
Hence,  for $|p|\geq L^{-k}$ and $L\geq 8$, we have
\begin{align}\label{eq:MkMk+1}
\begin{split}
4 \widehat{\Malt}_k(p)&=4\sum_{1\leq |\alpha|\leq M} L^{2k(|\alpha|-1)}|q(p)^{2\alpha}| \\
& \leq   4|q(p)|^{2}+\frac{4}{L^2}\sum_{2\leq |\alpha|\leq M} L^{2(k+1)(|\alpha|-1)}|q(p)|^{2|\alpha|} \leq \widehat{\Malt}_{k+1}(p).
\end{split}
\end{align}
We claim that there is a constant $k_0=k_0(\lambda)$ 
independent of $k$ and $N$ such that
\begin{align}\label{eq:sumCjMk} 
2\, \widehat{{\mathcal{C}}}_{k+2}^{N+1}(p)=2\sum_{k'=k+2}^{N+1} \widehat{\mathcal{C}}_{k'}(p)\leq \lambda \widehat{\Malt}_{k+1}(p)^{-1}                                                                                                                                  
\end{align}
for $|p|\geq L^{-k+k_0}$.
To prove this we observe that for $L^{-j-1}<|p|\leq L^{-j}$ and $j< k-k_0$ the sum on the left hand side is by \eqref{finalfrdupper}
dominated by a geometric series which implies
\begin{align}\begin{split}
\left|\sum_{k'=k+2}^{N+1} \widehat{\mathcal{C}}_{k'}(p)\right|&\leq C_1L^{2(d+\tilde{n})+1}L^{2j}L^{-(k+2-j)(d-1+n)} \\		
&=C_1L^{d+2\tilde{n}+2-n}L^{2j}L^{-(k+1-j)(d-1+n)}.
\end{split}
\end{align}
Note that 
\begin{align}\label{eq:qp_sum_bounded_power}
\sum_{|\alpha|=l} |q(p)^{2\alpha}|\leq |q(p)|^{2l}\leq |p|^{2l}
\end{align}	
	This implies that,  for $j\leq k$ and $|p|\leq L^{-j}$,  the right hand side of \eqref{eq:sumCjMk} satisfies
	\begin{align}\label{eq:boundMk+1}
		\widehat{\Malt}_{k+1}(p)\leq  \sum_{l=1}^M L^{2(l-1)(k+1)}L^{-2lj}\leq 2L^{-2j}L^{2(M-1)(k+1-j)}. 
	\end{align} 	
	Therefore we find that
	\begin{align}
		\lambda^{-1} \widehat{\Malt}_{k+1}(p)\left|\sum_{k'=k+2}^{N+1} \widehat{\mathcal{C}}_{k'}(p)\right|\leq 
		\frac{2C_1}{\lambda}L^{d+2\tilde{n}+2-n}L^{(k+1-j)(2M-1-d-n)}\leq \frac{1}{2}           
	\end{align}
	for $k-j>k_0$ with $k_0=k_0(\lambda)=\lceil\log_3(4C_1/\lambda)\rceil+d+2\tilde{n}+2-n$ where we used
	that $L\geq 3$ and  $n\geq 2M$  and thus $2M-1-d-n\leq -1$. 
	Note that the constant $C_1$ from \eqref{finalfrdupper} does not depend on $L$. 
	Hence, in particular, $k_0$ is independent of $L$.
	This proves \eqref{eq:sumCjMk}.
	The bounds  \eqref{eq:MkMk+1} and \eqref{eq:sumCjMk} thus, 
	for $|p|\geq L^{-k+k_0}$, $\delta<\frac{1}{4\lambda}$, and $\varepsilon<1$, jointly imply 
\begin{align}\begin{split}\label{eq:finalreglarge}
		\left(\lambda\widehat{\Malt}_k^{-1}(p)+ (1+\varepsilon)\widehat{{\mathcal{C}}}_{k+2}^{N+1}(p)\right)^{-1}+\delta \widehat{\Malt}_{k+1}(p)
		 & \leq                                                                                                   
		\frac{1}{\lambda}\widehat{\Malt}_k(p)+\delta\widehat{\Malt}_{k+1}(p)\\
		  & \leq \frac{1}{4\lambda}\widehat{\Malt}_{k+1}(p)+\frac{1}{4\lambda}\widehat{\Malt}_{k+1}(p) \\
		  & \leq \left(2\lambda\widehat{\Malt}^{-1}_{k+1}(p)\right)^{-1}                                     \\
		  & \leq\left(\lambda\widehat{\Malt}_{k+1}^{-1}(p)+(1+\varepsilon)\widehat{{\mathcal{C}}}_{k+2}^{N+1}(p)  
		\right)^{-1}.
		\end{split}
	\end{align}
	In the first and the last step we used the fact that the inversion of a Hermitian positive definite matrix 
	is a monotone operation (see Lemma~\ref{le:matrix_monotone}).
	This ends the proof for  $p\in\widehat{T}_N$ with $|p|\geq L^{-k+k_0}$.

	For $p\in \widehat{T}_N$ such that $|p|<L^{-k+k_0}$  we note that there are constants
	$\omega_1,\omega_2,\Omega_1,\Omega_2$ depending on $L$, $k_0(\lambda)$, and
	$\lambda$ such that
	\begin{align}\label{eq:specbounds}
		\omega_1|p|^{-2} & \leq\widehat{{{\mathcal{C}}}}_{k+2}^{N+1}(p)\leq 
		(1+\varepsilon)\widehat{{{\mathcal{C}}}}_{k+2}^{N+1}(p)\leq \Omega_1|p|^{-2} \text{ and } \\\label{eq:specbounds2}
		\omega_2|p|^{-2} & \leq  \lambda\widehat{{\Malt}}_{k+1}^{-1}(p)\leq                       
		\Omega_2|p|^{-2}.
	\end{align}
 Indeed, the upper bounds are trivial and even hold uniformly in  $k_0$ and $N$
	for all $p$ because
	$\widehat{\mathcal{C}}(p)\leq \Omega_1|p|^{-2}$ for some constant
	$\Omega_1$ by \eqref{Ahatestimate} and 
	$\boldsymbol M_{k+1}\geq -\Delta$. 
	The first lower bound follows from \eqref{finalfrdlower} which implies  the bound
	\begin{align}
		\widehat{{\mathcal{C}}}_{k+2}^{N+1}(p)\geq
		\widehat{{\mathcal{C}}}_j(p)\geq                        
		cL^{-2(d+\tilde{n})-1}L^{2j}\geq cL^{-2(d+\tilde{n})-3}|p|^{-2}
	\end{align}
	for $L^{-j-1}<|p|<L^{-j}$ and $j\geq k+2$.
	For $L^{-j-1}<|p|<L^{-j}$ and $k-k_0\leq j<k+2$ we use
	\begin{align}\begin{split}
	\widehat{{\mathcal{C}}}_{k+2}^{N+1}(p)&\geq	\widehat{\mathcal{C}}_{k+2}(p)\geq cL^{-2(d+\tilde{n})-1}L^{(k+2-j)(-d+1-n)}L^{2j}\\ 
	&		\geq cL^{-2(d+\tilde{n})-3}L^{(k_0+2)(-d+1-n)}|p|^{-2}. 
	\end{split}\end{align}
	Therefore the lower bound in \eqref{eq:specbounds} holds with 
	$\omega_1=cL^{-2(d+\tilde{n})-3+(k_0+2)(-d+1-n)}$.
	The second lower
	bound is a consequence of  \eqref{eq:qp_sum_bounded_power} which implies
	\begin{align}\begin{split}
		\widehat{\Malt}_{k+1}(p)  &\leq\sum_{l=1}^ML^{2(l-1)(k+1)}|p|^{2l}                    
		                      \\
		                      & \leq \sum_{l=1}^ML^{2(l-1)(k+1)}L^{2(l-1)(-k+k_0)} |p|^{2} 
		                       \leq                                                       
		2L^{2(M-1)(k_0+1)}|p|^2.
	\end{split}\end{align}
	if  $|p|<L^{-k+k_0}$. So the lower bound in \eqref{eq:specbounds2} holds with $\omega_2=\lambda(2L^{2(M-1)(k_0+1)})^{-1}$.

	Observe that $\widehat{{\mathcal{C}}}^{N+1}_{k+2}(p)$ and $\widehat{\Malt}_{k+1}(p)$ are Hermitian 
	and they commute because $\widehat{\Malt}_{k+1}(p)$ is a multiple of the identity. 
	Therefore we can work in a basis where both matrices are diagonal which  reduces
	the estimates to the scalar case $m=1$. 
	Then the bound we want to show  is essentially the estimate $(a-x)^{-1}-a^{-1}>x/a^2$ for $a>x>0$.
	In more detail, using \eqref{eq:specbounds} and the trivial estimate $\widehat{\Malt}_k(p)\leq \widehat{\Malt}_{k+1}(p)$, 
	we find for $|p|<L^{-k+k_0}$, $m=1$, and $0<\delta <(1+\varepsilon)/\mu$, 
	\begin{align}\begin{split}\label{eq:techbound}
		  & \left(\lambda\widehat{\Malt}_{k+1}^{-1}(p)+(1+\varepsilon-\mu\delta)\widehat{{\mathcal{C}}}_{k+2}^{N+1} 
		(p)\right)^{-1}-
		\left(\lambda\widehat{\Malt}_{k}^{-1}(p)+(1+\varepsilon)\widehat{{\mathcal{C}}}_{k+2}^{N+1}
		(p)\right)^{-1}\\
		  & \quad\geq                                                                                                  
		\left(\lambda\widehat{\Malt}_{k+1}^{-1}(p)+(1+\varepsilon-\mu\delta)\widehat{{\mathcal{C}}}_{k+2}^{N+1}
		(p)\right)^{-1}-
		\left(\lambda\widehat{\Malt}_{k+1}^{-1}(p)+(1+\varepsilon)\widehat{\mathcal{{C}}}_{k+2}^{N+1}
		(p)\right)^{-1}\\
		  & \quad\geq \frac{\mu\delta\widehat{{\mathcal{C}}}_{k+2}^{N+1}(p)}{                                             
		\left(\lambda\widehat{\Malt}_{k+1}^{-1}(p)+(1+\varepsilon-\mu\delta)\widehat{{\mathcal{C}}}_{k+2}^{N+1}
		(p)\right)
		\left(\lambda\widehat{\Malt}_{k+1}^{-1}(p)+(1+\varepsilon)\widehat{{\mathcal{C}}}_{k+2}^{N+1}
		(p)\right)}\\
		  & \quad\geq\frac{ \mu \delta \omega_1|p|^{-2}}{(\Omega_1+\Omega_2)^2|p|^{-4}}                              \\
		  & \quad\geq  \delta \widehat{\Malt}_{k+1}(p)                                                                
		\frac{\mu\omega_1\omega_2}{\lambda(\Omega_1+\Omega_2)^2}.
		\end{split}
	\end{align}
	Then for 
	\begin{align}
	\mu\geq \lambda\frac{(\Omega_1+\Omega_2)^2}{\omega_1\omega_2}
	\end{align}
	(where $\omega_1$, $\omega_2$, $\Omega_1$, and $\Omega_2$ where introduced in \eqref{eq:specbounds}
	and \eqref{eq:specbounds2}) the inequality \eqref{eq:basicweightineq} follows.
	For $m>1$ the claim follows by applying \eqref{eq:techbound} to each diagonal entry of the diagonalised matrices.
	The estimates \eqref{eq:finalreglarge} and \eqref{eq:techbound} imply the claim.
\end{proof}

\section{Basic properties of the operators \texorpdfstring{$\protect \headingAk$}{AkX} and \texorpdfstring{$\protect\headingAkkp$}{Ak:k+1X}}
Recall that $\weightzeta\in (0,\tfrac14)$ is a fixed parameter 
and $\Xi_{\mathrm{max}}$ was defined in \eqref{eq:thetamax}.
For given values of $h$, $\delta$, and $\mu$ (that will be specified later)
we define sequences
\begin{align}\label{eq:definition_hj_deltaj}
h_j=2^jh,\quad \delta_j=4^{-j}\delta, \quad  \weightzeta_j=2\weightzeta-\sum_{i=0}^j \mu\Xi_{\mathrm{max}}\delta_i.
\end{align} 

The following lemma proves the claims \ref{w:w1} and \ref{w:w2} from   Theorem~\ref{th:weights_final}.
\begin{lemma}\label{prop:W1}
Under the assumptions of Theorem~\ref{th:weights_final},
 for every $L\geq 2^{d+3}+16R$, there are constants $\lambda>0$, $\mu(L)>1$, 
	and $\delta(L)\in \left(0,\frac{\weightzeta}{2\mu\Xi_{\mathrm{max}}}\right)$
	 such that $\weightzeta_j\geq \weightzeta$ for all $j=0,\ldots ,N$, and 
	for all $0\leq k\leq N$:
\begin{enumerate}[label=(\roman*),leftmargin=0.7cm]
		\item \label{it:welldefined}
			The operators $\boldsymbol{A}_{k:k+1}^X$ and $\boldsymbol{A}_k^X$ are
		      well-defined, symmetric and non-negative operators on $\Xcal_N$ for any
		      $X\in\mathcal{P}_k$;
		\item \label{it:translation} Translation invariance: For any translation $\tau_a\p(x)=\p(x-a)$ with $a\in (L^k\mathbb{Z})^d/(L^N\mathbb{Z})^d$ 
		the equalities $(\p,\boldsymbol{A}_k^X\p)=(\tau_a\p,\boldsymbol{A}_k^{X+a}\tau_a\p)$ and 
		      $(\p,\boldsymbol{A}_{k:k+1}^X\p)=(\tau_a\p,\boldsymbol{A}_{k:k+1}^{X+a}\tau_a\p)$ hold;
		\item \label{it:locality} Locality: The operators $\boldsymbol{A}_k^X$ and $\boldsymbol{A}_{k:k+1}^X$ only depend on the values
		      of $\p$ in $X^{++}$ and the are shift invariant, i.e., they are measurable with respect to the $\sigma$-algebra generated by
		      $\nabla\p{\restriction_{\vec{E}(X^{++})}}$;
		\item \label{it:monotonicity} Monotonicity: For $Y\subset
		      X$ the inequalities $\boldsymbol{A}_k^Y\leq \boldsymbol{A}_k^X$ and $\boldsymbol{A}_{k:k+1}^Y\leq \boldsymbol{A}_{k:k+1}^X$  
		      hold in the sense of operators;
		\item \label{it:bound} Bounds: The weight functions are bounded from above as
		      follows 
		      \begin{align}\label{eq:akbound}
		      	\boldsymbol{A}_k^X       & \leq \left(\lambda \Malt_k^{-1}+ 	(1+\weightzeta_k)\sum_{j=k+1}^{N+1} \mathscr{C}_j \right)^{-1}   \\
		      	\label{eq:akk+1bound}
		      	\boldsymbol{A}_{k:k+1}^X & \leq \left(\lambda \Malt_{k}^{-1}+	(1+\weightzeta_k)\sum_{j=k+2}^{N+1} \mathscr{C}_j \right)^{-1}. 
		      \end{align}
	\end{enumerate}
\end{lemma}
\begin{proof}
Note first that the estimate $\zeta_j\geq \zeta$ is an immediate consequence of the definition of $\delta$.
	
The proof  is by induction on $k$.
First, for $k=0$, the
properties \ref{it:welldefined}, \ref{it:locality}, and \ref{it:monotonicity}  are obvious. Indeed, $\Qscr$ has range at most $R$, is
positive, and $D\p(x)$ can be expressed as a function of $\nabla\p{\restriction_{\vec{E}(x+[0,R]^d)}}$. Similarly  $\boldsymbol{M}_0^X$ is non-negative, symmetric,
and monotone in $X$ and $\boldsymbol{M}_0^X\p$ only depends on the values of $\nabla \p$ 
restricted to the bonds $\vec{E}((X^+ +[-M,M]^d)\cap T_N)$ and
 $(X^+ +[-M,M]^d)\cap T_N\subset X^{++}$ since $R\geq M$  by \eqref{eq:definition_R}
and \eqref{eq:defofMConstant}.

	Translation invariance for $k=0$ follows from the facts that the discrete derivatives commute with translations and
	$\tau_{-a}\boldsymbol{1}_{X+a}\tau_a=\boldsymbol{1}_X$ where $\boldsymbol{1}_X$ denotes
	the multiplication operator with the indicator function of $X$ which implies translation
	invariance of the operators $\boldsymbol{M}_k$ in the set variable that is
	$( \p,\boldsymbol{M}_k^X\p)=(\tau_{a}\p,\boldsymbol{M}_k^{X+a}\tau_a\p)$. A similar statement holds 
	for $\Qscr$.
	Finally, we establish the bound \eqref{eq:akbound}.
	First we note that there exist two constants  $\Omega,\omega>0$  independent of $L$, such that  the operator 
	$\mathscr{A}$ (see \eqref{eq:A_again}) satisfies the bounds
	\begin{align}\label{eq:MequivA}
		\omega\mathscr{A}\leq \Malt_0\leq \Omega\mathscr{A}. 
	\end{align}
	This is a consequence of the fact that both operators have the Fourier modes bounded uniformly by $|p|^2$ from above and below. 
	For $\mathscr{A}$ the bounds follow from \eqref{Ahatestimate} and for $\Malt_0$ the lower bound follows from $|q(p)|^2\geq\frac{|p|^2}{4}$ 
	while the upper bound follows from $|q(p)|\leq |p|$ and the fact that the dual torus is bounded.
	Then, for $\delta_0\leq \weightzeta/\Omega$, we estimate 
	\begin{align}
		\begin{split}                                                           
		\boldsymbol{A}_0^X\leq \boldsymbol{A}_0^{\TN}&= (1-4\weightzeta)\mathscr{A}+\delta_0\Malt_0             
		\\ &
		\leq (1-3\weightzeta)\mathscr{A}                                                 
		\leq \frac{1}{(1+3\weightzeta)\mathscr{A}^{-1}}                                  
		\leq \frac{1}{{\weightzeta\omega}\Malt_0^{-1}+  (1+2\weightzeta)\mathscr{C}}. 
		\end{split}                                                             
	\end{align}
	Hence \eqref{eq:akbound} holds for $k=0$ and $\lambda\leq {\weightzeta}{\omega}$. We now fix
	\begin{align}\label{eq:defoflambda}
	\lambda=\min\left({\weightzeta}{\omega},\frac14\right)
	\end{align}
where $\omega$ was introduced in \eqref{eq:MequivA} and 
	\begin{align}\label{eq:defofmu}
	\mu=\mu(\lambda,L)>1
	\end{align}
	as in  Lemma~ \ref{le:opineq}. We then set
\begin{align}\label{eq:defofdelta}
	\delta(L)=\min\left(\frac{\weightzeta}{\Omega},\frac{\weightzeta}{2\mu\Xi_{\mathrm{max}}}\right),
	\end{align}
		 where $\Xi_{\mathrm{max}}$ was introduced in \eqref{eq:thetamax}
	 
	Now we perform the induction step from  $\boldsymbol{A}_k^X$ to 
	$\boldsymbol{A}_{k:k+1}^X$.
	
	 First we show that $\boldsymbol{A}_{k:k+1}^X$ is well
	defined. By the induction hypothesis, the operator $\boldsymbol{A}_k^X$ is
	non-negative and symmetric and the
	bound \eqref{eq:akbound} implies
	\begin{align}
		\boldsymbol{A}_k^X\leq \left((1+\weightzeta_k)\mathscr{C}_{k+1}\right)^{-1}.
	\end{align} 
  Since $\mathscr{C}_{k+1}$ is also symmetric, Lemma~\ref{le:matrix_monotone} 
   implies that $\boldsymbol{A}_{k:k+1}^X$ is well defined using the extension
   defined in \eqref{eq:algebraic_integration2}
	and it 
	can be expressed as follows
	\begin{align}\label{eq:Akredefine}
		\boldsymbol{A}_{k:k+1}^X=
		\left(\boldsymbol{A}_k^X\right)^{\frac12}\left(1-\left(\boldsymbol{A}_k^X\right)^{\frac12}\mathscr C_{ k +1}
		\left(\boldsymbol{A}_k^X\right)^{\frac12}\right)^{-1}\left(\boldsymbol{A}_k^X\right)^{\frac12}.                 
	\end{align}
	
	This expression shows that the operator $\boldsymbol{A}_{k:k+1}^X$  is symmetric and, 
	again by Lemma~\ref{le:matrix_monotone}, also  non-negative.
	Moreover, the matrix monotonicity that was stated in Lemma~\ref{le:matrix_monotone}
	implies that the monotonicity $\boldsymbol{A}_{k:k+1}^Y\leq \boldsymbol{A}_{k:k+1}^X$
	for $Y\subset X$ follows from the induction
	hypothesis
	$\boldsymbol{A}_{k}^Y\leq \boldsymbol{A}_{k}^X$.
	 
	 To prove the claim \ref{it:translation} for $\boldsymbol{A}_{k:k+1}^X$, we use the induction hypothesis, 
	 the series representation \eqref{eq:algebraic_integration4}
	 for $\boldsymbol{A}_{k:k+1}^X$, and the translation invariance of the kernel $\mathscr{C}_{k+1}$, i.e., 
	 $[\tau_a,\mathscr{C}_{k+1}]=\tau_a \mathscr{C}_{k+1} - \mathscr{C}_{k+1}\tau_a = 0$. 
	The easiest way to show the locality of $\boldsymbol{A}_{k:k+1}^X$ stated in \ref{it:locality} is based on the observation that, 
	by Gaussian integration \eqref{eq:GaussianCalculus}, we get the identity 
	\begin{align}
	\begin{split}		
		\int_{\Xcal_N} e^{\frac{1}{2}(\p+\xi,\boldsymbol{A}_{k}^X(\p+\xi))}\, 
		\mu_{(1+\weightzeta)\mathscr{C}_{k+1}}(\d \xi)
		= &\frac{e^{\frac{1}{2}(\p, 
		((\boldsymbol{A}_{k}^X)^{-1}-(1+\weightzeta)\mathscr{C}_{k+1})^{-1}\p)}}
		{\det\left(\1-\left((1+\weightzeta)\mathscr{C}_{k+1}\right)^{\frac12}\boldsymbol{A}_k^X
		\left((1+\weightzeta)\mathscr{C}_{k+1}\right)^{\frac12}\right)^{\frac12}}\\
		=&\frac{e^{\frac12 (\p,\boldsymbol{A}_{k:k+1}^X\p)}}
		{\det\left(\1-\left((1+\weightzeta)\mathscr{C}_{k+1}\right)^{\frac12}\boldsymbol{A}_k^X
		\left((1+\weightzeta)\mathscr{C}_{k+1}\right)^{\frac12}\right)^{\frac12}}. 
		\end{split}                    
	\end{align}
By the induction hypothesis the left hand side
	is measurable with respect to the $\sigma$-algebra 
	generated by $\nabla\p{\restriction_{\vec{E}(X^{++})}}$, hence the same is true for the right hand side.

	For the proof of \ref{it:bound} for $\boldsymbol{A}_{k:k+1}^X$ we first note that, by the monotonicity \ref{it:monotonicity}, 
	it is sufficient to prove the bound for $X=\TN$. 
	This is an immediate consequence of the bound for $\boldsymbol{A}_k^{\TN}$, Lemma~\ref{le:matrix_monotone} \ref{it:matrix_monotone3}, 
	and the inequality $\weightzeta_k\geq \weightzeta$ which implies
	\begin{align}
		\begin{split}
		\boldsymbol{A}_{k:k+1}^{\TN}&=((\boldsymbol{A}_k^{\TN})^{-1}-(1+\weightzeta)\mathscr{C}_{k+1})^{-1} 
		\\
		& \leq \left( \lambda\boldsymbol{M}_{k}^{-1}+(1+\weightzeta_k)\sum_{j=k+1}^{N+1}\mathscr{C}_j-(1+\weightzeta)\mathscr{C}_{k+1}\right)^{-1}
		\\
		& \leq  \left( \lambda\boldsymbol{M}_{k}^{-1}+(1+\weightzeta_k)\sum_{j=k+2}^{N+1}\mathscr{C}_j\right)^{-1}.
		\end{split}
	\end{align}

	It remains to show the induction step from $\boldsymbol{A}_{k:k+1}^X$ to $\boldsymbol{A}_{k+1}^X$. 
	We begin with the observation
	that the operators $\boldsymbol{M}_{k+1}^X$ are well-defined, symmetric, non-negative, monotone in $X$,
	and translation invariant.
	Moreover $\boldsymbol{M}_{k+1}^X$ only depends on 
	$\nabla\p{\restriction_{\vec{E}(X^++[-M,M]^d)}}$  and $X^++[-M,M]^d\subset X^{++}$  once  $L\geq M$, the inequality that 
	follows from $M\leq R\leq L$.
	
	Now, the points \ref{it:welldefined} and \ref{it:translation}
	follow from the induction hypothesis applied to $\boldsymbol{A}_{k:k+1}^{X^\ast}$
	and the previous observation.
	The claim  \ref{it:monotonicity} follows from the induction hypothesis for $\boldsymbol{A}_{k:k+1}$ applied to $X^\ast\subset Y^\ast$
	for $X\subset Y$ and the monotonicity of
	$\boldsymbol{M}_{k+1}^X$.
To show \ref{it:locality}, it remains to check that $\boldsymbol{A}_{k:k+1}^{X^\ast}$
is measurable with respect to $\nabla\p{\restriction_{\vec{E}(X^{++})}}$.
Using the induction hypothesis we are left to show the
inclusion $(X^\ast)^{++}\subset X^{++}$.
Note that by \eqref{eq:nghbhdscompact},
\begin{align}
(X^\ast)^{++}=
\begin{cases}
X+[-2^d-3R,2^d+3R]^d\quad & \text{for $X\in \mathcal{P}_1$},\\
X+[-(2^d+2)L^{k-1},(2^d+2)L^{k-1}]^d \quad &\text{for $X\in \mathcal{P}_k$, $k\geq 2$}.
\end{cases}
\end{align}
Therefore $(X^\ast)^{++}\subset X^{++}$ holds for $X\in \mathcal{P}_k$, $k\geq 1$, and $L\geq 2^d+3R$.
	
	Finally, the bound for $\boldsymbol{A}_{k+1}^X$ is a direct consequence of Lemma~\ref{le:opineq} and our choice for $\delta$.
	Indeed, recall that $\delta\Xi_{\mathrm{max}}\leq \frac{\weightzeta}{2\mu}\leq \frac{1}{\mu}$ 
	and $\delta_{k+1}\leq \delta$, hence Lemma~\ref{le:opineq} and the induction hypothesis imply	
	\begin{align}\notag	
	\boldsymbol{A}_{k+1}^X\leq \boldsymbol{A}_{k+1}^{\TN}
	&=\boldsymbol{A}_{k:k+1}^{\TN}+\delta_{k+1}\boldsymbol{M}_{k+1}^{\TN}
	\\
	&\leq \Big( \lambda\Malt_k^{-1}+(1+\weightzeta_k)\sum_{j=k+2}^{N+1}\mathscr{C}_j\Big)^{-1}+\delta_{k+1}\Xi_{\mathrm{max}}\Malt_{k+1}
	\\
	&\leq \Big(\lambda\Malt_{k+1}^{-1}+(1+\weightzeta_k-\mu \delta_{k+1} \Xi_{\mathrm{max}})\sum_{j=k+2}^{N+1}\mathscr{C}_j\Big)^{-1}.	
	\end{align}
	The claim follows from $\weightzeta_k-\weightzeta_{k+1}=\mu\Xi_{\mathrm{max}}\delta_{k+1}$.
\end{proof}

\section{Subadditivity properties of the operators \texorpdfstring{$\protect\headingAk$}{AkX} and \texorpdfstring{$\protect \headingAkkp$}{Ak:k+1X}}

In this section we prove that the weight operators satisfy additivity properties that directly imply 
the statements \ref{w:w3}-\ref{w:w6} in Theorem~\ref{th:weights_final}.
 In Chapter~\ref{sec:normprop} we will also prove that they imply that the norms we defined in Section \ref{se:norms_new} are sub-multiplicative.

\begin{lemma}\label{prop:W2}
	The  weight operators $\boldsymbol{A}_{k}^X$ and $\boldsymbol{A}_{k:k+1}^X$
	and the strong norm weight functions $\boldsymbol{G}_k^X$ defined in \eqref{eq:strong_weight}  satisfy  for $0\leq k\leq N-1$, 
	under the same assumptions as in Theorem~\ref{th:weights_final} with $\delta$ and $\lambda$ as in    Lemma~\ref{prop:W1}, 
	the following (sub)additivity properties:
\begin{enumerate}[label=(\roman*),leftmargin=0.7cm]
	\item \label{it:subadditivityI} Additivity: For any strictly disjoint $X, Y\in \mathcal{P}_k$, the
		      equality
		      \begin{align}\label{eq:subadd1}
		      	\boldsymbol{A}_{k}^{X\cup Y} & =\boldsymbol{A}_{k}^X+\boldsymbol{A}_{k}^Y 
		      \end{align}
		      holds. For any $X,Y\in \mathcal{P}_k$ such that $\mathrm{dist}(X,Y)\geq \frac{3}{4}L^{k+1}$,  we have
		      \begin{align}
		      	\label{eq:subadd2}
		      	\boldsymbol{A}_{k:k+1}^{X\cup Y} & =\boldsymbol{A}_{k:k+1}^X+\boldsymbol{A}_{k:k+1}^Y. 
		      \end{align}
		\item \label{it:subadditivityII} Subadditivity: For any disjoint $k$-polymers $X, Y\in\mathcal{P}_k$, the inequality 
		      \begin{align}\label{eq:subadd3}
		      	\boldsymbol{A}_k^X+\boldsymbol{G}_k^Y\leq \boldsymbol{A}_k^{X\cup Y} 
		      \end{align}
		      holds if $h^{-2}<\delta$. For any $(k+1)$-polymer $U\in\mathcal{P}_{k+1}$ and
		      a $k$-polymer $X\in\mathcal{P}_k$ such that $\pi(X)=U$, the inequality 
		      \begin{align}\label{eq:subadd4}
		      	\boldsymbol{A}_{k:k+1}^X+2\boldsymbol{G}_{k}^{U^+}\leq \boldsymbol{A}_{k+1}^U 
		      \end{align}
		      holds if $8h^{-2}<\delta$.
	\end{enumerate}
\end{lemma}
\begin{proof}
We first prove \eqref{eq:subadd1} and proceed by induction. Note that for all $k\geq 0$ and
 any disjoint 
$X,Y\in \mathcal{P}_k$ we have
\begin{align}\label{eq:Mkadditive}
\boldsymbol{M}_k^{X\cup Y}=\boldsymbol{M}_k^X+\boldsymbol{M}_k^Y
\end{align}
since a block $B\in \mathcal{B}_k$ is contained in $X\cup Y$ if and only if either $B\subset X$ or
$B\subset Y$. 
From \eqref{eq:Mkadditive} with $k=0$ and \eqref{eq:defAk}, it follows  that 
\eqref{eq:subadd1} holds for $k=0$. Hence it suffices
to show that $\eqref{eq:subadd1}_k\Rightarrow \eqref{eq:subadd2}_k$ for $k\geq 0$
and $\eqref{eq:subadd2}_k\Rightarrow \eqref{eq:subadd1}_{k+1}$ for $k\geq 0$.

To prove the second statement,
$\eqref{eq:subadd2}_k\Rightarrow \eqref{eq:subadd1}_{k+1}$, we consider
strictly disjoint $X,Y\in \mathcal{P}_{k+1}$. Then $\mathrm{dist}(X,Y)\geq L^{k+1}$ and, by \eqref{eq:distXast},
 $X^\ast, Y^\ast\in \mathcal{P}_k$ satisfy 
\begin{align}
\mathrm{dist}(X^\ast, Y^\ast)\geq L^{k+1}-2(2^d+R)L^k\geq \frac{3}{4} L^{k+1}
\end{align}
for $L\geq 2^{d+3}+8R$.  
Then 
\begin{align}
\boldsymbol{A}_{k:k+1}^{X^\ast\cup Y^\ast}=\boldsymbol{A}_{k:k+1}^{X^\ast}+\boldsymbol{A}_{k:k+1}^{Y^\ast}
\end{align}
by $\eqref{eq:subadd2}_k$.
Together with \eqref{eq:Mkadditive} this implies $\boldsymbol{A}_{k+1}^{X\cup Y}=\boldsymbol{A}_{k+1}^X+\boldsymbol{A}_{k+1}^Y$.

To prove the statement $\eqref{eq:subadd1}_k\Rightarrow \eqref{eq:subadd2}_k$, we observe that by property \ref{it:locality} in Lemma~\ref{prop:W1}, 
the operator $\boldsymbol A_{k}^{X}$ is, for a $k$-polymer $X\in \mathcal{P}_k$, measurable with respect to the $\sigma$-algebra 
generated by $\nabla\p{\restriction_{\vec{E}(X^{++})}}$ and similarly $\boldsymbol{A}_k^{Y}\p$  is measurable with respect to the $\sigma$-algebra 
generated by $\nabla\p{\restriction_{\vec{E}(Y^{++})}}$. 
Let $X,Y\in \mathcal{P}_k$ be polymers such that $\mathrm{dist}(X,Y)\geq \frac34 L^{k+1}$. 
Note that $\mathrm{dist}(X^{++},Y^{++})\geq \mathrm{dist}(X,Y)-4L^k> L^{k+1}/2$ for $L>16$ and $k\geq 1$ 
and thus $\mathrm{dist}(X^{++},Y^{++})\geq \mathrm{dist}(X,Y)-4R> L/2$ for $k=0$ and $L\geq 16R$.
This implies that $\nabla\xi_{k+1}{\restriction_{\vec{E}(X^{++})}}$ and $\nabla\xi_{k+1}{\restriction_{\vec{E}(Y^{++})}}$
are independent under $\mu_{k+1}$ and therefore also under the measure $\mu_{(1+\weightzeta_{k+1})\mathscr{C}_{k+1}}$.
Hence the random variables $(\p+\xi_{k+1},\boldsymbol{A}_k^{X}(\p+\xi_{k+1}))$ and $(\p+\xi_{k+1},\boldsymbol{A}_k^{Y}(\p+\xi_{k+1}))$ 
are independent under the same measure for $\xi_{k+1}$ and any $\p$.
To simplify the notation we denote $C=(1+\weightzeta_{k+1})\mathscr{C}_{k+1}$.
Independence and the formula \eqref{eq:GaussianCalculus}  for Gaussian integration shows that there exist positive constants $c_X$, $c_Y$, and $c_{X\cup Y}$ such that
\begin{align}\begin{split}
		\frac{e^{\frac{1}{2}(\boldsymbol{A}_{k:k+1}^{X\cup Y}\p,\p)}}{c_{X\cup Y}}
		  &=\int_{\Xcal_N} e^{\frac{1}{2}(\boldsymbol{A}_{k}^{X\cup Y}(\p+\xi),\p+\xi)}\, \mu_{C}(\d \xi)  \\
		  & =\int_{\Xcal_N} e^{\frac{1}{2}((\boldsymbol{A}_{k}^{X}+\boldsymbol{A}_k^Y)(\p+\xi),\p+\xi)}\, \mu_{C}(\d \xi) \\
		  & =\int_{\Xcal_N} e^{\frac{1}{2}(\boldsymbol{A}_{k}^{X}(\p+\xi),\p+\xi)}\, \mu_{C}(\d \xi)              
		\int_{\Xcal_N} e^{\frac{1}{2}(\boldsymbol{A}_{k}^{Y}(\p+\xi),\p+\xi)}\, \mu_{C}(d \xi)\\
		  & =\frac{ e^{\frac{1}{2}(\boldsymbol{A}_{k:k+1}^{X}\p,\p)}}{c_X}\frac{e^{\frac{1}{2}(\boldsymbol{A}_{k:k+1}^{Y}\p,\p)}}{c_Y},               
		\end{split}
\end{align}
where,  in the second step,  we used the induction hypothesis. 
In the third step we usedthat the integral factors by the  finite range property of $\mathscr{C}_{k+1}$.
Evaluation for $\p\equiv 0$ shows that the constants must satisfy $c_{X\cup Y}=c_Xc_Y$ which implies
$\boldsymbol{A}_{k:k+1}^{X\cup Y}=\boldsymbol{A}_{k:k+1}^X+\boldsymbol{A}_{k:k+1}^Y$. 
The equality $c_{X\cup Y}=c_Xc_Y$  can also checked explicitly.
Using \eqref{eq:GaussianCalculus} we can rewrite
	\begin{align}
	\begin{split}
	 c_X^2c_Y^2
	 &=\det(\1-C^{\frac12}\boldsymbol A_k^XC^{\frac12})\det(1-C^{\frac12}\boldsymbol A_k^YC^{\frac12})
	 \\
	 &=\det(\1-C^{\frac12}(\boldsymbol A_k^X+\boldsymbol A_k^Y)C^{\frac12}+C^{\frac12}\boldsymbol A_k^XC\boldsymbol A_k^YC^{\frac12})
	 \\
	 &=\det(\1-C^{\frac12}\boldsymbol A_k^{X\cup Y}C^{\frac12})
	 =	 
	 c_{X\cup Y}^2.
	\end{split}
	\end{align}
Here we used the induction hypothesis for the linear term.
The  quadratic term vanishes because $\boldsymbol A_k^XC\boldsymbol{A}_k^Y=0$ which we now show.
 Note that $\mathrm{supp}\,(\boldsymbol A_k^Y\p)\subset X^{++}$. 
 Indeed, the symmetry of $\boldsymbol A_k^Y$ and the locality property \ref{it:locality} in Lemma~\ref{prop:W1} imply that 
 $(\psi,\boldsymbol A_k^Y\p)=0$ for any $\psi$ with $\mathrm{supp}\,\psi\cap Y^{++}=\emptyset$.
Since the kernel $\mathcal{C}(x)$ of $C$ is constant for $|x|_\infty\geq L^{k+1}/2$  we find that 
$C\p(x)=c$ for some constant $c$ for $x\notin B_{L^{k+1}/2}(\mathrm{supp}\, \p)$. 
Using $\mathrm{dist}(X^{++},Y^{++})\geq L^{k+1}/2$ we conclude that $\nabla C\boldsymbol{A}_k^Y\p{\restriction_{\vec{E}(X^{++})}}=0$
for all $\p \in \mathcal{X}_N$ and therefore $\boldsymbol A_k^XC\boldsymbol{A}_k^Y\p=0$.
	
Now we prove  \ref{it:subadditivityII}.
Recall that $\boldsymbol{G}_k^X$ was defined in \eqref{eq:strong_weight}.
We first observe that the following operator inequality is true for $h^{-2}<\delta$
\begin{align}\label{eq:relationMkGk}
	\delta_k\boldsymbol{M}_k^X\geq \delta_k \sum_{1\leq|\alpha|\leq M} L^{2k(|\alpha|-1)} (\nabla^\ast)^\alpha \boldsymbol{1}_{X^+}\nabla^\alpha
	\geq \delta_k h_k^2 \boldsymbol G_k^{X^+}\geq \boldsymbol G_k^X.                                                                                                                  
\end{align}
This implies \eqref{eq:subadd3} for $k=0$. 
For $k\geq 1$ the monotonicity   of $\boldsymbol{A}_{k-1:k}^X$ in $X$, the positivity of $\boldsymbol{M}_k^X$, 
and the additivity property \eqref{eq:Mkadditive} of  $\boldsymbol{M}_k^X$  imply
	\begin{align}
		\boldsymbol{A}_{k}^{X\cup Y}=\boldsymbol{A}_{k-1:k}^{(X\cup Y)^\ast}+\delta_k \boldsymbol{M}_k^{X\cup Y}\geq 
		\boldsymbol{A}_{k-1:k}^{X^\ast}+\delta_k\boldsymbol{M}_k^X+\delta_k\boldsymbol{M}_k^Y\geq                    
		\boldsymbol{A}_k^X+\boldsymbol{G}_k^Y.                                                               
	\end{align}
	
It remains to prove  \eqref{eq:subadd4}. 
Note that  $\delta_{k+1} h_{k+1}^2=\delta h^2\geq 8$.
Similar to \eqref{eq:relationMkGk} we conclude that for $U\in \mathcal{P}_{k+1}$
	\begin{align}\label{eq:prop4:1}
		\delta_{k+1}\boldsymbol{M}_{k+1}^U\geq \delta_{k+1} \sum_{1\leq|\alpha|\leq M} L^{2(k+1)(|\alpha|-1)} 
		(\nabla^\ast)^\alpha                                                                    
		\boldsymbol{1}_{U^+}\nabla^\alpha\geq \delta_{k+1} h_{k+1}^2                                
		\boldsymbol G_{k+1}^{U^+}\geq 8\boldsymbol G_{k+1}^{U^+}.                                               
	\end{align}
Recall that  by \eqref{eq:pisetincl} we have $X\subset X^\ast\subset U^\ast$ if $U=\pi(X)$.
Together with \eqref{eq:prop4:1} this  implies
	\begin{align}
		\boldsymbol{A}_{k+1}^{U}=\boldsymbol{A}_{k:k+1}^{U^\ast}+\delta_{k+1}\boldsymbol{M}_{k+1}^{U}\geq 
		\boldsymbol{A}_{k:k+1}^X+8\boldsymbol{G}_{k+1}^{U^+}                                      
		\geq \boldsymbol{A}_{k:k+1}^X+2\boldsymbol{G}_{k}^{U^+},                                   
	\end{align}
where we used in the last step that $h_{k+1}^2=4h_k^2$ and therefore $4\boldsymbol{G}_{k+1}^X\geq \boldsymbol{G}_k^X$. 
Note that in the last expression the operation $U^+$ in $\boldsymbol{G}_k^{U^+}$ is still on scale $k+1$.
\end{proof}

\section{Consistency of the weights under \texorpdfstring{$\protect\headingRk$}{Rk+1(q)}}
In this section we prove the necessary bounds that imply the integration property of the weights \ref{w:w7} and \ref{w:w8} in Theorem~\ref{th:weights_final}.
They follow from a Gaussian integration stated in \eqref{eq:GaussianCalculus} with
the operators $\boldsymbol{A}_k^X$ and the covariances $\mathscr{C}_{k+1}^{(\boldsymbol{q})}$.

\begin{lemma}\label{prop:W3}
 Under the same assumptions as in Theorem~\ref{th:weights_final} with $\delta$ and $\lambda$ as  in    
 Lemma~\ref{prop:W1} the operators $\boldsymbol{A}_k^X$ satisfy the following additional properties:
\begin{enumerate}[label=(\roman*),leftmargin=0.7cm]
		\item \label{it:determinant} 
		Let $\rho =(1+\weightzeta)^{\frac13}-1$ and $\overline{\rho}\in  [0, \rho]$. 
		There is a constant $A_{\mathcal{P}}$ depending on $\weightzeta$, and in addition on $L$ if $d=2$, 
		and a constant  $\kappa=\kappa(L)$ such that  for any $k$-polymer $X$ and 
		$\boldsymbol{q}\in B_\kappa=\{\boldsymbol{q}\in \mathbb{R}^{(d\times m)\times (d\times m)}_{\mathrm{sym}}|\; |\boldsymbol{q}|<\kappa\}$  
		the following estimate holds
		      \begin{align}\label{eq:determinant}
		      \det\left(\1-(1+ \overline{\rho})\left(\mathscr{C}_{k+1}^{(\boldsymbol{q})}\right)^{1/2}
		      \boldsymbol{A}_k^X\left(\mathscr{C}_{k+1}^{(\boldsymbol{q})}\right)^ {1/2}\right)^{-1/2}\leq \left(\frac{{A}_{\mathcal{P}}}{2}\right)^{|X|_k}.                            
		      \end{align}
		For blocks $X\in \mathcal{B}_k$ the same estimate holds for a constant $A_{\mathcal{B}}$ which does not depend on $L$.		    
		\item \label{it:integration} 
		Integration property: Let $A_{\mathcal{P}}$ and $\overline{\rho}$ be as in \ref{it:determinant}. 
		Then
		        \begin{align}\label{eq:weightintegral}
		      	\int_{\Xcal_N} e^{\frac{1+\overline{\rho}}{2}(\boldsymbol A_{k}^X(\p+\xi),\p+\xi)}\, \mu_{k+1}^{(\boldsymbol{q})}(\d \xi)\leq 
		      	\left(\frac{{A}_{\mathcal{P}}}{2}\right)^{|X|_k}e^{\frac{1+\overline{\rho}}{2}(\boldsymbol A_{k:k+1}^X\p,\p)}          
		      \end{align}
		for any polymer $X$ and the same bound with $A_{\mathcal{P}}$ replaced by $A_{\mathcal{B}}$ holds for any block $X\in \mathcal{B}_k$.
	\end{enumerate}
\end{lemma}
\begin{proof}
	 
The statement \ref{it:determinant}  can  be proved  similarly to Lemma~ 5.3 in \cite{AKM16}. 
We rely  on the abstract Gaussian calculus sketched at the beginning of this section.
One difficulty is the fact that we need to renormalize the covariance. 
Hence we later need the integration property \ref{it:integration}
not only for $\mu_{k+1}^{(0)}$ but also for $\boldsymbol{q}$ in a small neighbourhood $B_\kappa(0)$. 
As in Section~\ref{sec:FRD} we impose the condition $\kappa\leq \omega_0/2$.	
This condition ensures that the finite range decomposition of the covariance $\mathscr{C}^{(q)}$ 
is defined for $\boldsymbol{q}\in B_\kappa$ under the assumption \eqref{eq:conditionForQ} on $\Qscr$. 
Clearly the left hand side of \eqref{eq:determinant} is decreasing in $\overline{\rho}$. 
Thus we only need to consider $\overline{\rho} = \rho $ in the proof of \ref{it:determinant}. 
	
The first step is to bound the spectrum of the operator 
$(\mathscr{C}_{k+1}^{(\boldsymbol{q})})^{\frac12}\boldsymbol{A}_k^X(\mathscr{C}_{k+1}^{(\boldsymbol{q})})^{\frac12}$. 
This	is a necessary condition for the convergence of the integral in \eqref{eq:weightintegral} and it is also needed to bound the determinant.
We show that the covariance operators $\mathscr{C}_{k+1}^{(\boldsymbol{q})}$ and $\mathscr{C}_{k+1}^{(0)}$ are comparable for small $\boldsymbol{q}$.
Namely, for a sufficiently small neighbourhood $B_\kappa$ of the origin, the inequality ${\mathscr{C}}_{k+1}^{(\boldsymbol{q})}\leq (1+\rho)\mathscr{C}_{k+1}^{(0)}$
holds for $\boldsymbol{q}\in B_\kappa$.
Since both operators are block-diagonal in the Fourier space, it is sufficient to show the estimate for all Fourier modes. 
Indeed, we observe that for $p\in \widehat{T}_N$ and $\boldsymbol{q}$ satisfying $|\boldsymbol{q}|<\omega_0/2$,
the bound \eqref{keyquotientbound} with $\ell=1$ implies 
	\begin{equation} 
		\left|\widehat{\mathcal{C}}_{k+1}^{(\boldsymbol{q})}(p)-\widehat{\mathcal{C}}_{k+1}^{(0)}(p)\right|
		  \leq  \int_{0}^1 \left| \frac{\d}{\d t}\widehat{\mathcal{C}}_{k+1}^{(t\boldsymbol{q})}(p)\right|\,\d t \\
		   \leq |\boldsymbol{q}| K_1 L^{4(d+\tilde{n})+2} \frac{1}{\left|(\widehat{\mathcal{C}}_{k+1}^{(0)}(p))^{-1}\right|}.                                    
	\end{equation}
From this and the bound $\mathrm{Id}/|A^{-1}|\leq A$, we infer that
	\begin{align}
	\widehat{\mathcal{C}}_{k+1}^{(\boldsymbol{q})}(p)-\widehat{\mathcal{C}}_{k+1}^{(0)}(p)
	\leq |\boldsymbol{q}| K_1 L^{4(d+\tilde{n})+2} \widehat{\mathcal{C}}_{k+1}^{(0)}(p).
	\end{align}
The claim now follows for $\boldsymbol{q}\in B_\kappa$  where 
	\begin{align}\label{eq:defofkappa}
	\kappa = \min(\rho L^{-4(d+\tilde{n})-2}/K_1,\omega_0/2).
	\end{align}
Note that here, the lower bounds for the finite range decomposition are essential.
	We can rewrite ${\mathscr{C}}_{k+1}^{(\boldsymbol{q})}\leq (1+\rho)\mathscr{C}_{k+1}^{(0)}$ equivalently as
	 \begin{align}\label{eq:estimateCkqCk0}
	 	({\mathscr{C}}_{k+1}^{(\boldsymbol{q})})^{\frac12}
	 	({\mathscr{C}}_{k+1}^{(0)})^{-1}
	 	({\mathscr{C}}_{k+1}^{(\boldsymbol{q})})^{\frac12}\leq (1+\rho).
	 \end{align}
The constants $\weightzeta_k$ that appear in Lemma \ref{prop:W1} satisfy the inequality  $\weightzeta_k\geq \weightzeta$, 
and we assumed $\weightzeta\in (0,1/4)$. Thus we have $\rho \in (0,1/4)$.
Using this, the bounds  \eqref{eq:akbound} and \eqref{eq:estimateCkqCk0},  for $X\in \mathcal{P}_k$ we estimate
	\begin{align}\begin{split}\label{eq:posofspec}
		(1+\rho)\left(\mathscr{C}_{k+1}^{(\boldsymbol{q})}\right)^{\frac{1}{2}}\boldsymbol{A}_{k}^X
		\left({\mathscr{C}}_{k+1}^{(\boldsymbol{q})}\right)^{\frac{1}{2}}
		  & \leq                                                 
		\frac{(1+\rho)}{1+\weightzeta_k}\left({\mathscr{C}}_{k+1}^{(\boldsymbol{q})}\right)^{\frac{1}{2}}
		\left({\mathscr{C}}_{k+1}^{(0)}\right)^{-1}
		\left({\mathscr{C}}_{k+1}^{(\boldsymbol{q})} \right)^{\frac{1}{2}}(p)\\
		  & \leq \frac{(1+\rho)^2}{(1+\rho)^3}<1-\frac{\rho}{2}. 
		\end{split}\end{align}
Therefore we have shown that the determinant in \eqref{eq:determinant} is non-vanishing.

To complete the proof of \eqref{eq:determinant}, we bound the trace of 
$(\mathscr{C}_{k+1}^{(\boldsymbol{q})})^{\frac12}\boldsymbol{A}_k^X(\mathscr{C}_{k+1}^{(\boldsymbol{q})})^{\frac12}$.
Recall that the operators $\mathscr{C}_{k+1}^{(\boldsymbol{q})}$ and $\boldsymbol{A}_k^X$ can be extended to $\Vcal_N$ so that they annihilate constant fields. 
This extension does not change the trace.
Let $\eta_X:T_N\rightarrow \mathbb{R}$ be a cut-off function such that $\eta_X{\restriction_{X^{++}}}=1$, $\mathrm{supp}(\eta)\subset X^{+++}$, and 
$\eta_X$ satisfies the smoothness estimate 
		\begin{align}\label{eq:derivativeboundcutoff}
		|\nabla^l\eta_X|\leq \Theta L^{-lk}
	\end{align} 
for $l\leq 2M$ where $\Theta$ does not depend on $L$ or $X$. 
We use  $m_{\eta_X}$ to denote the  operator of multiplication by $\eta_X$.
First we note that 
	\begin{align}
		\begin{split}                             
		m_{\eta_X}\boldsymbol{A}_k^Xm_{\eta_X}=\boldsymbol{A}_k^X 
		\end{split}                               
	\end{align}
because $\boldsymbol{A}_k^X$ is self adjoint and  depends only on $\p(x)$ for $x\in X^{++}$.
We observe that, for symmetric operators, the inequality  $A\geq B$ implies that $\mathrm{Tr}\, A\geq \mathrm{Tr}\, B$ which by \eqref{eq:akbound} yields
	\begin{align}
		\begin{split}\label{eq:firsttracestep}
		\tr \left(\mathscr{C}_{k+1}^{(\boldsymbol{q})}\right)^{\frac{1}{2}}\boldsymbol{A}_k^X\left(\mathscr{C}_{k+1}^{(\boldsymbol{q})}\right)^{\frac{1}{2}}
		  & =    
		\tr \left(\mathscr{C}_{k+1}^{(\boldsymbol{q})}\right)^{\frac{1}{2}}
		m_{\eta_X}\boldsymbol{A}_k^Xm_{\eta_X}\left(\mathscr{C}_{k+1}^{(\boldsymbol{q})}\right)^{\frac{1}{2}}\\
		  & \leq 
		\frac{1}{\lambda} 
		\tr \left(\mathscr{C}_{k+1}^{(\boldsymbol{q})}\right)^{\frac{1}{2}}m_{\eta_X}
		\Malt_km_{\eta_X}\left(\mathscr{C}_{k+1}^{(\boldsymbol{q})}\right)^{\frac{1}{2}}\\
		  & =    
		\frac{1}{\lambda} 
		\tr m_{\eta_X}\Malt_k m_{\eta_X}\mathscr{C}_{k+1}^{(\boldsymbol{q})}.
		\end{split}
	\end{align}
Here $\lambda$ is the quantity defined in  \eqref{eq:defoflambda}. 
In particular, $\lambda$ does not depend on $L$.
	
The remaining part of the proof is, up to some minor details, the same as in \cite{AKM16}{Lemma~5.3}.
For the trace calculation we will use the orthonormal basis  $e_x^i(y)=\delta_x^ye^i$ of $\Vcal_N$, where $e^i\in \mathbb{R}^m$ is a standard basis vector. 
Note that 
	\begin{align}
		\left(m_{\eta_X}\mathscr{C}_{k+1}^{(\boldsymbol{q})}e_{x_0}^i\right)(x)=\eta_X(x)\mathcal{C}_{k+1}^{(\boldsymbol{q})}(x-x_0)e^i. 
	\end{align}
For the evaluation of the operator $\Malt_k$, we need the product rule for discrete derivatives that reads
	\begin{align}
		\nabla_i(fg)=\nabla_if S_ig+S_if\nabla_i f 
	\end{align}
where
	\begin{align}
		(S_if)(x)\coloneqq \frac12 f(x)+\frac12 f(x+e_i). 
	\end{align}
The operators $S_i$ commute with discrete derivatives and we use the usual multiindex notation $S^\alpha$ for $\alpha\in \mathbb{N}_0^d$.
This implies that 
	\begin{align}
		\begin{split}
		  & \Malt_k m_{\eta_X}\mathscr{C}_{k+1}^{(\boldsymbol{q})}e_{x_0}^i(\cdot)=       
		  \hspace{-0.2cm}\sum_{|\alpha|\leq M}\hspace{-0.2cm}L^{2k(|\alpha|-1)} 
		\hspace{-0.2cm}\sum_{\substack{\beta_1+\beta_2=\alpha \\ \gamma_1+\gamma_2=\alpha}}\hspace{-0.2cm} 
		K_{\gamma_1, \gamma_2}^{\beta_1, \beta_2}\times\\
		&\hspace{3cm}
		\times(S^{\beta_2})^\ast S^{\gamma_2}(\nabla^{\beta_1})^\ast\nabla^{\gamma_1} \eta(\cdot)
		(S^{\beta_1})^\ast S^{\gamma_1}(\nabla^{\beta_2})^\ast\nabla^{\gamma_2} \mathcal{C}_k^{(\boldsymbol{q})}(\cdot-x_0)e^i
		\end{split}
	\end{align}
	where $K_{\gamma_1, \gamma_2}^{\beta_1, \beta_2}$ is a combinatorial constant.	
	Note that $\lVert S_i\Vert= 1$ where, only here, we use  $\lVert \cdot\rVert$ 
	to denote the operator norm with respect to the maximum norm $\lVert \cdot\rVert_\infty$ on $\Vcal_N$.
	The bound \eqref{eq:discretebounds} for the discrete derivatives of $\mathcal{C}_{k+1}^{(\boldsymbol{q})}$ 
	and the choice of $n\geq 2M$ for the regularity parameter of the finite range decomposition, jointly imply that there is a constant
	$C_M =C_M(L)>0$ such that
	\begin{align}\label{eq:discreteboundagain}
		\sup_{x\in T_N} | \nabla^\alpha\mathcal{C}_{k+1}^{(q)}(x)|\leq C_ML^{-k(d-2+|\alpha|)}\qquad\text{for all $|\alpha|\leq 2M$,} 
	\end{align}
	where $C_M$ is independent of $L$ for $d>2$, but $C_M\propto \ln(L)$ for $d=2$.
	Using this combined with \eqref{eq:derivativeboundcutoff} and \eqref{eq:discreteboundagain}, we get
	\begin{align}
		\begin{split}\label{eq:linftybound}
		  & \lVert \boldsymbol{M}_k m_{\eta_X}\mathscr{C}_{k+1}^{(\boldsymbol{q})}e_{x_0}^i\rVert_{\infty}                              \\
		  & \quad\leq \sum_{|\alpha|\leq M}\hspace{-0.15cm}L^{2k(|\alpha|-1)}\hspace{-0.15cm} 
		\sum_{\substack{\beta_1+\beta_2=\alpha \\ \gamma_1+\gamma_2=\alpha}}\hspace{-0.15cm} K_{\gamma_1, \gamma_2}^{\beta_1, \beta_2}
		\lVert(S^{\beta_2})^\ast
		S^{\gamma_2}(\nabla^{\beta_1})^\ast\nabla^{\gamma_1} \eta\rVert_\infty\times\\
                 &\hspace{6cm}\times
		\lVert (S^{\beta_1})^\ast S^{\gamma_1}(\nabla^{\beta_2})^\ast\nabla^{\gamma_2} \mathcal{C}_{k+1}^{(\boldsymbol{q})}(\cdot-x_0)e^i\rVert_\infty \\
		 & quad\leq                                                                                            
		\sum_{|\alpha|\leq M}L^{2k(|\alpha|-1)}\hspace{-0.15cm}
		\sum_{\substack{\beta_1+\beta_2=\alpha \\ \gamma_1+\gamma_2=\alpha}}\hspace{-0.15cm} K_{\gamma_1, \gamma_2}^{\beta_1, \beta_2} 
		\lVert\nabla^{\beta_1+\gamma_1}\eta\rVert_\infty
		\lVert \nabla^{\beta_2+\gamma_2} \mathcal{C}_{k+1}^{(\boldsymbol{q})}e^i\rVert_{\infty}\\
		  & \quad\leq                                                                                            
		\sum_{|\alpha|\leq M}L^{2k(|\alpha|-1)}\hspace{-0.15cm}
		\sum_{\substack{\beta_1+\beta_2=\alpha \\ \gamma_1+\gamma_2=\alpha}}\hspace{-0.15cm} K_{\gamma_1, \gamma_2}^{\beta_1, \beta_2} 
		\Theta L^{-k(|\beta_1|+|\gamma_1|)}C_M L^{-k(d-2+|\beta_2|+|\gamma_2|)}\\	
		&\quad\leq C_M\Theta K\sum_{l=1}^M L^{2k(l-1)}L^{-k(d-2+2l)}  \leq C_M\Theta KM L^{-kd}=\Omega L^{-kd}.                                                    
		\end{split}
	\end{align}
	Here $K$ is a purely combinatorial constant depending on the $K_{\gamma_1, \gamma_2}^{\beta_1, \beta_2}$.
	
	Using \eqref{eq:firsttracestep} and \eqref{eq:linftybound}, this implies 
	\begin{align}
	\begin{split}
	\label{eq:thirdtracestep}
		\tr \left(\mathscr C_{k+1}^{(\boldsymbol{q})}\right)^{\frac{1}{2}}\boldsymbol A_k^X\left(\mathscr{C}_{k+1}^{(\boldsymbol{q})}\right)^{\frac{1}{2}}
		   &\leq \frac{1}{\lambda}\sum_{x\in \TN}\sum_{i=1}^m (e_x^i,m_{\eta_X}\boldsymbol M_k^{\TN}m_{\eta_X}\mathscr C_{k+1}^{(\boldsymbol{q})}e_x^i) \\
		  &   \leq                                                                                                                                  
		\frac{1}{\lambda}\sum_{x\in \mathrm{supp}(\eta_X)}\sum_{i=1}^m\lVert \boldsymbol M_k^{\TN}m_{\eta_X}\mathscr
		C_{k+1}^{(\boldsymbol{q})}e_x^i\rVert_{\infty}
		  \\ &
		   \leq \frac{\Omega m|X^{+++}| L^{-kd}}{\lambda}                                                                                   
		   \leq \frac{\Omega m(7R+1)^d}{\lambda}|X|_k=\Theta_1|X|_k.                                                                                           
	\end{split}\end{align}
	where $\Theta_1$ depends on $L$ if $d=2$.
	The factor $(7R+1)^d$ arises because $X^{+++}=(X+[-3R,3R]^d)\cap T_N$ for $X\in \mathcal{P}_0$. It could be strengthened to $7^d$ for $k\geq 1$.
	
	The appearance of an $L$-dependent term seems to be only an artefact of the use of a cutoff function.
	Let us show how we can get rid of the $L$-dependence if $X$ is a single block. 
	This shows the second part of the first statement. First, let us consider $0\leq k\leq N-1$.
	Note that by \eqref{eq:discretebounds} there is a constant $C_M'$  independent of $L$ such that
	\begin{align}\label{eq:discreteboundsfinal}
		\sup_{x\in T_N} | \nabla^\alpha\mathcal{C}_{k+1}(x)|\leq C_M'L^{-k(d-2+|\alpha|)}\qquad\text{for all $1\leq|\alpha|\leq 2M$} .
	\end{align}
	Consider the set
	\begin{align}
		T=\bigl\{0, 2L^k, 4L^k \cdots , { (L^{(N-k)}-3)} L^k\bigr\}^d. 	
	\end{align}
	
	Then the blocks $\tau_a(B)$ and $\tau_b(B)$ with  $a,b\in T, a\neq b$, have distance at least $L^k$ for $B\in \mathcal{B}_k$. 
	Therefore we find using properties \ref{it:translation}, \ref{it:monotonicity}, and \ref{it:bound} from 
	Lemma~\ref{prop:W1} and  \ref{it:subadditivityI} from Lemma~\ref{prop:W2}
	\begin{align}\begin{split}
		\tr \mathscr{C}_{k+1}^{\frac12}\boldsymbol A_k^B\mathscr{C}_{k+1}^{\frac12}
		 & =\frac{1}{|T|} \sum_{a\in T} \tr 
		\mathscr{C}_{k+1}^{\frac12}\boldsymbol A_k^{\tau_a(B)}\mathscr{C}_{k+1}^{\frac12}
		\\ &
		 \leq \frac{1}{|T|} \tr                
		\mathscr{C}_{k+1}^{\frac12}\boldsymbol A_k^{\TN}\mathscr{C}_{k+1}^{\frac12}
		 \leq \frac{1}{\lambda|T|} \tr    
		\mathscr{C}_{k+1}^{\frac12}\Malt_k\mathscr{C}_{k+1}^{\frac12}
	\end{split}\end{align}
	This trace is estimated similarly to \eqref{eq:linftybound} using \eqref{eq:discreteboundsfinal} as follows
	\begin{multline}   \label{eq:trace_bound_single_block}
		\frac{1}{\lambda|T|} \tr
		\mathscr{C}_{k+1}^{\frac12}\Malt_k\mathscr{C}_{k+1}^{\frac12}
		= \frac{1}{\lambda|T|} \sum_{x\in T_N}\sum_{i=1}^m (e_x^i, \Malt_k\mathscr{C}_{k+1}e_x^i) 
		 \leq \frac{mL^{Nd}}{\lambda|T|} \lVert\Malt_k\mathcal{C}_{k+1}\rVert_\infty               \\
		 \leq \frac{KmL^{Nd}}{\lambda|T|} \sum_{l=1}^{2M} L^{2l(k-1)}C_M'L^{-k(d-2+2l)}      
		 \leq \frac{C_M'KmL^{Nd}}{\lambda|T|} L^{-kd} ,                                         
	\end{multline}
	where $K$ denotes again a combinatorial constant and none of the constants depends on $L$.
	Using the bound $|T|=(L^{N-k}-1)^d/2^d\geq 4^{-d}L^{(N-k)d}$ we find that there is a constant $\Theta_2$  independent of 
	$L$ such that for all blocks $B\in\mathcal{B}_k$
	\begin{align}
		\tr \mathscr{C}_{k+1}^{\frac12}\boldsymbol A_k^B\mathscr{C}_{k+1}^{\frac12}\leq \Theta_2|B|_k=\Theta_2. 
	\end{align}
Note that for $k=N$ there is only one block and we can use the same argument with $T=\{0\}$.

	The estimate for the determinant is now standard.
	We denote the eigenvalues of the operator 
	$(1+\rho)\left(\mathscr C_{k+1}^{(\boldsymbol{q})}\right)^{\frac{1}{2}}\boldsymbol A_k^X\left(\mathscr C_{k+1}^{(\boldsymbol{q})}\right)^{\frac{1}{2}}$ by
	$\lambda_i$. 
	Recall that $\rho=\rho(\weightzeta)<1/4$ is a constant and the bound  \eqref{eq:posofspec}  on the spectrum implies that $\lambda_i\in [0,1-\rho/2]$.
	Concavity of the logarithm implies $\ln(1-x)\geq  - \frac{\ln(2/\rho)}{1-\rho/2}x$ for $x\in [0,1-\rho/2]$.
	Using this we obtain the bound
	\begin{multline}\label{eq:detbound}
		\ln \det\left(                                                                                                                                    
		\1-(1+\rho)\left(\mathscr C_{k+1}^{(\boldsymbol{q})}\right)^{\frac{1}{2}}\boldsymbol A_k^X\left(\mathscr C_{k+1}^{(\boldsymbol{q})}\right)^                                                
		{\frac{1}{2}}\right)                                                                                                                              
		=
		\sum_i \ln(1-\lambda_i)                                                                                                                           
		\geq \frac{\ln(\rho/2)}{1-\rho/2}\sum_i \lambda_i   \\
		=-\frac{\ln(2/\rho)}{1-\rho/2}(1+\rho)\tr \left(\mathscr C_{k+1}^{(\boldsymbol{q})}\right)^{\frac{1}{2}}\boldsymbol A_k^X
		\left(\mathscr C_{k+1}^{(\boldsymbol{q})}\right)^{\frac{1}{2}}.                                                                                                                                   
	\end{multline}
	From \eqref{eq:thirdtracestep} we conclude that, for $A_{\mathcal{P}}\geq 2\exp(\Theta_1(1+\rho)\ln(2/\rho)/(2(1-\rho/2))$,
	\begin{align}
	\begin{split}
	& \ln \det\left( \1-(1+\rho)\left(\mathscr C_{k+1}^{(\boldsymbol{q})}\right)^{\frac{1}{2}}\boldsymbol A^X_k
		\left(\mathscr C_{k+1}^{(\boldsymbol{q})}\right)^ {\frac{1}{2}}\right)   \\
        &\hspace{2cm} \geq -\frac{\ln(2/\rho)}{1-\rho/2}(1+\rho)\tr \left(\mathscr C_{k+1}^{(\boldsymbol{q})}\right)^{\frac{1}{2}}\boldsymbol A_k^X
        \left(\mathscr C_{k+1}^{(\boldsymbol{q})}\right)^{\frac{1}{2}} \geq -2|X|_k\ln\left(\frac{A_{\mathcal{P}}}{2}\right)                                                                                                                    
		\end{split}
	\end{align}
	which implies the claim \eqref{eq:determinant}. 
	Similarly we find the same statement for blocks $B\in \mathcal{B}_k$  for the constant 
	$A_{\mathcal{B}}\geq 2\exp(\Theta_2(1+\rho)\ln(2/\rho)/(2(1-\rho/2))$ which does not depend on $L$.
	
	The integration property \ref{it:integration} follows directly from Gaussian calculus (which is justified because of \eqref{eq:posofspec})
	and the previous point \ref{it:determinant},
	\begin{align}
		\begin{split}
		\int_{\Xcal_N} & e^{\frac{1+\overline{\rho}}{2}(\boldsymbol A_{k}^X(\p+\xi),\p+\xi)}\, \mu_{k+1}^{(\boldsymbol{q})}(\d \xi) \\
		                & =  \left(\det\, \1-(1+\overline{\rho})\left(\mathscr C_{k+1}^{(\boldsymbol{q})}\right)^{\frac{1}{2}}
		                \boldsymbol A_k^X\left(\mathscr C_{k+1}^{(\boldsymbol{q})}\right)^{\frac{1}{2}}\right)^{-\frac{1}{2}}\times\\
		&\hspace{3cm}\times\exp\left(\frac{1}{2}\left(\p,(((1+\overline{\rho})\boldsymbol A_k^X)^{-1}-\mathscr C_{k+1}^{(\boldsymbol{q})})^{-1}\p\right)\right)\\
		                & \leq \left(\frac{A_{\mathcal{P}}}{2}\right)^{|X|_k}
		\exp\left(\frac{1+\overline{\rho}}{2}(\p,((\boldsymbol A_k^X)^{-1}-(1+\rho)^2\mathscr C_{k+1}^{(0)})^{-1}\p)\right) \\ 
		&\leq	\left(\frac{A_{\mathcal{P}}}{2}\right)^{|X|_k}e^{\frac{1+\overline{\rho}}{2}(\p,\boldsymbol A_{k:k+1}^X\p)},
		\end{split}
	\end{align}
	where we again used the monotonicity of the inversion in combination  with
the bound $(1+\overline{\rho})\mathscr C_{k+1}^{(\boldsymbol{q})}\leq
	(1+\rho)^2\mathscr C_{k+1}^{(0)}\leq (1+\weightzeta)\mathscr
	C_{k+1}^{(0)}$
	for $\boldsymbol{q}\in B_\kappa$. If $X$ is a block we can replace $A_{\mathcal{P}}$ by $A_{\mathcal{B}}$.
\end{proof}
Finally, we prove the property Theorem~\ref{th:weights_final}\ref{w:w9}.
\begin{lemma}\label{prop:W4}
	Under the same assumptions 
	as in Theorem~\ref{th:weights_final} and with $\delta$ and $\lambda$ as 
	 in Lemma~\ref{prop:W1}, the norm for the fields can be bounded in terms of the weights as follows.
	 For any polymer $X\in \mathcal{P}_{k+1}$ the  bound
		      \begin{align}
		      	|\p|_{k+1,X}^2\leq                            
		      	(\p,\boldsymbol{A}_{k+1}^X\p)-(\p,\boldsymbol{A}_{k:k+1}^X\p) 
		      \end{align}
		      holds if $h\geq h_0 = (M'3^{2M'}S /\delta)^{1/2}\coloneqq c_d\delta^{-1/2}$
		      where $M'=\lfloor \frac{d}{2}\rfloor+1$, $S=S(d)$ is the Sobolev
		      constant in Lemma~ \ref{le:Sobolev}, and 
		      $\delta$ is the constant from Lemma~\ref{prop:W1}.
\end{lemma}
\begin{proof}
	This property follows  from the discrete Sobolev inequality stated in the next lemma.

	\begin{lemma}\label{le:Sobolev}
		Let $B_\ell=[0,\ell]^d\cap \mathbb{Z}^d$ and 
		$M'=\lfloor \frac{d}{2}\rfloor+1$. For $f:B_\ell\rightarrow \mathbb{R}$ we define the norm
		\begin{align}
			\lVert f\rVert_{B_\ell,2}=\left(\sum_{x\in B_\ell} |f(x)|^2\right)^{\frac{1}{2}}. 
		\end{align}
		Then the following Sobolev inequality holds for some constant $S(d)>0$ 
		\begin{align}
			\max_{x\in B_\ell}|f(x)|\leq S(d) \ell^{-\frac{d}{2}}\sum_{0\leq |\alpha|\leq M'} \lVert (\ell\nabla)^\alpha f\rVert_2 
		\end{align}
		where we assume that $f$ is defined in a neighbourhood of $B_\ell$ such that all discrete derivatives exist.
	\end{lemma}
	\begin{proof}
	Sobolev already considered such inequalities on lattices, see \cite{Sob40} for a similar statement. 
	Also, a similar claim with $d$ derivatives appeared in   \cite{BGM04}[Proposition~B2] and \cite{BS15I}[Lemma~6.6].
		For the statement above a proof can be found, e.g., in \cite{AKM16}{Appendix A}. 
	\end{proof}
	
	We apply this lemma to the function $\nabla^\alpha\p_i$ for $1\leq |\alpha|\leq p_\Phi=\left\lfloor d/2\right\rfloor + 2$ and the set $B^{\ast}$ for $B\in  \mathcal{B}_{k+1}$. 
	Using that $B\subset B^\ast \subset B^+$ for $B\in \mathcal{B}_{k+1}$ we obtain
	that the side-length of $B^\ast$ is contained in $[L^{k+1},3L^{k+1}]$
	and therefore
	\begin{align}
		\begin{split}  \label{eq:sobolev_applic}
		\max_{x\in B^\ast}|\nabla^\alpha \p_i(x)|^2
		&\leq M'S(d) (L^{k+1})^{-d}\sum_{0\leq |\beta|\leq M'} \lVert (3L^{k+1}\nabla)^\beta\nabla^\alpha\p_i\rVert_{B^\ast,2}^2 \\
		&\leq M'3^{2M'}S(d) L^{-(k+1)d}\sum_{1\leq |\gamma|\leq M} 
		L^{2(|\gamma|-|\alpha|)(k+1)}\left(\nabla^{\gamma}\p_i,\boldsymbol{1}_{B^\ast}\nabla^{\gamma}\p_i\right)                  \\
		&\leq M'3^{2M'}S(d) L^{-(k+1)(d+2|\alpha|-2)}(\p,\boldsymbol{M}_{k+1}^B\p).                                                                        
		\end{split}                                                                                                                    
	\end{align}
	Here we used that $M=M'+\pphi=2\lfloor d/2\rfloor +3$ by \eqref{eq:defofMConstant} and the definition of $\boldsymbol{M}_k^X$ in \eqref{eq:defofMk}. 
	Note that the definition of $\boldsymbol{M}_k^B$ involves the term $\boldsymbol{1}_{B^+}\geq \boldsymbol{1}_{B^\ast}$.
	  Using the definition \eqref{eq:primal_norm} of the primal norm we deduce
	\begin{align}
		\begin{split}
		 \hspace{-1cm}|\p|_{k+1,B}^2 & 
		 =\frac{1}{h_{k+1}^2}\max_{x\in B^\ast}\max_{1\leq i\leq m}\max_{1\leq |\alpha|\leq p_\Phi} L^{(k+1)(d-2+2j)}|\nabla^\alpha\p(x)|^2 \\
		               &  \hspace{1cm}\leq \frac{M'3^{2M'}\Xi(d)}{h_{k+1}^2} (\p,\boldsymbol{M}_{k+1}^B\p)                                       
		                \leq \delta_{k+1}(\p,\boldsymbol{M}_{k+1}^B\p)                                                        
		\end{split}
	\end{align}
	provided that $h^2\geq M'3^{2M'}S(d) /\delta$.
	We can now easily conclude for a general polymer $X\in \Pcal_{k+1}$,
	\begin{align}
		\begin{split}
		 \hspace{-1cm}|\p|_{k+1,X}^2 & =\max_{B\in \mathcal{B}(X)}|\p|_{k+1,B}^2                                            
		 \leq \sum_{B\in \mathcal{B}(X)} \delta_{k+1}(\p,\boldsymbol{M}_{k+1}^B\p)                      
		 \\ & 
		               =\delta_{k+1}(\p,\boldsymbol{M}_{k+1}^X\p)=(\p,\boldsymbol{A}_{k+1}^X\p)-(\p,\boldsymbol{A}_{k:k+1}^X\p). 
		\end{split}
	\end{align}
\end{proof}

\chapter{Estimates for the Basic Operations}\label{sec:normprop}

In this chapter we prove estimates for the basic operations
 in the definition of the RG Map $\boldsymbol{T}_k$ (defined in Chapter~\ref{sec:defRenormTransform}): 
 products, the integration map ${\boldsymbol R}^{(\boldsymbol q)}_{k+1}$ (see \eqref{eq:defRenormMap},
and the extraction map $\Pi_2$ (see \eqref{eq:def_projection}). 
In addition we study the relation of the norms on consecutive scales $k$ and $k+1$. 
In view of the definition \eqref{eq:def_circle_product} of the circle product $\circ$ we essentially only have to deal 
with products of functionals defined on disjoint polymers.
Estimates of such products can also be seen as submultiplicativity properties of the norms defined above.

Here and in the following chapters we assume that for any given $L$ our norms are defined using
 weights as in Theorem~\ref{th:weights_final} with  $\weightzeta$, $M$, $\pphi$, $n$, and $\tilde{n}$ as
 indicated in Chapter \ref{sec:tracking}.
 
\section{Pointwise properties of the norms}  \label{se:pointwise_submult}

Specialising the general properties of norms on Taylor polynomials described in Appendix~\ref{se:norms_polynomials}
to the (injective) tensor  norms defined in   \eqref{eq:tensor_norm}
and the dual norm in 
 \eqref{eq:def_k_X_Tp_norm}
we obtain the following result.

Recall (see Section~\ref{se:polymers}) that $\vec{E}(X^\ast)$ denotes the set of directed edges in the small neighbourhood $X^\ast$ of a polymer  $X \in \mathcal{P}_k$. 
 Lemma~\ref{le:shiftinvvsgradients} states that functional 
$F:\mathcal{P}_k\times\Vcal_N\rightarrow \mathbb{R}$
 is local and shift invariant  if and only if 
 for each $X\in \mathcal{P}_k$
 the map $\p \mapsto F(X,\p)$ is measurable with respect to the $\sigma$-algebra generated by
 $\nabla\p{\restriction_{\vec{E}(X^\ast)}}$. Recall also that we always assume $r_0 \ge 3$. 

\begin{lemma} \label{le:norms_pointwise}
Let $X \in \Pcal_k$,  $F,G \in C^{r_0}(\Vcal_N)$ and assume that $F$ and $G$
are measurable with respect to the $\sigma$-algebra generated by $\nabla\p{\restriction_{\vec{E}}(X^\ast)}$.

Then
\begin{align}  \label{eq:product_estimate_concrete}
|FG|_{k,X, T_\p} \le |F|_{k,X,T_\p} \, |G|_{k,X, T_\p}
\end{align}
and  
 \begin{align}  \label{eq:two_norm_concrete}
  |F|_{k+1,X,T_\p} \le (1 + |\p|_{k+1,X})^3  \big(   |F|_{k+1,X, T_0} + 16 L^{-\frac32 d} \sup_{0 \le t \le 1} |F|_{k,X, T_{t \p}}\big).
\end{align}
\end{lemma}

\begin{proof}
The first inequality will follow from Proposition~\ref{pr:product_estimate_taylor} applied to a certain quotient space.
In the following we will define this quotient space and show that
it is a Banach space on which the Taylor polynomials of $F$ and $G$ act.

We first note that $|{\psi}|_{k,X}=0$ implies that $\nabla\psi{\restriction_{\vec{E}(X^\ast)}}=0$
 and therefore by assumption $F(\p+\psi)=F(\p)$ for $\p\in \Vcal_N$.
Hence,  $F$ and $G$ have the property that 
 $(\tay_\p F)(\dot \p + \dot  \psi) = (\tay_\p F)(\dot \p )$ and   $(\tay_\p G)(\dot \p + \dot  \psi) = (\tay_\p G)(\dot \p)$ for all $\dot \p \in \Vcal_N$ 
 and all $\dot \psi \in \Vcal_N$ with $|\dot \psi|_{k,X} = 0$ (recall that $\tay_\p$ denotes the Taylor polynomial of order $r_0$ at $\p$).
 
This implies that the norms in \eqref{eq:product_estimate_concrete}
are finite (see the remark after 
 \eqref{eq:def_k_X_Tp_norm})
 and that the Taylor polynomials  act on the quotient space $\Vcal_N/\! \! \sim$ and on $\oplus_{r=0}^{r_0} (\Vcal_N/ \! \! \sim)^{\otimes r}$
 where $\p \sim \xi$ if and only if $|\xi - \p|_{k,X} = 0$. Moreover $|\cdot|_{k, X}$ is a norm on this quotient space. 
Thus  the assertion follows from 
Proposition~\ref{pr:product_estimate_taylor}.

Similarly for the second inequality we again use that 
$F$ acts on the quotient space $\Vcal_N/ \! \! \sim$ where $\p \sim \xi$ if $|\p - \xi|_{k,X}=0$. 
Since $|\p|_{k,X}=0\Leftrightarrow|\p|_{k+1,X}=0$ 
both $|\cdot|_{k,X}$ and $|\cdot|_{k+1,X}$ define norms on $\Vcal_N/ \! \! \sim$.
We may thus apply 
the two norm estimate  \eqref{eq:two_norm_new}  in Proposition~\ref{pr:two_norm_new}  with the  norms $|g|_{k,X}$ and $|g|_{k+1,X}$
 and $\overline{r}=2$. 
 It follows directly from the definition of the norms $|g|_{j,X}$ in  \eqref{eq:definition_weights},
   \eqref{eq:tensor_norm} and  \eqref{eq:definition_weights_r} (and the fact that $|\alpha_i| \ge 1$)
 that 
 $$ | g^{(r)}|_{k, X} \le 2^r  L^{-r \frac{d}{2}} \,  |g^{(r)}|_{k+1, X} \quad \forall g^{(r)} \in \Vcal_N^{\otimes r}.$$
Here we used in particular that $h_{k+1}/ h_k = 2$.
 Thus the quantity $\rho^{(3)}$ in Proposition~\eqref{pr:two_norm_new} satisfies
 $$\rho^{(3)} \le 16 L^{-\frac32 d}.$$
 Therefore the   two norm estimate   \eqref{eq:two_norm_new} 
with $\overline{r}=2$
implies \eqref{eq:two_norm_concrete}.
\end{proof}

\begin{lemma}\label{le:submultofsimplenorm}Let $\p \in  \Vcal_N$. Then 
\begin{enumerate}[label=\roman*)]
	\item 
	for any $F_1, F_2\in M(\mathcal{P}_k)$ and any (not necessarily disjoint) $X_1,X_2\in \mathcal{P}_k$, we
	have
	\begin{align}   \label{eq:submultofsimplenorm1}
		|F_1(X_1)F_2(X_2)|_{k,X_1\cup X_2,T_\p}\leq |F_1(X_1)|_{k,X_1,T_\p} \, |F_2(X_2)|_{k,X_2,T_\p};
	\end{align}
	\item 
	for any $F\in M(\mathcal{P}_k)$ and any polymer $X\in\mathcal{P}_k$ the bound 
	\begin{align}\label{eq:Fk+1smFk}
		|F(X)|_{k+1,\pi(X),T_\p}   \leq  |F(X)|_{k,X\cup \pi(X),T_\p }   \leq |F(X)|_{k,X,T_\p} 
	\end{align}
	holds if $L\geq 2^d+R$.
	\end{enumerate}
\end{lemma} 
\begin{proof} In view of \eqref{eq:product_estimate_concrete} (applied with $X = X_1 \cup X_2$, $F =F_1$ and $G= F_2$) the first inequality follows from the bound
	\begin{align}  \label{eq:pointwise_monotone_in_X}
		|F(X)|_{k,X\cup Y,T_\p} \leq|F(X)|_{k,X,T_\p} 
	\end{align}
	which itself is a consequence of the estimate $|{\p}|_{k,X}\leq |{\p}|_{k,X\cup Y}$
	and the definition of the dual norm $|\cdot|_{k,X,T_\p}$ in \eqref{eq:def_k_X_Tp_norm}.
	
	The second inequality in  \eqref{eq:Fk+1smFk} follows from   \eqref{eq:pointwise_monotone_in_X}. To prove the first inequality in 
	 \eqref{eq:Fk+1smFk}
	it is sufficient to show that for any polymer $X\in\mathcal{P}_k$ and any $\p\in \Vcal_N$ the primal norms satisfy the estimate 
		\begin{align} \label{eq:Fk+1smFk_dual}
		|\p|_{k,X\cup \pi(X)}\leq  |\p|_{k+1,\pi(X)} 
	\end{align}
for $L\geq 2^d+R$. 
Note that by \eqref{eq:pisetincl} the condition $L\geq 2^d+R$
	implies that $X^\ast \subset \pi(X)^\ast$. 
This fact and  the bound
	\begin{align}
	h_{k+1}^{-1}L^{(k+1)(\frac{d-2}{2}+|\alpha|)}\geq
	h_k^{-1}L^{k(\frac{d-2}{2}+|\alpha|)}\frac{L}{2}\geq h_k^{-1}L^{k(\frac{d-2}{2}+|\alpha|)}\end{align} 
	for $|\alpha|\geq 1$ imply  \eqref{eq:Fk+1smFk_dual}.
	
\end{proof}

\section{Submultiplicativity of the norms}  \label{se:norms_submult}

\begin{lemma}\label{le:submult}
	Assume that $L\geq 2^{d+3}+16R$ odd, and $h\geq h_0(L)$ where
	$h_0(L)$  is specified in  \eqref{eq:definition_h0} in Theorem~\ref{th:weights_final}.
	Let $K\in M(\mathcal{P}_k)$
	factor at scale $0\leq k\leq N-1$ and let $F, F_1, F_2, F_3\in M(\mathcal{B}_k)$. Then the following
	bounds hold:
\begin{enumerate}[label=(\roman*),leftmargin=0.7cm]
		\item \label{it:norm1}
			For every $X\in \Pcal_k$	
		 \begin{align}
		 \lVert K(X)\rVert_{k,X}\leq \prod_{Y\in \mathcal{C}(X)}
		      \lVert
		      K(Y)\rVert_{k,Y}
		     \\
		      \lVert K(X)\rVert_{k:k+1,X}\leq \prod_{Y\in \mathcal{C}(X)} \lVert
		      K(Y)\rVert_{k:k+1,Y}.
		      \end{align}
		      More generally the same bounds hold for any decomposition
		      $X=\bigcup_i Y_i$ such that the $Y_i$ are strictly disjoint;
		\item \label{it:norm2}
		For every $X,Y\in \mathcal{P}_k$ with $X$ and $Y$ disjoint
		\begin{align}
		\lVert K(Y) F^X\rVert_{k,X\cup Y}\leq \lVert
		      K(Y)\rVert_{k,Y}\vertiii{F}_k^{|X|_k};
		\end{align}
		\item \label{it:norm2b}
		For any polymers $X,Y,Z_1,Z_2\in \mathcal{P}_k$ such that $X\cap Y=\emptyset$,  $Z_1\cap Z_2=\emptyset$, and $Z_1, Z_2\subset
		      \pi(X\cup Y)\cup X\cup Y$
		      \begin{align}
		      \lVert F_1^{Z_1}F_2^{Z_2}F_3^{X}K(Y)\rVert_{k+1,\pi({X\cup Y})}\leq \lVert
		      K(Y)\rVert_{k:k+1,Y}\vertiii{F_1}_k^{|Z_1|_k}\vertiii{F_2}_k^{|Z_2|_k}
		      \vertiii{F_3}_k^{|X|_k};
		      \end{align}
		\item \label{it:norm3}For  $B\in \mathcal{B}_k$ 
		\begin{align}
		\vertiii{\mathbb{1}(B)}_{k,B}=1.
		\end{align}
	\end{enumerate}
\end{lemma}

\begin{proof} 
	The proof is the same as in \cite{AKM16} with the difference that the definition of the weight functions changed.
	The submultiplicativity from Lemma \ref{le:submultofsimplenorm} reduces the proof to the factorisation of the weight functions
	stated  in Theorem~\ref{th:weights_final}\ref{w:w3}.
	Indeed, for \ref{it:norm1} we observe that
	\begin{align}
		\begin{split}
		\lVert K(X)\rVert_{k,X} & =\sup_{\p\in \Vcal_N} \frac{|K(X)|_{k,X,T_\p}}{w_k^X(\p)}=                                    
		\sup_{\p\in \Vcal_N}
		\frac{\left|\prod_{Y\in\mathcal{C}(X)}K(Y)\right|_{k,X,T_\p}}{\prod_{Y\in\mathcal{C}(X)}w_k^Y(\p)}\\
		                        & \leq \sup_{\p\in \Vcal_N} \prod_{Y\in\mathcal{C}(X)}\frac{\left|K(Y)\right|_{k,Y,T_\p}}{w_k^Y(\p)} 
		\leq \prod_{Y\in\mathcal{C}(X)} \lVert K(Y)\rVert_{k,Y}.
		\end{split}
	\end{align}
	The same proof applies for a general decomposition $X=\bigcup_i Y_i$ into strictly disjoint sets $Y_i$. To prove the 
	estimate for the $\lVert\cdot\rVert_{k:k+1,X}$ norm  it suffices to use property   \ref{w:w4} in  
	Theorem~\ref{th:weights_final} instead of property  \ref{w:w3}.

	The proof of \ref{it:norm2} relies on Theorem~\ref{th:weights_final}~\ref{w:w5} which,
	together with   \eqref{eq:submultofsimplenorm1}, implies
	\begin{align}
		\begin{split}
		 \hspace{-1cm}\lVert K(Y)F^X\rVert_{k,X\cup Y} & \leq \sup_{\p\in                                                                                                            
		\Vcal_N}\frac{|K(Y)|_{k,Y,T_\p}\prod_{B\in\mathcal{B}_k(X)}|F(B)|_{k,B,T_\p}}{w_k^{Y}(\p)W_k^X(\p)} \\
		&  \hspace{-2cm}\leq \sup_{\p\in \Vcal_N}\frac{|K(Y)|_{k,Y,T_\p}}{w_k^X(\p)}\prod_{B\in\mathcal{B}_k(X)}\frac{|F(B)|_{k,B,T_\p}}{W_k^B(\p)} 
		      \leq \lVert K\rVert_{k,X}\vertiii{F}^{|X|_k}_k.                                                                             
		\end{split}
	\end{align}
	
	To prove \ref{it:norm2b} we use property \ref{w:w6}  in Theorem~\ref{th:weights_final} and estimate  \eqref{eq:Fk+1smFk}
	to get 
	\begin{align}
		\lVert F_1^{Z_1}&F_2^{Z_2}F_3^{X}K(Y)\rVert_{k+1,\pi(X\cup Y),T_\p} 
		     \leq   \sup_{\p \in \Vcal_N}  \frac{|F_1^{Z_1}F_2^{Z_2}F_3^{X}  K(Y)|_{k,X\cup Y\cup \pi(X\cup Y), T_\p}}{w_{k:k+1}^{X\cup 
		Y}(\p) \left(W_{k}^{\pi(X\cup Y)^+}(\p)\right)^2}.
	\end{align}
	where for $U \in  \mathcal{P}_{k+1}$ the neighbourhood $U^+$ is given by $U^+ = U + [-L^{k+1}, L^{k+1}]^d \cap T_N$, see \eqref{eq:nghbhdscompact}.
	Now  Theorem~\ref{th:weights_final}~\ref{w:w1} implies that $w_{k:k+1}^{X \cup Y} \ge w_{k:k+1}^Y$. 
	Moreover we have $X \subset \pi(X\cup Y)\cup X\cup Y$,  $Z_1 \cup Z_2\subset \pi(X\cup Y)\cup X\cup Y$  and $Z_1 \cap Z_2 = \emptyset$. Thus the factorisation property  \eqref{eq:w4b}  of the strong weight function yields
	\begin{align}
	 \left(W_{k}^{\pi(X\cup Y)^+}(\p)\right)^2 \ge  W_k^{Z_1 \cup Z_2}(\p)  W_k^{X}(\p)  =   W_k^{Z_1}(\p)  W_k^{ Z_2}(\p)  W_k^{X}(\p).
	\end{align}
Together with    \eqref{eq:submultofsimplenorm1} we get 
\begin{align}
\begin{split}
		\lVert F_1^{Z_1}F_2^{Z_2}F_3^{X}K(Y)&\rVert_{k+1,\pi(X\cup Y),T_\p} \\ 
		&  \hspace{-3cm}\leq                                                                                  
		\sup_{\p\in
		\Vcal_N}\hspace{-0.1cm}\frac{\prod_{B\in\mathcal{B}_k(Z_1)}|F_1(B)|_{k,B,T_\p}\prod_{B\in\mathcal{B}_k(Z_2)}|F_2(B)|_{k,B,T_\p}}
		{W_k^{Z_1}(\p)W_k^{Z_2}(\p)    }
		 \frac{\prod_{B\in\mathcal{B}_k(X)}|F_3(B)|_{k,B,T_\p}|K(Y)|_{k,Y,T_\p}}{ W_k^X(\p)   w_{k:k+1}^{Y}(\p)}	\\
		&  \hspace{-3cm}\leq\vertiii{F_1}^{|Z_1|_k} \vertiii{F_2}^{|Z_2|_k}\vertiii{F_3}^{|X|_k} \lVert K(Y)\rVert_{k:k+1,Y}  .         
		\end{split}
	\end{align}
	where we used the definition of the norms in the last inequality. 
\end{proof}

\section{Regularity of the integration map}
The next lemma gives the  bound for the renormalisation maps $\boldsymbol{R}_k^{(\boldsymbol{q})}$. Moreover it states regularity of the renormalisation map with
respect to the parameter $\boldsymbol{q}$. This is one of the major differences compared to \cite{AKM16} where the authors have to deal with a
loss of regularity for the $\boldsymbol{q}$ derivatives. The regularity we obtain here is a consequence of the new finite range decomposition from
Theorem \ref{thm:frd} that was constructed in \cite{Buc18}.
\begin{lemma}\label{le:keyboundRk}
Assume that $L\geq 2^{d+3}+16R$ 	
and let $\AP=\AP(L)\geq 1$, $\kappa=\kappa(L) >0$  be the constants  from Theorem~\ref{th:weights_final}
Then for  $\boldsymbol{q}\in B_\kappa$ and $X\in \mathcal{P}_k$
\begin{align}
\label{eq:bdRkl0}
\lVert(\myR_{k+1}^{(\boldsymbol{q})}K)(X)\rVert_{k:k+1,X}\leq 
{\AP}^{|X|_k}\lVert K(X)\rVert_{k,X}.        
\end{align}
Let $X\in\mathcal{P}_k$ be a polymer such that $\pi(X)\in \Pckp$. Then for  $\ell \ge 1$ and $\boldsymbol{q}\in B_\kappa$
\begin{align}
\label{eq:bdRkl1}
\sup_{|\dot{\boldsymbol{q}}| \le 1} 	\lVert \partial^{\ell}_{\boldsymbol{q}}(\myR_{k+1}^{(\boldsymbol{q})}K)(X)
(\dot{\boldsymbol{q}}, \ldots, \dot{\boldsymbol{q}})\rVert_{k:k+1,X}\leq 
C_{\ell}(L) {\AP}^{|X|_k}\lVert K(X)\rVert_{k,X}.                        
\end{align}
	
The same bounds hold with $\AP$ replaced by $\AB$ if $X\in \mathcal{B}_k$ is a single block. 
\end{lemma}

\begin{remark}
The additional constraint that $\pi(X)\in \Pckp$  is necessary to bound 
derivatives of $\myR_{k+1}^{(\boldsymbol{q})}$ with respect to $\boldsymbol{q}$.
Indeed, Theorem~\ref{prop:finalsmoothness} bounds the derivative for a polymer $X\in \mathcal{P}_k$ in terms of its diameter $D$. But for arbitrary polymers $X$ the diameter can be of order $L^N$ even when $|X|_k=2$ which makes the bound in 
Theorem~\ref{prop:finalsmoothness} useless. We will see later that it sufficient to only have a bound for those $X\in \mathcal{P}_k$ such that $\pi(X)\in \Pckp$.
We expect that Theorem~\ref{prop:finalsmoothness}  is suboptimal and a bound only involving $|X|_k$ instead of the diameter of $X$ holds. Then we could
get rid of the assumption $\pi(X)\in \Pckp$ in  Lemma~\ref{le:keyboundRk}.
\end{remark}

\begin{proof}
We first consider $\ell=0$. Here we argue similar to  \cite{AKM16}.
Since Taylor expansion commutes with convolution we have
\begin{align}
|(\myR_{k+1}^{(\boldsymbol{q})} K)(X)|_{k,X,T_\p} 
\le \int_{\Xcal_N} |K(X)|_{k,X,T_{\p + \xi} } \, \,   \mu_{k+1}^{(\boldsymbol{q})}(\d \xi).
\end{align}
It follows that 
\begin{align}\begin{split}
\lVert (\boldsymbol{R}_{k+1}^{(\boldsymbol{q})}K)(X)\rVert_{k:k+1,X} & \leq \sup_\p 
w_{k:k+1}^{-X}(\p)\int      \left|K(X)\right|_{k,X,T_{\p + \xi}}\,                                       
\mu_{k+1}^{(\boldsymbol{q})}(\d\xi)\\
& \leq \sup_\p w_{k:k+1}^{-X}(\p)\int \lVert K(X)\rVert_{k,X}\; w_k^X(\p+\xi)\,\mu_{k+1}^{(\boldsymbol{q})}(\d\xi) \\
& \leq \left(\frac{\AP}{2}\right)^{|X|_k} \lVert K(X)\rVert_{k,X}                     
\end{split}\end{align}
where we used Theorem~\ref{th:weights_final}~\ref{w:w7}.
 Using  Theorem~\ref{th:weights_final}~\ref{w:w8}, the constant $\AP$ can be replaced by $\AB$ for single blocks.
		
For the derivatives we argue similarly.  First we bound  the diameter of $X$. 
Note that we have $B\cap X\neq \emptyset$ for any block $B\in \mathcal{B}_{k+1}(\pi(X))$ by definition of $\pi$.
This implies $|\pi(X)|_{k+1}\leq |X|_k$. By \eqref{eq:pisetincl}
we have $X^\ast\subset \pi(X)^\ast$. We get using \eqref{eq:distXast}
\begin{align}  \label{eq:integration_bound_diam}
\mathrm{diam}(X^\ast)\leq \mathrm{diam}(\pi(X)^\ast)\leq L^{k+1}|\pi(X)|_{k+1}+2(2^d+R)L^k\leq 2L^{k+1}|X|_k.
\end{align}
	
Next we claim that for $\ell \ge 1$, $p > 1$ and  $D=\mathrm{diam}(X^\ast)$
\begin{align}
\begin{split}  \label{eq:commute_Taylor_diff_Q}
 \hspace{-1cm}\sup_{|\dot{\boldsymbol{q}}| \le 1} &\abs*{ \frac{d^\ell}{dt^{\ell}}_{| t= 0} \boldsymbol{R}_{k+1}^{(\boldsymbol{q} + t  \dot{\boldsymbol{q}})} K(X) }_{k, X, T_\p}
\\ &   \hspace{1cm}\leq C_{p,\ell}(L) (DL^{-k})^{\frac{d \ell }{2}}    \Bigl(\int_{\Xcal_N}   \, |K(X)|_{k,X,T_{\p +\xi}}^p \,  \, \mu_{k+1}^{(\boldsymbol{q})}(d\xi) \Bigr)^{1/p}.
\end{split}	
\end{align}
Indeed, we have
\begin{align}\label{eq:TayRK}
\Big\langle {\textstyle{\tay_\p}}  \frac{d^\ell}{dt^{\ell}}_{| t= 0}   (\boldsymbol{R}_{k+1}^{(\boldsymbol{q} + t \dot{\boldsymbol{q}})} K)(X,\p), g \Big\rangle
=  \frac{d^\ell}{dt^{\ell}}_{| t= 0} \int_{\Xcal_N} \langle \textstyle{\tay_{\p + \xi}} K(X), g \rangle\, \,  \mu_{k+1}^{(\boldsymbol{q} + t \dot{\boldsymbol{q}})}(\d\xi).
\end{align}
Denote the integrand in \eqref{eq:TayRK} by $F(\xi) =  F_{\p, g}(\xi)$ and note that we have
$| F(\xi)| \le |K(X)|_{k,X,T_{\p + \xi}}  \, |g|_{k,X}$.
Passing to absolute values and using Theorem \ref{prop:finalsmoothness} 
with  $\boldsymbol{Q}_1(z)$ being  the generator of the quadratic form  $z \mapsto - (\dot{\boldsymbol q} z^\nabla, z^\nabla)$, we  get 
 \begin{align*}
&   \hspace{-2cm}\,  \left|
\big  \langle  \textstyle{   \tay_\p  } \frac{d^\ell}{dt^{\ell}}_{| t= 0}    \boldsymbol{R}_{k+1}^{(\boldsymbol{q} + t \dot{\boldsymbol{q}})} K(X), g  \big\rangle \right|^p  \\
 \hspace{1cm}\le & \, C_{\ell,p}^p(L) (DL^{-k})^{\frac{d \ell p}{2}}\,   \|  F\|_{L^p(\Xcal, \mu_{k+1}^{(\boldsymbol{q})})}^p \, \lVert\dot{\boldsymbol{q}}\rVert^{p\ell} \\
 \hspace{1cm}\le & \,   C_{\ell,p}^p(L) (DL^{-k})^{\frac{d \ell p}{2}}  \int_{\Xcal_N}   \, |K(X)|_{k,X,T_{\p +\xi}}^p \,  \, \mu_{k+1}^{(\boldsymbol{q})}(d\xi) \,  |g|_{k,X}^p \, \lVert\dot{\boldsymbol{q}}\rVert^{p\ell}.
\end{align*}
Taking the supremum over $g$ with $|g|_{k, X} \le 1$ and over $\dot{\boldsymbol{q}}$ with $\lVert\dot{\boldsymbol{q}}\rVert \le 1$ we get \eqref{eq:commute_Taylor_diff_Q}.

Now   \eqref{eq:integration_bound_diam} implies that $DL^{-k} \le 2 L |X|_k$. 
Using that $x^{d\ell/2}2^{-x}$ is bounded and  \eqref{eq:commute_Taylor_diff_Q} we see that there is another
constant $C_{\ell,p}'(L)$ such that 
\begin{align}
\begin{split}
\sup_{\lVert\dot{\boldsymbol{q}}\rVert \le 1}& \left| \frac{d^\ell}{dt^{\ell}}_{| t= 0} \boldsymbol{R}_{k+1}^{(\boldsymbol{q} + t \dot{\boldsymbol{q}})} K(X) \right|_{k, X, T_\p} \\	
&  \leq C_{\ell,p}'(L)2^{|X|_k}\lVert K(X)\rVert_{k,X}\left(\int  |w_k^X(\p+\xi)|^p\mu_{k+1}^{(\boldsymbol{q})}(\d\xi)\right)^{\frac{1}{p}}.                                                                              
\end{split}
\end{align}
Now we set $p=1+\rho$ where $\rho=(1+\weightzeta)^{1/3}-1$. Then  Theorem~\ref{th:weights_final}~\ref{w:w7} implies 
\begin{align}
\begin{split}
\sup_{\lVert\dot{\boldsymbol{q}}\rVert \le 1} 	\left| \frac{d^\ell}{dt^{\ell}}_{| t= 0} \boldsymbol{R}_{k+1}^{(\boldsymbol{q} + t \dot{\boldsymbol{q}})} K(X) \right|_{k, X, T_\p} 
& \leq C_{\ell}(L)\,\AP^{|X|_k}\lVert K(X)\rVert_{k,X}\; w_{k:k+1}^X(\p).
\end{split}
\end{align}
The conclusion follows by multiplying with $w_{k:k+1}^{-X}(\p)$ and then taking the supremum over $\p$. Again, using Theorem~\ref{th:weights_final}\ref{w:w8} we can replace
$\AP$ by $\AB$ for single blocks.
\end{proof}


\section{The projection \texorpdfstring{$\Pi_2$}{P2} to relevant Hamiltonians}  \label{se:projection_Pi2}
 In this section we introduce the projection $\Pi_2$ to relevant Hamiltonians and prove its 
key properties. The argument is based on a natural duality between relevant monomials in the fields and monomials  on $\Z^d$.
The projection $\Pi_2$ is a very special case of the operator $\loc$ (in fact $\loc_B$)
introduced by Brydges and Slade \cite{BS15II}, except that we do not need to symmetrise between
forward and backward derivatives. Since our situation is much simpler than the general case considered
in \cite{BS15II} we give a self-contained exposition, which follows the strategy  in \cite{BS15II},
  for the convenience of the reader. For $d \le 3$ a more simple-minded proof of the boundedness and contraction 
  properties of $\Pi_2$ was given in Lemma 6.2 and Lemma 7.3  of \cite{AKM16}. This argument can be extended to the case $d > 3$,
  but we prefer to follow the more elegant approach of \cite{BS15II}.
As pointed out in \cite{BS15II}, related questions are discussed in the paper \cite{dBR92} by de Boor and Ron.

Regarding dependencies on the various parameters we recall our convention that we do not indicate dependence 
on the fixed parameters described in Chapter~\ref{sec:tracking}. The parameter $A$ does not enter at all, so we only indicate dependence
on $L$ and $h$.  For the contraction estimate which involves norms on scales $k$ and $k+1$ we use that the ratio $h_{k+1}/ h_k$ is bounded, 
in fact with our choice $h_{k+1}/ h_k = 2$. 
An inspection of the proofs shows that the constants which appear in the rest of this section depend only
on the spatial dimension $d$ the number of components $m$ and the parameter $R = \max(R_0, M)$ where $R_0$ is the range
of the interaction and $M = \pphi + \lfloor d/2 \rfloor + 1 = 2 \lfloor d/2 \rfloor + 3$.

We follow closely the notation of \cite{BS15II}, with the following exception. Since we only deal with forward derivatives
we set 
$$\Ucal = \{e_1, \ldots, e_d\} \simeq \{1, \ldots, d\}$$
 (while in \cite{BS15II} $\Ucal$ is the set $\{ \pm e_1, \ldots, \pm e_d\}$
and we drop various subscripts $+$ which refer to forward derivatives. 

\paragraph{ Relevant monomials in the fields.}
Recall that we declared the following monomials to be relevant in Section~\ref{se:polymers}.
\begin{itemize}[leftmargin=0.4cm]
\item The constant monomial $ \Mscr_{\emptyset}(\{x\})(\p) \equiv 1$;
\item the linear monomials $\Mscr_{i, \alpha}(\{x\})(\p) := \nabla^{i, \alpha} \phi(x) := \nabla^\alpha \phi_i(x)$ \quad for $1 \le |\alpha| \le \lfloor d/2 \rfloor + 1$;
\item the quadratic monomials $ \Mscr_{(i, \alpha), (j, \beta)}(\{x\})(\p) = \nabla^{\alpha}\p_i(x) \, \nabla^\beta \p_j(x)$ \quad for $ |\alpha| = |\beta|= 1$.
\end{itemize}
We introduced the corresponding index sets (recall that $\Ucal = \{e_1, \ldots, e_d\}  \simeq \{1, \ldots, d\}$)
\begin{align}
\begin{split}
\mf v_0 := \{ \emptyset\}, 
\quad \mf v_1 := \{ (i, \alpha) : 1 \le i \le m, \, \alpha \in {\mathbb{N}}_0^{\mathcal U}, 1 \le |\alpha| \le \lfloor d/2\rfloor +1 \},
\\
\mf v_2 := \{ \big(i, \alpha), (j, \beta)\big)   : 1 \le i,j  \le m, \,  \alpha, \beta \in {\mathbb{N}}_0^{\mathcal U},\,  |\alpha|= |\beta| = 1, \, (i, \alpha) \le (j, \beta)  \}.
\end{split}
\end{align}
and 
$\mf v = \mf v_0 \cup \mf v_1 \cup \mf v_2$. 
Here $(i, \alpha) \le (j, \beta)$ refers to any ordering on $\{1, \ldots, m\} \times \{e_1, \ldots, e_d\}$, e.g. lexicographic. 
We use ordered indices to avoid double counting since $ \Mscr_{(i, \alpha), (j, \beta)}(\{x\})(\p) =  \Mscr_{(j, \beta), (i, \alpha)}(\{x\})(\p)$.

In the following we will always consider levels $k$ with
$$ 0 \le k \le N-1.$$
For  a $k$-block $B$ and  $\mpzc \in \mf v$ we define
\begin{align} \Mscr_\mpzc(B) = \sum_{x \in B}  \Mscr_\mpzc(\{x\}).
\end{align}
We denote by $\Rcal = \Rcal_0 \oplus  \Rcal_1 \oplus \Rcal_2 $ the space of relevant Hamiltonians, with
\begin{align}
\Rcal_0 = \R, \quad \Rcal_1 = \Span \{ M_\mpzc(B) : \mpzc \in \mf v_1\}, 
\quad \Rcal_2 = \Span \{ M_\mpzc(B) : \mpzc \in \mf v_2\}.
\end{align}

Given a local functional $K(B)$ we want to extract a 'relevant' part $H= \Pi_2 K(B) \in \Rcal$
in such a way that the functional $K(B) - \Pi_2 K(B)$ measured in the next scale norm $\| \cdot \|_{k+1, B}$
is much smaller than $K(B)$ measured in the $\| \cdot \|_{k,B}$ norm, see Lemma~\ref{le_contraction_I} below.
 This is not true without extraction as can 
be seen by considering the constant functional. In fact we need to gain a factor which is small compared to $L^{-d}$
(to compensate the effect of reblocking which combines $L^d$ blocks on the scale $k$ to a single block on the scale $k+1$)
and for this we need to extract exactly the elements of $\Rcal$. 

We will show that $H = \Pi_2 K(B)$ can be characterised in the following way. 
Let $K^{(0)} + K^{(1)} 
+ K^{(2)} $ denote the second order Taylor polynomial of $K$ at $0$ written as 
a sum of the constant, linear and quadratic part. We will show that there exist unique 
$H^{(i)} \in \Rcal_i$ such that
\begin{align}
H^{(0)} = &\, K^{(0)};   \label{eq:condition_H0}\\
H^{(1)}(\p) = &\,  K^{(1)}(\p) \quad \hbox{for all $\p$ such that  $\p{\restriction_{B^+}}$ is a polynomial of degree $\le \lfloor d/2\rfloor + 1$;}      \label{eq:condition_H1} \\
H^{(2)}(\p) = & \, K^{(2)}(\p)  \quad \hbox{for all $\p$ such that $\p{\restriction_{B^+}}$ is a linear map.}  \label{eq:condition_H2}
\end{align}
Here the large set neighbourhood $B^+$ was defined in  \eqref{E:X+}.
We  then define $H = \Pi_2 K$ by $H= H^{(0)} + H^{(1)} + H^{(2)}$.

We can write this in a more concise notation by using the dual pairing $\langle K, g \rangle_0$ introduced in  
\eqref{eq:pairing_P_g} and \eqref{eq:pairing_F_g}.
Before we do so we note that both $H(\p)$ and $K(B)(\p)$ depend only on values of the field on the set  $B^+$ if $L\geq 2^d+R$
(see Section \ref{se:polymers}).

Since $k \le N-1$ the enlarged block  $B^{++}$ does not wrap around 
the torus $T_N$ for $L\geq 7$ and we can view $B^{++}$ as a subset of $\Z^d$ rather than 
of $T_N$. Note that $\nabla^\alpha \p_i(x)$ for $|\alpha|\leq \pphi$
and $x\in B^\ast$ only depends on $\p{\restriction_{B^{++}}}$ for $L\geq 2^d+R$
since by \eqref{XastsubsetXplus} $B^\ast\subset B^+$.

 We will thus consider in this section the space of fields
\begin{equation}   \label{eq:fields_on_Zd}  \Xcal =(\R^m)^{B^{++}}/ N_{k,B}  
\end{equation}
equipped with the norm 
$$|\p|_{k, B} = \frac{1}{h_k}\sup_{x \in B^*} \sup_{1 \le |\alpha| \le \pphi} \sup_{1 \le i \le m} L^{k |\alpha|}  L^{k \frac{d-2}{2}} |\nabla^\alpha \p_i(x)|
$$
where 
$$ N_{k,B} = \{ \p \in (\R^m)^{B^+} : | \p|_{k, B} = 0\}.$$
Note that $N_{k, B}$ contains in particular the constant functions. 

\paragraph{ Polynomials on $\Z^d$.}
We introduce a convenient basis for polynomials on $\Z^d$ as follows.
For $t \in \Z$ and $k \in \NN$ we define the polynomial
$$ t \mapsto \binom{t}{k} := \frac{     t (t-1) \ldots (t-k+1)    }{k!}    $$
and we extend this by $\binom{t}{0} =1$ and $\binom{t}{k} = 0$ if $k  \in \Z \setminus \NN_0$. Then 
$\nabla \binom{t}{k} = \binom{t}{k-1}$ where $\nabla$ denotes  the one dimensional forward difference
operator.
For a multiindex $\alpha \in \NN_0^{ \{ 1, \ldots, d  \} }$ and $z \in \Z^d$ define
\begin{equation} \label{eq:basic_polynomials}
 b_\alpha(z) = \binom{z_1}{\alpha_1}  \ldots \binom{z_d}{\alpha_d}. 
 \end{equation}
Then
\begin{equation}   \label{eq:duality_basic_polynomials}
 \nabla^\beta b_\alpha = b_{\alpha-\beta}.
\end{equation}
This relation leads to a natural duality between monomials in $\nabla$ and polynomials on $\Z^d$.
Finally we set 
\begin{align}
b_{(i, \alpha)}(z) = b_\alpha(z) e_i,
\end{align}
where $e_1, \ldots e_m$ is the standard basis of $\R^m$, and
\begin{align} \label{eq:def_polynom_tensor_new}
b_{\mpzc} = b_{i, \alpha} \otimes b_{j, \beta}    \quad \hbox{for $\mpzc =((i, \alpha), (j, \beta))$}.
\end{align}

We also define the normalized symmetrised tensor products
\begin{align} \label{eq:def_polynom_tensor}
f_{\mpzc} = N_{\mpzc} b_{\mpzc} = N_\mpzc \frac12 \big(  b_{i, \alpha} \otimes b_{j, \beta}  +  b_{j, \beta} \otimes b_{i, \alpha} \big)
  \quad \hbox{for $\mpzc = ((i, \alpha), (j, \beta))$.}
\end{align}
where 
\begin{align}  \label{eq:tensor_normalisation}  
N_{(i, \alpha), (j,\beta)} = \begin{cases} 1 & \hbox{if $(i, \alpha) = (j, \beta)$,}\\
2 & \hbox{if $(i, \alpha) \ne (j, \beta)$.}
\end{cases}
\end{align}

This agrees with the much more general  definition $N_{\mpzc} = \frac{|\overset{\to}{\Sigma}(\mpzc)|}{|\Sigma_0(\mpzc)|}$.
in (3.9) of \cite{BS15II}. There  $\overset{\to}{\Sigma}(\mpzc)$ denotes  the group of permutation that fix the species and $\overset{\to}{\Sigma}_0$
is the subgroup that fixes $\mpzc = (\mpzc_1, \mpzc_2)$. In our case there is only one species so that  $\overset{\to}{\Sigma}(\mpzc)$ is simply the
group of permutations of two elements and $\overset{\to}{\Sigma}_0(\mpzc) =  \overset{\to}{\Sigma}(m)$ if $\mpzc_1 = \mpzc_2$
and $\overset{\to}{\Sigma}_0 = \{ \mathrm{id} \}$ otherwise.

We now define the subspaces $\Pcal_k \subset {\Xcal}^{\otimes k}$ of (equivalence classes of) functions by
\begin{align}
\Pcal_0 := \R,  \quad \Pcal_1= \Span \{ b_{(i, \alpha)} : (i, \alpha) \in \mf v_1\},  \quad 
\Pcal_2 := \Span \{ f_\mpzc : \mpzc \in \mf v_2\}.
\end{align}
and we set $\Pcal = \Pcal_0 \oplus \Pcal_1 \oplus \Pcal_2$.

\paragraph{ Definition and properties of the projection $\Pi_2$.}
\begin{lemma} \label{le:exists_Pi2}  Let $K \in M(\Pck)$  and let $B$ be a $k$-block. Then there exists one and only one $H \in \Rcal$ such 
that
\begin{align} \label{eq:def_Pi2}
\langle H, g \rangle_0 = \langle K(B), g \rangle_0 \quad \forall g \in \Pcal.
\end{align}
\end{lemma}

We remark in passing that  \eqref{eq:def_Pi2} is equivalent to
\eqref{eq:condition_H0}--\eqref{eq:condition_H2}. For $H^{(0)}$ we simply evaluate at $\p = 0$, for $H^{(1)}$
we use test functions $\p$ such that $\p_{B^+}$ is a polynomial of degree $\lfloor d/2\rfloor +1$.
For $H^{(2)}$
 the implication 
 \eqref{eq:def_Pi2} $\Longrightarrow$   \eqref{eq:condition_H2} follows by taking $g = \p \otimes \p$ for a linear function $\p$. For the converse
 implication one can use polarisation, i.e., the identity $\frac{d}{ds} \frac{d}{dt}_{| s=t=0} (H^{(2)} - K(B))(s b_{i, \alpha} + t b_{j, \beta}) = 0$.

\begin{definition} \label{de:Pi2}
We define  $\Pi_2 K(B) = H$ where $H$ is given by Lemma~\ref{le:exists_Pi2}.
\end{definition}

We now state the main properties of  $\Pi_2$: the maps $\Pi_2$ is bounded on a fixed scale
and $1 - \Pi_2$ is a  contraction  under change of scale.

Recall that on relevant Hamiltonians $H = \sum_{\mpzc \in v} a_{\mpzc}  \Mscr_{\mpzc}(B)$ we defined
in \eqref{hamiltoniannorm}  the norm
\begin{align}  \label{hamiltoniannorm2}
\| H \|_{k,0} = L^{kd} |a_\emptyset| + \sum_{(i, \alpha) \in \mf v_1}
h_k L^{kd} L^{- k \frac{d-2}{2}} L^{-k |\alpha|} |a_{i, \alpha}| + \sum_{\mpzc \in \mf v_2}  
h_k^2 |a_\mpzc|.
\end{align}

\begin{lemma}[Boundedness of $\Pi_2$] \label{le:Pi2_bounded}
There exists a constant $C$ 
such that for $L\geq 2^d+R$ and $0\leq k\leq N-1$
\begin{align} \label{eq:bound_Pi2}
\| \Pi_2 K(B) \|_{k, 0} \le C | K(B)|_{k, B,T_0}.
\end{align}
\end{lemma}

Since $\Pi_2 H = H$ for $H \in \Rcal$, 
Lemma~\ref{le:Pi2_bounded}
shows in particular that
$\| H\|_{k, 0} \le C |H|_{k,T_0} \le C\vertiii{H}_{k,B}$. We can also  prove the converse estimate, in fact a slightly stronger 
result  which will be useful to bound $e^H$ (see Lemma~\ref{le:smoothness_exp} below). 
Define
\begin{align}  \label{eq:norl_ell2_phi}\begin{split}
 | \p|_{k, \ell_2(B)}^2 :&=  \frac{1}{h_k^2}\sup_{(i, \alpha)  \in \mf v_1}  \frac{1}{L^{kd}} \sum_{x \in B} L^{2 k |\alpha| }  L^{k (d-2) }   |\nabla^\alpha \p_i(x)|^2
 \\ &
=  \frac{1}{h_k^2}\sup_{(i, \alpha)  \in \mf v_1} \sum_{x \in B}  L^{2k (|\alpha| - 1)} |\nabla^\alpha \p_i(x)|^2.
\end{split}
\end{align}
Then it follows directly from the definition of $|\p|_{k,B}$ in  \eqref{eq:primal_norm} that 
\begin{align} \label{eq:ell_phi_vs_sup_phi}
 | \p|_{k, \ell_2(B)} \le |\p|_{k,B}.
\end{align}

\begin{lemma}  \label{le:Htp_vs_Hk0} For $H \in M_0(\Bcal_k)$, $L\geq 3$, and $0\leq k\leq N$ we have
\begin{align}  \label{eq:est_HTp}
|H|_{T_\p} \le \big(1 + | \p|_{k, \ell_2(B)}\big)^2 \,    \|H\|_{k,0} \le 2( 1+ | \p |_{k, \ell_2(B)}^2) \, \| H\|_{k,0}
\end{align}
and in particular
\begin{align}  \label{eq:strong_vs_k0_norm}
\vertiii{H}_{k,B} \le 4 \|H \|_{k,0}.    
\end{align}
\end{lemma}

\begin{lemma}[Contraction estimate] \label{le_contraction_I} 
 There exists a constant $C$   
   such that
 for all $L\geq 2^d+R$ and 
$0 \le k \le N-1$
\begin{align}  \label{eq:contraction_T0}
|(1- \Pi_2)K(B)|_{k+1, B, T_0} \le C  L^{-(d/2 + \lfloor d/2 \rfloor +1)} |K(B)|_{k, B, T_0}.
\end{align}
\end{lemma}

\paragraph{Proofs.}
\begin{proof}[Proof of Lemma~\ref{le:exists_Pi2} (existence and uniqueness of $\Pi_2$)]
Clearly $H^{(0)} = K^{(0)} = K(0)$. 

\textit{Step 1:} \quad There exist one and only one    $H^{(2)}  \in \Rcal_2$ such that 
\begin{align}  \label{eq:eqn_for_H2}
\langle H^{(2)}, g \rangle_0 = \langle K(B) , g \rangle_0 \quad \forall g \in \Pcal_2.
\end{align}
Indeed each $H^{(2)} \in \Rcal_2$ is of the form 
$H^{(2)} = \sum_{\mpzc \in \mf v_2} a_{\mpzc} \Mscr_{\mpzc}$. 
Now $\Mscr_{(i, \alpha), (j, \beta)}(B)$ defines a unique symmetric element  
of
$(\Xcal \otimes \Xcal)'$ via (see Lemma~\ref{le:extension_polynomial})
$$ \langle \Mscr_{(i, \alpha), (j, \beta)}(B), \p \otimes \psi\rangle = \frac12 
\sum_{x \in B}  \nabla^\alpha \p_i(x) \nabla^\beta \psi_j(x) +  \nabla^\beta \p_j(x) \nabla^\alpha \psi_i(x).
$$
Thus in view of   \eqref{eq:duality_basic_polynomials},  \eqref{eq:def_polynom_tensor} and \eqref{eq:tensor_normalisation}
we get
\begin{align}
\langle \Mscr_{\mpzc}(B), f_{\mpzc'} \rangle_0 = L^{kd} \delta_{\mpzc \mpzc'} \quad \forall \mpzc, \mpzc' \in \mf v_2. 
\end{align}
It follows that there is one and only one $H^{(2)}$ which satisfies   \eqref{eq:eqn_for_H2}
and the coefficients are given by
\begin{align}  \label{eq:coefficients_H2}
a_{\mpzc} = L^{-dk} \langle K(B), f_{\mpzc} \rangle_0 =L^{-kd}  \langle K^{(2)}, f_{\mpzc} \rangle_0  \quad \forall \mpzc \in \mf v_2.
\end{align}

\textit{Step 2:}   There exists one and only one    $H^{(1)}  \in \Rcal_1$ such that 
\begin{align}  
\langle H^{(1)}, \p \rangle_0 = \langle K(B) , \p \rangle_0 \quad \forall \p \in \Pcal_1.
\end{align}
Writing $H^{(1)} = \sum_{(i, \alpha) \in \mf v_1}  a_{i, \alpha}  \, \Mscr_{i, \alpha}(B)$ 
and testing against the
basis $\{b_{i', \alpha'} : (i', \alpha') \in \mf v_1\}$ of $\Pcal_1$ we see that the condition for $H^{(1)}$ is
equivalent to 
\begin{equation}  \label{eq:equation_coeff_H1}
\sum_{\mpzc \in \mf v_1} B_{\mpzc' \mpzc} \,  a_{\mpzc}= \langle K(B), b_{\mpzc'} \rangle_0  \quad  
\forall \mpzc' \in \mf v_1
\end{equation}
where
\begin{align}  \label{eq:coefficient_upper_triangular}
B_{\mpzc' \mpzc} = \sum_{x \in B}  \nabla^\alpha  b_{\alpha'}(x) \, \delta_{i i'} = \sum_{x \in B} b_{\alpha' - \alpha}(x)  \, \delta_{ii'} 
\quad \hbox{for $\mpzc = (i, \alpha)$, $\mpzc' = (i', \alpha')$.}
\end{align}
In particular
\begin{align}
B_{\mpzc \mpzc} = L^{dk} \qquad \hbox{and} \qquad  B_{\mpzc' \mpzc} = 0 \quad \hbox{if $|\alpha| > |\alpha'|$}.
\end{align}
Thus if we order the indices $(i, \alpha)$ in such a way that $(i, \alpha) < (i, \alpha')$ if $|\alpha| < |\alpha'|$ 
then $B$ is a triangular matrix with entries $L^{dk}$ on the diagonal. Therefore $B$ is invertible
and hence the coefficients of $H^{(1)}$ are uniquely determined.
\end{proof}

\begin{proof}[Proof of Lemma~\ref{le:Pi2_bounded} (boundedness of $\Pi_2$)] We have
\begin{align} \label{eq:bound_H0}
L^{kd} |a_\emptyset| = |H^{(0)}| = |K(0)|.
\end{align}
Since $L\geq 2^d+R$ we can again view $B^{++}$ as a subset of $\mathbb{Z}^d$.
Moreover, since the space of polynomials of a certain degree is invariant by translation we assume without loss of generality
that $0 \in B$. This implies that 
 $$ |b_{i, \alpha}|_{k, B} = \frac{1}{h_k} L^{k \frac{d}{2}} \quad \hbox{if  $|\alpha|=1$}
 \quad \hbox{and thus } 
 \quad |f_{\mpzc}|_{k,B}
 \le 2 \frac{1}{h_k^2} L^{kd}  \quad \forall \mpzc \in \mf v_2.
 $$
 Then  \eqref{eq:coefficients_H2} implies that
 \begin{align*} |a_{\mpzc}| \le L^{-dk}  |K^{(2)}|_{k, B,T_0} \,  | f_{\mpzc}|_{k, B} \le 2 \frac{1}{h_k^2}  |K^{(2)}|_{k, B,T_0} 
 \end{align*}
and therefore
\begin{align} \label{eq:bound_H2}
\sum_{\mpzc \in  \mf v_2}  h_k^2 |a_\mpzc| \le 2 \, \# \mf v_2 \, |K^{(2)}|_{k, B, T_0}.
\end{align}

To estimate the coefficients of $H^{(1)}$ we note that the system
 \eqref{eq:coefficient_upper_triangular} for the coefficients $a_{(i, \alpha)}$ decouples
 for different $i$ since $B_{(i, \alpha)(i', \alpha')} = C_{\alpha \alpha'} \delta_{i i'}$. 
 Hence it is sufficient to prove the estimate for the scalar case $m=1$. 
 It convenient to work in a rescaled basis. 
 Using again that $0 \in B$ we get 
for $|\alpha'| \ge |\alpha|$
$$ \sup_{x \in B^*} |\nabla^\alpha b_{\alpha'}(x)| = \sup_{x \in B^*} |b_{\alpha' - \alpha}(x)|
\le (\diam_\infty B^*)^{|\alpha'| - |\alpha|}$$
and the left hand side vanishes for $|\alpha'| < |\alpha|$. 
Thus
$$ |b_{\alpha'}|_{k, B}  \le \sup_{1 \le |\alpha| \le |\alpha'|} \frac{1}{h_k} L^{k \frac{d-2}{2}} L^{k |\alpha|} (\diam_\infty B^*)^{|\alpha'| - |\alpha|}
\le C' \frac{1}{h_k} L^{k \frac{d-2}{2}} L^{k |\alpha'|}$$
where 
$$ C' :=  ( L^{-k} \diam_\infty B^*)^{\pphi-1} =  ( L^{-k} \diam_\infty B^*)^{\lfloor d/2 \rfloor+1}$$
depends only on $d$ and $R$ (the dependence from $R$ arises from the fact that for $k =0$ we have
$\diam_\infty B^* = 2R+1$).

Now we use the basis of test functions given by
$$ \widetilde b_{\alpha'} = h_k  L^{-k \frac{d-2}{2}} L^{-k |\alpha'|} b_{\alpha'}.$$
Then 
\begin{equation} \label{eq:bound_rescaled_b}
| \widetilde b_{\alpha'}|_{k,B} \le C'.
\end{equation}
We define rescaled coefficients
$$ \widetilde a_\alpha = h_k L^{dk} L^{-k \frac{d-2}{2}} L^{-k |\alpha|}a_\alpha $$
In these new quantities  \eqref{eq:equation_coeff_H1} can be rewritten as
$$\sum_{\alpha \in \mf v_1} A_{\alpha' \alpha} \,  \widetilde a_\alpha = \langle K, \widetilde b_{\alpha'} \rangle.$$
with
\begin{align*}
 A_{\alpha', \alpha}  =  \, 
h_k^{-1}&L^{-dk}  L^{k \frac{d-2}{d}} L^{k |\alpha|} \, \, \, \,  h_k  L^{-k \frac{d-2}{2}} L^{-k |\alpha'|} \,  \, \, \, B_{\alpha' \alpha}
\\ 
\underset{  \eqref{eq:coefficient_upper_triangular}}{=}  
 L&^{-dk} L^{k (|\alpha| - |\alpha'|)} \sum_{x \in B} b_{\alpha'-\alpha}(x).
 \end{align*}
Hence
$$ 
A_{\alpha' \alpha} =\delta_{\alpha' \alpha}  \quad \hbox{if $|\alpha'|= |\alpha|$,} \quad  
|A_{\alpha' \alpha}|\le  \frac{1}{(\alpha'- \alpha)!} \quad \hbox{if $\alpha' - \alpha \in \NN_0^{\{1, \ldots, d\}} \setminus \{0\}$} 
$$
and $A_{\alpha' \alpha} = 0$  if $\alpha'_i <  \alpha_i$ for some $i \in \{ 1, \ldots, d\}.$  
This implies that $(A - \1)^{\lfloor d/2\rfloor + 1} = 0$. Indeed, let $V_\ell := \Span(e_\alpha: |\alpha| \le \ell)$. Then
$A^T- \1$ acts on $V_{\lfloor d/2\rfloor + 1}$ and we have $(A- \1)^T V_l \subset V_{l_1}$ 
and $(A^T - \1)  V_1 = \{0\}$. 
Thus
$$ 
A^{-1} = (\1 + (A - \1))^{-1} = \1 + \sum_{r=1}^{\lfloor d/2 \rfloor} (A- \1)^r.
$$
Since the matrix elements of $A-\1$ are bounded this implies that
$$ |\tilde a_\alpha| \le C\sup_{\alpha' \in \mf v_1} \langle K(B), \tilde b_{\alpha'} \rangle_0 
\underset{ \eqref{eq:bound_rescaled_b}}{\le} 
C C' |K^{(1)}(B)|_{k, B, T_0}.$$
Here $C$ is a combinatorial constant which depends only on the dimension $d$. 
Thus
\begin{equation}
\| H^{(1)} \|_{k, 0} = \sum_{\alpha \in \mf v_1} |\widetilde a_\alpha|  \le 
C  C' \,
 \#\mf v_1\,  |K^{(1)}(B)|_{k, B, T_0}
\end{equation}
in the scalar case $m=1$. For $m >1$ the equations for  the different components $i$ decouple and thus 
the estimate holds with an additional factor $m$. 
Combining this with  \eqref{eq:bound_H0} and  \eqref{eq:bound_H2} we get
$\|H \|_{k, 0} \le 
C
 \sum_{r=0}^2 |K^{(r)}(B)|_{k, B, T_0} \le C|K(B)|_{k, B, T_0}$.
\end{proof}

\begin{proof}[Proof of Lemma~\ref{le:Htp_vs_Hk0}]
The assertion \eqref{eq:strong_vs_k0_norm} follows from \eqref{eq:est_HTp}, the definition of the strong norm in
\eqref{strongnorm} and  \eqref{eq:strong_weight}
as well as  the estimates 
$|\p|_{k, \ell_2(B)}^2 \le \boldsymbol{G}^B_k(\p)$ and $(1+t) \le 2 e^{t/2}$ for $t \ge 0$.

To prove   \eqref{eq:est_HTp}
we use that 
 $|M_\emptyset(\{x\})|_{k,T_0} = 1$ and that by \eqref{eq:estimate_monomial_appendix}
and  \eqref{eq:weight_appendix} we have
\begin{align}  \label{eq:estimate_monomials}
|\Mscr_{i, \alpha}(\{x\})|_{k,B, T_0} \le h_k L^{-k|\alpha|} L^{-k \frac{d-2}{2}} \quad \hbox{and}  \quad
|\Mscr_{\mpzc}(\{x\})|_{k,B, T_0} \le h_k^2  L^{-kd}  \quad \forall \mpzc \in \mf v_2.
\end{align}

Now for $\p = 0$ the estimate   \eqref{eq:est_HTp} follows directly by summing \eqref{eq:estimate_monomials} over $x \in B$. 
For $\p\ne 0$ we use that for the decomposition of $H  = H_0 + H_1 + H_2$ in constant, linear and quadratic terms we get
$$ \textstyle{\tay_\p H} = \tay_0 H + (H_1(\p) + H_2(\p)) + L_\p$$
where $H_1(\p) + H_2(\p)$ is a constant term and $L_\p$ is the linear functional defined 
by $L_\p(\psi) = 2 \overline H_2(\p \otimes \psi)$
 or explicitly by
$$L_\p(\psi) =  \sum_{x \in B}  \, \,   \sum_{(i, \alpha) \le (j, \beta), |\alpha| = |\beta|=1} a_{(i, \alpha), (j, \beta)} \big( 
 \nabla^\alpha \p_i(x)  \nabla^\beta \psi_j(x) +  \nabla^\beta \p_j(x)  \nabla^\alpha \psi_i(x) \big).$$
 Since $\nabla^{\alpha} \psi_i(x) = \Mscr_{i, \alpha}(\{x\})(\psi)$ we get from   \eqref{eq:estimate_monomials} 
 (with $|\alpha| = 1$) and the Cauchy-Schwarz inequality for $\sum_{x \in B}$
 \begin{align}   \notag
  |L_\p|_{k, T_0} \le& \,  \sum_{x \in B}  \, \, \sum_{(i, \alpha) \le (j, \beta), |\alpha| = |\beta|=1} |a_{(i, \alpha), (j, \beta)}| 
\big(  | \nabla^{\alpha} \p_i(x)| + |\nabla^{\beta} \p_j(x)|    \big) \, \,  h_k  L^{-k \frac{d}{2}} \\
\le &  2\,  \sup_{(i, \alpha), |\alpha|=1}  \frac{1}{h_k} \Big( \sum_{x \in B}| \nabla^\alpha \p_i(x)|^2\Big)^{1/2}    \sum_{\mpzc \in \mf v_2} h_k^2   |a_{\mpzc}|  
\label{eq:Htp_vs_Hk0_aux1}
\end{align}
 It follows directly from the definition of $H_2$ and the inequality $|ab| \le \frac12 a^2 +\frac12 b^2$ applied to $\nabla^\alpha \p_i \nabla^\beta \p_j $
 that 
 \begin{align}   \label{eq:Htp_vs_Hk0_aux2}
 |H_2(\p)| \le \sup_{(i, \alpha), |\alpha|=1}  \frac{1}{h_k^2 } \Big( \sum_{x \in B}| \nabla^\alpha \p_i(x)|^2\Big)  \sum_{\mpzc \in \mf v_2} h_k^2   |a_{\mpzc}|.
 \end{align}
 Finally the Cauchy-Schwarz inequality for $\sum_{x \in B}$ gives
 \begin{align}  \label{eq:Htp_vs_Hk0_aux3}
 |H_1(\p)| \le  \sup_{(i, \alpha) \in \mf v_1}  \frac{1}{h_k} \Big( \sum_{x \in B}  L^{2k (|\alpha|-1)} \, | \nabla^\alpha \p_i(x)|^2\Big)^{1/2}
   \hspace{-0.1cm}\sum_{(i, \alpha) \in \mf v_1} h_k L^{k\frac{d}{2}}  L^{-k (|\alpha|-1)}  \,  |a_{(i, \alpha)}|  
 \end{align}
 Now   \eqref{eq:Htp_vs_Hk0_aux1}--\eqref{eq:Htp_vs_Hk0_aux3} imply that
 $$ |\textstyle{\tay_{\p} H - \tay_0} H|_{k, T_0} \le  \big(2 \| \p\|_{k, \ell_2(B)} +  \| \p\|^2_{k, \ell_2(B)}  \big)\, \|H \|_{k,0}.$$  
 Together with the estimate for $\p=0$, i.e.,  $|H|_{k,T_0} \le \|H \|_{k,0}$,  this concludes the proof.
 \end{proof}

\begin{proof}[Proof of  Lemma~\ref{le_contraction_I} (contraction estimate)]
This 
will easily follow from a duality argument given below  and the following result.
\end{proof}

\begin{lemma} \label{le:discrete_taylor_approximation}  There exists a constant $C$   such that for all $L \ge 2^d + R$
\begin{align}   \label{eq:contraction_taylor_remainder}
 \min_{P \in \Pcal_1} |\p - P|_{k, B} \le C L^{-(d/2 + \lfloor d/2 \rfloor +1)} |\p|_{k+1, B}  \quad \forall  \p \in \BX
\end{align}
and 
\begin{align}  \label{eq:contraction_taylor_remainder_g}
\min_{P \in \Pcal_2} |Sg - P|_{k, B} \le C L^{-(d+1)} |g|_{k+1, B}  \quad \forall g \in \BX \otimes \BX.
\end{align}
Here $S$ is the symmetrisation operator, defined by $S (\p \otimes \psi) = \frac12 (\p \otimes \psi) + \frac12 (\psi \otimes \p)$ and
linear extension. 
\end{lemma}

\begin{proof}
Since $L\geq 2^d+R$ we can view $B^{++}$ as a subset of $\Z^d$.
We first show \eqref{eq:contraction_taylor_remainder}. It suffices to consider the scalar case $m=1$ since the estimate
can be done component by component. 
The small set neighbourhood  $B^\ast$ can be written as 
\begin{align}  \label{eq:Bast_discrete_taylor}
B^\ast = a + [0, \rho]^d  \quad \hbox{with}  \quad  L^{-k} \rho \le C
\end{align}
where 
  $C =  
      \max(2R+1, 3)$.
We will apply Lemma~\ref{le:taylor_remainder} for the estimate of the remainder term in the Taylor expansion  with
$$
s := \lfloor d/2 \rfloor + 1  = p_\Phi - 1
$$ 
and 
$$ 
M_s:= M_{s, \rho} = \sup \{ |\nabla^\alpha \p(x)| :  |\alpha| = s+1, \,  x \in \Z^d \cap \big( a + [0, \rho]^d\big)  \}.
 $$

Then it follows from the definition of the field norm $|\p|_{k+1, B}$ that
\begin{align} 
\label{eq:est_contraction_M}
 M_s \le h_{k+1}  L^{-(k+1) (s+1)} L^{-(k+1)  \frac{d-2}{2}}  \,  |\p|_{k+1, B}. 
\end{align}

Let $P = \tay_a^s \p$ be the discrete Taylor polynomial of  order $s$ of $\p$ at $a$.
Then by Lemma~\ref{le:taylor_remainder} we have for $t = |\beta| \le s$
\begin{align}\begin{split} \label{eq:taylor_remainder_norms}
\left| \nabla^\beta[ \p(x) - P(x)] \right| &\le M_s  \binom{|x-a|_1}{s-t+1}  \le M_s (d \rho)^{s+1-t} 
\\ &
\le M_s  C L^{k(s+1-t)}  \quad \hbox{for all $x \in \Z^d \cap \big( a + [0, \rho]^d\big)$.}
\end{split}
\end{align}
{Here $C = C(d,R)$ and}
 we used that  $|x-a|_1 \le d \rho$ 
{as well as  \eqref{eq:Bast_discrete_taylor}.}
 Taking into account that for $|\beta| =s+1$ we have $\nabla^\beta (\p -P) = \nabla^\beta \p$
and using   \eqref{eq:taylor_remainder_norms},  the definition of $M_s$, 
and the fact that $h_{k+1}/ h_k = 2$, we get
\begin{align}
\begin{split}
|\p - P|_{k, B} \le C  \frac{1}{h_k} L^{k \frac{d-2}{2} }  L^{k (s+1)} M_s
&\underset{   \eqref{eq:est_contraction_M}     }{\le} C     L^{-\frac{d-2}{2}} L^{-(s+1)}|\p|_{k+1, B} 
\\
&= C
 L^{-(s + \frac{d}{2})}|\p|_{k+1, B}.
\end{split}
\end{align}
This finishes the proof of  \eqref{eq:contraction_taylor_remainder}.

The proof of the second estimate is similar. We consider the space
\begin{align*}
\widetilde \Pcal_2 :=\Span \{ b_{i,\alpha} \otimes b_{j, \beta} : |\alpha| = |\beta| =1  \}
\end{align*}
Thus $\widetilde \Pcal_2$ is the non symmetrised counterpart of $\Pcal_2$.
In particular $S \widetilde \Pcal_2 =  \Pcal_2$ where $S$ is the symmetrisation operator. 
For $\widetilde P \in \widetilde \Pcal_2$ and $|\alpha|  + |\beta| \ge 3$ 
we have $(\nabla^{i, \alpha} \otimes \nabla^{j, \beta}) \widetilde P = 0$.
Using again that $h_{k+1}/ h_k = 2$ we deduce that
\begin{align} 
\begin{split}
 & \, h_k^{-2} L^{k (|\alpha| + |\beta|)}  L^{k (d-2)}  |(\nabla^{i, \alpha} \otimes \nabla^{j, \beta}) (g- \widetilde P)(x,y)|    \\
&\qquad\quad  \qquad\le  \, 4 L^{- (|\alpha| + |\beta| + d -2)} |g|_{k+1, B} \le 4 L^{-(d+1)} |g|_{k+1, B}   \qquad \hbox{if $|\alpha| + |\beta| \ge 3$.}
   \label{eq:remainder_g_high_r_contraction}
   \end{split}
\end{align}
To prove \eqref{eq:contraction_taylor_remainder_g} it only remains to estimate
 $\nabla^{i, \alpha} \otimes \nabla^{j, \beta} (g- \widetilde  P)$
for $|\alpha| = |\beta| = 1$. 
We define $\widetilde P \in \widetilde \Pcal_2$ by 
\begin{align}
\widetilde P = \sum_{(i', \alpha'), (j', \beta'), |\alpha'| = |\beta'|=1}  (\nabla^{i', \alpha'} \otimes \nabla^{j', \beta'} g)(a,a) \, \, \, b_{i', \alpha'} \otimes b_{j', \beta'}.
\end{align}
Then $\nabla^{i, \alpha} \otimes \nabla^{j, \beta} \widetilde  P = \hbox{const} = (\nabla^{i, \alpha} \otimes \nabla^{j, \beta} g)(a,a)$ for $|\alpha| = |\beta| =1$.

We now define
\begin{align}  \label{eq:define_M_g}
M_2 := \sup \{| (\nabla^{i,\alpha} \otimes \nabla^{j, \beta} g)(x, y)| : |\alpha| \ge 1, \, |\beta| \ge 1,\,  |\alpha| + |\beta| =3, \, x,y \in a + [0, \rho]^d \}.
\end{align}
Then 
\begin{align} \label{eq:bound_M_g}
M_2 \le h_{k+1}^2  L^{-3(k+1)} L^{-(k+1)(d-2)} |g|_{k+1,B}
\end{align}
We claim that for $|\alpha| = |\beta| =1$
\begin{align}    \label{eq:remainder_g_low_r_contraction}
| (\nabla^{i,\alpha} \otimes \nabla^{j,\beta}  g)(x,y) - 
(\underbrace{\nabla^{i,\alpha} \otimes \nabla^{j,\beta}  g)(a,a)}_{=\nabla^{i,\alpha} \otimes \nabla^{j,\beta}  \widetilde P}   |
 \le M_2( |x-a|_1 + |y-a|_1) \le 2 d \rho M_2 
 \end{align}
This estimate is a special case of the Taylor remainder estimate in Lemma 3.5.\ of \cite{BS15II}, but it can also be easily verified as follows. 
For $h: \R^{B^{++}} \times \R^{B^{++}} \to \R$ the difference $h(x,y) - h(a,a)$ can be estimated in $B^\ast \times B^\ast$ 
 by the maximum of the first order forward derivatives
of $h$ in $B^\ast$ times $|x-a|_1 + |y-a|_1$. Now apply this with $h = \nabla^{i,\alpha} \otimes \nabla^{j,\beta}  g$.

Since $\rho \le C L^k$ it follows from  \eqref{eq:remainder_g_low_r_contraction},   \eqref{eq:bound_M_g}
and  \eqref{eq:remainder_g_high_r_contraction} that 
$|g - \widetilde P|_{k, B} \le C L^{-(d+1)} |g|_{k+1, B}$. Application of the symmetrisation operator $S$ does not increase
the norm (see 
Lemma~\ref{le:symmetry_estimate})
and thus
$| Sg - S \widetilde P|_{k, B} \le   C L^{-(d+1)}$. Since $P:= S \widetilde P \subset \Pcal_2$ the assertion
\eqref{eq:contraction_taylor_remainder_g} follows.
\end{proof}

\begin{proof}[Proof of Lemma~\ref{le_contraction_I} (continued)]
It follows from the definition of the norm  $|g |_{j, B}$ for $j \in \{k , k+1\}$ and $g \in \BX^{\otimes r}$ 
in   \eqref{eq:tensor_norm}  and the fact that $h_{k+1}/ h_k = 2$ that
\begin{align} \label{eq:contraction_g_high_r}
|g|_{k, B} \le 8 L^{- \frac32 d} |g|_{k+1, B} \quad \forall g \in \BX^{\otimes r} \quad \forall r \ge 3. 
\end{align}
Since $\Pi_2 K(B)$ depends only on the second order Taylor polynomial of $K$ we get the estimate
\begin{align}\begin{split}  \label{eq:contraction_K_high_r}
|\langle (\1 - \Pi_2) K, g \rangle_0 |&= |\langle K, g \rangle_0| \le |K|_{k, B, T_0} \, | g|_{k,B} 
\\ &
\le 8 L^{-\frac32 d} |K|_{k, B, T_0} \, | g|_{k+1,B} \quad \forall g \in \BX^{\otimes r} \quad \forall r \ge 3.
\end{split}
\end{align}
Now for $\p \in \Xcal$ we have by the definition of $\Pi_2$, the boundedness of $\Pi_2$  and Lemma~\ref{le:discrete_taylor_approximation} 
\begin{align} 
\begin{split}
  |  \langle (\1 - \Pi_2) K(B), \p \rangle_0|& = \min_{P \in \Pcal_1} | \langle (\1 - \Pi_2) K(B), \p - P\rangle_0|   
\\&
 \le  \,  |  (\1 - \Pi_2) K(B)|_{k, B, T_0} \,  \,  \min_{P \in \Pcal_1} |  \p - P|_{k, B}   
 \\ &
\le  \,  C  \,   |K(B)|_{k, B, T_0}  \,  L^{-(d/2 + \lfloor d/2 \rfloor +1)}\,  |\p|_{k+1,B}.     \label{eq:contraction_K_r=1}
\end{split}
\end{align}
Since the pairing  $\langle (\1 - \Pi_2) K(B), g \rangle_0$ depends only on $Sg$ we get similarly for $g \in \Xcal \otimes \Xcal$
\begin{align} 
\begin{split} \label{eq:contraction_K_r=2}
 |  \langle (\1 - \Pi_2) K(B), g \rangle_0| &= \min_{P \in \Pcal_2} | \langle (\1 - \Pi_2) K(B),Sg - P\rangle_0|    
  \\&
\le  \,  |  (\1 - \Pi_2) K(B)|_{k, B, T_0} \,  \,  \min_{P \in \Pcal_2} |  Sg - P|_{k, B}    
\\ &
\le  \,  
C     \,  |K(B)|_{k, B, T_0}  \,  L^{-(d+1)} \, |g|_{k+1,B}.  
\end{split}
\end{align}
The desired assertion follows from 
  \eqref{eq:contraction_K_high_r}-- \eqref{eq:contraction_K_r=2}
and the definition   \eqref{eq:def_k_X_Tp_norm}
 of $|K(B)|_{k+1,B,T_0}$. 
\end{proof}

\chapter{Smoothness of the Renormalisation Map}\label{sec:smoothness}

In this chapter we prove smoothness of the renormalisation map $\myS_k$ formulated  in Theorem 6.7. 
The strategy is to write the  map $\myS_k$ as a composition of simpler maps and  to show smoothness for those maps. 
For this chapter we fix a scale $k$. 
To simplify typography, the functions  $H$ and $K$ without a scale index will in the following denote the quantities on scale $k$ 
while primed $H'$ and $K'$ are the quantities on the next scale $k+1$. 
In final Chapter 12 devoted to the  determination of the renormalized Hamiltonian and dealing with  an application of the implicit function theorem 
to the whole trajectory of relevant and irrelevant interactions $(H_k, K_k)$ we come back to explicit notation with scale indices.

\section{Decomposition of the renormalisation map}

Recall from Section~\ref{se:polymers} that the space of functionals $K \in M(\mathcal P_k)$ that factorise over connected components can be identified
with the space $M(\mathcal P_k^c)$ via the map $\iota_2: M(\mathcal P_k^c) \to M(\mathcal P_k)$ given by $(\iota_2 K)(X) = \prod_{Y \in \mathcal C(X)} K(Y)$. 
We often  do not distinguish between $K$ and $\iota_2K$. 
Similarly the space of functionals $F$  which factorise over $k$-blocks can be identified with the elements of 
$M(\mathcal B_k)$  via  $F^X := (\iota_1 F)(X) := \prod_{B \in \mathcal B_k(X)} F(B)$.

To simplify the notation we introduce the following abbreviations  from \cite{AKM16} for the Banach spaces involved in the decomposition of the map $S_k$:
\begin{align}
\begin{split}\label{eq:def_of_Banach_spaces}
\boldsymbol{M}_{}^{(A)}    & =(M(\mathcal{P}_{k}^c),\lVert\cdot\rVert_{k}^{(A)}),                 \\
{\boldsymbol{M}_{}'}^{(A)} & =(M(\mathcal{P}_{k+1}^c),\lVert\cdot\rVert_{k+1}^{(A)}),         \\
\boldsymbol{M}_{0}                 & =(M(\mathcal{B}_{k}),\lVert\cdot\rVert_{k,0}),                                 \\
\boldsymbol{M}_{|||}               & =(M(\mathcal{B}_k),\vertiii{\cdot}_k),                                                \\
B_{\kappa}          & =\left\{\boldsymbol{q}\in\mathbb{R}^{(d\times m)\times (d\times m)}_{\textrm{sym}}:\, | \boldsymbol{q}|_{\mathrm{op}}<\kappa\right\}. 
\end{split}
\end{align}
Here it is understood that $\boldsymbol{M}_{}^{(A)}$ consists of those elements of $M(\mathcal{P}_{k}^c)$
for which the norm $\lVert\cdot\rVert_{k}^{(A)}$ is finite and similarly for the other spaces.
The prime symbol in ${\boldsymbol{M}_{}'}^{(A)}$ indicates that the scale $k+1$.
The abbreviations ${\boldsymbol M}^{(A)}$ etc.  should  not be confused with the notation for the quadratic forms that appeared in Chapter~\ref{sec:weights}.

We also need a slight modification of the spaces 
$\boldsymbol{M}_{}^{(A)}$ because the renormalisation map does not preserve the factorisation on  the scale $k$, 
i.e., in general
for $K\in M(\mathcal{P}_k^{c})$
\begin{align}\label{eq:non_factorization}
\boldsymbol{R}_{k+1}^{(\boldsymbol{q})}K(X,\p)\neq \prod_{Y\in \mathcal{C}(X)}\boldsymbol{R}_{k+1}^{(\boldsymbol{q})}K(Y,\p)
\end{align}
(here we identified $K$ in $\iota_2 K$). 
In other words, $\boldsymbol{R}_{k+1}^{(\boldsymbol{q})}K$ cannot be identified with an element of $M(\mathcal{P}_k^c)$. 

In \cite{AKM16} this problem is solved by the use of the embedding $M(\mathcal{P}^c_k)\to M(\mathcal{P}_k)$ for $K$. 
This embedding is bounded by the submultiplicativity estimates from Lemma \ref{le:submult}. 
Moreover, the renormalisation map $\boldsymbol{R}_{k+1}^{(\boldsymbol{q})}$ is bounded on 
$M(\mathcal{P}_k)$ by the first part of Lemma~\ref{le:keyboundRk}, i.e., using \eqref{eq:bdRkl0}.
Their approach does not require  derivatives with respect to $\boldsymbol{q}$ of the  renormalisation map $\boldsymbol{R}_{k+1}^{(\boldsymbol{q})}$. 
Instead the authors of \cite{AKM16} are  dealing with a rather cumbersome loss of regularity. 

Here, we work only with differentiable functions.
Then the same approach can no longer be used because  for the concatenation $\boldsymbol{R}_{k+1}^{(\boldsymbol{q})} \iota_2 K$ 
we cannot bound the derivatives with respect to $\boldsymbol{q}$ of the renormalisation map
 $\boldsymbol{R}_{k+1}^{(\boldsymbol{q})}$. Indeed, our bound \eqref{eq:bdRkl1}
 in Lemma~\ref{le:keyboundRk} only applies to polymers in $\mathcal{P}_k$ satisfying $\pi(X)\in \Pckp$. 
 Therefore, we need to introduce an intermediate space between $M(\mathcal{P}_k^c)$ and $M(\mathcal{P}_k) $ which will be called $M(\Pclk)$. 
 This space is smaller than $M(\mathcal{P}_k) $ so that we can bound the derivatives of $\boldsymbol{R}_{k+1}^{(\boldsymbol{q})}K$ for $K\in M(\Pclk)$
  but larger than $M(\mathcal{P}_k^c)$ so that $\boldsymbol{R}_{k+1}^{(\boldsymbol{q})}$  factorizes over suitable decompositions of polymers $X\in \mathcal{P}_k$.
Roughly the space $M(\Pclk)$ consists of functionals that live on the scale $k$ but factor only on the  scale $k+1$.

 More precisely, we use the following definition. 
Recall the definition of the map $\pi: \mathcal P_k \to \mathcal P_{k+1}$ in \eqref{eq:defofpi} and  \eqref{eq:pifactor}.
\begin{definition} We say that $X \in \mathcal P_k \setminus \emptyset$ is a cluster,  $X\in\Pclk$,  if $\pi(X) \in \Pckp$. 
For $X \in \Pcal_k$, $Y\subset X$ is a cluster of $X$ if there is $U \in \mathcal C_{k+1}(\pi(X))$ 
such that
\begin{equation}  \label{eq:formula_cprime_components2}
Y = \bigcup_{Z \in \mathcal C(X): \pi(Z) \subset U}  Z 
\end{equation}
We use $\Ccl(X)$ to denote the set of all clusters of $X$.
\end{definition}

\begin{lemma}  \label{le:factorise_cprime} Assume that $L \ge 2^{d+2}+4R$. Let $X \in \mathcal P_k \setminus \emptyset$. 
Then
\begin{enumerate}[label=(\roman*),leftmargin=0.7cm]
\item For any $U \in \Ccal_{k+1}(\pi(X))$, there is a cluster $Y$ of $X$, $Y\in \Ccl(X)$, such that $\pi(Y)=U$;
\item $ X = \bigcup_{Y\in \Ccl(X)} Y$;
\item Two clusters of $Y_1, Y_2 \in \Ccl(X)$ are either identical or
strictly disjoint on scale $k$;
\item \label{it:factorise_cprime4} $\sum_{Y\in \Ccl(X)} |\Ccal(Y)| = |\Ccal(X)|$;
\item
If $K\in M(\Pcal_k)$ factors over connected components on the scale $k$ 
then
\begin{align}  \label{eq:factorise_next_scale}
	(\boldsymbol{R}_{k+1}^{(\boldsymbol{q})}K)(X,\p)=\prod_{Y\in \Ccl(X)} (\boldsymbol{R}_{k+1}^{(\boldsymbol{q})}K)(Y,\p). 
\end{align}
\end{enumerate}
\end{lemma}

\begin{proof}
Let $X \in \mathcal P_k$ and   $ U = \pi(X)$.
 By definition \eqref{eq:pifactor} of $\pi$ we have
 \begin{equation} \label{eq:cprime_surjective}  U = \bigcup_{Z \in \mathcal C(X)} \pi(Z).
 \end{equation}
Note first that a component of $X$ cannot be shared between two components of $U$:
\begin{equation} \label{eq:pi_connected_to_connected}
Z \in \Pck \text{ implies that }  \pi(Z) \in \Pckp.
\end{equation}
Indeed, if $Z \in \mathcal S_k$ then $\pi(Z)$ is a single block and hence connected. 
If $Z \in \Pck \setminus \mathcal S_k$ then $\pi(Z) = \overline Z$ and, in particular,
$Z \subset \pi(Z)$ and every block $B\in \Bcal_{k+1}(\pi(Z))$ contains at least one point from $Z$. 
For any two points $x,y\in  \pi(Z)$ consider $x'\in Z\cap B_x$, where  $B_x\in\Bcal_{k+1}(\pi(Z))$ is the block that contains the point $x$
and similarly $y'\in Z\cap B_y$. Given that $Z$  as well as any block are connected, there exist a path joining $x$ with $y$ via $x'$ and $y'$.
 
Thus, in view of  \eqref{eq:cprime_surjective} and the fact that a connected set cannot be contained in a union of two nonempty disjoint sets, we get
 \begin{equation} \label{eq:pi_of_component}
 Z \in \mathcal C(X) \text{ implies that }  \pi(Z) \text{ is contained in one  component of $U$ (on scale $k+1$)}
 \end{equation}
 For a connected component $U_1\in \Ccal(U)$ we consider the corresponding cluster $Y_1$ defined by \eqref{eq:formula_cprime_components2}, i.e., 
\begin{align}
Y_1=\bigcup_{Z\in \mathcal{C}(X):\,\pi(Z)\subset U_1} Z.
\end{align}
Then \eqref{eq:cprime_surjective} and \eqref{eq:pi_of_component} jointly
imply $\pi(Y_1)=U_1$ thus proving the first claim. Moreover, \eqref{eq:pi_of_component} also implies the second claim.
To prove the third  claim let $U_1$ and $U_2$ be two different components of $\pi(X)$. 
Again by \eqref{eq:pi_of_component}.
 the corresponding clusters $Y_1$ and $Y_2$ defined by  \eqref{eq:formula_cprime_components2}
 are disjoint. Since $Y_1$ and $Y_2$ are unions of $k$-components of $X$ they must by strictly
 disjoint on scale $k$. 
 
 The fourth claim follows now from the fact that clusters are union of distinct elements of $\Ccal(X)$. 
 
To prove the last claim,  it is sufficient to show that for different components $U_1$ and $U_2$ of 
$U = \pi(X)$ with the corresponding clusters $Y_1\subset U_1$ and $Y_2\subset U_2$,
the fields $\nabla\xi_{k+1}{\restriction_{Y_1^\ast}}$ and $\nabla\xi_{k+1}{\restriction_{Y_2^\ast}}$  are independent if 
$\xi_{k+1}$ is distributed according to  $\mu_{k+1}$.
Note that by \eqref{eq:pisetincl} $Y_{i}^\ast\subset {U_i}^\ast$ and
by \eqref{eq:distU1U2} $\mathrm{dist}(U_1^\ast,U_2^\ast)\geq \frac{L^{k+1}}{2}$
for $L\geq
2^{d+2}+4R$
which implies the independence of the gradient fields.
Therefore we find for any polymer $X\in \Pck$ and $K\in M(\Pclk)$
the identity  \eqref{eq:factorise_next_scale}.
\end{proof}

The space of functionals that factorise over clusters  can again be identified with
the space $M(\Pclk)$. Now we need to equip this space with a norm. It turns out that we need norms that involve, 
in addition to the  parameter $A$ that regulates the growth
depending on the number of blocks, another parameter $B$ that regulates the growth depending on the number of connected components of the polymer.
For $K\in M(\Pclk)$ and $A,B>1$, we define
\begin{align}
	\lVert K\rVert_{k}^{(A,B)}=\sup_{X\in \Pclk} A^{|X|_k}B^{|\mathcal{C}(X)|}\lVert K(X)\rVert_{k,X}. 
\end{align}
We also consider the norm $\lVert\cdot\rVert_{k:k+1}^{(A,B)}$ obtained by replacing, on the right hand side above, 
the norm $\lVert\cdot\rVert_{k,X}$  by the norm $\lVert\cdot\rVert_{k:k+1,X}$. 

Again we introduce abbreviations  for the corresponding normed spaces
\begin{align}
\begin{split}
\widehat{\boldsymbol{M}}^{(A,B)}     & =\{M(\Pclk),\lVert\cdot\rVert_{k}^{(A,B)}\},     \\
\widehat{\boldsymbol{M}}_{:}^{(A,B)} & =\{M(\Pclk),\lVert\cdot\rVert_{k:k+1}^{(A,B)}\} .
\end{split}
\end{align}

Recall  the definition of $K'=\myS_k(H,K,\boldsymbol q)$ in  \eqref{eq:define_Sk} and \eqref{eq:defofKk+1}: for   $U\in \mathcal{P}_{k+1}$ we have
\begin{align}
	K'(U,\p)=\sum_{X\in \mathcal{P}_k} \chi(X,U) \tilde{I}^{U\setminus X}(\p)
	\tilde{I}^{-X\setminus U}(\p)\int_{\Xcal_N}(\tilde{J}(\p)\circ 
	P(\p+\xi))(X)\,\mu_{k+1}^{(\boldsymbol q)}(\d\xi),                                                                                     
\end{align}
where $\tilde{I}=e^{-\tilde{H}}$, $\tilde{J}=1-\tilde{I}$,  $I=e^{-H}$, $P=(I-1)\circ K$,  and \\

\noindent
$\tilde{H}(B,\p)=(\Pi_2\boldsymbol{R}_{k+1}^{(\boldsymbol{q})}H)(B,\p)-(\Pi_2\boldsymbol{R}_{k+1}^{(\boldsymbol{q})}K)(B,\p)$.
\smallskip

Using first the definition of the circle product $\circ$  and then the factorisation property
 \eqref{eq:factorise_next_scale} (and \eqref{eq:circfactor} for $P_2=(I-1)\circ K$ to verify its assumption)  we get
 \begin{align}
 \begin{split}
&K'(U,\p)=\sum_{\substack{X_1,X_2\in \mathcal{P}_k             \\ 
X_1\cap X_2=\emptyset}} \chi(X_1 \cup X_2 ,U)\tilde{I}^{U\setminus (X_1 \cup X_2)}\tilde{I}^{-X_1 \cup X_2\setminus U}(\p) 
\tilde{J}^{X_1}(\p) \times  \\[-0.6cm]
& \hspace{8cm}\times\int_{\Xcal_N} P(X_2, \p+\xi)  \,\mu_{k+1}^{(\boldsymbol q)}(\d\xi)     \\
= &    \sum_{\substack{X_1,X_2\in \mathcal{P}_k             \\ X_1\cap X_2=\emptyset}}
\chi(X_1 \cup X_2 ,U)\tilde{I}^{U\setminus (X_1 \cup X_2)}\tilde{I}^{-X_1 \cup X_2\setminus U}(\p)        
 \tilde{J}^{X_1}(\p) 
   \prod_{Y\in \Ccl(X_2)}   (\myR_{k+1}^{(\boldsymbol{q})}P)(Y,\p)                                                                
 \end{split}
\end{align}
\bigskip

It is now easy to see that the map $\myS_k$ can be  rewritten as a composition of the following maps.
The exponential map 
\begin{align}    \label{eq:exponential_map_def} \begin{split}
E:\boldsymbol{M}_{0}\rightarrow \boldsymbol{M}_{|||},\quad E(H)=\exp(H),
\end{split}                                               
\end{align}
three polynomial maps
\begin{align}
\begin{split}
& P_1:\boldsymbol{M}_{|||}\times \boldsymbol{M}_{|||}\times \boldsymbol{M}_{|||}\times \widehat{\boldsymbol{M}}_{:}^{(A/(2\AP),B)}   \rightarrow {\boldsymbol{M}_{}'}^{(A)},\\
& P_1(I_1,I_2,J,K)(U,\p)\\	  
&  =\sum_{\substack{X_1,X_2\in \mathcal{P}_k  \\ X_1\cap X_2=\emptyset}} \chi(X_1\cup X_2,U)
I_1^{U\setminus (X_1\cup X_2)}(\p)I_2^{(X_1\cup X_2)\setminus U}J^{X_1}(\p) \prod_{Y \in \Ccl(X_2)} K(Y,\p),\\
& P_2:\boldsymbol{M}_{|||}\times \boldsymbol{M}_{}^{(A)}\rightarrow \boldsymbol{M}_{}^{(A/2)},\quad  P_2(I,K)=(I-1)\circ K,\\
& P_3:\boldsymbol{M}^{(A/2)}\rightarrow \widehat{\boldsymbol{M}}^{(A/2,B)},\quad P_3K(X,\p)=\prod_{Y\in \mathcal{C}(X)} K(Y,\p), 
\end{split}
\end{align}
and, finally,  two maps which
include an
integration with respect to $\mu_{k+1}^{(\boldsymbol q)}$. This is the point where regularity is lost for derivatives
in $\boldsymbol q$ direction if the original 
finite range decomposition from \cite{AKM13} is used.
These maps are given by
\begin{align}\begin{split}\label{defs1}
	  & R_1:\widehat{\boldsymbol{M}}_{}^{(A/2,B)}\times B_{\kappa}\rightarrow \widehat{\boldsymbol{M}}_{:}^{(A/(2\AP),B)} ,\\
	  & R_1(P,\boldsymbol{q})(X,\p)=(\myR^{(\boldsymbol{q})}_{k+1}P)(X,\p)=\int_{\Xcal_N}P(X,\p+\xi)\,\mu_{k+1}^{(\boldsymbol{q})}(\d\xi)   
	\end{split}\end{align}
	and
		\begin{align}
	\begin{split}\label{defs2}
	  & R_2:\boldsymbol{M}_{0}\times\boldsymbol{M}_{}^{(A)}\times B_{\kappa}\rightarrow \boldsymbol{M}_{0} ,                       \\
	  & R_2(H,K,\boldsymbol{q})(B,\p)=\Pi_2\left((\myR_{k+1}^{(\boldsymbol{q})}H)(B,\p)-(\myR_{k+1}^{(\boldsymbol{q})}K)(B,\p)\right).                   
	\end{split}
\end{align}
Here, the constant $\AP$ is as specified in Theorem~\ref{th:weights_final}~\ref{w:w7}.
In terms of these maps the map $\myS_k$ can be expressed  as 
\begin{multline}
	\label{eq:decomposition_S}
	\myS_k(H,K,\boldsymbol{q})  = \\
	P_1\Big(E\big(-R_2(H,K,\boldsymbol{q})\big),
	E\big(R_2(H,K,\boldsymbol{q})\big),1-E\big(-R_2(H,K,\boldsymbol{q})\big),R_1\big(P_3\big(P_2(E(-H),K)\big),\boldsymbol{q}\big)\Big).
	\end{multline}
Note that when we insert  in the arguments $I_1$ and $I_2$ of $P_1$ we find $I_1=I_2^{-1}$. 
Since the inversion is not continuous for the strong norm we  have to introduce the two terms as different arguments of $P_1$. 
They are,  however,  Hequal to $E(H)$ and $E(-H)$ for some $H$  and we clearly  have $\lVert H\rVert_{k,0}=\lVert -H\rVert_{k,0}$.

Compared to \cite{AKM16} the smoothness estimates for $R_1$ and $R_2$ change. Actually
they become much simpler because the bulk of the work has been done in \cite{Buc18}.
The estimate for $P_1$ changes slightly because of the slight changes in the combinatorics.
The proof for the smoothness of $E$ has been simplified.
The remaining smoothness estimates are very similar.

To control the polynomial maps $P_2$ and $P_3$ we will use the assumptions on 
$L$ and $h$ in Lemma~\ref{le:submult}, i.e., 
\begin{equation} \label{eq:hL_for_products}
L \ge 2^{d+3}+16R, \quad h \ge h_0(L),
\end{equation}
where $h_0(L)$ is as in \eqref{eq:definition_h0}.
For $P_1$ we need a slightly stronger assumption for $L$
\begin{align}\label{eq:hL_for_P1}
L \ge \max(2^{d+3}+16R, 4d(2^d+R)), \quad h \ge h_0(L).
\end{align}
For the maps $R_1$ and $R_2$ we use the assumption
\begin{align}  \label{eq:L_for_R1R2}
\begin{split}
L \ge  2^{d+3}+16R, 
\end{split}
\end{align}
in Lemma~\ref{le:keyboundRk}.
Finally, for the map $E$ we use the weaker condition
\begin{equation} \label{eq:L_for_map_E}
L \ge {3}.
\end{equation}

\section{The immersion \texorpdfstring{$E$}{E}}

\begin{lemma}  \label{le:smoothness_exp} 
Assume  \eqref{eq:L_for_map_E}.
Then the  map 
$$E : B_{\frac18}(M_0(\Bcal_k), \| \cdot \|_{k,0}) \to (M(\Bcal_k),       \vertiii{ \cdot }_{k,B}) \qquad \hbox{defined by $E(H) = e^{H}$}
$$
 is smooth
and the $r$-th derivative  (viewed as a  map from $B_{\frac18}(M_0(\mathcal B_k))$ to the set of  $r$-multilinear forms on  
$M_0(\Bcal_k)$ with values in  $M(\Bcal_k)$) is uniformly bounded.  More precisely if we set 
$$ \| D^r E(H)\| 
:= \sup \bigl\{    \opnorm[\big]{D^r E(H)(\dot H_1, \ldots, \dot H_r)}_{k, B} :  \| \dot H_i\|_{k,0} \le 1 \hbox{ for $i=1, \ldots, r$} \bigr\}$$
and
$$ C_r := 2^r e^{\frac14}  \max_{t \ge 0} e^{-\frac{t}{4}} (1+t)^{r},$$
then
$$ D^r E(H)(\dot H_1, \ldots, \dot H_r) = e^{H} \dot H_1 \ldots \dot H_r$$
and 
\begin{equation} \label{eq:bound_D_e_to_H}
 \| D^r E(H)\|  \le C_r\    \text{ for any } \  H \in B_{\frac18}(M_0, \| \cdot \|_{k,0}).
 \end{equation}
  Moreover, 
 \begin{equation} \label{eq:bound_exp_H}
 \opnorm[\big]{e^{H}- 1 }_{k, B} \le 8 \|H\|_{k,0}\  \text{ for any }\  H \in B_{\frac18}(M_0, \| \cdot \|_{k,0}).
 \end{equation}
\end{lemma}

\begin{proof}
We first recall some notation.
In   \eqref{eq:norl_ell2_phi} we defined the (semi)norm on fields
\begin{align}
 | \p|_{k, \ell_2(B)}^2 
=  \frac{1}{h_k^2}\sup_{(i, \alpha)  \in \mf v_1} \sum_{x \in B}  L^{2k (|\alpha| - 1)} |\nabla^\alpha \p_i(x)|^2.
\end{align}
Since   \eqref{eq:L_for_map_E} holds we can apply  Lemma~\ref{le:Htp_vs_Hk0} guaranteeing that
\begin{align}  \label{eq:est_HTp_bis}
|H|_{k,B,T_\p} \le 2( 1+ | \p |_{k, \ell_2(B)}^2) \, \| H\|_{k,0}  \text{ for all }  H \in M_0(\Bcal_k).
\end{align}
The strong norm $   \vertiii{ \cdot }_{k,B}$ is defined using the weight $W^B_k = e^{\frac12 (\p,\boldsymbol{G}_k^B \p)}$ where
\begin{align}
 (\p, \boldsymbol{G}_k^B \p) \underset{ \eqref{eq:strong_weight}}{=}  
 \frac{1}{h_k^2} \sum_{1\leq |\alpha|\leq \lfloor d/2 \rfloor +1}L^{2k(|\alpha|-1)}(\nabla^\alpha \p,\mathbb{1}_B      \nabla^\alpha\p)  \ge  |\p|^2_{k, \ell_2(B)}.
\end{align}
Thus
\begin{align}  \label{eq:immersion_bound_strong_norm}
  \vertiii{ F }_{k,B}
  \underset{\eqref{strongnorm},  \eqref{eq:strong_weight}}{=} 
  \sup_\p e^{-\frac12  (\p, \boldsymbol{G}_k^B \p)} |F|_{k,B, T_\p} \le \sup_\p e^{-\frac12 |\p|^2_{k, \ell_2(B)}} |F|_{k,B,T_\p}.
\end{align}

To prove the differentiability  we argue by induction. The main point is to show that
\begin{equation}  \label{eq:proof_diff_e_to_H}
\lim_{\dot H \to 0}  \frac{1}{\| \dot H\|_{k,0}}   \,  \sup_{ \|\dot H_i\|_{k,0} \le 1} 
\opnorm[\Big]{\underbrace{(e^{H + \dot H} -e^{H} - e^{H} \dot H) \dot H_1 \ldots \dot H_r}_{ =: f(\dot H)}}_{k,B}
 = 0.
\end{equation} 
We have 
$$f(\dot H) = e^{H} ( e^{\dot H} - 1 - \dot H)  \dot H_1 \ldots \dot H_r$$
In the following we assume, without loss of generality,  that
$$\| \dot H \|_{k,0} \le \frac1{16}.$$
Combining the equality 
$$ e^{\dot H} - 1 - \dot H = \sum_{m=2}^\infty \frac{1}{m!} \dot H^m = \dot H^2 \sum_{m=0}^\infty \frac{1}{(m+2)!} \dot H^m,$$
with  the product property of the $T_\p$ norm, the estimate $\sum_{m=0}^\infty \frac{1}{(m+2)!} x^m \le e^x$ valid for $x \ge0$, 
and  \eqref{eq:est_HTp_bis},  we infer  that
$$ | e^{\dot H} - 1 - \dot H|_{k,T_\p} \le  | \dot H|_{k,T_\p}^2  e^{| \dot H|_{k,T_\p}} \le 
\| \dot H\|_{k,0}^2 \, 4 (1 + | \p|_{k, \ell_2(B)}^2)^2    \, e^{\frac18 (1 + | \p|_{k, \ell_2(B)}^2)}.
$$
Thus, using  again the product property,  the assumptions $\|H \|_{k, 0} \le \frac18$ and $\| \dot H_i\|_{k,0} \le 1$,
as well as  \eqref{eq:est_HTp_bis},     we get 
\begin{align*} |f(\dot H)&|_{k, T_\p} \le   \, 
e^{|H|_{k, T_\p}} \, \, \, \| \dot H\|_{k,0}^2 \, 4 (1 + | \p|_{k, \ell_2(B)}^2)^2    \, e^{\frac18 (1 + | \p|_{k, \ell_2(B)}^2)} 
\, \,\,  \prod_{j=1}^r |\dot H_j|_{k, T_\p}  \\
\le & \, 
e^{\frac14 (1+ | \p|_{k, \ell_2(B)}^2)} \, \,  \, \| \dot H\|_{k,0}^2 \,  4 (1 + | \p|_{k, \ell_2(B)}^2)^2    \, e^{\frac18 (1 + | \p|_{k, \ell_2(B)}^2)} 
\, \, \, 2^r ( (1 + | \p|_{k, \ell_2(B)}^2)^r   \\  
\le & \,  \| \dot H\|_{k,0}^2  \,  2^{r+2}  
(1 + | \p|_{k,\ell_2(B)}^2)^{(r+2)} e^{\frac38}  e^{\frac38 | \p|_{k, \ell_2(B)}^2} \\
 \le  & \, 2^{r+2} C'_r e^{\frac38}  \| \dot H\|_{k, 0}^2 \, e^{
\frac12 |\p|^2_{k, \ell_2(B)}} ,
\end{align*}
where $C'_r = \sup_{t \ge 0}  e^{-\frac{t}{8}}  (1+ t)^{r+2}$.
Using \eqref{eq:immersion_bound_strong_norm} we get $\vertiiiReg{f(\dot H)}_{k,B}  \le 2^{r+2} C'_r e^{\frac38} \| \dot H\|_{k,0}^2$ 
and the assertion  \eqref{eq:proof_diff_e_to_H} follows.

To prove the bound  \eqref{eq:bound_D_e_to_H} we use the product property of the $T_\p$ norm to deduce that
\begin{align*}
&\,   | D^r E(H)(\dot H_1, \ldots, \dot H_r)|_{k,T_\p} \\
\le & \,  e^{|  H|_{k,T_\p}}  \, | \dot H_1 |_{k,T_\p} \ldots | \dot H_r |_{k,T_\p}\\
\le & \,e^{\frac14} e^{\frac14 | \p|_{k, \ell_2(B)}^2}  \, \| \dot H_1\|_{k,0} \ldots \| \dot H_r\|_{k,0} \, 2^r  (1+  | \p|_{k, \ell_2(B)}^2)^{r} , \\
\le &\,  C_r     e^{\frac12 | \p|_{k, \ell_2(B)}^2}  \,  \| \dot H_1\|_{k,0} \ldots \| \dot H_r\|_{k,0}.
\end{align*} 
Dividing  both sides by $ e^{\frac12 | \p|_{k, \ell_2(B)}^2}$,  taking the supremum over $\p$, and using    the definition\eqref{eq:immersion_bound_strong_norm},   we get  \eqref{eq:bound_D_e_to_H}.
Finally,  the bound   \eqref{eq:bound_exp_H} follows from \eqref{eq:bound_D_e_to_H} with $r=1$ since 
$ e^{H} - 1 = \int_0^1 DE(tH)(H) \, dt$
and $C_1 =2 e^{1/4}  \max_{ t \ge 0}   e^{-\frac{t}4 }(1+t) = 8 e^{-1/2}$. 
\end{proof}

\section{The map \texorpdfstring{$P_2$}{P2}}
We next consider the map
$$ P_2(I,K) = (I-1) \circ K.$$

\begin{lemma}\label{le:P2} Let $L$ and $h$ satisfy the lower bounds \eqref{eq:hL_for_products}.
Then the map $P_2$ restricted to  $B_{\rho_1}(1)\times B_{\rho_2}\subset \boldsymbol{M}_{|||}\times \boldsymbol{M}_{}^{(A)}$ with 
$\rho_1<(2A)^{-1}$ and $\rho_2<\frac{1}{2}$  is smooth for any $A\geq 2$ and satisfies the bounds
\begin{align}\label{eq:P2.keybd}
\frac{1}{j_1!j_2!}\lVert(D_I^{j_1}D_{K}^{j_2}P_2)(I,K)(\dot{I},\ldots,\dot{I},\dot{K},\ldots,\dot{K})\rVert_{k}^{(A/2)}\leq\bigl(2A\dot{\vertiii{I}}_k\bigr)^{j_1} 
\bigl(2 \lVert\dot{K}\rVert_{k}^{(A)}\bigr)^{j_2}.                                                                                               
\end{align}
 In particular, for $I\in B_{\rho_1}(1)$ and $K\in B_{\rho_2}$ this implies
\begin{align}\label{eq:P2.secbd}
\lVert P_2(I,K)\rVert_{k}^{(A/2)}\leq 2A\vertiii{I-1}_k+2\lVert K\rVert_{k}^{A}.
\end{align}	 
\end{lemma}

On right hand side of \eqref{eq:P2.keybd} we used the convention
$$ a^0 = 1$$
that we will use   also in the rest of this section. 

\begin{proof}
We have
\begin{align}
\begin{split}
P_2(I,K)(X) & =((I-1)\circ K)(X)=\sum_{Y\in \mathcal{P}_k(X)}(I-1)^{X\setminus Y}K(Y) \\                           
& =\sum_{Y\in \mathcal{P}_k(X)}\prod_{B\in \mathcal{B}_k(X\setminus Y)}(I(B)-1)\prod_{Z\in \mathcal{C}(Y)}K(Y). 
\end{split}
\end{align}
Using \ref{it:norm1} and \ref{it:norm2} of Lemma \ref{le:submult} and $\Gamma_{k,A}(Y)=A^{|Y|_k}$ 
we get  
\begin{align}
\begin{split}\label{eq:P2.1}                                                                         
\lVert P_2(I,K)(X)\rVert_{k,X}
&\leq \sum_{Y\in \mathcal{P}(X)}        \vertiii{I-1}_k^{|X\setminus Y|}
\bigl(    \underbrace{ \lVert  K\rVert_{k}^{(A)}    }_{\le 1}      \bigr)^{|\mathcal{C}(Y)|} A^{-|Y|} \\ 
&\leq\Bigl( \frac1{2A} + \frac1A\Bigr)^{|X|} \le 1
\end{split}                                                                                          
\end{align}
where we used that $\sum_{Y \in \mathcal P(X)} a^{|X \setminus Y|} \, b^{|Y|} = (a+b)^{|X|}$ and $A \ge 2$. 

The derivatives of $P_2$ are given by
\begin{align}
	\begin{split}
	\frac{1}{j_1!j_2!}&(D_1^{j_1}D_2^{j_2}  P_2(I,K)(\dot{I},\ldots,\dot{I},\dot{K},\ldots,\dot{K}))(X)   \\
& =\hspace{-.2cm}\sum_{\substack{Y\in \mathcal{P}(X),Y_1\in \mathcal{P}(X\setminus Y), |Y_1|=j_1 \\ \mathcal{J}\subset \mathcal{C}(Y), 
|\mathcal{J}|=j_2}}\hspace{-.2cm} (I-1)^{X\setminus (Y\cup Y_1)}\dot{I}^{Y_1}\prod_{Z\in \mathcal{C}(Y)\setminus\mathcal{J}}K(Z)\prod_{Z\in \mathcal{J}}\dot{K}(Z).
	\end{split}
          \end{align}
Using  the bound $\binom{n}{j}\leq 2^n$, we can estimate the norm of the expression above similarly as in \eqref{eq:P2.1},
\begin{align}
	\begin{split}
	&\frac{1}{j_1!j_2!}  \norm{(D_1^{j_1}D_2^{j_2}P_2(I,K)(\dot{I},\ldots,\dot{I},\dot{K},\ldots,\dot{K}))(X)}_{k,X}       \\
	 & \leq \sum_{Y\in \mathcal{P}(X)} \tbinom{|X\setminus Y|}{j_1}\vertiii{I-1}_k^{|X\setminus Y|-j_1}\opnorm[\big]{\dot{I}}^{j_1}_k            
	\tbinom{|\mathcal{C}(Y)|}{j_2}\bigl(\norm{K}_{k}^{(A)}\bigr)^{|\mathcal{C}(Y)|-j_2}\bigl(\lVert \dot{K}\rVert_{k}^{(A)}\bigr)^{j_2}A^{-|Y|}\\
	& \leq \sum_{Y\in \mathcal{P}(X)} 2^{|X\setminus Y|}(2A)^{-|X\setminus Y|+j_1} \opnorm[\big]{\dot{I}}^{j_1}_k. 2^{|\mathcal{C}(Y)|}
	2^{-|\mathcal{C}(Y)|+j_2}\bigl(\lVert \dot{K}\rVert_{k}^{(A)}\bigr)^{j_2}A^{-|Y|}\\
	& \leq \Bigl(\frac{A}{2}\Bigr)^{-|X|}\bigl(2A\opnorm[\big]{\dot{I}}_k\bigr)^{j_1}\bigl(2\lVert \dot{K}\rVert_{k}^{(A)}\bigr)^{j_2} .
	\end{split}
\end{align}
Equation \eqref{eq:P2.secbd} follows from 
\begin{align}
	\frac{\d}{\d t}P_2(1+t(I-1),tK)=D_1P_2(1+t(I-1),tK)(I-1)+D_2P_2(1+t(I-1),tK)K
\end{align}
using that $P_2(1,0)=0$.

\end{proof}

\section{The map \texorpdfstring{$P_3$}{P3}}
 The smoothness of the maps $P_3$ is implied by similar estimates, but simpler,  as those for $P_2$.
\begin {lemma}\label{le:P3}
Let $L$ and $h$ satisfy the lower bounds \eqref{eq:hL_for_products} and let
$A\geq 2$, $B\geq 1$. Consider the map $P_3:\boldsymbol{M}^{(A/2)}\rightarrow \widehat{\boldsymbol{M}}^{(A/2,B)}$ given by
\begin{align}
	P_3K(X)=\prod_{Y\in \mathcal{C}(X)}K(Y). 
\end{align}
Its restriction to $B_\rho=\{K\in \boldsymbol{M}^{(A/2)}:\, \lVert K\rVert_{k}^{(A/2)}\leq \rho\} $ is smooth 
for any $\rho$ satisfying
\begin{align}
	\rho\leq (2B)^{-1}.
\end{align}
Moreover the following estimate holds for all $j\geq 0$
\begin{align}
	\frac{1}{j!}\lVert (D^jP_3K)(\dot{K}, \ldots,\dot{K})\rVert_{k}^{(A/2,B)} 
	\leq                                                                      
	\bigl(2B\lVert \dot{K}\rVert_{k,r}^{(A/2)})\bigr)^j.                     
\end{align}

\end {lemma}
\begin{proof}
We note that 
\begin{align}
	\frac{1}{j!}D^jP_3(K)(X)(\dot{K},\ldots,\dot{K})=\sum_{\substack{\mathcal{J}\subset \mathcal{C}(X) \\ \abs{\mathcal{J}}}=j}
	\prod_{Z\in \mathcal{C}(X)\setminus \mathcal{J}}K(Z)\prod_{Z\in \mathcal{J}}\dot{K}(Z).                 
\end{align}
Using the bound $\binom{|\mathcal{C}(X)|}{j}\leq 2^{|\mathcal{C}(X)|}$ and \ref{it:norm1} from Lemma \ref{le:submult}   for $K\in B_\rho$, we get 
\begin{align}
\begin{split}
	B^{|\mathcal{C}(X)|}\Bigl(\frac{A}{2}\Bigr)^{|X|}& \frac{1}{j!}\norm{(D^jP_3K)(\dot{K}, \ldots,\dot{K})(X)}_{k,X}\\
	& \leq (2B)^{|\mathcal{C}(X)|}\bigl( \lVert K\rVert_{k}^{(A/2)}\bigr)^{|\mathcal{C}(X)|-j} 
	\bigl( \lVert \dot{K}\rVert_{k,r}^{(A/2)}\bigr)^{j}
	 \leq \bigl( 2B \lVert \dot{K}\rVert_{k,r}^{(A/2)}\bigr)^{j}.                     
\end{split}
\end{align}
\end{proof}

\section{The map \texorpdfstring{$P_1$}{P1}}

Next we show smoothness of the outermost map $P_1$
given by
\begin{align}
\begin{split}
 P_1(&I_1,I_2,J,K)(U,\p)\\
 &=\sum_{\substack{X_1,X_2\in \mathcal{P}_k            \\ X_1\cap X_2=\emptyset}}
  \chi(X_1\cup X_2,U)I_1^{U\setminus (X_1\cup
	X_2)}(\p)I_2^{(X_1\cup X_2)\setminus U}J^{X_1}(\p) 
	\prod_{Y \in \Ccl(X_2)} K(Y,\p)
\end{split}
\end{align}

\begin{lemma}\label{le:P1} 
Let $L$ and $h$ satisfy the lower bounds  \eqref{eq:hL_for_P1}
and
\begin{equation} \label{eq:A0_in_P1_estimate}
A_0(L) =({ 48}\AP)^{\frac{L^d}{\upalpha}}
\end{equation}
with $\upalpha(d)=(1+2^d)^{-1}(1+6^d)^{-1}$ and $\AP$ as in Theorem~\ref{th:weights_final}~\ref{w:w7}.
Further, let $A\geq A_0(L)$, $B=A$ and
\begin{align}
	\rho_1 = \rho_2\leq \frac{1}{2}, \quad  \rho_3\leq A^{-2},\quad  \rho_4\leq 1.
\end{align}
Then the map $P_1$ restricted to  the neighbourhood 
$$
U=B_{\rho_1}(1) \times B_{\rho_2}(1) \times B_{\rho_3}(0) \times B_{\rho_4}(0)\subset \boldsymbol{M}_{|||}\times \boldsymbol{M}_{|||}\times 
\boldsymbol{M}_{|||}\times \widehat{\boldsymbol{M}}_:^{(A/(2\AP),B)}
$$
is smooth with the bound on derivatives,
\begin{align}\begin{split}\label{eq:P1:Claim}
	\frac{1}{i_1!i_2!j_1!j_2!}\lVert D^{i_1}_{I_1}D^{i_2}_{I_2}D^{j_1}_JD^{j_2}_K\,
	& P_1(I_1,I_2,J,K)  
	(\dot{I}_1,\ldots,\dot{I}_2,\ldots,  \dot{J},\ldots,\dot{K},\ldots)\rVert_{k+1,r}^{(A)}\\
	& \leq \dot{\vertiii{I_1}}^{i_1}\dot{\vertiii{I_2}}^{i_2}(A^2\dot{\vertiii{J}})^{j_1}(\rVert 
	\dot{K}\rVert_{k:k+1}^{(A/(2\AP),B)})^{j_2}.
	\end{split}
\end{align}
\end{lemma}
\begin{proof}
We first note some simple inequalities for polymers.
Recall from Lemma~\ref{le:factorise_cprime} 
  that
\begin{equation}  \label{eq:factorise_cprime4}
\sum_{Y \in \mathcal \Ccl(X)} |\mathcal C(Y)| = |\mathcal C(X)|.
\end{equation}	
Next let $X \in \mathcal P_k$ and $U = \pi(X)$.
Then by \eqref{eq:pisetincl} we have  $X \subset U^\ast$ and hence
\begin{equation}  \label{eq:blocks_X_Uast}
|X \setminus U|_k + |U \setminus X|_k \le |U^*|_k, \quad |X|_k \le |U^*|_k.
\end{equation}
We also have
\begin{equation} \label{eq:estimate_Ustar_blocks}
|U^\ast|_k\leq 2|U|_k  \text{ if } L \ge 4d (2^d + R).
\end{equation}
Indeed for $B'\in \mathcal{B}_{k+1}$ and $k \ge 1$
\begin{align}  \label{eq:estimate_Bstar_blocks}
	|{B'}^\ast|_k\leq (L+2^{d+1})^d\leq L^d\Bigl(1+\frac{1}{2d}\Bigr)^d\leq L^d e^{\frac12},
\end{align}
while for $k=0$,
\begin{align*}
	 |{B'}^\ast|_0\leq (L+2^{d+1} + 2R)^d\leq L^d\Bigl(1+\frac{1}{2d}\Bigr)^d\leq L^d e^{\frac12}.
\end{align*}
Finally, for $X_2, X \in \mathcal P_k$ with $X_2 \subset X$ we use the identity
\begin{equation} \label{eq:count_components_smoothness}
	|\mathcal C(X_2)| =
	\sum_{Y \in \mathcal C(X)} |\mathcal C(X_2 \cap Y)|.
\end{equation}
It suffices to show that each connected component of $X_2$ is a connected component
of $X_2 \cap Y$ for some $Y \in \mathcal C(X)$ (with $Y \cap X_2 \ne \emptyset$) and vice versa. 
Now if $Z \in \mathcal C(X_2)$ then $Z$ is a connected subset of $X$ and hence contained in 
exactly one component $Y$ of $X$. Thus $Z$ is a connected subset of $X_2 \cap Y$.
In fact $Z \in \mathcal C(X_2 \cap Y)$ because
$\dist_\infty(Z, (X_2 \cap Y) \setminus Z) \ge \dist_\infty(Z, X_2 \setminus Z) \ge L^k$
as $Z \in \mathcal C(X_2)$. 

Conversely consider $Y \in \mathcal C(X)$ with $X_2 \cap Y \ne \emptyset$ and   $Z \in \mathcal C(X_2 \cap Y)$.
Then $Z$ is a connected subset of $X_2$. 
Moreover 
$ \dist_\infty(Z, (X_2 \setminus Y) \setminus Z)  \ge \dist_\infty(Y, X \setminus Y) \ge L^k$
and $\dist_\infty(Z, (X_2 \cap Y) \setminus Z) \ge L^k$. Thus $\dist(Z, X_2 \setminus Z) \ge L^k$
and therefore $Z \in \mathcal C(X_2)$. This concludes the proof of \eqref{eq:count_components_smoothness}.

Now let  $U\in\mathcal{P}_{k+1}^c$ be a connected polymer.
	Lemma~\ref{le:submult}  implies that 
\begin{align}
\begin{split}\label{uglyestimate}
  &\lVert P_1(I_1,I_2,J,K)\rVert_{k+1,U,r} \\
	&\hspace{0.5cm} \leq                                    
	\sum_{\substack{X_1,X_2\in \mathcal{P}_k\\ X_1\cap X_2=\emptyset}} \chi(X_1\cup X_2,U) \, 
	\vertiii{I_2}_k^{|(X_1\cup X_2)\setminus U|} \, 
	\vertiii{I_1}_k^{|U\setminus (X_1\cup X_2)|}  \, \times
	\\ &
	\hspace{6cm}\times 	\vertiii{J}_k^{|X_1|} \, \, \,\Big\lVert   \prod_{Y \in \Ccl(X_2)} K(Y)\Big\rVert_{k:k+1,X_2}
		\\
	& \hspace{0.2cm} \underset{   \eqref{eq:factorise_cprime4}   }{\leq}                                    
	\sum_{\substack{X_1,X_2\in \mathcal{P}_k\\ X_1\cap X_2=\emptyset}} \chi(X_1\cup X_2,U) \, 
	2^{|(X_1\cup X_2)\setminus U|}  \, 
	2^{|U\setminus (X_1\cup X_2)|} \, 
	A^{-2|X_1|}\times\\
	& \hspace{5cm}                            
	\times\Bigl(\frac{A}{2\AP}\Bigr)^{-|X_2|}B^{-|\mathcal{C}(X_2)|}\bigl(\lVert K\rVert_{k:k+1}^{(A/(2\AP),B)}\bigr)^{|\Ccl(X_2)|}
		\\
	&\hspace{0.2cm} \underset{    \eqref{eq:blocks_X_Uast}  }{\leq} 2^{ 2 |U^\ast|_k} \, (\AP)^{|U^\ast|_k}                
	\sum_{\substack{X_1,X_2\in \mathcal{P}_k\\ X_1\cap X_2=\emptyset}} \chi(X_1\cup X_2,U) \, 
	A^{-2|X_1|-|X_2|}B^{-|\mathcal{C}(X_2)|}\\
	& \underset{  \eqref{eq:estimate_Ustar_blocks}, \eqref{eq:count_components_smoothness}    }{\leq}
	\hspace{-.2cm}( {4} \AP)^{2|U|_k}                        
	\sum_{\substack{X_1,X_2\in \mathcal{P}_k\\ X_1\cap X_2=\emptyset}} \chi(X_1\cup X_2,U)
	\hspace{-.2cm}\prod_{Y\in \mathcal{C}(X_1\cup X_2)} \hspace{-.2cm} A^{-2|X_1\cap Y|-|X_2\cap Y|-|\mathcal{C}(X_2\cap Y)|} ,
\end{split}
\end{align}
	where we used $B=A$ to get the last inequality.

	Now we use the crucial fact that connected polymers $X$ with $X\notin \mathcal{B}_k$ satisfy
	the bound $|\pi(X)|_{k+1}<c|X|_k$  for some $c<1$ . For the precise formulation, we refer to the standard inequality  \eqref{eq:app1} 
	in Lemma \ref{le:app1}  (Appendix C). It states that for  polymers $X\in \mathcal{P}_k^c\setminus \mathcal{B}_k$, we have
\begin{align}
	|X|_k\geq (1+2\upalpha(d))|\pi(X)|_{k+1}, 
\end{align}
where $0<\upalpha(d)=((1+2^d)(1+6^d))^{-1}<1$ is a positive constant.
This implies, for $Y\in \mathcal{C}(X_1\cup X_2)$ such that $Y\notin \mathcal{B}_k$, that
\begin{align}\label{eq:P3comb1}\begin{split}
	2|X_1\cap Y|+|X_2\cap Y|+|\mathcal{C}(X_2\cap Y)|
	&\geq 
	|X_1\cap Y| +|X_2\cap Y|= |Y|
	\\ &
	\geq  (1+2\upalpha(d))|\pi(Y)|_{k+1}.
	\end{split}
\end{align}
	 If $Y \in \mathcal B_k$ we note that
	either $Y\subset X_1$ or $Y\subset X_2$. In either case we get
\begin{align}\label{eq:P3comb2}
	2|X_1\cap Y|+|X_2\cap Y|+|\mathcal{C}(X_2\cap Y)|=2= 2|\pi(Y)|_{k+1}\geq (1+2\upalpha(d))|\pi(Y)|_{k+1} 
\end{align}
where we used $\pi(Y)\in \mathcal B_{k+1}$.

Inserting \eqref{eq:P3comb1} and \eqref{eq:P3comb2} into \eqref{uglyestimate}, we get
\begin{align}
\begin{split}\label{eq:P1:1}
	\lVert P_1(I_1&,I_2,J,K)\rVert_{k+1,U}
	\\
	& \leq   (     { 4}      \AP)^{2|U|_k}
	\sum_{\substack{X_1,X_2\in \mathcal{P}_k\\ X_1\cap X_2=\emptyset}} \chi(X_1\cup X_2,U)
	\prod_{Y\in \mathcal{C}(X_1\cup X_2)}A^{-(1+2\upalpha)|\pi(Y)|_{k+1}} \\
	& \leq                                                                          
	(    { 4} \AP)^{2|U|_k} \, 3^{|U^\ast|_k} \, A^{-(1+2\upalpha)|U|_{k+1}}
	\underset{   \eqref{eq:estimate_Ustar_blocks} }{\leq} (   { 12} \AP)^{2|U|_{k}}
	A^{-(1+ 2 \upalpha)|U|_{k+1}}\\
	& \leq \Bigl(\frac{(   { 12} \AP)^{2L^d}}{A^{2\upalpha}}\Bigr)^{|U|_{k+1}}A^{-|U|_{k+1}} .
\end{split}
\end{align}
For the second inequality we used that $X_1 \cup X_2 \subset U^\ast$ if $\chi(X_1 \cup X_2, U) \ne 0$
and that there are $3^{|U^\ast|_k}$ possibilities for partitions of $U^\ast$ into three disjoint sets
$X_1$, $X_2$ and $X_3 = U^\ast \setminus (X_1 \cup X_2)$. We also used
that  by the  definition of $\pi$ we have $\pi(X) = \cup_{Y \in \mathcal C(X)} \pi(Y)$ and  thus
$|\pi(X)|_{k+1} \le \sum_{Y \in \mathcal C(X)} |\pi(Y)|_{k+1}$. 

Thus we get for $A \ge (  12 \AP)^{\frac{L^d}{\upalpha}}$ 
\begin{align}
	\lVert P_1(I_1,I_2,J,K)\rVert_{k+1}^{(A)}\leq 1.
\end{align}

Let us now proceed to   the bounds for  derivatives. Similarly to the  derivatives of $P_3$ in Lemma \ref{le:P2}, we get
\begin{align}
	{\frac1{j!}} 	D^j \Big(\prod_{Y \in \Ccl(X)} K(Y) \Big)(\dot{K},\ldots,\dot{K})=\sum_{\substack{\mathcal{J}\subset \Ccl(X) \\ |\mathcal{J}|=j}}\prod_{Y\in \mathcal{J}}\dot{K}(Y)
		\prod_{Y\in \Ccl(X)\setminus \mathcal{J}}K(Y).                                
\end{align}
For  $\lVert K\rVert_{k:k+1}^{(A/(2\AP),B)}\leq 1 $ we use Lemma~\ref{le:submult}  to get,
\begin{align}
\begin{split}
	\frac{1}{j!}\Big\lVert D^j \big(&\prod_{Y \in \Ccl(X)} K(Y) \big)(\dot{K},\ldots,\dot{K})\Big\rVert_{k:k+1,X} 
	\\ & 
	\leq 
	\sum_{\substack{\mathcal{J}\subset \mathcal{C}'(X)\\ |\mathcal{J}|=j}}
	\prod_{Y\in \mathcal{J}}\lVert \dot{K}(Y)\rVert_{k:k+1,Y}
	\prod_{Y\in \mathcal{C}'(X)\setminus \mathcal{J}}\lVert K(Y)\rVert_{k:k+1,Y}
	\\
	& \underset{  \eqref{eq:factorise_cprime4}}{\leq} 
	\tbinom{|\mathcal{C}'(X)|}{j} \Bigl(\frac{A}{2\AP}\Bigr)^{-|X|}B^{-|\mathcal{C}(X)|} \bigl(\lVert \dot{K}\rVert_{k:k+1}^{(A/(2\AP),B)}\bigr)^j.
\end{split}
\end{align}
A similar bound holds for the factors of $I_1$, $I_2$, and $J$.
Therefore, similarly to \eqref{uglyestimate},  we bound 
\begin{align}
\begin{split}
	& \frac{1}{i_1!i_2!j_1!j_2!}\lVert D_{i_1}^{I_1}D_{i_2}^{I_2}D_{j_1}^JD_{j_2}^K P_1(I_1,I_2,J,K)\rVert_{k+1,U} \\
	& \leq                                                                                                         
	\sum_{\substack{X_1,X_2\in \mathcal{P}_k\\ X_1\cap X_2=\emptyset}} \chi(X_1\cup X_2,U)
	\tbinom{|(X_1\cup X_2)\setminus U|}{i_2}\vertiii{I_2}_k^{|(X_1\cup X_2)\setminus U|-i_2}\dot{\vertiii{{I}_2}}_k^{i_2}
	\times\\
	& \qquad\qquad \times                                                                                          
	\tbinom{|U\setminus (X_1\cup X_2)|}{i_1}\vertiii{I_1}_k^{|U\setminus (X_1\cup X_2)|-i_1}\dot{\vertiii{{I}_1}}_k^{i_1}
	\tbinom{|X_1|}{j_1}\vertiii{J}_k^{|X_1|-j_1}\dot{\vertiii{{J}}}_k^{j_1}\times\\
	& \qquad\qquad                                                                                                 
	\times\tbinom{\mathcal{C}'(X_2)}{j_2}\Bigl(\frac{A}{2\AP}\Bigr)^{-|X_2|}B^{-|\mathcal{C}(X_2)|} 
	\bigl(\lVert \dot{K}\rVert_{k:k+1}^{(A/(2\AP),B)}\bigr)^{j_2}.
\end{split}
\end{align}
	Assume that $\chi(U,X_1\cup X_2)=1$. Then  $X_1\cup X_2\subset U^\ast$.
and we can bound the combinatorial factor by
\begin{align}
\begin{split}
	\tbinom{|(X_1\cup X_2)\setminus U|}{i_2}
	\tbinom{|U\setminus (X_1\cup X_2)|}{i_1}
	\tbinom{|X_1|}{j_1}
	\tbinom{|\mathcal{C}'(X_2)|}{j_2} &
	\leq 2^{|(X_1\cup X_2)\setminus U|+|U\setminus (X_1\cup X_2)|+|X_1|+|X_2|} \\
	   &                                                                            
	\underset{\eqref{eq:blocks_X_Uast}  }{\leq} 2^{2|U^\ast|_k}
	\underset{ \eqref{eq:estimate_Ustar_blocks}  }{\leq} 4^{2|U|_k}.
\end{split}
\end{align}
Then we bound, exactly as in \eqref{uglyestimate},
\begin{align}
\begin{split}\label{uglyestimate2}
	& \frac{1}{i_1!i_2!j_1!j_2!}\lVert D_{i_1}^{I_1}D_{i_2}^{I_2}D_{j_1}^JD_{j_2}^K P_1(I_1,I_2,J,K)\rVert_{k+1,U} \\
	& \leq(  {  16} \AP)^{2|U|_k}                                                                                           
	\sum_{\substack{X_1,X_2\in \mathcal{P}_k\\ X_1\cap X_2=\emptyset}} \chi(X_1\cup X_2,U)
	\prod_{Y\in \mathcal{C}(X_1\cup X_2)}A^{-2|X_1\cap Y|-|X_2\cap Y|-|\mathcal{C}(X_2\cap Y)|}\times\\
	& \qquad \qquad\times                                                                                          
	\bigl(\tfrac12 \dot{\vertiii{{I}_1}}_k\bigr)^{i_1}\bigl(\tfrac12 \dot{\vertiii{{I}_2}}_k\bigr)^{i_2}\bigl(A^2\dot{\vertiii{{J}}}_k\bigr)^{j_1}
	\bigl(\lVert \dot{K}\rVert_{k:k+1}^{(A/(2\AP),B)}\bigr)^{j_2}.
\end{split}
\end{align}

Now, we can conclude as in \eqref{eq:P1:1} that 
\begin{align}
\begin{split}
	&\frac{1}{i_1!i_2!j_1!j_2!}\lVert D_{i_1}^{I_1}D_{i_2}^{I_2}  D_{j_1}^JD_{j_2}^K P_1(I_1,I_2,J,K)\rVert_{k+1,U,r} \\
	&\qquad \leq                                                
	\Bigl(\frac{(   {48} \AP)^{2L^d}}{A^{2\upalpha}}\Bigr)^{|U|_{k+1}}A^{-|U|_{k+1}}
	\dot{\vertiii{{I}_2}}_k^{i_2}
	\dot{\vertiii{{I}_1}}_k^{i_1}
	\bigl(A \dot{\vertiii{{J}}}_k\bigr)^{j_1}
	\bigl(\lVert \dot{K}\rVert_{k:k+1,r}^{(A/(2A'),B)}\bigr)^{j_2}
	\\
	&\qquad \leq                                                
	A^{-|U|_{k+1}}
	\dot{\vertiii{{I}_1}}_k^{i_1}
	\dot{\vertiii{{I}_2}}_k^{i_2}
	\bigl(A^2\dot{\vertiii{{J}}}_k\bigr)^{j_1}
	\bigl(\lVert \dot{K}\rVert_{k:k+1,r}^{(A/(2A'),B)}\bigr)^{j_2}
\end{split}
\end{align}
once $A>(48 \AP)^{\frac{L^d}{\upalpha}}$. This implies the claim \eqref{eq:P1:Claim}.
\end{proof}

\section{The map \texorpdfstring{$R_1$}{R1}}
Next we discuss the smoothness of the maps $R_1$ and $R_2$ which depend explicitly on $\boldsymbol{q}$.
The proofs are similar to those in \cite{AKM16} however we do not have to deal
with the $\boldsymbol{q}$ derivatives explicitly because we already controlled them in Lemma \ref{le:keyboundRk}.
Let us begin with the map $R_1$ which is defined by 
$$ R_1(P,\boldsymbol{q})(X,\p)=(\myR^{(\boldsymbol{q})}_{k+1}P)(X,\p)=\int_{\Xcal_N}P(X,\p+\xi)\,\mu_{k+1}^{(\boldsymbol{q})}(\d\xi).$$

\begin{lemma}\label{le:R1}  Let $L$ and $h$ satisfy the lower bound  \eqref{eq:L_for_R1R2} 
and let $\kappa = \kappa(L)$ be the constant in Theorem~\ref{th:weights_final}.
For $B\geq 1$ and any $A\geq 4\AP$  the map $R_1$ restricted to $\widehat{\boldsymbol{M}}_{k}^{(A/2,B)}\times \mathcal{U}_{\kappa}$  
is smooth and satisfies
\begin{align}\label{eq:R1:Claim0}
	\lVert  D^j_P R_1(P,\boldsymbol{q})(X,\cdot)(\dot{P},\ldots \dot{P})\rVert_{k:k+1}^{(A/(2\AP),B)} \leq                                                                                                                    
	(\lVert \dot{P} \rVert_{k}^{(A/2)})^{j}(\lVert P \rVert_{k}^{(A/2)})^{1-j}.                        
\end{align}
and
\begin{align}\label{eq:R1:Claim}
\begin{split}
	\lVert D^\ell_{\boldsymbol{q}} D^j_P R_1(P,\boldsymbol{q})(X,\cdot)&(\dot{\boldsymbol{q}},\ldots,\dot{\boldsymbol{q}},\dot{P},\ldots \dot{P})\rVert_{k:k+1}^{(A/(2\AP),B)} 
	\\ &
	\leq C_\ell(L)   \lVert \dot{\boldsymbol{q}}\rVert^\ell (\lVert \dot{P} \rVert_{k}^{(A/2)})^{j}(\lVert P \rVert_{k}^{(A/2)})^{1-j}.        
	\end{split}                
\end{align}
for $\ell \geq 1 $ and $0\leq j\leq 1$. The constants $C_\ell(L)$ do not depend on $h$ or $A$.
The derivatives vanish for $j>1$.
\end{lemma}

\begin{proof} 	
Note first that the map $R_1$ is linear in $P$. Therefore the statement for the derivative in $P$ direction is trivial and we only need to consider the $\boldsymbol{q}$ derivative. Note that $X\in \Pclk$ is equivalent to the condition that $\pi(X)$ is connected.
Therefore we can apply Lemma \ref{le:keyboundRk}. From \eqref{eq:bdRkl0} we get
\begin{align}
	\lVert(\myR_{k+1}^{(\boldsymbol{q})}K)(X)\rVert_{k:k+1,X}\leq {\AP}^{|X|_k}\lVert K(X)\rVert_{k,X}.        
\end{align} 
and hence 
\begin{align}
\begin{split}
	\lVert   \boldsymbol{R}^{(\boldsymbol{q})}_{k+1}K\rVert_{k:k+1}^{(A/(2\AP),B)} 
	& =                                                                                                                           
	\sup_{X\in \Pclk}B^{|\mathcal{C}(X)|}\Bigl(\frac{A}{2\AP}\Bigr)^{|X|_k}\lVert 
	\boldsymbol{R}_{k+1}^{(\boldsymbol{q})}K(X)\rVert_{k:k+1,X}\\
	& \leq \sup_{X\in \Pclk}B^{|\mathcal{C}(X)|}\Bigl(\frac{A}{2\AP}\Bigr)^{|X|_k}{\AP}^{|X|_k}\lVert K(X)\rVert_{k,X} \\
	& = \lVert K(X)\rVert_{k:k+1}^{(A/2,B)}.                                                                                
\end{split}
\end{align}
Similarly, for $\ell \ge 1$, we get	
\begin{align}
\begin{split}
	\lVert D_{\boldsymbol{q}}^\ell \boldsymbol{R}^{(\boldsymbol{q})}_{k+1}K\rVert_{k:k+1}^{(A/(2\AP),B)} 
	& =                                                                                                                           
	\sup_{X\in \Pclk}B^{|\mathcal{C}(X)|}\Bigl(\frac{A}{2\AP}\Bigr)^{|X|_k}\lVert D_{\boldsymbol{q}}^\ell
	\boldsymbol{R}_{k+1}^{(\boldsymbol{q})}K(X)\rVert_{k:k+1,X}\\
	& \leq C_\ell(L)\sup_{X\in \Pclk}B^{|\mathcal{C}(X)|}\Bigl(\frac{A}{2\AP}\Bigr)^{|X|_k}{\AP}^{|X|_k}\lVert K(X)\rVert_{k,X} \\
	& =C_\ell(L) \lVert K(X)\rVert_{k:k+1}^{(A/2,B)}.                                                                                
\end{split}
\end{align}
\end{proof}

\section{The map \texorpdfstring{$R_2$}{R2}}

\begin{lemma}\label{lemmaS2}   Let $L$ and $h$ satisfy the lower bound   \eqref{eq:L_for_R1R2}.
For any $h\geq 1$ and $A\geq 1$ the map $R_2$ defined in \eqref{defs2} is smooth. 
Moreover  there exist a constant $C_0$ (which is independent of $L,h$ and $A$) and  for each $\ell \ge 1$  there exist   a constant $C_\ell(L)$ 
(which is independent of $h$ an $A$) such that
\begin{align}\label{boundss2}
	\lVert D^{j_1}_HD^{j_2}_KD^\ell_{\boldsymbol{q}} R_2(H,K,\boldsymbol{q})(\dot{H},\dot{K},\dot{\boldsymbol{q}})\rVert_{k,0}  \leq C_\ell(L)                                                                           
\begin{cases}     \lVert H\rVert_{k,0}+\lVert K\rVert_{k}^{(A)}\;\text{if}\;j_1=j_2=0   \\
\lVert \dot{H}\rVert_{k,0}\;\text{if}\;j_1=1,j_2=0                                     \\
\lVert \dot{K}\rVert_{k}^{(A)}\;\text{if}\;j_1=0,j_2=1,                                
\end{cases}                                                                            
\end{align}
and 
\begin{align}\label{trivialzeros2}
	D^{j_1}_HD^{j_2}_KD^\ell_{\boldsymbol{q}} R_2(H,K,\boldsymbol{q})(\dot{H},\ldots,\dot{H},\dot{K},\ldots,\dot{K},\dot{\boldsymbol{q}},\ldots,\dot{\boldsymbol{q}})=0 
	\quad\text{if}\;j_1+j_2\geq 2.                                                                                         
\end{align}
\end{lemma}
\begin{proof}
First we observe that
$R_2(H,K,\boldsymbol{q})=R_{2,a}^{(\boldsymbol{q})}H-R^{(\boldsymbol{q})}_{2,b}K$ where both 
$R_{2,a}^{(\boldsymbol{q})}$ and $R^{(\boldsymbol{q})}_{2,b}$ are linear maps given by 
\begin{align}
	R_{2,a}^{(\boldsymbol{q})}H=\Pi_2\myR_{k+1}^{(\boldsymbol{q})}H,\qquad R_{2,b}^{(\boldsymbol{q})}K=\Pi_2\myR_{k+1}^{(\boldsymbol{q})}K. 
\end{align}
This implies \eqref{trivialzeros2}. 
Due to the linearity with respect to $H$ and $K$ the bounds for the derivatives with respect to $H$ and $K$
follow from the case without derivatives in $H$ or $K$ direction.
We consider the two operators separately. 

The estimate for the operator $R_{2,a}^{(\boldsymbol{q})}$ is simple because its action on Hamiltonians can be calculated explicitly.
It only changes the constant part 
$a_\emptyset \rightarrow a_\emptyset + 	\sum_{(i, \alpha), (j, \beta) \in \mf v_2}   
a_{(i, \alpha), (j, \beta)}    \,    (\nabla^\beta)^\ast \nabla^\alpha \mathcal{C}^{(\boldsymbol{q})}_{k+1, ij}(0)$,
see
Proposition \ref{prop:contractivity}.
Using the bound \eqref{eq:discretebounds}  and the definition  \eqref{hamiltoniannorm}, we get
\begin{align}\label{raestimate}
	\lVert D_{\boldsymbol{q}}^\ell R_{2,a}^{(\boldsymbol{q})}H\rVert_{k,0}\leq \lVert H\rVert_{k,0}+c_{2,\ell}h_k^{-2}\lVert H\rVert_{k,0}\leq  
	(1 +  c_{2,\ell\eqref{eq:discretebounds}  }) \lVert H\rVert_{k,0} 
\end{align}
if $h \ge 1$.
	
Further, let us consider the map $R^{(\boldsymbol{q})}_{2,b}$. 
From the linearity of $\Pi_2$ and Lemma~\ref{le:Pi2_bounded} we get 
\begin{align}\begin{split}
	\lVert D^\ell_{\boldsymbol{q}}\Pi_2\boldsymbol{R}_{k+1}^{(\boldsymbol{q})}K(B,\cdot)(\dot{\boldsymbol{q}},\ldots, \dot{\boldsymbol{q}})\rVert_{k,0} 
	& \leq \lVert \Pi_2(D^\ell_{\boldsymbol{q}}\boldsymbol{R}_{k+1}^{(\boldsymbol{q})}K(B,\cdot))(\dot{\boldsymbol{q}},\ldots, \dot{\boldsymbol{q}})\rVert_{k,0}\\
	&\leq C_{\eqref{eq:bound_Pi2}}
	|D^\ell_{\boldsymbol{q}}\boldsymbol{R}_{k+1}^{(\boldsymbol{q})}K(B,0)(\dot{\boldsymbol{q}},\ldots,\dot{\boldsymbol{q}})|_{k, B, T_0}\\
	& \leq C_\eqref{eq:bound_Pi2}  \, 
	\lVert D^\ell_{\boldsymbol{q}}\boldsymbol{R}_{k+1}^{(\boldsymbol{q})}K(B)(\dot{\boldsymbol{q}},\ldots,\dot{\boldsymbol{q}})\rVert_{k:k+1,B}.
\end{split}
\end{align}
In the last step we used that by definition \eqref{middlenorm},
\begin{align}
	\lVert F(B)\rVert_{k:k+1,B}=\sup_{\p}   w_{k:k+1}^{-B}(\p) \, |F(B)|_{k,B, T_0} \geq |F(B)|_{k,B,T_0} 
\end{align}
since  $w_{k:k+1}^{-B}(0)=1$.
Now, Lemma \ref{le:keyboundRk} for $\ell \ge 1$ yields 
\begin{align}  \label{eq:R2b_bound} 
\begin{split}   
	\lVert D^\ell_{\boldsymbol{q}}\Pi_2\boldsymbol{R}_{k+1}^{(\boldsymbol{q})}K(B,\cdot)(\dot{\boldsymbol{q}},\ldots, \dot{\boldsymbol{q}})\rVert_{k,0} 
	& \leq C_\eqref{eq:bound_Pi2}    \,                                                                       
	\lVert  D^\ell_{\boldsymbol{q}}\boldsymbol{R}_{k+1}^{(\boldsymbol{q})}K(B)(\dot{\boldsymbol{q}},\ldots,\dot{\boldsymbol{q}})\rVert_{k:k+1,B}\\
	& \leq C_\eqref{eq:bound_Pi2} C_{\ell,\eqref{eq:bdRkl1}(L)} {\AB}   \, \lVert\dot{\boldsymbol{q}}\rVert^\ell     \, \lVert K(B)\rVert_{k,B}              \\                                                 	& \leq \frac{C_{\eqref{eq:bound_Pi2} C_{\ell,\eqref{eq:bdRkl1}}(L) }{\AB}\lVert\dot{\boldsymbol{q}}\rVert^\ell}{A}  \,  \lVert K\rVert_{k}^{(A)}. 
\end{split}
\end{align} 
	
This implies that 
\begin{align}\label{rbestimate}
	\lVert D^\ell_{\boldsymbol{q}}\boldsymbol{R}_{k+1}^{(\boldsymbol{q})}K\rVert_{k,0}\leq C_\ell(L) \lVert \dot{\boldsymbol{q}}\rVert^\ell\lVert K\rVert_{k}^{(A)}
\end{align}
for $\ell \ge 1$. 
The bounds  \eqref{raestimate} and \eqref{rbestimate} jointly yield the desired estimate for $\ell \ge 1$.
For $\ell = 0$ we get, instead of  \eqref{eq:R2b_bound}, a slightly sharper estimate,
\begin{align}  \label{eq:R2b_bound_0} 
	\lVert \Pi_2\boldsymbol{R}_{k+1}^{(\boldsymbol{q})}K(B,\cdot) \rVert_{k,0} & \leq C_{\eqref{eq:bound_Pi2}}    \,                                                                       
	\lVert  \boldsymbol{R}_{k+1}^{(\boldsymbol{q})}K(B)    \rVert_{k:k+1,B}	                                                                                                                                                         
	 \leq \frac{   C_{\eqref{eq:bound_Pi2}} {\AB}	  }{A}  \,  \lVert K\rVert_{k}^{(A)}. 
\end{align} 	
Together with  \eqref{raestimate} and the assumption $A \ge 1$ this implies the desired estimate for $\ell=0$ with
\begin{equation} \label{eq:bound_R2_ell0}
	C_0 = 1 +   c_{2,\ell, \eqref{eq:discretebounds}  } +    C_{\eqref{eq:bound_Pi2}}\frac{   \AB}{A}.
\end{equation}

\end{proof}
\begin{corollary}\label{cor:bdAB}
	The operators $\boldsymbol{A}_k^{(\boldsymbol{q})}$ and $\boldsymbol{B}_k^{(\boldsymbol{q})}$ satisfy the estimate \eqref{eq:qderivABC}.
\end{corollary}
\begin{proof}
The operators $\boldsymbol{A}_k^{(\boldsymbol{q})}$ and $\boldsymbol{B}_k^{(\boldsymbol{q})}$ satisfy the identities
\begin{align}
\begin{split}
	\boldsymbol{A}_k^{(\boldsymbol{q})}H(B',\p) & =\sum_{B\in \mathcal{B}(B')}R_{2,a}^{(\boldsymbol{q})}H(B,\p)   \\
	\boldsymbol{B}_k^{(\boldsymbol{q})}K(B',\p) & =-\sum_{B\in \mathcal{B}(B')}R_{2,b}^{(\boldsymbol{q})}K(B,\p). 
\end{split}
\end{align}
Hence, the claim follows from  bounds \eqref{raestimate} and \eqref{rbestimate}.
\end{proof}

\section{Proof of Theorem~\ref{PROP:SMOOTHNESSOFS}}  \label{se:proof_smoothness_S}

\begin{proof}[Proof of Theorem  \ref{prop:smoothnessofS}]
The assertion follows from the smoothness of the individual maps $E$, $P_1$, $P_2$, $P_3$, $R_1$, and $R_2$
and the chain rule. To get an estimate for a neighbourhood $U_{\rho, \kappa} \ni 0$ on which the map $\myS_k$ is smooth
and to see on which parameters the constants $\rho$ and $\kappa$ depend,
we  sequentially trace the dependence back to the neighbourhoods on which the individual maps and their compositions are smooth. 

First, we fix  a constant  $A \ge A_0(L)$, where 
\begin{equation}  \label{eq:bound_A0_smooth}
	A_0(L)=({  48} \AP(L))^{\frac{L^d}{\upalpha}} \quad \hbox{with} \quad  \upalpha(d)=(1+2^d)^{-1}(1+6^d)^{-1}
\end{equation}
is as in  Lemma~\ref{le:P1} and  set 
$$
B=A.
$$
Thus, by   Lemma~\ref{le:P1},  the map   $P_1$ is smooth in a neighbourhood 
$O_1=B_{\rho_1}(1)\times B_{\rho_2}(1)\times B_{\rho_3}(0)\times B_{\rho_4}(0)$ with
$$
\rho_1 = \rho_2 = \tfrac12, \  \rho_3 = A^{-2}, \text{ and }  \rho_4=1.
$$
Using  Lemma \ref{le:R1}, we find a neighbourhood $O_2=B_{\rho_5}\times B_{\kappa}$ of the origin
such that  $R_1$ is smooth on $O_2$ and $R_1(O_2)\subset {B_{\rho_4}}$.
Indeed, we may take
\begin{equation}  \label{eq:bound_kappa_smooth}
	\kappa= \kappa(L) \text{ to be the constant } \kappa \text{ in  Theorem~\ref{th:weights_final} }
\end{equation}
and 
$$ 
\rho_5 =  \rho_4 = 1. 
$$
Similarly, by Lemma~\ref{le:smoothness_exp}, there  exists a neighbourhood $O_3 = B_{\rho_6}(0)$  
such that $E$ is smooth on $O_3$ and $E(O_3) \subset B_{\rho_1}(1) \cap B_{\rho_2}({1}) \cap B_{\rho_3}({1})$.
Indeed, since $A \ge A_0(L) \ge 2$,  it suffices to take 
$$
\rho_6 =  \tfrac18  \min(1, \rho_1, \rho_2, \rho_3) = \tfrac18 A^{-2}.
$$
In view of Lemma~\ref{lemmaS2}, there exists a neighbourhood $O_4 = B_{\rho_7}(0) \times B_{\rho_8}(0) \times B_{\kappa}$
such that $R_2(O_4) \subset B_{\rho_6}$. Indeed, we may take
$$ 
\rho_7 = \rho_8 = \frac{\rho_6}{C_{0, \eqref{boundss2}}} = \frac1{8 A^2 \, \,   C_{0, \eqref{boundss2}}}. 
$$
This defines the first restriction on the final neighbourhood $U_{\rho, \kappa}$, namely,
\begin{equation} \label{eq:first_cond_U_rho_kappa}
	U_{\rho, \kappa} \subset B_{\rho_7}(0) \times B_{\rho_8}(0) \times B_{\kappa}(0).
\end{equation}
	 	
The second restriction comes from the condition 
\begin{equation} \label{eq:second_cond_U_rho_kappa}
	P_3 (P_2 (E(-H), K))  \in B_{\rho_5}(0).
\end{equation}
To satisfy this condition, we note that by Lemma~\ref{le:P3} there exists a neighbourhood $O_5 =B_{\rho_9}(0)$ such that $P_3(O_5) \subset B_{\rho_5}$. 
It suffices to take 
$$ 
\rho_9 =\frac1{2B} \min(\rho_5, 1) = \frac1{2A}. 
$$
By Lemma~\ref{le:P2}, there exists a neighbourhood $O_6 = B_{\rho_{10}}(1) \times B_{\rho_{11}}(0)$
with $P_2(O_6) \subset B_{\rho_9}(0)$. Taking into account the bound  \eqref{eq:P2.secbd} and the fact that  $\rho_9 \le 1$, we may take
$$ 
\rho_{10} =  \frac{\rho_9}{4A} = \frac{1}{8 A^2},   \quad \rho_{11} = \frac{\rho_9}{4} = \frac{1}{8A}.
$$
Using once more Lemma~\ref{le:smoothness_exp}, we see that  the condition \eqref{eq:second_cond_U_rho_kappa} holds if 
$$
(H,K) \in B_{\rho_{12}}(0) \times B_{\rho_{11}}(0) \quad \hbox{with} \quad  \rho_{12} = \tfrac18 \rho_{10} = \frac{1}{64 A^2}.
$$
Combining this with \eqref{eq:first_cond_U_rho_kappa}, we see that the map $\myS_k$ is $C^\infty$ in the set 
$$
U_{\rho, \kappa} = B_\rho(0) \times B_\rho(0) \times B_{\kappa}(0)
$$
once
$$ 
\rho = \min(\rho_7, \rho_8, \rho_{11}, \rho_{12}) = \frac{1}{8 A^2} \min\Bigl(\frac{1}{C_{0, \eqref{boundss2}}}, \frac{1}{8}\Bigr).
$$
Since $A \ge A_0 \ge \AP \ge \AB$ we deduce from \eqref{eq:bound_R2_ell0} that
$$ 
C_{0, \eqref{boundss2}} \le  1 +   c_{2,\ell, \eqref{eq:discretebounds}  } +  C_{\eqref{eq:bound_Pi2}}.
$$
Thus we may take 
\begin{equation}  \label{eq:bound_rho_smooth}
	\rho = \frac{1}{8 A^2} 
	\min  \Bigl(  \frac{1}{  1 +   c_{2,\ell, \eqref{eq:discretebounds}}   +  C_{\eqref{eq:bound_Pi2}}}, 
	\tfrac{1}{8}\Bigr).
\end{equation}
Finally, the chain rule implies the estimate $\eqref{eq:propS:Claim}$.    
\end{proof}	


\chapter{Linearisation of the Renormalisation Map}\label{sec:contraction}
In this chapter we prove the bounds for the operator norms stated in Theorem  \ref{prop:contractivity}.
These contraction estimates make precise the idea  that $H$ and $K$ collect the relevant and the irrelevant terms, respectively. 
Throughout this chapter we assume that
$$
\boldsymbol{q} \in B_{ \kappa}
\text{ where } B_{\kappa} = B_\kappa(0) \text{ and } \kappa =  \kappa(L)\text{ is  defined in  Theorem~\ref{th:weights_final}.}
$$

\section{Bounds for the operator \texorpdfstring{$\protect\headingC$}{Cq}}
\label{subsec:contrC}

By  \eqref{eq:defofCk} we have, for $K\in M(\Pck)$,
\begin{align}  \label{eq:decomp_C}
(\boldsymbol{C}^{(\boldsymbol{q})} K)(U, \p) = F(U,\p) + G(U, \p)
\end{align}
where
$F\in M(\Pckp)$ is defined by
		\begin{align}\label{eq:definition_F_U_p}
			F(U,\p)=\sum_{\substack{X\in \Pck\setminus \mathcal{B} \\ \pi(X)=U}}\int_{\Xcal_N}K(X,\p+\xi)\,\mu_{k+1}^{(\boldsymbol{q})}(\d\xi).
		\end{align}
and 
\begin{align}  \label{eq:contraction_GBprime}
G(B', \p) =  \sum_{B \in \Bcal_k(B')}   G(B)(\p) 
\end{align}
with
 \begin{align}   \label{eq:contraction_GB}
 G(B)(\p) := (1- \Pi_2) \boldsymbol{R}_{k+1}^{(\boldsymbol{q})} K(B, \p).
 \end{align} 
if $B'$ is a $k+1$ block
while 
\begin{align}
G(U, \p) = 0 \quad \forall U \in \Pckp \setminus \mathcal{B}_{k+1}.
\end{align}
For ease of reading we  restate the key bound from Theorem~\ref{prop:contractivity} as Lemma~\ref{le:contr} below.
Recall the definition of the parameter $R$ in \eqref{eq:definition_R}
and let $\AB$ and $\AP(L)$ 
denote the constants which appear in the integration estimates in
Theorem~\ref{th:weights_final}~\ref{w:w7} and \ref{w:w8}.
Recall also that $\eta \in (0, \tfrac23]$ is a fixed parameter. This parameter actually controls  the contraction rate of the flow, see the definition of the norm \eqref{eq:norm_mcZ} in the set-up for the final fixed point argument.

\begin{lemma}\label{le:contr}
There exists an  $L_0$   such that for  each  $L\geq L_0$  there exists an $A_0(L)$ and a $h_0(L)$
with the property that for all $A \ge A_0(L)$ and all $h \ge h_0(L)$
\begin{align}  \label{eq:contraction_Ck}
\lVert\boldsymbol{C}^{(\boldsymbol{q})}\rVert^{(A)}=\sup_{\lVert K\rVert_{k}^{(A)}\leq 1}\lVert \boldsymbol{C}^{(\boldsymbol{q})}K\rVert_{k+1}^{(A)}
\leq \frac34 \eta   \qquad \hbox{for all  $\boldsymbol q \in B_\kappa$.}
\end{align}
We may take
 \begin{equation} \label{eq:L_0_single_block_contraction}
 L_0 = \max\bigl( (    4 \eta^{-1}   C' \AB C_1)^{\frac{1}{d'-d}}, 
  ( 32 \eta^{-1}       C' \AB (C_2 +1))^{\frac{2}{d}},  { 2^{d+3}+16R}  \bigr),
 \end{equation}
\begin{align}  \label{eq:lower_bound_A_contraction_rate} 
A_0(L) :=  \max\Bigl(   \frac{8}{\eta}  {\AP}^2L^d(2^{d+1}+1)^{d2^d}, 
\bigl( \frac{    8  \AP}{  {\eta} \updelta    }\bigr)^{\frac{1+2\upalpha}{ 2 \upalpha}} \Bigr)
\end{align}
 and $h_0(L)$ as in \eqref{eq:definition_h0} in  Theorem~\ref{th:weights_final}.
Here  $C_1 $ is the constant in the estimate
 $| (1- \Pi_2) K(B)|_{k+1,B,0} \le C_1 L^{-d'}  |K(B)|_{k,B,0}$
 in Lemma~\ref{le_contraction_I} 
 and $C_2$ is the constant in the estimate 
 $| \Pi_2 K(B)|_{k+1,B,0} \le C_2  |K(B)|_{k,B,0}$ 
 in Lemma~\ref{le:Pi2_bounded}.
Moreover  $d' = \lfloor d/2\rfloor + d/2  + 1$,
$ C' = \max_{x \ge 0}  (1+x)^5  e^{-\frac12 x^2}$, 
		and  $\upalpha$ and $\updelta$
		are the constants from Lemma \ref{le:app1} and Lemma \ref{le:app2}, respectively. 
\end{lemma}

The estimate \eqref{eq:L_0_single_block_contraction} says, roughly speaking, that the parameter $\eta$, 
which is a proxy for the contraction rate of the renormalisation flow, 
can be bounded by $\eta \le C L^{-1/2}$ for odd dimensions $d$ and by $C L^{-1}$ for even dimensions.
Actually,  the situation is more subtle:
to achieve the contraction rate  $L^{-1/2}$ or $L^{-1}$ in odd/even dimensions,   
our method of proof requires us  to choose the weighting parameter $A$ and the parameter $h_0$ sufficiently large, 
in dependence of $L$, see  \eqref{eq:lower_bound_A_contraction_rate}  and \eqref{eq:definition_h0}. 
The decay of $\eta$ as a function of $L$ could possibly be improved by extracting more
linear monomials in the definition of the projection $\Pi_2$, but we are not really interested in optimising the single-step contraction rate. 
Any value of $\eta \in (0, \frac23)$ is sufficient for the convergence of the overall renormalisation scheme.

There are two  different mechanisms that ensure contractivity of the  map $\boldsymbol{C}^{(\boldsymbol{q})}$. 
For the operator $F$ defined in \eqref{eq:definition_F_U_p} we use that the operation $\pi$ reduces the number of blocks, i.e., $|\pi(X)|_{k+1}<|X|_k$. 
The definition of the norm ensures that we gain a factor of $A^{|X|_k-|\pi(X)|_{k+1}}$ which can be used to cancel the combinatorial explosion of the number of terms. Here we need to choose $A$ sufficiently large to achieve a good contraction rate.
For the operator $G$, i.e., the contributions of single blocks,  this is not possible.
 For single block we use instead that $(1- \Pi_2) K$ measured at scale $k+1$ is much smaller than $K$ measured at scale
 $k$ (see Lemma~\ref{le_contraction_I}  and Lemma~\ref{le:contraction_single_block_prime} below). Here the decay like  $L^{-1/2}$ or $L^{-1}$ in odd and even dimensions, respectively, appears.  
 For the construction of the relevant weights for the large field regulator we need that $h$ is sufficiently large, 
 depending on $L$, \eqref{eq:definition_h0} in  Theorem~\ref{th:weights_final}. 
\medskip

We first consider the simpler large polymer term $F$.
\begin{lemma}\label{le:contrlarge}
Let $L\geq 2^{d+3}+16R$ and define   	
\begin{align}
A_0(L)  :=  \max \Bigl( \frac{8}{   {\eta}   } {\AP}^2L^d(2^{d+1}+1)^{d2^d}, 
\bigl(\frac{8 \AP}{ {\eta}    \updelta}\bigr)^{\frac{1+2\upalpha}{2\upalpha}} \Bigr)
\end{align}
where $\AP$ is the constant from     Theorem~\ref{th:weights_final}~\ref{w:w7} 
and  $\upalpha$ and $\updelta$
are the constants from Lemma \ref{le:app1} and Lemma \ref{le:app2}, respectively. 		
Then for all $A \ge A_0(L)$
\begin{align}
\lVert F\rVert_{k+1}^{(A)}\leq \tfrac{1}{4}{\eta}\lVert K\rVert_{k}^{(A)}. 
\end{align}
\end{lemma}
	
	\begin{proof}
		Lemma~\ref{le:submultofsimplenorm} 
		states that for $U = \pi(X)$ 
		\begin{align}
			\abs{\boldsymbol{R}_{k+1}^{(\boldsymbol{q})}K(X,\p)}_{k+1,U,T_\p}\leq
			 \abs{\boldsymbol{R}^{(\boldsymbol{q})}_{k+1}K(X,\p)}_{k,X,T_\p}. 
		\end{align}
		The inequality \eqref{eq:w6} in      Theorem~\ref{th:weights_final}~\ref{w:w6} 
		 implies that
		\begin{align}
			w_{k:k+1}^X(\p)\leq w_{k+1}^U(\p). 
		\end{align}
		We conclude that
		\begin{align}\begin{split}
			\lVert\boldsymbol{R}_{k+1}^{(\boldsymbol{q})}&K(X,\p)\rVert_{k+1,U}
			   =                                            
			\sup_{\p\in \Xcal_N} \frac{\left|\boldsymbol{R}_{k+1}^{(\boldsymbol{q})}K(X,\p)\right|_{k+1,U,T_\p}}{w_{k+1}^U(\p)}	
			  \\ &
			   \leq                                         
			\sup_{\p\in \Xcal_N}\frac{\left|\boldsymbol{R}_{k+1}^{(\boldsymbol{q})}K(X,\p)\right|_{k,X,T_\p}}{w_{k:k+1}^X(\p)}
		 =\lVert \boldsymbol{R}_{k+1}^{(\boldsymbol{q})}K(X,\p)\rVert_{k:k+1,X}. 
			\end{split}\end{align}
			Using this bound we can estimate
			\begin{align}\label{eq:C:1}
\begin{split}			
			A^{|U|_{k+1}}&\lVert F(U)\rVert_{k+1,U}
			\\ &
			\leq 
			A^{|U|_{k+1}}\bigg(\sum_{\substack{X\in\Pck\setminus \mathcal{S}_k\\
			\pi(X)=U}} \lVert \boldsymbol{R}_{k+1}^{(\boldsymbol{q})}K(X)\rVert_{k:k+1,X}+
			\sum_{\substack{X\in \mathcal{S}_k\setminus \mathcal{B}_k\\
			\pi(X)=U}} \lVert \boldsymbol{R}_{k+1}^{(\boldsymbol{q})}K(X)\rVert_{k:k+1,X}\bigg).
	\end{split}	\end{align}
		We bound the two summands separately. 		
		For the first term we use the bound $|\pi(X)|_{k+1} \le \frac{1}{1 + 2 \upalpha} |X|_k$ in  Lemma \ref{le:app1} and Lemma \ref{le:app2}. 
		Bounding in addition the map $\boldsymbol{R}_{k+1}^{(\boldsymbol{q})}$ 
		using  Lemma~\ref{le:keyboundRk} we infer that
		\begin{align}  \label{eq:bound_large_small_F}
			\begin{split}
			&A^{|U|_{k+1}}\sum_{\substack{X\in\Pck\setminus \mathcal{S}_k\\
			\pi(X)=U}} \lVert \boldsymbol{R}_{k+1}^{(\boldsymbol{q})}K(X)\rVert_{k:k+1,X}
\\ &			 
			 \leq                                                                                         
			A^{|U|_{k+1}}\sum_{\substack{X\in\Pck\setminus \mathcal{S}_k\\
			\pi(X)=U}} {\AP}^{|X|_k}\lVert K(X) \rVert_{k,X}
		 \leq A^{|U|_{k+1}}\sum_{\substack{X\in\Pck \setminus \mathcal{S}_k \\
			\pi(X)=U}} \lVert K\rVert_{k}^{(A)}\Bigl(\frac{\AP}{A}\Bigr)^{|X|_k}
	\\ &
			 \leq                                                                                         
			\lVert K\rVert_{k}^{(A)}\sum_{\substack{X\in\Pck \setminus \mathcal{S}_k\\
			\pi(X)=U}} \bigl(\AP{A}^{-\frac{2\alpha}{1+2\alpha}}\bigr)^{|X|_k}
			 \leq                                                                                         
			\tfrac{1}{8}    \eta  \lVert K\rVert_{k}^{(A)}
			\end{split}
		\end{align}
		for $A\geq \bigl(\frac{8     \AP}{     \eta      \updelta}\bigl)^{\frac{1+2\upalpha}{2\upalpha}}$.
		For the second contribution we observe that $\pi(X)$ is a single block for $X\in \mathcal{S}_k$, i.e.,
		 the second summand in 
		\eqref{eq:C:1} is only non-zero if $U\in \mathcal{B}_{k+1}$. 
		Moreover we can bound the number of small polymers $X$ that intersect a block $B'\in \mathcal{B}_{k+1}$ by
		$L^d(2^{d+1}+1)^{{d}2^d}$.
		 Indeed there are $L^d$ possibilities to pick the first block $B$ of $X$ and 
		 then all further blocks are 	contained in 
		a cube of side-length $(2^{d+1}+1)L^{k}$ centred at $B$ and there are at most $2^d$ of them.
		
		This implies for $U\in \mathcal{B}_{k+1}$ and $A\geq \AP$
		\begin{align}
			\begin{split}
			&A^{|U|_{k+1}}\hspace{-0.25cm}
			\sum_{\substack{X\in \mathcal{S}_k\setminus \mathcal{B}_k\\
			\pi(X)=U}} \lVert \boldsymbol{R}_{k+1}^{(\boldsymbol{q})}K(X)\rVert_{k:k+1,X}
			   \leq 
			A
			\sum_{\substack{X\in \mathcal{S}_k\setminus \mathcal{B}_k\\
			\pi(X)=U}} {\AP}^{|X|_k}\lVert K(X)\rVert_{k,X}
			\\			  & 
			  \leq 
			A\lVert K\rVert_{k}^{(A)}
			\sum_{\substack{X\in \mathcal{S}_k\setminus \mathcal{B}_k\\
			\pi(X)=U}} \Bigl({\frac{\AP}{A}}\Bigr)^{|X|_k}
			 \leq 
			A\lVert K\rVert_{k}^{(A)} L^d(2^{d+1}+1)^{{ d}2^d}\frac{{\AP}^2}{A^2}
			 \leq 
			\tfrac{1}{8}   \eta    \lVert K\rVert_{k}^{(A)}
			\end{split}
		\end{align}
		for $A\geq 8 \eta^{-1}  {\AP}^2L^d(2^{d+1}+1)^{{ d}2^d}$.
	\end{proof}
	
	Next we consider the contribution from single blocks.
	Recall  from   \eqref{eq:contraction_GB}   that for a $k$-block $B$ we defined $G(B)(\p) = (1- \Pi_2) \boldsymbol{R}_{k+1}^{(\boldsymbol{q})} K(B, \p)$.
	
\begin{lemma}  \label{le:contraction_single_block_new} Assume that $L\geq 2^{d+3}+16R$. Then we have
 \begin{equation} \label{eq:main_pointwise_contraction_estimate}
 |G(B)|_{k+1,B, T_\p} \le \AB  (1 + |\p|_{k+1,B})^5 \, \bigl(C_1 L^{-d'} +   8 (C_2+1) L^{-\frac32 d}  w^B_{k:k+1}(\p) \bigr) \,   \|K\|_{k,B}. 
 \end{equation}
 where 
 \begin{equation}
 d' = d/2 + \lfloor d/2\rfloor +1 > d
 \end{equation}
and where  $\AB$ is the constant which appears in the integration estimate for the weights in
Theorem~\ref{th:weights_final}~\ref{w:w8}.
 The constant $C_1$ is the constant in the estimate
 $| (1- \Pi_2) K(B)|_{k+1,B,0} \le C_1 L^{-d'}  |K(B)|_{k,B,0}$
 in Lemma~\ref{le_contraction_I} while the constant $C_2$ is the constant in the estimate 
 $| \Pi_2 K(B)|_{k+1,B,0} \le C_2  |K(B)|_{k,B,0}$ 
 in Lemma~\ref{le:Pi2_bounded}.
  \end{lemma}
 
 \begin{proof}  From the two norm estimate   \eqref{eq:two_norm_concrete}
 and  the contraction estimate   \eqref{eq:contraction_T0}
  we get 
 \begin{align}
\begin{split} 
  \,  |G(B)|_{k+1, B, T_\p}  &\le \,  (1 + |\p|_{k+1,B})^3  \,  \bigl( | (1 - \Pi_2) \boldsymbol R_{k+1} K(B) |_{k+1,B, T_0} 
  \\ &
 \hspace{1.5cm} +
 8  L^{-\frac32 d} \sup_{0 \le t \le 1}  |(1- \Pi_2) \boldsymbol R_{k+1}K(B)|_{k,B,T_{t\p}} \bigr)     \\
 &\le \,  (1 + |\p|_{k+1,B})^3  \,  \bigl(   C_1  L^{-d'}  |\boldsymbol R_{k+1} K(B)|_{k,B,T_0} 
\\ & 
 \hspace{1.5cm} + 
 8  L^{-\frac32 d} \sup_{0 \le t \le 1}  |(1- \Pi_2) \boldsymbol R_{k+1}K(B)|_{k,B,T_{t\p}} \bigr)
 \label{eq:contraction_estimate_G}
\end{split}
 \end{align}
 Now by Jensen's inequality and the estimate \eqref{eq:w8} in 
Theorem~\ref{th:weights_final}~\ref{w:w8}
  with $\p=0$ 
 \begin{align}
   |\boldsymbol R_{k+1} K(B)|_{k,B,T_0} \le & \,  \int_{\BX_N} |K(B)|_{k,B, T_\xi} \, \mu_{k+1}(\d \xi) \notag   \\
  \le & \,  \int_{\BX_N} \|K\|_{k,B} \, w_{k}^B(\xi) \, \, \mu_{k+1}(\d \xi) 
  \le  \AB \,   \| K\|_{k,B}.
  \label{eq:contraction_estimate_RK0}
 \end{align}

 The second term is bounded similarly.
 By \eqref{eq:est_HTp}, \eqref{eq:ell_phi_vs_sup_phi}, Lemma~\ref{le:Pi2_bounded} and  \eqref{eq:contraction_estimate_RK0} we get, for all $t \in [0,1]$,
 \begin{align}
& \,  | \Pi_2 \boldsymbol R_{k+1}K(B)|_{k,B,T_{t\p}} \le (1 +|\p|_{k,B})^2 \| \Pi_2 \boldsymbol R_{k+1} K(B)\|_{k, 0}   \notag \\
 \le & \,  C_2 (1 +|\p|_{k,B})^2 | \boldsymbol R_{k+1} K(B)|_{k, B, T_0}  
 \le C_2  (1 +|\p|_{k,B})^2 \AB \,   \| K\|_{k,B}.
 \label{eq:contraction_estimate_PRK}
 \end{align}

 Using the monotonicity of $t \mapsto w_{k:k+1}(t\p)$ we get  
  from   \eqref{eq:bdRkl0}  in Lemma~\ref{le:keyboundRk}
 \begin{align}
  | \boldsymbol R_{k+1}K(B)|_{k,B,T_{t\p}} 
 \le  w_{k:k+1}^B(\p) \,  \|  \boldsymbol R_{k+1} K(B)  \|_{k:k+1,B}   
 \le      \AB \,  w_{k:k+1}^B(\p) \,  \| K(B)  \|_{k,B}.
 \label{eq:contraction_estimate_RK}
 \end{align}
Since $|\p|_{k,B} \le |\p|_{k+1, B}$ the estimate  \eqref{eq:main_pointwise_contraction_estimate}
now follows from \eqref{eq:contraction_estimate_G},    \eqref{eq:contraction_estimate_RK0},
 \eqref{eq:contraction_estimate_PRK}  and 
\eqref{eq:contraction_estimate_RK}.
 \end{proof}

\begin{lemma}   \label{le:contraction_single_block_prime}
Assume that $L\geq 2^{d+3}+16 R$ and that $h\geq h_0(L)$ where $h_0(L)$ satisfies \eqref{eq:definition_h0}.
Let 
 $B' \in \mathcal B_{k+1}$ be a  $k+1$ block   and recall that 
$ G(B') = \sum_{B \in   \Bcal_k(B')} G(B)$.  Then 
\begin{equation} \label{eq:main_pointwise_contraction_estimate2}
 |G(B)|_{k+1,B',T_\p} \le C'  \AB    \bigl(C_1 L^{-d'} + 8  (C_2+1) L^{-\frac32 d}  \bigr) \,  w_{k+1}^{B'}(\p) \,   \|K\|_{k,B}. 
 \end{equation}
 and
 \begin{equation}
 \| G(B') \|_{k+1,B'} \le  \label{eq:main_pointwise_contraction_estimate3}
C' \AB   \bigl(C_1 L^{d-d'} +  8 (C_2+1)  L^{-\frac12 d}  \bigr) \,   \|K\|_{k,B}.
 \end{equation}
 where 
 $$ C' = \max_{x \ge 0}  (1+x)^5  e^{-\frac12 x^2}.$$
 In particular there exists an $L_0$   
  such that for $L \ge L_0$ and $h \ge h_0(L)$   
 \begin{align}   \label{eq:contraction_G_blocks}
 \| G \|_{k+1}^{(A)} \le   \tfrac12 \eta    \|K\|_k^{(A)}  \quad \forall A \ge 1.
 \end{align}
 We may take
 \begin{equation} \label{eq:L_0_single_block_contraction2}
 L_0 = \max\bigl(  
     (   4 \eta^{-1}       C' \AB C_1)^{\frac{1}{d'-d}},  
 (     32 \eta^{-1}      C' \AB (C_2 +1))^{\frac{2}{d}}, 2^{d+3}+16R \bigr).
 \end{equation}
 \end{lemma}

\begin{proof}
Indeed by 
  \eqref{eq:w9} in Theorem~\ref{th:weights_final}~\ref{w:w9}     
and the definition of $C'$ we have 
$$(1+ |\p|_{k+1, B'})^5 \le (1+ |\p|_{k+1, B'})^5 \,  w_{k:k+1}^{B'}(\p)  \le C' w_{k+1}^{B'}(\p). $$
Since $|\p|_{k+1,B} \le |\p|_{k+1,B'}$ and $w^B_{k:k+1} \le   w^{B'}_{k:k+1}$
the estimate  \eqref{eq:main_pointwise_contraction_estimate2} follows from 
Lemma~\ref{le:contraction_single_block_new}.
Now
 \eqref{eq:main_pointwise_contraction_estimate3}
follows from \eqref{eq:main_pointwise_contraction_estimate2} after summing over $B$, dividing by
$w_{k+1, B'}(\p)$ and taking the supremum over $\p$. 
Finally   \eqref{eq:contraction_G_blocks} holds  if we  take $L_0$ so large that
\begin{equation}  \label{eq:bound_L0_lemma_single_block}
C' \AB C_1 L_0^{d-d'} \le 
\frac14 \eta
 \quad  \hbox{and}   \quad 8 C' \AB (C_2 +1) L_0^{-d/2} \le 
\frac14 \eta.
\end{equation}
Clearly both conditions are satisfied if $L$ satisfies $L \ge L_0$ and $L_0$ is the number in 
\eqref{eq:L_0_single_block_contraction2}.
\end{proof}

\begin{proof}[Proof of Lemma~\ref{le:contr}]
This follows from   \eqref{eq:decomp_C},
Lemma~\ref{le:contrlarge} and Lemma~\ref{le:contraction_single_block_prime}.
\end{proof}

\section{Bound for the operator \texorpdfstring{$\bigl(\protect\headingA\bigr)^{-1}$}{(Aq)-1}   }
\label{se:bound_Aq} 

\begin{lemma}\label{le:Aq} Let $C_{2,0}$ be the constant in 
 \eqref{eq:discretebounds} for $\ell = 0$. Then for  
 \begin{align}   \label{eq:condition_h_bound_A_inverse}
 h \ge \sqrt{C_{2,0}}  
 \end{align}
and $h_k=2^kh$ the operator $\boldsymbol{A}^{(\boldsymbol{q})}: (M_0(\Bcal_k), \| \cdot \|_{k,0})\to
(M_0(\Bcal_{k+1}), \| \cdot \|_{k+1})$ satisfies 
	\begin{align}
		\norm{\bigl(\boldsymbol{A}^{(\boldsymbol{q})}\bigr)^{-1}} \leq \frac{3}{4}. 
	\end{align}
\end{lemma}
\begin{proof}
	Let  $H'=  \boldsymbol{A}^{(\boldsymbol{q})}   H$. 
	As before we  denote the coefficients of  the expansion of $H$  and $H'$ in monomials by $a_\mpzc$ and $a'_\mpzc$, respectively. 
	Here $\mpzc \in \mf v$. 
	By \eqref{eq:defofAk}  we have $a'_\mpzc = a_\mpzc$ for $\mpzc \neq \emptyset$ and
	\begin{align}
	a'_\emptyset = a_\emptyset  + 		\sum_{(i, \alpha), (j, \beta) \in \mf v_2}   
		a_{(i, \alpha), (j, \beta)}    \,    (\nabla^\beta)^\ast \nabla^\alpha \mathcal{C}^{(\boldsymbol{q})}_{k+1, ij}(0).
	\end{align}
	Thus $\boldsymbol A := \boldsymbol A^{(\boldsymbol{q})}$ is invertible and by the definition \eqref{hamiltoniannorm} of the $\| \cdot \|_{k,0}$ norm
	in  connection with the relations $h_{k+1} \ge 2 h_k$  and $L \ge 2$ we get 
	 \begin{align} 
	 \begin{split}
 \| H \|_{k,0} = & \, L^{kd} \, |a_\emptyset| + \sum_{(i, \alpha) \in \mf v_1}
h_k L^{kd} L^{- k \frac{d-2}{2}} L^{-k |\alpha|} \, |a_{i, \alpha}| + \sum_{\mpzc \in \mf v_2}  
h_k^2\,  |a_\mpzc| \\
\le & \, L^{kd} \, |a'_\emptyset| + \sum_{(i, \alpha) \in \mf v_1}
h_k L^{kd} L^{- k \frac{d-2}{2}} L^{-k |\alpha|} \, |a'_{i, \alpha}| + \sum_{\mpzc \in \mf v_2}  
h_k^2\,  |a'_\mpzc|   \\
& \, + L^{kd} \sum_{(i, \alpha), (j, \beta) \in \mf v_2}   
		|a_{(i, \alpha), (j, \beta)} |   \,    |(\nabla^\beta)^\ast \nabla^\alpha \mathcal{C}^{(\boldsymbol{q})}_{k+1, ij}(0)| \\
\le & \,  \tfrac12 \| H'\|_{k+1,0} + L^{kd} \sum_{(i, \alpha), (j, \beta) \in \mf v_2}   
		|a_{(i, \alpha), (j, \beta)} |   \,    |(\nabla^\beta)^\ast \nabla^\alpha \mathcal{C}^{(\boldsymbol{q})}_{k+1, ij}(0)| 
		\end{split}
\end{align}
	The bound \eqref{eq:discretebounds} implies that for $((i, \alpha), (j, \beta)) \in \mf v_2$ 
	\begin{align}
		   \left|   (\nabla^\beta)^\ast \nabla^\alpha \mathcal{C}^{(\boldsymbol{q})}_{k+1, ij}(0) \right|
     \leq C_{2,0}\,  L^{-kd}. 
	\end{align}
	Using in addition that
	\begin{align}
	\sum_{(i, \alpha), (j, \beta) \in \mf v_2}   
		|a_{(i, \alpha), (j, \beta)}  |  
		= \sum_{(i, \alpha), (j, \beta) \in \mf v_2}   
		|a'_{(i, \alpha), (j, \beta)}  |  
		  \leq \frac{\lVert H'\rVert_{k+1,0}}{h_{k+1}^2} 
	\end{align}
	and $h_{k+1} = 2^{k+1} h \ge 2 h$ 
	we conclude that 
	\begin{align}
		\Vert \boldsymbol{A}^{-1}H'\rVert_{k,0}\leq \frac{1}{2}\lVert H'\rVert_{k+1,0}+
		\frac{C_{2,0} \lVert H'\rVert_{k+1,0}}{h_{k+1}^2}\leq \frac{3}{4}\lVert H'\rVert_{k+1,0}                                                                                                                              
	\end{align}
	provided that $h^2 \ge C_{2,0}$.
\end{proof}

\section{Bound for the operator \texorpdfstring{$\protect\headingB$}{Bq}}
\label{se:bound_Bq}

Recall from   \eqref{eq:defofBk}  that $\boldsymbol{B}_k^{(\boldsymbol{q})}: (M(\Pck), \| \cdot\|_k^{(A)}) \to (M_0(\Bcal_{k+1}), \| \cdot \|_{k+1,0})$ is defined by
\begin{align}
		(\boldsymbol{B}_k^{(\boldsymbol{q})}  {K})(B',\p) & =-\sum_{B\in\mathcal{B}_k(B')}
		\Pi_2\Bigl(\int_{\Xcal_N}   {K}(B,\p+\xi)\, \mu_{k+1}^{(\boldsymbol{q})}(\d\xi)\Bigr) 
\end{align}

\begin{lemma}  \label{le:bound_BQ}
Assume that 
\begin{equation} \label{eq:condition_L_bound_Bq}
L \ge  2^{d+3}+16R,
\end{equation}
and
\begin{equation}   \label{eq:condition_A_bound_Bq}
A\geq A_0 :=   3 C_2  \AB L^d    
\end{equation}
	where $C_2$ is the constant in Lemma~\ref{le:Pi2_bounded} and $\AB$ is the constant in
	Theorem~\ref{th:weights_final}~\ref{w:w8}.      
	Then the  operator norm of $\boldsymbol{B}^{(\boldsymbol{q})}$ satisfies
	\begin{align}
		\lVert \boldsymbol{B}^{(\boldsymbol{q})}\rVert\leq \frac{C_2 \AB L^d}{A}\leq \frac{1}{3}.
	\end{align}
\end{lemma}

\begin{proof}

Set $H'(B') = (\boldsymbol{B}_k^{(\boldsymbol q)} K)(B')$. For a $B \in \mathcal B_k(B')$ set $H(B) = - \Pi_2 {\boldsymbol R}^{(\boldsymbol q)}K(B)$. 
Then $H(B)$ can be written as 
$$ 
H(B) =   \sum_{x \in B}  \sum_{\mpzc  \in \mf v}   a_\mpzc   \, \Mscr_{\mpzc}(\{x\}).  
$$
By translation invariance $H'(B')$ can be written as
$$
 H'(B') =   \sum_{x \in B'}  \sum_{\mpzc  \in \mf v}  a_\mpzc   \, \Mscr_{\mpzc}(\{x\}) 
 $$
with the \emph{same} coefficients $a_\mpzc$.
Thus it follows from the definition  \eqref{hamiltoniannorm} of the norm $\| \cdot \|_{k,0}$
on relevant Hamiltonians  and the  relation  $h_{k+1} = 2 h_k$ that 

\begin{align}
\| \boldsymbol{B}^{(\boldsymbol{q})} K \|_{k+1, 0} \le \max(L^d, 2 L^{d/2}, 4)  \, \| \Pi_2 \boldsymbol R_{k+1}^{(\boldsymbol{q})} K(B)\|_{k,0}
\le L^d \, \| \Pi_2 \boldsymbol R_{k+1}^{(\boldsymbol{q})} K(B)\|_{k,0}.
\end{align}

Lemma~\ref{le:Pi2_bounded} and   \eqref{eq:contraction_estimate_RK0}
(which is a consequence of  
\eqref{eq:w8}) imply that
\begin{align}
\norm{\Pi_2(\boldsymbol{R}_{k+1}^{(\boldsymbol{q})}K)(B)}_{k,0}
  \le C_2           |{\boldsymbol{R}_{k+1}^{(\boldsymbol{q})} }K(B)|_{k,B,T_0} 
 \le C_2        \AB \|K {(B)}\|_{k,B}. 
\end{align}
Since $\|K{(B)}\|_{k, B} \le A^{-1} \|K\|_k^{(A)}$ the desired assertion follows.
\end{proof}

\chapter{Proofs of the Main Results}  \label{sec:proofs}

In this chapter we first state the final  representation formula for the partition functions in
\eqref{thmequality} and \eqref{thmequality2} below (cf. \eqref{eq:Zpertcomp_general}).
We then show that our 
main results on the free energy ${\mathcal W}_N(\mathcal K,\mathcal Q)$ (Theorem ~\ref{th:pertcomp_E}) and on the scaling limit (Theorem~\ref{th:scalinglimit}) follow easily from these representation formulas.

\section{Main result of the renormalisation analysis}
We fix $\zeta \in (0,1)$ and we recall from \eqref{eq:normE} that the Banach space $\boldsymbol{E}$ consists of functions
$\mathcal K:\Gcal=\left(\mathbb{R}^m\right)^{\Ical}\rightarrow \mathbb{R}$ such that that the following norm is finite
\begin{align}
	\norm{\mathcal  K }_\zeta = \sup_{z\in \Gcal} \sum_{|\alpha|\leq r_0} \frac1{\alpha !}\abs{\partial^\alpha 
	\mathcal K(z)}e^{-\frac{1}{2}(1-\zeta)\Qscr(z)} .                                           
\end{align}
 Recall that $\eta\in (0,\frac23 ]$ is a parameter controlling the rate of contraction of the renormalisation flow.
 (Cf. the definition of the norm \eqref{eq:norm_mcZ} introduced in the next chapter formalising this notion.)
 
\begin{theorem}\label{maintheorem}\label{MAINTHEOREM}
Let $\kappa = \kappa(L)$ be as in Theorem~\ref{th:weights_final}.
Let $L_0$, $h_0(L)$, $A_0(L)$, $\rho(A)$, $C_{j_1,j_2,j_3}(L,A)$, and $C_\ell(L,A)$ be such that the conclusions of 
 Theorem~\ref{prop:smoothnessofS}  and Theorem~\ref{prop:contractivity} hold for every triplet $(L,h,A)$ 
 with $L \ge L_0$, $h \ge h_0(L)$, $A \ge A_0(L)$.
  Assume also that
  \begin{equation}
 h_0(L) \ge \delta(L)^{-\frac12}
 \end{equation}
 where $\delta(L)$ is the constant from  Lemma~\ref{prop:W1} 
 chosen in \eqref{eq:defofdelta}.
 
 Then for every triplet $(L,h,A)$ with $L\geq L_0$, $h\geq h_0(L)$, and $A\geq A_0(L)$  
 there exists a $\rhoMT = \rhoMT(L,h,A) > 0$ such that for each $N \ge 1$ there are $C^\infty$ maps  
 $\widehat  e_N:B_\rhoMT(0)\subset \boldsymbol{E}\rightarrow \mathbb{R}$, 
$\widehat{\boldsymbol{q}}_N:B_\rhoMT(0)\subset\boldsymbol{E}\rightarrow B_\kappa(0)\subset\mathbb{R}^{(d\times m)\times (d\times m)}_{\textrm{sym}}$ 
and $\widehat K_N:B_\rhoMT(0) \subset \boldsymbol{E}\rightarrow \boldsymbol{M}_N^{(A)}$ 
(defined in \eqref{eq:def_of_Banach_spaces}) with the following properties.  
For each  $\mathcal{K}\in B_\rhoMT(0) \subset \boldsymbol{E}$,
	\begin{align}\label{thmequality} \begin{split}
		\int_{\Xcal_N} \sum_{X\subset T_N}\prod_{x\in X} &\mathcal{K}(D\p(x))\,\mu^{(0)}(\d \p)
		\\ & 
		=                                                        
		    \frac{ Z_N^{ (\widehat{\boldsymbol q}_N(\mathcal K))}  		%
		e^{L^{Nd} \widehat  e_N(\mathcal{K})}  }{   Z_N^{(0)}    }
		\int_{\Xcal_N}\left(1+
		\widehat K_N(\mathcal{K})(\TN,\p)\right)\,  \, \mu_{N+1}^{(   \widehat {\boldsymbol q}_N(\mathcal{K}))  }(\d\p),
		\end{split}
	\end{align}
where $Z_N^{ (\boldsymbol q)}$ denotes the normalisation introduced in \eqref{eq:measurenormalisation}.
The derivatives of these maps satisfy bounds that are uniform in $N$
and the map  $\widehat K_N$ is contracting in the sense that there are constants $\overline{C}_{\ell}(L, h, A)>0$ such that 
	\begin{align}\label{eq:keyboundKN}
	\frac{1}{\ell!} 	\lVert \partial^\ell_{\mathcal{K}}\widehat K_N(\mathcal{K})(\dot{\mathcal K}, \ldots, \dot{\mathcal K})\rVert_{N}^{(A)}\leq \overline{C}_{\ell}(L, h, A) \; \eta^N \, \| \dot{\mathcal K}\|_{\zeta}^\ell
	\end{align}
for every  $\ell \ge 0$. 

Moreover, 
	\begin{equation} \label{eq:keyboundKN_ell0}
          \int_{\Xcal_N} |  \widehat K_N(\mathcal K)(\TN, \p) | \, 
          \mu_{N+1}^{(\widehat {\boldsymbol q}_N(\mathcal K))}(d\p) \le \frac12.
	\end{equation}
More generally, the following identity holds for $f_N\in \Xcal_N$ and $\mathcal K \in B_\rhoMT(0)$,
	\begin{align}\label{thmequality2}
		\begin{split}
		\int_{\Xcal_N} e^{(f_N,\p)}\sum_{X\subset T_N}  &\prod_{x\in X} \mathcal{K}(D\p(x))\,\mu^{(0)}(\d \p) 
		                                                 =  \,                                               
		e^{\frac{1}{2}(f_N,\mathscr{C}^{(\widehat {\boldsymbol{q}}_N(\mathcal{K}))}f_N)}
		\frac{Z_N^{(\widehat {\boldsymbol{q}}_N(\mathcal{K}))}
		e^{L^{Nd}\widehat e_N(\mathcal{K})}}{Z_N^{(0)}}
		\\ &
		\int_{\Xcal_N}\left(1+\widehat K_N(\mathcal{K})(\TN	,	
			\p + \mathscr { C} ^ { (\widehat{\boldsymbol{q}}_N(\mathcal{K})) } f_N)\right)\, 
			\mu_ { N+1 } ^ { (\widehat{\boldsymbol{q}}_N(\mathcal{ K } )) } (\d\p).
		\end{split}
	\end{align}
\end{theorem}
Actually the proof shows that we may take $\rhoMT$ as the minimum $\tilde{\rhoMT}$ in Lemma~\ref{lemmafixedpoint}
and $\frac{ A}{2 \AB C_{1, \eqref{eq:keyboundKN}}} \eta^{-N}$.
Thus for $N \ge N_0(L,h,A)$ we may take $\rhoMT$ simply as in Lemma~\ref{lemmafixedpoint}.

We will prove this theorem at the end of  Chapter~\ref{sec:finetuning}.
In the remainder of the current  chapter we show how Theorem~\ref{maintheorem} implies the main results in Chapter~\ref{sec:setting}.

\section{Proof of the main theorem}  \label{sec:proof_GGM}

\begin{proof}[Proof of Theorem \ref{th:pertcomp_E}]
Choose  the parameter $\rhoMT$ in the statement of Theorem~\ref{th:pertcomp} as the number $\rhoMT(L, h_0(L), A_0(L))$ in Theorem \ref{maintheorem}. 
We apply first \eqref{eq:defvarsigmaN_general} and then  \eqref{eq:Zpertcomp_general}
and \eqref{thmequality} from Theorem~\ref{maintheorem} and obtain that the perturbative free energy can be expressed as
	\begin{align}
		\begin{split}\label{eq:pertcompdecomp}
		\overline{\myW}_N(\Kcal) \ &{=} -\frac{1}{L^{Nd}}\ln \myZ_N(\Kcal,\Qscr,0) \\
	         &  {=}                                 
		-\widehat  e_N(\mathcal{K})-\frac{1}{L^{Nd}}\ln\Bigl(\frac{Z_N^{(\widehat{\boldsymbol{q}}_N(\mathcal{K}))}}{Z_N^{(0)}}\Bigr)\\ 
		&  \qquad - \frac{1}{L^{Nd}}\ln\Bigl(\int_{\Xcal_N}\left(1+\widehat K_N(\mathcal{K})(\TN , \p)\Bigr)\,
		\mu_ { N+1}^{ (  \widehat{\boldsymbol{q}}_N(\mathcal { K })) } (\d\p)\right).
		\end{split}
	\end{align}
	The first term is $C^{\infty}$ uniformly in $N$ by Theorem \ref{maintheorem}.
	Similarly the second term is $C^{\infty}$ uniformly in $N$ by Theorem \ref{maintheorem},
	 Lemma~\ref{le:smoothness_Zq_div_Z0_term} below, and the chain rule.

	To address the last term we introduce the shorthand
	\begin{align}
		G(K_N,\boldsymbol{q})=\int_{\Xcal_N} K_N(X,\p)\,\mu_{N+1}^{(\boldsymbol{q})}(\d \p)=\myR_{N+1}^{(\boldsymbol{q})}K_N(\TN,0). 
	\end{align}
	Then the last term equals
	 $ L^{-Nd} \ln (1+G(\widehat K_N(\mathcal{K}),  \widehat{\boldsymbol{q}}_N(\mathcal{K})))$.
	Note that for any positive function $G$ the  derivative
	$D^k\ln (1+G)$ is given by a polynomial in derivatives of $G$ divided by $(1+G)^k$.
	It follows from  \eqref{eq:keyboundKN_ell0} that $1 + G \ge \frac12$. 
By the chain rule it is sufficient to show that $G: B_\kappa(0)\times \boldsymbol{M}_N^{(A)}\to \mathbb{R}$ 
is smooth because $\widehat{\boldsymbol{q}}(\Kcal)$ and $\widehat{K}_N(\Kcal)$ are smooth functions.  
For the derivatives with respect to $\boldsymbol{q}$ we use \eqref{eq:bdRkl1} from Lemma~\ref{le:keyboundRk} to estimate
	\begin{align}
		\begin{split}
		|\partial_{\boldsymbol{q}}^\ell G(K_N,\boldsymbol{q}))| 
		& \leq  \lVert \partial^\ell_{\boldsymbol{q}} \myR_{N+1}^{(\boldsymbol{q})}	K_N(\TN)\rVert_{N:N+1,\TN}   \\
		& \leq \overline{C}_{\ell}\,  \frac{\AB}{A}\lVert \widehat K_{N}\rVert_{N}^{(A)}.
		\end{split}
	\end{align}
Thus, we have established that $\overline{\mathcal W}_N $ is $C^{\infty}$ with uniform bounds.
\end{proof}
To show smoothness  of the second term on the right hand side of \eqref{eq:pertcompdecomp} we used the following result.
\begin{lemma}  \label{le:smoothness_Zq_div_Z0_term}
Let $F_N(\boldsymbol q) = \frac{1}{L^{Nd}}\ln\bigl(\frac{Z^{(\boldsymbol{q})}}{Z^{(0)}}\bigr). $
Then $F_N \in C^\infty(B_{\omega_0/2}(0))$ and the derivatives of $F_N$ can be bounded uniformly in $N$. 
\end{lemma}

\begin{proof}
To emphasise the dependence on $N$,  we temporarily use $\mathscr{A}^{(\boldsymbol{q})}_N$ 
to denote the operator  $\mathscr{A}^{(\boldsymbol{q})}$  on $L^2(\Xcal_N)$ defined in \eqref{eq:definition_Aq}.   
Fourier transform diagonalises this operator in the scalar case $m=1$ and block-diagonalises it for general $m$ (with $m \times m$ blocks). 
By \eqref{eq:fourierA} the Fourier transform is given by 
\begin{align}\label{eq:fourierA_repeated}
	\widehat{\mathscr{A}^{(\boldsymbol{q})}_N}(p)=\sum_{\alpha,\beta\in \Ical} \overline{q}(p)^\alpha \boldsymbol{Q}_{\alpha\beta}
	q(p)^\beta { +} \sum_{|\alpha| = |\beta|=1} \overline{q}(p)^\alpha \boldsymbol q_{\alpha \beta}      q(p)^\beta,                                                   
\end{align}
where $\boldsymbol q_{\alpha \beta}$ denotes the $m\times m$ matrix with entries $\boldsymbol q_{(\alpha,i) (\beta,j)}$
and the $j$-th component of $q(p)$ is given by $q_j(p) = e^{i p_j} -1$. 
Since $\boldsymbol q \in B_{\omega/2}(0)$, it follows from  \eqref{Ahatestimate} that 
$\mathscr{A}^{(\boldsymbol{q})}_N$ is positive definite and Gaussian calculus gives
	\begin{align}
		\begin{split}\label{eq:detQuotient}
		F_N(\boldsymbol q) 
		  & =\frac{1}{L^{Nd}}  {   \frac12}\ln \frac{\det \mathscr{A}^{(\boldsymbol{q})}_N}{    \det \mathscr{A}^{(0)}_N}        
		 =\frac{1}{L^{Nd}} {  \frac12}\ \sum_{p\in \widehat{T}_N \setminus \{0\}} \ln 
		  \det\bigl(   \widehat{   \mathcal{A}^{(\boldsymbol q)}_N}(p)\bigr)  
		 - \ln \det  \bigl(   \widehat{\mathcal{A}^{(0)}_N}(p)    \bigr)\\
		  & =\frac{1}{L^{Nd}} {  \frac12}\ \sum_{p\in \widehat{T}_N \setminus \{0\}} \ln 
		  \det\bigl(  \tfrac1{|p|^2}    \widehat{   \mathcal{A}^{(\boldsymbol q)}_N}(p)\bigr)  
		 - \ln \det  \bigl(   \tfrac1{|p|^2}   \widehat{\mathcal{A}^{(0)}_N}(p)    \bigr).
		\end{split}
	\end{align}
	
Now it follows from \eqref{eq:fourierA_repeated} and  \eqref{Ahatestimate} that $\boldsymbol q \mapsto \frac{1}{|p|^2}  \widehat{\mathcal{A}_N^{(\boldsymbol{q})}}(p)$
 is linear in $\boldsymbol q$ and both $\frac{1}{|p|^2}  \widehat{\mathcal{A}_N^{(\boldsymbol{q})}}(p)$ and its inverse are bounded 
  uniformly in $N$ and $p \in \widehat T_N \setminus \{0\}$.
In particular,  $ \det\bigl(  \frac1{|p|^2}    \widehat{   \mathcal{A}^{(\boldsymbol q)}_N}(p)\bigr)$ lies in a fixed compact subset of $(0, \infty)$. 
The determinant is smooth and the logarithm is smooth away from 0.
Since the sum contains $L^{Nd}-1$ terms it follows that  the function $F_N$ is smooth and the derivatives are bounded uniformly in $N$. 
\end{proof}

\section{Proof of the scaling limit}\label{sec:proof_scaling_limit}
In the setting of \cite{AKM16} the scaling limit was derived by Hilger \cite{Hil15}.
Here we follow a similar strategy.
In this section $\Kcal$ is fixed and we use the abbreviations
\begin{equation} \label{eq:abbr_scaling_limit}
e_N = \widehat {e}_N(\mathcal K), \quad 
\boldsymbol q_N = \widehat {\boldsymbol q}_N(\mathcal K), \quad K_N = \widehat K_N(\mathcal K).
\end{equation}

\begin{proof}[Proof of Theorem~\ref{th:scalinglimit}] 
Recall that we consider $f \in C^\infty(\mathbb{T}^d; \R^m)$ where $\mathbb{T}^d = \R^d/ \Z^d$ and  define rescaled functions on $\TN = \Z^d/ (L^N \Z^d)$ by 
\begin{equation}
f_N(x) = L^{-N \frac{d+2}{2}} f(L^{-N} x) -  c_N
\end{equation}
where the constant $c_N$ is chosen so that 
\begin{equation}  \label{eq:mainthm_zero_avg}
\sum_{x \in T_N} f_N(x) = 0.
\end{equation}
Note that in the statement of Theorem~\ref{th:scalinglimit} we did not subtract 
to constant from $f_N$. Since $(c_N, \p) = 0$ for all $\p \in \Xcal_N$ subtracting the constant 
does, however, not affect the statement of Theorem~\ref{th:scalinglimit}.

We rewrite the right hand side of equation \eqref{eq:scaling_limit_abstract}
using the definition \eqref{eq:Zpertcomp_general} and \eqref{thmequality2} from Theorem~\ref{maintheorem},
	\begin{align}\label{eq:scalingstarting}
		\frac{\Zscr_N(\Kcal,\Qscr,f_N)}{\Zscr_N(\Kcal,\Qscr,0)}=                                                                       
		e^{\frac12( f_N,    \mathscr{C}^{({\boldsymbol{q}}_N)}_N    f_N)}
		\frac{\int_{\Xcal_N}\left(1+K_N(\TN,\p  +\mathscr{C}^{({\boldsymbol{q}}_N)} f_N)\right)
		\,\mu_{N+1}^{({\boldsymbol{q}}_N )}(\d\p)}
		{\int_{\Xcal_N}\left(1+K_N(\TN,\p)\right)\,\mu_{N+1}^{({\boldsymbol{q}}_N   )}(\d\p)}.                                                                                                                                  
	\end{align}
	Here we used that the contribution of the  term $Z^{( {\boldsymbol{q}}_N)}_N e^{L^{Nd}e_N}/Z^{(0)}_N$ in \eqref{thmequality2}
	cancels in the above ratio since this term  does not depend on $ f_N$.
	The matrix $\boldsymbol{q}_N$ depends on $N$, but it is bounded uniformly in  $N$.
	Thus, we find a subsequence $N_\ell \to \infty$ such that $\boldsymbol q_{N_{\ell}}$  converges to $\boldsymbol q$. 
	In the following we only consider this subsequence, but for ease of notation we still write $\boldsymbol q_N$. 
	One can actually show convergence of the whole sequence  \cite{Hil18} using the techniques from \cite{BS15V}. 
	
	We consider the two terms on the right hand side of  \eqref{eq:scalingstarting} in two steps. 
	First we show that the second term converges to 1 by showing that this holds for the numerator and the denominator  separately.
	Actually, it suffices to show convergence for the numerator since the denominator corresponds to the special case with $ f_N = 0$. 
	  Theorem~\ref{th:weights_final}\ref{w:w8} and the bound \eqref{eq:keyboundKN}  imply that
	\begin{align}
		\begin{split}\label{eq:errweight1}
		&\left| \int_{\Xcal_N}K_N(\TN,\p +  \mathscr{C}^{({\boldsymbol{q}}_N)} f_N)\,
		\mu_{N+1}^{({\boldsymbol{q}}_N)}(\d\p)  \right|
		 \\ & \hspace{3cm}
		  \leq 
		\frac{\lVert K_N\rVert_N^{(A)}}{A}  
		\int_{\Xcal_N}w_N^{\TN}(\p + \mathscr{C}^{({\boldsymbol{q}}_N)} f_N)\,
		\mu_{N+1}^{({\boldsymbol{q}}_N)}(\d\p)
		\\
		  & \hspace{3cm}\leq \overline{C}_0  \eta^N
		\frac1A  \,   \AB  \, \,
		 w_{N:N+1}^{\TN}(\mathscr{C}^{({\boldsymbol{q}}_N)} f_N).
		\end{split}
	\end{align}
	The  weight function can be bounded using  Theorem~\ref{th:weights_final}\ref{w:w2},
	\begin{align}
		\begin{split}\label{eq:errweight2}
		\ln(w_{N:N+1}^{\TN}(\mathscr{C}^{({\boldsymbol{q}}_N)} f_N))
		  & \leq 
		\frac{1}{2\lambda} \left(\mathscr{C}^{{(\boldsymbol{q}}_N)}_N f_N, \Malt_{N}\mathscr{C}^{({\boldsymbol{q}}_N)}_N  f_N\right).
		\end{split}
	\end{align}
	By \eqref{Ahatestimate} 
	 the Fourier modes of the kernel of $\mathscr{C}^{({\boldsymbol{q}}_N)}$ satisfy 
	 $|\widehat{\mathcal{C}}^{({\boldsymbol{q}}_N)}(p)|\leq C|p|^{-2}\leq CL^{2N}$. 
	Recall that $q_i(p)=e^{ip_i}-1$ and $q(p)^\alpha=\prod_{i=1}^d q_i(p)^{\alpha_i}$ for any multiindex $\alpha\in \mathbb{N}^d$. 
	Using the Plancherel
	identity \eqref{eq:Plancherel}, we get 
		\begin{align}
		\begin{split}\label{eq:errweight3}
		\bigl( \mathscr{C}^{({\boldsymbol{q}}_N)}  f_N &, \boldsymbol{M}_{N}  \mathscr{C}^{({\boldsymbol{q}}_N)}  f_N\bigr)
		\\ 
		  & = L^{-Nd}\sum_{1\le |\alpha|\leq M} \sum_{p\in \widehat{T}_N}                           
		(\widehat{\mathcal{C}}^{({\boldsymbol{q}}_N)}(p)\widehat{ f}_N(p), L^{2N(|\alpha|-1)}|q(p)^{2\alpha}|
		\widehat{\mathcal{C}}^{({\boldsymbol{q}}_N)}(p)\widehat{ f}_N(p))
		\\
		  & = L^{-Nd-2N}\sum_{1 \le|\alpha|\leq M}  \sum_{p\in \widehat{T}_N}                          
		L^{2N|\alpha|}|q(p)^{2\alpha}|\,|  \widehat{\mathcal{C}}^{({\boldsymbol{q}}_N)}(p)\widehat{ f}_N(p)|^2
		\\
		  & \leq                                                                                    
		CL^{-Nd+2N}\sum_{1 \le|\alpha|\leq M}\sum_{p\in \widehat{T}_N}
		L^{2N|\alpha|}|q(p)^{2\alpha}||\widehat{ f}_N(p)|^2
		\\
		  & =CL^{2N}\sum_{1 \le  |\alpha| \leq M} L^{2N|\alpha|}(\nabla^\alpha  f_N,\nabla^\alpha  f_N). 
		\end{split}
	\end{align}
	To estimate the discrete derivatives at $x$ we apply a Taylor expansion of $ f$ of order $r$. 
	This gives
	\begin{align}  \label{eq:taylor_mainthm}
		\begin{split}
		 f_N(x+a) & =L^{-N\frac{d+2}{2}}  f(L^{-N}x+L^{-N}a)                                                                     
		\\
		         & =L^{-N\frac{d+2}{2}}\Bigl( \sum_{0\leq \beta\leq r} \frac{(L^{-N}a)^\beta}{\beta!}\partial^\beta  f(L^{-N}x) 
		+R_r\Bigr)
		\end{split}
	\end{align}
	where $R_r$ denotes the remainder that can be bounded by  $ C_{r+1}   |\nabla^{r+1}f|_\infty |L^{-N}a|^{r+1}$.
	Since the discrete derivative of order $|\alpha|$ annihilates polynomials up to order $|\alpha|-1$ 
	and since the discrete derivative is a bounded operator,  the identity \eqref{eq:taylor_mainthm} implies that 
	\begin{align}\label{eq:discderivfN}
		|\nabla^{\alpha} f_N(x)|\leq C_{|\alpha|} L^{-N\frac{d+2}{2}}|\nabla^{|\alpha|} f|_\infty L^{-N|\alpha|} 
	\end{align}
	and thus
	\begin{align}\label{eq:errweight4}
		L^{2N}\sum_{1\le  |\alpha| \leq M} L^{2N|\alpha|}(\nabla^\alpha  f_N,\nabla^\alpha  f_N) 
		\leq C\sum_{r=0}^M |\nabla^r f|_{\infty}.                                             
	\end{align}
	Combining \eqref{eq:errweight1}, \eqref{eq:errweight2}, \eqref{eq:errweight3}, and \eqref{eq:errweight4}, we conclude that
	\begin{align}
		\left|\int_{\Xcal_N}K_N(\TN,\p + \mathscr{C}^{({\boldsymbol{q}}_N)} f_N)
		\,\mu_{N+1}^{({\boldsymbol{q}}_N )}(\d\p)\right|
		  & \leq \overline{C}_0  \eta^N \frac{\AB}{A}  \exp\left(C\sum_{r=0}^M |\nabla^r f|_{\infty}\right)\to 0
	\end{align}
	as $N \to \infty$. 
	This implies that the numerator on the right hand side of \eqref{eq:scalingstarting} converges to 1.
	
	The second step is to prove the convergence of  the prefactor 
	\begin{align}
	\frac12( f_N,\mathscr{C}^{({{\boldsymbol{q}_N}})}_N\ f_N)
	\to \frac{1}{2} (f,\mathscr{C}_{\mathbb{T}^d}f).
	\end{align}
	To show this we change the scaling of the system. 
	Instead of considering the system size growing with $N$ with the distance between the atoms remaining fixed, 
	here we fix  the system size and   let the distance between the atoms go to zero.
	
	We define the rescaled torus $T_N'$ and the corresponding dual torus $\widehat{T}_N'$ in Fourier space by 
	\begin{align}
		T_N'=L^{-N}T_N, \quad \widehat{T}_N'=L^N\widehat{T}_N.
	\end{align}
	Recall from \eqref{eq:definition_dual_torus}
	that 
	\begin{equation}   \label{eq:rescaled_dual_torus}
	\widehat{T}_N' = \{ \xi \in (2 \pi \mathbb{Z})^d : | \xi|_\infty \le (L^N -1) \pi \}
	\end{equation}
	 (here we use that $L$ is odd and hence $L^N-1$ is even and we identify the dual torus with its fundamental domain). 
	To make the notation clearer we will write $x$ and $z$ for coordinates in $T_N$ and $T_N'$, respectively,
	and similarly $p$ and $\xi$ for coordinates in $\widehat{T}_N$ and $\widehat{T}_N'$, respectively.
	Note that there  is an inclusion $T_N'\to (\mathbb{R}/\mathbb{Z})^d = \mathbb{T}^d$.
For a function $g: T_N'  \to \mathbb{C}$ we define the discrete Fourier transform by
	\begin{equation} \label{eq:discrete_FT_rescaled}
	\hat g(\xi) := L^{-dN} \sum_{z \in T_N'} g(z) e^{-i \xi \cdot z}  \quad \text{ for any } \xi \in \widehat{T}_N'.
	\end{equation}
	The prefactor $L^{-dN}$ is chosen  so that, for $g \in C^0(\mathbb{T}^d)$, the sum is the Riemann sum which corresponds
	to the integral for the coefficient in the Fourier series of $g$. 
	For brevity  we write for the rest of this section 
	\begin{equation} \label{eq:abbr_CqN_mainthm}
	\mathcal C_N = \mathcal C^{({\boldsymbol q}_N)}.
	\end{equation}
	This quantity should not be confused with the finite range decomposition at scale $N$. 
	We define the rescaled functions $f_N':T_N'\to \R^m$, 
	\begin{align}
		\begin{split}
		f_N'(z)                  & =   L^{N\frac{(d+2)}{2}} f_N(L^Nz)=f(z),       \\ 
			{\mathcal C}_N'(z)&  = L^{N (d-2)} {\mathcal C}_N(L^N z).  
		\end{split}
	\end{align}
	Note that the rescaling of $\mathcal C_N$ reflects the expected behaviour of the Green's function of the Laplacian,
	namely $\mathcal C_N(x) \sim |x|^{2-d}$.
	Then the corresponding Fourier transforms $\widehat{f}_N':\widehat{T}_N'\to \mathbb{C}^m$,
	and $\widehat{\mathcal{C}}_N':\widehat{T}_N'\to \mathbb{C}^{m\times m}_{\mathrm{her}}$
	satisfy 
	\begin{align}  \label{eq:rescaled_f_Nprime}
	      \begin{split}
	     & \widehat{f}_N'(\xi)        =L^{-Nd}\sum_{z\in T_N'}f(z)e^{-iz\xi}
		=L^{-N\frac{d-2}{2}}\sum_{x\in T_N} f_N(x)e^{-iL^{-N}\xi x}
		=L^{-N\frac{d-2}{2}}\widehat{f}_N(L^{-N}\xi) 
		\\
		&\text{ and } \;
		\widehat{\mathcal{C}}_N'(\xi)  =L^{-2N}   \widehat{\mathcal{C}}_N(L^{-N}\xi).  
	      \end{split}
	\end{align}
	Using  this rescaling, Plancherel identity,  and the zero-average condition  \eqref{eq:mainthm_zero_avg}, we find that
	\begin{align}
		\begin{split}\label{eq:scalrhs}                                                                     
		(f_N, \mathscr{C}_N f_N)                                                                                 
		=\frac{1}{L^{Nd}}\sum_{p\in \widehat{T}_N \setminus \{0\}}(\widehat{f}_N(p), 
		\widehat{\mathcal{C}}_N(p)\widehat{f}_N(p)) 
		=\sum_{\xi \in \widehat{T}_N' \setminus\{0\}} (\widehat{f}_N'(\xi),  \widehat{\mathcal{C}}_N'(\xi)\widehat{f}_N'(\xi)).  
		\end{split}                                                                                         
	\end{align}
	On the other hand, the Plancherel identity and the fact that $f$ has average $0$  yield
	\begin{align}\label{eq:scallhs}
		(f,\mathscr{C}f)=\sum_{\xi\in (2\pi\mathbb{Z})^d\setminus \{0\}} (\widehat{f}(\xi),\widehat{\mathcal{C}}_{\mathbb{T}^d}(\xi)\widehat{f}(\xi)),
	\end{align}
	where the Fourier modes are given  by
	\begin{align}   \label{eq:definition:mathcalChat}
		\begin{split}
		\widehat{f}(\xi)      & =\int_{\mathbb{T}^d} f(z)e^{-i\xi z} \text{ and }                               \\
		\widehat{\mathcal{C}}_{\mathbb{T}^d}(\xi) & =\left(\sum_{i,j=1}^d \xi_i\xi_j(\boldsymbol{Q}+\boldsymbol{q})_{ij}\right)^{-1}.
		\end{split}
	\end{align}
	The last expression is well defined because $\boldsymbol{Q}+\boldsymbol{q}$ is positive definite.
	Now we show the pointwise convergence
	\begin{align}\label{eq:scalingptwise}
		\lim_{N\to \infty} (\widehat{f}_N'(\xi),  \widehat{\mathcal{C}}_N'(\xi)\widehat{f}_N'(\xi))= 
		(\widehat{f}(\xi), \widehat{\mathcal{C}}_{\mathbb{T}^d}(\xi)\widehat{f}(\xi))                             
	\end{align}
	for all $\xi \in (2 \pi \Z)^d \setminus \{0\}$. 
	First, note that  $\widehat{f}_N'(\xi)\to \widehat{f}(\xi)$ for
	 all $\xi\in(2\pi\mathbb{Z})^d  \setminus \{0\}$
	since $\widehat{f}_N'(\xi)$ is a Riemann sum approximation of the integral for $\widehat{f}(\xi)$.
	For the covariance we observe that, by \eqref{eq:fourierA_repeated},
	\begin{align}
	\begin{split}  \label{eq:covariance_scaling_higher}
	&	\widehat{\mathcal{C}}_N'(\xi)  =\widehat{\mathcal{C}}_N(L^{-N}\xi)L^{-2N} \\
&		=\Bigl(\sum_{\alpha,\beta\in \Ical} L^N\overline{q}(L^{-N}\xi)^{\alpha}\boldsymbol{Q}_{\alpha \beta}L^Nq(L^{-N}\xi)^\beta
		+
		\hspace{-.3cm}\sum_{|\alpha|=|\beta|=1}^d\hspace{-.3cm} L^N\overline{q}(L^{-N}\xi)^\alpha\boldsymbol{q}_{\alpha \beta}L^Nq(L^{-N}\xi)^\beta\Bigr)^{-1}\hspace{-.25cm}.
		\end{split}
	\end{align}
	With  $N\to \infty$, we have $L^N q(L^{-N}\xi)^\alpha=L^N(e^{iL^{-N}\xi_j}-1)\to i\xi_j$  for
	$\alpha=e_j$ and $L^N q(L^{-N}\xi)^\alpha\to 0$ for $|\alpha|\geq 2$.
	Then the assumption that $\boldsymbol{q}_N\to \boldsymbol{q}$ along the considered subsequence 
	and the fact that the inversion of matrices is continuous yield
	\begin{align}
	\widehat{\mathcal{C}}_N'(\xi)\to \widehat{\mathcal{C}}_{\mathbb{T}^d}(\xi)\quad \text{as $N\to\infty$}.
	\end{align}
	This establishes \eqref{eq:scalingptwise}.
	
	Next we show that the Fourier modes are uniformly bounded from above.
	Note that $|\widehat{\mathcal{C}}_N'(\xi)|=L^{-2N}|\widehat{\mathcal{C}}_N(L^{-N}\xi)|\leq C |\xi|^{-2}$
	by \eqref{Ahatestimate}.
	The definition of $q(p)$ and the discrete integration by parts yield
	\begin{align}
		|q(p)|^{2r}\widehat{f}_N(p)=\sum_{x\in T_N} f_N(x)\Delta^r e^{-ipx}= 
		\sum_{x\in T_N} \Delta^r f_N(x)e^{-ipx}.                             
	\end{align}
	The bound $|p|\leq 2|q(p)|$  with the rescaling \eqref{eq:rescaled_f_Nprime}
	and \eqref{eq:discderivfN} 	 implies  
	\begin{align}
		\begin{split}
		|\xi|^{2r}|\widehat{f}_N'(\xi)|
		  & =L^{2rN}|p|^{2r}L^{-N\frac{d-2}{2}}|\widehat{f}_N(p)|  \leq
		  C_r \, L^{2rN}    L^{-N\frac{d-2}{2}}  \sum_{x \in T_N} |\Delta^r f_N(x)| 
		  \\
		  & \leq                                                  
    		C_r \, L^{2rN}L^{-N\frac{d-2}{2}}\sum_{x\in T_N} L^{-N\frac{d+2}{2}}
		\, |\nabla^{2r}f|_\infty  \,  L^{-2rN}
		\leq C_r|\nabla^{2r}f|_\infty
		\end{split}
	\end{align}
	for $\xi\in \widehat{T}_N'$ and  $p = L^{-N} \xi$.
	Note that by   \eqref{Ahatestimate} and \eqref{eq:rescaled_f_Nprime}, we have 
	$| \widehat{\mathcal{C}}_N'(\xi)| \le C L^{-2N} L^{2N} |\xi|^{-2} \le C |\xi|^{-2}$. 
	Hence,
	\begin{align}
		(\widehat{f}_N'(\xi),  \widehat{\mathcal{C}}_N'(\xi)   \widehat{f}_N'(\xi)) 
		\leq C_r|\xi|^{-2r { - 2}} \,  |\nabla^{2r}f|_\infty                     
		\qquad \text{ for any }   \xi\in\widehat{T}_N' \setminus \{0\}.
	\end{align}
		For $r \ge  \lfloor \frac{d}{2}\rfloor $,  the right hand side is  summable over  $\xi \in (2\pi \Z)^d \setminus \{0\}$
	and the dominated convergence theorem, with  the pointwise convergence \eqref{eq:scalingptwise}, implies that
	\begin{align}
	\begin{split}
	 \sum_{\xi \in \widehat{T}_N' \setminus\{0\}} (\widehat{f}_N'(\xi),  \widehat{\mathcal{C}}_N'(\xi)\widehat{f}_N'(\xi))
	&\underset{   \eqref{eq:rescaled_dual_torus} }{=}
	 \sum_{ \xi \in (2 \pi \Z)^d \setminus \{0\} }1_{|\xi|_\infty \le (L^N-1) \pi} \, \, 
	  (\widehat{f}_N'(\xi),  \widehat{\mathcal{C}}_N'(\xi)
	 \widehat{f}_N'(\xi))\\
&	\to  \sum_{\xi\in (2\pi\mathbb{Z})^d\setminus \{0\}}  (\widehat{f}(\xi),\widehat{\mathcal{C}}_{\mathbb{T}^d}(\xi)\widehat{f}(\xi)).
	\end{split}
	\end{align}
	 Now  \eqref{eq:scalrhs} and   \eqref{eq:scallhs}   show that 
	$ (f_N, \mathscr{C}_N f_N) \to (f, \mathscr{C}_{\mathbb{T}^d} f)$ (along the subsequence considered).
\end{proof}

\chapter{Fine-tuning of the Initial Conditions}\label{sec:finetuning}

In this chapter we prove Theorem \ref{maintheorem} by the use of a stable manifold theorem
and an additional application of the implicit function theorem to determine the renormalized Hamiltonian. 
The setting for the stable manifold theorem is very similar to the situation in 
Theorem 2.16 of \cite{Bry09} but for completeness and for the convenience of the reader we provide a detailed proof.

\section{The renormalisation maps as a dynamical system} 
\label{sec:TonZ}

The stable manifold theorem boils down to an application of the implicit function theorem
to the whole trajectory of relevant and irrelevant interactions $(H_k, K_k)$. 
We define the Banach space\begin{align}
	\mathcal{Z}=\{Z=(H_0,H_1,\ldots,H_{N-1},K_1,\ldots,K_N):H_k\in M_0(\mathcal{B}_k),K_k\in M(\mathcal{P}_k^c)\} 
\end{align}
equipped with the norm
\begin{align}   \label{eq:norm_mcZ}
	\lVert Z\rVert_{\mathcal{Z}}=\max\left( \max_{0\leq k\leq N-1} \frac{1}{\eta^k}\lVert H_k\rVert_{k,0},
	\max_{1\leq k\leq N}   \frac{1}{\eta^k}      \lVert K_k\rVert_{k}^{(A)}    \right)                                                                           
\end{align}
where
\begin{equation} \label{eq:choice_eta}
\eta \in \left(0,\tfrac23\right].
\end{equation}
is a fixed parameter.
Note that a bound on $\| Z \|_{\mathcal Z}$ implies exponential decay of the norms of $H_k$ and $K_k$ in $k$. 
The functionals $H_N$ and $K_0$ do not appear in $\mathcal{Z}$ because we want to achieve the final condition 
$H_N = 0$ and we treat $K_0$  as a fixed initial condition, see \eqref{defofk0} below.

We define a dynamical system $\mathcal{T}$ on $\mathcal{Z}$.
The map $\mathcal{T}$ depends in addition  on two parameters, a relevant  Hamiltonian 
$\mathcal{H}\in M(\mathcal{B}_0)$
and the interaction $\mathcal{K}\in \boldsymbol{E}$.
Here we fix  $\zeta \in (0,1)$ and we recall from \eqref{eq:normE} that the Banach space $\boldsymbol{E}$ consists of functions
$K:\Gcal=\left(\mathbb{R}^m\right)^{\Ical}\rightarrow \mathbb{R}$ 
so
  that the following norm is finite
\begin{align}
	\norm{ K }_\zeta = \sup_{z\in \Gcal} \sum_{|\alpha|\leq r_0} \frac1{\alpha !}\abs{\partial^\alpha 
	K(z)}e^{-\frac{1}{2}(1-\zeta)\Qscr_0(z)} .                                           
\end{align}
The Hamiltonian $\mathcal H$  will eventually allow us  
to extract the correct Gaussian part in the measure (the renormalized covariance).

More precisely we consider  a map
$\mathcal{T}:\boldsymbol{E}\times M(\mathcal{B}_0)\times \mathcal{Z}\rightarrow \mathcal{Z}$
defined by $\mathcal{T}(\mathcal{K},\mathcal{H},Z)=\tilde{Z}$ where the coordinates of $\tilde{Z}$  
are given by 
\begin{align}
\tilde{H}_0(\mathcal K, \mathcal H, Z) 
& =(\boldsymbol{A}_0^{(\boldsymbol{q}(\mathcal H))})^{-1}\left(H_{1}-\boldsymbol{B}_{0}^{(\boldsymbol{q}(\mathcal H))} \widehat K_0(\mathcal K, \mathcal H)\right),
\label{eq:tilde_H0_def}   \\
\tilde{H}_k(\mathcal K, \mathcal H, Z) 
& =(\boldsymbol{A}_k^{(\boldsymbol{q}(\mathcal H))})^{-1}(H'-\boldsymbol{B}_k^{(\boldsymbol{q}(\mathcal H))}K_k), \quad \hbox{for} \quad  1 \le k \le N-2, \\
\tilde H_{N-1}(\mathcal K, \mathcal H, Z) 
&= - (\boldsymbol{A}_{N-1}^{(\boldsymbol{q}(\mathcal H))})^{-1} \boldsymbol{B}_{N-1}^{(\boldsymbol{q}(\mathcal H))}K_{N-1}, 
\label{eq:tilde_H{N-1}_def} \\
\tilde{K}_{k+1}(\mathcal K, \mathcal H, Z) & =\myS_k(H_k,K_k,\boldsymbol{q}), \quad    \hbox{for} \quad 1 \le k \le N-1,         
\label{eq:tilde_Kk+1_def}       \\
\tilde{K}_1(\mathcal K, \mathcal H, Z) &= \myS_0(H_0, \widehat{K}_0(\mathcal K, \mathcal H), \boldsymbol q(\mathcal H)).                
\end{align}
Here the map $\widehat K_0$ is defined by
\begin{align}\label{defofk0}
	\widehat K_0(\mathcal K, \mathcal H) (X,\p)=\exp\left(-\mathcal{H}(X,\p)\right)\prod_{x\in X} \mathcal{K}(D \p(x)). 
\end{align}
and  $\boldsymbol{q}(\mathcal{H})$ is the projection on the coefficients of the quadratic part of $\mathcal{H}$, 
i.e., $\boldsymbol{q}_{(i,\alpha)(j,\beta)}=\frac12 a_{(i,\alpha),(j,\beta)}$ for $(i,\alpha)< (j,\beta)$,
$\boldsymbol{q}_{(i,\alpha)(j,\beta)}=\frac12 a_{(j,\beta),(i,\alpha)}$ for $(i,\alpha)> (j,\beta)$, and
$\boldsymbol{q}_{(i,\alpha)(j,\beta)}=a_{(i,\alpha),(j,\beta)}$ for $(i,\alpha)=(j,\beta)$ 
where $a_{(i,\alpha),(j,\beta)}$ denotes the coefficients of the quadratic term of $\mathcal{H}$.
The factor $\frac12$ arises because $\boldsymbol{q}$ is symmetric.
Note that the definition  \eqref{eq:tilde_H{N-1}_def} of $H_{N-1}$  reflects the final condition $\HN=0$.

One easily sees that 
\begin{align}
\begin{split} &  \mathcal{T}(\mathcal{K},\mathcal{H},Z)=Z  \quad 
\text{ \emph{if and only if} }  \boldsymbol{T}_k(H_k,K_k,\boldsymbol{q}(\mathcal H))=(H_{k+1},K_{k+1}) \\
&\hspace{2cm} \text{ for all } 0\leq k\leq N-1 \text{ with } \HN = 0 \text{ and } K_0 = \widehat K_0(\mathcal K, \mathcal H).
\end{split}
\end{align}
Here $\boldsymbol{T}_k$ is
the
renormalisation group map defined 
in Definition \ref{def:Tk}.
Proposition \ref{pr:properties_RG_map} and   \eqref{E:R...R} imply
 that a fixed point of $\mathcal{T}$ satisfies 
\begin{align}\label{eqfirstfixedpointid}
	\begin{split}
	\int_{\Xcal_N} (e^{-H_0}\circ \widehat{K_0}(\mathcal K, \mathcal H))&(\TN,\p+\psi)\; \mu^{(\boldsymbol{q}(\mathcal H))}(\d \p)\\
	&=\int_{\Xcal_N} \bigl(1 +K_N(\TN,\p+\psi)\bigr)\; \mu^{(\boldsymbol{q}(\mathcal H))}_{N+1}(\d \p). 
\end{split}\end{align}

\section{Existence of a fixed point of the map \texorpdfstring{$\mathcal{T}(\mathcal{K},\mathcal{H},\cdot)$}{T(K,H,.)}}
\label{sec:existfixpoint}

Theorem~\ref{propfixedpoint} below states that for sufficiently small $\mathcal H$ and $\mathcal K$ there is a
unique fixed point $\hat Z(\mathcal K, \mathcal H)$ which depends smoothly on $\mathcal K$ and  $\mathcal H$. 
In particular \eqref{eqfirstfixedpointid} holds with $H_0 = \Pi_{H_0} \widehat{Z}(\mathcal K, \mathcal H)$ and
$K_N = \Pi_{K_N} \widehat Z(\mathcal K, \mathcal H)$ where $\Pi_{H_0} Z$ and $\Pi_{K_N} Z$ denote the projection
onto the $H_0$ component and the $K_N$ component, respectively.
Now, the right hand side of \eqref{eqfirstfixedpointid} deviates from $1$ only by an error of order $O( \eta^N)$
and
the left hand side of \eqref{eqfirstfixedpointid} looks very similar to the functional
\begin{equation}  \label{eq:fine_tuning_real_functional}
\int_{\Xcal}  \sum_{X \subset \TN}  \prod_{x \in X} \mathcal K(D \p(x)) \, \mu^{(0)}(d\p)
\end{equation}
which we want to study, but is in general not identical to it due to the presence of the terms
$\Pi_{H_0} \widehat Z(\mathcal K, \mathcal H)$ and $\boldsymbol q(\mathcal H)$. 
Another application of the implicit function theorem leads to 
Lemma~\ref{lemmafixedpoint}  below which shows that 
there exist an $\mathcal H = \widehat {\mathcal H}(\mathcal K)$ such that
$\Pi_{H_0} \widehat Z(\mathcal K, \mathcal H) = \mathcal H$.
Then a short calculation shows that  the left hand side of 
\eqref{eqfirstfixedpointid} agrees 
with  the expression \eqref{eq:fine_tuning_real_functional}
up to an explicit scalar factor which involves the constant term in $\mathcal H$ 
and the ratio $Z^{(\boldsymbol q(\mathcal H))}/ Z^{(0)}$
of the Gaussian partition functions,  see  \eqref{finalidentity} below.
From this representation we will easily deduce the main theorem of the previous chapter, 
Theorem~\ref{maintheorem}.

 Recall the convention that, say, $C_{   \eqref{eq:bound_mcH_to_q} }$  denotes  the constant 
 which appears in  equation   
\eqref{eq:bound_mcH_to_q}.

\begin{theorem}\label{propfixedpoint} Let $\kappa = \kappa(L)$ be as in Theorem~\ref{th:weights_final}.
Let $L_0$, $h_0(L)$, $A_0(L)$, $\rho(A)$, $C_{j_1,j_2,j_3}(L,A)$, and $C_\ell(L,A)$ be such that the conclusions of 
 Theorem~\ref{prop:smoothnessofS}  and Theorem~\ref{prop:contractivity} hold for every triplet $(L,h,A)$
 with $L \ge L_0$, $h \ge h_0(L)$, $A \ge A_0(L)$.
 Assume also that
  \begin{equation} \label{eq:conditions_L_h_propfixedpoint}
h_0(L) \ge   \max(\delta(L)^{-1/2}, 1)
 \end{equation}
 where $\delta(L)$ is the constant introduced in \eqref{eq:defofdelta}.  
 Then for every triplet $(L,h,A)$
 that satisfies  $L \ge L_0$, $h \ge h_0(L)$, $A \ge A_0(L)$  there exist constants $\rho_1 = \rho_1(h,A) >0$, $\rho_2 = \rho_2(L) >0$
and $\overline{C}_{j_1, j_2, j_3}$  such that
 $\mathcal T$ is smooth in $B_{\rho_1}(0) \times B_{\rho_2}(0) \times B_{\rho(A)}(0) \subset
  \mathscr (M(\mathcal B_0); \| \cdot \|_{0,0}) \times {\boldsymbol E} \times \mathcal Z$,
 \begin{align}  \label{eq:bounds_T}
 \begin{split}
& \frac{1}{ j_1!  j_2!  j_3!}    \| D^{j_1}_{\mathcal K} D^{j_2}_{\mathcal H} D^{j_3}_{Z} \mathcal{T}(\mathcal K, \mathcal H,  Z)
 (\dot{\mathcal K}, \ldots \dot{\mathcal H}, \ldots,  \dot{Z}) \|_{\mathcal Z}
 \le  \overline{C}_{j_1, j_2, j_3}(L,A) \,  \| \dot{\mathcal K}\|^{j_1}_{\zeta} 
 \,   \|\dot{\mathcal H}\|^{j_2}_{0,0}  \,  \|\dot Z\|^{j_3}_{\mathcal Z}\\
& \text{ for all } (\mathcal K, \mathcal H, Z) \in B_{\rho_1}(0) \times B_{\rho_2}(0) \times B_{\rho(A)}(0)\text{ and }
 \end{split}
 \end{align}
 \begin{align} \label{eq:range_q_of_mcH}
 \boldsymbol q(\mathcal H) \in B_{\kappa}(0)  \text{ for all } \mathcal H \in B_{\rho_2}(0).
 \end{align}

Moreover there exist $\epsilon = \epsilon (L,h,A) >0$, 
 $\epsilon_1 = \epsilon_1(L,h,A) > 0$, $\epsilon_2 = \epsilon_2 (L,h,A) > 0$, and $C_{j_1,j_2}(L,A)>0$ such that,
for each $(\mathcal K, \mathcal H) \in B_{\epsilon_1}(0) \times B_{\epsilon_2}(0)$, there exists a unique 
$Z = \widehat Z(\mathcal K, \mathcal H)$ in $B_{\epsilon}(0)$ that satisfies 
\begin{align}\label{firstfixedpoint}
\mathcal{T}(\mathcal{K},\mathcal{H},\hat{Z}(\mathcal{K},\mathcal{H}))=\hat{Z}(\mathcal{K},\mathcal{H}).
\end{align}
The map $\hat Z$ is smooth on  $B_{\epsilon_1}(0) \times B_{\epsilon_2}(0)$ and  satisfies the  bounds
\begin{align}  \label{firstfixedpoint_bounds}
\begin{split}
\frac{1}{j_1! j_2!}	\lVert D^{j_1}_{\mathcal K} D^{j_2}_{\mathcal H} \hat{Z}(\mathcal{K},\mathcal{H})
(\dot{\mathcal K}, \ldots, \dot{\mathcal H}) \rVert_{\mathcal Z}
& \leq  C_{j_1,j_2}(L,h,A) \, \| \dot{\mathcal K}\|^{j_1}_{\zeta} \, 
\| \dot{\mathcal H}\|^{j_2}_{0,0}  \\
&   \text{ for all }\  (\mathcal K, \mathcal H) \in B_{\epsilon_1}(0) \times B_{\epsilon_2}(0).
\end{split}
\end{align}
\end{theorem}
An explicit choice of diameters $\rho_1$ and $\rho_2$ is
 \begin{equation} \label{eq:smoothness_mcT_neighbourhood}
 \rho_1(h,A) = \frac{\rho(A)}{2^{R_0 dr_0 +3} h^{r_0} A}  \text{ and }
 \rho_2(L) = \min\Bigl( \frac18, \frac{\kappa(L)}{C_{   \eqref{eq:bound_mcH_to_q} }(m,d)}\Bigr).
 \end{equation}	
The parameters $\epsilon$, $\epsilon_1$ and $\epsilon_2$ can be bounded from below by
$\rho_1, \rho_2, \rho(A)$, and the bounds on the first and second derivatives of $\mathcal T$.
We may take
\begin{align} \label{eq:conditions_eps_i_A} 
\begin{split}
&\epsilon = \min\left(    \frac{1}{48
\overline{C}_{0,0,2}}, \frac{\rho(A)}{2}   \right),  \quad
\epsilon_1 = \min\left(\frac{1}{24 \overline{C}_{1,0,1}}, \frac{\epsilon}{8 
 \overline{C}_{1,0,0}},\rho_1    \right), \quad
 \\ &
 \epsilon_2 = \min\left(\frac{1}{24 
\overline{C}_{0,1,1}}, \rho_2    \right)
\end{split}
\end{align}
where $  \overline{C}_{j_1, j_2, j_3}$ are the constants in  \eqref{eq:bounds_T}.
\medskip

The condition  \eqref{eq:conditions_L_h_propfixedpoint} is implied by the conditions we use to prove
Theorem~\ref{prop:smoothnessofS}  and Theorem~\ref{prop:contractivity}. We added it here since in principle
the conclusions of these theorems might hold under weaker conditions on $L$ and $h$.
Condition  \eqref{eq:conditions_L_h_propfixedpoint} is used in Lemma~\ref{le:smoothness_K0}
which ensures smoothness of the map $(\mathcal K, \mathcal H) \mapsto K_0$.

\begin{proof}[Proof of Theorem~\ref{propfixedpoint}; Set-up]
The proof is mostly along the lines of the proof of Proposition 8.1 in \cite{AKM16}.
The situation here is, however, much simpler than in   \cite{AKM16} because   no loss of regularity
occurs when we take derivative with respect to $\boldsymbol q$ (or $\mathcal H$).
Thus we can use the usual implicit function theorem in Banach spaces
which can be found, e.g., in   Chapter X.2 in \cite{Die60} or Theorem 4.E. in \cite{Zei95}.
To apply the implicit function theorem we verify its assumptions.
	
Here Theorem~\ref{prop:smoothnessofS}  and Theorem~\ref{prop:contractivity} are the 
key ingredients. The first result gives smoothness of the maps $\tilde K_k$ (except for $k=1$)
while Theorem~\ref{prop:contractivity}
will be used to show that the derivatives of $\mathcal{T}$ are small. 
Then we can apply the implicit function theorem to the map $\mathcal{T}-\pi_3$ where $\pi_3$
is the projection on the third component. The main remaining  point in showing smoothness of the 
map $\mathcal T$ is to show smoothness of the maps $(\mathcal K, \mathcal H) \mapsto K_0$. 
We first state and prove this result. Then we will continue the proof of Theorem~\ref{propfixedpoint}.
\end{proof}
	
\begin{lemma} \label{le:smoothness_K0}
Assume that $L\geq 5$ and
\begin{equation}  \label{eq:K0_restriction_on_h}
h \ge \max(\delta(L)^{-1/2}, 1   )
\end{equation}
where $\delta(L)$ is the constant defined  in \eqref{eq:defofdelta}.
Set 
\begin{align} \label{eq:rho_for_K0}
\rho_1 = (2^{R_0  d r_0 + 3} h^{r_0} A)^{-1}, \quad \rho_2 = \frac18. 
\end{align}
Then  the map $\widehat{K}_0:(\boldsymbol{E},\lVert\cdot\rVert_\zeta)\times(M_0(\mathcal{B}_0),\lVert\cdot\rVert_{0,0})\rightarrow
M(\mathcal{P}^{\rm c}_0),\lVert\cdot\rVert_{0})$  defined in \eqref{defofk0} is smooth on  $B_{\rho_1}(0)\times B_{\rho_2}(0)$ 
and there exist numerical constants $C_{j_2}$ and $C_{j_1, j_2}$ such that
\begin{align}\begin{split} \label{eq:bound_K0_ell0}
&\frac{1}{j_2!} \| D^{j_2}_{\mathcal H} \widehat K_0(\mathcal K, \mathcal H)(\dot{\mathcal H}, \ldots, \dot{\mathcal H})\|_0^{(A)} 
\\ &\qquad
\le C_{j_2} \,  2^{R_0 d r_0 +2} h^{r_0} A \, \| \mathcal K\|_{\zeta} \,  \, \| \dot{\mathcal H}\|^{j_2}
\text{ for all } (\mathcal K, \mathcal H) \in B_{\rho_1}(0) \times B_{\rho_2}(0), 
\end{split}
\end{align}  
with 
\begin{align}   \label{eq:bound_K0_ell0_bis }
C_0 = 1
\end{align}
and, for $j_1 \ge 1$,
\begin{align}
\begin{split}
&\| D^{j_1}_{\mathcal K} D^{j_2}_{\mathcal H} \widehat K_0(\mathcal K, \mathcal H)(\dot{\mathcal K}, \ldots \dot{\mathcal H})\|_0^{(A)}
\\ &\qquad \le  C_{j_1, j_2} (2^{R_0 d r_0 +3} h^{r_0} A)^{j_1}       \,  
 \|\dot{ \mathcal K}\|_{\zeta}^{j_1} \, \| \dot{\mathcal H}\|^{j_2}  
 \text{ for all } (\mathcal K, \mathcal H) \in B_{\rho_1}(0) \times B_{\rho_2}(0).   
\end{split}
\end{align}
\end{lemma}
	
To prove this lemma we decompose $K_0$ into a series of maps and show smoothness for each of them.
Then the chain rule implies the claim.
It is convenient to introduce  the weight function
\begin{align}\label{eq:w-1}
w_{-1:0}^X(\p)=\exp\bigl(\tfrac12(1-\zeta)\sum_{x\in X}    \Qscr(D\p(x))  \bigr)
\end{align}
and define $\lVert\cdot\rVert_{-1:0}^{(4A)}$ as in \eqref{globalweaknorm} and \eqref{middlenorm}.
We can write $K_0(\mathcal{K},\mathcal{H})=P_4(I(\mathcal{K}),E(\mathcal{H}))$, 
where $E$ is the exponential defined in  \eqref{eq:exponential_map_def} 
and where the inclusion map  $I$ and the product map $P_4$ are given by
\begin{align}\label{eq:defofI}
I   & :(\boldsymbol{E},\lVert\cdot\rVert_{\zeta})\rightarrow (M(\mathcal{P}_0^c),\lVert \cdot\rVert_{-1:0}^{(4A)}),\; 
I(\mathcal{K})(X,\p)=\prod_{x\in X}\mathcal{K}(D\p(x)),\\ \label{eq:defofP4}
P_4 & :(M(\mathcal{P}_0^c),\lVert \cdot\rVert_{-1:0}^{(4A)})
\times (M(\mathcal{B}_0),\opnorm{\cdot}_0)            
\rightarrow (M(\mathcal{P}_0^c),\,\lVert \cdot\rVert_{0}^{(A)})), 
\\ 
 P_4&(K,F)(X,\p)=K(X,\p)F^X(\p).
\end{align}
Smoothness of $E$ was established in Lemma~\ref{le:smoothness_exp}.
We will now show smoothness of $I$ and of $P_4$
in Lemma~\ref{le:I} and Lemma~\ref{le:smoothness_P4}, respectively,  and then conclude the proof of
Lemma~\ref{le:smoothness_K0}.
		
\begin{lemma}\label{le:I}
Let  $I$ be the map defined in \eqref{eq:defofI}. Assume that
\begin{equation} 
\rho_1 \le (2^{R_0 d r_0+3} h^{r_0} A)^{-1}   \text{ and } h \ge 1.
\end{equation}
Then $I$ is smooth on  $B_{\rho_1}(0)\subset \boldsymbol{E}$ and, for all $\mathcal K \in B_{\rho_1}(0)$,
\begin{align}  \label{eq:bound_I_ell0}
\| I(\mathcal K) \|_{-1:0}^{(4A)} \le 2^{R_0  d r_0 +2} h^{r_0} A \, \| \mathcal K\|_{ \zeta},
\end{align}
\begin{equation}  \label{eq:bound_I_ell}
\frac{1}{j!}\lVert D^jI(\mathcal{K})(\dot{\mathcal K}, \ldots, \dot{\mathcal K}) \rVert^{(4A)}_{-1:0}\le  
(2^{R_0 dr_0 + 3} \, h^{r_0} A)^j  \,  \|\dot{\mathcal K}\|^j_{\zeta}.
\end{equation}
\end{lemma}
\begin{remark}  \label{re:h-dependence}
We could avoid  $h$-dependence of the constants and neighbourhoods
here and in all other statements in  this chapter as well as in Theorem~\ref{maintheorem}
if we work with the norm
\begin{align}  
\| \mathcal K\|_{\zeta, h} := \sup_{z\in \Gcal} \sum_{|\alpha|\leq r_0} \frac1{\alpha !} h^{|\alpha|}  \abs{\partial^\alpha 
K(z)}e^{-\frac{1}{2}(1-\zeta)\Qscr_0(z)},
\end{align}
This gives slightly better results, because, roughly speaking our current setting leads to 
conditions of the type  '$h^{r_0} \| \mathcal K\|_\zeta$ is small' while it suffices that
$\| \mathcal K\|_{\zeta,h}$ is small which is a weaker condition on the low derivatives of $\mathcal K$. 
We prefer, however,   to keep the notation in Chapter~\ref{sec:setting}
simple and not to introduce a  more complicated norm with another parameter.
\end{remark}
		
\begin{proof}
Note that the functional $I(\mathcal K)$ is  translation invariant, shift invariant and local.
Thus $I(\mathcal K)$ is an element of  $M(\mathcal{P}_0^c)$.
We  first estimate $|I(\mathcal{K})(\{x\})|_{0,\{x\},T_\p}$.
Let us introduce the set  $ \Ical_m = \{1, \ldots, m\}\times \Ical $
where we recall that $\Ical\subset\{0,\ldots,R_0\}^d\setminus \{0,\ldots,0\}$ (see \eqref{eq:def_Ical}).
We consider multiindices $\gamma\in \mathbb{N}_0^{\Ical_m}$.
Recall that for    $\mpzc=(i,\alpha)\in   \Ical_m$ we defined the monomials
\begin{equation} \label{eq:I_fine_tuning_Mm}
\Mscr_\mpzc(\{x\})(\dot{\psi}) := \nabla^\mpzc\dot{\psi}(x) := \nabla^\alpha\dot{\psi}_i(x).
\end{equation}
The Taylor expansion of order $r_0$ of  $I(\mathcal{K})(\{x\})$ is given by
\begin{align}
\tay_\p I(\mathcal{K})(\{x\})(\dot{\psi})=\sum_{|\gamma|\leq r_0}\frac{1}{\gamma!}
\partial^\gamma \mathcal{K}(D\p(x)) \prod_{\mpzc\in \mathcal{I}_m} (\nabla^\mpzc\dot{\psi}(x))^{\gamma_\mpzc}.
\end{align}
Hence we have 
\begin{align}
\tay_\p I(\mathcal{K})(\{x\}) = \sum_{|\gamma|\leq r_0}\frac{1}{\gamma!}
\partial^\gamma \mathcal{K}(D\p(x)) \prod_{\mpzc\in \mathcal{I}_m} ( \Mscr_\mpzc(\{x\})  )^{\gamma_\mpzc}.
\end{align}
The triangle inequality and the product property in  Lemma~\ref{le:norms_pointwise} imply
\begin{align}
\begin{split}\label{eq:K0:1_new}
|I(\mathcal{K})(\{x\})|_{0,\{x\},T_\p}&\leq 
\sum_{|\gamma|\leq r_0}\frac{1}{\gamma!}
\abs{\partial^\gamma  \mathcal{K}(D\p(x))}   
\, \, \Bigl|\prod_{\mpzc\in \mathcal{I}_m}( \Mscr_\mpzc(\{x\})  )^{\gamma_\mpzc}\Bigr|_{0,\{x\},T_0} \\
&\leq \sum_{|\gamma|\leq r_0}\frac{1}{\gamma!} 
\abs{\partial^\gamma  \mathcal{K}(D\p(x))} \, \, \, \,  \prod_{\mpzc\in \mathcal{I}_m} 
\abs*{ \Mscr_\mpzc(\{x\})}^{\gamma_\mpzc}_{0,\{x\},T_0}.
\end{split}	
\end{align}		
	
Next, we give a crude estimate for  $\bigl|  \Mscr_\mpzc(\{x\}) \bigr|_{0,\{x\},T_0}$.
The definition of $\{x\}^\ast=\{x\}+[-R,R]^d$ ensures that  the reiterated difference quotient
$\nabla^\alpha \dot{\p}(x)$ for $\alpha\in \Ical$ can be written as a linear combination
of values $\nabla_i\p(y)$ with $y\in \{x\}^\ast$ involving  at most  $2^{|\alpha|-1}$ terms.
Using an induction argument we easily see that, for $\alpha \in \Ical$,
\begin{align}\label{eq:K0:2_new}
|\nabla^\alpha_i \dot{\psi}(x)|\leq 2^{|\alpha|-1} \, \sup_{y\in \{x\}^\ast}|\nabla_i\dot{\psi}(y)| \leq 2^{R_0 d}\,  h \,  |\dot{\psi}|_{0,\{x\}},                                                  
\end{align}
where we used the definition  \eqref{eq:primal_norm}  of $| \cdot|_{0, X}$
and the fact that for $j=0$ the weights
in   \eqref{eq:definition_weights} reduce to $\wpzc_0(i,\alpha)  = h$. 
Now, the bound \eqref{eq:K0:2_new}
and the definition of the norm $| \cdot|_{0,\{x\}, T_0}$  by duality yield the estimate
\begin{align}
\abs{ \Mscr_\mpzc(\{x\}) }_{0,\{x\},T_0} \le  2^{R_0 d}\,  h.
\end{align}

From \eqref{eq:K0:1_new}, \eqref{eq:K0:2_new}, the condition $h \ge 1$, and the definition \eqref{eq:normE}, we infer that
\begin{align}
\frac{|I(\mathcal{K})(\{x\})|_{0,\{x\},T_\p}}{w_{-1:0}^{\{x\}}(\p)}\leq 2^{R_0 dr_0} \, h^{r_0} \, \lVert \mathcal{K}\rVert_{\zeta}
 \frac{e^{\frac12(1-\zeta) \Qscr_0(D\p(x))}}{e^{\frac12(1-\zeta) \Qscr_0(D\p(x))}}   
= 2^{R_0 dr_0} \, 	h^{r_0}\lVert \mathcal{K}\rVert_{\zeta}.
 \end{align}
 Given that $w_{-1,0}^X$ factors over any polymer, 
 the submultiplicativity estimate  \eqref{eq:product_estimate_concrete} combined with the trivial
estimate $| \cdot |_{0, X, T_\p} \le |\cdot |_{0, x, T_\p}$ whenever $x \in X$  implies that
\begin{align}
\frac{|  I(\mathcal{K})(X)|_{0,X,T_\p}}{w_{-1:0}^{X}(\p)}\leq 
 ( 2^{R_0 dr_0} \, h^{r_0})^{|X|} \, \lVert \mathcal{K}\rVert_{\zeta}^{|X|}.
\end{align}
Thus, for $ \rho_1 \le  (2^{R_0 dr_0} \, h^{r_0} 4A)^{-1}$,
\begin{align}   \label{eq:estimate_I_without_deriv}
\| I(\mathcal{K})(X) \|_{-1:0}^{(4A)}  \le  \sup_{X \in \mathcal P_0^c} \left(2^{R_0 dr_0} \, h^{r_0} 4A \,  \lVert \mathcal{K}
\rVert_{\zeta}\right)^{|X|}  \le 2^{R_0dr_0} \, h^{r_0} 4A\,   \lVert \mathcal{K}\rVert_{\zeta} \le 1.
\end{align}
This proves  \eqref{eq:bound_I_ell0}.
			
The derivatives are estimated similarly as in the proof of Lemma \ref{le:P3}.		
For $ \rho_1 \le  (2^{R_0dr_0+1} \, h^{r_0} 4A)^{-1}$ we get
\begin{align}   \label{eq:estimate_I_with_deriv}
\begin{split}
& (4A)^{|X|} \, \,   \frac{1}{j!}\lVert D^jI(\mathcal{K})(\dot{\mathcal{K}},\ldots, \dot{\mathcal{K}})(X)\rVert_{0,X} 
\le \tbinom{|X|}{j}      (2^{R_0 dr_0} \, h^{r_0} 4A)^{|X|} \, \,    \lVert \mathcal{K}\rVert_{\zeta}^{|X|-j}  \, \, 
\lVert \dot{\mathcal{K}}\rVert_{\zeta}^{j}   \\
&  \le   (2^{R_0 dr_0 + 1} \, h^{r_0} 4A)^{|X|}  \, \,   \lVert \mathcal{K}\rVert_{\zeta}^{|X|-j}  \, \, 
\lVert \dot{\mathcal{K}}\rVert_{\zeta}^{j}  \le 
\bigl(2^{R_0 dr_0 + 1} \, h^{r_0} 4A     \lVert \dot{\mathcal{K}}\rVert_{\zeta}\bigr)^{j}.
\end{split}
\end{align}
This shows that $I$ is smooth on $B_{\rho_1}(0)$ and the estimate \eqref{eq:bound_I_ell} holds.
\end{proof}
			
\begin{lemma}  \label{le:smoothness_P4}
Assume that 
\begin{equation} \label{eq:restriction_h_P4}
h \ge \max(\delta(L)^{-1/2}, 1),
\end{equation}				
where $\delta(L)$ was introduced in \eqref{eq:defofdelta}.
Then the  map $P_4$ defined in \eqref{eq:defofP4} is smooth on the set
$\boldsymbol{M}_{-1:0}^{(4A)}\times B_1(1)$ with $B_1(1)\subset (M(\mathcal{B}_0),\opnorm{\cdot}_0)$.
Moreover, on that set we  have
\begin{equation} \label{eq:P4_bound_ell0}
\| P_4(K, F)\|_0^{(A)} \le \| K\|_{-1:0}^{(4A)},
\end{equation}
\begin{align} \label{eq:P4_bound_ell}
\frac{1}{j_2!} \| D^{j_1}_K  D^{j_2}_F P_4(K, F)\|_0^{(A)} \le
\bigl(  \| K\|_{-1:0}^{(4A)}\bigr)^{1-j_1}   \bigl(  \| \dot K\|_{-1:0}^{(4A)}\bigr)^{j_1} \opnorm{\dot{F}}_0^j.
\end{align}
for $j_1 \in \{0,1\}$; the left hand side vanishes for $j _1 \ge 2$. 
\end{lemma}
		
\begin{proof} 
For brevity we write $\delta$ instead of $\delta(L)$. 
It follows from the definitions of the quadratic forms  $\boldsymbol{M}_0^X$ and $\boldsymbol{G}_0^X$ in  \eqref{eq:defofMk}
and \eqref{eq:strong_weight} and assumption  \eqref{eq:restriction_h_P4} that  $\delta \boldsymbol M_0^X\geq \boldsymbol G_0^X$.
aking into account that  in \eqref{eq:defAk} we have  $\delta_0 = \delta$ and $4 \weightzeta = \zeta$
(see \eqref{eq:definition_weightzeta}) we deduce from \eqref{eq:defAk} that
\begin{align}  \label{eq:consistency_weightzeta}
(\p, \boldsymbol{A}_0^X \p) \ge (1- \zeta) \sum_{x \in X} \Qscr(D\p(x)) + (\boldsymbol{G}_0^X \p, \p).
\end{align}
Since $w_0^X(\p) = e^{\frac12 (\boldsymbol{A}_0^X \p, \p)}$ and $W_0^X(\p) = e^{\frac12 (\boldsymbol{G}_0^X \p, \p)}$ 
the definition of $w_{-1:0}^X$ in \eqref{eq:w-1} implies that 
\begin{align*}
w_0^X \ge w_{-1:0}^X W_0^X = w_{-1:0}^X \prod_{B \in \mathcal{B}_0(X)} W_0^B.
\end{align*}
Together with Lemma~\ref{le:norms_pointwise} we get 
\begin{align}
\begin{split}
\lVert P_4(K,F)(X)\rVert_{0,X}
& =\sup_{\p}\frac{|F^X(\p)K(X,\p)|_{0,X,T_\p}}{w_0^X(\p)}                                                     
\\  &
\leq                                                                                                       
\sup_{\p}\frac{|K(X,\p)|_{0,X,T_\p}}{w_{-1:0}^X(\p)}\hspace{-0.2cm}\prod_{B\in \mathcal{B}_0(X)}
\sup_{\p}\frac{|F(B,\p)|_{0,B,T_\p}}{W_0^B(\p)}  \\
& \leq \lVert K\rVert_{-1:0}^{(4A)}  \, (4A)^{-|X|}  \, \opnorm{F}^{|X|}_0
\end{split}
\end{align}
where we used that $F \in B_1(1) \subset B_2(0) \subset (M(\mathcal{B}_0),\opnorm{\cdot}_0)$. 			
Multiplying by $A^{|X|}$ and taking the supremum over $X$ we get \eqref{eq:P4_bound_ell0}.
			
To estimate  the derivatives, we observe that $P_4$ is linear in $K$. 
Therefore it is sufficient to note that
\begin{align}
\frac{1}{j!}\lVert D^j_F P_4(K,F)(\dot{F},\ldots,\dot{F})(X)\rVert_{0,X}  
\leq \lVert K\rVert_{-1:0}^{(4A)}\, (2A)^{-|X|} 2^{|X|} \,  \opnorm{\dot{F}}_0^j ,
\end{align}
where the additional factor $2^{|X|}$ is again the combinatorial factor of the derivatives.
Hence 
\begin{equation}
\frac{1}{j!}\lVert D^j_F P_4(K,F)(\dot{F},\ldots,\dot{F})\rVert_{0}^{(A)} \le
\lVert K\rVert_{-1:0}^{(4A)}\, \,     \opnorm{\dot{F}}_0^j.
\end{equation}
\end{proof}

\begin{proof}[Proof of Lemma~\ref{le:smoothness_K0}]
To see that the $\widehat K_0$ is smooth on $B_{\rho_1} \times B_{\rho_2}$ for the given values of $\rho_1$ and $\rho_2$,
it suffices to note that $I$ is smooth on $B_{\rho_1}$ and $E$ maps $B_{\frac18}(0)$
to $B_1(1)$  (see \eqref{eq:bound_exp_H}). Then the assertion follows from the fact that
$P_4$ is smooth on $\boldsymbol{M}_{-1:0}^{(4A)}\times B_1(1)$.
The bound  \eqref{eq:bound_K0_ell0} for $j_2 = 0$ with $C_0 =1$ follows from 
\eqref{eq:P4_bound_ell0} and  \eqref{eq:bound_I_ell0}. The other bounds follow from the bounds in 
Lemma~\ref{le:smoothness_exp},  Lemma~\ref{le:I}, and Lemma~\ref{le:smoothness_P4},  combined with the chain rule.
\end{proof}
	
\begin{proof}[Proof of Theorem~\ref{propfixedpoint}; Conclusion]
We first note that the map $\mathcal H \mapsto \boldsymbol q(\mathcal H)$ is linear and satisfies the bound
\begin{equation} \label{eq:bound_mcH_to_q}
|\boldsymbol q(\mathcal H)| \le \frac{C}{h^2} \norm{\mathcal H}_{0,0} \le C \norm{\mathcal H}_{0,0}. 
\end{equation}
This follows from the definition of the norm $\| \cdot \|_{0,0}$ in 
\eqref{hamiltoniannorm} and the fact that all norms on  $\mathbb{R}^{(d\times m)\times (d\times m)}_{\mathrm{sym}}$ are equivalent. 
	
Next, we establish smoothness of the coordinate maps for $\tilde{H}_k$ and $\tilde{K}_k$
in a neighbourhood of the origin.
First, we consider the maps $\tilde K'$ with $k \ge 1$. 
Then $\tilde K'(\mathcal K, \mathcal H, Z) = \myS_k(H_k, K_k, \boldsymbol q(\mathcal H))$ and, in particular,
$\tilde K'$ does not depend on $\mathcal K$. Smoothness of $\tilde K'$ follows from the smoothness of $\myS_k$
(see Theorem~\ref{prop:smoothnessofS}) and \eqref{eq:bound_mcH_to_q} as long as 
\begin{equation}  \label{eq:first_condition_rho2_fixed_point}
\rho_2 \le \frac{\kappa(L)}{C_{   \eqref{eq:bound_mcH_to_q} }    }.
\end{equation}
Regarding the bounds on the derivatives of $\tilde K'$ we have
\begin{align} \label{eq:bounds_tildeKk_k_ge_2}
\begin{split}
& \frac{1}{j_2! j_3!} \norm*{\frac{1}{\eta^{k+1}}  D^{j_2}_{\mathcal H} D^{j_3}_Z \tilde K_{k+1}(\mathcal K, \mathcal H, Z)
(\dot{\mathcal H}, \ldots, \dot{\mathcal Z}) }_{k+1}^{(A)} \\
&\hspace{3cm}\le  C_{j_2, j_3}(L,A) \frac1\eta \, \, \frac{1}{\eta^k} 
\left(     \| \dot K_k\|_{k}^{(A)} + \| \dot H_k\|_{k,0}\right)^{j_3} C_{   \eqref{eq:bound_mcH_to_q} }^{j_2} \, \|\dot{\mathcal H}\|_{0,0}^{j_2}\\
&\hspace{3cm} \le  C_{j_2, j_3}(L,A) \,  \| \dot{Z}\|_{\mathcal Z}^{j_3}  \, \, \|\dot{\mathcal H}\|_{0,0}^{j_2}.
\end{split}
\end{align}
Here, we used the convention that we  indicate the dependence of constants on fixed parameters like $\eta$.
Similarly, smoothness  of $\tilde H_k$ and the bounds on the derivatives follow from \eqref{eq:bound_mcH_to_q},
\eqref{eq:contractionABC} and \eqref{eq:qderivABC}.
	
The main point is to show smoothness of the composition map
$\tilde K_1(\mathcal K, \mathcal H, Z) = \myS_0(\widehat{K}_0(\mathcal K, \mathcal H), H_0, \boldsymbol q(\mathcal H))$ 
and to bound its derivatives. 
We first note that, for $\rho_1$ and $\rho_2$ given by \eqref{eq:smoothness_mcT_neighbourhood},
\begin{equation} \label{eq:tildeK_1_neighbourhoods}
\widehat K_0(B_{\rho_1}(0) \times B_{\rho_2}(0)) \subset B_{\rho(A)}(0).
\end{equation}
Indeed, this follows directly from \eqref{eq:bound_K0_ell0} with $j_2 =0$.
Now the desired properties of $\tilde K_1$ follow from Lemma~\ref{le:smoothness_K0},
Theorem~\ref{prop:smoothnessofS}, and the chain rule. 
	
Next  we show that, at the origin, the differential of the map $Z \mapsto \mathcal T(\mathcal K, \mathcal H, Z)$ is contraction.
From the definition of the maps $\tilde K_k$ and $\tilde H_k$  
in \eqref{eq:tilde_H0_def}--\eqref{eq:tilde_Kk+1_def},  in combination with
\eqref{eq:triangular_derivative} in Theorem~\ref{prop:contractivity}, it follows that
\begin{align}\label{derivativesofcomplicatedmap}
D_{H_{k+1}} \tilde H_k(0,0,0)         = (\boldsymbol A_k^{(0)})^{-1} \quad \hbox{for} \quad  0 \le k \le N-2, \\
D_{K_k} \tilde H_k(0,0,0) = -(\boldsymbol A_k^{(0)})^{-1}\boldsymbol B_k^{(0)},\text{ for }  1 \le k \le N-1,\\
D_{K_k} \tilde K_{k+1}(0,0,0) =\boldsymbol C_k^{(0)}  \text{ for } 1 \le k \le  N-1,                                                                                       
\end{align}
and all other derivatives vanish.
To estimate the operator norm of $D_Z \mathcal{T}(0,0,0)$, 
let $\dot Z\in\mathcal{Z}$ with $\lVert \dot{Z}\rVert_{\mathcal{Z}}\leq 1$ and set
\begin{align} 
Z' = D_Z \mathcal{T}(0,0,0) \dot Z.
\end{align}
We denote the coordinates of $\dot Z$ by $\dot H_k$ and $\dot K_k$ and the coordinates of $Z_{k+1}$ by $H_{k+1}$ and $K_{k+1}$.
The definition of the norm on $\mathcal{Z}$ implies that
$\lVert \dot{H}_k\rVert_{k,0}\leq \eta^k $ and $\lVert \dot{K}_k\rVert_{k}\leq \eta^k$.
The bounds from Theorem  \ref{prop:contractivity} implies for $1\leq k\leq N-2$ that
\begin{align}  
\begin{split}
\eta^{-k}\lVert {H}_k'\rVert_{k,0}           & \leq                                                       
\eta^{-k}\lVert (\boldsymbol{A}_k^{(0)})^{-1}\rVert\eta^{k+1}+
\eta^{-k}\lVert (\boldsymbol{A}_k^{(0)})^{-1}\rVert     \, \, \lVert \boldsymbol{B}_k^{(0)}\rVert   \eta^k 
\leq \tfrac{3}{4}(\eta+\tfrac13),
\end{split}
\end{align}
and for $2\leq k\leq N$ 
\begin{align}
\eta^{-k}\lVert {K}_k'\rVert                 & \leq \eta^{-k}\lVert \boldsymbol{C}_{k-1}^{(0)}\rVert {\eta^{k-1}}\leq   
\tfrac{1}{\eta} \, \tfrac34 \eta = \tfrac34,
\end{align}
and finally for the boundary terms
\begin{align}\begin{split}
\norm{{H}_0'}_{0,0}     & \leq \norm{(\boldsymbol{A}_0^{(0)})^{-1}}\eta\leq \tfrac{3}{4}\eta, \\
\eta^{-(N-1)}\lVert {H}_{N-1}'\rVert_{N-1,0} & \leq                                                       
\eta^{-(N-1)}\lVert (\boldsymbol{A}_{N-1}^{(0)})^{-1}\rVert   \, \, \lVert \boldsymbol{B}_{N-1}^{(0)}\rVert \eta^{N-1}
\leq \tfrac{3}{4}\cdot\tfrac{1}{3}=\tfrac14,
\\
\eta^{-1}\lVert {K}_1'\rVert                 & =0,                                                   
\\
\end{split}
\end{align}
Since $\eta \le \frac23$, these estimates  imply that
\begin{align}  \label{eq:contraction_mcT}
\norm{D_Z \mathcal{T}(0,0,0)} \le \tfrac34.                                                                         
\end{align}
Thus the assumptions of the implicit function theorem are satisfied for the map $\mathcal{T}-\pi_3$ 
and since $\mathcal{T}$  is  smooth (with bounds on the derivatives that are independent of $N$),  
the implicit function $\hat{Z}$ is defined in a neighbourhood
$B_{\epsilon_1}\times B_{\epsilon_2}\subset \boldsymbol{E}\times M(\mathcal{B}_0)$ 
with $\epsilon_1$ and $\epsilon_2$ independent of $N$ and the derivatives of $\widehat Z$ 
can be bounded independent of $N$. 
\end{proof}
	
	To show that the  choice of the constants $\epsilon_1$, 
$\epsilon_2$, and $\epsilon$ specified after Theorem~\ref{propfixedpoint} is sufficient, 
assume that 
\begin{equation} \label{eq:zeroth_condition_eps_i}
\epsilon_1 \le \rho_1, \quad \epsilon_2 \le \rho_2, \quad \epsilon \le \rho(A)/2
\end{equation}
and that
\begin{equation} \label{eq:first_condition_eps_i}
2 \overline C_{0,0,2}\,  \epsilon  + \overline C_{1,0,1} \, \epsilon_1 + \overline C_{0,1,1} \,  \epsilon_2 \le \tfrac18.	
\end{equation}
Then, for $(\mathcal K, \mathcal H, Z) \in B_{\epsilon_1} \times B_{\epsilon_2} \times \overline{B_\epsilon}$, we have 
\begin{equation} \label{eq:mcT_general_contraction}
\|D_Z \mathcal T(\mathcal K, \mathcal H, Z)\| \le \tfrac34 + \| D_{Z} \mathcal T(\mathcal K, \mathcal H, Z) - D_Z \mathcal T(0,0,0) \| \le \tfrac78.
\end{equation}
Note that  the definition of $\mathcal{T}$ implies that $\mathcal T(0,\mathcal H,0) = 0$  for all $\mathcal H \in B_{\epsilon_2}(0) \subset B_{\rho_2}(0)$.  
Thus, if in addition 
\begin{equation} \label{eq:second_condition_eps_i}
\overline C_{1,0,0}  \, \epsilon_1 \le \tfrac18 \epsilon,
\end{equation}
we have 
\begin{equation} \label{eq:mcT_general_selfmap}
 \| \mathcal T(\mathcal K, \mathcal H, 0) \| < \tfrac18 \epsilon   \text{ for all } \, (\mathcal H, \mathcal K) 
\in B_{\epsilon_1}(0) \times B_{\epsilon_2}(0).
\end{equation}
It follows from \eqref{eq:mcT_general_contraction} and  \eqref{eq:mcT_general_selfmap} that, 
for all  $(\mathcal K, \mathcal H) \in B_{\epsilon_1}(0) \times B_{\epsilon_2}(0)$, the map
$Z \mapsto \mathcal T(\mathcal K, \mathcal H, Z)$ is a contraction and maps the closed ball
$\overline{B_\epsilon}$ to itself. Thus, by the Banach fixed point theorem, there is a unique
$\widehat Z(\mathcal K, \mathcal H) \in \overline{B_\epsilon}$ such that
\begin{equation}
\mathcal T(\mathcal K, \mathcal H, \widehat Z(\mathcal K, \mathcal H)) = \widehat Z(\mathcal K, \mathcal H).
\end{equation}
Moreover,  $\| \widehat Z(\mathcal K, \mathcal H) \|  \le 8  \mathcal \| \mathcal{T}(\mathcal K, \mathcal H, 0)\| < \epsilon$
so that $\widehat Z(\mathcal K, \mathcal H) \in B_\epsilon(0)$. 
It follows from the implicit function theorem applied 
at the point $(\mathcal K, \mathcal H, \widehat Z(\mathcal H, \mathcal K))$ that the function $\widehat  Z$ is locally smooth. 
By the uniqueness, $\widehat  Z$ is smooth in $B_{\epsilon_1}(0) \times B_{\epsilon_2}(0)$.
Finally, one easily sees that the choices in  \eqref{eq:conditions_eps_i_A} imply
\eqref{eq:zeroth_condition_eps_i}, \eqref{eq:first_condition_eps_i}, and \eqref{eq:second_condition_eps_i}.

\section{Existence of a fixed point of the map \texorpdfstring{$\Pi_{H_0}\hat{Z}(\mathcal{K},\cdot)$}{PiZ(K,.)}}
\label{sec:exist2fixpoint}

Theorem~\ref{propfixedpoint} and the definition of the  norm $\| \cdot \|_{\mathcal Z}$  show the existence of  a sequence of  maps
 $H_k$ and $K_k$ such that 
 $\myR_{k+1}^{(\boldsymbol{q})}\left( e^{-H_k}\circ K_k\right)(\p)=\left(e^{-H_{k+1}}\circ K_{k+1}\right)(\p)$,
 the coordinate $K_k$ is exponentially decreasing, and $\HN=0$.  
 In particular,  equation \eqref{eqfirstfixedpointid} holds.
But this sequence  will  in general  not satisfy the correct initial condition because the $H_0$ 
coordinate of the fixed point is only given implicitly by the fixed point equation.
We can, however,  use the artificially inserted coordinate $\mathcal{H}$
and apply the  implicit function theorem once more to show that we can choose $\mathcal{H}$ such that $H_0=\mathcal{H}$.
Then a simple calculation shows that this fixed point satisfies the correct initial condition up to an explicit
scalar factor, see  \eqref{finalidentity} and   \eqref{finalidentity2} below.

We use the same notation as in Theorem~\ref{propfixedpoint}. 
In particular  $\widehat Z: B_{\epsilon_1}(0) \times B_{\epsilon_2}(0) \to \mathcal Z$ denotes the fixed point map.  We 
use $\Pi_{H_0}$ to denote  the bounded  linear map $\Pi_{H_0}:\mathcal{Z}\rightarrow M_0(\mathcal{B}_0)$ 
that extracts the coordinate $H_0$ from $Z$.
\begin{lemma}\label{lemmafixedpoint}
Under the assumptions of  Theorem~\ref{propfixedpoint}, 
there is a constant $\tilde \rhoMT >0$ that can be chosen independently of $N$ 
and a map  $\widehat{ \mathcal{H}}:B_{\tilde \rhoMT}(0)\subset \boldsymbol{E} \rightarrow B_{\epsilon_2}(0) \subset M_0(\mathcal{B}_0)$
 such that
\begin{equation} 
\Pi_{H_0}\hat{Z}(\mathcal{K},\widehat{\mathcal{H}}(\mathcal{K}))=\widehat{\mathcal{H}}(\mathcal{K})
\text{ and }  \boldsymbol q(\widehat{\mathcal H}(\mathcal K)) \subset B_{\kappa}(0)
\text{ for all } \mathcal K \in B_{\tilde \rhoMT}(0).
\end{equation} 
Moreover $\widehat{\mathcal{H}}$ is smooth in $B_{\tilde \rhoMT}(0)$ and its derivatives can be bounded
uniformly in $N$. 
We may take
\begin{equation}  \label{eq:rho_second_fixed_point}
{\tilde \rhoMT} = \min \Bigl(  \frac{1}{4 C_{1,1}},  \frac{\rho'}{2 C_{1,0 }},  \epsilon_1  \Bigr) \text{ where }  
\rho' = \min \Bigl(   \frac{1}{8C_{0,2}},  \frac{\epsilon_2}{2}\Bigr).
\end{equation}
Here, $C_{j_1, j_2}$ are the constants in the estimate 
 \eqref{firstfixedpoint_bounds} for the derivatives of $\widehat Z$. 
\end{lemma}
Note that the conditions $\widehat{\mathcal H}(B_{\tilde \rhoMT}(0)) \subset B_{\epsilon_2}(0)$,
\eqref{eq:range_q_of_mcH}, and the fact that $\epsilon_2 \le \rho_2$,
 imply that 
 \begin{align} \label{eq:range_q_bis}
  \boldsymbol q(\widehat{\mathcal H}(\mathcal K)) \in B_{\kappa}(0)\text{ for all  } \mathcal K \in B_\rho(0).
 \end{align}

\begin{proof} We first note that
$\mathcal T(0, \mathcal H, 0) = 0$. Hence by uniqueness of the fixed point we get
\begin{equation}  \label{eq:hatZ_zero}
\widehat Z(0, \mathcal H) = 0   \quad \hbox{  for all    $\mathcal H \in B_{\epsilon_2}(0)$}
\end{equation}
and in particular
\begin{equation} \label{eq:D2_hatZ}
D_{\mathcal H} \hat{Z}(0,0) = 0.
\end{equation}
We  now  consider the function 
$$
f =\Pi_{H_0}  \circ \widehat{Z}-\pi_2: B_{\epsilon_1} \times B_{\epsilon_2} 
\subset \boldsymbol{E}\times M_0(\mathcal{B}_0)\rightarrow M_0(\mathcal{B}_0),
$$
where $\pi_2(\mathcal K, \mathcal H) := \mathcal H$.
Condition  \eqref{eq:D2_hatZ} implies that 
$D_{\mathcal H} f(0,0) \dot{\mathcal H} = - \dot{\mathcal H}$.
 Hence we can apply
the implicit function theorem to $f$ and find a constant ${\tilde \rhoMT} >0$ and   
a smooth  function 
$\widehat{\mathcal{H}}:B_{\tilde \rhoMT}(0)\subset \boldsymbol{E}\rightarrow M_0(\mathcal{B}_0)$ such that
$f(   \Pi_{H_0}\hat{Z}(\mathcal{K}, \widehat {\mathcal H}(\mathcal{K})), \widehat{\mathcal H}(\mathcal K)     ) = 0$, i.e., 
$$
\Pi_{H_0}\hat{Z}(\mathcal{K}, \widehat {\mathcal H}(\mathcal{K}))=\widehat{ \mathcal H}(\mathcal{K}).
$$
Given that the derivatives of $\hat{Z}$ are bounded uniformly in $N$, we can choose  ${\tilde \rhoMT}$  independent of  $N$.

It only remains to show that the choice \eqref{eq:rho_second_fixed_point} for ${\tilde \rhoMT}$ is admissible
and $\widehat{\mathcal H}(B_\rho(0)) \subset B_{\epsilon_2}(0)$. 
To see this we argue exactly as in the proof of Theorem~\ref{propfixedpoint}. 
First, assume that  $\rho' \le \epsilon_2/2$ and
\begin{equation} \label{eq:first_condition_rho_second_fix}
2 C_{0,2}\,  \rho' + C_{1,1} \,  {\tilde \rhoMT}  \le \tfrac12.
\end{equation}
Then   $D_{\mathcal H} \widehat Z(0,0) = 0$   implies that
$$
\| D_{\mathcal H}  (\Pi_{H_0} \circ \hat Z)\| \le \| D_{\mathcal H} \widehat Z\| \le \tfrac12 \text{ in }
B_{{\tilde \rhoMT}}(0) \times B_{\rho'}(0).
$$
If, in addition,
\begin{equation} \label{eq:first_condition_rho_second_fix2}
 C_{1,0} \, {\tilde \rhoMT}\le \tfrac12 \rho',
\end{equation}
then $\|  (\Pi_{H_0} \circ \hat Z)(\mathcal K, 0)\| \le \frac12 \rho'$ for $\mathcal K \in B_\rho(0)$. Thus, for such
$\mathcal K$, the map $\mathcal H \mapsto (\Pi_{H_0} \circ \widehat Z)(\mathcal K, \mathcal H)$ is a contraction 
and maps $\overline{B_{\rho'}(0)}$ to itself. Hence, this map has a unique fixed point
$\widehat{\mathcal H}(\mathcal K) \in \overline{B_{\rho'}(0)} \subset B_{\epsilon_2}(0)$.
        Smoothness of $\widehat{\mathcal H}$ follows from the implicit function theorem.
\end{proof}

\section{Proof of Theorem~\ref{MAINTHEOREM}}
\label{sec:proofmain}

\begin{proof}
The heart of the matter is the identity  \eqref{finalidentity} below.
In combination with Lemma~\ref{lemmafixedpoint} and the identity \eqref{eqfirstfixedpointid}  it immediately yields
  the representation
\eqref{thmequality}. The further assertions  in Theorem~\ref{maintheorem}
then follow from the properties of the map $\widehat{\mathcal H}$ stated in Lemma~\ref{lemmafixedpoint}.
  To simplify the notation we write $\hat e$ and $\widehat{\boldsymbol q}$ instead of $\hat e_N$ and
 $\widehat{\boldsymbol q}_N$ for the maps whose existence is claimed in  Theorem~\ref{maintheorem}.

  Recall that for an ideal Hamiltonian $\mathcal H$ we denote  the matrix which defines the quadratic part by $\boldsymbol q(\mathcal H)$.
We denote the constant part by $e(\mathcal H)$.
Then 
\begin{align}  \label{eq:eq_ideal_hamiltonian}
\sum_{x\in \TN}\mathcal{H}(\mathcal{K})(x,\p)= e(\mathcal H) L^{Nd}
+\frac{1}{2}{\sum_{x\in\TN} \langle \boldsymbol{q}(\mathcal H) \nabla\p(x),\nabla\p(x)\rangle}, 
\end{align}
where we used that the sum over the linear terms in the field vanishes  because $\sum_{x\in\TN}\nabla^\alpha \p_i(x)=0$
for any $\p\in\Xcal_N$ and any multiindex $\alpha$ and $1\leq i\leq m$, due to the periodic boundary conditions. 
Recall that $\lambda$ is the Hausdorff measure on $\Xcal_N$. The definition of the partition function 
$Z^{(\boldsymbol q)}$ implies that 
\begin{align}
\begin{split}
e^{\frac12	\sum_{x\in\TN} \langle \boldsymbol{q}\nabla\p(x),\nabla\p(x)\rangle}\,\mu^{(0)}(\d\p)& = 
\frac{Z^{(\boldsymbol{q})}e^{-\frac{1}{2}\sum_{x\in\TN} \Qscr(D\p(x)) -
\langle   \boldsymbol{q}\nabla\p(x),\nabla\p(x)\rangle}\lambda(\d\p) }{Z^{(\boldsymbol{q})}Z^{(0)}}
\\ & =                                                                               
\frac{Z^{(\boldsymbol{q})}}{Z^{(0)}}\,\mu^{(\boldsymbol{q})}(\d\p).
\end{split}        
\end{align}
Recall also that $\mathcal K(X, \p) = \prod_{x \in X} \mathcal K(D \p(x))$ (see \eqref{eq:explanation_K}). Thus, 
by the definition \eqref{defofk0} of $\widehat{K}_0(\mathcal K, \mathcal H)$,
\begin{align}  \label{eq:prefinalidentity}
\begin{split}
(\widehat{K}_0(\mathcal K, \mathcal H)  \circ e^{ - \mathcal H}   )(\TN, \p) &= 
  (\mathcal K e^{  -  \mathcal H}  \circ e^{  -  \mathcal H})(\TN, \p)  \\
&=\sum_{X \subset \TN} \mathcal K(X, \p) e^{ - \mathcal H}(X, \p) \, e^{ - \mathcal H}(\TN\setminus X, \p)
\\ &
= \sum_{X \subset \TN} \mathcal K(X, \p)    \, \, e^{ - \sum_{x \in \TN} \mathcal H(x, \p)}.
\end{split}
\end{align}
Using the  identities  \eqref{eq:eq_ideal_hamiltonian}--\eqref{eq:prefinalidentity}, we get
\begin{align} \begin{split}  \label{finalidentity}
\int_{\Xcal_N} \sum_{X\subset \TN} & \mathcal{K}(X, \p)\,\mu^{(0)}(\d\p)\\             
&=\int_{\Xcal_N} \left(  \widehat K_0(\mathcal K, \mathcal H)\circ e^{   - \mathcal{H}   }\right)
\left(\TN,\p\right)\cdot   e^{\sum_{x\in \TN}\mathcal{H}(x,\p)}\,\mu^{(0)}(\d\p)\\
& =\frac{Z^{(\boldsymbol{q(\mathcal H)})}}{Z^{(0)}} e^{L^{Nd}     e(\mathcal H)}
 \int_{\Xcal_N}   \left(\widehat K_0(\mathcal K, \mathcal H) \circ  e^{   - \mathcal H}\right)(\TN,\p)\,\mu^{(\boldsymbol{q(\mathcal H)})}(\d\p).
\end{split}\end{align}
Now let $0<\rhoMT <\tilde{\rhoMT}$  with $\tilde{\rhoMT}$  as in Lemma~\ref{lemmafixedpoint} 
and define the following  maps on $B_\rhoMT(0) \subset \boldsymbol E$,
\begin{equation}
\widehat{\boldsymbol q}(\mathcal K) := \boldsymbol q( \widehat{\mathcal H}(\mathcal K)), \quad
\widehat{e}(\mathcal K) := e( \widehat{\mathcal H}(\mathcal K)), \quad
\widehat K_N(\mathcal K) := \Pi_{K_N} \widehat{Z}(\mathcal K, \widehat{\mathcal H}(\mathcal K)).
\end{equation}
Here $\Pi_{K_N}$ denotes the projection from $Z$ to the $K_N$ coordinate of $Z$. 
By Lemma~\ref{lemmafixedpoint} we have
\begin{align}
\Pi_{H_0} \widehat Z(\mathcal K, \widehat{\mathcal H}(\mathcal K)) = \widehat{\mathcal H}(\mathcal K).
\end{align}
Using the abbreviation $H_0 = \Pi_{H_0} \widehat Z(\mathcal K, \widehat{\mathcal H}(\mathcal K))$, we get 
\begin{align}   \label{finalidentity2}
\begin{split}
\int_{\Xcal_N} \left( \widehat K_0(\mathcal K,  \widehat{\mathcal H}(\mathcal K) )\circ e^{     -\widehat{\mathcal H}(\mathcal K)} \right)    
&(\TN,\p)\, 
\mu^{(\boldsymbol{q}(\widehat{\mathcal H}(\mathcal K))}(\d\p)\\
=&
 \int_{\Xcal_N}  \left(   \widehat K_0(\mathcal K, \widehat{\mathcal H}(\mathcal K)) \circ e^{- H_0}\right)(\TN,\p)
\,     \mu^{ (\boldsymbol{q}(\widehat{\mathcal H}(\mathcal K))}(\d\p) \\
 \underset{\eqref{eqfirstfixedpointid}  }{=}  \, 	 
&\int_{\Xcal_N}   \left( 1 + \widehat K_N(\mathcal K)(\TN, \p)  \right) \, \mu_{N+1}^{(\boldsymbol{q}(\widehat{\mathcal H}(\mathcal K))      }(\d\p).
 \end{split}	
\end{align}
Taking $\mathcal H = \widehat{\mathcal H}(\mathcal K)$ in \eqref{finalidentity} 
and using that $\boldsymbol q(\widehat{\mathcal H}(\mathcal K))= \widehat{\boldsymbol q}(\mathcal K)$ and 
$e(\widehat{\mathcal H}(\mathcal K)) = \widehat{e}(\mathcal K)$ we obtain the desired representation  \eqref{thmequality}.

Smoothness of maps 
$\widehat{\boldsymbol q} $, $\widehat{e}$ and $\widehat{K}_N$  as well
as bounds on the derivatives which are independent on $N$ follow from the same property 
for $\widehat{\mathcal H}$ and $\widehat{ Z}$ as well  the linearity  and uniform boundedness of the projections 
$\mathcal H \mapsto  \boldsymbol q(\mathcal H)$, $\mathcal H  \mapsto e(\mathcal H)$ and $Z \mapsto K_N$. 
In particular uniform bounds on the derivatives of $\mathcal K \mapsto \widehat{ Z}( \mathcal K, \widehat{\mathcal H}(\mathcal K))$
and the definition $\| \cdot \|_{\mathcal Z}$ imply that
\begin{align} \frac{1}{\eta^N}
\frac{1}{l!}    \|  D^\ell_{\mathcal K} \widehat{K}_N(\mathcal K)( \dot{\mathcal K}, \ldots, \dot{\mathcal K}) \|_N^{(A)} 
 \le C_\ell(L,h,A)  \, \| \dot{\mathcal K}\|_\zeta^\ell.
\end{align}
This proves \eqref{eq:keyboundKN}.
To show \eqref{eq:keyboundKN_ell0} we note that the definition of $\| \cdot\|_{N}^{(A)}$ and
Theorem~\ref{th:weights_final}~\ref{w:w8}
yield
\begin{align}
\begin{split}
	 \quad   \int_{\Xcal} \widehat K_N(\TN, \p) \, \mu_{N+1}^{(\boldsymbol q)}(d\p)
&\le  \int_{\Xcal}   \frac1A \| \widehat K_N\|_N^{(A)}    w_N(\p) \,  \mu_{N+1}^{(\boldsymbol q)}(d\p) \\
& \le \, \frac1A \| \widehat K_N\|_N^{(A)}    \AB  \,  w_{N:N+1}(0) 
=  \frac{\AB}{A}  \| \widehat K_N\|_N^{(A)}.
\end{split}
\end{align}
Given that $\widehat{K}_N(0) = 0$, it follows from   \eqref{eq:keyboundKN}   with $\ell = 1$ that 
$ \| \widehat K_N\|_N^{(A)} \le  C_{1, \eqref{eq:keyboundKN}}   \eta^N \| \mathcal K\|_\zeta$. 
Thus, the bound \eqref{eq:keyboundKN_ell0} holds if $\rhoMT$ satisfies, in addition, the bound 
\begin{equation}
\frac{\AB}{A} C_{1, \eqref{eq:keyboundKN}}   \eta^N \rhoMT \le \frac12.
\end{equation}

Finally,  the representation  \eqref{thmequality2}  can be derived arguing as in 
 \eqref{finalidentity} and \eqref{finalidentity2} and using Gaussian calculus.
More precisely, we use that for every positive quadratic form $\mathscr{C}$
\begin{align}
(f_N, \p + \mathscr{C} f_N) - \frac12 (\mathscr{C}^{-1} (\p + \mathscr C f_N) , \p + \mathscr C f_N) =
\frac12 (f_N, \mathscr C f_N) - \frac12 (\mathscr{C}^{-1} \p, \p). 
\end{align}
Since the Hausdorff measure $\lambda$ on $\Xcal_N$ is translation invariant, this implies that
\begin{align} \label{eq:translation_gaussian_main}
\int_{\Xcal_N} e^{(f_N, \p)} G(\p) \, \mu^{(\boldsymbol q)}(d\p) 
= e^{\frac12 ( f_N, {\mathscr C}^{(\boldsymbol q)} f_N)} \, \int_{\Xcal_N} G(\p + {\mathscr C}^{(\boldsymbol q)} f_N) \, \mu^{(\boldsymbol q)}(d\p).
\end{align}
Using now, first  \eqref{eq:prefinalidentity} as in  \eqref{finalidentity} and then  \eqref{eq:translation_gaussian_main},
we get
\begin{align} \label{finalidentity3}
\begin{split}
\int_{\Xcal_N} &e^{(f_N,\p)} \sum_{X\subset \TN} \mathcal{K}(X, \p)\,\mu^{(0)}(\d\p)  \\	                                                 
 =&  \, \frac{Z^{(\boldsymbol{q}(\mathcal H))}}{Z^{(0)}}e^{L^{Nd}{e}(\mathcal H)}
\int_{\Xcal_N} e^{(f_N,\p)}                                       
\bigl(\widehat K_0(\mathcal K, \mathcal H) \circ  e^{- \mathcal H} \bigr)(\TN,\p)\,\mu^{(\boldsymbol{q}(\mathcal H))}(\d\p)
\\
=& \,                                                                                                              
e^{\frac{1}{2}(f_N,\mathscr{C}^{(    \boldsymbol{q}(\mathcal H))}f_N)}    \,   
\frac{Z^{(\boldsymbol{q}(\mathcal H))}      }{Z^{(0)}}        e^{L^{Nd}  e(\mathcal H)}  
\\ &\hspace{2.5cm}  \int_{\Xcal_N}
\bigl(\widehat K_0(\mathcal K, \mathcal H) \circ      e^{- \mathcal H}  \bigr)(\TN,\p +\mathscr{C}^{(\boldsymbol{q}(\mathcal H))}f_N)
\,\mu^{(\boldsymbol{q}(\mathcal H))}(\d\p).
\end{split}
\end{align}
Taking as before  $\mathcal H = \widehat{\mathcal H}(\mathcal K)$, using the abbreviation
$H_0 = \Pi_{H_0} (\widehat Z(\mathcal K, \widehat{\mathcal H}(\mathcal K)) = \mathcal H$,
the relations $\boldsymbol q(\mathcal H) = \widehat{\boldsymbol q}(\mathcal K)$ and $e(\mathcal H) = \widehat{e}({ \Kcal})$,
and,  finally, the equality \eqref{eqfirstfixedpointid},  we see that the right hand side of 
 \eqref{finalidentity3} equals
 \begin{align}
e^{\frac{1}{2}(f_N,\mathscr{C}^{(  \widehat{  \boldsymbol{q}}({ \Kcal}))}f_N)}     \frac{Z^{(    \widehat{\boldsymbol{q}}(\mathcal K)    )}}{Z^{(0)}}e^{ L^{Nd}  
 \widehat{e}(\mathcal K)}
 \int_{\Xcal_N} \bigl(1+\widehat K_N(\mathcal K) \bigr)(\TN,\p  + {\mathscr C}^{(\widehat{\boldsymbol q}(\mathcal K))} f_N )  )
 \,\mu_{N+1}^{(   \widehat{\boldsymbol{q}}(\mathcal K))}(\d\p). 
 \end{align}
This concludes the proof of \eqref{thmequality2}  and thus of Theorem~\ref{maintheorem}.
\end{proof}
\newpage

\begin{appendix}
\chapter{Norms on Taylor Polynomials}  \label{se:norms_polynomials}

The following material is essentially contained in \cite{BS15I}. We include it 
for the convenience of the reader because the notation is simpler than in \cite{BS15I} (since we 
do not have to deal with fermions) and because we would like to emphasise that 
the basic results (product property, polynomial property and two-norm estimate) 
follow from general features of tensor products and are not dependent on the special choice 
of the norm  in \eqref{eq:primal_norm_appendix}.

Before we start on the details let us put this appendix more precisely into  context. The uniform smoothness estimates
for the polynomial maps and the exponential map in Chapter~\ref{sec:smoothness} rely heavily
on the submultiplicativity of the norms on the functionals $K(X, \p)$. This submultiplicativity
in turn is based on two ingredients: submultiplicativity of the weights (see Theorem~\ref{th:weights_final}~\ref{w:w3}-\ref{w:w6}
in Chapter~\ref{sec:weights}) and the choice of a submultiplicative norm on Taylor polynomials
which we address in this appendix. 

For smooth functions on $\R^p$ one can easily check that 
a suitable $\ell_1$ type norm on the Taylor coefficients (see  \eqref{eq:T0_norm_for_Rp_new} below) is submultiplicative.  
We deal with smooth maps on the space $\Xcal_N$ of fields and, more importantly, we
want the norm on Taylor polynomials to reflect the typical behaviour of the field on different scales
$k$, i.e., under the measure $\mu_{k+1}^{(\boldsymbol q)}$. In this setting a more systematic approach
to the construction of the norms is useful.

 The main idea is to view 
a homogeneous polynomial of degree $r$  on a finite dimensional space $\BX$
as a linear functional on the tensor product 
$\BX^{\otimes r}$. A norm on $\BX$ induces in a natural way norms on the tensor products
(see Definition  \eqref{eq:T0_norm_for_Rp_new}) and by duality on polynomials (see
 \eqref{eq:norm_g_all_degrees}, \eqref{eq:norm_P_decompose_untruncated} 
 and \eqref{eq:norm_P_decompose}).
 This norm automatically satisfies submultiplicativity (see Propositions~\ref{pr:product_estimate_new}
and~\ref{pr:product_estimate_taylor})
and in addition we get useful properties such as the
polynomial property  in Proposition~\ref{pr:polynomial_prop_gen} and 
the two-norm estimate in Proposition~\ref{pr:two_norm_new}.

\section{Norms on polynomials}
Let $\BX$ be a finite dimensional space vector space. For definiteness we consider only 
vector spaces over $\R$, but the arguments apply also to vector spaces over $\C$. 
The main idea is to linearise the action of polynomials on $\BX$. 
We say that $P : \BX \to \R$ is a polynomial if, in some (and hence in any) 
basis, $P$ is a polynomial in the coordinate with respect to that basis.
For $r$-homogeneous polynomials we can use the following representation
(alternatively this representation can be used as a
coordinate-free  definition of an $r$-homogeneous polynomial).

\begin{lemma}  \label{le:extension_polynomial}
Let $P$ be an $r$-homogeneous polynomial on $\BX$. 
Then there exist a unique symmetric 
element $\overline P$ of the dual space $(\BX^{\otimes r})'$
such that 
$ P(\xi) = \langle \overline P,  \xi \otimes \ldots \otimes \xi \rangle$.
\end{lemma}
Here we write $\langle \cdot, \cdot \rangle$ to denote the dual pairing of
$(\BX^{\otimes r})'$ and $\BX^{\otimes r}$. We say that $g \in \BX^{\otimes r}$ is symmetric 
if $Sg = g$ where the symmetrisation operator $S$  is defined in  \eqref{eq:symmetry_operator}.

\begin{proof}
Existence: define
$\langle \overline P,  \xi_1 \otimes \ldots \otimes \xi_r \rangle  = \frac{1}{r!} \frac{d}{dt_r} \ldots \frac{d}{dt_1}_{| t_i= 0}
P(\xi(t))$ where $\xi(t) = 
\sum_{i=1}^r  t_i \xi_i$
and where the $\xi_i$ run through a basis. Then  extend $\overline P$ by linearity.
Homogeneity implies that $$\overline P(\xi \otimes \ldots \otimes \xi) = \frac{1}{r!} \frac{d}{dt_r} \ldots \frac{d}{dt_1}_{| t_i= 0} (t_1 + \ldots + t_r)^r P(\xi) =P(\xi).$$
 Uniqueness: if $\overline P, \overline Q \in (\BX^{\otimes r})'$
are symmetric and $\langle \overline P - \overline Q, \xi(t) \otimes \ldots \otimes \xi(t) \rangle = 0$ then applying 
$\frac{d}{dt_r} \ldots \frac{d}{dt_1}_{| t_i= 0}$ we deduce that $\overline P - \overline Q=0$.
\end{proof}

We denote by $\bigoplus_{r=0}^\infty \BX^{\otimes r}$ the space of sequences $(g^{(0)}, g^{(1)}, \ldots)$
with $g^{(r)} \in \BX^{\otimes r}$ for which only finitely many of the $g^{(r)}$ are non-zero.
By writing a general polynomial $P$ as a sum of homogeneous polynomials
we can associate to $P$ a linear map on $\bigoplus_{r=0}^\infty \BX^{\otimes r}$
via\footnote{Actually polynomials act even more naturally on the space of symmetric tensor products
$\oplus_{m=0}^\infty \odot_m \BX$, see Chapters 1.9 and 1.10 in \cite{Fed69}, but the easier duality with $\oplus_{r=0}^\infty \BX^{\otimes r}$
is good enough of us.}
\begin{equation} \label{eq:duality_pairing_P_g}  \langle \overline P, g \rangle = \sum_{r=0}^\infty \langle 
  \overline {P^{(r)}}, g^{(r)}  \rangle.
\end{equation}
Here $\BX^{\otimes 0} := \R$ and $P^{(0)}$ is the constant term $P(0)$. 
We will define a norm on $\bigoplus_{r=0}^\infty \BX^{\otimes r}$. This induces a  norm on $P$ by duality.
The point is to define  the  norm on $\bigoplus_{r=0}^\infty \BX^{\otimes r}$ in such a way that the norm on $P$ 
enjoys the product property: $\| PQ \| \le \|P \| \, \|Q\|$.

Here we consider only finite dimensional spaces $\BX_i$. The study   of tensor products
of (infinite dimensional) Banach spaces has been a very active field of research
beginning with Grothendieck's seminal work \cite{Gro53}, see, e.g., \cite{DF93,Flo99,Rys02,DFS08,CG11,Pis12}.

Let $\BX_i$ be finite dimensional normed vector spaces over $\R$ and with dual spaces $\BX'_i$.
We say that an element of $\xi \in \BX_1 \otimes \ldots \otimes \BX_r$ is \textit{simple} if
$$ \xi = \xi_{1} \otimes \ldots \otimes \xi_{r}  \quad \hbox{with $\xi_i \in \BX_i$} \quad 
\hbox{and we define} \quad
 \| \xi_{1} \otimes \ldots \otimes \xi_{r} \| = \| \xi_{1}\|  \ldots \| \xi_{r}\|.$$
Note that by definition of the tensor product every element of 
$ \BX_1 \otimes \ldots \otimes \BX_r $ can be written as a finite combination of simple elements. 
We recall the definition of two standard norms on tensor products.

\begin{definition}  \label{de:projective_injective_norm}
The  projective norm (or largest reasonable norm)  on $ \BX_1 \otimes \ldots \otimes \BX_r $  is given by
$$ \| g \|_\wedge = \inf \left\{ \sum_{n}    \, \| \xi_n \|   : g  = \sum_n  \xi_n \,  \, 
\hbox{with  $\xi_n$ simple} \, \,  \right\}
$$
Here   the infimum  is taken over  finite sums. 
The injective norm (or smallest reasonable norm) is given by 
$$ \| g\|_{\vee} = \sup \left\{ 
\langle \xi'_1 \otimes  \ldots \otimes \xi'_r, g \rangle: \| \xi_i'\|_{\BX'_i} \le 1 \hbox{ for all $i =1, \ldots, r$}   \right \}.
$$
\end{definition}
There is a third important norm based on the Hilbertian structure, but we will not use this here.

One easily sees that 
\begin{equation} \label{eq:vee_less_wedge}
 \| g\|_\vee \le \|g\|_\wedge
\quad \hbox{and} \quad 
 \| \xi_{1} \otimes \ldots \otimes \xi_{r} \|_\vee = \| \xi_{1} \otimes \ldots \otimes \xi_{r} \|_\wedge
=  \| \xi_{1}\|  \ldots \| \xi_{r}\|.
\end{equation} 
Therefore for simple elements we  write $\| g\|$ instead of $\| g\|_\vee$ or $\| g \|_\wedge$. 

\begin{example} \label{ex:Rp_with_l_infty} We show that the injective norm on $(\R^p, |\cdot|_\infty)^{\otimes r}$ is 
the $\ell_\infty$ norm and the projective norm on $(\R^p, |\cdot|_1)^{\otimes r}$ is the $\ell_1$ norm.

Let $e_1, \ldots, e_p$ be the standard basis of $\R^p$. For $\p = \sum_{j=1}^p \p_j e_j$ set
$|\p|_\infty = \max_{1 \le j \le p} |\p_j|$ and consider $\BX = (\R^p, | \cdot|_\infty)$.
Denote the dual basis by $e'_j$, i.e.
$e'_j(\p) = \p_j$. 
Then the dual space consists of functionals of the form $\ell  =\sum_{j=1}^p a_j e'_j$ and the dual norm is
given by $|\ell |_{\BX'} = |a|_1 = \sum_{j=1}^p |a_j|$. Thus $\BX'$ is isometrically isomorphic to
$(\R^p, |\cdot|_1)$.
Let $E = \{1, \ldots, p\}$. Then an element $g \in \BX^{\otimes r}$ can be identified with an element
of $\R^{E^r}$ via $g = \sum_{(j_1, \ldots, j_r) \in E^r} g_{j_1 \ldots j_r} \, e_{j_1} \otimes \ldots \otimes e_{j_r}$.
Similarly $L \in  (\BX')^{\otimes r}$ can be uniquely expressed as
$L = \sum_{(j_1, \ldots, j_r) \in E^r} a_{j_1 \ldots j_r}  \, e'_{j_1} \otimes \ldots \otimes e'_{j_r}$.
We claim that
\begin{align}
\| g\|_\vee  = & \, |g|_\infty := \max_{(j_1, \ldots, j_r) \in E^r} |g_{j_1 \ldots j_r}|, 
       \label{eq:injective_Rp}\\
\|L\|_\wedge = & \, |L|_1 := \sum_{(j_1, \ldots, j_r) \in E^r} |a_{j_1 \ldots j_r}|.
\label{eq:projective_Rp_prime}
\end{align}
Indeed $\| L \|_\wedge \le |L|_1$ since $e'_{j_1} \otimes \ldots \otimes e'_{j_r}$ is simple. 
On the other hand for every simple $L = l_1 \otimes \ldots \otimes l_r$  with $l_i = \sum_{j_i=1}^p a^{(i)}_{j_i} e'_{j_i}$,
since   $L = \sum_{(j_1, \ldots, j_r) \in E^r} \bigl(\prod_{i=1}^r a^{(i)}_{j_i} \bigr) \, e'_{j_1} \otimes \ldots \otimes e'_{j_r}$, 
we have
\begin{align}  \label{eq:est_product_simple_Rp}
|L|_1 =\hspace{-.2cm} \sum_{(j_1, \ldots, j_r) \in E^r} \Bigl|\prod_{i=1}^r  a^{(i)}_{j_i} \Bigr|=
\hspace{-.2cm}\sum_{(j_1, \ldots, j_r) \in E^r} \prod_{i=1}^r  \abs{a^{(i)}_{j_i}}= \prod_{i=1}^r   \sum_{j_i=1}^p |a^{(i)}_{j_i}| =   \prod_{i=1}^r |\ell_i|_{\BX'}= \| L \|_\wedge.
\end{align}
Thus $|L|_1 = \|L\|_\wedge$ for all simple $L$ and by definition of $\| \cdot \|_\wedge$ this implies
$|L|_1 \le\|L \|_\wedge$ for all $L$. 

To prove    \eqref{eq:injective_Rp} we first note that
\begin{align*}
\pm g_{j_1 \ldots j_r} = \langle \pm e'_{j_1} \otimes \ldots \otimes e'_{j_r}, g \rangle \le \| g\|_\vee
\end{align*}
and hence $|g|_\infty \le \|g\|_\vee$. 
To prove the converse inequality we note that  for $l_i \in \BX'$ as above  using $\langle e'_{j},e_k\rangle= \delta_{j,k}$ and thus
\begin{align*}
 \langle \ell_1 \otimes \ldots \otimes \ell_r, g \rangle 
=  \sum_{(j_1, \ldots, j_r) \in E^r} g_{j_1 \ldots j_r} \prod_{i=1}^r  a^{(i)}_{j_i} \, \,  
\le |g|_\infty\prod_{i=1}^r   \sum_{j_i=1}^p |a^{(i)}_{j_i}| = |g|_\infty \prod_{i=1}^r \| \ell_i \|_{\BX'}.
\end{align*}
Thus  $\| g \|_\vee \le |g|_\infty$.
\end{example}
\bigskip

Define  dual norms on $(\otimes_{i=1}^r \BX_i)'$ by
\begin{align}\begin{split}
 \|L\|'_\vee : &= \sup \{ \langle L, g \rangle : g \in \otimes_{i=1}^r \BX_i , \, \| g \|_\vee \le 1\},
\\ 
 \|L\|'_\wedge : &= \sup \{ \langle L, g \rangle : g \in \otimes_{i=1}^r \BX_i, \, \| g \|_\wedge \le 1\}.
\end{split}\end{align}
The dual space $(\otimes_{i=1}^r \BX_i)'$ can be identified with
$(\otimes_{i=1}^r \BX'_i)$. Indeed, let $\xi'_i \in \BX'_i$, let $\xi_i$ run through a basis
of $\BX_i$ and define
$$ 
\iota(\xi'_1 \otimes \ldots \otimes \xi'_r)( \xi_{1} \otimes \ldots \otimes \xi_{r})= \prod_{i=1}^r   \langle \xi'_i, \xi_{i} \rangle.
$$
By linearity $\iota(\xi'_1 \otimes \ldots \otimes \xi'_r)$ can be extended to a linear functional on 
$\otimes_{i=1}^r \BX_i$, i.e., to an element of $(\otimes_{i=1}^r \BX_i)'$. Now let $\xi'_i$ run through a basis of $\BX'_i$.
Then $\iota$ can be extended to a unique  linear map from $(\otimes_{i=1}^r \BX'_i)$ to $(\otimes_{i=1}^r \BX_i)'$
and one easily checks that $\iota$ is injective and hence bijective since both spaces have the same dimension. 
With this identification and using the fact that the closed unit ball in the projective norm is the convex hull
$C = \mathop{\mathrm{conv}}(\{ \xi : \hbox{$\xi$ simple, $\| \xi\| \le 1 $} \})$, the Hahn-Banach separation theorem
and the fact that for finite dimensional spaces $\BX'' = \BX$  one easily verifies that
\begin{equation}  \label{eq:dual_tensor_norms}
\| L \|'_\wedge = \|L\|_\vee 
 \quad \hbox{and}  \quad
  \| L\|'_\vee = \| L \|_\wedge.
\end{equation}
One can also easily check that the projective and the injective norm are associative with respect to iterated tensorisation.

\begin{lemma} 
Assume that $\square = \vee$ or $\square = \wedge$. Then the following properties hold
\begin{enumerate}
\item (Tensorisation estimate) 
For $g \in {\BX}^{\otimes r}$, $h \in {\BX}^{\otimes s}$ and $L \in ({\BX}^{\otimes r})'$,
$M \in (\BX^{\otimes r})'$,
\begin{equation}  \label{eq:tensor_gen_consistency_primal}
 \| g \otimes h \|_\square  \le   \| g \|_\square    \| h \|_\square    \quad  \hbox{and} \quad 
\| L \otimes M \|'_\square \le \| L \|'_{\square} \, \| M \|'_\square.
\end{equation}
\item (Contraction estimate) 
For $L \in (\BX^{\otimes (r+s)})'$ and $h \in \BX^{\otimes s}$ define $M \in (\BX^{\otimes r})' $
 by
 $ \langle M, g \rangle  = \langle L, (g \otimes h)\rangle$.
 Then 
 \begin{equation}  \label{eq:tensor_contraction}
 \| M \|'_\square  \le   \| L \|'_\square    \, \| h \|_{\square}.
\end{equation} 
\end{enumerate}
\end{lemma}

\begin{proof}
To prove the first estimate in  \eqref{eq:tensor_gen_consistency_primal} for $\square = \vee$,
assume that $\| \xi'_i\| \le 1$
for $i \in \{ 1, \ldots, r+s\}$. Then
\begin{align*}
& \, \langle \xi'_1 \otimes \ldots  \otimes \xi'_{r+s} , g \otimes h \rangle
=  \langle \xi'_1 \otimes \ldots \otimes \xi'_r, g \rangle  \, \,  \langle \xi'_{r+1} \otimes \ldots \otimes \xi'_{r+s}, h \rangle
\le \|g\|_\vee     \, \, \|h\|_\vee.
\end{align*}
Next we
 consider $\square = \wedge$. 
For each $\delta >0$ there  exist $g_i, h_k$ simple such that
$$ \sum_i \   \|g_i \| \le (1+ \delta) \|g\|_\wedge, \quad 
 \sum_k     \|h_k \| \le (1+ \delta) \|h\|_\wedge.$$
 Now $g_i \otimes h_k$ is simple and thus
 $ \| g_i \otimes h_k\|_\wedge =  \| g_i \|  \, \| h_k \|$.
The assertion follows from the triangle inequality and  fact that
 $$ \sum_i \sum_k   \| g_i\| \, \| h_k \| =  \sum_i \| g_i\|   \, \, 
  \sum_k   \| h_k\|   \le (1+\delta)^2 \|g\|_\vee \, \| h\|_\vee.$$
 
 The second estimate in  \eqref{eq:tensor_gen_consistency_primal} follows from the first (applied to 
 $\BX'$ instead of $\BX$) and   \eqref{eq:dual_tensor_norms}.
Finally \eqref{eq:tensor_contraction} follows from  \eqref{eq:tensor_gen_consistency_primal}
and the definition of the dual norm.
\end{proof}

On $\BX^{\otimes r}$ we define the symmetrisation operator by
\begin{align}  \label{eq:symmetry_operator}
 S(\xi_1  \otimes \ldots \otimes   \xi_r) = \frac{1}{r!} \sum_{\pi \in S_r}  \xi_{\pi(1)} \otimes \ldots \otimes \xi_{\pi(r)},
 \end{align}
where the sum runs over all permutation of the set $\{1, \ldots, r\}$,
and extension by linearity.
Similarly we can define $S$ on $\BXp^{\otimes r} = (\BX^{\otimes r})'$.
Then
\begin{equation} \label{eq:dual_symmetrisation}
\langle S L, g \rangle = \langle L,  Sg \rangle. 
\end{equation}
Indeed the identity holds if $g$ is simple and hence by linearity for all $g$.

\begin{lemma} \label{le:symmetry_estimate}
For $\square = \vee$ or $\square = \wedge$ we have
\begin{equation}  \label{eq:symmetry_estimate_primal}
 \| Sg \|_\square \le \|g \|_\square \quad \forall g \in \BX^{\otimes r}
 \quad \hbox{and}
\quad 
  \| SL \|'_\square \le \|L \|'_\square   \quad \forall L \in \BX^{\otimes r}.
  \end{equation}
\end{lemma}

\begin{proof}
The second assertion follows from the first and \eqref{eq:dual_symmetrisation}.
To prove the first assertion  for $\square = \wedge$ it suffices to note that $S$ maps
a simple element of norm $1$  to a convex combination of simple elements
of norm $1$. 
For $\square = \vee$ we use \eqref{eq:dual_symmetrisation} to get
$ \langle \xi'_1 \otimes \ldots \xi'_r,  Sg\rangle  = \langle S (\xi'_1 \otimes \ldots \xi'_r),  g\rangle.$
Now we use again that $S$ maps
a simple element of norm $1$  to a convex combination of simple elements
of norm $1$.  
\end{proof}

We now define a norm on $\oplus_{r=0}^\infty \BX^{\otimes r}$ by
\begin{equation}    \label{eq:norm_Phi_untruncated}
\| g \|_{\BX, \square} := \sup_{r }  \| g^{(r)} \|_{\BX, \square}
\end{equation}
Here $\| g^{(0)}\| = |g^{(0)}|_\R$ where $| \cdot |_\R$ is the absolute value on $\R$. 
 For a polynomial  $P = \sum_r P^{(r)}$ written as a sum of homogeneous polynomials of degree $r$, 
the norm is defined by 
\begin{equation}  \label{eq:norm_g_all_degrees}
\| P \|'_{\BX, \square} = \sup \{ \langle \overline P, g \rangle  : \| g \|_{\BX, \square}  \le 1 \}
\end{equation}
 where $\langle \overline P, g \rangle$ was defined in \eqref{eq:duality_pairing_P_g}.
We have
\begin{equation}  \label{eq:norm_P_decompose_untruncated}
\| P \|'_{\BX, \square} = \sum_{r=0}^\infty   \| P^{(r)}\|'_{\BX, \square} = \sum_{r_0}^\infty  \| \overline{P^{(r)}}\|'_\square.
\end{equation}
Similarly we can define a seminorm by considering only test functions $g$ in the space
\begin{equation} 
 \Phi_\pn := \{ g \in \oplus_{r=0}^\infty \BX^{\otimes r} : g^{(r)} = 0 \quad \forall r > \pn \}.
\end{equation}
Then
\begin{equation}  \label{eq:norm_P_decompose}
\| P \|'_{\pn, \BX, \square} := \sup \{ \langle P, g \rangle : g \in \Phi_\pn, \, \| g \| \le 1\} = \sum_{r=0}^\pn   \| P^{(r)}\|'_{\BX, \square}.
\end{equation}
This defines is a seminorm on the space of all polynomials and a norm on  polynomials of degree $\le \pn$. 
When $\pn$ and $\BX$ and $\square$ are clear we simply write
$ \| P \| = \| P \|'_{\pn, \BX, \square}$.

\begin{proposition}[Product property]  \label{pr:product_estimate_new}  
Let $\pn \in \NN_0 \cup \{\infty\}$.
Assume that $\square \in \{ \vee, \wedge\}$.
 Let $P$ and $Q$
be polynomials on $X$. Then
\begin{equation}
\| P Q\|  \le \|P \|  \, \| Q\|.
\end{equation}
\end{proposition}

\begin{proof} We first show the assertion for an $r$-homogeneous polynomial $P$ and 
a $(s-r)$-homogeneous polynomial $Q$. If $r=0$ or $s-r=0$ the assertion is clear. 
We hence assume $r \ge 1$ and $s -r \ge 1$. 
We first note that 
 $\overline {PQ} = S (\overline  P  \otimes  \overline Q)$
where $S$ is the symmetrisation operator introduced above. 
Indeed both sides are symmetric elements of $\BX^{\otimes k}$ and they
agree on $\xi \otimes \ldots \otimes \xi$. 
Thus the desired identity  follows from the uniqueness statement in Lemma~\ref{le:extension_polynomial}.
Now  it follows from  the second estimate  in \eqref{eq:tensor_gen_consistency_primal}
and   \eqref{eq:symmetry_estimate_primal} 
that 
$
\| PQ \|  \le \| \overline P \|'_\square \, \| \overline Q\|'_\square = \| P \|  \, \| Q\| $.
This finishes the proof for homogeneous polynomials. 

 Finally consider general $P, Q$ and their decompositions into homogeneous
polynomials $P = \sum_r P^{(r)}$, $Q = \sum_s Q^{(s)}$.
Then it follows from \eqref{eq:norm_P_decompose}
 and the triangle inequality 
that 
$$ \| PQ \|  \le \sum_{s=0}^{\pn}  \sum_{r=0}^s  \| P^{(r)} Q^{(s-r)} \|  \le 
\sum_{s=0}^{\pn}  \| P^{(r)} \|   \|Q^{(s-r)} \|
\le \|P \|   \, \|Q\|.$$
\end{proof}

\section{Norms on polynomials in several variables}  \label{se:polynomials_several_variables}
The product property for polynomials can be easily extended to 
polynomials in several variables. To simplify the notation we illustrate this
for the case of two variables. A polynomial $P(x,y)$ on $\BX \times \BY$
which is $r$-homogeneous in $x$ and $s$-homogeneous in $y$
can be identified with an element $\overline P$ of $\BX^{ \otimes r} \otimes \BY^{\otimes s}$
which is symmetric in the sense that
$$ \overline P(\xi_{\pi(1)} \otimes \ldots \otimes \xi_{\pi(r)} \otimes \eta_{\pi'(1)} \otimes \ldots \otimes \eta_{\pi'(s)})
= \overline P(\xi_{1} \otimes \ldots \otimes \xi_{r} \otimes \eta_{1} \otimes \ldots \otimes \eta_{s})
$$
for all permutations $\pi$ and $\pi'$. 
We define a space of test functions 
$$\Phi_{\pn, s_0} := \{ 
 g \in \oplus_{r,s \in \NN_0} \BX^{\otimes r} \otimes \BY^{\otimes s} : g^{(r,s)}= 0 \quad \hbox{if $r > \pn$ or $s > s_0$} \}$$
 with the norm
 $$ \|g \|_\square := \sup_{r,s \in \NN}   \| g^{(r,s)} \|_{\BX, \BY, \square}.$$
 Decomposing a general polynomial in homogeneous pieces $P^{(r,s)}$ we define the pairing
 $$ \langle P, g \rangle = \sum_{r,s \in \NN_0}  \langle \overline{P^{(r,s)}}, g^{(r,s)}  \rangle$$
 and the dual norm
 $$ \|P\|'_\square = \|P \|'_{\pn, s_0, \BX, \BY, \square}  =
 \sup \{ \langle P, g \rangle : g \in \Phi_{\pn, s_0}, \, \| g\|_{ \square} \}.$$
 Then 
 $$ \| P \|'_\square = \sum_{r=0}^\pn \sum_{s=0}^{s_0}  \| P^{(r,s)} \|'_\square =  \sum_{r=0}^\pn \sum_{s=0}^{s_0}  \|\overline{ P^{(r,s)}} \|'_\square$$
 where $P^{(r,s)}$ are the $(r,s)$-homogeneous pieces of $P$. 
 
 For $M \in (\BX^{\otimes r_1} \otimes \BY^{\otimes s_1})'$ and $L \in (\BX^{\otimes r_2} \otimes \BY^{\otimes s_2})'$
 we define the tensor product $M \otimes L$ in $(\BX^{\otimes r_1 + r_2} \otimes \BY^{\otimes s_1 + s_2})'$
 by
 \begin{align*} & \,  \langle M \otimes L, \xi_1 \otimes \ldots \otimes \xi_{r_1+r_2} \otimes \eta_1 \otimes \ldots \otimes \eta_{s_1 + s_2} \rangle\\
 = & \, 
\langle M,  \xi_1 \otimes \ldots \otimes \xi_{r_1} \otimes \eta_1 \otimes \ldots \otimes \eta_{s_1} \rangle \, \, 
\langle L,  \xi_{r_1+1} \otimes \ldots \otimes \xi_{r_1+r_2} \otimes \eta_{s_1 +1} \otimes \ldots \otimes \eta_{s_1+s_2} \rangle.
\end{align*}
Then the same argument as before shows that 
\begin{align*} \| M \otimes L \|_\square \le \|M \|_\square \, \| L \|_\square \quad \hbox{for $\square \in \{ \vee, \wedge\}$}.
\end{align*}
We also define a symmetrisation operator $S_{\BX, \BY}$  which symmetrises separately in the variables on $\BX$ and the ones in $\BY$,
i.e., 
\begin{align*} & \, 
S (\xi_{1} \otimes \ldots \otimes \xi_{r} \otimes \eta_{1} \otimes \ldots \otimes \eta_{s})  
\\ & \hspace{3cm} :=  \frac1{r!} \frac1{s!} \sum_\pi \sum_{\pi'}
(\xi_{\pi(1)} \otimes \ldots \otimes \xi_{\pi(r)} \otimes \eta_{\pi'(1)} \otimes \ldots \otimes \eta_{\pi'(s)}).
\end{align*}
Again it is easy to see that $S$ has norm $1$. 
Thus for two homogeneous polynomials $P$ and $Q$ one sees as before
$$ \| P Q\| = \| S (\overline P \otimes \overline Q)\| \le \| \overline P \| \, \| \overline Q\| = \| P \| \, \| Q\|.$$
Now the product property for polynomials is obtained as before  by decomposing $P$ and $Q$ in $(r,s)$-homogeneous polynomials.

\section{Norms on Taylor polynomials}
\begin{definition} Let $p_0 \in \NN_0$, let $U \subset \BX$ be open and let $F \in C^\pn(U)$. 
For $\p \in U$ denote the Taylor polynomial of $F$ at $\p$ by $\tay_\p F$
and define
\begin{equation} \label{eq:norm_T_phi}
\| F\|_{T_\p} = \| \hbox{$\tay_\p$}  F \|'_{\pn,\BX,  \square}.
\end{equation}
where $\square$ refers to the norm used for the tensor products.
\end{definition}
When the norm on the tensor products is clear we often drop $\square$.

\begin{example}\label{Ex:T0norm}  Let $\BX = (\R^p, |\cdot|_\infty)$ and set
 $E = \{ 1, \ldots,  p\}$. 
 In   \eqref{eq:injective_Rp} we have seen that the injective norm of $g \in \BX^{\otimes r}$ is given by 
  $\| g \|_\vee = \max_{(j_1, \ldots j_r) \in E^r} |g_{j_1 \ldots j_r}| = | g |_\infty$. 
Let $F \in C^\pn(\BX)$. The Taylor polynomial of order $\pn$ at zero can be written as
$$ P(\p) = \sum_{r=0}^{r_0} \frac{1}{r!}   \sum_{j_1, \ldots, j_r=1}^{p}  \frac{\partial^r F}{\partial \p_{j_1} \ldots \partial \p_{j_r}}(0)  \,  \prod_{i=1}^r  \p_{j_i}
= \sum_{|\gamma|_1 \le \pn}   \frac{1}{\gamma!} \partial^\gamma F(0) \, \,   \p^\gamma$$
where the  last sum  runs over  multiindices $\gamma \in \NN_0^E$ and $|\gamma| := \sum_{j \in E} \gamma(j)$. 
The term corresponding to  $r=0$ is defined as $F(0)$.
We claim that
\begin{equation}  \label{eq:T0_norm_for_Rp_new}
\| F \|_{T_0} = \sum_{r=0}^{\pn}   \frac{1}{r!}   \sum_{j_1, \ldots, j_r=1}^{ p} \left| \frac{\partial^r F}{\partial \p_{j_1} \ldots \partial \p_{j_r}  }(0) \right|
=  \sum_{|\gamma| \le \pn}   \frac{1}{\gamma!} \left| \partial^\gamma F(0) \right|.
\end{equation}
 Indeed it suffices to verify the first identity, the second follows by the usual combinatorics. Denote the middle term in  \eqref{eq:T0_norm_for_Rp_new}
 by $M$. 
Since we use the $\ell_\infty$ norm on $\BX^{\otimes r} = \R^{E^r}$ we get for all $g \in \Phi_\pn$
$$ \langle F, g \rangle_0 =  \sum_{r=0}^{\pn} \frac{1}{r!}   \sum_{j_1, \ldots, j_r=1}^{\p}  \frac{\partial^r F}{\partial \p_{j_1} \ldots \partial \p_{j_r}}(0) \, 
g_{j_1 \ldots j_r} \le M   \sup_{0 \le r \le \pn} | g^{(r)}|_{\infty}  \le M \| g\|_{\BX, \vee}$$
The inequality becomes sharp if we take $g_{j_1  \ldots j_r} = \sgn  \frac{\partial^r F}{\partial \p_{j_1} \ldots \partial \p_{j_r}}(0)$.
This proves  \eqref{eq:T0_norm_for_Rp_new}.
\end{example}

\begin{proposition}[Product property, see \cite{BS15I}, Proposition 3.7]  \label{pr:product_estimate_taylor} 
Let  $U \subset \BX$ be open and let $F \in C^\pn(U)$. Then
$$ \| F G \|_{T_\p} \le \| F \|_{T_\p} \, \| G \|_{T_\p}.$$
\end{proposition}

\begin{proof} This follows from Proposition~\ref{pr:product_estimate_new}
and  the fact that the Taylor polynomial of the product is the product of the Taylor polynomials, truncated at  degree $\pn$. 
\end{proof}

 By the considerations in Section~\ref{se:polynomials_several_variables} the product property also holds for 
 polynomials in several variables.

\begin{proposition}[Polynomial estimate, see \cite{BS15I}, Proposition 3.10] \label{pr:polynomial_prop_gen} 
Assume that $\square \in \{ \vee, \wedge\}$. Let $F$ be a polynomial of degree $\overline{r} \le \pn$. 
Then
\begin{equation} \label{eq:polynomial_estimate}
\| F \|_{T_\p} \le (1 + \| \p\|)^{\overline{r}} \| F \|_{T_0}.
\end{equation}
\end{proposition}

\begin{proof} Let $F$ be a polynomial of degree $\overline{r}$ with homogeneous pieces $F^{(r)}$. Then
$ F(\p) = \sum_{r=0}^{\overline{r}}  \langle \overline{F^{(r)}}, \p \otimes \ldots \otimes  \p \rangle$.
Set
$ G( \xi) = F(\p + \xi)$. 
For $r > s$ define  
$B_{r}^{(s)} \in (\BX^{\otimes s})'$  by  
 $$\langle B_{r}^{(s)}, g \rangle = \langle \overline{F^{(r)}},  g \otimes 
 \p \otimes  \ldots  \otimes \p \rangle  \quad \text{for all $g \in \BX^{\otimes s}$.}
 $$
  Set $B_{s}^{(s)} = \overline{F^{(s)}}$. 
 Since the $\overline{F^{(r)}}$ are symmetric we get  
 $$ G(\xi) = \sum_{s=0}^{\overline{r}}   \langle B^{(s)}, \xi  \otimes \ldots \otimes  \xi \rangle \quad \hbox{where $B^{(s)} =
   \sum_{r=s}^{\overline{r}}      \binom{r}{s}  B_r^{(s)}$.}
 $$
 Now  by the contraction estimate   \eqref{eq:tensor_contraction} we have
 $ \| B_{r}^{(s)} \|'_\square  \le    \|\overline{F^{(r)}}\|'_\square   \, \| \p\|_{\BX}^{r-s}$.
Thus
 \begin{align*}
  \|G \|_{T_0} 
 & \le  \,  \sum_{s=0}^{\overline{r}} \sum_{r=0}^{\overline{r}} 1_{r \ge s} \binom{r}{s} \,  
  \|\overline{F^{(r)}}|'_{\square} \, \| \p\|_{\BX}^{r-s} \,  \, 1^s
  \\ &
 \le  \sum_{r=0}^{\overline{r}} (1 + \| \p\|_{\BX})^r \,   \|\overline{F^{(r)}}\|'_{\square} \leq (1 + \| \p\|_{\BX})^{\overline{r}} \,  \|F\|_{T_0}.
 \end{align*}
 Since  $ \| F\|_{T_\p} =  \|G \|_{T_0}$ this concludes the proof.
  \end{proof}

\begin{proposition}[Two norm  estimate, see \cite{BS15I}, Proposition 3.11] \label{pr:two_norm_new} 
Let $F \in C^\pn(\BX)$. 
Assume that $\square \in \{ \vee, \wedge\}$.
 Let  $\| \cdot \|_{\BX, \square}$ and $\| \cdot \|_{\BXt, \square}$ 
 denote norms on the tensor products
 $ \BX^{\otimes r}$  based on norms $\| \cdot \|_\BX$ and $\| \cdot \|_{\BXt}$. 
 Denote the corresponding norms of the Taylor polynomials of $F$ by 
 $ \| F\|_{T_\p}$ and $\| F\|_{\tilde T_\p}$.
 Define
 \begin{equation}  \label{eq:rho_n}
  \rho(\overline{r}) := 2 \sup \{  \| g  \|_{\BX, \square} : 
  g \in \BX^{\otimes r}, \,   \| g  \|_{\BXt, \square} \le 1, \, \overline{r}+1 \le r \le \pn \}.
 \end{equation}
 Then, for any ${\overline{r}} < \pn$,
 \begin{equation} \label{eq:two_norm_new}
\| F \|_{\tilde T_\p} \le (1 + \| \p \|_{\BXt})^{{\overline{r}}+1} \, \bigl( \|F\|_{\tilde T_0} + \rho(\overline{r})  \sup_{0 \le t \le 1}  \|F\|_{T_{t \p} }\bigr).
\end{equation}
\end{proposition}

\begin{proof} 
Let $P$ denote the Taylor polynomial of order ${\overline{r}}$ of  $F$ computed at $0$. 
By Proposition~\ref{pr:polynomial_prop_gen}  and the trivial estimate $\|P\|_{\tilde T_0}
\le \|F\|_{\tilde T_0}$ we have 
$$ \| P\|_{\tilde T_\p}   \le (1 + \| \p \|_{\BXt})^{{\overline{r}}} \,     \|P\|_{\tilde T_0}  
  \le (1 + \| \p \|_{\BXt})^{{\overline{r}}+1} \,     \|F\|_{\tilde T_0}.$$
Let $R = F- P$. It thus suffices to show that 
\begin{equation}  \label{eq:two_norm_r_new}
 \| R\|_{\tilde T_\p}   \le (1 + \| \p \|_{\BXt})^{{\overline{r}}+1} \, \rho(\overline{r})  \sup_{0 \le t \le 1}  \|F\|_{T_{t \p}}. 
\end{equation}
To abbreviate, set 
$$ M := \sup_{0 \le t \le 1} \| F\|_{T_{t \p}} = \sup_{0 \le t \le 1}  \sum_{s=0}^\pn    \frac{1}{s!} \| D^s F(t \p)\|'_{\BX, \square}.$$
Here we view $D^s F(\p)$ as an element of $(\BX^{\otimes s})'$. 
For $s \ge {\overline{r}}+1$ we have $D^s R = D^s F$ and
\begin{align*}
\langle D^s R(\p), g \rangle  = &  \, \langle D^s F(\p), g \rangle \le
\| D^s F\|'_{\BX, \square}  \, \, \| g\|_{\BX, \square} 
\le  \|D^s F\|'_{\BX, \square}\, \,  \frac12 \,   \rho(\overline{r})  \| g \|_{\BXt, \square} 
\end{align*}
and thus
\begin{equation}  \label{eq:two_norm_high_derivatives_new}
 \sum_{s={\overline{r}}+1}^\pn  \frac{1}{s!}  \| D^sR(\p) \|'_{ \BXt, \square} \le \frac12  \, \rho(\overline{r}) M.
 \end{equation}
For $s \le {\overline{r}}$ we apply the Taylor formula with remainder term in integral form to 
$\langle D^s R, g \rangle$
and get 
\begin{align}  \label{eq:two_norm_low_derivatives}
& \, | \langle D^s R(\p), g \rangle|\\
 \le & \,  \notag
\int_0^1   \frac{1}{({\overline{r}}-s)!} (1-t)^{{\overline{r}}-s} \,  | \langle D^{{\overline{r}}+1} F(t \p), g \otimes  \p \otimes  \ldots \otimes  \p \rangle| \, dt  \\
\le & \,   \notag
M \frac{({\overline{r}}+1)!}{({\overline{r}}+1-s)!}  \, \, 
\| g \otimes  \p \otimes  \ldots \otimes  \p \|'_{\BX, \square}\\
\le & \, \notag
\frac 12   \rho(\overline{r})   M    \frac{({\overline{r}}+1)!}{({\overline{r}}+1-s)!} \,\, 
\| g \otimes  \p \otimes  \ldots \otimes  \p \|'_{\BXt, \square}\\
\le & \, \notag
\frac 12   \rho(\overline{r})   M    \frac{({\overline{r}}+1)!}{({\overline{r}}+1-s)!} \, \,   \|g\|'_{\BXt, \square} \, \| \p\|_{\BXt}^{{\overline{r}}+1-s}.
\end{align}
Thus
\begin{align*}
\frac{1}{s!} \|D^sR(\p)\|'_{\BXt, \square} \le 
\frac 12    \rho(\overline{r})   M    \binom{{\overline{r}}+1}{s}   \| \p \|_{\BXt}^{{\overline{r}}+1-s} \, \, 1^s.
\end{align*}
Summing this from $s=0$ to ${\overline{r}}$ we get
$$ \sum_{s=0}^{\overline{r}} \frac{1}{s!} \|D^s R(\p)\|'_{ \BXt, \square} \le \frac 12    \rho(\overline{r})   M  (1 +   \| \p \|_{\BXt})^{{\overline{r}}+1}.$$ 
Together with \eqref{eq:two_norm_high_derivatives_new} this concludes the proof of 
 \eqref{eq:two_norm_r_new}
\end{proof}

\section{Examples with a more general  injective norm on \texorpdfstring{$\BX^{\otimes \protect\headingr}$}{Br}}  \label{se:standard_example}
We will be mostly interested in the case that the norm on $\BX$ is defined by a specific family of linear functionals
on $\BX$ (abstractly one can always define the norm in this way since for finite dimensional space  $\BX'' = \BX$).
Then the injective norm on $\BX^{\otimes r}$ is defined by the tensor products of these functionals
(see Proposition \ref{pr:concrete_tensor_norms} below).

Let $E$ be a finite set. On $\R^E$ consider 
a finite family $\myB$ of linear functionals $\ell : \R^{E} \to \R$. 
Let 
$$ N_\myB := \{ \p  \in \R^{E} : \ell(\p) = 0 \, \, \forall \ell \in \myB \}.$$
Then the linear functionals induce a norm on 
$ \BX := \R^{E} / N_{\myB}$,
namely
$$ \| \p \|_\BX := \sup \{ |\ell(\p)| : \ell \in \myB \}.$$

\begin{proposition}
The dual space of $\BX$ is given by
$ \BXp := \Span \{ \ell : \ell \in \myB\}$
and the norm on $\BXp$ is given by
\begin{equation}  \label{eq:concrete_dual_norm}
 \| \xi' \|_{\BXp} = \inf \big\{  \sum_n |\lambda_n|  :
  \xi' = \sum_n \lambda_n  \ell_n , \, \, \ell_n \in \myB, \, \lambda_n \in \R \big\}.
  \end{equation}
  In particular $\| \ell \|_{\BX'} \le 1$ for all $\ell \in \myB$. 
\end{proposition}

\begin{proof} Let $C$ denote the closed convex hull of $ \myB \cup -\myB$.
It follows from the definition of norm on $\BX$ that
$ C \subset \overline{ B_1(\BXp)}$.
For the reverse inclusion one uses that points  $\xi' \notin C$ can be separated
by a linear functional, i.e., that there exist a $g \in \BX$ such that
$ \xi'(g) > 1$ and $\tilde\xi'(g) \le 1  \quad \forall \tilde\xi' \in C$. This implies $\| g \| \le 1$ and hence
$\| \xi' \| > 1$. 
\end{proof}

\begin{proposition} \label{pr:concrete_tensor_norms}
The injective norm on $\BX^{\otimes r}$ can be characterized by
\begin{equation}  \label{eq:concrete_vee_norm}
 \| g \|_\vee = \sup \{| \langle \ell_1 \otimes \ldots \otimes \ell_r, g \rangle | : \ell_i \in \myB\}.
 \end{equation}
\end{proposition}
Note that in the special case $E = \{1, \ldots, p\}$ and $\myB = \{ e'_1, \ldots, e'_p\}$ we recover
    \eqref{eq:injective_Rp}.

\begin{proof} Denote the right hand side of \eqref{eq:concrete_vee_norm}  by $m$. 
 Since   $\| \ell \|_{\BXp} \le 1$ for all $\ell \in \myB$ we get  $m \le \| g\|_\vee$.
  To prove the reverse inequality, let $\delta > 0$ and assume that  $\| \xi'_k\| \le 1$. 
 Then by \eqref{eq:concrete_dual_norm}
 there exist $\lambda_{k,n} \in \R$ and $\ell_{k,n} \in \myB$ such that
 $\xi'_k = \sum_n \lambda_{k,n} \ell_{k,n}$ and $\sum_j |\lambda_{k,n}| \le 1+ \delta.$
 Thus 
 $ |\langle \xi'_1 \otimes \ldots \otimes \xi'_r, g\rangle| \le (1+\delta)^r m$
 and hence $\| g \|_\vee \le (1+\delta)^r m$. Since $\delta > 0$ was arbitrary
 we conclude that $\| g\|_\vee \le  m$.  
\end{proof}

\section{Main example}

We now come to our main example. Consider  the torus  $\Lambda = \Z^d/ L^N \Z^d$ and  set $\boldsymbol \Lambda = \{ 1, \ldots, m\}\times \Lambda$.
The elements of $\R^{\boldsymbol \Lambda} =
 \R^m  \otimes \R^\Lambda$ can be viewed as maps from $\boldsymbol \Lambda$ to $\R$
or as maps from $\Lambda \to \R^{m}$. We will use both viewpoints interchangeably.

We are interested in linear functionals $\R^{\boldsymbol \Lambda}$  which are based  on discrete derivatives.
More precisely let $e_1, \ldots, e_d$ denote the standard unit vectors in $\Z^d$ and set
\begin{equation} \label{eq:set_mcU+}
\mathcal U = \{ e_1, \ldots, e_d\}.
\end{equation}
We remark in passing that here our notation differs from  \cite{BS15II}.
There $\mathcal U$ denotes the set $\{ \pm e_1, \ldots, \pm e_d\}$.
For $e \in \mathcal U$ and $f :\Lambda \to \R$ the forward difference operator is given by
\begin{equation} \label{eq:difference_operator_loc}
\nabla^e f(x) = f(x+e) - f(x).
\end{equation}
For a multiindex $\alpha \in \NN_0^{\mathcal U}$ we write
\begin{equation}
\nabla^\alpha = \prod_{e \in \mathcal U} (\nabla^e)^{\alpha(e)}, \qquad \nabla^0 = \Id.
\end{equation}
For a pair $(i, \alpha) \in \{1, \ldots, m\} \times \NN_0^{\mathcal U} $ and $x \in \Lambda$ we define
\begin{equation}
\nabla^{i,\alpha}_x \p = \nabla^\alpha \p_i(x). 
\end{equation}
We set $N_\Lambda = \{ \p: \Lambda \to \R^p : \hbox{$\p$ constant}  \}$.
Given weights $w(i, \alpha)  > 0$ we define a norm on $\BX = \R^{\boldsymbol \Lambda}/ N_\Lambda$ by
\begin{equation}  \label{eq:primal_norm_appendix}
 \| \p \|_{\BX} = \sup_{x \in \Lambda} \sup_{1 \le i \le m} \sup_{1 \le |\alpha| \le \pphi}  w(i, \alpha)^{-1}  \nabla_x^{i, \alpha} \p.
\end{equation}
Here and in the following we always use the $\ell_1$ norm for multiindices
\begin{equation}
|\alpha| = |\alpha|_1 = \sum_{ i \in \Ucal} \alpha_i.
\end{equation}
On the scale $k$ we will usually use the weight
\begin{equation}  \label{eq:weight_appendix}
w_k(i, \alpha) = L^{-k |\alpha|} \hpzc_k,  \quad \hpzc_k = h_k L^{-k \frac{d-2}{2}}, \quad h_k = 2^k h.
\end{equation}

Note that for an element $\p \in \BX$ we cannot define a pointwise value $\p(x)$ but the 
derivative $\nabla^\alpha \p(x)$ are well defined if $\alpha \ne 0$. Indeed $\p$ is uniquely determined
by the derivatives with $|\alpha|_1 = 1$. 
We can choose a unique representative  $\widetilde \p$ of the equivalence class $\p + N$ 
by requiring $\sum_{x \in \Lambda} \widetilde \p(x)= 0$ and we sometimes identify
the space $\BX = \R^{\boldsymbol \Lambda}/ N$ with the space $\{ \psi  \in \R^{\boldsymbol \Lambda} : \sum_{x \in \Lambda} \psi(x) = 0\}$.

The tensor product $\BX \otimes \BX$ is the quotient of $\R^{\boldsymbol \Lambda} \otimes \R^{\boldsymbol \Lambda}$
by the vector space $ \Span \{ \hbox{constants} \otimes \p, \p \otimes \hbox{constants} : \p \in \R^{\boldsymbol \Lambda} \}$. Again an element $g^{(2)} \in 
\BX \otimes \BX$ does not have pointwise values $g_{ij}(x,y)$ but the derivatives $\nabla^{i,\alpha}  \otimes \nabla^{j, \beta} g^{(2)}(x,y)
=\nabla^\alpha \otimes \nabla^\beta g_{ij}(x,y)$
are well defined (for $\alpha \ne 0$ and $\beta \ne 0$) and the derivatives with $|\alpha|_1 = |\beta|_1=1$ determine $g^{(2)}$  uniquely. 
Here $\nabla^{i,\alpha}$ acts on the first argument of $g^{(2)}$ and $\nabla^{j, \beta}$ on the second argument. 
Similar reasoning applies to $\BX^{\otimes r}$ and by Proposition \ref{pr:concrete_tensor_norms} the injective norm on $\BX^{\otimes r}$ is given by
\begin{equation}   \label{eq:tensor_norm_appendix}
\| g^{(r)} \|_{\BX,\vee} = \sup_{x_1,  \ldots, x_r \in \Lambda}  \, \sup_{\mpzc \in \mf m_{\pphi,r}} \,  w(\mpzc)^{-1} \,  
\nabla^{\mpzc_1} \otimes \ldots \otimes \nabla^{\mpzc_r} g^{(r)}(x_1, \ldots, x_r).
\end{equation}
Here
\begin{equation}
w(\mpzc) = \prod_{j=1}^r w(\mpzc_j)
\end{equation}
and  $\mf m_{\pphi, r}$ is the set of $r$-tuples $\mpzc =(\mpzc_1, \ldots, \mpzc_r)$ with
$\mpzc_j = (i_j,  \alpha_j)$ and $1 \le |\alpha_j| \le \pphi$. 
Note that here each $\alpha_j$ is a multiindex, i.e., an element of $\mathbb{N}_0^{\Ucal}$, not a number.
For $\mpzc \in \mf m_{\pphi,r}$ consider the monomial
\begin{equation}  \label{eq:appendix_monomial}
\Mscr_\mpzc(\{x\})(\p) := \prod_{j=1}^r  \nabla^{\mpzc_j} \p(x).
\end{equation}
Then the element 
$\overline{\Mscr_\mpzc(\{x\}) }
\in (\BX^{\otimes r})'$ 
which corresponds to 
$\Mscr_\mpzc(\{x\})$ is
given by the symmetrisation $S (\nabla_x^{\mpzc_1} \otimes \ldots \otimes \nabla_x^{m_r})$. Thus
in view of    \eqref{eq:tensor_norm_appendix} and  \eqref{eq:symmetry_estimate_primal} we get
\begin{equation} \label{eq:estimate_monomial_appendix}
\| \Mscr_\mpzc(\{x \} )  \|_{T_0} = 
\| \overline {\Mscr_\mpzc(\{x\}) }\|'_{\BX, \vee} \le w(\mpzc).
\end{equation}

We consider functionals $F$  localised near a polymer $X \subset \Lambda$, i.e. 
$F(\p) = F(\psi)$ if $\p= \psi$ in $X^\ast$ where $X^\ast$ is the small set neighbourhood of $X$,
see \eqref{eq:nghbhdscompact}. 
Thus it is natural to work with field norms which are also localised.
There are different ways to do that.
We follow the approach in \cite{AKM16} and define
\begin{equation}  \label{eq:appendix_A_localised_norm}
\| \p \|_{\BX, X} := \sup_{x \in X^\ast} \sup_{1 \le i \le p} \sup_{1 \le |\alpha| \le \pphi} w(i, \alpha)^{-1} \,  |\nabla^{i, \alpha} \p(x)|,
\end{equation}
 see  \eqref{eq:primal_norm}.
Brydges and Slade 
take a more abstract approach and define 
$$ \| \p\|^\sim_{\BX, X} = \inf \{ \| \p - \xi\|_{\BX}  :  \xi_{ |\{1, \ldots, p\} \times X} = 0 \},$$
see eqns. (3.37)--(3.39) in   \cite{BS15I}.

The two approaches are very closely related. Indeed,
if the weights on scale $k$ satisfy $w(i,\alpha) = L^{-k|\alpha|} h_k(i)$ then one can use a cut-off argument
and the discrete Leibniz rule  to show that
$$  \| \p\|^\sim_{\BX, X}  \le
 C  \sup_{x \in X^+} \sup_{1 \le i \le p} \sup_{1 \le |\alpha| \le \pphi} w(i, \alpha)^{-1} \,  |\nabla^{i, \alpha} \p(x)|,$$
 see Lemma 3.3. in \cite{BS15II} for a similar result.

Conversely it follows directly from the definitions that for any set $X^\square$
$$ \| \p\|_{\BX, X} \le \| \p\|^\sim_{\BX, X^\square} \quad \text{as long as $X^\ast + [0, \pphi]^d  \subset X^\square$.} $$
Note that by  \eqref{eq:def_R}  we have $\pphi \le R$. Thus the definitions of $X^\ast$ and $X^+$ in  \eqref{eq:nghbhdscompact}
  imply in particular  that
$$  \| \p\|_{\BX, X}  \le  \| \p\|^\sim_{\BX, (X^+)^\ast}   \le  \| \p\|^\sim_{\BX, (X^+)^+} \, . $$

\section{Example with the projective norm on \texorpdfstring{$\BX^{\otimes \protect\headingr}$}{Xr}}   \label{se:appendix_sym_norm}
For the convenience of the reader we show that the approach taken in \cite{AKM16} 
fits in the current framework if we use the projective norm on
\texorpdfstring{$\BX^{\otimes r}$}{Xr}.
Since $\BX$ is finite dimensional the bidual $\BX''$ equals $\BX$. 
If we use the projective norm $\| \cdot \|_\wedge$  on $\BX^{\otimes r}$
then dual norm on  $L \in (\BX^{\otimes r})'$ is given by
\begin{equation}
\| L\|'_\wedge = \| L \|_\vee = \sup \{ \langle L, \xi_1 \otimes \ldots \otimes \xi_r \rangle : \| \xi_i \|_\BX  \le 1 \hbox{ for $i=1, \ldots, r$} \}. 
\end{equation}
This is the usual norm of multilinear maps. If $L$ is symmetric, i.e., $SL = L$,  then one can also define the symmetric norm
\begin{equation}
\|L\|_{\sym} := \sup \{ \langle L, \xi \otimes \ldots \otimes \xi \rangle : \| \xi \|_\BX  \le 1\}.
\end{equation}
For an $r$ homogeneous polynomial $P_r$ we have 
$\| \overline{P_r} \|_{\sym} = \sup \{ P_r(\xi) : \| \xi \| \le 1\}$. 
We claim that 
\begin{equation}  \label{eq:equivalence_sym_vee}
\| L \|_{\sym} \le \|L \|_\vee \le  \frac{r^r}{r!} \|L \|_{\sym}  \quad \forall L \in (\BX^{\otimes r})'.
\end{equation}
The first inequality is trivial and the second follows by polarisation.
Indeed, assume that $\| \xi_i \| \le 1$  consider the Rademacher functions $R_1, \ldots R_r : [0,1] \to \{-1,1\}$,
set $\xi(t) = \sum_{i=1}^r R_i(t) \xi_i$ and use that $\| \xi(t) \| \le r$ to deduce that 
$$ r! \,  \langle L, \xi_1 \otimes \ldots \otimes \xi_r\rangle = \int_0^1 \langle L, \otimes^r \xi(t) \rangle \, \prod_{i=1}^r R_i(t) \, \d t 
\le \| L \|_{\sym}  \,  r^r.$$
The second estimate in   \eqref{eq:equivalence_sym_vee} is sharp for 
$\BX= \R^r$ equipped with the $\ell_1$ norm and the 
permanent $ \langle L, \xi^{(1)} \otimes \ldots \otimes \xi^{(r)} \rangle := \sum_{\pi \in S_r}   \prod_{i=1}^r  \xi^{(i)}_{\pi(i)}$.
To get the upper bound  $\| L \|_{\sym} \le \frac{r!}{r^r}$ use the geometric-arithmetic mean inequality. To get the lower
bound $\|L \|_\vee \ge 1$ take $\xi^{(i)} = e_i$. 

One can define a norm on general polynomials by
\begin{equation}
\| P \|_{T_0, \sym} := \sum_{r=0}^{r_0}  \| \overline{P_r}\|_{\sym}.
\end{equation}
and a corresponding norm on $\|F \|_{T_\p, \sym}$ on the Taylor polynomials of $F$. 
This is the approach taken in \cite{AKM16}. It is easy to see that  $\| S( L \otimes M) \|_{\sym}
\le \|  L\|_{\sym} \,  \| M \|_{\sym}$ for symmetric $L$ and $M$ and that the product property
$\|  PQ \|_{T_0, \sym} \le \| P \|_{T_0, \sym}$ holds,  see the proof of Proposition  \ref{pr:product_estimate_new} above 
or \cite{AKM16} eqn.\ (5.2).

However,  the contraction estimate  \eqref{eq:tensor_contraction} does not hold in general for $\| \cdot \|_{\sym}$ and thus the polynomial
estimate and the two-norm estimate  need not hold  for $\| \cdot \|_{T_\p, \sym}$.
Since $\|  \cdot \|_{\sym}$ and $\| \cdot \|_\vee$ are equivalent by \eqref{eq:equivalence_sym_vee} 
   these estimates  do of course hold if one includes an additional multiplicative constant $C(r_0)$.
   
To see that the contraction estimate need not hold consider $\BX = \R^2$ with the $\ell_\infty$ norm
$|\xi|_\infty =\max(|\xi_1|, |\xi_2|)$. Let $e^1, e^2$ denote the dual basis and set 
$L = e^1 \otimes e^1 - e^2 \otimes e^2$, $\p = e_1 - e_2$ and $M(\xi) = L(\xi \otimes \p)$. 
Then $\langle L, \xi \otimes \xi\rangle )= \xi_1^2 - \xi_2^2$ and hence $\| L \|_{\sym} = 1$. 
Moreover $|\p|_\infty = 1$, but for $\xi = e_1 + e_2$ we have $M( \xi)  = \xi_1 + \xi_2 = 2$ and thus $\|M \|_{\sym} \ge 2$
(in fact, the equality holds). One easily obtains a counterexample for the polynomial estimate by taking $F(\p) = \langle L, \p \otimes \p\rangle$.

\chapter{Estimates for Taylor Polynomials in \texorpdfstring{$\Z^{\protect\headingd}$}{Zd}}
Here we give a proof of  the remainder estimate which was the key 
ingredient in proving, in Section~\ref{subsec:contrC}, the contraction estimate for the linearised operator $\boldsymbol{C}^{(\boldsymbol q)}$.
Recall that, for $f : \Z^d \to \R$,  the discrete  $s$-th order Taylor polynomial at $a$ is given  by
$$ \hbox{$\tay^s_a$} f(z) := \sum_{|\alpha| \le s} \nabla^\alpha f(a) \,   \, b_\alpha(z-a)$$
where 
$$ b_\alpha(z) = \prod_{i=1}^d \binom{z_j}{\alpha_j}   \quad \hbox{and} \quad 
\binom{z_j}{\alpha_j} = \frac{z_j \cdots (z_j - \alpha_j +1)}{\alpha_j!}.$$
It is easy to see that $\nabla^\beta b_\alpha =b_{\alpha-\beta}$ with the conventions
$b_0 \equiv 1$ and $b_{\alpha-\beta} = 0$ if $\alpha - \beta \notin \NN_0^{\{1, \ldots, d\}}$.
Recall that $\mathcal U = \{ e_1, \ldots, e_d\}$.

\begin{lemma}  \label{le:taylor_remainder}
Let $s \in \NN_0$, $\rho \in \NN$ and define
$$ M_{\rho, s}
 = \sup \{ |\nabla^\alpha f(z)| :  |\alpha| = s+1, \,  z \in \Z^d \cap \big( a + [0, \rho]^d\big)  \}.$$
Then for all $\beta \in \mathbb{N}_0^{\mathcal U}$ with $t = |\beta|  \le s$
\begin{equation} \label{eq:taylor_remainder}
\left| \nabla^\beta[ f(z) - \hbox{$\tay_a^s$} f(z)] \right| \le 
M_{\rho, s}
 \binom{|z-a|_1}{s-t+1} \quad \forall z \in \Z^d \cap \big( a + [0, \rho]^d\big).
\end{equation}
\end{lemma}

The estimate is sharp for $a=0$ and $t=0$  since  the function $f(z) = \binom{z_1 + \ldots + z_d}{t+1}$ satisfies 
$\nabla^{\alpha} f = 1$ for all $\alpha$ with $|\alpha| = t+1$ (see proof).

\begin{proof}
This result in classical and is a (very) special case of Lemma 3.5 in \cite{BS15II}.
Since the notation here is simpler we include the short proof along the lines of \cite{BS15II}  for the 
convenience of the reader. We may assume that $a=0$.
It suffices to show  \eqref{eq:taylor_remainder} for $t=0$. Indeed, if the result is known for $t=0$ we can
use that $\nabla^\beta \tay_0^s f = \tay_0^{s-t} \nabla^\beta f$ and deduce that
$|\nabla^\beta f(z) - \tay_0^{s-t} \nabla^\beta f(z)| \le 
M_{\rho, s} \binom{|z|_1}{s-t+1}$.
Here we used that $M_{\rho, s-t}(\nabla^\beta f) \le M_{\rho, s}(f)$.

The proof for $t=0$ is by induction over the dimension $d$. 
We first note that for $z \in \NN_0^d$
$$ \binom{|z|_1}{s+1} = b_{s+1}(z_1 +\ldots + z_d) = \sum_{|\alpha| = s+1} b_\alpha(z).$$
Indeed  the first identity follows immediately from the definition of $b_{s+1}$ 
(as a polynomial on $\mathbb{Z}$) since $z_i \ge 0$. To prove the second identity we show that both side have the same discrete derivatives at $z=0$. 
Indeed the discrete derivative $\nabla^\beta$ of the left hand side evaluated at zero  is given by
$ b_{s+1-|\beta|}(0)$. This equals $1$ for $|\beta| = s+1$ and $0$ if $|\beta| \ne s+1$. The same assertion is
true for the right hand side. 

Thus it suffices to show that 
\begin{equation} \label{eq:remainder_induction_hypothesis}
| f(z) - \hbox{$\tay_0^s$} f(z) | \le 
M_{\rho, s}  \sum_{|\alpha| = s+1} b_\alpha(z).
\end{equation}
Note that if $z_j \in \Z$ and $0 < z_j < \alpha_j$ for some $j$ then $b_\alpha(z) = 0$. Thus
$$ b_\alpha(z) \ge 0 \quad \forall z \in \Z^d \cap [0, \rho]^d.$$

For $d=1$ we use the discrete Taylor formula with remainder
$$ f(z) = \sum_{r=0}^s  \nabla^r f(0) \, b_r(z) + \sum_{z'=0}^{z-1} b_s(z-1-z') \, \nabla^{s+1} f(z').$$
This formula is easily proved using induction over $s$ and the summation by parts formula
\begin{align*}  \sum_{z'=0}^{z-1} b_s(z-1-z') \, g(z')
=    b_{s+1}(z) g(0) + \sum_{z''=0}^{z-1}  b_{s+1}(z-1-z'') (g(z''+1) - g(z'')).
\end{align*}

Since
  $\nabla b_{s+1} = b_s$ we have
 $$\sum_{z'=0}^{z-1}| b_s(z-1-z')|  = \sum_{z'=0}^{z-1} b_s(z-1-z') = \sum_{z'=0}^{z-1}  b_s(z') =  b_{s+1}(z)$$
 and thus
  the Taylor formula with remainder implies   \eqref{eq:remainder_induction_hypothesis}
for $d=1$.

Now assume that  \eqref{eq:remainder_induction_hypothesis} holds for $d-1$.
Set $\alpha' = (\alpha_1, \ldots, \alpha_{d-1})$ and $\alpha = (\alpha', \alpha_d)$ and similarly $z = (z', z_d)$. 
Then the  induction hypothesis gives (for $z_j \ge 0$)
\begin{align}  \label{eq:appendixB_step_a}
 \Big| f(z', z_d) - \sum_{|\alpha'| \le s} \nabla^{\alpha'} \!   f(0,z_d)\,   b_{\alpha'}(z')   \Big| \le 
M_{\rho, s}
 \sum_{|\alpha'| = s+1} b_{\alpha',0}(z). 
 \end{align}
 Now by the result for $d=1$ applied to the $z_d$ direction 
$$ \Big|    \nabla^{\alpha'}  f(0,z_d) -   
 \sum_{\alpha_d \le s- |\alpha'|} \nabla^{(\alpha', \alpha_d)} f(0) \,   b_{\alpha_d}(z_d)  \Big|
\le 
M_{\rho, s}
    \, b_{s+1- |\alpha'|}(z_d).$$
   Since $b_{\alpha'}(z') \,  b_{\alpha_d}(z_d) = b_\alpha(z)$ it follows that
    \begin{align}  \label{eq:appendixB_step_b}
    \begin{split}
 &  \, \,     \Big|  \sum_{|\alpha'| \le s} \nabla^{\alpha'}  \! f(0,z_d) \,  b_{\alpha'}(z') - 
 \sum_{|\alpha| \le s} \nabla^\alpha f(0) \,   b_\alpha(z)  \Big| \\
   \le &  \, \, 
M_{\rho, s}
 \sum_{|\alpha'| \le s}      b_{\alpha'}(z') \, b_{s+1- |\alpha'|}(z_d) 
 =  M_{\rho, s}
 \sum_{|\alpha| = s+1, |\alpha'| \le s }   b_{\alpha}(z) . 
 \end{split}
 \end{align}
 Combining \eqref{eq:appendixB_step_a}   and \eqref{eq:appendixB_step_b} we see that 
  \eqref{eq:remainder_induction_hypothesis} holds for $d$. 
\end{proof}

\chapter{Combinatorial Lemmas}
In this appendix we state two lemmas that are used in the reblocking step.
\begin{lemma}\label{le:app1}
Let $X\in\Pck\setminus \mathcal{B}_k$ and $\upalpha(d)=(1+2^d)^{-1}(1+6^d)^{-1}$. 
Then 
\begin{align}\label{eq:app1}
|X|_k\geq (1+2\upalpha(d))|\pi(X)|_{k+1}. 
\end{align}
\end{lemma}
\begin{proof}
Recall that the map $\pi:\Pcal_k\to \mathcal{P}_{k+1}$ 
was defined in \eqref{eq:defofpi} and \eqref{eq:pifactor} and 
determines to which polymer on the next scale the contribution of a polymer $X$ is assigned. 
By definition, it satisfies $\pi(X)=\overline{X}$ for $X\in \Pck\setminus \mathcal{S}_k$ and in this case \eqref{eq:app1} was shown in Lemma 6.15 in \cite{Bry09}.
For $X\in \Scal_k\setminus \Bcal_k$ we use that $\pi(X)\in \Bcal_{k+1}$ to conclude
\begin{align}
|X|_k\geq 2 =2|\pi(X)|_{k+1}\geq (1+2\upalpha)|\pi(X)|_{k+1}.
\end{align}
\end{proof}
\begin{lemma}\label{le:app2}
There exists $\updelta(d,L)<1$ such that
\begin{align}\label{eq:app2}
\sum_{\substack{X\in \Pck\setminus \mathcal{S}_k \\
\pi(X)}=U}\updelta^{|X|_k}\leq 1        
\end{align}
for any $k\in \mathbb{N}$ and $U\in\Pckp$.
\end{lemma}
\begin{proof}
Recall that $\pi(X)=\overline{X}$ for $X\in \Pck\setminus \mathcal{S}_k$. Now this is Lemma 6.16 in \cite{Bry09}.
\end{proof}
\end{appendix}
\bigskip

{\centerline  {\textbf{Acknowledgements}}}
\smallskip

We are most grateful to David Brydges for explaining us both the grand picture and many of the fine details of his approach to rigorous renormalisation and for his constant encouragement.
SB and SM have been supported by the 
Deutsche Forschungsgemeinschaft (DFG, German Research Foundation) 
through the Hausdorff Center for Mathematics
(GZ EXC 59 and 2047/1, Projekt-ID 3906$ $85813) and the collaborative research centre
'The mathematics of emerging effects'
(CRC 1060, Projekt-ID   2115$ $04053 ). This work was initiated through the research group FOR 718 
'Analysis and Stochastics in Complex Physical Systems'.
The research of RK was  supported by the grants GA\v CR 16-15238S and 20-08468S.
RK and SB would also like to thank the Isaac Newton Institute for Mathematical Sciences for support and hospitality during the programme \emph{Scaling limits, rough paths, quantum field theory} (supported by EPSRC Grant Number EP/R014604/1) where some work on the final version of the paper was undertaken. 

\backmatter

\end{document}